\documentclass[
11pt, 
english, 
onehalfspacing, 
headsepline, 
]{MastersDoctoralThesis} 

\usepackage[T1]{fontenc} 

\usepackage{palatino} 

\usepackage{amsmath,amssymb}
\usepackage{caption,subcaption}
\usepackage{amsthm}
\newtheorem{theorem}{Theorem}
\usepackage{epigraph}
\usepackage{verbatim}
\usepackage{cancel}
\usepackage{color}
\usepackage{titlesec}

\usepackage{natbib}

\usepackage[autostyle=true]{csquotes} 

\thesistitle{Theory of one-dimensional Vlasov-Maxwell equilibria: \\
with applications to collisionless current sheets and flux tubes} 
\supervisor{Professor Thomas \textsc{Neukirch}} 
\examiner{} 
\degree{Ph.D.} 
\author{Oliver \textsc{Allanson}} 
\addresses{} 

\subject{Mathematics} 
\university{\href{https://www.st-andrews.ac.uk}{University of St Andrews}} 
\department{\href{https://www.st-andrews.ac.uk/maths/}{School of Mathematics and Statistics}} 
\group{\href{http://www-solar.mcs.st-andrews.ac.uk}{Solar and Magnetospheric Theory}} 
\faculty{\href{https://www.st-andrews.ac.uk/study/ug/options/faculties/science/}{Faculty of Science}} 

\begin{document}

\frontmatter 

\pagestyle{plain} 

\begin{titlepage}

\begin{center}
\vspace*{-1cm}

\HRule \\[0.4cm] 
{\huge \bfseries \ttitle\par}\vspace{0.4cm} 
\HRule \\[1.5cm] 

{\huge {\bf Oliver Douglas Allanson}}
\vspace{-1cm}

\includegraphics[width=10cm]{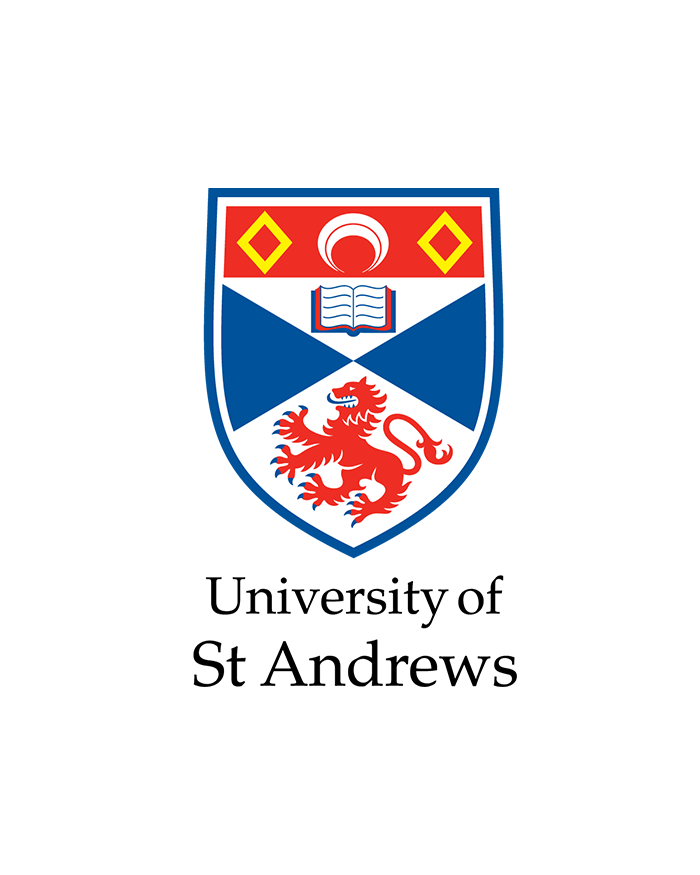}
\vspace{-1cm}

{\large This thesis is submitted in partial fulfilment for the degree of PhD\\ at the University of St Andrews}

\today
 
\end{center}

\end{titlepage}

\cleardoublepage

\section*{\huge \centering Abstract}
\noindent Vlasov-Maxwell equilibria are characterised by the self-consistent descriptions of the steady-states of collisionless plasmas in particle phase-space, and balanced macroscopic forces. We study the theory of Vlasov-Maxwell equilibria in one spatial dimension, as well as its application to current sheet and flux tube models.

\noindent The `inverse problem' is that of determining a Vlasov-Maxwell equilibrium distribution function self-consistent with a given magnetic field. We develop the theory of inversion using expansions in Hermite polynomial functions of the canonical momenta. Sufficient conditions for the convergence of a Hermite expansion are found, given a pressure tensor. For large classes of DFs, we prove that non-negativity of the distribution function is contingent on the magnetisation of the plasma, and make conjectures for all classes.

\noindent The inverse problem is considered for nonlinear `force-free Harris sheets'. By applying the Hermite method, we construct new models that can describe sub-unity values of the plasma beta ($\beta_{pl}$) for the first time. Whilst analytical convergence is proven for all $\beta_{pl}$, numerical convergence is attained for $\beta_{pl}=0.85$, and then $\beta_{pl}=0.05$ after a `re-gauging' process. 

\noindent We consider the properties that a pressure tensor must satisfy to be consistent with `asymmetric Harris sheets', and construct new examples. It is possible to analytically solve the inverse problem in some cases, but others must be tackled numerically. We present new exact Vlasov-Maxwell equilibria for asymmetric current sheets, which can be written as a sum of shifted Maxwellian distributions. This is ideal for implementations in particle-in-cell simulations.

\noindent We study the correspondence between the microscopic and macroscopic descriptions of equilibrium in cylindrical geometry, and then attempt to find Vlasov-Maxwell equilibria for the nonlinear force-free `Gold-Hoyle' model. However, it is necessary to include a background field, which can be arbitrarily weak if desired. The equilibrium can be electrically non-neutral, depending on the bulk flows.\\

\noindent \textbf{PhD supervisor:} \\ Prof Thomas Neukirch, University of St Andrews \\
\noindent \textbf{Viva examiners:} \\ Dr Andrew Wright, University of St Andrews \\ Prof Alexander Schekochihin, University of Oxford.


\cleardoublepage


\section*{\huge\centering  Candidate's Publications}

From the different research projects in my doctoral studies, the following papers have been published, or have been submitted:
\begin{enumerate}
\item O. Allanson, T. Neukirch, F. Wilson \& S. Troscheit: \\An exact collisionless equilibrium for the Force-Free Harris Sheet with low plasma beta, \emph{Physics of Plasmas}, {\bf 22}, 102116, 2015 
\item O. Allanson, T. Neukirch, S. Troscheit \& F. Wilson: \\From one-dimensional fields to Vlasov equilibria: theory and application of Hermite polynomials, \emph{Journal of Plasma Physics}, {\bf 82}, 905820306, 2016 
\item O. Allanson, F. Wilson \& T. Neukirch: \\Neutral and non-neutral collisionless plasma equilibria for twisted flux tubes: The Gold-Hoyle model in a background field, \emph{ Physics of Plasmas} {\bf 23}, 092106, 2016
\item O. Allanson, F. Wilson, T. Neukirch, Y.-H. Liu, and J. D. B. Hodgson:\\ Exact Vlasov-Maxwell equilibria for asymmetric current sheets, \emph{Geophysical Research Letters}, {\bf 44}, 2017
\item O. Allanson, T. Neukirch and S. Troscheit:\\ The inverse problem for collisionless plasma equilibria, \emph{Invited paper for The IMA Journal of Applied Mathematics}, submitted

\end{enumerate}

\cleardoublepage



\section*{\huge \centering Acknowledgements} 

There are people without whom this PhD would not have been possible, and there are those without whom it would not have been the same. \\

\noindent Those in the latter camp:
\begin{description}
\item[Thomas Neukirch] Chief amongst them. Thank you Thomas for supporting me wholeheartedly at every stage, for making sure that I didn't make a complete mess out of this, and for this adventure in Vlasov theory. One day we'll convince everyone that it is so much more interesting than MHD. 
\item[Roomies] To Sophie Dawe, Thomas Elsden and Cara Fraser (with an honourable mention for Jonathan Fraser, you were there often enough). Thank you for teaching me that you can never have too much of a good thing. And thank you for the impromptu - but regular - dancing to the house song. A particular thank you to Tom for the camaraderie of our days spent at Martyr's Kirk; the gleeful mid-morning breaks spent shivering over bacon sandwiches and the Bialetti are crystallised in my memory.
\item[The YRM2016 team] Thomas Bourne, Zo\"{e} Sturrock, Sascha Troscheit, Cristina Evans, Daniel Bennett, and Fiona MacFarlane. That was all very fun wasn't it, and extremely efficient!? 
\item[Sascha Troscheit] I don't know how you managed to get mentioned twice? Thank you for the innumerable distracting visits to my office, and for implicitly agreeing to make them useful, by engaging with me in my plasma physics problems. One day I'll help you with slippery devil's staircases. 
\end{description}

\noindent I would like to thank the STFC for allowing me to spend three and a half more years as a student, enabling me to fiddle around in plasma physics. I would also like to acknowledge grants from the National Science Foundation, the Vlasovia conference, and the Royal Astronomical Society for making possible my trips to the 2016 AGU Chapman conference in Dubrovnik, the 2016 Vlasovia meeting in Calabria, and the 2017 American Geophysical Union fall meeting in San Francisco, respectively.


\cleardoublepage




\section*{\huge \centering Candidate's Declarations}

I, Oliver Douglas Allanson, hereby certify that this thesis, which is approximately 38000 words in length, has been written by me, and that it is the record of work carried out by me, or principally by myself in collaboration with others as acknowledged, and that it has not been submitted in any previous application for a higher degree. 

\vspace{1cm}

\noindent I was admitted as a research student in September 2013 and as a candidate for the degree of Ph.D in September 2014; the higher study for which this is a record was carried out in the University of St Andrews between 2013 and 2017. 

\vspace{1cm}

\noindent Date: \\
\rule[0.5em]{25em}{0.5pt} 
\vspace{1cm}

\noindent Signature of candidate:\\
\rule[0.5em]{25em}{0.5pt} 

\cleardoublepage

\section*{\huge \centering Supervisor's Declaration}

I hereby certify that the candidate has fulfilled the conditions of the Resolution and Regulations appropriate for the degree of Ph.D in the University of St Andrews and that the candidate is qualified to submit this thesis in application for that degree. 

\vspace{1cm}

\noindent Date: \\
\rule[0.5em]{25em}{0.5pt} 
\vspace{1cm}

\noindent Signature of supervisor:\\
\rule[0.5em]{25em}{0.5pt} 


\cleardoublepage


\section*{\huge  \centering Permission for Publication} 

In submitting this thesis to the University of St Andrews I understand that I am giving permission for it to be made available for use in accordance with the regulations of the University Library for the time being in force, subject to any copyright vested in the work not being affected thereby. I also understand that the title and the abstract will be published, and that a copy of the work may be made and supplied to any bona fide library or research worker, that my thesis will be electronically accessible for personal or research use unless exempt by award of an embargo as requested below, and that the library has the right to migrate my thesis into new electronic forms as required to ensure continued access to the thesis. I have obtained any third-party copyright permissions that may be required in order to allow such access and migration, or have requested the appropriate embargo below. 
\vspace{1cm}

\noindent The following is an agreed request by candidate and supervisor regarding the publication of this thesis:
\vspace{1cm}

\noindent PRINTED COPY

\noindent No embargo on print copy 
\vspace{1cm}

\noindent ELECTRONIC COPY

\noindent No embargo on electronic copy

\vspace{1cm}

\noindent Date: \\
\rule[0.5em]{25em}{0.5pt} 
\vspace{1cm}

\noindent Signature of candidate:\\
\rule[0.5em]{25em}{0.5pt} 
\vspace{1cm}

\noindent Signature of supervisor:\\
\rule[0.5em]{25em}{0.5pt} 


\cleardoublepage



\noindent 

\section*{\huge \centering Some important notation}
 \begin{table*}[ht]
 \centering
  \begin{center}
  \begin{tabular}{lll} 
 \noindent $\boldsymbol{x}$ & (Particle) position &  \\
$\boldsymbol{v}$ & (Particle) velocity & $\boldsymbol{v}=d\boldsymbol{x}/dt$  \\
$\boldsymbol{A}$&Magnetic vector potential &\\
$\boldsymbol{B}$ & Magnetic field& $\boldsymbol{B}=\nabla\times\boldsymbol{A}$\\
$\phi$ & Electrostatic scalar potential & \\
$\boldsymbol{E}$ & Electric field & $\boldsymbol{E}=-\nabla\phi-\partial\boldsymbol{A}/\partial t$\\
$\alpha(\boldsymbol{r})$&Force-free parameter&$\alpha(\boldsymbol{r})=\boldsymbol{B}\cdot(\nabla\times\boldsymbol{B})/|\boldsymbol{B}|^2$\\
$s$ & Particle species $s$&\\
$m_s$ & Particle mass&\\
$q_s$ & Particle charge& $e=q_i=-q_e$\\
$f_s$ & Particle distribution function (DF) &\\
$n_s$ & Particle number density & $n_s=\int f_s d^3v$\\
$\rho_s$ & Mass density& $\rho_s=m_sn_s$\\
$\sigma$ & Electric charge density& $\sigma=\sum_s\sigma_s=\sum_sq_sn_s$\\
$\boldsymbol{V}_s$&Bulk flow& $\boldsymbol{V}_s=n_s^{-1}\int\boldsymbol{v}f_sd^3v$ \\
$\boldsymbol{j}$&Electric current density& $\boldsymbol{j}=\sum_s \boldsymbol{j}_s=\sum_sq_sn_s\boldsymbol{V}_s$\\
$\boldsymbol{w}_s$&Particle flow relative to the bulk&$\boldsymbol{w}_s=\boldsymbol{v}-\boldsymbol{V}_s$\\
$P_{ij}$& Thermal pressure tensor& $P_{ij}=\sum_sP_{ij,s}$\\
&&$=\sum_s\int w_{is}w_{js}f_sd^3v$\\
$p$ & Scalar thermal pressure &$p=\text{Tr}(P_{ij})/3$  \\
$H_s$& Particle Hamiltonian (energy)&$H_s=m_s\boldsymbol{v}^2/2+q_s\phi$\\
$\boldsymbol{p}_{s}$& Particle canonical momenta&$\boldsymbol{p}_{s}=m_s\boldsymbol{v}+q_s\boldsymbol{A}$\\
$\beta_s$&Thermal beta&$\beta_s=1/(m_sv_{th,s}^2)$\\
$v_{\text{th},s}$&Particle thermal velocity&\\
$r_{Ls}$ &Thermal Larmor radius&$r_{Ls}=m_sv_{th,s}(e|\boldsymbol{B}|)^{-1}$\\
$T_s$&Temperature&$T_s=1/(k_B\beta_s)$\\
$L$&Macroscopic length scale& e.g. current sheet width\\
$\delta_s$&Magnetisation parameter&$\delta_s=r_{Ls}/L$\\
$\beta_{pl}$&Plasma beta&$\beta_{pl}=\sum_s\beta_{pl,s}$  \\
&&$=\sum_sn_sk_BT_s/(\boldsymbol{B}^2/(2/\mu_0))$\\
$\lambda_D$&Debye radius&$\lambda_D=\sqrt{\epsilon_0k_BT_e/(n_ee^2)}$\\

\end{tabular}
\end{center}
\end{table*}


\cleardoublepage



\section*{\huge \centering Physical constants (SI units)}
\begin{table*}[ht]
 
  \begin{center}
  \begin{tabular}{lll} 

\noindent Boltzmann's constant&$k_B$&\SI{1.3807e-23}{\joule\per\kelvin}\\
Speed of light in a vacuum& $c$ & \SI{2.9979e8}{\meter\per\second} \\
Permittivity of free space&$\epsilon_0$&\SI{8.8542e-12}{\farad\per\metre}\\
Permeability of free space&$\mu_0$&\SI{4\pi e-7}{\henry\per\metre}\\
Proton mass&$m_i$&\SI{1.6726e-27}{\kilogram}\\
Electron mass&$m_e$&\SI{9.1094e-31}{\kilogram}\\
Elementary charge&$e$&\SI{1.6022e-19}{\coulomb}\\

\end{tabular}
\end{center}
\end{table*}

\section*{\huge \centering Abbreviations}
\begin{table*}[ht]
 \centering
  \begin{center}
  \begin{tabular}{ll}

\noindent \textbf{DF} & \textbf{D}istribution \textbf{F}unction \\
\textbf{VM} & \textbf{V}lasov-\textbf{M}axwell \\
\textbf{MHD} & \textbf{M}agneto\textbf{H}ydro\textbf{D}ynamics\\
\textbf{GEM} & \textbf{G}eospace \textbf{E}nvironmental \textbf{M}odelling \\
\textbf{1D} & \textbf{1 D}imensional\\
\textbf{2D} & \textbf{2 D}imensional \\
\textbf{RHS} & \textbf{R}ight \textbf{H}and \textbf{S}ide\\
\textbf{LHS} & \textbf{L}ight \textbf{H}and \textbf{S}ide\\
\textbf{FT}&\textbf{F}ourier \textbf{T}ransform\\
\textbf{IFT}&\textbf{I}nverse \textbf{F}ourier \textbf{T}ransform\\
\textbf{PIC} & \textbf{P}article-\textbf{I}n-\textbf{C}ell\\
\textbf{FFHS} & \textbf{Force}-\textbf{F}ree \textbf{H}arris \textbf{S}heet \\
\textbf{MMS} & \textbf{M}agnetospheric \textbf{M}ulti\textbf{S}cale mission\\
\textbf{AHS} & \textbf{A}symmetric \textbf{H}arris \textbf{S}heet\\
\textbf{AH+G} & \textbf{A}symmetric \textbf{H}arris plus \textbf{G}uide\\
\textbf{GH} & \textbf{G}old-\textbf{H}oyle\\
\textbf{GH+B} & \textbf{G}old-\textbf{H}oyle plus \textbf{B}ackground\\

\end{tabular}
\end{center}
\end{table*}


\cleardoublepage




{\large \centering
\noindent This thesis is dedicated to\\Sophie Dawe's love, and levity \\Phil Michaels' life, and spirit \\my parents' support, and tender care\\my family present, passed and in-law \\my friends for bringing me to life, \\
and cutting me down to size. \\

\vspace{3mm}
\noindent\rule{0.4\linewidth}{1pt}
\vspace{1.5cm}

\noindent
At quite uncertain times and places, \\
The atoms left their heavenly path, \\
And by fortuitous embraces, \\
Engendered all that being hath.\\ 
\vspace{3mm}
\noindent And though they seem to cling together,\\ 
And form "associations" here, \\
Yet, soon or late, they burst their tether,\\ 
And through the depths of space career.\\





\vspace{3mm}
\noindent Soon, all too soon, the chilly morning, \\\
This flow of soul will crystallize, \\
Then those who Nonsense now are scorning,\\ 
May learn, too late, where wisdom lies.\bigbreak

\vspace{3mm}

\noindent James Clerk Maxwell\\
\noindent Molecular evolution (abridged)\\
 \emph{Nature}, {\bf 8}, 205, \textit{page 473} (1873)\\}


\cleardoublepage


\tableofcontents 

\cleardoublepage


\mainmatter 

\pagestyle{thesis} 



\chapter{Introduction} \label{Intro} 

\epigraph{\emph{Most important part of doing physics is the knowledge of approximation.}}{\textit{Lev Landau}}

\section{The hierarchy of plasma models}
More than $99\%$ of the known matter in the Universe is in the plasma state \citep{Baumjohannbook}, by far the most significant material constituent of stellar, interplanetary, interstellar and intergalactic media. Not only is a deep understanding of plasmas then clearly necessary to understand the physics of our universe, but plasmas are also of real interest to us on Earth. Nuclear fusion experiments - and in principle, future power stations - necessarily exploit the plasma state to work, either using high-temperature plasmas confined by strong magnetic fields, or plasmas formed by the laser ablation of a solid fuel target. 

Plasmas are often known as the `fourth' state of matter, lying after the `third', and more familiar gaseous state. At a temperature above $100,000$K, most matter exists in an ionised state, however plasmas can exist at much lower temperatures should ionisation mechanisms exist, and if the density is sufficiently low \citep{Krallbook-1973}. Figures \ref{fig:zoo} and \ref{fig:zoo2} display some examples from the rich array of plasma environments in temperature-density scatter plots; from the relatively cool and diffuse plasmas of interstellar space, to the incredibly dense and hot plasmas of stellar and laboratory fusion. Since there is such variety in the physical conditions able to sustain plasmas, the `plasma state' may best describe \emph{collective behaviours}, the characteristics that persist despite the range of physical conditions that can sustain plasmas (we see from Figure \ref{fig:zoo2} that even the free electrons in metals can be considered, or modelled, as a plasma). Matter is in a plasma state when the degree of ionisation is sufficiently high that the dynamical behaviour of the particles is dominated by electromagnetic forces \citep{Fitzpatrickbook}, and this can even be the case for ionisation levels as low as a fraction of a percent \citep{Peratt-1996}. 
\begin{figure}
    \centering
     \begin{subfigure}[b]{0.7\textwidth}
        \includegraphics[width=\textwidth]{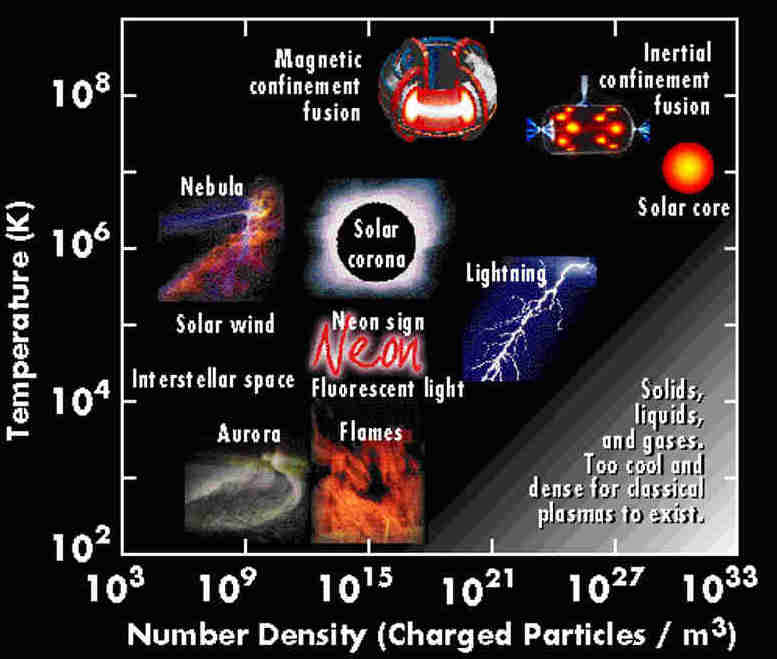}
    \caption{{\small The plasma `zoo': A density-temperature plot displaying various plasma environments and phenomena, and their contrast to solids, liquids and gases. {\bf Image copyright:} \href{http://www.cpepphysics.org}{Contemporary Physics Education Project}, (reproduced with permission).}}\label{fig:zoo}
     \end{subfigure}
    \begin{subfigure}[b]{0.7\textwidth}
     \includegraphics[width=\textwidth]{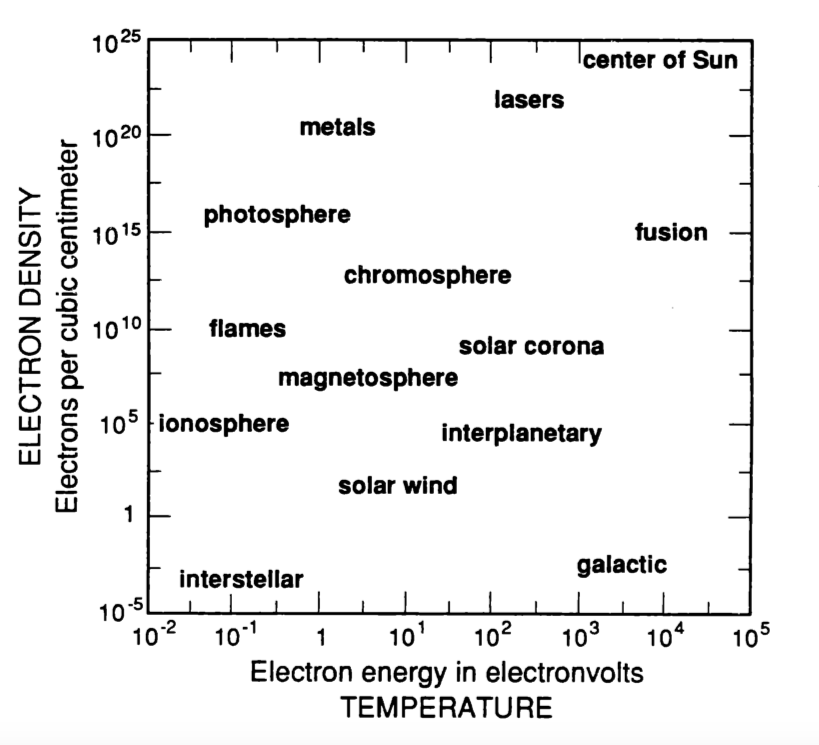}
    \caption{{\small A temperature-density plot reproduced from \citet{Peratt-1996}, focussing on the environments in which plasmas appear. {\bf Image Copyright:} Springer, \href{http://link.springer.com/journal/10509}{\emph{Astrophysics and Space Science}} {\bf 242}, 1-2, (1996), pp. 93-163., copyright (1996), (reproduced with permission). }}\label{fig:zoo2}
     \end{subfigure}
     \caption{\small The variety of plasma conditions and environments.}\label{}
\end{figure}
Whilst many of these plasmas possess some shared tendencies and behaviours, it is not possible to capture all the detailed physics of the entire variety of plasma processes with one particular mathematical toolkit or model. Not only may some models fail to capture certain aspects of the physics by virtue of the approximations made, but they may be inefficient, or in fact insoluble when applied in practice. Hence, plasma physics is a discipline with a rich variety of perspectives and methods. Within each of these paradigms we make certain approximations and ordering assumptions, in order to capture the essence of the problem at hand.

\subsection{Single particle motion}\label{sec:guiding}
Taking the viewpoint of particulate matter as the fundamental approach, then a `full' description of plasmas is found by solving the (Lorentz) equation of motion of each individual particle, written in classical form as 
\begin{equation}
\boldsymbol{F}_{s}(\boldsymbol{x}(t),\boldsymbol{v}(t);t)=q_s(\boldsymbol{E}(\boldsymbol{x},t)+\boldsymbol{v}(t)\times\boldsymbol{B}(\boldsymbol{x},t)),\label{eq:Lorentz}
\end{equation}
with the force, $\boldsymbol{F}_s$, on a test particle of species $s$, of charge $q_s$, at position $\boldsymbol{x}$, and with velocity $\boldsymbol{v}$, when under the influence of electric and magnetic fields, $\boldsymbol{E}$ and $\boldsymbol{B}$. One can in principle integrate in time to calculate the trajectory of the particle for all future times (e.g. see \cite{Vekstein-2002}),
\[
\boldsymbol{x}(t)=\int_{t_0}^t \boldsymbol{v} (t^\prime)dt^\prime,
\]
for $\boldsymbol{v}(t_0)$ some initial condition. However, in all but the simplest electromagnetic field geometries these integrals may not even be able to be written down, and/or one might have to resort to numerical methods to calculate the trajectory. One more complication is the effect of the charged particles on the electromagnetic fields, $\boldsymbol{E}$ and $\boldsymbol{B}$, and this shall be discussed in Section \ref{sec:machine}.

If a plasma is sufficiently magnetised it has small parameters  
\begin{equation}
\frac{r_L}{L}\ll 1,\hspace{3mm} \frac{1/\Omega}{\tau}\ll 1,\nonumber
\end{equation}
for $r_L$ and $\Omega$ the characteristic values of the Larmor radius and gyrofrequency of individual particle gyromotion respectively, and $L$ and $\tau$ the characteristic length and time scales upon which the electromagnetic fields vary. In such a case there is a well understood treatment for particle orbits, namely \emph{Guiding Centre} theory (e.g. see \cite{Northrop-1961,Littlejohn-1983,Cary-2009}). Guiding centre theory models particle motion as a superposition of rapid gyromotion and a comparitively slow \emph{secular} drift (e.g. see \cite{Morozov-1966}). This gyromotion is depicted in Figure \ref{fig:Northrop}, reproduced from \citet{Northrop-1963}; in which the notation $\boldsymbol{\rho}$ and $\rho$ are used for the gyroradius `vector' and magnitude respectively (in contrast to the use of $r_{L}$ herein); $\boldsymbol{r}$ is the particle position; and $\boldsymbol{R}$ is the guiding center position, such that $\boldsymbol{r}=\boldsymbol{R}+\boldsymbol{\rho}$.
\begin{figure}
    \centering
        \includegraphics[width=0.7\textwidth]{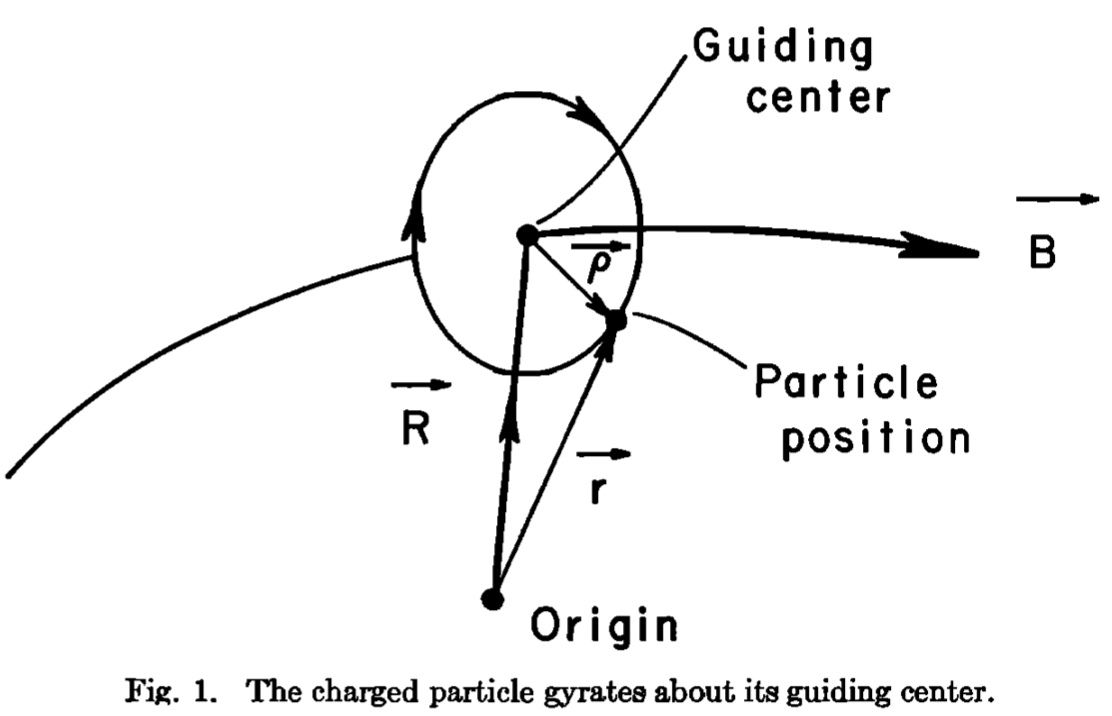}
    \caption{{\small A figure from \citet{Northrop-1963}. This figure depicts the gyromotion about the local magnetic field of a positively charged particle. {\bf Image copyright:} \href{http://sites.agu.org}{American Geophysical Union}  (reproduced with permission).}}
    \label{fig:Northrop}
\end{figure}
The local gyromotion is governed by the conservation (to lowest order) of the \emph{magnetic moment}, 
\begin{equation}
\mu =\frac{m_sv_{\perp}^2}{2|\boldsymbol{B}|},\nonumber
\end{equation}
for $m_s$ the mass of a particle, and $v_{\perp}^2 $ the square magnitude of the particle velocity normal to the local magnetic field. This theory is very useful for heuristic understanding of individual particle motion, and for the study of `test particles' embedded in a system of interest (e.g. see \cite{Threlfall-2015, Borissov-2016}), however not for `building up' a theory that models the evolution of the particles and electromagnetic fields self-consistently. In a situation in which many particles are present, the self-consistent modelling of all of the particles would in practice require knowledge of the individual particle interactions via the electromagnetic fields of mixed origin (microscopic/self-generated and macroscopic/external fields), and in principle collisions, which is mathematically unwieldy. However, we note here that it is possible - whilst unconventional - to use $N$-body particle dynamics to study collective effects in plasma physics (e.g. see \cite{Pines-1952,Escande-2016}), including the recent work of Dominique Escande and collaborators, who have taken an N-body approach to `re-deriving' physical phenomena, such as Debye shielding and Landau Damping (see Figure \ref{fig:Escande} for a representation of how their work `sidesteps' the more traditional routes). 

\begin{figure}
    \centering
        \includegraphics[width=0.7\textwidth]{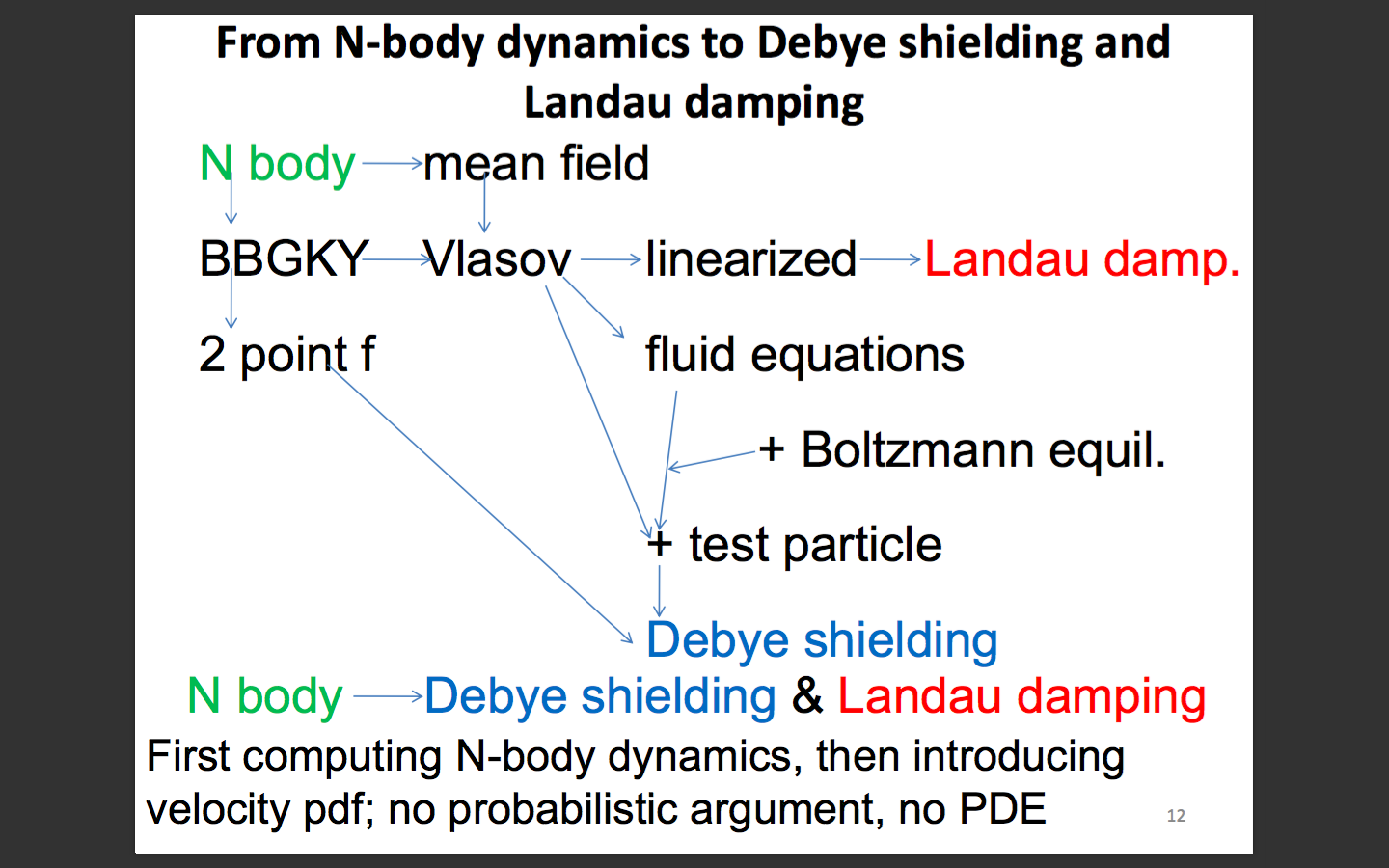}
    \caption{{\small A representation of the different models, approaches and phenomena in plasma physics. {\bf Image copyright:} Dominique Escande: From his presentation at the EPS Conference on Plasma Physics 2015 in Lisbon (reproduced with permission).}}\label{fig:Escande}
\end{figure}

\subsection{Kinetic theory}\label{sec:machine}
To move forward we require a mean-field/statistical formalism that allows for a self-consistent set of evolution equations, involving the quantities that both describe the particles and electromagnetic fields. The electromagnetic fields are governed by Maxwell's equations, and given in free space as 
\begin{eqnarray}
\nabla\cdot\boldsymbol{E}&=&\frac{\sigma}{\epsilon_0},\label{eq:Gaussequation}\\
\nabla\times\boldsymbol{E}&=&-\frac{\partial\boldsymbol{B}}{\partial t},\label{eq:Faradaylaw}\\
\nabla\times\boldsymbol{B}&=&\mu_{0}\boldsymbol{j}+\frac{1}{c^2}\frac{\partial\boldsymbol{E}}{\partial t},\\
\nabla\cdot\boldsymbol{B}&=&0,
\end{eqnarray}
for $\sigma$ and $\boldsymbol{j}$ the charge and current densities respectively (e.g. see \cite{Griffithsbook}). The electric permittivity and magnetic permeability \emph{in vacuo} are given by $\epsilon_{0}$ and $\mu_{0}$ respectively, and they are related by $c^{2}=1/(\mu_{0}\epsilon_{0})$, for $c$ the speed of light in free space. The electric and magnetic fields are defined as derivatives of the electrostatic scalar potential, $\phi$, and the magnetic vector potential, $\boldsymbol{A}$, according to
\begin{eqnarray}
\boldsymbol{E}&=&-\nabla\phi-\frac{\partial \boldsymbol{A}}{\partial t},\label{eq:Edef}\\
\boldsymbol{B}&=&\nabla\times\boldsymbol{A}.\label{eq:Bdef}
\end{eqnarray}
The potential functions are themselves `sourced' by $\sigma$ and $\boldsymbol{j}$, respectively,
\begin{eqnarray}
\phi (\boldsymbol{x},t)&=&\frac{1}{4\pi\epsilon_{0}}\int_{-\infty}^\infty \int_{-\infty}^\infty\int_{-\infty}^\infty \frac{\sigma (\boldsymbol{x}^\prime,t_r)}{|\boldsymbol{x}-\boldsymbol{x}^\prime|} d^3x^\prime,\label{eq:phisource}\\
\boldsymbol{A}(\boldsymbol{x},t)&=&\frac{\mu_0}{4\pi}\int_{-\infty}^\infty \int_{-\infty}^\infty\int_{-\infty}^\infty \frac{\boldsymbol{j} (\boldsymbol{x}^\prime,t_r)}{|\boldsymbol{x}-\boldsymbol{x}^\prime|}d^3x^\prime,\label{eq:Asource}
\end{eqnarray}
for $t_r=t-|\boldsymbol{x}-\boldsymbol{x}^\prime|/c$ the \emph{retarded time} \citep{Griffithsbook}. The charge and current densities can be calculated by taking \emph{moments} of the \emph{1-particle distribution functions} (DF), $f_s(\boldsymbol{x},\boldsymbol{v};t)$ for particle species $s$ (e.g. see \cite{Krallbook-1973,Schindlerbook}), over velocity space
\begin{eqnarray} 
\sigma(\boldsymbol{x},t)&=&\sum_s q_sn_s=\sum_s q_s\int_{-\infty}^\infty \int_{-\infty}^\infty\int_{-\infty}^\infty f_sd^3v,\label{eq:fmomentzero}\\
\boldsymbol{j}(\boldsymbol{x},t)&=&\sum_s q_sn_s\boldsymbol{V} _s=\sum_sq_s\int_{-\infty}^\infty\int_{-\infty}^\infty\int_{-\infty}^\infty \boldsymbol{v}f_sd^3v,\label{eq:fmomentone}
\end{eqnarray}
with $n_s$ and $\boldsymbol{V} _s$ the number density and bulk velocity of particle species $s$ respectively. Hereafter we us the notation $d^3x$ and $d^3v$ to imply triple integration over all position and velocity space respectively,
\begin{eqnarray}
\int d^3x&:=& \int_{-\infty}^\infty\int_{-\infty}^\infty\int_{-\infty}^\infty d^3x,\nonumber\\
\int d^3v&:= &\int_{-\infty}^\infty\int_{-\infty}^\infty\int_{-\infty}^\infty d^3v,\nonumber
\end{eqnarray}
unless otherwise stated. The DF, $f_{s}$, represents the number density of particles in a microscopic volume of six-dimensional phase-space at a particular time, such that
\begin{eqnarray}
f_s(\boldsymbol{x},\boldsymbol{v};t)d^3xd^3v&=&\# \text{ of particles in volume} \,d^3x  \,\text{centred on} \,\boldsymbol{x}\nonumber\\
&&\text{with velocities in the range}\, (\boldsymbol{v},\, \boldsymbol{v}+d\boldsymbol{v}).\nonumber
\end{eqnarray} 
Note that one can instead use the Klimontovich-Dupree description to exactly describe the particles using Dirac-Delta functions in phase space, but this approach is really only useful for formal considerations \citep{Krallbook-1973}. 

Now we are in a position to imagine the `machine' behind nature's self-consistent evolution of the particles and fields in the plasma, in the following way:
\begin{description}
\item[Statistical description:] $f_s (\boldsymbol{x},\boldsymbol{v},t_r)$ is found by `coarse graining' (or `ensemble averaging') the exact positions and velocities of the particles of species $s$ at time $t_r$ \citep{Krallbook-1973,Fitzpatrickbook} 
\item[Source terms:] $\sigma (\boldsymbol{x},t_r)$ and $\boldsymbol{j}(\boldsymbol{x},t_r)$ are then found by integrating $f_s(\boldsymbol{x},\boldsymbol{v},t_r)$ over velocity space (Equations \ref{eq:fmomentzero} and \ref{eq:fmomentone})
\item[Potentials:] $\phi (\boldsymbol{x},t)$ and $\boldsymbol{A}(\boldsymbol{x},t)$ are found by integrating $\sigma (\boldsymbol{x},t_r)$ and $\boldsymbol{j}(\boldsymbol{x},t_r)$ (Equations \ref{eq:phisource} and \ref{eq:Asource})
\item[Forces:] $\boldsymbol{F}_s(t)$ is found by differentiating the $\phi (\boldsymbol{x},t)$ and $\boldsymbol{A}(\boldsymbol{x},t)$ (Equations \ref{eq:Edef} and \ref{eq:Bdef})
\item[Velocities:] $\boldsymbol{v}(t+\delta t)$ is found by integrating the Lorentz force, $\boldsymbol{F}_s(\boldsymbol{x},t)$, for $\delta t$ some infinitesimal time (Equation \ref{eq:Lorentz})
\item[Positions:]  $\boldsymbol{x}(t+2\delta t)$ is found by integrating $\boldsymbol{v}(t+\delta t)$
\item [Statistical description:] $f_s (\boldsymbol{x},\boldsymbol{v},t+2\delta t)$ is found by ... and so the cycle continues.
\end{description}

To put these ideas on a firm mathematical footing, we need to understand the evolution of $f_s$ in phase space, $(\boldsymbol{x},\boldsymbol{v};t)$. The DF evolves according to an equation typically known as the \emph{Boltzmann equation},
\begin{equation}
\frac{\partial f_s}{\partial t}+\boldsymbol{v}\cdot\frac{\partial f_s}{\partial \boldsymbol{x}}+\frac{q_s}{m_s}\left(\boldsymbol{E}+\boldsymbol{v}\times\boldsymbol{B}\right)\cdot\frac{\partial f_s}{\partial \boldsymbol{v}}=\frac{\partial f_s}{\partial t}\bigg\rvert_c,\label{eq:Boltz}
\end{equation}
with the right-hand side (RHS) of the equation describing the evolution of the DF according to `collisions' (e.g. binary Coulomb collsions, see \cite{Fitzpatrickbook}). Properly, this equation is specifically named after the form of collision operator assumed, e.g. Boltzmann, Fokker-Planck or Lenard-Balescu \citep{Schindlerbook}. If the collision operator chosen is a function of $f_s$ alone, then the Boltzmann equation and Maxwell's equations form a closed set, and the plasma is said to be in a \emph{kinetic regime} \citet{Schindlerbook}. In its general form, the Boltzmann equation can be obtained by integrating the Liouville equation for the N-particle DF in $6N$ dimensional phase-space,
\begin{equation}
\frac{d F_s(\boldsymbol{x}_1,...,\boldsymbol{x}_N,\boldsymbol{v}_1,...,\boldsymbol{v}_N;t)}{dt}=0,\nonumber
\end{equation}
over the positions and velocities of all but one particle \citep{Krallbook-1973} (made possible by the fact that particles of a particular species are identical \citep{Tong}). This also involves some assumptions made about the weak nature of the particle coupling in the plasma, characterised by
\begin{equation}
g=\frac{4\pi}{ 3 \Lambda_p}=\frac{1}{n_e\lambda_D^3}\ll 1,\nonumber
\end{equation}
for the small parameter $g$, i.e. a \emph{weakly coupled plasma} \citep{Schindlerbook,Krallbook-1973}. Here, $\Lambda_p$ is the \emph{plasma parameter}, equal to the number of electrons in the \emph{Debye sphere}, a sphere of radius $\lambda_D$ beyond which charge density inhomogeneities are shielded \citep{Krallbook-1973,Fitzpatrickbook}. The small parameter $g$ is used as the ordering parameter in an infinite hierarchy of statistical equations - the so called \emph{BBGKY hierarchy} - for which closure is achieved by neglecting terms of the desired order in $g^s$ \citep{Krallbook-1973}. The standard collisional framework is achieved by neglecting terms of order $g^2$ and above.

\subsection{Quasineutrality}\label{sec:quasi}
It is a feature common to many weakly coupled plasmas that typical spatial variations, $L$, are much larger than a quantity known as the \emph{Debye radius}, $\lambda_D$,
\begin{equation}
\epsilon=\frac{\lambda_D}{L}\ll 1,\hspace{3mm}\text{s.t.}\hspace{3mm}\lambda_D=\sqrt{\frac{\epsilon_0k_{B}T_e}{ne^2}},\nonumber
\end{equation}
for $k_{B}$ Boltzmann's constant, $T_e$ the electron temperature, and $e$ the fundamental charge. In such a situation the plasma is considered to be \emph{quasineutral} \citep{Schindlerbook}, typically taken to mean that
\begin{equation}
n_i=n_e\iff\sigma=0.
\end{equation}
Note that this is in an asymptotic sense, and formally does not imply that $\nabla\cdot\boldsymbol{E}$ vanishes, see e.g. \citet{Freidbergbook, Schindlerbook, Harrison-2009POP}. To see how this works, first notice that if one normalises Poisson's equation by
\[
\phi=\phi_{0}\tilde{\phi},\hspace{3mm}\nabla=\frac{1}{L}\tilde{\nabla},\hspace{3mm}\sigma=en_0\tilde{\sigma},
\]
for characteristic values $\phi_0$, $L$ and $n_0$ of the scalar potential, length scales and number densities, then one obtains
\[
\epsilon^2\tilde{\nabla}^2\tilde{\phi}=-\tilde{\sigma},
\]
for $\phi_0=k_BT_0/e$, and $\epsilon=\lambda_D/L$. In the quasineutral limit the $\epsilon^2$ parameter is vanishingly small. If one then makes an expansion of small parameters
\[
\tilde{\phi}=\sum_{n=0}^\infty \epsilon^{2n}\tilde{\phi}_n,\hspace{3mm}\tilde{\sigma}=\sum_{n=0}^\infty \epsilon^{2n}\tilde{\sigma}_n,
\]
then one sees that formally, for $\lambda_D/L\ll 1$,
\begin{eqnarray}
\tilde{\sigma}_0&=&0,\nonumber\\
\tilde{\nabla}^2\tilde{\phi}_0&=&-\tilde{\sigma}_1.\nonumber\\
&\vdots&\nonumber
\end{eqnarray}
As such, letting $\sigma=0$ is an approximation to the quasineutral limit, valid to \emph{first order}.

It should also be mentioned that quasineutrality implies that the characteristic frequencies are much less than the (electron) plasma frequency, 
\begin{equation}
\omega_p=\sqrt{\frac{n_ee^2}{\epsilon_0 m_e}}.
\end{equation}
Quoting \citet{Freidbergbook} directly: \emph{``For any low-frequency macroscopic charge separation that tends to develop, the electrons have more than an adequate time to respond, thus creating an electric field which maintains the plasma in local quasineutrality''}. The assumption of quasineutrality is consistent with neglecting the displacement current in Maxwell's equations \citep{Schindlerbook}. These ordering assumptions give the quasineutral `low-frequency/pre-Maxwell' equations that are commonly used in plasma physics
\begin{eqnarray}
\nabla\times\boldsymbol{B}=\mu_{0}\boldsymbol{j},&&\hspace{3mm}\text{Amp\`{e}re's Law}\nonumber\\
\nabla\times\boldsymbol{E}=-\frac{\partial\boldsymbol{B}}{\partial t},&&\hspace{3mm}\text{Faraday's Law}\nonumber\\
\nabla\cdot\boldsymbol{B}=0 &&\hspace{3mm}\text{Solenoidal constraint},\nonumber
\end{eqnarray}
and
\begin{equation}
\Bigg(\nabla\cdot\boldsymbol{E}=\frac{\sigma}{\epsilon_{0}},\hspace{3mm}\text{Gau{\ss}' Law, s.t.}\hspace{3mm}\frac{\epsilon_0\nabla\cdot\boldsymbol{E}}{\sigma}\ll 1\Bigg).\nonumber
\end{equation}
In practice, Gau{\ss}' Law is often not considered as a `core equation' in plasma physics, and is implicitly `replaced' by $\sigma=0$. Faraday's law is also often `reformulated' by eliminating the electric field using some version of Ohm's law (e.g. see \cite{Schindlerbook, KulsrudMHD, Freidbergbook, Krallbook-1973, Fitzpatrickbook}).

\subsection{Fluid Models}
Fluid models are the next step in the hierarchy after kinetic models, and are characterised by variables that depend only on space and time. Hence, the fluid equations are calculated by integrating over velocity space: taking velocity space \emph{moments} of the kinetic equation at hand \citep{Schindlerbook}. This process was laid down in the seminal work of \citet{Braginskii-1965}, giving the collisional transport (or Braginskii) equations
\begin{eqnarray}
\frac{\partial \rho_e}{\partial t}+\rho_e\nabla\cdot\boldsymbol{V}_e=0,\hspace{3mm}&&\text{\emph{Electron mass transport}}\nonumber\\
\rho_e\frac{d\boldsymbol{V}_e}{dt}+\nabla p_e+\nabla\cdot\boldsymbol{\pi}_e-\sigma_e(\boldsymbol{E}+\boldsymbol{V}_e\times\boldsymbol{B})=\boldsymbol{F}_{\text{fr},e},\hspace{3mm}&&\text{\emph{Electron mom. transport}}\nonumber\\
\frac{3}{2}\frac{dp_e}{dt}+\frac{5}{2}p_e\nabla\cdot\boldsymbol{V_e}+\boldsymbol{\pi}_e:\nabla\boldsymbol{V}_e+\nabla\cdot\boldsymbol{q}_e=W_e,\hspace{3mm}&&\text{\emph{Electron energy transport}}\nonumber
\end{eqnarray}
for electrons, and
\begin{eqnarray}
\frac{\partial \rho_i}{\partial t}+\rho_i\nabla\cdot\boldsymbol{V}_i=0,\hspace{3mm}&&\text{\emph{Ion mass transport}}\nonumber\\
\rho_i\frac{d\boldsymbol{V}_i}{dt}+\nabla p_i+\nabla\cdot\boldsymbol{\pi}_i-\sigma_i(\boldsymbol{E}+\boldsymbol{V}_i\times\boldsymbol{B})=-\boldsymbol{F}_{\text{fr},i},\hspace{3mm}&&\text{\emph{Ion mom. transport}}\nonumber\\
\frac{3}{2}\frac{dp_i}{dt}+\frac{5}{2}p_i\nabla\cdot\boldsymbol{V}_i+\boldsymbol{\pi}_i:\nabla\boldsymbol{V}_i+\nabla\cdot\boldsymbol{q}_i=W_e,\hspace{3mm}&&\text{\emph{Ion energy transport}}\nonumber
\end{eqnarray}
for ions, using the notation from \citet{Fitzpatrickbook}. In these equations $\rho_s=m_sn_s$ defines the mass density, $p_s=\frac{1}{3}\text{Tr}(\boldsymbol{P}_s)$ the scalar pressure for species $s$, defined by the trace of the pressure tensor of species $s$
\begin{equation}
P_{ij,s}=m_s\sum_s\int f_sw_{is}w_{js}d^3v\hspace{3mm}\text{s.t.}\hspace{3mm}P_{ij}=\sum_sP_{ij,s},\nonumber
\end{equation}
for $\boldsymbol{w}_s=\boldsymbol{v}-\boldsymbol{V}_s$ the velocity of a particle relative to the bulk flow, and for which
\begin{equation}
\boldsymbol{\pi}_s=\boldsymbol{P}_s-p_s\boldsymbol{I},\nonumber
\end{equation}
is the stress/generalised viscosity tensor. The vector $\boldsymbol{q}_s$,
\begin{equation}
\boldsymbol{q}_s=\frac{m_s}{2}\int w_s^2\boldsymbol{w}_sf_sd^3v,\nonumber
\end{equation}
is the heat flux density. Finally, $\boldsymbol{F}_{\text{fr},s}$ and $W_s$ are found by taking the momentum- and energy- moments of the collision operator (the RHS of the Boltzmann equation), and represent the collisional friction force, and collisional energy change, respectively. 

These are the \emph{two-fluid} equations. They describe the spatio-temporal evolution of the moments of the ion and electron DFs resepctively, and these are coupled by the EM fields. In their current form they are not closed: there are more unknowns than equations \citep{Freidbergbook}. It is not the purpose of this introduction to explore the subtle details of fluid closure, two-fluid, single fluid and magnetohydrodynamic (MHD) theories. For details on these topics see \citet{Schindlerbook, KulsrudMHD, Freidbergbook, Krallbook-1973, Fitzpatrickbook}.

\section{Collisions in plasmas}
\subsection{Collisional plasmas}\label{sec:collisions} 
The collisionality of a plasma species is characterised in time and space by two quantities \citep{Fitzpatrickbook}: the collision rate/frequency, $\nu_s$; and the mean free path $\lambda_{\text{mfp},s}$, such that
\begin{eqnarray}
\nu_s&\approx&\sum_{s^\prime}\nu_{ss^\prime},\nonumber\\
\lambda_{\text{mfp},s}&=&v_{\text{th},s}/\nu_s,\nonumber\\
T_i=T_e&\implies&\nu_e\sim \sqrt{\frac{m_i}{m_e}}\nu_i.\nonumber
\end{eqnarray}
That is to say that the total collision rate for a species is made up of the collision rates with all species (including its own), the mean free path measures the typical distance a particle travels between collisions, and that in the case of an isothermal plasma the collision rate for electrons is much greater than that for ions. The thermal velocity, $v_{\text{th},s}$, gives the energy of random particle motion $E_{\text{random}}=m_sv_{\text{th},s}^2$, such that in thermal equilibrium $k_BT_s=E_{\text{random}}$ \citep{Schindlerbook}. We note here that a collision is classified as a $\ge 90^\circ$ scattering event, and as such a particle may have numerous `small-angle' scattering (i.e. $<90^\circ$) events before a successful `collision' \citep{Fitzpatrickbook}. 

A collision dominated plasma is one for which the mean free path is much smaller than typical plasma length scales, $L$
\begin{equation}
\lambda_{\text{mfp}}\ll L\nonumber ,
\end{equation}
with the opposite limit indicating a collisionless plasma. The collisional frequency typically has magnitude
\begin{equation}
\nu_e\sim\frac{\ln\Lambda_p}{\Lambda_p}\omega_p,\nonumber
\end{equation}
\citep{Fitzpatrickbook} and as such 
\begin{equation}
\nu_e\ll\omega_p\iff\Lambda_p\gg 1\iff g\ll 1.\nonumber
\end{equation}
That is to say that weakly coupled plasmas are those for which collisions are not able to prevent plasma oscillations from regulating charge separation. In the case of a sufficiently collisional plasma characterised by 
\begin{eqnarray}
\frac{1}{\nu_s}\frac{\partial \langle v^kf_s\rangle}{\partial t}&&\ll  \langle v^kf_s\rangle,\nonumber\\
\lambda_{\text{mfp},s}\nabla \langle v^kf_s\rangle&&\ll \langle v^kf_s\rangle,\nonumber\\
\lambda_{\text{mfp},s}e|\boldsymbol{E}|&&\ll k_BT_s\nonumber
\end{eqnarray}
for which $\langle v^kf_s\rangle$ is a $k^{\text{th}}-$order velocity moment of the DF, then the plasma is in a \emph{local thermal equilibrium} (e.g. see \cite{Cowleynotes}), characterised by a temperature $T_s(\boldsymbol{x},t)$, and the DF can be written as a Maxwellian of the form 
\begin{equation}
f_{s}(\boldsymbol{x},\boldsymbol{v};t)=\frac{n_s(\boldsymbol{x};t)}{(2\pi k_BT_s(\boldsymbol{x};t)/m_s)^{3/2}}e^{-m_s(\boldsymbol{v}-\boldsymbol{V}_s(\boldsymbol{x};t))^2/(k_BT_s(\boldsymbol{x};t))},\label{eq:localtherm}
\end{equation}
to lowest order. This DF describes a plasma species with local number density $n_s(\boldsymbol{x},t)$ and local bulk velocity $\boldsymbol{V}_s(\boldsymbol{x},t)$. The DF in Equation (\ref{eq:localtherm}) is clearly not an equilibrium solution, since the number density, bulk flow and temperature explicitly depend on time. Given sufficient time, Boltzmann's \emph{H-Theorem} implies that collisions will always attempt to drive a system towards thermal equilibrium (e.g. see \cite{Grad-1949b, Brush-2003}), defined by a DF of the form
\begin{equation}
f_{s}(\boldsymbol{v})=\frac{n_{s}}{(2\pi k_BT_s/m_s)^{3/2}}e^{-m_s(\boldsymbol{v}-\boldsymbol{V}_s)^2/(k_BT_s)}.\label{eq:therm}
\end{equation}
The DF in Equation (\ref{eq:therm}) is of the same form as that in Equation (\ref{eq:localtherm}), but is now independent of space and time. The temperature is constant and a non-zero bulk flow is permitted.

\subsection{Collisionless plasmas}
The statement that collisionless plasmas are those for which $\lambda_{\text{mfp}}\gg L$ is rather truistic, and not particularly helpful in physical terms. Using the definition of the plasma parameter  \citep{Fitzpatrickbook},
\begin{equation}
\Lambda_p=\frac{4\pi}{n_e^{1/2}}\left(\frac{\sqrt{\epsilon_0T_e}}{e}\right)^3,\nonumber
\end{equation}
we see that the collision frequency behaves like
\begin{equation}
\nu_e\sim\frac{e^4n_e\ln\Lambda_p}{4\pi\epsilon_0^2m^{1/2}T_e^{3/2}}=\frac{e^4}{4\pi\epsilon_0^2m^{1/2}}\frac{n_e}{T_e^{3/2}}\ln\left(\frac{4\pi}{n_e^{1/2}}\left(\frac{\sqrt{\epsilon_0T_e}}{e}\right)^3\right).\nonumber
\end{equation}
Hence, dense and low temperature plasmas are more likely to be collisional, whereas diffuse and high temperature plasmas tend to be collisionless. In such situations, it is reasonable to neglect the RHS of the Boltzmann equation (Equation (\ref{eq:Boltz})), giving the \emph{Vlasov equation} \citep{Vlasov-1968},
\begin{equation}
\frac{\partial f_s}{\partial t}+\boldsymbol{v}\cdot\frac{\partial f_s}{\partial \boldsymbol{x}}+\frac{q_s}{m_s}\left(\boldsymbol{E}+\boldsymbol{v}\times\boldsymbol{B}\right)\cdot\frac{\partial f_s}{\partial \boldsymbol{v}}=0.\label{eq:Vlasov}
\end{equation}
In closed form this equation can be written, using Hamilton's equations \citep{Tong}, as 
\begin{eqnarray}
\frac{df_s}{dt}&=&\frac{\partial f_s}{\partial t}+\frac{\partial f_s}{\partial \boldsymbol{x}}\cdot\frac{ d\boldsymbol{x}}{dt}+\frac{\partial f_s}{\partial \boldsymbol{v}}\cdot\frac{ d\boldsymbol{v}}{dt}=0,\nonumber\\
&=&\frac{\partial f_s}{\partial t}+\frac{\partial f_s}{\partial \boldsymbol{x}}\cdot\frac{ \partial H_s}{\partial\boldsymbol{p}_s}-\frac{\partial f_s}{\partial \boldsymbol{p}_s}\cdot\frac{ \partial H_s}{\partial \boldsymbol{x}}=0,\nonumber\\
&=&\frac{\partial f_s}{\partial t}+\left\{ f_s , H_s\right\}_{PB}=0.
\end{eqnarray}
Here, the Hamiltonian is given by $H_s$, the canonical momenta by $\boldsymbol{p}_s$, and the brackets $\{ \, , \,\}_{PB}$ are Poisson brackets, whose definition can be inferred from above. We can go from using velocity variables in the first line, to momentum variables in the second since $d\boldsymbol{p}_s=m_sd\boldsymbol{v}$. The Vlasov equation essentially states that the DF is conserved along a particle trajectory in phase-space \citep{Schindlerbook}, since the characteristics of the Vlasov equation are the single particle equations of motion,
\begin{eqnarray*}
\frac{d}{dt}\boldsymbol{x}(t)&=&\boldsymbol{v}(t),\\
\frac{d}{dt}\boldsymbol{v}(t)&=&\frac{q_s}{m_s}(\boldsymbol{E}+\boldsymbol{v}\times\boldsymbol{B}).
\end{eqnarray*}
The solutions of this equation are in principle completely reversible in time, and hence entropy conserving \citep{Krallbook-1973}.

\section{Collisionless plasma equilibria}\label{sec:collisionless}
A Vlasov equilibrium is obtained when the DF satisfies
\begin{equation}
\frac{\partial f_s}{\partial t}=0\implies\left\{ H_s , f_s\right\}_{PB}=0.
\end{equation}
This statement does not mean that there are no macroscopic particle flows or currents; density, pressure or temperature gradients; or even heat fluxes, for example. That is to say that the moments of the DF can still have gradients in space. Rather, it is an equilibrium in the sense of a particle distribution. This means that the value of the DF at each individual point in phase-space is independent of time.

It is a standard result in classical mechanics that constants of motion, $C_s(\boldsymbol{x}(t),\boldsymbol{p}(t))$, (that do not depend explicitly on time) are in `involution' with/commute with the Hamiltonian \citep{TongMech},
\begin{equation}
\left\{ H_s , C_s\right\}_{PB}=0.
\end{equation}
Using this result, and the linearity of the Poisson bracket, we see that any function of the constants of motion is a Vlasov equilibrium DF, since
\begin{equation}
\left\{ H_s , f_s(C_{1s},...,C_{ns})\right\}_{PB}=\sum_{j=1}^n \frac{\partial f_s}{\partial C_{js}}\left\{ H_s , C_{js}\right\}_{PB}=\sum_{j=1}^n \frac{\partial f_s}{\partial C_{js}}\, \times \, 0=0.
\end{equation}

 We can also show that the reverse is true, namely that any Vlasov equilibrium DF is a function of the constants of motion. First consider a Vlasov equilibrium DF $f_s(G_1,G_2,...,G_n)$ for arbitrary linearly independent functions $G_j(\boldsymbol{x}(t),\boldsymbol{p}(t))$. Then by linearity of the Poisson Bracket,
\begin{equation}
\left\{ H_s , f_s\right\}_{PB}=\sum_{j=1}^n \frac{\partial f_s}{\partial G_j}\left\{ H_s , G_j\right\}_{PB}.
\end{equation}
This sum must be zero for an equilibrium, and since the $G_j$ are linearly independent, that implies that each of the Poisson brackets must be zero independently. Hence the $G_j$ must be constants of motion and so
\[
\text{\emph{``}$f_s$ \emph{is a Vlasov equilibrium DF}}\iff\text{$f_s$ \emph{is a function of the constants of motion''}}.
\]

It is clear that a Vlasov equilibrium DF also satisfies the time-dependent Vlasov equation itself \citet{Schindlerbook}, since
\begin{equation}
\frac{df_s}{dt}=\frac{\partial f_s}{\partial t}+\left\{ f_s , H_s\right\}_{PB}=0+0.
\end{equation}
Using this fact, one can construct time-dependent solutions for `nonlinear' propagating structures to the Vlasov equation by using a frame transformation \citep{Schamel-1979}. Then one can solve for Vlasov equilibria in the wave frame, e.g. the famous BGK modes \citep{Bernstein-1957} and Schamel's theory \citep{Schamel-1986}, amongst other examples, e.g. see \citet{Abraham-Shrauner-1968, Ng-2005,Vasko-2016,Hutchinson-2017}. 

\subsection{The `forward' and `inverse' approaches}\label{sec:for_inv}
As described above, one can easily construct equilibrium solutions of the Vlasov equation provided that at least one constant of motion has been identified. Any differentiable function of the constants of motion is an equilibrium solution of the Vlasov equation \citep{Schindlerbook}, and is physically meaningful provided all velocity moments exist,
\begin{equation}
\bigg| \int v_1^iv_2^jv_3^k\, f_s\, dv_{1}\, dv_{2}\, dv_{3}\bigg| \,<\,\infty\, \forall \, i, \, j,\, k \in 0,1,2,... \, ,\nonumber
\end{equation}
and the function is non-negative over all phase-space,
\[
f_s(\boldsymbol{x},\boldsymbol{v})\ge 0\, \forall\, \boldsymbol{x},\, \boldsymbol{v}.\nonumber
\]
Whilst such a function may well satisfy these mathematical/microscopic conditions, the next question to ask is of the macroscopic electromagnetic fields that are consistent with such a function. Through Equations (\ref{eq:fmomentzero}) and (\ref{eq:fmomentone}), we see that the distribution of particles in phase-space determines the charge and current densities respectively,  in configuration-space. These charge and current densities are consistent with certain electric and magnetic fields through Maxwell's equations (Equations (\ref{eq:Gaussequation}) - (\ref{eq:Faradaylaw})). Hence, a full understanding of the macroscopic and microscopic physics of a plasma necessitates a self-consistent `solution' of the Vlasov-Maxwell (VM) system.

From these considerations, it should be clear that there are two possible routes to follow, in the absence of a comprehensive self-consistent theory, namely
\begin{itemize}
\item `Inverse': Given some or all of the macroscopic fields $(\phi,\boldsymbol{A})$, can we find a self-consistent DF, $f_s$? (e.g. see discussions in  \cite{Alpers-1969, Channell-1976, Mynick-1979a,Greene-1993, Harrison-2009POP, Belmont-2012, Allanson-2016JPP})
\item `Forward': Given a DF, $f_s$, can we find some set of self-consistent macroscopic fields, $(\phi,\boldsymbol{A})$? (e.g. see discussions in \cite{Grad-1961, Harris-1962, Sestero-1964, Sestero-1965, Lee-1979JGR, Schindlerbook, Kocharovsky-2010, Vasko-2013})
\end{itemize}
The forward approach is the one that is most frequently seen in the literature. This is partly due, mathematically, to the fact that this involves solving differential equations, as opposed to the often less tractable inversion of integral equations in the case of the inverse approach. But also, as argued in Section \ref{sec:collisions}, it is reasonable on physical grounds to assume that - for sufficiently collisional \citep{Cowleynotes} and `not-too-turbulent' plasmas \citep{Alpers-1969} - that the DF is (locally) Maxwellian, and then to proceed with the forwards approach from thereon. 

In the case of collisionless plasmas, there are an infinite class of equilibrium solutions in principle, and hence the forwards approach would have to be predicated on some prior knowledge of the DF. In-situ observations of DFs have only recently become available with spatio-temporal resolution on kinetic scales, for example the NASA Multiscale Magnetospheric (MMS) mission \citep{Hesse-2016}, and the ESA candidate mission: Turbulent Heating ObserveR (THOR) \citep{Vaivads-2016}. 

Due to the ubiquitous nature and reasonable validity of the MHD approach in many environments, and the relative wealth and long history of magnetic field measurements, the equilibrium structures and dynamics of electromagnetic fields are better understood and more often used as the fundamental basis, or object, of plasma physics discussions and theory. Hence, it is of use, and necessity, to consider the inverse approach. 

\begin{figure}
    \centering
        \includegraphics[width=0.7\textwidth]{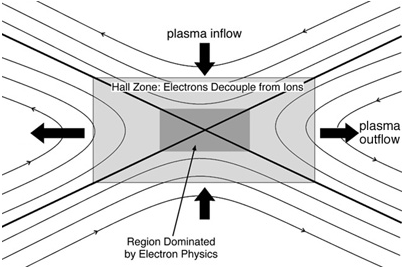}
    \caption{{\small A diagrammatic representation of the local structure of a magnetic reconnection event, and the `electron diffusion region', in which the electrons decouple from the magnetic field. {\bf Image copyright:}  \href{http://mms.space.swri.edu}{NASA MMS-SMART Investigation}, (reproduced with permission).}}\label{fig:edr}
\end{figure}

\subsection{Motivating translationally invariant Vlasov-Maxwell (VM) equilibria}
\subsubsection{Current sheets}\label{sec:currentsheets}\label{sec:HarrisDF}
In a planar geometry, localised electric currents in a plasma are known as current sheets: frequently considered to be the initial state of wave processes \citep{Fruit-2002}, instabilities \citep{Schindlerbook}, reconnection \citep{Yamada-2010} and various dynamical phenomena in laboratory \citep{Beidler-2011}, space \citep{Zelenyi-2011} and astrophysical \citep{DeVore-2015} plasmas. The formation of current sheets is ubiquitous in plasmas. They can form between plasmas of different origins that encounter each other, such as at Earth's magnetopause between the magnetosheath plasmas and magnetospheric plasmas (e.g. see \cite{Dungey-1961, Phan-1996}); or they can develop spontaneously in magnetic fields that are subjected to random external driving (e.g. see \cite{Parker-1994}), such as in the solar corona. 

As to be introduced in Section \ref{sec:reconnection}, localised electric currents are an important ingredient for magnetic reconnection: acting as a signature of sheared magnetic fields, and reconnection electric fields (e.g. see \cite{Biskamp-2000,Hesse-2011}). As per Poynting's theorem \citep{Poynting-1884}, with $\boldsymbol{S}=\mu_0^{-1}\boldsymbol{E}\times\boldsymbol{B}$, and neglecting electric field energy,
\[
\frac{\partial B^2/((2\mu_0)}{\partial t}=-\nabla\cdot\boldsymbol{S}-\boldsymbol{j}\cdot\boldsymbol{E},
\]
intense current sheets are ideal locations for magnetic energy conversion and dissipation \citep{Birn-2010, Zenitani-2011}. The dominant mechanisms that release the free energy include magnetic reconnection, and various plasma instabilities.

The currents themselves are usually considered synonymous with a stressed and/or anti-parallel magnetic field configuration, since in a quasineutral plasma (or a plasma in equilibrium), the current density is given by
\[
\boldsymbol{j}=\frac{1}{\mu_0}\nabla\times\boldsymbol{B}.
\]
Perhaps the most used current sheet equilibrium model is represented in Figure \ref{fig:Harris}: the Harris sheet \citep{Harris-1962},
\begin{eqnarray}
\boldsymbol{B}&=&B_0\left(\tanh\left(\frac{z}{L}\right),0,0\right),\label{eq:harrissheet}\\
\frac{1}{\mu_0}\nabla\times\boldsymbol{B}=\boldsymbol{j}&=&\frac{B_0}{\mu_0L}\left(0,\text{sech}^2\left(\frac{z}{L}\right),0\right),\nonumber\\
\frac{dp}{dz}=-j_yB_x\implies p&=&\frac{B_0^2}{2\mu_0}\text{sech}^2(z/L),\nonumber
\end{eqnarray}
with $L$ the current sheet `width', normalising $z$; $B_0$ the asymptotic values of the magnetic field, normalising $B_x$;  $j_{y0}=B_0/(\mu_0L)$ and $p_0=B_0^2/(2\mu_0)$ normalising the current density and scalar pressure respectively. The maximum shear of $B_x$ is localised in the region $-L<z<L$, and this is where we see the maximum values of the current density: the current sheet itself. A Vlasov equilibrium DF self-consistent with the Harris sheet is given by
\begin{equation}
f_s=\frac{n_{0s}}{(\sqrt{2\pi}v_{\text{th},s})^3}e^{-\beta_s(H_s-u_{ys}p_{ys})}, \label{eq:HarrisDF}
\end{equation}
with $\beta_s=1/(m_sv_{\text{th},s}^2)$; $n_{0s}$ a constant with dimensions of spatial number density (and not necessarily representing the number density itself); and with $u_{ys}$ a bulk flow parameter, that in this case coincides with the bulk flow itself, i.e. $u_{ys}=V_{ys}$. Note that one can derive other equilibrium DFs for the Harris sheet, e.g. the Kappa ($\kappa$) DF \citep{Fu-2005}.

\subsubsection{Harris-type distribution functions (DFs)}\label{sec:Harristype}
If we were to `generalise' the DF in Equation (\ref{eq:HarrisDF}) to one that supports two current density components (and hence a DF self-consistent with a different magnetic field), then we have
\[
f_s=\frac{n_{0s}}{(\sqrt{2\pi}v_{\text{th},s})^3}e^{-\beta_s(H_s-u_{xs}p_{xs}-u_{ys}p_{ys})}.
\]
One particularly nice feature of a DF that is a function of $(H_s-u_{xs}p_{xs}-u_{ys}p_{ys})$,
\[
f_s=f_s(H_s-u_{xs}p_{xs}-u_{ys}p_{ys})
\]
is that the bulk flows are directly related to the flow parameters, i.e. $V_{xs}=u_{xs}$ and $V_{ys}=u_{ys}$. This is seen by the following argument. If we define $\mathcal{H}_s=H_s-\boldsymbol{u}_s\cdot\boldsymbol{p}_s$ for
\[
\boldsymbol{u}_s=(u_{xs},u_{ys},0),\hspace{3mm}\boldsymbol{p}_s=(p_{xs},p_{ys},0),
\]
then $f_s=f_s(\mathcal{H}_s)$ and 
\[
\mathcal{H}_s=\frac{m_s}{2}\boldsymbol{\mathcal{U}}_s^2-\frac{m_s}{2}\boldsymbol{u}_{s}^2-q_s(A_x+A_y)\hspace{3mm}\text{s.t.}\hspace{3mm}\boldsymbol{\mathcal{U}}_s=\boldsymbol{v}-\boldsymbol{u}_s.
\]
If we now consider the first-order moment of $f_s$ by $\boldsymbol{\mathcal{U}}_s$, the result must be zero since $f_s$ only depends on $\boldsymbol{\mathcal{U}}_s^2$, through $\mathcal{H}_s$. Consequently
\[
\int \boldsymbol{\mathcal{U}}_s f_s(\mathcal{H}_s)d^3\mathcal{U}_{s}=0=\underbrace{\int \boldsymbol{v} f_s d^3v }_{n_s\boldsymbol{V}_s}-\underbrace{\boldsymbol{u}_s\int f_s d^3v}_{n_s\boldsymbol{u}_s},
\]
and hence $\boldsymbol{V}_s=\boldsymbol{u}_s=(u_{xs},u_{ys},0)$.

\begin{figure}
    \centering
   
    \begin{subfigure}[b]{0.45\textwidth}
        \includegraphics[width=\textwidth]{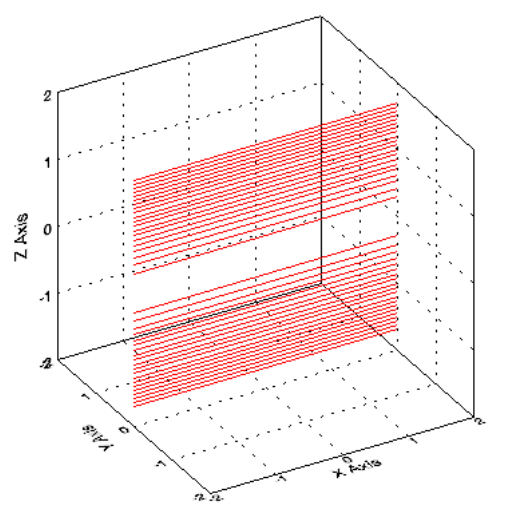}
        \caption{\small The Harris sheet magnetic field}
        \label{fig:Harris1}
    \end{subfigure}
    \begin{subfigure}[b]{0.35\textwidth}
        \includegraphics[width=\textwidth]{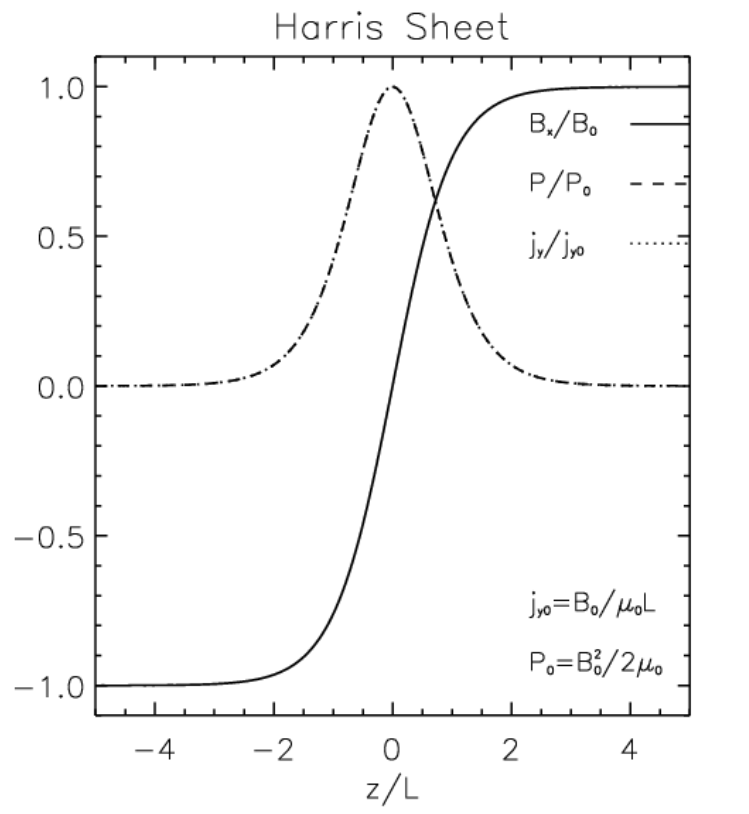}
        \caption{\small The Harris sheet equilibrium magnetic field, current density and scalar pressure. }
        \label{fig:Harris2}
    \end{subfigure}
      
    \caption{ {\small Figure \ref{fig:Harris1} represents the magnetic field lines for the Harris sheet magnetic field. Figure \ref{fig:Harris2} shows the normalised $B_x, j_y$, and scalar pressure $p$ for the Harris sheet equilibrium characterised by $j_y=dB_x/dz$, and $dp/dz=-j_yB_x$.  {\bf Image's copyright:} M.G. Harrison's PhD thesis \citep{Harrison-thesis}, (reproduced with permission).}  }\label{fig:Harris}
\end{figure}

\subsubsection{Other applications}
Current sheets are by no means the only application of the work on translationally invariant VM equilibria in this thesis. As indicated in Section \ref{sec:lit}, translationally invariant VM equilibria are of use for numerous other applications in plasma physics. Examples include nonlinear waves (e.g. see \cite{Bernstein-1957, Ng-2012}); electron holes, ion holes and double layers (e.g. see \cite{Schamel-1986}); and colllisionless shock fronts (e.g. see \cite{Montgomery-1969, Burgessbook}).

\subsection{Magnetic Reconnection}\label{sec:reconnection}
Magnetic reconnection is a ubiquitous phenomenon in solar, space, astrophysical and laboratory plasmas, and now considered to be \emph{``among the most fundamental unifying concepts in astrophysics, comparable in scope and importance to the role of natural selection in biology.''} \citep{Moore-2015}: see authoritative discussions of `classical' reconnection in \citet{Schindlerbook, Priest-2000, Biskamp-2000, Hesse-2011}; on modern theories of `fast' reconnection and `turbulent/stochastic reconnection' in \citet{Lazarianbook, Loureiro-2016}; and `fractal reconnection' in \citet{Shibata-2001}. The literature on the topic is vast and there are many complex concepts to consider regarding the precise mathematical definition (e.g. see \cite{Hesse-1988,Priest-2014}) of reconnection and its physical behaviour in different dimensions and plasma environments. The phenomenon also appears in physical environments as numerous as the number of plasma environments themselves, e.g. solar corona, planetary and pulsar magnetospheres, magnetic dynamos, gamma-ray bursts, geomagnetic storms and sawtooth crashes in tokamaks. However, there are common features that are agreed upon:
\begin{description}
\item[Topology:] There is a change in the topology of the magnetic field, caused by processes in non-ideal ($\boldsymbol{E}+\boldsymbol{V}\times\boldsymbol{B}\ne\boldsymbol{0}$) regions of plasma with strong localised electric currents and parallel electric fields. \item[Diffusion region:] This region is termed the diffusion region (e.g. see \cite{Hesse-2001, Schindlerbook, Hesse-2011}), and is represented locally, and in an idealised geometry in Figure \ref{fig:edr}. 
\item[Decoupling:] Ideal MHD breaks down within the diffusion region, kinetic physics is dominant, and the plasma decouples from the magnetic field, enabling stored magnetic energy to be released to the physical medium.
\end{description}
Hence, magnetic reconnection explicitly couples (via the transmission of energy) the macroscopic ideal MHD picture of relatively slow-evolving and large scale neutral, conducting fluids to the small-scale, short-timescale and non-neutral kinetic plasma physics. Reconnection can of course occur in many different ways. It could occur in one of following ways
\begin{description}
\item[Incidental:] One physical phenomenon out of many (and not necessarily dominant), occurring in a dynamical plasma, e.g. small scale reconnection in a turbulent plasma (e.g. \cite{Lazarian-1999});
\item[Steady-state:] A continuous reconnection phenomenon that generates kinetic energy with no significant macroscopic structural changes, e.g. the Sweet-Parker \citep{Parker-1957, Sweet-1958} and Petschek models \citep{Petschek-1964};
\item[Instability:] The result of an instability, i.e. the system was perturbed from equilibrium, reconnection was initiated, and the system does not return to the initial equilibrium, e.g. the tearing mode instability (e.g. see \cite{Furth-1963, Drake-1977}). 
\end{description}

\subsubsection{Approximate equilibria in particle-in-cell (PIC) simulations}\label{sec:flowshift}
Magnetic reconnection processes can critically depend on a variety of length and time scales, for example on lengths of the order of the Larmor orbits and below that of the mean free path (e.g. see \cite{Biskamp-2000, Birn-2007}). In such situations a collisionless kinetic theory could be necessary to capture all of the relevant physics, and as such an understanding of the differences between using MHD, two-fluid, hybrid, Vlasov and other approaches is of paramount importance, for example see \citet{Birn-2001, Birn-2005} for discussions of this problem in the context of one-dimensional (1D) current sheets: the `Geospace Environmnetal Modelling (GEM)' and `Newton' challenges. 

In the absence of an exact collisionless kinetic equilibrium solution, one has to use non-equilibrium DFs to start kinetic simulations, without knowing how far from the true equilibrium DF they are. In such cases, non-equilibrium drifting Maxwellian distributions are frequently used (see \cite{Swisdak-2003, Hesse-2005, Pritchett-2008, Malakit-2010, Aunai-2013, Hesse-2013, Guo-2014, Hesse-2014, Liu-2016} for examples), 
\begin{equation}
f_{\text{Maxw},s}=\frac{n_s(\boldsymbol{x})}{(\sqrt{2\pi}v_{\text{th},s})^3}\exp\left[\frac{\left(\boldsymbol{v}-\boldsymbol{V}_s(\boldsymbol{x})\right)^2}{2v_{\text{th},s}^2}\right], \label{eq:Mshift}
\end{equation}   
with $v_{\text{th},s}$ a characteristic value of the thermal velocity, $n_s(\boldsymbol{x})$ the number density, and $\boldsymbol{V}_{s}$ the bulk velocity of species $s$ . These DFs can reproduce the same moments $n_s, \boldsymbol{V}_s$ (and $p=n_sk_BT_s$, typically with $n_i=n_e$) necessary for a fluid equilibrium, maintained by the gradient of a scalar pressure, 
\[
\nabla p=\boldsymbol{j}\times\boldsymbol{B}.
\]
However, the DF, $f_{\text{Maxw,s}}$,  in Equation (\ref{eq:Mshift}) is not an exact solution of the Vlasov equation and hence does not describe a kinetic equilibrium. The macroscopic force balance self-consistent with a quasineutral Vlasov/kinetic equilibrium is maintained by the divergence of a rank-2 pressure tensor, $P_{ij}=P_{ij}(A_x(z),A_y(z))$ (e.g. see \cite{Channell-1976, Mynick-1979a, Schindlerbook}), according to 
\[
\nabla\cdot\boldsymbol{P}=\boldsymbol{j}\times\boldsymbol{B}.
\] 
As explained in \citet{Aunai-2013} on the subject of PIC simulations, the fluid equilibrium characterised by a drifting Maxwellian can evolve to a quasi-steady state \emph{``with an internal structure very different from the prescribed one''}, and as demonstrated in \citet{Pritchett-2008}, undesired electric fields, \emph{``coherent bulk oscillations''}, and other perturbations may form, in nature's attempt to maintain force-balance. Figure \ref{fig:Pritchett} is taken from \citet{Pritchett-2008}, and demonstrates this phenomenon. Each of the panels relates, in principle, to a 1D MHD equilibrium characterised by $dp/dx=j_yB_z$, in which the PIC simulation is intialised with a DF of the form of that in Equation (\ref{eq:Mshift}). Panel (a) demonstrates how the initial condition is self-consistent with a magnetic field profile and number density that are very close to those prescribed by the fluid equilibrium. However, panel (b) shows an electric field that forms due to the non-equilbrium initial state, and panel (c) demonstrates the resultant disparity between the exact/`fluid' current density (black), and that derived from the PIC simulation (red). 

\begin{figure}
    \centering
        \includegraphics[width=0.7\textwidth]{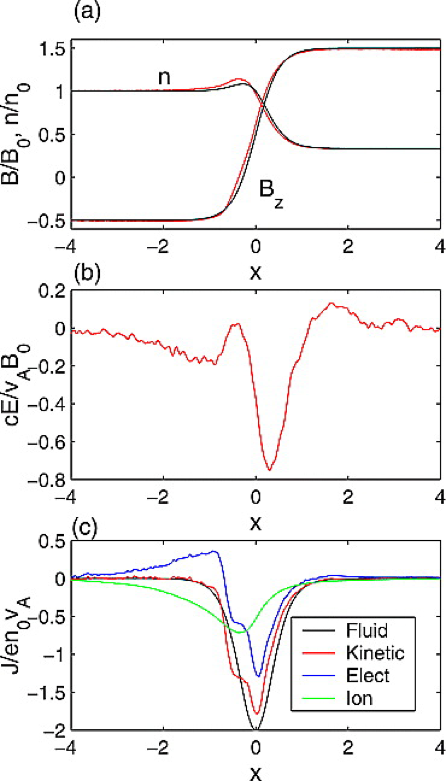}
    \caption{{\small A figure from \citet{Pritchett-2008}. Profiles in $x$ across a 1D current layer: (a) magnetic field $B_z(x)$ and density $n(x)$ from a PIC simulation at `time' 20 (red curves), and from the fluid equilibrium (black curves); (b) electric field $E_x(x)$ from a PIC simulation at `time' 20; (c) current density $J_y(x)$ determined from a PIC simulation at `time' 20 carried by the electrons (blue curve), ions (green curve), and the electrons and ions combined (red curve) and the fluid current density corresponding to the magnetic field (black curve). \\
    {\bf Image copyright:} \href{http://sites.agu.org}{American Geophysical Union} (reproduced with permission).}}    \label{fig:Pritchett}
\end{figure}

The knowledge of exact VM equilibria thus provides the chance to initialise PIC simulations in full confidence, with the intended macroscopic quantities reproduced. Exact VM equilibria would also permit analytical and numerical studies of the linear phase of collisionless instabilities \citep{Gary-2005}, such as the tearing mode (e.g. see \cite{Drake-1977,Quest-1981A}). This sort of exact analysis is formally out of reach without an exact initial condition since - as discussed by e.g. \citet{Pritchett-2008, Aunai-2013} - a non-exact Vlasov solution creates perturbations itself, by virtue of not being an equilibrium.

Of course, one could make an argument on the basis of ordering arguments that a non-exact equilibrium DF such as that in Equation (\ref{eq:Mshift}) allows the study of the nonlinear (and perhaps the linear) phase dynamics of plasma instabilities, such as the tearing mode. This sort of argument would be based on the assumption that a drifting Maxwellian such as that in Equation (\ref{eq:Mshift}) is \emph{sufficiently} close to a VM equilibrium so as not to significantly affect the physical processes. However, it is generally unclear how far such an initial condition is from exact equilibrium.

\subsection{Forward approach for one-dimensional (1D) VM equilibria}
To give context and to demonstrate the contrast, I will briefly introduce the `forward approach' in VM equilibria, as used and discussed in e.g. \citet{Grad-1961, Harris-1962, Sestero-1967, Lee-1979JGR, Schindlerbook}. In these - and other - works, a self-consistent solution to the VM system is found first by specifying the equilibrium DF as a function of the constants of motion. For example, a 1D system with $\partial/\partial x=\partial/\partial y=0$, has the Hamiltonian, and two canonical momenta as the constants of motion,
\begin{eqnarray}
H_s(\phi(\boldsymbol{x}), \boldsymbol{v})=H_s(z, \boldsymbol{v})&=&m_{s}\boldsymbol{v}^2/2+q_{s}\phi (z),\label{eq:hs}\\
p_{xs}(A_x(\boldsymbol{x}), \boldsymbol{v})=p_{xs}(z,v_x)&=&m_{s}v_{x}+q_{s}A_x(z),\label{eq:pxs}\\
p_{ys}(A_y(\boldsymbol{x}), \boldsymbol{v})=p_{ys}(z,x_y)&=&m_{s}v_{y}+q_{s}A_y(z),\label{eq:pys}
\end{eqnarray}
These quantities are constants of motion in the sense that for an individual particle trajectory (the characteristics of the Vlasov equation) parameterised by $t$,
\begin{eqnarray}
\frac{d}{dt}H_s (z(t),\boldsymbol{v}(t))=\frac{d}{dt}p_{xs}(z(t),\boldsymbol{v}(t))=\frac{d}{dt}p_{ys}(z(t),\boldsymbol{v}(t))=0,\nonumber
\end{eqnarray}
where the $d/dt$ is in fact an operator involving derivatives over phase-space,
\[
\frac{d}{dt}=\cancelto{0}{\frac{\partial}{\partial t}}+\frac{dz}{d t}\frac{\partial}{\partial z}+\frac{d\boldsymbol{v}}{d t}\cdot\frac{\partial}{\partial \boldsymbol{v}}.
\]

Using these relationships, it is now clear how one can justify writing the equilibrium DF as a function of the constants of motion
\[
f_s(\boldsymbol{x},\boldsymbol{v})=f_s(z,\boldsymbol{v})=f_s(H_s(z,\boldsymbol{v}),p_{xs}(z,\boldsymbol{v}),p_{ys}(z,\boldsymbol{v})),
\]
and a solution of Vlasov's equation. Note how the second equality above demonstrates that the non-uniqueness of the correspondences,
\begin{eqnarray}
z&=&z(H_s,p_{xs},p_{ys}),\nonumber\\
\boldsymbol{v}&=&\boldsymbol{v}(H_s,p_{xs},p_{ys}),
\end{eqnarray}
could play a role in this problem, see e.g. \citet{Grad-1961, Belmont-2012} for discussions of this problem.

In order to now satisfy the equilibrium Maxwell equations, scalar and vector potentials must be found that satisfy the following,
\begin{eqnarray}
-\epsilon_0 \frac{d^2}{dz^2}\phi (z)=\sigma (\phi(z),A_x(z),A_y(z))&=&\sum_s \int f_s (H_s,p_{xs},p_{ys}) \,d^3v,\nonumber\\
-\frac{1}{\mu_0}\frac{d^2}{dz^2}A_x(z)=j_x (\phi(z),A_x(z),A_y(z))&=&\sum_s \int v_x \,f_s (H_s,p_{xs},p_{ys}) \,d^3v,\nonumber\\
-\frac{1}{\mu_0}\frac{d^2}{dz^2}A_y(z)=j_y (\phi(z),A_x(z),A_y(z))&=&\sum_s \int v_y \,f_s (H_s,p_{xs},p_{ys}) \,d^3v.\nonumber
\end{eqnarray}
Since the RHS of the above equations are in principle now known functions of $(\phi,A_x,A_y)$, the problem of finding a VM equilibrium has been reduced to solving 3 coupled (ordinary) differential equations, subject to boundary conditions, e.g. the asymptotic values of the potentials at $z=\pm \infty$.

\subsubsection{A route through the forward problem} \label{sec:forwardexample}
To demonstrate how the forward problem works, we give an example for a form of DF that could be used,
\[
f_s=\frac{n_{0s}}{(\sqrt{2\pi}v_{\text{th},s})^3}e^{-\beta_sH_s}(a_se^{\beta_su_{xs}p_{xs}}+b_se^{\beta_su_{ys}p_{ys}})
\]
for the constants $a_s$ and $b_s$. This form is chosen as it is directly relatable to those considered in e.g. \citet{Harris-1962, Schindlerbook}, and has properties like that discussed in Section \ref{sec:Harristype}. With this form of DF, the charges and current densities become
\begin{eqnarray}
&&\sigma=-\epsilon_0 \frac{d^2\phi}{dz^2}=\sum_s  q_se^{-q_s\beta_s\phi}\left[n_{as}e^{ q_s\beta_su_{xs}A_x}+n_{bs}e^{q_s\beta_su_{ys}A_y} \right] ,\label{eq:sigmaharris}\\
&&j_x=-\frac{1}{\mu_0}\frac{d^2A_x}{dz^2}=\sum_s   q_s n_{as}u_{xs} e^{ -q_s\beta_s(\phi-u_{xs}A_x)}  ,\label{eq:jxharris}\\
&&j_y=-\frac{1}{\mu_0}\frac{d^2A_y}{dz^2}=\sum_s   q_s n_{bs}u_{ys}e^{-q_s\beta_s(\phi      -u_{ys}A_y)} ,\label{eq:jyharris}
\end{eqnarray}
for $n_{as}=n_{0s}a_s\exp (u_{xs}^2/(2v_{\text{th},s}^2))$ and $n_{bs}=n_{0s}b_s\exp (u_{ys}^2/(2v_{\text{th},s}^2))$. If we now make the assumption of quasineutrality - on the level of $\sigma (\phi,A_x,A_y)=0$ - then from consideration of Equation (\ref{eq:sigmaharris}), we see that one possible solution for $\phi=\phi(A_x,A_y)$ is as
\begin{equation}
\phi(A_x,A_y)=\frac{1}{\beta_e+\beta_i}(\beta_iu_{xi}+\beta_eu_{xe})A_x+\text{const.}=\frac{1}{\beta_e+\beta_i}(\beta_iu_{yi}+\beta_eu_{ye})A_y+\text{const.},\label{eq:phiharris}
\end{equation}
when
\begin{eqnarray}
\beta_iu_{xi}A_x&=&\beta_iu_{yi}A_y,\nonumber\\
\beta_eu_{xe}A_x&=&\beta_eu_{ye}A_y.\nonumber
\end{eqnarray}
Upon substituting Equation (\ref{eq:phiharris}) into Equations (\ref{eq:jxharris}) and (\ref{eq:jyharris}), the problem has now been reduced to solving two second order nonlinear ODEs in $A_x$ and $A_y$, 
\begin{eqnarray}
&&(j_x=)-\frac{1}{\mu_0}\frac{d^2A_x}{dz^2}= j_{x0}e^{\alpha_xA_x},\nonumber\\
&&(j_y=)-\frac{1}{\mu_0}\frac{d^2A_y}{dz^2}= j_{y0}e^{\alpha_yA_y},\nonumber
\end{eqnarray}
for constants $\alpha_x,\alpha_y, j_{x0}$ and $j_{y0}$. For examples/discussions of solutions to ODEs such as these, see \citet{Harris-1962, Schindlerbook, Tassi-2008, Vasko-2013}. Note that Harris treats a problem like this in 1D, but with only one current density component; Schindler treats a 2D problem with only one current density component; Tassi treats a 2D problem in an MHD context and exploiting Lie Point symmetries, but with some 1D solutions; and Vasko also treats the 2D problem with a group theory approach, and only one current density component.

\subsection{Inverse approach for 1D VM equilibria}\label{sec:inverseapproach}
As demonstrated by the above example, the `forward approach' necessarily restricts the choice of electromagnetic fields that one can describe in a VM equilibrium, by the solution of differential equations. The inverse approach bypasses this restriction, since it begins with the prescription of the (electro-)magnetic fields themselves. The counterpoint to this - since the calculation of charge and current densities involves definite integration and hence a loss of information - is that there are in principle an infinite number of possible VM equilibrium DFs for a given macroscopic fluid equilibrium, e.g. see \citet{Wilson-2011} for an explicit demonstration of this feature. 

The inverse approach is used in \citet{Alpers-1969, Channell-1976, Greene-1993, Harrison-2009PRL} to obtain analytical solutions of VM equilibria, and in \citet{Mynick-1979a, Belmont-2012} for numerical ones. All of these works consider 1D Cartesian coordinates, which are very frequently used in the study of waves, instabilities and reconnection (e.g. see \cite{Schindlerbook}). In this work, and without loss of generality, $z$ is taken to be the spatial coordinate on which the system depends, and so $\nabla=(0,0,\partial/\partial z)$. Thus the particle Hamiltonian, $H_s$, and two of the canonical momenta $p_{xs}$ and $p_{ys}$ are conserved, see Equations (\ref{eq:hs} - \ref{eq:pys}).

\subsubsection{Existence of a Vlasov equilibrium}\label{sec:inversekey}
Resembling discussions in e.g. \citet{Bertotti-1963, Channell-1976, Mynick-1979a, Greene-1993, Schindlerbook, Harrison-2009POP}, we now consider the theory that describes macroscopic equilibria in one dimension, given the existence of a Vlasov equilibrium DF. The first velocity moment of the Vlasov equation in Cartesian coordinates
\[
\int_{-\infty}^\infty \boldsymbol{v} \, \left(   \boldsymbol{v}\cdot\frac{\partial f_s}{\partial \boldsymbol{x}}+\frac{q_s}{m_s}(\boldsymbol{E}+\boldsymbol{v}\times\boldsymbol{B})\cdot\frac{\partial f_s}{\partial \boldsymbol{v}}       \right)\,d^3v=\boldsymbol{0},
\]
will, after a little algebra, yield the macroscopic/fluid equation of motion
\[
\nabla\cdot\boldsymbol{P}=\sigma\boldsymbol{E}+\boldsymbol{j}\times\boldsymbol{B}.
\]
In our 1D equilibrium geometry $B_z=j_z=E_x=E_y=0$ automatically, for
\[
\boldsymbol{B}=\nabla\times\boldsymbol{A},\hspace{3mm}
\boldsymbol{E}=-\nabla\phi,
\]
and so this implies that force-balance is maintained by
\begin{eqnarray}
\frac{d}{dz}P_{zx}&=&0,\nonumber \\
\frac{d}{dz}P_{zy}&=&0,\nonumber\\
\frac{d}{dz}P_{zz}&=&\sigma E_z+j_xB_y-j_yB_x.\label{eq:zforcebalance}
\end{eqnarray}
We note here that this type of equilibrium is known as a \emph{tangential equilibrium} (e.g. see \cite{Mottez-2004}), and is characterised by
\[
\boldsymbol{B}\cdot\nabla =0,\hspace{3mm}\boldsymbol{V}_s\cdot\nabla =0,
\]
i.e. the magnetic field and bulk plasma flows are normal to the gradient direction(s). 

If we now consider the dynamic component of the pressure tensor,
\[
P_{zz}=\sum_s m_s \int _{-\infty}^\infty v_z^2\, f_s(H_s(\boldsymbol{v}^2,\phi),p_{xs}(v_x, A_x),p_{ys}(v_y,A_y))\,d^3v,
\]
then we see that $P_{zz}=P_{zz}(\phi,A_x,A_y)$. Note that the pressure tensor is usually found by taking moments by $w_{zs}=v_z-V_{zs}$. But since $f_s$ is only a function of $v_z$ through $H_s$ (which is a function of $v_z^2$), then automatically the $v_z$ moment of $f_s$ is zero, and so the bulk flow $V_{zs}=0$, giving $w_{zs}=v_z$. Using this knowledge of the form of $P_{zz}$ gives
\begin{equation}
\frac{d}{dz}P_{zz}=\frac{d \phi}{dz}\frac{\partial P_{zz}}{\partial \phi}+\frac{dA_x}{dz}\frac{\partial P_{zz}}{\partial A_x}+\frac{dA_y}{dz}\frac{\partial P_{zz}}{\partial A_y},\label{eq:chainrule}
\end{equation}
by the chain rule. A term-by-term comparison of Equation (\ref{eq:zforcebalance}) with Equation (\ref{eq:chainrule}) yields
\begin{eqnarray}
\sigma&=&-\frac{\partial P_{zz}}{\partial\phi},\nonumber\\
j_x&=&\frac{\partial P_{zz}}{\partial A_x},\nonumber\\
j_y&=&\frac{\partial P_{zz}}{\partial A_y},\nonumber
\end{eqnarray}
and so we see that the existence of a Vlasov equilibrium implies the existence of a \emph{potential} function, $P_{zz}$, from which the charge and current densities can be calculated.

The above equations demonstrate that a reasonable first step in an attempt to find a VM equilibrium DF self-consistent with a given set of electromagnetic fields is to first find a  $P_{zz}$ function that is compatible. For example, in the case of a \emph{force-free field} for which $\boldsymbol{j}\times\boldsymbol{B}=\boldsymbol{0}$, there is a simple procedure one can follow to calculate an expression for $P_{zz}(A_x,A_y)$ (for details relevant to force-free fields, see e.g. \citet{Harrison-2009POP} and Chapter \ref{Sheets}).

\subsubsection{Equilibrium DF}
The Vlasov equation can be solved by any differentiable function $f_s(H_s,p_{xs},p_{ys})$, with the additional `physical' constraints being that $f_s$ is also normalisable, non-negative and has velocity moments of arbitrary order \citep{Schindlerbook}. In line with numerous previous works in 1D \citep{Sestero-1967, Alpers-1969, Channell-1976, Harrison-2009PRL, Abraham-Shrauner-2013}, the work in this thesis on VM equilibria in Cartesian coordinates (Chapters \ref{Vlasov}, \ref{Sheets} and \ref{Asymmetric}) shall consider DFs of the form
\begin{equation}
f_s=\frac{n_{0s}}{(\sqrt{2\pi}v_{\text{th},s})^3}e^{-\beta_sH_s}g_s(p_{xs},p_{ys}),\label{eq:F_form}
\end{equation}
for $g_s$ an as yet unknown function, to be determined. This form is chosen for the DF for the following reasons:
\begin{description}
\item[Integrability:] $e^{-\beta_sH_s}$ scales like $e^{-v^2/(2v_{\text{th},s}^2)}$, implying that for a reasonable $g_s$ function, all moments of $f_s$ will be integrable, as necessary,
\item[Solving integrals:] The $e^{-v^2/(2v_{\text{th},s}^2)}$ dependence lends itself to not only being integrable, but to having known definite integrals when multiplied by many functions,
\item[Physical meaning:] As discussed in Section \ref{sec:collisions}, the unique equilibrium DF for a collisional plasma is a Maxwellian. As such it is clear how this Vlasov/collisionless equilibrium DF relates to a collsional equilibrium DF, 
\item[Elegance:] A zero-flow Maxwellian DF is reproduced when $g_{s}=1$.
\end{description}

\subsubsection{Scalar and vector potentials}\label{sec:pseudo}
As demonstrated in Section \ref{sec:forwardexample}, the combination of quasineutrality, \[n_{i}(A_{x},A_{y},\phi)=n_{e}(A_{x},A_{y},\phi),\] and a DF of the form in Equation (\ref{eq:F_form}) results in a scalar potential that is implicitly defined as a function of the vector potential, e.g. \citet{Harrison-2009POP,Schindlerbook,Tasso-2014,Kolotkov-2015}:
\begin{equation}
\phi_{qn}(A_x,A_y)=\frac{1}{e(\beta_e+\beta_i)}\ln (n_i/n_e).\label{eq:quasineutral}
\end{equation} 
In Chapters \ref{Vlasov}, \ref{Sheets} and \ref{Asymmetric}, and as in e.g. \citet{Channell-1976}, parameters will be chosen such that $n_i=n_e$ as functions over $(A_x,A_y)$ space, and so `strict neutrality' is satisfied, implying $\phi_{qn}=0$. This choice of parameters is mathematically equivalent to the condition used to derive the `micro-macroscopic' parameter relationships, which will be discussed later.

It has been commented in e.g. \citet{Grad-1961, Bertotti-1963,Nicholson-1963, Sestero-1966, Mynick-1979a,Attico-1999,Harrison-2009POP,}, that the 1D VM equilibrium problem is analagous to that of a particle moving under the influence of a potential; with the relevant component of the pressure tensor, $P_{zz}$, taking the role of the potential; $(A_x,A_y)$ the role of position and $z$ the role of time. This analogy is demonstrated by
\begin{eqnarray}
\frac{d^2A_x}{dz^2}=-\mu_0\frac{\partial P_{zz}}{\partial A_x},\label{eq:Amp1}\\
\frac{d^2A_y}{dz^2}=-\mu_0\frac{\partial P_{zz}}{\partial A_y}.\label{eq:Amp2}
\end{eqnarray} 
The LHS of the above equations take the role of acceleration, and the RHS take the role of force, as the gradient of a potential. Through this analogy, the task of finding a consistent $P_{zz}$ function - as discussed in Section \ref{sec:inversekey} - can be reformulated as finding a `potential function' $P_{zz}$, such that a `particle trajectory' follows $(A_x(z),A_y(z))$.

\subsubsection{The inverse problem}\label{sec:Channell}
\citet{Channell-1976} developed the theory of the inverse problem in a general sense, with the assumption of zero scalar potential from the offset. It is shown therein that a DF of the form of Equation (\ref{eq:F_form}) implies that the relevant component of the pressure tensor, $P_{zz}$, is a 2-D integral transform of the unknown function $g_{s}$, given by
\begin{eqnarray}
&&P_{zz}(A_x,A_y)=\frac{\beta_{e}+\beta_{i}}{\beta_{e}\beta_{i}}\frac{n_{0s}}{2\pi m_{s}^2v_{\text{th},s}^2}\nonumber\\
&&\times\int_{-\infty}^\infty\int_{-\infty}^\infty \; e^{-\beta_{s}\left((p_{xs}-q_{s}A_x)^2+(p_{ys}-q_{s}A_y)^2\right)/(2m_{s})}g_{s}(p_{xs},p_{ys})dp_{xs}dp_{ys}. \label{eq:Channell}
\end{eqnarray}
This equation together with Equations (\ref{eq:Amp1}) and (\ref{eq:Amp2}) define the inverse problem at hand, viz. `for a given macroscopic equilibrium described by $(A_x(z),A_y(z))$, can we find a self-consistent $P_{zz}(A_x,A_y)$ according to Equations (\ref{eq:Amp1}) and (\ref{eq:Amp2}), and can we then invert the integral transform in Equation (\ref{eq:Channell}) to solve for the unknown function $g_{s}$?' Observe that the LHS of Equation (\ref{eq:Channell}) is species-independent, whereas the RHS seems not to be. In fact, the consistency of this equation for both ions and electrons is one more condition that is implicit in `Channell's method', and is formally compatible with the condition of strict neutrality, $\phi=0$.
 
 \subsubsection{Inversion by Fourier transforms}\label{sec:ftransform}
As written, Equation (\ref{eq:Channell}) is almost exactly a 2D \emph{convolution} of the functions $e^{-(t_1^2+t_2^2)/2}$ and $g(t_1,t_2)$, for a convolution of functions $h_1(t_1,t_2)$ and $h_2(t_1,t_2)$ defined as
\begin{eqnarray}
h_1\star h_2\, (\tau_1,\tau_2)&=&\int_{t_1=-\infty}^{t_1=\infty}\int_{t_2=-\infty}^{t_2=\infty} h_1(\tau_1-t_1,\tau_2-t_2)h_2(t_1,t_2)dt_2dt_1\nonumber ,\\
&=&\int_{t_1=-\infty}^{t_1=\infty}\int_{t_2=-\infty}^{t_2=\infty} h_1(t_1,t_2)h_2(\tau_1-t_1,\tau_2-t_2)dt_2dt_1.\label{eq:cconvolution}
\end{eqnarray}
There is a useful result regarding the Fourier transform,
\[
\text{FT} [h](\omega)=\frac{1}{\sqrt{2\pi}} \int_{-\infty}^\infty e^{-i t \omega} h(t) dt,
\]
of a convolution. The convolution theorem states that
\[
\text{FT} [h_1\star h_2] (\omega_1,\omega_2)= \text{FT} [h_1]  (\omega_1) \, \text{FT} [h_2 ](\omega_2),
\]
\citep{Zayedbook}. That is to say that the Fourier transform of a convolution of functions is the product of the transforms of the individual functions. By making some simple changes of variables, $\boldsymbol{\mathcal{A}}=\boldsymbol{A}/q_s$, Equation (\ref{eq:Channell}) can be manipulated into the form of Equation (\ref{eq:cconvolution}),
\begin{equation}
P_{zz}\left(\frac{A_x}{q_s},\frac{A_y}{q_s}\right)=P_{zz}(\mathcal{A}_{xs},\mathcal{A}_{ys})=\frac{\beta_{e}+\beta_{i}}{\beta_{e}\beta_{i}}\frac{n_{0s}}{2\pi m_{s}^2v_{\text{th},s}^2} e^{-\beta_s(p_{xs}^2+p_{ys}^2)/(2m_s)}\star g_s.\label{eq:pformnew}
\end{equation}
As such, and using the convolution theorem, $g_s$ can - at least formally - be written
\begin{equation} 
g_s(p_{xs},p_{ys})=\frac{\beta_{e}\beta_{i}}{\beta_{e}+\beta_{i}}\frac{2\pi m_{s}^2v_{\text{th},s}^2} {n_{0s}}\text{IFT}\left[    \frac{\text{FT}[P_{{zz}}](\omega_1,\omega_2)    }{\text{FT}\left[   e^{-\beta_s(t_1^2+t_2^2)/(2m_s)}\right] (\omega_1,\omega_2)     }           \right],\label{eq:formalfourier}
\end{equation}
for IFT the inverse Fourier transform,
\[
\text{IFT}[h](t)=\frac{1}{\sqrt{2\pi}}\int_{-\infty}^\infty e^{i t \omega} \text{FT} [h](\omega) d\omega .
\]
Note that the $t_1,t_2,\omega_1,\omega_2$ variables used in Equation (\ref{eq:formalfourier}) are in a sense dummy variables, and do not in fact represent time/frequency in this example, but were used for consistency with the rest of the discussion. For dimensional consistency the conjugate variables to the $p_{xs},p_{ys}$ variables should have dimensions of ``1/momentum''. 

This Fourier transform method is used in \citet{Channell-1976, Harrison-2009PRL} to derive VM equilibrium DFs for 1D macroscopic equilibria, and in a sense this is the most natural method for the problem. At least, one can always formally write down the solution. However, there are two main difficulties:
\begin{description}
\item[Integrability:] Since the Fourier transform of a Gaussian is a Gaussian \citep{Erdelyi}, part of the RHS of Equation (\ref{eq:formalfourier}) is an exponential of a positive quadratic. Formally, the integrability of the RHS places serious restrictions on the nature of $\text{FT}[P_{zz}:(\omega_1,\omega_2)] $, and hence the validity of the method. We note here that despite this formal restriction on the use of the Fourier transform, it is in effect possible to bypass this problem by inspection. For example, in \citet{Neukirch-2009, Abraham-Shrauner-2013} the $g_s$ function is found `by inspection'/using known integrals, that give the same result that the (invalid) Fourier transform method would have.
\item[Integrals:] It may be that certain $P_{zz}$ functions in equation (\ref{eq:pformnew}) have no analytic expression for the Fourier transform, or that the argument of the RHS of Equation (\ref{eq:formalfourier}) has no analytic expression for the inverse Fourier transform.
\end{description}

\subsection{Previous work on VM equilibria}\label{sec:lit}
In this thesis we shall consider theory and examples of exact self-consistent solutions of the VM system for magnetised plasmas, including some non-trivial solutions of Poisson's equation such that the plasma can be either neutral or non-neutral, in Chapter \ref{Cylindrical}. Our focus will be on translationally invariant equilibria in Cartesian geometry in Chapters \ref{Vlasov}, \ref{Sheets} and \ref{Asymmetric}, and on rotationally symmetric equilibria in cylindrical geometry in Chapter \ref{Cylindrical}. These solutions can either describe equilibria of the VM system, such that the one-particle DF for species $s$, $f_s$, satisfies the steady-state Vlasov equation in particle phase space $(\boldsymbol{x},\boldsymbol{v})$,
\begin{equation}
\frac{df_s (\boldsymbol{x},\boldsymbol{v};t)}{dt}=0=\frac{\partial f_s (\boldsymbol{x},\boldsymbol{v};t)}{\partial t},\nonumber
\end{equation}
or as aforementioned in Section \ref{sec:collisionless}, nonlinear wave solutions that satisfy the above equation when Galilei-transformed to the wave frame (e.g. see \cite{Bernstein-1957, Abraham-Shrauner-1968}), by making a transformation
\begin{eqnarray}
\boldsymbol{x}&\to&\boldsymbol{x}-\boldsymbol{u}t,\nonumber\\
\boldsymbol{v}&\to&\boldsymbol{v}-\boldsymbol{u},\nonumber
\end{eqnarray}
for $\boldsymbol{u}$ the phase velocity of the travelling wave. 

Knowledge of exact solutions to the VM system are of value in the study of a wide variety of phenomena in collisionless plasmas, and a comprehensive review and description of all the potential applications is beyond the scope of this thesis. However, we shall survey the theoretical works most relevant to ours, and some applications. Broadly speaking there are three approaches in the literature: on electrostatic and un-magnetised; electrostatic and magnetised; and neutral magnetised plasmas. Of course these `streams' have some overlap, and theoretically the boundary between them is `woolly' by the Lorentz invariance of Maxwell's equations. Specifically, since Galilean frame transformations, $\boldsymbol{u}$ - in the non-relativistic scenario where $u\ll c$ - can send 
\begin{eqnarray}
\boldsymbol{E}^\prime=\boldsymbol{0}\to\boldsymbol{E}=\boldsymbol{u}\times\boldsymbol{B}, \hspace{3mm}\text{or}\label{eq:Gal_E}\\
\boldsymbol{B}^\prime=\boldsymbol{0}\to\boldsymbol{B}=-\frac{1}{c^2}\boldsymbol{u}\times\boldsymbol{E},\label{eq:Gal_B}
\end{eqnarray}
(e.g. see \cite{Griffithsbook,Landaufields}). We interpret Equations (\ref{eq:Gal_E}) and (\ref{eq:Gal_B}) as follows. Consider two coordinate systems: the stationary laboratory, $K$, and one moving at a constant velocity $\boldsymbol{u}$ relative to the laboratory, $K^\prime$. In these two coordinate systems, the electromagnetic fields are denoted without and with primes, respectively. Then Equation (\ref{eq:Gal_E}) says that if in the frame $K^\prime$ the electric field is measured to be $\boldsymbol{E}^\prime=\boldsymbol{0}$, then it measured to be given by $\boldsymbol{u}\times\boldsymbol{B}$ in the $K$ frame. Likewise, Equation (\ref{eq:Gal_B}) says that if in the frame $K^\prime$ the magnetic field is measured to be $\boldsymbol{B}^\prime=\boldsymbol{0}$, then it measured to be given by $c^{-2}\boldsymbol{u}\times\boldsymbol{E}$ in the $K$ frame.

Not only that, but the differences/distinctions between the following frequently assumed states:
\begin{itemize}
\item `strict neutrality' (e.g. see \cite{Grad-1961,Channell-1976}),
\[
\phi=0\implies\sigma=0;\nonumber
\] 
\item quasineutrality, i.e. $\sigma=0$ to first order, as introduced in Section \ref{sec:quasi}), and typically achieved in the literature (e.g. see \cite{Harrison-2009POP}) by
\[
\phi=\phi(\boldsymbol{A}(\boldsymbol{x}))\hspace{3mm}\text{s.t.}\hspace{3mm}\sigma=0;\nonumber
\]
\item non-neutrality (e.g. see \cite{Davidsonbook} for the authoritative text),
\[
\phi=\phi(\boldsymbol{x})\hspace{3mm}\text{s.t.}\hspace{3mm}\sigma\ne 0,\nonumber
\]
\end{itemize}
are subtle (e.g. see \cite{Bertotti-1963, Greene-1993,Schindlerbook}). Given these considerations, we shall make some crude distinctions, and given that the electrostatic literature is relatively self-contained and seemingly the one that gained maturity the quickest, we describe this first.

The seminal work on electrostatic solutions of the VM system in the absence of a magnetic field is that of \citet{Bernstein-1957}, in which an inductive method is developed that calculates the DF of trapped electrons in a nonlinear travelling electrostatic wave (\emph{Bernstein-Greene-Kruskal (BGK) waves}), for a given 1D scalar potential, $\phi$, in the wave frame. This work was developed upon in particular by \citet{Schamel-1971, Schamel-1972JPP} with particular emphasis on the necessary condition of positivity of the DF. Other theoretical works in a 1D geometry include those on ion-acoustic waves (e.g. see \cite{Schamel-1972PPCF}, ion/electron holes and double layers (e.g. see \cite{Schamel-1986,Schamel-2000}), generalisations and extensions of BGK theory (e.g. see \cite{Lewis-1984,Karimov-1999}), and `three-dimensional BGK waves' (e.g. see \cite{Ng-2005,Ng-2006}). One particular application of this theory is the phenomena of collisionless shocks (e.g. see \cite{Burgessbook, Marcowith-2016}), relevant in astrophysical, laboratory, and laboratory astrophysical contexts (e.g. see \cite{Montgomery-1969,Forslund-1970,Forslund-1971, Eliasson-2006JETPL, Spitkovsky-2008APJA,Stockem-2014, Cairns-2014, Svedung-Wettervik-2016}).

There exists a similarly rich literature for magnetised quasi-neutral and non-neutral solutions (the majority of which is quasi-neutral), much of which is collected in the articles by \citet{Roth-1996, Zelenyi-2011, Artemyev-2013}. Perhaps the most ubiquitous work in the context of current sheets is that of \citet{Harris-1962}, in which it is demonstrated that the DF consistent with the 1D Harris current sheet and for a plasma with zero scalar potential can, by using a post-hoc Galilean transformation, also describe a non-neutral configuration (the Harris sheet equilibrium is considered in the relativistic case in \cite{Hoh-1966}). The foundational work in the realm of magnetised and electrostatic collisionless shocks is that of \citet{Sagdeev-1966}, in which analogies are drawn between the equations describing solitary waves, and the motion of a particle in a potential: the Sagdeev potential. General theoretical treatments on quasi-neutral and non-neutral VM equilibria include, for
\begin{itemize}
\item 1D plasmas: \citet{Tonks-1959,Sestero-1964,Sestero-1966,Sestero-1967, Lam-1967, Abraham-Shrauner-1968,Lemaire-1976,Lee-1979JPP,Mitchell-1979, Greene-1993, Mottez-2003, Yoon-2006, Balikhin-2008, Artemyev-2011}, 
\item Two-dimensional (2D) plasmas: \citet{Hewett-1976,Mynick-1979a,Kan-1979PSS,Otto-1984,Muschietti-2000,Schindler-2002,Eliasson-2006PPCF, Suzuki-2008, Kocharovsky-2010, Schindlerbook,Ng-2012, Vasko-2013},
\item With applications to magnetospheres for 1D plasmas: \citet{Davies-1968, Davies-1969, Su-1971, Kan-1979JGRa,Stern-1981JGR,Stern-1981NASA, Rogers-1988, DeVore-2015},
\item With applications to magnetospheres for 2D plasmas: \citet{Kan-1979JGRb,Lee-1979JGR, Birn-2004} .
\end{itemize}

For theoretical treatments that treat the plasma as strictly neutral ($\phi=0$), see \citet{Grad-1961, Hurley-1963, Nicholson-1963, Schmid-Burgk-1965,Moratz-1966, Lerche-1967, Alpers-1969, Channell-1976, Bobrova-1979, Lakhina-1983, Attico-1999, Bobrova-2001, Fu-2005, Yoon-2005, Harrison-2009PRL, Neukirch-2009, Panov-2011, Wilson-2011, Belmont-2012, Janaki-2012, Abraham-Shrauner-2013, Ghosh-2014, Kolotkov-2015,Allanson-2015POP,Allanson-2016JPP}.

We should indicate that there also exists a substantial literature on magnetised neutral and non-neutral VM solutions in cylindrical geometry (for example flux tubes, mono-energetic beams, laboratory pinches and astrophysical jets), with \citet{Davidsonbook,,Vinogradov-2016, Allanson-2016POP} and references therein, as well as Chapter \ref{Cylindrical} providing a suitable starting point for an interested reader.

\section{Thesis motivation and outline}
The importance of understanding the equilibrium states permitted by a given system is common to most physical disciplines, and this is - broadly speaking - the motivation for the work in this thesis. Specifically, I shall consider electromagnetic structures that - by the balance of electromagnetic, inertial, and thermal pressure forces - confine the mass and electric currents in a plasma. These equilibrium configurations will be considered in Cartesian and cylindrical geometries, namely current sheets and flux tubes. There are many potential applications for current sheet and flux tube equilibria, and these shall be discussed in Chapters \ref{Sheets}, \ref{Asymmetric}, and then \ref{Cylindrical} respectively. However, the main/most timely application of the work in this thesis could be to studies of magnetic reconnection, for which localised currents are a pre-condition. 

\subsection{Outline of the thesis}
This thesis is structured as follows:
\begin{itemize}

\item Chapter \ref{Vlasov}: {\bf The use of Hermite polynomials for the inverse problem in one-dimensional Vlasov-Maxwell equilibria}\\
By expressing the unknown functions, $g_{s}$, of the canonical momenta as (infinite) expansions of Hermite polynomials, we establish a one-to-one correspondence between the coefficients of expansion, and those of a Maclaurin expansion of the pressure tensor. We then find a sufficient condition for the convergence of the Hermite representation, contingent on the Maclaurin expansion coefficients of the pressure tensor. For certain classes of DFs, we prove results on the non-negativity of the $g_{s}$ function, and make a conjecture for all other classes. 

\item Chapter \ref{Sheets}: {\bf One-dimensional nonlinear force-free current sheets}\\
Using pressure transformation techniques, we find a new pressure tensor self-consistent with the force-free Harris sheet magnetic field, for any value of the plasma beta, and crucially sub-unity values that could not be accessed before. Then we use the Hermite polynomial expansion technique established in Chapter \ref{Vlasov} to calculate a Vlasov equilibrium DF consistent with the low beta force-free Harris sheet. Next, the Hermite expansion is proven to be analytically convergent, using the sufficient condition derived in Chapter \ref{Vlasov}, and we confirm that the DF satisfies the conjectured condition for non-negativity of the Hermite representation of a DF, also from Chapter \ref{Vlasov}.
 
We conduct a preliminary analysis on the physical properties of the DF, but encounter numerical difficulties for the parameter range of interest when attempting to make plots for $\beta_{pl}<0.85$. In response to this difficulty, we `re-gauge' the vector potential, allowing for numerical convergence of the Hermite expansions for much lower values of the plasma beta, $\beta_{pl}=0.05$. As before, we establish the necessary convergence and non-negativity of the DF, and present new plots.

\item Chapter \ref{Asymmetric}: {\bf One-dimensional asymmetric current sheets}\\
We first consider the mathematical problem for a pressure tensor consistent with an `asymmetric' current sheet equilibrium. Using these results, we present possible examples of pressure tensors self-consistent with asymmetric equilibria, and discuss the inverse problem. It becomes apparent that for certain representations, the problem is not analytically soluble, and numerical techniques are necessary. Using representations for the pressure tensor that give soluble solutions, we present exact analytic VM equilibria for an asymmetric Harris sheet with guide field, and a preliminary analysis

\item Chapter \ref{Cylindrical}: {\bf Neutral and non-neutral flux tube equilibria}\\ 
This is a departure from the previous work on translationally invariant systems. First we consider the problem of constructing one-dimensional VM equilibria in cylindrical geometry, and establish the fluid equation(s) of motion from the Vlasov equation in cylindrical geometry. We include an analysis of the microscopic origin of the macroscopic forces in the resultant equation of motion. 

Next, there is discussion on the attempts to construct VM equilibria for the exact Gold-Hoyle model, a force-free flux tube. These attempts do not yield solutions, and there seems to be good physical reasoning behind the mathematical difficulties. By making a small change to the macroscopic magnetic field, we are able to find a consistent VM equilibrium for the Gold-Hoyle model embedded in a uniform background field. We present a preliminary analysis of the equilibrium, including a consideration of multiple maxima in velocity space, and the non-neutrality of the macroscopic configuration.

\item Chapter \ref{Discussion}: {\bf Discussion}\\
We briefly summarise the main results from this thesis, and place them in the context of current plasma physics research. In particular we focus on open questions and avenues that merit further investigation.


\end{itemize}

\null\newpage

\chapter{The use of Hermite polynomials for the inverse problem in one-dimensional Vlasov-Maxwell equilibria} \label{Vlasov} 

\epigraph{\emph{Boltzmann's is still the most beautiful equation in the world, but Vlasov's isn't too shabby!}}{\textit{C\'{e}dric Villani }}

\noindent Much of the work in this chapter is drawn from \citet{Allanson-2015POP, Allanson-2016JPP}.

\section{Preamble}
In this chapter, the aim is to make a contribution to the theory of exact equilibrium solutions to the Vlasov-Maxwell system, in 1D Cartesian geometry. In particular, we consider a solution method for the inverse problem in collisionless equilibria, namely that of calculating a VM equilibrium for a given macroscopic (fluid) equilibrium. Using Jeans' theorem \citep{Jeans-1915}, the equilibrium DFs are expressed as functions of the constants of motion, in the form of a stationary Maxwellian multiplied by an unknown function of the two conserved canonical momenta. In this case it is possible to reduce the inverse problem to inverting Weierstrass transforms, which we achieve by using expansions over Hermite polynomials. A sufficient condition on the pressure tensor is found which guarantees the convergence of the candidate solution when satisfied, and as a result the existence of velocity moments of all orders. This condition is obtained by elementary means, and it is clear how to put it into practice. We also argue that for a given pressure tensor for which our method applies, there always exists a non-negative DF for a sufficiently magnetised plasma. This argument is in fact proven for certain classes of DFs, and in the form of conjecture for others.

\section{Introduction}\label{vlasovintro}
\subsection{Hermite polynomials in fluid closure}
\[
f=\frac{n(\boldsymbol{x},t)}{(\sqrt{2\pi k_BT(\boldsymbol{x},t)/m} )^3}e^{-\boldsymbol{w}(\boldsymbol{x},t)^2/(2k_BT(\boldsymbol{x},t)/m)}\sum_{n=0}^\infty a^{(n)}(\boldsymbol{x},t)\mathcal{H}^{(n)}(\boldsymbol{w}),
\]
for $\mathcal{H}^{(n)}$ the $n$-dimensional Hermite ``polynomial'', and in fact a rank-$n$ tensor, defined by
\begin{eqnarray}
\mathcal{H}^{(n)}(\boldsymbol{w})&=&\frac{(-)^n}{\mathcal{W}(\boldsymbol{w})}\frac{\partial ^n}{\partial w_{i_1}...\partial w_{i_n}}\mathcal{W}(\boldsymbol{w}),\nonumber\\
\text{s.t.}\hspace{3mm}\mathcal{W}(\boldsymbol{w})&=&\frac{1}{(2\pi)^{3/2}}e^{-w^2/(2k_BT(\boldsymbol{x},t)/m)},
\end{eqnarray}
with each of the $i_n$-indices running over $\{x,y,z\}$. Note that - by the commutativity of partial derivatives - the labelling of the n-dimensional Hermite polynomials is somewhat degenerate, e.g. $H^{(2)}_{xy}=H^{(2)}_{yx}=w_xw_y$.

In this representation $a^{(n)}\mathcal{H}^{(n)}$ is the scalar product of two rank-$n$ tensors, with the $a$ `coefficients' relating directly to the velocity moments of the DF, and as such they neatly `index' the relationship between the particle distributions and certain macroscopic quantities:
\begin{eqnarray}
a^{(0)}&=&1\iff \int f d^3v =n(\boldsymbol{x},t) ,\nonumber\\
\boldsymbol{a}^{(1)}&=&(0,0,0)\iff \int w_i fd^3v=0\nonumber\\
a^{(2)}_{ij}&=&\pi_{ij}/p,\iff \pi_{ij}=P_{ij}(\boldsymbol{x},t)-\delta_{ij}p(\boldsymbol{x},t) \nonumber\\
a^{(3)}_{ijk}&=&S_{ijk}/(pv_{\text{th}})\iff \int w_iw_jw_kfd^3v=S_{ijk} (\boldsymbol{x},t) .\nonumber\\
&\vdots&\nonumber
\end{eqnarray}
for $\delta_{ij}$ the Kronecker delta and $S_{ijk}$ the heat flux tensor. By substituting this expanded form of the DF into Boltzmann's equation (Equation (\ref{eq:Boltz})), multiplying by $\mathcal{H}^{n}(\boldsymbol{w})$ and then integrating over velocity space $d^3v$, Grad obtains an infinite hierarchy of differential equations that describe the spatial-temporal evolution of the $a^{(n)}$ coefficients, and in turn the moments of the DF. By truncating to third order (i.e. up to $S_{ijk}$), Grad then develops the ``13-moment'' equations for the variables $n, \boldsymbol{V}, T, \pi_{ij}$ and $S_i=pv_{\text{th}}a^{(3)}_{ijj}$.

Grad uses Hermite polynomials (or generalisations thereof) in gas kinetic theory because of their orthogonality properties with respect to Gaussian functions, and this is what allows each term of order $n$ in the expansion of the DF to be directly related to $n^{\text{th}}$ order velocity-space moments of the DF. It is for this very reason that Hermite polynomials have a long history in plasma physics.

\subsection{Hermite polynomials in VM plasma theory}
The most typical approach in collisionless and weakly collisional plasma kinetic theory is to use expansions in `scalar' Hermite polynomials, defined by
\begin{eqnarray}
H_n(v)&=&(-1)^ne^{v^2}\frac{d^{n}}{dv^{n}}e^{-v^2},\label{eq:hdef}\\
\int_{-\infty}^\infty H_m(v)H_n(v)e^{-v^2}dv&=&\delta_{mn}2^nn!\sqrt{\pi}.\label{eq:orthogonal}
\end{eqnarray}
Hermite polynomials are a complete orthogonal set of polynomials for $f\in L^2(\mathbb{R},e^{-v^2}dv)$ \citep{Arfkenbook}. That is to say that for any piecewise continuous $f$, such that
\begin{equation}
\int_{-\infty}^\infty |f|^2 e^{-v^2}dv <\infty ,\label{eq:squareintegrable}
\end{equation}
then there exists an (infinite) expansion in Hermite polynomials, $\sum_{n=0}^{\infty} c_n H_n(v)$, such that
\begin{equation}
\lim_{k\to\infty}\int_{-\infty}^\infty\bigg|f-\sum_{n=0}^k c_n H_n(v)\bigg|^2e^{-v^2}dv=0.\label{eq:parity}
\end{equation}
Whereas Equations (\ref{eq:hdef}) and (\ref{eq:squareintegrable}) are the standard definitions relevant to the use of Hermite polynomials, it will be of use in this work to consider the scaled function $H_n(v/(\sqrt{2}v_{\text{th},s}))$, since Maxwellian DFs scale with $e^{-v^2/(2v_{\text{th},s}^2)}$, as opposed to $e^{-v^2/(v_{\text{th},s}^2)}$. This slight modification results in changes to Equations (\ref{eq:hdef}), (\ref{eq:orthogonal}), (\ref{eq:squareintegrable}) and (\ref{eq:parity}), easily achieved by substitution.

\subsubsection{Hermite polynomials in velocity space}\label{sec:phasespace}
As intimated above, expansions in Hermite polynomials are a natural choice for representing the velocity space structure of a DF in equilibrium and near-equilibrium plasmas, be the (near-)equilibrium collisional and hence (near-)thermal; or collisionless, and hence not necessarily (near-)thermal at all. Their suitability is epitomised by Equation (\ref{eq:orthogonal}), and is demonstrated as follows. 

First consider a quite general DF, written explicitly as a function over phase space $(\boldsymbol{x},\boldsymbol{v};t)$, and of the form
\begin{eqnarray}
&&f_s(\boldsymbol{x},\boldsymbol{v},t)=\frac{n_s(\boldsymbol{x},t)}{(\sqrt{2\pi}v_{\text{th},s}(\boldsymbol{x},t) )^3}e^{-v^2/(2(v_{\text{th},s}(\boldsymbol{x},t)^2))}\nonumber\\
&&\times\sum_{ij} a_{ij}(\boldsymbol{x},t)H_i\left(\frac{v_x}{\sqrt{2}v_{\text{th},s}(\boldsymbol{x},t)}\right)H_j\left(\frac{v_y}{\sqrt{2}v_{\text{th},s}(\boldsymbol{x},t)}\right),\label{eq:fvelocity}
\end{eqnarray}
where we define a time and space dependent thermal velocity by $v_{\text{th},s}(\boldsymbol{x},t)=k_BT_s(\boldsymbol{x},t)/m_s$. Expansions such as these are used in \citet{Hewett-1976, Camporeale-2006, Suzuki-2008}, for example. This form of the DF implies that a velocity space moment with respect to the $(i,j)^{\text{th}}$-order Hermite polynomials is directly related to the $(i,j)^{\text{th}}$-order coefficient of expansion,
\[
\int f_s H_i\left(\frac{v_x}{\sqrt{2}v_{\text{th},s}}\right)H_j\left(\frac{v_y}{\sqrt{2}v_{\text{th},s}}\right) d^3v\propto n_s(\boldsymbol{x},t)a_{ijk}(\boldsymbol{x},t).
\]
A DF expanded in Hermite polynomials in the manner of Equation (\ref{eq:fvelocity}) also possesses the feature that `normal' velocity moments yield simple results, since the velocity space moments can be determined using the following definite integral \citep{Gradshteyn},
\begin{equation}
\int_{-\infty}^\infty v^n e^{-v^2}H_n(v)dv=n!\sqrt{\pi}.    \label{eq:Gradshteyn}
\end{equation} 
For example, the charge density and current density are directly related to the $a_{ij}$ coefficients according to
\begin{eqnarray}
\sigma(\boldsymbol{x},t)&\propto &\sum_s q_sn_sa_{00},\nonumber\\
j_x(\boldsymbol{x},t)&\propto &\sum_s q_sn_sv_{\text{th},s}a_{10},\nonumber\\
j_y(\boldsymbol{x},t)&\propto &\sum_s q_sn_sv_{\text{th},s}a_{01}.\nonumber
\end{eqnarray}

\subsubsection{Hermite polynomials in momentum space}
The usefulness of Hermite polynomial expansions is not necessarily restricted to writing them as explicit functions of velocity space. If one considers VM equilibria, then as aforementioned the equilibrium DF is a function of phase space $(\boldsymbol{x},\boldsymbol{v})$ \emph{through} its dependence on the constants of motion. In such circumstances one could write the DF as a stationary Maxwellian multiplied by an expansion in Hermite polynomials in the canonical momenta. For example, in the case of a 1D plasma such that $\nabla=(0,0,\partial/\partial z)$, one could write
\begin{equation}
f_s=\frac{n_{0s}}{(\sqrt{2\pi}v_{\text{th},s} )^3}e^{-\beta_sH_s}\sum_{ij}a_{ij}H_i\left(\frac{p_{xs}}{\sqrt{2}m_sv_{\text{th},s}}\right)H_j\left(\frac{p_{ys}}{\sqrt{2}m_sv_{\text{th},s}}\right),\label{eq:fmomenta}
\end{equation}
(e.g. see \cite{Abraham-Shrauner-1968,Channell-1976} for expansions such as these). Despite the fact that the Maxwellian factor, $e^{-\beta_sH_s}$, is a function of $v^{2}$, and the Hermite polynomials are functions of the momenta, one can still exploit the orthogonality properties of the Hermite polynomials. To see this, we can use the identity mentioned in \citet{Hermite}
\begin{equation}
H_j(x+y)=(H+2x)^j,\hspace{3mm} \text{s.t.}\hspace{3mm}H^j:=H_j(y).\label{eq:Hermite}
\end{equation}
The identity in Equation (\ref{eq:Hermite}), and proven below, is useful since we can associate $X=x+y$ with $p_{js}=m_sv_j+q_sA_j$. This allows us to re-write the DF from Equation (\ref{eq:fmomenta}), and to separate the dependence on velocity and vector potential. Since the vector potential is a function of space ($z$) only, the phase-space variables have also been `separated' allowing us to use results such as those explained in Section (\ref{sec:phasespace}).

We now prove this identity, since it seems fairly non-standard, and the above reference cites personal communication as the source:
\begin{proof}
We first make use of the generating function for Hermite polynomials \citep{Arfkenbook}
\begin{eqnarray}
\exp(2Xt-t^2)&=&\sum_{j=0}^\infty H_j(X)\frac{t^j}{j!}.\label{eq:genfunc}
\end{eqnarray}
By substituting $X=x+y$ into Equation (\ref{eq:genfunc}) we see that
\begin{eqnarray}
\exp(2(x+y)t-t^2)&=&\sum_{j=0}^\infty H_j(x+y)\frac{t^j}{j!},\nonumber\\
&=&\exp(2xt)\sum_{i=0}^\infty H_i(y)\frac{t^i}{i!}.\nonumber
\end{eqnarray}
Then, expanding $\exp(2xt)$ as an infinite series implies that
\begin{equation}
\sum_{i=0}^\infty\sum_{k=0}^\infty\frac{(2xt)^k}{k!}\frac{H_i(y)t^i}{i!}=\sum_{j=0}^\infty H_j(x+y)\frac{t^j}{j!}.
\end{equation}
To isolate the $H_j(x+y)$ term, we now need to pick the terms such that $i+k=j$:
\begin{eqnarray}
 \frac{H_j(x+y)}{j!}&=&\sum_{k=0}^j\frac{(2x)^k}{k!(j-k)!}H_{j-k}(y),\\
\implies H_j(x+y)&=&\sum_{k=0}^j {j \choose k}(2x)^kH_{j-k}(y),\\
\implies H_j(x+y)&=&(H+2x)^j,\hspace{3mm} H^j:=H_j(y) .
\end{eqnarray}
\end{proof}

\subsection{Hermite polynomials for exact VM equilibria}
In the work by \citet{Abraham-Shrauner-1968}, expansions in Hermite polynomials of the canonical momentum are used to solve the VM system for the case of `stationary waves' in a manner like that to be described in this chapter. These correspond not to Vlasov equilibria, but rather to nonlinear waves that are stationary in the wave frame, as discussed in Section \ref{sec:lit}. Abraham-Shrauner considers a 1D plasma with only one component of current density, first in a general sense, and then considers three different magnetic field configurations. \citet{Alpers-1969} also presents a somewhat general discussion on the use of Hermite polynomials for 1D VM equilibria, and proceeds to consider models suitable for the magnetopause, with both one component of the current density, and with two. In the work by \citet{Channell-1976}, two methods are presented for the solution of the inverse problem with neutral VM equilibria, by means of example. These two methods are inversion by Fourier transforms and -- once again -- expansion over Hermite polynomials respectively. Channell uses Hermite polynomials in the canonical momenta, but this time with two components of the current density, for the specific case of a magnetic field that is especially suitable to be considered as a stationary wave solution.

In contrast to \citet{Abraham-Shrauner-1968, Alpers-1969, Channell-1976}, the works by \citet{Hewett-1976,Suzuki-2008} both consider the forwards problem in VM equilibria, and use Hermite polynomial expansions in velocity space, for 1D and 2D plasmas respectively. \citet{Hewett-1976} assume a representation for the DF similar to that in Equation (\ref{eq:fvelocity}) but with only one current density component, and ensure self-consistency with Maxwell's equations numerically, whereas \citet{Suzuki-2008} use an analytical approach, e.g. demonstrating that the Hermite polynomial approach can reproduce known equilibria such as the Harris sheet \citep{Harris-1962}, and the Bennett Pinch \citep{Bennett-1934}.

To give a subset of (modern) examples outside the realm of equilibrium studies \emph{per se}, Hermite polynomial expansions are used by \citet{Daughton-1999} to assess the linear stability of a Harris current sheet; by \citet{Camporeale-2006} also on the linear stability problem, using a truncation method somewhat like that of \cite{Grad-1949b}, and managing to bypass the traditional approach of integrating over the `unperturbed orbits' \citep{Coppi-1966, Drake-1977, Quest-1981A, Daughton-1999}; by \citet{Zocco-2015} on linear collisionless Landau damping  \citep{Landau-1946,Mouhot-2011}; and by \citet{Schekochihin-2016} on the problem of the free-energy associated with velocity-space moments of the DF, in the problem of plasma turbulence. 

\subsubsection{Mathematical criteria}
Since a DF represents a probability (in phase space), it clearly must satisfy the property
\begin{equation}
f_s\ge 0 \, \forall \, \boldsymbol{x},\boldsymbol{v},t ,\label{eq:constraint1}
\end{equation}
and since a DF found using a Hermite polynomial method could in principle include an infinite series of polynomials in momenta/velocity that does not represent a known function in closed form, it is by no means clear if Equation (\ref{eq:constraint1}) will be satisfied. This issue is recognised by \citet{Abraham-Shrauner-1968, Hewett-1976}. Not only is the non-negativity in question, but it is not obvious whether a given expansion in Hermite polynomials even converges, and this question was also raised by \citet{Hewett-1976}. Finally, even if the Hermite expansion converges, it must -when multiplied by the Maxwellian factor - produce a DF for which velocity moments of all orders exist, as discussed in Section \ref{sec:collisions}. In order to have full confidence in the Hermite polynomial method we need to address these issues of non-negativity, convergence, and the existence of moments.

Crucially, none of the above references tackle the questions of non-negativity and convergence of an infinite series of Hermite polynomials in a systematic way, or of the boundedness of the resultant DF. The method presented in this chapter should be seen as a rigorous extension, or generalisation, of the Hermite Polynomial discussed previously by these authors. 

We should mention that the \emph{reverse} questions are well established, i.e. if one \emph{a priori} knows the DF in closed form, or at least if Equation (\ref{eq:squareintegrable}) is satisfied. In such circumstances, one can represent a given non-negative DF as a Maxwellian multiplied by an expansion in Hermite polynomials provided the $g_s$ function grows at a rate below $e^{v^2/(4v_{\text{th},s}^2)}$ \citep{Grad-1949b, Widder-1951}.

The structure of the rest of this chapter is as follows. Section \ref{sec:formal} contains the details of a formal solution to the inverse problem, by using known methods of inverting \emph{Weierstrass transforms} with possibly infinite series of Hermite polynomials. For the formal solution to meaningfully describe a DF however, these series must be convergent, positive and bounded. A sufficient condition for convergence that places a restriction on the pressure tensor is obtained in Section \ref{sec:convergence}. In Section \ref{sec:non-neg} we argue that for an appropriate pressure function, there always exists a positive DF, for a sufficiently magnetised plasma, including proofs for a certain class of function.

\section{Formal solution by Hermite polynomials}\label{sec:formal}
It was demonstrated in Section \ref{sec:inversekey} that the pressure tensor component $P_{zz}$ can be seen as the `key' to solving the inverse problem for VM equilibria. In a 1D $z-$dependent geometry, the inverse problem is encapsulated by Equation (\ref{eq:Channell}), repeated below, 
\begin{eqnarray}
&&P_{zz}(A_x,A_y)=\frac{\beta_{e}+\beta_{i}}{\beta_{e}\beta_{i}}\frac{n_{0s}}{2\pi m_{s}^2v_{\text{th},s}^2}\nonumber\\
&&\times\int_{-\infty}^\infty\int_{-\infty}^\infty \; e^{-\beta_{s}\left((p_{xs}-q_{s}A_x)^2+(p_{ys}-q_{s}A_y)^2\right)/(2m_{s})}g_{s}(p_{xs},p_{ys})dp_{xs}dp_{ys}, \nonumber
\end{eqnarray}
along with Amp\`{e}re's Law and quasineutrality (in this chapter we shall assume strict neutrality),
\begin{eqnarray}
\frac{\partial P_{zz}}{\partial A_x}&=&-\frac{1}{\mu_0}\frac{d^2A_x}{dz^2},\nonumber\\
\frac{\partial P_{zz}}{\partial A_y}&=&-\frac{1}{\mu_0}\frac{d^2A_y}{dz^2},\nonumber\\
\phi &=&0.\nonumber
\end{eqnarray}
The subsequent work in this chapter assumes that such a function, $P_{zz}(A_x,A_y)$, has been found. To make mathematical progress, we shall make the assumption that the $P_{zz}$ function found is of either `summative' or `multiplicative' separability, i.e. that $P_{zz}(A_x,A_y)$ is of the form
\begin{equation}
P_{zz}=\frac{n_0(\beta_e+\beta_i)}{\beta_e\beta_i}\left(\tilde{P}_1(A_x)+\tilde{P}_2(A_y)\right)\;{\rm or}\;P_{zz}=\frac{n_0(\beta_e+\beta_i)}{\beta_e
\beta_i}\tilde{P}_1(A_x)\tilde{P}_2(A_y).\label{eq:pform}
\end{equation}
The constants $n_0, \beta_e$ and $\beta_i$ are present in order to give the correct dimensions to the $P_{zz}$ expression, in a species independent manner, such that the `components' of the pressure, $\tilde{P}_1(A_x)$ and $\tilde{P}_2(A_y)$, are dimensionless. These assumptions are commensurate with 
\begin{equation}
g_{s}=g_{1s}(p_{xs};v_{\text{th},s})+g_{2s}(p_{ys};v_{\text{th},s})\;{\rm or}\;g_{s}=g_{1s}(p_{xs};v_{\text{th},s})g_{2s}(p_{ys};v_{\text{th},s})\label{eq:gsep},
\end{equation}
respectively, and allow separation of variables according to 
\begin{eqnarray}
\tilde{P}_1(A_x)=\frac{1}{\sqrt{2\pi} m_{s}v_{\text{th},s}}\int_{-\infty}^{\infty}\; e^{-\beta_{s}\left(p_{xs}-q_{s}A_x\right)^2/(2m_{s})}g_{1s}(p_{xs};v_{\text{th},s})dp_{xs},\label{eq:p1tog1}\\ 
\tilde{P}_2(A_y)=\frac{1}{\sqrt{2\pi} m_{s}v_{\text{th},s}}\int_{-\infty}^{\infty}\; e^{-\beta_{s}\left(p_{ys}-q_{s}A_y\right)^2/(2m_{s})}g_{2s}(p_{ys};v_{\text{th},s})dp_{ys}.\label{eq:p2tog2}
\end{eqnarray}
The separation constant is set to unity in the case of multiplicative separability, and zero in the case of additive separability, without loss of generality. We have included the parametric dependence on the thermal velocity, $v_{\text{th},s}$, in the $g_s$ functions to highlight the fact that the $g_s$ functions must behave in such a way that the RHS of Equations (\ref{eq:p1tog1}) and (\ref{eq:p2tog2}) must, after integration, be independent of species as discussed in Section \ref{sec:Channell}. This would be impossible if $g_s$ did not depend on $v_{\text{th},s}$.

The components of the pressure are now represented by 1D integral transforms of the unknown parts of the DF, namely Weierstrass transforms.

\subsection{Weierstrass transform}\label{sec:Weierstrass}
The Weierstrass transform, $u(x,t)$ of $u_0(y)$, is defined by
\begin{equation}
u(x,t):=\mathcal{W}\left[u_0\right] (x,t)=\frac{1}{\sqrt{4\pi t}}\int_{-\infty}^\infty \, e^{-(x-y)^2/(4t)}\,u_0(y)\,dy,\label{eq:Weierstrass}
\end{equation}
see \citet{Bilodeau-1962} for example. This is also known as the Gau{\ss} transform, Gau{\ss}-Weiertrass transform and the Hille transform \citep{Widder-1951}. As the Green's function solution to the heat/diffusion equation, 
\begin{eqnarray}
\frac{\partial u}{\partial t}-\frac{\partial^2 u}{\partial x^2} &=&0\nonumber,\\
\text{such that} \hspace{3mm}u(x,0)&=&u_0(x), \, \forall x\in (-\infty,\infty),\nonumber\\
\implies u(x,t)&=&\mathcal{W}\left[u_0\right] (x,t) \nonumber ,
\end{eqnarray}
$u(x,1)$ represents the temperature/density profile of an infinite rod one second after it was $u_0(x)$, see \citet{Widder-1951}. Hence the Weierstrass transform of a positive function is itself a positive function.

\subsection{Two interpretations with respect to our equations }
Give or take some constant factors, Equations (\ref{eq:p1tog1}) and (\ref{eq:p2tog2}) express $\tilde{P}_1$ and $\tilde{P}_2$ as Weierstrass transforms of $g_{1s}$ and $g_{2s}$ respectively. To discuss this problem in generality, the following discussions in this chapter will make regular use of the subscript $j\in\{1,2\}$. This index will indicate the following components for the vector potential and canonical momenta,
\begin{eqnarray}
(A_1,A_2)&:=&(A_x,A_y),\nonumber\\
(p_{1s},p_{2s})&:=&(p_{xs},p_{ys})\nonumber
\end{eqnarray}
Otherwise, the indexing of $P_1, P_2, g_{1s}, g_{2s}$ will remain ``as is''. As such the inverse problem is now characterised by the following equation,
\[
\tilde{P}_{j}(A_j)=\frac{1}{\sqrt{2\pi} m_{s}v_{\text{th},s}}\int_{-\infty}^{\infty} \; e^{-\beta_s(p_{js}-q_{s}A_j)^2/(2m_s)}g_{js}(p_{js};v_{\text{th},s})dp_{js}
\]
To be precise, there are two different interpretations of the equations that could be made here, namely:
\begin{description}
\item[Dimensionality retained and `time' is a variable:]
\begin{equation}
\tilde{P}_{j}(A_j)=:\mathcal{I}_{j}(\mathcal{A}_{js})=\frac{1}{\sqrt{4\pi\varepsilon_s}}\int_{-\infty}^\infty e^{-(p_{js}-\mathcal{A}_{js})^2/(4\varepsilon_s)}g_{js}(p_{js};\varepsilon_s)dp_{js},\label{eq:timeanalogy}
\end{equation}
for $\varepsilon_s=m_s^2v_{\text{th},s}^2/2$ and $\mathcal{A}_{js}=q_{s}A_{j}$. This first interpretation is depicted by Equation (\ref{eq:timeanalogy}) and casts the inverse problem in direct comparison with the Weiertrass transform, making a correspondence between space and time in the heat equation, $(x, t)$, to $(\mathcal{A}_{js},\varepsilon_s)$ in our inverse problem. However, one difference is that the $g_s$ function must - at least parametrically - depend on `time', $\varepsilon_s$, in contrast to the initial condition (i.e. time-independent function) that is part of the integrand in Equation (\ref{eq:Weierstrass}). We know that $g_s$ must depend on a species-dependent parameter, i.e. $\varepsilon_s$, since the result of the integral (the LHS) must be independent of $\varepsilon_s$, in a similar vein to the discussion in Section \ref{sec:Channell}.
\item[Dimensionless variables and `time' is fixed:]
\begin{equation}
\tilde{P}_j\left(\text{sgn}(q_s)\delta_s A_j \right)=:\mathcal{J}_{js}(\tilde{A}_{j};\delta_s)=\frac{1}{\sqrt{2\pi}}\int_{-\infty}^\infty e^{-(\tilde{p}_{js}-\tilde{A}_j)^2/2} \bar{g}_{js}(\tilde{p}_{js};\delta_s) d\tilde{p}_{js},\label{eq:analogy}
\end{equation}
with ${\rm sgn}(q_e)=-1$ and ${\rm sgn}(q_i)=1$, and for 
\begin{eqnarray}
\delta_s&=&\frac{m_sv_{\text{th},s}}{eB_0L},\nonumber\\
\tilde{p}_{js}&=&\frac{p_{js}}{m_sv_{\text{th},s}},\nonumber\\
\tilde{A}_j&=&\frac{A_j}{B_0L}\nonumber\\
\bar{g}_{js}(\tilde{p}_{js};\delta_s)&=&g_{js}(p_{js};v_{\text{th},s})\nonumber .
\end{eqnarray}
The species-dependent magnetisation parameter, $\delta_s$ (e.g. see \cite{Fitzpatrickbook}), is defined by 
\[
\delta_s=\frac{r_{Ls}}{L}=\frac{m_{s}v_{\text{th},s}}{eB_0L}.
\]
It is the ratio of the thermal Larmor radius, $r_{Ls}=v_{\text{th},s}/|\Omega_s|$, to the characteristic length scale of the system, $L$. The gyrofrequency of particle species $s$ is $\Omega_s=q_{s}B_0/m_{s}$. The magnetisation parameter is also known as the fundamental ordering parameter in gyrokinetic theory (see \cite{Howes-2006, Abel-2013} for example). In particle orbit theory, $\delta_s\ll 1$ implies that a guiding centre approximation will be applicable for that species, e.g. see \citet{Northrop-1961} and Section \ref{sec:guiding}.

This second interpretation is depicted by Equation (\ref{eq:analogy}) and once again casts the inverse problem in direct comparison with the Weiertrass transform, making a correspondence between space in the heat equation, $x$, to $\tilde{A}$ in our inverse problem. But in this case the `time' is evaluated at $t=1/2$. Since the LHS of Equation (\ref{eq:analogy}) is now a function of $\delta_s$, we have included the parametric dependence on $\delta_s$ in $\bar{g}_s$. 
\end{description}

\subsubsection{The `backwards heat equation'}
The first interpretation is the one that I believe carries the most meaning for the problem considered in this thesis. Since the integral transform described by Equation (\ref{eq:Channell}) must leave the LHS independent of species-dependent parameters, it makes sense that the transformed function, $g_s$, is not directly analogous to an initial condition. If $g_s$ was an `initial condition' and independent of `time', $\varepsilon_s$, then the outcome of the evolution (transform) would surely give a time-dependent solution, i.e. one that depends on $\varepsilon_s$. But that is not what occurs. The correct analogy is to view the $g_s$ function not as an initial condition, but \emph{as the `heat distribution' $\varepsilon_s$ `seconds' ago, such that when evolved (transformed) forward by $\varepsilon_s$ `seconds', the resultant `heat distribution' is $P_{zz}$}. In that sense, we are considering the heat equation but with a \emph{final condition}, as opposed to an initial condition: the `backwards heat equation'. Similar topics are discussed in the `backwards uniqueness of the heat equation' (see e.g. \cite{Evansbook}).

\subsection{Formal inversion of the Weierstrass transform}
Formally, the operator for the inverse Weierstrass transform is $\; e^{-D^2}$, with D the differential operator and the exponential suitably interpreted, see \citet{Eddington-1913, Widder-1954} for two different interpretations of this operator. 

A second, and perhaps more computationally `practical' method employs Hermite polynomials, see \citet{Bilodeau-1962}. The Weierstrass transform of the $n^{th}$ Hermite polynomial $H_n(y/2)$ at $t=1$ is $x^n$. Hence if one knows the coefficients of the Maclaurin expansion of $u(x,1)$ in Equation (\ref{eq:Weierstrass}), 
\[
u(x,1)=\sum_{j=0}^\infty\eta_jx^j,
\]
then the Weierstrass transform can immediately be inverted to obtain the formal expansion
\begin{equation}
u_0(y)=\sum_{j=0}^\infty\eta_jH_j\left(y/2\right).\label{eq:inverseweier}
\end{equation}
For this method to be useful in our problem, the pressure function must have a Maclaurin expansion that is convergent over all $(A_x,A_y)$ space. Then, its coefficients of expansion must `allow' the Hermite series to converge. 

\subsubsection{Formal inversion of our problem}
The following discussion applies to pressure functions of both summative and multiplicative form, with Maclaurin expansion representations (convergent over all $(A_x,A_y)$ space) given by
\begin{equation}
\tilde{P}_1(A_x)=\sum_{m=0}^\infty a_m\left(\frac{A_x}{B_0L}\right)^m,\;\;\tilde{P}_2(A_y)=\sum_{n=0}^\infty b_n\left(\frac{A_y}{B_0L}\right)^n,\label{eq:Pmacro}
\end{equation}
with $B_0$ and $L$ the characteristic magnetic field strength and spatial scale respectively. In line with the discussion on inversion of the Weierstrass transform in Section \ref{sec:formal}, we solve for $g_{s}$ functions represented by the following expansions
\begin{eqnarray}
g_{1s}(p_{xs};v_{\text{th},s})&=&\displaystyle\sum_{m=0}^\infty C_{ms}H_m\left(\frac{p_{xs}}{\sqrt{2}m_{s}v_{\text{th},s}}\right),\label{eq:gform1}\\
g_{2s}(p_{ys};v_{\text{th},s})&=&\displaystyle\sum_{n=0}^\infty D_{ns}H_n\left(\frac{p_{ys}}{\sqrt{2}m_{s}v_{\text{th},s}}\right),\label{eq:gform2}
\end{eqnarray}
with currently unknown species-dependent coefficients $C_{ms}$ and $D_{ns}$. We cannot simply `read off' the coefficients of expansion as in Equation (\ref{eq:inverseweier}), since our integral equations are not quite in the `perfect form' of Equation (\ref{eq:Weierstrass}). Upon computing the integrals of Equations (\ref{eq:p1tog1}) and (\ref{eq:p2tog2}) with the above expansions for $g_{s}$, we have
\begin{equation}
\tilde{P}_1(A_x)=\displaystyle\sum_{m=0}^\infty \left(\frac{\sqrt{2}q_{s}}{m_{s}v_{\text{th},s}}\right)^{m}C_{ms}\,A_x^{m},\;\;\tilde{P}_2(A_y)=\displaystyle\sum_{n=0}^\infty \left(\frac{\sqrt{2}q_{s}}{m_{s}v_{\text{th},s}}\right)^{n}D_{ns}\,A_y^n.\label{eq:Pmicro}
\end{equation}
This result appears species dependent. However, to ensure self-consistency with quasineutrality ($n_i(A_x,A_y)=n_e(A_x,A_y)$) - as in \citet{Channell-1976, Harrison-2009PRL, Wilson-2011} - we have to fix the pressure function to be species independent. It clearly must also match with the pressure function that maintains equilibrium with the prescribed magnetic field. The conditions to be derived here are critical for making a link between the macroscopic description of the equilibrium structure  with the microscopic one of particles. These requirements imply - by the matching of Equations (\ref{eq:Pmacro}) and (\ref{eq:Pmicro}) - that
\begin{eqnarray}
\left(\frac{\sqrt{2}q_{s}}{m_{s}v_{\text{th},s}}\right)^{m}C_{ms} & = & \left(\frac{1}{B_0L}\right)^{m}a_{m}\implies C_{ms}={\rm sgn}(q_{s})^m\left(\frac{\delta_s}{\sqrt{2}}\right)^{m}a_{m}, \label{eq:ccoeff}\\
\left(\frac{\sqrt{2}q_{s}}{m_{s}v_{\text{th},s}}\right)^{n}D_{ns} & =  &\left(\frac{1}{B_0L}\right)^nb_n\implies D_{ns}={\rm sgn}(q_{s})^n\left(\frac{\delta_s}{\sqrt{2}}\right)^{n}b_{n}\label{eq:dcoeff}.
\end{eqnarray}

\section{Mathematical validity of the method}\label{sec:convergence}

\subsection{Convergence of the Hermite expansion}
Here we find a sufficient condition that, when satisfied, guarantees that the Hermite series representations in (\ref{eq:gform1}) and (\ref{eq:gform2}) converge. This provides some answers to questions on the convergence of Hermite Polynomial representations of Vlasov equilibria dating back to \citet{Hewett-1976}, and implicit in the work of e.g. \citet{Alpers-1969,Channell-1976,Suzuki-2008}.
\begin{theorem}\label{thm:HermiteConvergence}
Consider a Maclaurin expansion of the form
\begin{equation}
\tilde{P}_j(A_{j})=\sum_{m=0}^\infty a_m\left(\frac{A_{j}}{B_0L}\right)^m
\end{equation}
that is convergent for all $A_j$. Then for $\varepsilon_s=m_{s}^2v_{\text{th},s}^2/2$ the function $g_{js}$, calculated in the inverse problem defined by the association
\begin{equation}
\tilde{P}_{j}(A_j):=\tilde{P}_{\text{INT},j}(A_j)=\frac{1}{\sqrt{4\pi\varepsilon_s}}\int_{-\infty}^{\infty} \; e^{-(p_{js}-q_{s}A_j)^2/(4\varepsilon_s)}g_{js}(p_{js};v_{\text{th},s})dp_{js}.\label{eq:sample}
\end{equation}
of the form
\begin{equation}
g_{js}(p_{js};v_{\text{th},s})=\sum_{m=0}^\infty a_m\,{\rm sgn}(q_{s})^m\left(\frac{\delta_s}{\sqrt{2}}\right)^mH_{m}\left(\frac{p_{js}}{\sqrt{2}m_{s}v_{\text{th},s}}\right)\label{eq:gj}
\end{equation}
converges for all $p_{js}$, provided 
\begin{equation}
\lim_{m\to\infty}\sqrt{m}\left|\frac{a_{m+1}}{a_m}\right|<1/\delta_s,
\end{equation}
in the case of a series composed of both even- and odd-order terms, or
\begin{equation}
\lim_{m\to\infty}\,m\,\left|\frac{a_{2m+2}}{a_{2m}}\right|<1/(2\delta_s^2),\hspace{3mm}\lim_{m\to\infty}\,m\,\left|\frac{a_{2m+3}}{a_{2m+1}}\right|<1/(2\delta_s^2),
\end{equation}
in the case of a series composed only of even-, or odd-order terms, respectively.
\end{theorem}
\begin{proof}
For a series composed of even- and odd-order terms, we have that
\begin{equation}
g_{js}(p_{js};v_{\text{th},s})=\sum_{m=0}^{\infty}a_{m}\,{\rm sgn}(q_{s})^m\left(\frac{\delta_s}{\sqrt{2}}\right)^mH_{m}\left(\frac{p_{js}}{\sqrt{2}m_{s}v_{\text{th},s}}\right)\label{eq:g2}.
\end{equation}
An upper bound on Hermite polynomials (see e.g.  \cite{Sansonebook}) is provided by the identity 
\begin{equation}
| H_{j}(x)|<k\sqrt{j!}2^{j/2}\exp\left(x^2/2\right)\; \mbox{\rm s.t.\ } \; k=1.086435\, .\label{eq:hermbound}
\end{equation}
This upper bound implies that
\[
0<|a_m|\left(\frac{\delta_s}{\sqrt{2}}\right)^m\bigg|H_m\left(\frac{p_{s}}{\sqrt{2}m_{s}v_{\text{th},s}}\right)\bigg|<k|a_m|\delta_s ^m\sqrt{m!}\exp\left(\frac{p_{js}^2}{4m_{s}^2v_{\text{th},s}^2}\right).
\]
Let us now compose a series of the upper bounds,
\[
g_{js,\text{upper}}=k\exp\left(\frac{p_{js}^2}{4m_{s}^2v_{\text{th},s}^2}\right)\sum_{m=0}^\infty |a_m|\delta_s ^m\sqrt{m!}.
\]
By the use of the ratio test \citep{Bartle}, a sufficient condition for convergence of $g_{js,\text{upper}}$ is found by
\begin{eqnarray}
\lim_{m\to\infty}\Bigg|\frac{a_{m+1}}{a_m}\Bigg|\sqrt{m+1}<1/\delta_s,\nonumber\\
\implies \lim_{m\to\infty}\Bigg|\frac{a_{m+1}}{a_m}\Bigg|\sqrt{m}<1/\delta_s, \label{eq:criterion}
\end{eqnarray}
for a given $\delta_s\in (0,\infty)$. If the $a_m$ satisfy the criteria in Equation (\ref{eq:criterion}) then $g_{js,\text{upper}}$ is a convergent series, and hence by the comparison test \citep{Bartle}, 
\[
g_{js,\text{absolute}}=\sum_{m=0}^{\infty}|a_{m}|\left(\frac{\delta_s}{\sqrt{2}}\right)^m\bigg|H_{m}\left(\frac{p_{js}}{\sqrt{2}m_{s}v_{\text{th},s}}\right)\bigg|,
\]
is a convergent series. This then implies that 
\[
\sum_{m=0}^{\infty}a_{m}{\rm sgn}(q_{s})^m\left(\frac{\delta_s}{\sqrt{2}}\right)^mH_{m}\left(\frac{p_{js}}{\sqrt{2}m_{s}v_{\text{th},s}}\right)(=g_{js}(p_{js};v_{\text{th},s}))
\]
is an absolutely convergent series, and in turn a convergent series. We can now confirm that $g_{js}(p_{js};v_{\text{th},s})$ is a convergent series \citep{Bartle}.

An analogous argument holds for those series with only even or odd order terms, with the ratio test giving
\begin{equation}
\lim_{m\to\infty}\Bigg|\frac{a_{2m+2}}{a_{2m}}\Bigg|m<1/(2\delta_s^2), \hspace{2mm}\text{or}\hspace{2mm}\lim_{m\to\infty}\Bigg|\frac{a_{2m+3}}{a_{2m+1}}\Bigg|m<1/(2\delta_s^2), \label{eq:criterion2}
\end{equation}
respectively. By the same argument as above, the comparison test implies that if the condition of (\ref{eq:criterion2}) is satisfied, that since the series composed of upper bounds will converge, so must $g_{js}(p_{js})$.
\end{proof}

\subsubsection{Decay rate of the coefficients}
In order to get a better understanding of the meaning of Theorem \ref{thm:HermiteConvergence}, it is instructive to recapitulate the results in a continuous setting. One could imagine the modulus of the coefficients, $|a_{m}|$, as a subset of the codomain of a continuous function of the independent variable $m$,
\begin{eqnarray}
&&|a_m|, \, m=0,1,2,....\nonumber\\
\to &&a=a(m),\, m\in[0,\infty), \hspace{3mm} \text{s.t.}\hspace{3mm}a(0)=|a_0|, a(1)=|a_1| ... \, .\nonumber
\end{eqnarray}
In this case, we require 
\[
a(m)=\mathcal{O}(a_{\text{u}}(m)),\hspace{3mm}\text{s.t.}\hspace{3mm}a_u(m)=(\delta_s^2 m)^{-m/2},
\]
since the function $a_{\text{u}}$ satisfies the restrictions of Equations (\ref{eq:criterion}) and (\ref{eq:criterion2}), i.e
\begin{eqnarray}
\mathcal{O}\left(\bigg|\frac{a_{\text{u}}(m+1)}{a_{\text{u}}(m)}\bigg|\right)&=&\frac{1}{\delta_s\sqrt{m}},\nonumber\\
\mathcal{O}\left(\bigg|\frac{a_{\text{u}}(2m+2)}{a_{\text{u}}(2m)}\bigg|\right)&=&\frac{1}{2\delta_s^2m},\nonumber\\
\mathcal{O}\left(\frac{a_{\text{u}}(2m+3)}{a_{\text{u}}(2m+1)}\bigg|\right)&=&\frac{1}{2\delta_s^2m}.\nonumber
\end{eqnarray}
Hence the modulus of the coefficients, $|a_m|$ must `fall below' the graph of $(\delta_s^2 m)^{-m/2}$ for large $m$, and depicted in Figure \ref{fig:decay}.
\begin{figure}
    \hspace{1cm}
    \includegraphics[width=0.7\textwidth]{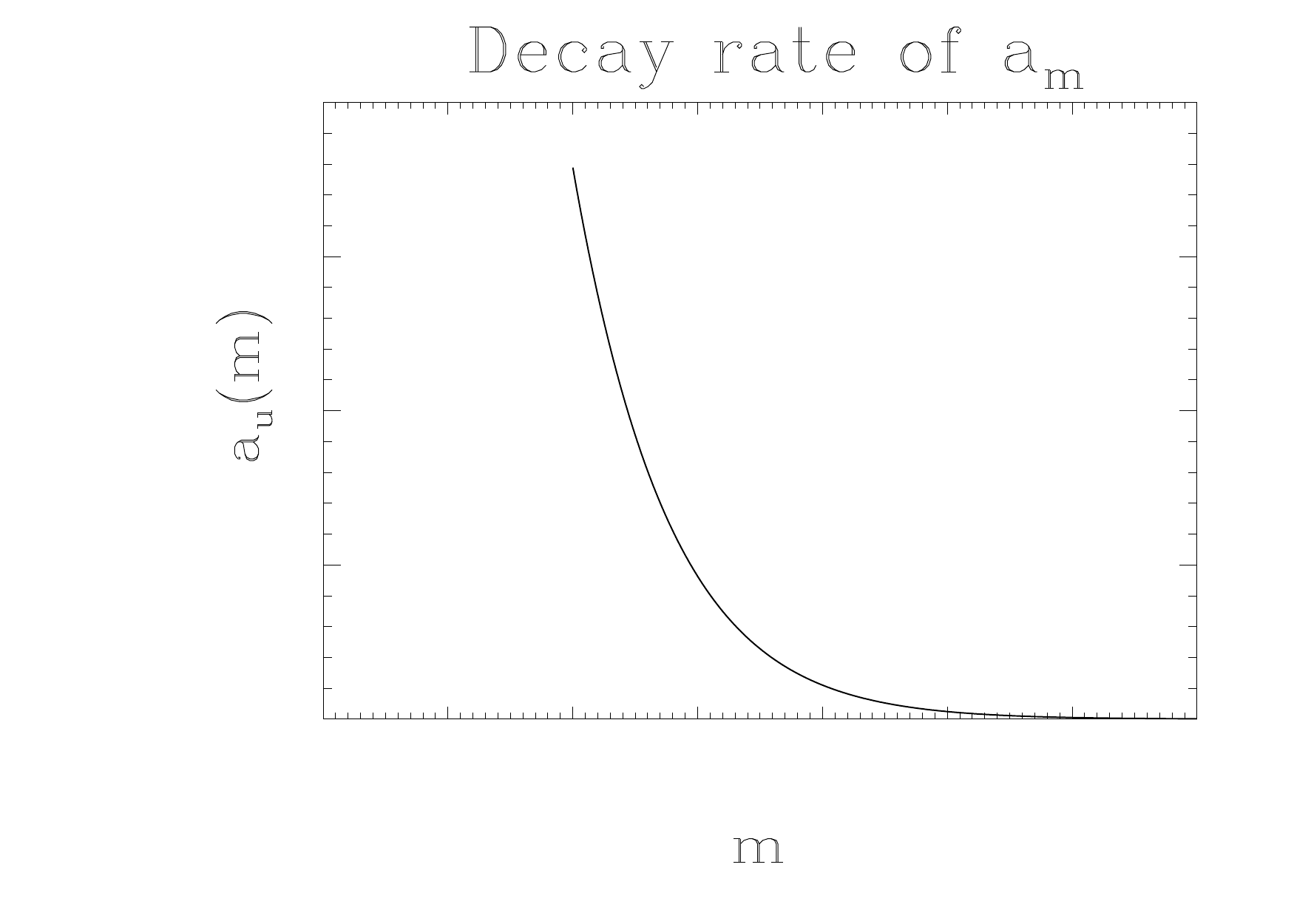}
        \caption{{\small Theorem \ref{thm:HermiteConvergence} states that if the modulus of the coefficients, $|a_m|$, `fall below' the graph of $(\delta_s^2 m)^{-m/2}$ as $m\to\infty$, then the Hermite series of Equation (\ref{eq:gj}) will converge. } }
        \label{fig:decay}
        \end{figure}

\subsubsection{The existence of velocity moments}\label{sec:boundedness}
Once the convergence of the Hermite polynomial is established, then one can begin to consider the boundedness of the DF, and the existence of velocity moments. If $g_{s}(p_{s};v_{\text{th},s})$ is a convergent series, then by using Equation (\ref{eq:hermbound}) we see that
\[
|g_{js}(p_{js};v_{\text{th},s})|<\mathcal{L}_{js}\exp\left(\frac{p_{js}^2}{4m_{s}^2v_{\text{th},s}^2}\right) \hspace{3mm}\forall\, p_{js},
\]
and for $\mathcal{L}_{js}$ a finite, positive constant, independent of space and momentum. By now using the form of the DF from Equation (\ref{eq:F_form}) and the separability conditions of Equation (\ref{eq:gsep}), we see that 
\begin{eqnarray}
&&|f_s|<\exp\left[-(p_{xs}-q_sA_x)^2/(2m_s^2v_{\text{th},s}^2)-(p_{ys}-q_sA_y)^2/(2m_s^2v_{\text{th},s}^2)-v_z^2/(2v_{\text{th},s}^2)\right]\nonumber\\
&&\times\left(\mathcal{L}_{xs}e^{p_{xs}^2/(4m_{s}^2v_{\text{th},s}^2)}+\mathcal{L}_{ys}e^{p_{ys}^2/(4m_{s}^2v_{\text{th},s}^2)}    \right),\nonumber
\end{eqnarray}
in the case of additive separability, or
\begin{eqnarray}
&&|f_s|<\exp\left[-(p_{xs}-q_sA_x)^2/(2m_s^2v_{\text{th},s}^2)-(p_{ys}-q_sA_y)^2/(2m_s^2v_{\text{th},s}^2)-v_z^2/(2v_{\text{th},s}^2)\right]\nonumber\\
&&\times\left(\mathcal{L}_{xs}\mathcal{L}_{ys}e^{p_{xs}^2/(4m_{s}^2v_{\text{th},s}^2)}e^{p_{ys}^2/(4m_{s}^2v_{\text{th},s}^2)}    \right),\nonumber
\end{eqnarray}
in the case of multiplicative separability. In either case, we see that boundedness in momentum space (and hence velocity space) is guaranteed. The reasoning is as follows. Since $p_{js}=m_sv_j+q_sA_j$, the arguments of the exponentials scale like
\begin{equation}
\exp\left(-\frac{v_j^2}{4v_{th,s}^2}\right), \label{eq:dfscale}
\end{equation}
in $v_j$ velocity space. There is also a spatial dependence in the argument of the exponential, through $A_j(z)$, but this does not affect the velocity moment at a given $z$ value. The scaling described by Expression (\ref{eq:dfscale}) not only ensures boundedness, but guarantees that velocity moments of all order exist, since
\[
\bigg|\int_{-\infty}^\infty  v^k e^{-v^2/(4v_{th,s})^2} dv \bigg|\, <\,\infty \, \forall \, k \, \in \, 0,1,2, ...
\]

\subsubsection{Summary}
In this Section we have shown that for a DF of the form 
\begin{equation}
f_s(H_s,p_{xs},p_{ys})=\frac{n_{0s}}{(\sqrt{2\pi}v_{\text{th},s})^3}e^{-\beta_sH_s}g_s(p_{xs},p_{ys};v_{\text{th},s}),\nonumber
\end{equation}
with
\begin{equation}
g_{s}=g_{1s}(p_{xs};v_{\text{th},s})+g_{2s}(p_{ys};v_{\text{th},s})\;{\rm or}\;g_{s}=g_{1s}(p_{xs};v_{\text{th},s})g_{2s}(p_{ys};v_{\text{th},s})\nonumber,
\end{equation}
and 
\begin{eqnarray}
g_{1s}(p_{xs};v_{\text{th},s})&=&\sum_{m=0}^\infty a_m\,{\rm sgn}(q_{s})^m\left(\frac{\delta_s}{\sqrt{2}}\right)^mH_{m}\left(\frac{p_{xs}}{\sqrt{2}m_{s}v_{\text{th},s}}\right)\nonumber,\\
g_{2s}(p_{ys};v_{\text{th},s})&=&\sum_{n=0}^\infty b_n\,{\rm sgn}(q_{s})^n\left(\frac{\delta_s}{\sqrt{2}}\right)^nH_{n}\left(\frac{p_{ys}}{\sqrt{2}m_{s}v_{\text{th},s}}\right)\nonumber,
\end{eqnarray}
the $g_s$ functions are convergent provided the criteria on the growth rates of the coefficients of expansion from Theorem \ref{thm:HermiteConvergence} are satisfied:
\begin{equation}
\lim_{m\to\infty}\sqrt{m}\left|\frac{a_{m+1}}{a_m}\right|<1/\delta_s,\nonumber
\end{equation}
in the case of a series composed of both even- and odd-order terms, or
\begin{equation}
\lim_{m\to\infty}\,m\,\left|\frac{a_{2m+2}}{a_{2m}}\right|<1/(2\delta_s^2),\hspace{3mm}\lim_{m\to\infty}\,m\,\left|\frac{a_{2m+3}}{a_{2m+1}}\right|<1/(2\delta_s^2),\nonumber
\end{equation}
in the case of a series composed only of even-, or odd-order terms, respectively, and this in turn implies that velocity moments of the DF of all order exist.

\section{Non-negativity of the Hermite expansion}\label{sec:non-neg}
In this Section, we consider the non-negativity of the Hermite series representation of $g_{s}$ -- given by Equations (\ref{eq:gform1}) and (\ref{eq:gform2}) -- and hence  positivity of the DF. As such this Section responds to questions on the positivity of DF representation by Hermite polynomials raised by \citet{Abraham-Shrauner-1968, Hewett-1976}, and implicit in the work of e.g. \citet{Alpers-1969, Channell-1976, Suzuki-2008}. 

\subsection{Possible negativity of the Hermite expansion}\label{sec:negativity}
For an example of a $g_{js}$ function that is not necessarily always positive despite the pressure function being positive, consider a pressure function (e.g. from \cite{Channell-1976}) that is quadratic in the vector potential. In our notation, the pressure function considered by Channell is
\[
\tilde{P}=\frac{1}{2}\left(a_0+a_2\left(\frac{A_x}{B_0L}\right)^2\right)+\frac{1}{2}\left(a_0+a_2\left(\frac{A_y}{B_0L}\right)^2\right),
\]
for $a_0,a_2>0$. The resultant $g_{s}$ function is of the form
\[
g_{s}\propto \frac{1}{2}\left[a_0+a_2\left(\frac{\delta_s}{\sqrt{2}}\right)^2H_2\left(\frac{p_{xs}}{\sqrt{2}m_{s}v_{\text{th},s}}\right)\right]+\frac{1}{2}\left[a_0+a_2\left(\frac{\delta_s}{\sqrt{2}}\right)^2H_2\left(\frac{p_{ys}}{\sqrt{2}m_{s}v_{\text{th},s}}\right)\right].
\]
Once these Hermite polynomials are expanded, and by substituting $p_{xs}=p_{ys}=0$, we see that positivity of $g_s$ is -- for given values of $a_0$ and $a_2$ -- contingent on the size of $\delta_s$, 
\begin{eqnarray}
g_s(0,0)=a_0-a_2\delta_{s}^{2},\nonumber\\
\therefore g_s(0,0)\ge 0\implies \delta_s^2\le \frac{a_0}{a_2}.\nonumber
\end{eqnarray}
However, there is not necessarily anything `special' about the origin, as compared to other points in momentum-space. For example, consideration of the pressure function
\[
\tilde{P}_j=\left(a_0+a_2\left(\frac{A_j}{B_0L}\right)^2 +a_4\left(\frac{A_j}{B_0L}\right)^4       \right),
\]
gives a $g_{js}$ function that can, for given values of $a_0,a_2,a_4$ and for $\delta_s$ sufficiently large, be positive at $p_{js}=0$, and negative at some other points. 

It is worth considering how a $g_{js}$ function that is negative for some $p_{js}$ can transform in the manner of (\ref{eq:p1tog1}) and (\ref{eq:p2tog2}) to give a positive $\tilde{P}_j(A_j)$. One might expect that for certain values of $A_j$ such that the Gaussian 
\[
e^{-(p_{js}-q_{s}A_j)^2/(4\varepsilon_s)}
\]
is centred on the region in $p_{js}$ space for which $g_{js}$ is negative, that a negative value of $\tilde{P}_j(A_j)$ could be the result.

Essentially, the Gaussian will only `successfully sample' a negative region of $g_{js}$ to give a negative value of $\tilde{P}_j(A_j)$ if the Gaussian is narrow enough -- for a given value of $\varepsilon_s$ --  to `resolve' a negative patch of $g_{js}$. In other words, if the Gaussian is too broad, it won't `see' the negative patches of $g_{js}$, and hence $\tilde{P}_j(A_j)$ will be positive. Hence the non-negativity of $\tilde{P}_j(A_j)$ is a restriction on the possible shape of $g_{js}$, and how that shape must scale with $\varepsilon_s$.

\subsection{Detailed arguments}
When considering the non-negativity of the Hermite expansion, it is instructive to rewrite (\ref{eq:sample}) in the form
\begin{equation}
\sum_{n=0}^\infty a_n\left(\text{sgn}(q_s)\delta_s\tilde{A}_j\right)^n=\frac{1}{\sqrt{2\pi}}\int_{-\infty}^\infty e^{-(\tilde{p}_{js}-\tilde{A}_j)^2/2}\bar{g}_{js}(\tilde{p}_{js};\delta_s)d\tilde{p}_{js},\label{eq:newsample}
\end{equation}
by using the following associations
\[
\tilde{A}_j=\frac{A_j}{B_0L},\hspace{3mm}\tilde{p}_{js}=\frac{p_s}{\sqrt{2\varepsilon_s}},\hspace{3mm}g_{js}(p_{js};\varepsilon_s)=\bar{g}_{js}(\tilde{p}_{js};\delta_s).
\]
The formal solution as an expansion in Hermite polynomials can be written as
\begin{equation}
\bar{g}_{js}(\tilde{p}_{js};\delta_s)=\sum_{n=0}^\infty a_n\text{sgn}(q_s)^n\left(\frac{\delta_s}{\sqrt{2}}\right)^nH_n\left(\frac{\tilde{p}_{js}}{\sqrt{2}}\right).\label{eq:gnorm}
\end{equation}
We shall assume that the right-hand side of (\ref{eq:gnorm}) represents a differentiable function. Note that the Gaussian in (\ref{eq:newsample}) is of fixed width $2\sqrt{2}$ (defined at $1/e$), in contrast to the Gaussian of variable width defined in (\ref{eq:sample}). 

\subsubsection{Boundedness below zero of the Hermite expansion}
If the Hermite series satisfies the condition in Theorem 1 then it is convergent, so Equation (\ref{eq:hermbound}) gives
\begin{equation}
\left|\bar{g}_{js}(\tilde{p}_{js};\delta_s)\right|<\mathcal{L}_{js}e^{\tilde{p}_{js}^2/4}      
\nonumber
\end{equation}
for some finite and positive $\mathcal{L}_{js}$, determined by the sum of the (possibly infinite) series. Note that these bounds automatically imply integrability of $f_s$ since as can be seen from Equation (\ref{eq:newsample}), for some finite $L^\prime>0$, we have that $\left|\bar{g}_{js}(\tilde{p}_{js};\delta_s)\right|<L^\prime e^{\tilde{p}_{js}^2/2}$  implies integrability, which is a less strict condition. 

The bounds on $\bar{g}_{js}$ given above demonstrate that $\bar{g}_{js}$ can not tend to $\pm \infty$ for finite $\tilde{p}_{js}$. Hence, if it reaches $-\infty$ at all, it can only do so as $|\tilde{p}_{js}|\to\infty$.  We argue however that the positivity of the pressure prevents the possibility of $\bar{g}_{js}$ being without a finite lower bound. The heuristic reasoning is as follows: the expression on the RHS of Equation (\ref{eq:newsample}) treats -- in the language of the heat/diffusion equation -- the $\bar{g}_{js}$ function as the initial condition for a temperature/density distribution on an infinite 1-D line, and the left-hand side represents the distribution at some finite time later on. Were $\bar{g}_{js}$ to be unbounded from below, this would imply for our problem that a smooth `temperature/density' distribution that is initially unbounded from below could, in some finite time, evolve into a distribution that has a positive and finite lower bound. This seems entirely unphysical since this would imply that an infinite negative `sink' of heat/mass would somehow be `filled in' above zero level in a finite time. 

\subsubsection{Proofs and arguments by contradiction}
Here we give some technical remarks that support our claim that $\bar{g}_{js}$ (and hence $g_{js}$) is bounded below, using an argument by contradiction. First of all consider a smooth $\bar{g}_{js}$ function that is unbounded from below in positive momentum space. Then, depending on the number and nature of stationary points, either 
\begin{itemize}
\item Case 1: There will be some $\tilde{p}_{j0,s}$ such that $\bar{g}_{js}<c<0$  for all $\tilde{p}_{js}>\tilde{p}_{j0,s}$. This is a trivial statement if $\bar{g}_{js}$ has only a finite number of stationary points, whereas in the case of an infinite number of stationary points, all maxima of $\bar{g}_{js}$ for $\tilde{p}_{js}>\tilde{p}_{j0,s}$ must be `away' from zero by a finite amount.
\item Case 2:  In this case the (infinite number of) maxima either can rise above zero, or tend to zero from below in a limiting fashion.
\end{itemize}

If $\bar{g}_{js}$ is of the type described in Case 1, then we can create an `envelope' $g_{\text{env},j}$ for $\bar{g}_{js}$ such that $g_{\text{env},j}>\bar{g}_{js}$ for all $\tilde{p}_{js}$. The envelope we choose is 
\begin{equation}
g_{\text{env},j}=
\begin{cases}
\mathcal{L}_{js}e^{\tilde{p}_{js}^2/4},\text{   for   } \tilde{p}_{js}\le \tilde{p}_{j0,s},\\
c\text{   for   } \tilde{p}_{js} > \tilde{p}_{j0,s}.
\end{cases}
\end{equation}
The $\mathcal{L}_{js}e^{\tilde{p}_{js}^2/4}$ form for the profile is chosen because this represents the absolute upper bound for our convergent Hermite expansions, at a given $\tilde{p}_{js}$ as seen from Equation (\ref{eq:hermbound}). If we then substitute the $g_{\text{env},j}$ function for $\bar{g}_{js}$ in Equation (\ref{eq:newsample}) the integrals give combinations of error functions, 
\begin{eqnarray}
&&\frac{1}{\sqrt{2\pi}}\int_{-\infty}^\infty e^{-(\tilde{p}_{js}-\tilde{A}_j)^2/2}\bar{g}_{\text{env},j}d\tilde{p}_{js}=\nonumber\\
&&\frac{\mathcal{L}_{js}e^{\tilde{A}_j^2/2}}{\sqrt{2}} \left( \text{erf} \left(\frac{\tilde{p}_{j0,s} - 2 \tilde{A}_j}{2}\right) + 1\right)+\frac{c}{2}\left(   \text{erf}\left( \frac{\tilde{A}_j - \tilde{p}_{j0,s}}{\sqrt{2}}\right) + 1\right)\nonumber
\end{eqnarray}
from which it is seen that one obtains a negative result, i.e. $c$, as $\tilde{A}_j\to\infty$. This is a contradiction since the left-hand side of Equation (\ref{eq:newsample}) is positive for all $\tilde{A}_j$. Hence we can discount the $\bar{g}_{js}$ functions of the variety described in Case 1, as we have a contradiction.

Case 2 is less simple to treat. The fact that there exists an infinite number of local minima and that the infimum of $\bar{g}_{js}$ is $-\infty$ implies that there exists an infinite sequence of points in momentum space, $\mathcal{S}_p=\{ \tilde{p}_k\, : \, k=1,2,3 ...\}$, that are local minima of $\bar{g}_{js}$, such that $\bar{g}_{js}(\tilde{p}_{k+1})<\bar{g}_{js}(\tilde{p}_{k})$. Essentially there are an infinite number of minima `lower than the previous one'. For sufficiently large $k=l$, we have that the magnitude of the minima is much greater than the width of the Gaussian, i.e.
\[
|\bar{g}_{js}(\tilde{p}_{l})| \gg 2\sqrt{2}.
\]
In this case the only way that the sampling of $\bar{g}_{js}$ described by Equation (\ref{eq:newsample}) could give a positive result for a Gaussian centred on the minima is if $\bar{g}_{js}$ rapidly grew to become sufficiently positive, in order to compensate the negative contribution from the minimum and its local vicinity. However, this seems to be at odds with the condition that $\bar{g}_{js}$ is smooth, since the function would have to rise in this manner for ever more negative values of the minima (and hence rise ever more quickly) as $k\to \infty$. We claim that this can not happen, and hence we discount the $\bar{g}_{js}$ functions of the variety described in Case 2.

Since there is no asymmetry in momentum-space in this problem, the arguments above hold just as well for for a $\bar{g}_{js}$ function that is unbounded from below in negative momentum space. It should be clear to see that if $\bar{g}_{js}$ can not be unbounded from below in either the positive or negative direction, then it can not be unbounded in both directions either.

\subsubsection{Behaviour with respect to the magnetisation}
If $\bar{g}_{js}$ (and hence $g_{js}$) is indeed bounded below then that means that one can always add a finite constant to $g_{js}$ to make it positive, should the lower bound be known. However this constant contribution would directly correspond to raising the pressure (through the zeroth order Maclaurin coefficient $a_0$). 

If we wish to consider a pressure function that is `fixed', then we have a fixed $a_0$, and so it is not immediately obvious whether or not we can obtain a $g_{js}$ that is positive over all momentum space. We have already seen some examples in Section \ref{sec:negativity} for which the sign of $g_{js}$ depended on the value of $\delta_s$. 

Consider $\bar{g}_{js}$ evaluated at some particular value of $\tilde{p}_{js}$. We see from Equation (\ref{eq:gnorm}) that positivity requires
\[
a_0+c_1\delta_s+c_2\delta_s^2+...>0,
\]
for $c_1,c_2,...$ finite constants. We also know that $a_0>0$ since $P_j(0)>0$, i.e. the pressure is positive. This clearly demonstrates that positivity of $g_{js}$ places some restriction on possible values of $\delta_s$.

Let us now suppose that for a given value of $\delta_s$, that there exists some regions in $\tilde{p}_{js}$ space where $\bar{g}_s<0$. Our claim that $\bar{g}_{js}$ has a finite lower bound, combined with the expression in Equation (\ref{eq:gnorm}) implies that the $\bar{g}_s$ function is bounded below by a finite constant of the form $a_0+\delta_s \mathcal{M}$, with
\[
\mathcal{M}=\frac{1}{\sqrt{2}}\inf_{\tilde{p}_{js}}\sum_{n=1}^\infty a_n\text{sgn}(q_s)^n\left(\frac{\delta_s}{\sqrt{2}}\right)^{n-1}H_n\left(\frac{\tilde{p}_{js}}{\sqrt{2}}\right),
\]
and finite (and for $\inf$ the \emph{infimum}, i.e. the greatest lower bound). By letting $\delta_s\to 0$ we see that $\bar{g}_{js}$ will converge uniformly to $a_0$, with
\[
\lim_{\delta_s\to 0}\bar{g}_{js}(\tilde{p}_{js},\delta_s)=a_0>0.
\]
Hence, there must have existed some critical value of $\delta_s=\delta_c$ such that for all $\delta_s<\delta_c$ we have positivity of $\bar{g}_{js}$. Note that if the negative patches of $\bar{g}_{js}$ do not exist for any $\delta_s$, then trivially $\delta_c=\infty$ as a special case.

\subsection{Summary}
To summarise, we claim  -- provided $g_s$ is differentiable and convergent -- that for values of the magnetisation parameter $\delta_s$ less than some critical value $\delta_c$, according to $0<\delta_s<\delta_c\le \infty$, $g_s$ is positive for any positive pressure function. The crucial step in this work was to prove/argue that $g_s$ is bounded from below by a constant for all values of the momenta.

We have in fact proven this result for the class of $g_s$ functions for which the number of stationary points is finite, or if infinite for which the stationary points are `away' from zero by a finite amount. We have also presented arguments based on the differentiability of $g_s$, that support this result for other classes of $g_s$ function.

\section{Illustrative case of the use of the method: \\
correspondence with the Fourier transform method}\label{Sec:Channell}
Here we give an example of the use of the solution method to a pressure function that was first discussed in \citet{Channell-1976}. In that paper, Channell actually solved the inverse problem by the Fourier transform method, and showed that the solution was valid given certain restrictions on the parameters. We tackle the problem via the Hermite Polynomial method, and find that for the resultant DF to be convergent, we require exactly the same restrictions as Channell. This parity between the validity of the two methods is reassuring, and implies that the necessary restrictions on the parameters are in a sense `method independent', and are the result of fundamental restrictions on the inversion of Weierstrass transformations. 

The magnetic field considered by Channell can not be given analytically, but is of the form 
\begin{equation}
\boldsymbol{B}=(B_{x}(z)\,,0\,,0),\hspace{3mm}\text{s.t.}\hspace{3mm}B_{x}(-\infty)=B_0,\label{eq:channellnumeric1}
\end{equation}
and self-consistent with a pressure function 
\begin{equation}
P_{zz}=P_0e^{-\gamma \tilde{A}_{y}^{2}}\label{eq:channellnumeric2}
\end{equation}
for $P_0,B_0$ and $L$ characteristic values of the pressure, magnetic field and length scales, $\tilde{A}_{y}=A_{y}/(B_{0}L)$ and $\gamma>0$ dimensionless. The magnetic field and self-consistent number density profiles for this equilibrium are shown in Figure \ref{fig:Channell}, reproduced from \citet{Channell-1976}.
\begin{figure}
    \centering
    \includegraphics[width=0.7\textwidth]{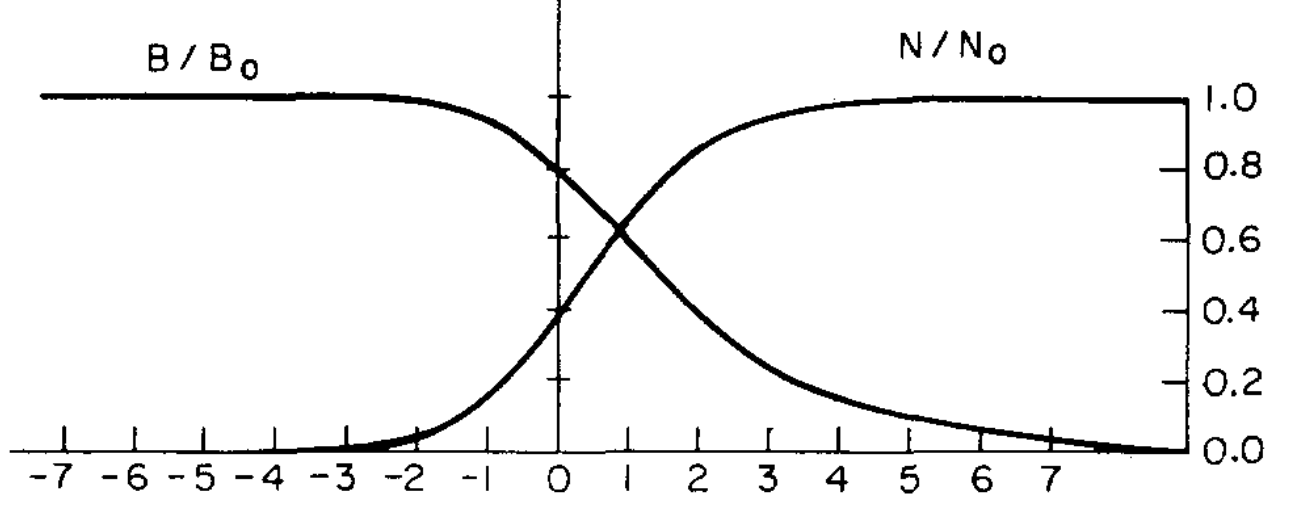}
        \caption{\small A figure from \citet{Channell-1976} that displays the magnetic field and number density consistent with Equations (\ref{eq:channellnumeric1}) and (\ref{eq:channellnumeric2}). {\bf Image Copyright:} AIP, \href{http://aip.scitation.org}{\emph{Physics of Fluids}} {\bf 19}, 1541, (1976), copyright (1976), (reproduced with permission).}
        \label{fig:Channell}
        \end{figure}
Note that the $\gamma$ used by Channell has dimensions equivalent to $1/(B_{0}^{2}L^{2})$. We can now write the details of the inversion. The equation we must solve, for a DF given by 
\[
f_{s}=\frac{n_{0}}{(\sqrt{2\pi}v_{\text{th},s})^{3}}e^{-\beta_{s}H_{s}}g_{s}(p_{ys};v_{\text{th},s})
\]  
is
\[
P_0\exp \left(-\gamma \frac{A_{y}^{2}}{B_{0}^{2}L^{2}}\right)=\frac{n_0 (\beta_e+\beta_i)}{\beta_e \beta_i} \frac{1}{\sqrt{2\pi}m_{s}v_{\text{th},s}}\int_{-\infty}^\infty {\rm e}^{-(p_{ys}-q_{s}A_y)^2/(2m_{s}^2v_{\text{th},s}^2)}g_{s}dp_{ys}.
\]
We can immediately formally invert this equation as per the methods described in this Chapter, given the Maclaurin expansion of the pressure 
\[ 
P_{zz}=P_0\sum_{m=0}^\infty a_{2m}\left(\frac{A_y}{B_0L}\right)^{2m}\hspace{3mm}\text{s. t.} \hspace{3mm}a_{2m}=\frac{(-1)^m\gamma ^m}{m!},
\]
to give
\[
g_{s}(p_{ys})=\sum_{m=0}^\infty \left(\frac{\delta_s}{\sqrt{2}}\right)^{2m}a_{2m}H_{2m}\left(\frac{p_{ys}}{\sqrt{2}m_{s}v_{\text{th},s}}\right).
\]
Let us turn to the question of convergence. Theorem 1 states that if
\[
\lim_{m\to\infty} m\left|\frac{a_{2m+2}}{a_{2m}}\right|<1/(2\delta_s^2), 
\]
then the $g_{s}$ function is convergent. This is readily seen to imply that $\gamma$ must satisfy
\[
\gamma< \frac{1}{2\delta_s^2},
\]
for the Hermite series representation of  $g_{s}$ to be convergent. This condition is exactly equivalent to the one derived by Channell (Equation (28) in the paper). Note that now that we have established convergence for particular $\gamma$, then boundedness results follow as per the results in Section \ref{sec:boundedness}. One more question remains, namely how does the $g_{s}$ function derived compare to the Gaussian $g_{s}(p_{ys})$ function derived by Channell 
\[
g_{s}\propto e^{-4\gamma^2\delta_s^4p_{ys}^2/(1-4\gamma^2\delta_s^4)}
\]
(in our notation) using the method of Fourier transforms? In fact, one can see by setting $y=0$ in Mehler's Hermite Polynomial formula \citep{Watson-1933}
\[
\frac{1}{\sqrt{1-\rho^2}}\exp\left[\frac{2xy\rho-(x^2+y^2)\rho^2    }{1-\rho^2}   \right]=\sum_{n=0}^{\infty}\frac{\rho ^n}{2^n n!} H_n(x)H_n(y)  , 
 \]
and using 
\begin{equation}
H_m(0)=
\begin{cases}
0 & \text{if}\hspace{3mm} m\hspace{3mm} \text{is odd}, \\
  (-1)^{m/2}m!/(m/2)!      & \text{if}\hspace{3mm} m\hspace{3mm} \text{is even},
\end{cases}\nonumber
\end{equation}
(see \cite{Gradshteyn} for example), that the Hermite series represents a Gaussian function in the range $| \rho |<1$. This is equivalent to the condition derived above for convergence, $\gamma< 1/(2\delta_s^2)$. Hence, we have shown that for this specific example - solvable by using both Hermite polynomials and Fourier transforms - the two methods used to solve the inverse problem give equivalent functions with equivalent ranges of mathematical validity.

\section{Summary}
The primary result of this chapter is the rigorous generalisation of a solution method that exactly solves the `inverse problem' in 1-D collisionless equilibria, for a certain class of equilibria. Specifically, given a pressure function, $P_{zz}(A_x,A_y)$, of a separable form, neutral equilibrium DFs can be calculated that reproduce the prescribed macroscopic equilibrium, provided $P_{zz}$ satisfies certain conditions on the coefficients of its (convergent) Maclaurin expansion, and is itself positive. 

The DF has the form of a Maxwellian modified by a function $g_s$, itself represented by -- possibly infinite -- series of Hermite polynomials in the canonical momenta. It is crucial that these series are convergent and positive for the solution to be meaningful. A sufficient condition was derived for convergence of the DF by elementary means, namely the ratio test, with the result a restriction on the rate of decay of the Maclaurin coefficients of $P_{zz}$. For DFs that are written as an expansion in Hermite polynomials, multiplied by a stationary Maxwellian, we have demonstrated that the necessary boundedness results follow.

We also argue that for such a pressure function that is also positive, that the Hermite series representation of the modification to the Maxwellian is positive, for sufficiently low values of the magnetisation parameter, i.e. lower than some critical value. This was actually proven for a certain class of $g_s$ functions, and differentiability of $g_s$ was assumed. It would be interesting in the future to investigate whether this critical value of the magnetisation parameter can be determined. It is also desirable that the result is proven for all reasonable function classes.

We have demonstrated the application of the solution method in Section \ref{Sec:Channell}. This particular example already has a known solution and range of validity in parameter space, obtained by a Fourier transform method in \citep{Channell-1976}. We obtain a solution with an alternate representation using the Hermite Polynomial method. The Hermite series obtained is shown to be equivalent to the representation obtained by Channell, and to have the exact same range of validity in parameter space. It is not clear if this equivalence between solutions obtained by the two different methods is true in general. Our problem is somewhat analagous to the heat/diffusion equation, and in that `language' the question of the equivalence of solutions is related to the `backwards uniqueness of the heat equation' (see e.g. \citep{Evansbook}). The degree of similarity between our problem and the one described by Evans, and its implications, are left for future investigations.

Also, whilst we have assumed that the pressure is separable (either summatively or multiplicatively), the method should be adaptable in the `obvious way' for pressures that are a `superposition' of the two types. Interesting further work would be to see if the method can be adapted to work for pressure functions that are non-separable, i.e. of the form
\begin{equation*}
P_{zz}=\sum_{m,n}\mathcal{C}_{mn}\left(\frac{A_x}{B_0L}\right)^{m}\left(\frac{A_y}{B_0L}\right)^{n}.
\end{equation*}

\null\newpage

\chapter{One-dimensional nonlinear force-free current sheets} 

\label{Sheets} 

 \epigraph{\emph{We have to keep an eye on the electrons.}}{\textit{Thomas Neukirch}}
 
\noindent Much of the work in this chapter is drawn from \citet{Allanson-2015POP, Allanson-2016JPP}.

\section{Preamble}
In this chapter we present new exact collisionless equilibria for a 1D nonlinear force-free magnetic field, namely the force-free Harris sheet. In contrast to previous solutions \citep{Harrison-2009PRL, Wilson-2011, Abraham-Shrauner-2013, Kolotkov-2015}, the solutions that we present allow the plasma beta ($\beta_{pl}$) to take any value, and crucially values below unity for the first time. In the derivations of the equilibrium DFs it is found that the most typical approach of Fourier Transforms can not be applied, and so we use expansions in Hermite polynomials, making use of the techniques developed in Chapter \ref{Vlasov}. Using the convergence criteria developed therein, we verify that the Hermite expansion representation of the DFs are convergent for all parameter values. As shown in Chapter \ref{Vlasov}, this also implies boundedness, and the existence of velocity moments of all orders. 

Despite the proven analytic convergence, initial difficulties in attaining numerical convergence mean that plots of the DF can be presented for the plasma beta only modestly below unity. In the effort to model equilibria with much lower values of the plasma, we use a new gauge for the vector potential, and calculate the DF consistent with this gauge, confirming the properties of convergence velocity moments. This new gauge makes attaining numerical convergence possible for lower values of the plasma beta, and we present results for $\beta_{pl}=0.05.$

\section{Introduction}
Force-free equilibria, with fields defined by
\begin{eqnarray}
\boldsymbol{j}\times\boldsymbol{B}=\frac{1}{\mu_0}(\nabla\times\boldsymbol{B})\times\boldsymbol{B}&=&\boldsymbol{0},  \label{eq:force-free}
\end{eqnarray}
are of particular relevance to the solar corona (e.g. see \cite{Priest-2000, Wiegelmann} and Figure \ref{fig:beta}); current sheets in the Earth's magnetotail (e.g. \cite{Vasko-2014,Petrukovich-2015}), the Earth's magnetopause (e.g. \cite{Panov-2011}) and in the Jovian magnetotail (e.g. \cite{Artemyev-2014}); other astrophysical plasmas (e.g. \cite{Marshbook}); scrape-off layer currents in tokamaks (e.g. \cite{Fitzpatrick-2007}); and `Taylor-relaxed' magnetic fields in fusion experiments (e.g. \cite{Taylor-1974,Taylor-1986}). Equation (\ref{eq:force-free}) implies that the current density is everywhere-parallel to the magnetic field;
\begin{equation}
\mu_0\boldsymbol{j}=\alpha(\boldsymbol{x})\boldsymbol{B},\label{eq:alpha}
\end{equation}
or zero in the case of potential fields, and with $\alpha$ the \emph{force-free parameter}. If $\nabla\alpha\neq 0$ then the force-free field is nonlinear, whereas a constant $\alpha$ corresponds to a linear force-free field. Note that 
\[
\nabla\cdot(\nabla\times\boldsymbol{B})=0\implies\boldsymbol{B}\cdot\nabla\alpha =0,
\]
and hence $\alpha$ is a constant along a magnetic field line, but will vary from field line to field line in the case of nonlinear force-free fields. Extensive discussions of force-free fields are given in \citet{Sakurai-1989, Marshbook}. 

\subsection{Force-free equilibria and the plasma beta}
Equation (\ref{eq:force-free}) presents the force-free condition in purely geometric terms, i.e. an equilibrium force-free magnetic field has field lines obeying certain geometrical constraints, such that a particular combination of spatial derivatives vanish. In order to gain some physical insight, consider a generic plasma equilibrium (in the absence of a gravitational potential),
\begin{equation}
\nabla \cdot\boldsymbol{P}=\sigma\boldsymbol{E}+\boldsymbol{j}\times\boldsymbol{B}.\label{eq:sheetbalance}
\end{equation}
Next, normalise each of the quantities according to
\begin{eqnarray}
\nabla\cdot\boldsymbol{P}&=&\frac{p_0}{L_P}\,\tilde{\nabla}\cdot\tilde{\boldsymbol{P}},\nonumber\\
\sigma\boldsymbol{E}&=&\sigma_0E_0\,\tilde{\sigma}\tilde{\boldsymbol{E}},\nonumber\\
\boldsymbol{j}\times\boldsymbol{B}&=&\frac{B_0^2}{\mu_0L_B}\,\tilde{\boldsymbol{j}}\times\tilde{\boldsymbol{B}},\nonumber
\end{eqnarray}
for $L_P,L_B$ typical values of the length scales associated with the pressure and magnetic fields respectively; and with $p_0, \sigma_0,E_0, B_0$ typical values of the thermal pressure, charge density, electric and magnetic field respectively. Furthermore, since $\boldsymbol{E}=-\nabla\phi$ and $\nabla^{2}\phi=-\sigma/\epsilon_0$, we define
\begin{eqnarray}
\sigma_0&=&-\frac{\epsilon_0\phi_0}{L_\phi^2},\nonumber\\
E_0&=&-\frac{\phi_0}{L_\phi},\nonumber\\
\text{s.t.}\hspace{3mm}\phi_0&=&\frac{k_BT_0}{e},\nonumber
\end{eqnarray}
for $T_0$ a typical value of the temperature, and $L_\phi$ the length scale associated with the scalar potential. Written in dimensionless form, the force balance equation (Equation (\ref{eq:sheetbalance})) can now be written as
\begin{equation}
\frac{\beta_{pl}}{2}L_B\left[\frac{1}{L_P}\,\tilde{\nabla}\cdot\tilde{\boldsymbol{P}}-\frac{1}{L_\phi}   \frac{\lambda_D^2}{L_\phi^2}\,\tilde{\sigma}\tilde{\boldsymbol{E}}\right]=\tilde{\boldsymbol{j}}\times\tilde{\boldsymbol{B}},\nonumber
\end{equation}
for $\beta_{pl}=2\mu_0p_0/B_0^2$ the plasma beta, and $\lambda_D=\sqrt{\epsilon_0k_BT_0/(n_0e^2)}$ the Debye radius. Note that we have made use of $p_0=n_0k_BT_0$. This equation demonstrates that - in principle - 
\[
\beta_{pl}\ll1 \cancel{\iff}\boldsymbol{j}\times\boldsymbol{B}=\boldsymbol{0},
\]
for $\cancel{\iff}$ to read as `not equivalent', i.e. force free equilibria need not necessarily have a vanishing plasma beta, or vice versa. However, we see that for a quasineutral plasma in which $\epsilon=\lambda_D/L_\phi \ll1$, the second term on the LHS is - for a given value of $\beta_{pl}$ - almost certainly of a lower order than the first term on the LHS, due to the $\epsilon^2$ dependence. 
\begin{figure}
    \centering
        \includegraphics[width=0.7\textwidth]{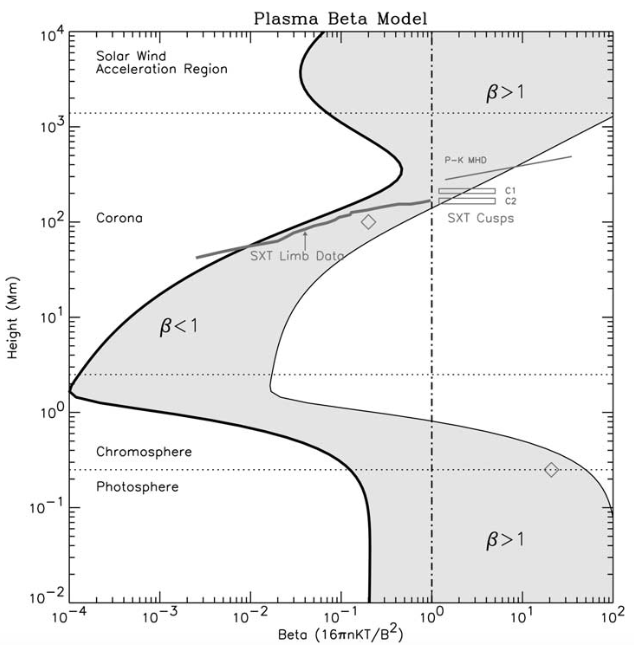}
    \caption{{\small A figure from \citet{Gary-2001} that displays a representative $\beta_{pl}$ model over solar active regions, derived from a range of sources. {\bf Image Copyright:} Springer, \href{http://link.springer.com/journal/11207}{\emph{Solar Physics}} {\bf 203}, 1, (October 2001), pp. 71-86., copyright (2001), (reproduced with permission). }}
 \label{fig:beta}     
\end{figure}
Hence we see that for a quasineutral equilibrium
\begin{equation}
\frac{\beta_{pl}}{2}\,\frac{L_B}{L_P}\,\tilde{\nabla}\cdot\tilde{\boldsymbol{P}}=\tilde{\boldsymbol{j}}\times\tilde{\boldsymbol{B}},\nonumber
\end{equation}
and so it would now seem fair to say that $\beta_{pl}\ll 1 \iff \boldsymbol{j}\times\boldsymbol{B}=\boldsymbol{0}$ for a quasineutral equilibrium, unless the thermal pressure varies with respect to very fine length scales. For a similar discussion to the above, including the gravitational acceleration but not the electric field, see \citet{Neukirch-2005}. Figure \ref{fig:beta} is reproduced from \citet{Gary-2001} and shows a model for the plasma beta in the solar atmosphere, compiled from observational data. The figure demonstrates that $\beta_{pl}$ can take sub-unity and vanishing values in the solar chromosphere and the corona, as well as values above one (contrary to the most typical assumptions). As such, much of the solar corona magnetic field is modelled as force-free \citep{Wiegelmann}.

\subsection{1D force-free equilibria}
1D magnetic fields can be represented without loss of generality by 
\begin{equation}
\boldsymbol{B}=\left(B_x(z),B_y(z),0\right)=\left(-\frac{dA_y}{dz},\frac{dA_x}{dz},0\right).
\end{equation}
The force-free condition then implies that
\begin{equation}
\boldsymbol{j}\times\boldsymbol{B}=\boldsymbol{0}\implies\frac{d}{dz}\left(\frac{B_x^2}{2\mu_0}+\frac{B_y^2}{2\mu_0}\right)=0,\label{eq:bconstant}
\end{equation}
and hence the magnetic field is necessarily of uniform magnitude. Considering the equation of motion for a quasineutral plasma, now given by
\[
\frac{d}{dz}\left( P_{zz}+\frac{B^2}{2\mu_0}\right)=0,
\]
we see that the thermal pressure is also of constant magnitude,
\begin{equation}
\frac{d}{dz}P_{zz}=0\implies P_{zz}={\rm const.}
\end{equation}

As demonstrated in Section \ref{sec:inversekey}, the (assumed) existence of a VM equilibrium implies - through the dependence of the DF on the constants of motion - that the pressure tensor is a function of the vector and scalar potentials. Hence, we see that for a quasineutral plasma in which $\phi_{qn}=\phi(A_x,A_y)$, the force-free equilibrium fields correspond to a  \emph{trajectory}, $\boldsymbol{A}_{ff}(z)=(A_x(z),A_y(z),\phi_{qn}(A_x(z),A_y(z)))$, that is itself a \emph{contour};
\begin{equation} 
\frac{d}{dz}P_{zz}(A_x(z),A_y(z))=0, \label{eq:Contour}
\end{equation}
of the \emph{potential}, $P_{zz}$ \citep{Harrison-2009POP,Harrison-2009PRL}. As such, the construction of a $P_{zz}$ function that satisfies Equation (\ref{eq:Contour}), given some $(A_x(z),A_y(z)$ is the first step in the inverse method for 1D force-free equilibria. 

In fact, Equation (\ref{eq:Contour}) compactly defines the entire macroscopic problem, since 
\begin{equation}
\frac{\partial P_{zz}}{\partial\boldsymbol{A}}=\boldsymbol{j},\label{eq:pressurecurrent}
\end{equation}
implies that
\begin{eqnarray}
\frac{d}{dz}P_{zz}(A_x(z),A_y(z))=\underbrace{\frac{\partial P_{zz}}{\partial A_x}}_{j_x}\underbrace{\frac{dA_x}{dz}}_{B_y}+\underbrace{\frac{\partial P_{zz}}{\partial A_y}}_{j_y}\underbrace{\frac{dA_y}{dz}}_{-B_x}=0,\nonumber\\
=j_xB_y-j_yB_x,\nonumber\\
=(\boldsymbol{j}\times\boldsymbol{B})_z.
\end{eqnarray}
This demonstrates that - in a 1D quasineutral plasma - the existence of a VM equilibrium that is self-consistent with a spatially uniform pressure tensor directly implies that the magnetic field is force-free.  

\subsubsection{Pressure tensor transformation theory}\label{sec:pressuretrans}
The inverse problem is not only non-unique regarding the form of the DF for a particular macroscopic equilibrium (as discussed in Section \ref{sec:inverseapproach}), but also for the form of $P_{zz}(A_x,A_y)$ for a particular magnetic field. Given a specific force-free magnetic field, i.e. a specific $\left(A_x,A_y\right)$, and a known $P_{zz}$ that satisfies Equations (\ref{eq:Contour}) and (\ref{eq:pressurecurrent}), one can construct infinitely many new $\bar{P}_{zz}$ functions that also satisfy them;
\begin{equation}
\bar{P}_{zz}=\frac{1}{\psi^\prime(P_{ff})}\psi(P_{zz}),\label{eq:Ptrans}
\end{equation} 
for differentiable and non-constant $\psi$, provided the LHS is positive, and for which the value of $P_{zz}$ evaluated on the force-free contour, $\boldsymbol{A}_{ff}$, is the constant, $P_{ff}$ \citep{Harrison-2009POP}. These $\bar{P}_{zz}$ functions maintain a force-free equilibrium with the \emph{same magnetic field} as $P_{zz}$, since
\begin{eqnarray}
\frac{\partial \bar{P}_{zz}}{\partial\boldsymbol{A}}\bigg|_{\boldsymbol{A}_{ff}}=\frac{1}{\psi^\prime(P_{ff})}\frac{\partial{\psi}}{\partial P_{zz}}\frac{\partial P_{zz}}{\partial \boldsymbol{A}}\bigg|_{\boldsymbol{A}_{ff}}=\frac{\partial P_{zz}}{\partial \boldsymbol{A}}\bigg|_{\boldsymbol{A}_{ff}}=\boldsymbol{j}_{ff},\nonumber
\end{eqnarray}
for $\boldsymbol{j}_{ff}$ the current density derived from $\boldsymbol{A}_{ff}$.

\section{Force-free current sheet VM equilibria}
Since current sheets are extremely important for reconnection studies (e.g. see \cite{Priest-2000}), and it is appropriate in many circumstances to model the magnetic field as force-free, a natural step is to construct VM equilibria for force-free current sheets. The archetypal 1D current sheet structure used to model reconnection is the Harris sheet \citep{Harris-1962} (see Section \ref{sec:currentsheets}), 
\[
\boldsymbol{B}=B_0(\tanh(z/L),0,0),
\]
for which an exact VM equilibrium DF is well-known. However, the Harris sheet has $\boldsymbol{j}\perp\boldsymbol{B}$ and hence is not force-free, with thermal pressure gradients balancing those of the magnetic pressure. It is possible to approximate a force-free field with the addition of a uniform guide field
\[
\boldsymbol{B}=(B_{x0}\tanh(z/L), B_{y0},0),
\]
for $B_{x0},B_{y0}$ constants. This magnetic field configuration is frequently chosen as the initial condition in PIC simulations of magnetic reconnection (e.g. see \cite{Pritchett-2004}), and the VM equilibrium is easily implemented since it is the same as that for the Harris sheet (Equation (\ref{eq:HarrisDF})).

In principle, this magnetic field does approach a force free configuration for $B_{y0}\gg B_{x0}$, since $\boldsymbol{j}$ is approximately parallel to $\boldsymbol{B}$. However, the current density, $j_y$, is completely independent of the magnitude of the guide field, and so it is quite unlike an exact force-free field, for which the field-aligned current is related to the shear of the magnetic field. The equilibrium force balance is still maintained by the balance between gradients in the thermal pressure and the magnetic pressure,
\[ 
\frac{1}{2\mu_0}\left(B_{x0}^2\tanh^2\left(\frac{z}{L}\right)+B_{y0}^2\right),
\]
unlike for an exact force-free field. Finally, the addition of the guide field adds no extra free energy to the system \citep{Harrison-thesis}. Hence it is of value to consider VM equilibria self-consistent with exact force-free magnetic fields because of their distinct physical nature, with one motivation in mind to see how these differences affect the magnetic reconnection process.       

As discussed in e.g. \citet{Bobrova-2001,Vekstein-2002}, Equation (\ref{eq:bconstant}) implies that a 1D force-free field can be written without loss of generality as
\begin{equation}
\boldsymbol{B}(z)=B_0(\cos (S(z)),\sin (S(z)),0), \label{eq:forcefreerepresentation}
\end{equation}
where $S(z) =\int\alpha (z)dz$, for $\alpha$ defined in Equation (\ref{eq:alpha}). 1D linear force-free fields then, necessarily, have $S(z)$ as a linear function of $z$, i.e. $S_0z+S_1$. As a result, Equation (\ref{eq:forcefreerepresentation}) then implies that that the magnetic field configuration for linear force-free fields will be periodic in the $z$ direction, and hence there will be an infinite sequence of current sheet structures,
\[
\boldsymbol{j}=\frac{-B_0S_0}{\mu_0L} ( \sin (S_0z+S_1), \cos(S_0z+S_1)   ,0  ).
\]
Figure \ref{fig:lff} displays the magnetic field from Equation (\ref{eq:forcefreerepresentation}), and its current density, for $S(z)=z-\pi/2$. 
\begin{figure}
    \centering
   
    \begin{subfigure}[b]{0.45\textwidth}
        \includegraphics[width=\textwidth]{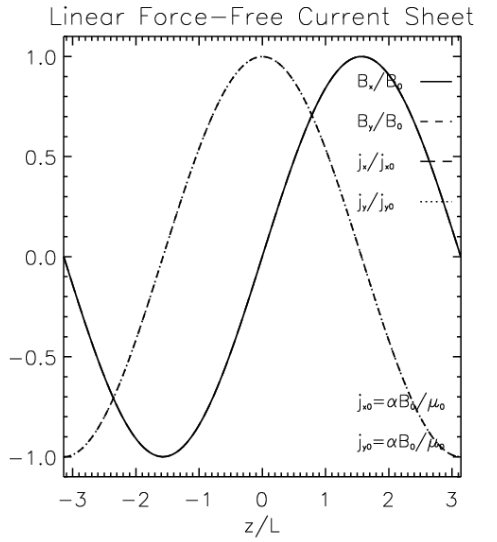}
        \caption{}
        \label{fig:lff1}
    \end{subfigure}
    \begin{subfigure}[b]{0.45\textwidth}
        \includegraphics[width=\textwidth]{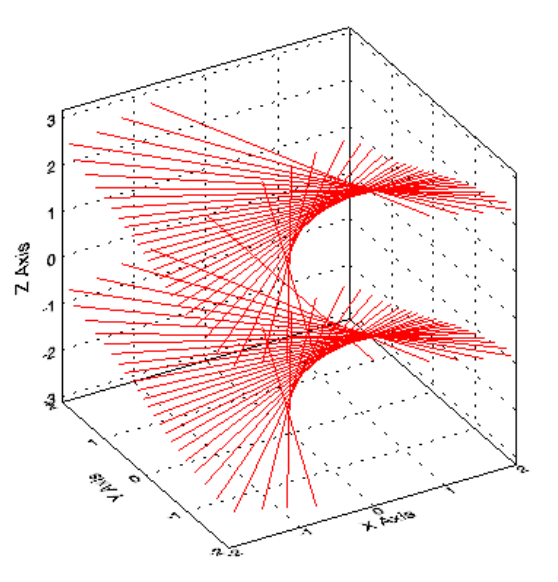}
        \caption{ }
        \label{fig:lff2}
    \end{subfigure}

    \caption{\small Figure \ref{fig:lff1} shows the magnetic field and current density components for a linear force-free field with $S(z)=z-\pi/2$. Figure \ref{fig:lff2} shows the magnetic field lines.  {\bf Images copyright:} M.G. Harrison's PhD thesis \citep{Harrison-thesis}, (reproduced with permission).  }      \label{fig:lff}
\end{figure}

In contrast to linear force-free fields, nonlinear force-free fields admit - in principle - all reasonable varieties of differentiable $S(z)$ functions, and hence are able to describe single, localised and intense current sheet structures.

\subsection{Known VM equilibria for force-free magnetic fields}\label{sec:known}
The first VM equilibria self-consistent with linear force-free fields were found approximately fifty years ago, \citep{Moratz-1966, Sestero-1967}, with further examples of equilibria in \citet{Channell-1976, Bobrova-1979,Correa-Restrepo-1993, Attico-1999,Bobrova-2001} (note that \citet{Channell-1976,Attico-1999} don't actually make the connection to force-free fields, but write down DFs that are self-consistent with such fields). A limited number of PIC studies with exact VM equilibria for linear force-free fields as initial conditions have been conducted in \citet{Bobrova-2001, Li-2003, Nishimura-2003, Sakai-2004, Bowers-2007, Harrison-thesis}.

In contrast, exact VM equilibria for nonlinear force-free fields were only discovered in \citet{Harrison-2009PRL} (see also \cite{Neukirch-2009}), with subsequent solutions in \citet{Wilson-2011, Abraham-Shrauner-2013, Kolotkov-2015}, and `nearly force-free' equilibria in \citet{Artemyev-2011}. As a result, the investigations of the linear and nonlinear dynamics of such configurations are at an early stage \citep{Harrison-thesis, Wilson-thesis, Wilson-2017}, with the first fully kinetic simulations of collisionless reconnection with an initial condition that is an exact Vlasov solution for a nonlinear force-free field conducted by \citep{Wilson-2016}, and using the DF derived by \citet{Harrison-2009PRL}.

\subsubsection{The force-free Harris sheet}
The nonlinear force-free VM equilibrium solutions derived by \citet{Harrison-2009PRL, Wilson-2011, Kolotkov-2015} are self-consistent with the force-free Harris sheet (FFHS), defined by
\begin{eqnarray}
\boldsymbol{B}&=&B_0\left(\text{tanh}\left(z/L\right), \text{sech}\left(z/L\right),0\right), \label{eq:FFHSmagfield}\\
\boldsymbol{j}&=&\frac{B_0}{\mu_0L}\frac{1}{\cosh(z/L)}\left(\text{tanh}\left(z/L\right), \text{sech}\left(z/L\right),0\right),\label{eq:FFHScurrent}\\
P_{zz}(z)&=&P_T-\frac{B_0^2}{2\mu_0}=\text{const.}
\end{eqnarray}
with $L$ the width of the current sheet, $B_0$ the constant magnitude of the magnetic field, $\alpha (z)= L^{-1}\text{sech}(z/L)$ and $P_T$ the total pressure. The magnetic field and current density for the FFHS are displayed in Figure \ref{fig:ffhs}. 

The DF found by \citet{Abraham-Shrauner-2013} is consistent with magnetic fields more general than the FFHS, described by Jacobi elliptic functions, 
\[
\boldsymbol{B}=B_0 \left(\text{sn} \left(\frac{z}{L},k\right), \text{cn} \left(\frac{z}{L},k\right), 0\right),
\]
with $\text{sn}$ and $\text{cn}$ doubly periodic generalisations of the trigonometric functions. The parameter $k$ is a real number such that as $k\to 0$, $\text{sn}\to\sin$ and $\text{cn}\to\cos$; whereas for $k\to 1$, $\text{sn}\to\tanh$ and $\text{cn}\to\text{sech}$. As such the FFHS is a special case, as is the linear force-free case when $k\to 0$. We also note work on `nearly' force-free equilibria \citep{Artemyev-2011}, with the FFHS modified by adding a small $B_z$ component. 
\begin{figure}
    \centering
     \begin{subfigure}[b]{0.45\textwidth}
        \includegraphics[width=\textwidth]{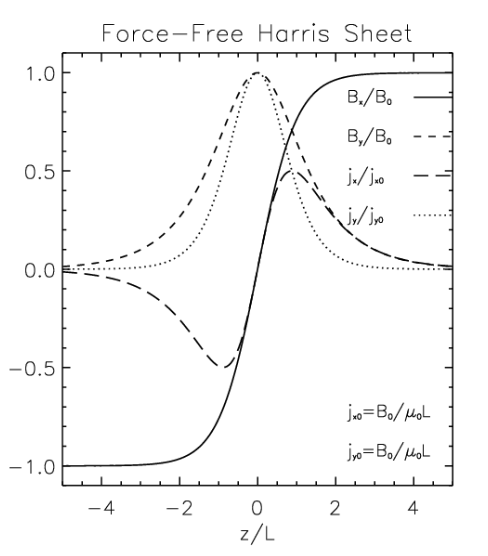}
        \caption{}
        \label{fig:ffhs1}
    \end{subfigure}
    \begin{subfigure}[b]{0.45\textwidth}
        \includegraphics[width=\textwidth]{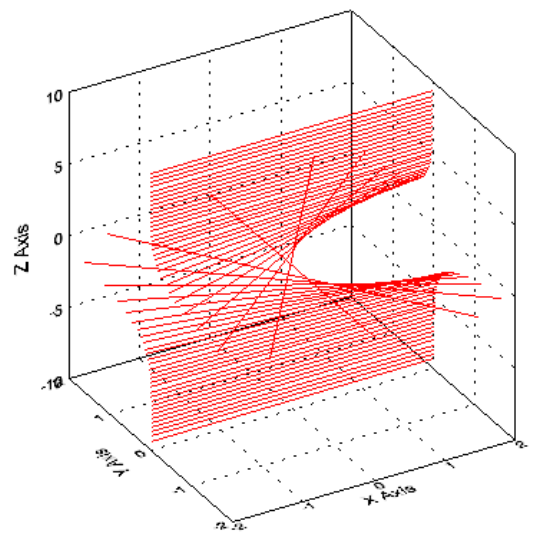}
        \caption{ }
        \label{fig:ffhs2}
    \end{subfigure}

    \caption{\small Figure \ref{fig:ffhs1} shows the magnetic field and current density components for the FFHS. Figure \ref{fig:ffhs2} shows the magnetic field lines.  {\bf Images copyright:} M.G. Harrison's PhD thesis \citep{Harrison-thesis}, (reproduced with permission).  }      \label{fig:ffhs}
\end{figure}

As demonstrated by \citet{Harrison-2009PRL, Neukirch-2009}, the assumption of summative separability for $P_{zz}$ (the first option in Equation (\ref{eq:pform})), determines the components of the pressure according to 
\begin{eqnarray}
P_{zz}(A_x,A_y)+\frac{B_0^2}{2\mu_0}=P_T,\nonumber\\
P_1(A_x)+\frac{1}{2\mu_0}B_y^2(A_x)=P_{T1},\;\;P_2(A_y)+\frac{1}{2\mu_0}B_x^2(A_y)=P_{T2}\label{eq:p1p2balance}
\end{eqnarray}
for $P_{T1},P_{T2}$ constants such that $P_{T1}+P_{T2}=P_{T}$ is the total pressure. We choose to write $B_x$ and $B_y$ as functions of $A_y$ and $A_x$ since $B_x=-dA_y/dz$ and $B_y=dA_x/dz$. In the `particle in a potential' analogy - as discussed in Section \ref{sec:pseudo} - this corresponds to writing $v_x=v_x(x(t))$, and $v_y=v_y(y(t))$.

The expressions in Equation (\ref{eq:p1p2balance}) can now be used as the left-hand side of the integral Equations (\ref{eq:p1tog1}) and (\ref{eq:p2tog2}), and one could attempt to invert the Weierstrass transforms. They were used by \citet{Harrison-2009PRL} to derive a summative pressure for the FFHS. The gauge chosen for the magnetic field was
\begin{equation}
\boldsymbol{A}=B_0L\left(2\arctan\left(\exp\left(\frac{z}{L}\right)\right),\ln\left(\text{sech}\left(\frac{z}{L}\right)\right),0\right),\label{eq:Agauge1}
\end{equation}
and as such the pressure tensor is given by
\begin{equation}
P_{zz}=\frac{B_0^2}{2\mu_0}\left[\frac{1}{2}\text{cos}\left(\frac{2A_x}{B_0L}\right)+\text{exp}\left(\frac{2A_y}{B_0L}\right)+b\right]. \label{eq:P_Harrison}
\end{equation}
The constant $b>1/2$ contributes to a `background' pressure consistent with a Maxwellian distribution, required for positivity. Figure \ref{fig:Harrisonpzz} shows the $P_{zz}$ function as defined by Equation (\ref{eq:P_Harrison}), with the overlaid contour delineating the `path' followed by $\boldsymbol{A}=(A_x(z),A_y(z),0)$ according to Equation (\ref{eq:Agauge1}), and such that $dP_{zz}/dz=0$.
\begin{figure}
    \centering
        \includegraphics[width=0.7\textwidth]{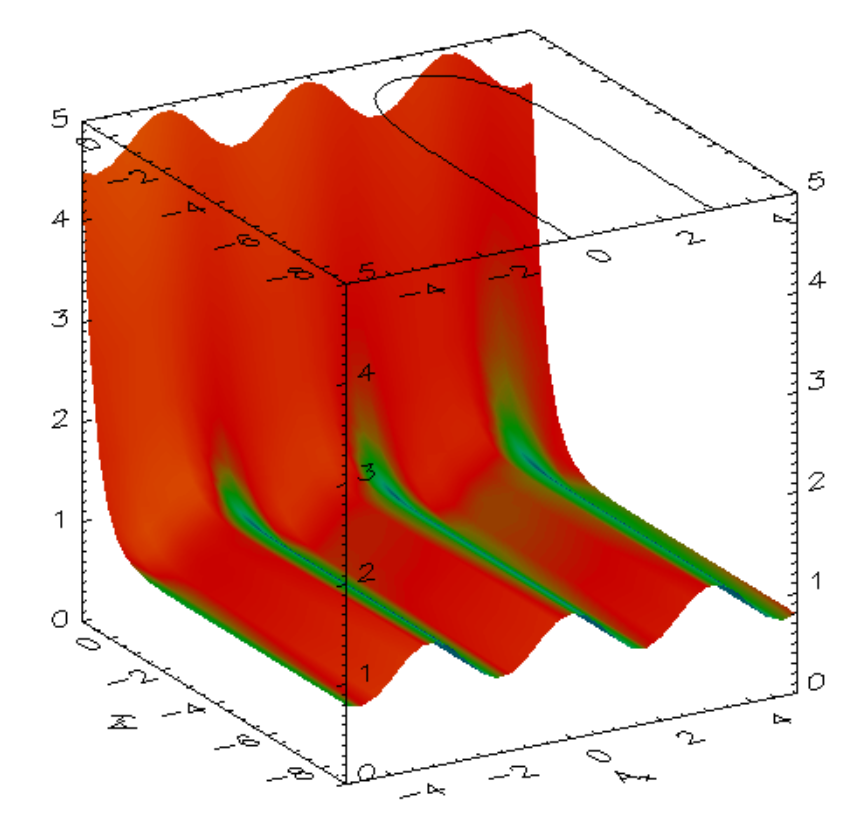}
        \caption{\small The \emph{Harrison-Neukirch} pressure function $P_{zz}$, with overlaid contour delineating the path in $(A_x(z),A_y(z))$ on which $dP_{zz}/dz=0$. {\bf Images copyright:} M.G. Harrison's PhD thesis \citep{Harrison-thesis}, (reproduced with permission).}
        \label{fig:Harrisonpzz}
   \end{figure}
Using either Fourier transforms or inspection to invert the Weierstrass transforms, the DF calculated to correspond to the $P_{zz}$ in Equation (\ref{eq:P_Harrison}) was given by
\[
f_s=\frac{n_{0s}}{(\sqrt{2\pi}v_{\text{th},s})^3}e^{-\beta_sH_s}\left(a_s\cos\left(\beta_su_{xs}p_{xs}\right)+e^{\beta_su_{ys}p_{ys}}+b_s\right).
\]
In this representation, $u_{xs}$ and $u_{ys}$ are bulk flow parameters in the $x$ and $y$ directions respectively, with 
\begin{eqnarray}
V_{xs}&=&\frac{u_{ys}\sinh (z/L)}{(b+1/2)\cosh ^2(z/L)},\nonumber\\
V_{ys}&=&\frac{u_{ys}}{(b+1/2)\cosh ^2(z/L)},\nonumber
\end{eqnarray}
and $|u_{xs}|=|u_{ys}|$.

\subsubsection{Summative pressures and the plasma beta}
A free choice of the plasma beta is not possible in the summative \emph{Harrison-Neukirch} equilibrium DF: it is bounded below by unity. In fact it is a feature generally observed that for pressure tensors (that correspond to force-free fields) constructed in this manner \citep{Harrison-2009PRL, Wilson-2011, Abraham-Shrauner-2013, Kolotkov-2015}, that the plasma-beta is bounded below by unity. By combining Equations (\ref{eq:forcefreerepresentation}) and (\ref{eq:p1p2balance}) we see that under the following assumptions,
\begin{enumerate}
\item  $P_1(A_x)\ge 0$ and $P_2(A_y)\ge 0$
\item  $\exists \, z_{1},z_{2}$ s.t. $\sin ^2S(z_{1})=1, \sin^2S(z_2)=0, \cos ^2S(z_{2})=1, \cos^2S(z_1)=0$.
\end{enumerate}
We justify Assumption 1. by the following argument. Whilst formally we only require the sum $P_{zz}=P_1(A_x(z))+P_2(A_y(z))\ge 0$ (since pressure can't be negative), we do in fact require $P_1(A_x)\ge 0$ and $P_2(A_y)\ge 0$ individually. The inverse problem defined by Equation (\ref{eq:Channell}) ties together the dependence of $P_{zz}$ on $A_x$ and $A_y$ to the dependence of the DF on $p_{xs}$ and $p_{ys}$ respectively. As the DF must be positive with respect to the independent variation of $p_{xs}$ or $p_{ys}$, so must $P_{zz}$ be with respect to independent variations of $A_x$ and $A_y$. 

Assumption 2. is trivially true in the case of a 1D linear force-free field, since $S(z)$ is a linear function of $z$. For the case of a nonlinear force-free field in which one of the magnetic field components goes through $0$, and the other tends to 0 at $\pm\infty$, Assumption 2. will hold, and this is the case for the FFHS. 

If we combine Assumptions 1. and 2., then the following inequalities will hold,
\begin{eqnarray}
P_{T1}&=&P_1(A_x(z_1))+\frac{B_0^2}{2\mu_0}\sin^2S(z_1)\ge \frac{B_0^2}{2\mu_0},\label{eq:pt1}\\
P_{T2}&=&P_2(A_y(z_2))+\frac{B_0^2}{2\mu_0}\cos^2S(z_2)\ge\frac{B_0^2}{2\mu_0}.\label{eq:pt2}
\end{eqnarray}
In fact, since $P_{zz}(z)=\text{const.}$, and $P_{T1}$ and $P_{T2}$ are independent of each other through the separation of variables, we see that the inequalities in Equations (\ref{eq:pt1}) and (\ref{eq:pt2}) must in fact hold true for all $z$. Using this knowledge, and equations (\ref{eq:p1p2balance}), we conclude that
\[
P_{T}=P_{T1}+P_{T2}\ge 2\frac{B_0^2}{2\mu_0}\implies P_1(A_x)+P_2(A_y) +\frac{B_0^2}{2\mu_0}\ge 2\frac{B_0^2}{2\mu_0},
\]
and then, upon dividing through by $B_0^2/(2\mu_0)$ that
\[
\beta_{pl}+1\ge 2\implies \beta_{pl}\ge 1.
\]

\subsubsection{Exponential pressure transformation}\label{sec:exppressure}
The lower bound of unity on the $\beta_{pl}$ for the DFs considered by \citet{Harrison-2009PRL, Wilson-2011, Abraham-Shrauner-2013, Kolotkov-2015} could be considered a problem for modelling the solar corona. Formally, $\beta_{pl}$ is defined as the ratio of the thermal energy density to the magnetic energy density;
\begin{equation}
\beta_{pl}=\sum_s\beta_{pl,s}=\frac{2\mu_0k_B}{B_0^2}\sum_sn_sT_s, \label{eq:plasmabeta}
\end{equation} 
for $n_s$ and $T_s$ the number density and temperature - of species $s$ - respectively. In a 1D Cartesian geometry, and for a DF of the form of Equation (\ref{eq:F_form}), the following relation holds
\[
P_{zz,s}=\frac{n_{s}}{\beta_{s}}=n_{s}k_{B}T_{s},
\]
e.g. see \citet{Channell-1976, Harrison-2009POP}. As a result the plasma beta can be written in the more familiar form, 
\[
\beta_{pl}=\frac{2\mu_0P_{zz}}{B_0^2}, \hspace{3mm}\text{s.t.}\,P_{zz}=\sum_sP_{zz,s}
\]
In this chapter we take the $P_{zz}$ used in \citet{Harrison-2009PRL, Neukirch-2009, Wilson-2011, Kolotkov-2015}, which is given by Equation (\ref{eq:P_Harrison}), and transform it as in Equation (\ref{eq:Ptrans}) with the exponential function according to 
\begin{equation}
\psi(P_{zz})=\exp\left[\frac{1}{P_0}\left(P_{zz}-P_{ff}\right)\right],\label{eq:Pfunc}
\end{equation}
with $P_0$ a freely chosen positive constant. This gives $\bar{P}_{zz,ff}=P_0$, and so the plasma pressure can be as low or high as desired. \citet{Channell-1976} showed that under the assumptions used in this chapter,
\begin{equation}
P_{zz}(A_x,A_y)=\frac{\beta_e+\beta_i}{\beta_e\beta_i}n(A_x,A_y),\label{eq:Ppropn}
\end{equation}
where $n=n_i=n_e$. Equation (\ref{eq:plasmabeta}) then gives
\begin{equation}
\beta_{pl}=\frac{2\mu_0P_{zz,ff}}{B_0^2}=\frac{2\mu_0P_0}{B_0^2}.\nonumber
\end{equation} 
Hence, a freely chosen $P_0$ corresponds directly to a freely chosen $\beta_{pl}$.

We note here that this pressure transformation can also be implicitly seen for the different linear force-free cases presented in the literature, although this connection has never been made. For example, the pressure function in \citet{Sestero-1967} (and implicitly in  \cite{Bobrova-2001}) is an exponentiated version of that in \citet{Channell-1976,Attico-1999}. A further interesting aspect is that the momentum dependent parts of the DFs are also related to each other exponentially in the linear force-free case.

Obviously, even if integral Equation (\ref{eq:Channell}) can be solved for the original function $P_{zz}(A_x, A_y)$ it is by no means clear that this is possible for the transformed function $\bar{P}_{zz}$. Usually one would expect that solving equation (\ref{eq:Channell}) for $g_{s}$ is much more difficult after the transformation to $\bar{P}_{zz}$.

\section{VM equilibria for the force-free Harris sheet:\\ $\beta_{pl}\in (0,\infty)$}\label{sec:newdf1}
\subsection{Calculating the DF}
The pressure function in Equation (\ref{eq:P_Harrison}) describes $\beta_{pl}\ge 1$ regimes, and we are to transform according to Equations (\ref{eq:Ptrans}) and (\ref{eq:Pfunc}) in order to realise $\beta_{pl}<1$, resulting in
\begin{equation}
\bar{P}_{zz}=P_0\exp\left\{\frac{1}{2\beta_{pl}}\left[\cos\left(\frac{2A_x}{B_0L}\right)+2\exp\left(\frac{2A_y}{B_0L}\right)              -1 \right]                   \right\}.\nonumber
\end{equation}
The $-1/(2\beta_{pl})$ term comes from the fact that ${P_{ff}=B_0^2/(2\mu_0)(1+(b-1/2))}$, readily seen for $z=0$, for example. Note that $P_{zz}$ is constant over $z$, and so we can evaluate at any $z$ to calculate $P_{ff}$. Exponentiation of $P_{zz}$ has clearly resulted in a complicated LHS of Equation (\ref{eq:Channell}), i.e.
\begin{eqnarray}
&&P_0\exp\left\{\frac{1}{2\beta_{pl}}\left[\cos\left(\frac{2A_x}{B_0L}\right)+2\exp\left(\frac{2A_y}{B_0L}\right)              -1 \right]\right\}   = \frac{\beta_{e}+\beta_{i}}{\beta_{e}\beta_{i}}\frac{n_{0s}}{2\pi m_{s}^2v_{\text{th},s}^2}\nonumber\\
&&\times\int_{-\infty}^\infty\int_{-\infty}^\infty  e^{-\beta_{s}\left((p_{xs}-q_{s}A_x)^2+(p_{ys}-q_{s}A_y)^2\right)/(2m_{s})}g_{s}(p_{xs},p_{ys})dp_{xs}dp_{ys},    \label{eq:challenging}
\end{eqnarray}
and so the inverse problem defined above is mathematically challenging.  

Since exponentiation of the `summative' pressure function results in a `multiplicative' one, we shall exploit separation of variables by assuming $g_s\propto g_{1s}(p_{xs})g_{2s}(p_{ys})$, whilst noting that $\bar{P}_{zz}\propto \bar{P}_1(A_x)\bar{P}_2(A_y)$. This assumption leads to integral equations of the form of those in Equations (\ref{eq:p1tog1}) and (\ref{eq:p2tog2}),
\begin{eqnarray}
\bar{P}_1(A_x)&\propto&\int_{-\infty}^{\infty}e^{-\beta_{s}\left(p_{xs}-q_sA_x\right)^2/(2m_s)}g_1(p_{xs})dp_{xs},\label{eq:P1tog1}\\
\bar{P}_2(A_y)&\propto&\int_{-\infty}^{\infty}e^{-\beta_{s}\left(p_{ys}-q_sA_y\right)^2/(2m_s)}g_2(p_{ys})dp_{ys},\label{eq:P2tog2}
\end{eqnarray}  
in which the LHS are formed of exponentiated cosine and exponential functions, respectively. From Equation (\ref{eq:challenging}), we see that the inverse problem now defined by Equations (\ref{eq:P1tog1}) and (\ref{eq:P2tog2}) is not analytically soluble by Fourier transform methods. Hence, we resolve to use the Hermite polynomial method from Chapter \ref{Vlasov}.

The first step is to Maclaurin expand the exponentiated pressure function of Equation (\ref{eq:P_Harrison}) according to Equations (\ref{eq:Ptrans}) and (\ref{eq:Pfunc}). Exponentiation of a power series is a combinatoric problem, and was tackled by E.T. Bell in \citet{Bell-1934}. If $h(x)=\exp k(x)$, and $k(x)$ is given by the power series 
\begin{equation}
k(x)=\sum_{n=1}^\infty \frac{1}{n!}\zeta_nx^n,\nonumber
\end{equation} 
then
\begin{equation}
h(x)=\sum_{n=0}^\infty \frac{1}{n!}Y_n(\zeta_1,\zeta_2,...,\zeta_n)x^n,\nonumber
\end{equation}
for $Y_n$ the $n$th Complete Bell Polynomial (CBP), with $Y_0=1$. These can be defined explicitly for $n\ge 1$ by Fa\`{a} di Bruno's determinant formula as the determinant of an $n\times n$ matrix \citep{Johnson-2002},
 \begin{equation}
Y_n(\zeta_1,\zeta_2,...\zeta_n)=\left|\begin{matrix}
  {n-1 \choose 0}\zeta_1 &{n-1 \choose 1}\zeta_2  &{n-1 \choose 2}\zeta_3  &\hdots &{n-1 \choose n-2}\zeta_{n-1} &{n-1 \choose n-1}\zeta_n \\
  -1 &{n-2 \choose 0}\zeta_1  & {n-2 \choose 1}\zeta_2 &\hdots &{n-2 \choose n-3}\zeta_{n-2} &{n-2 \choose n-1}\zeta_{n-1} \\
  0 & -1 & {n-3 \choose 0}\zeta_1  & \hdots&{n-3 \choose n-4}\zeta_{n-3} &{n-3 \choose n-3}\zeta_{n-2}\\
  \vdots  &\vdots &\vdots & &\vdots &\vdots\\
0& 0& 0& \hdots& {1\choose 0}\zeta_1&{1\choose 1}\zeta_2\\
0& 0& 0& \hdots& -1& {0\choose 0}\zeta_1\\
 \end{matrix}\right|.\label{eq:Bruno}
\end{equation}
For example $Y_1(\zeta_1)=\zeta_1$ and $Y_2(\zeta_1,\zeta_2)=\zeta_1^2+\zeta_2$. We include this determinant form here since this is the representation we use to plot the DF. Instructive references on CBPs can be found in \citet{Riordan-1958, Comtet-1974, Kolbig-1994, Connon-2010}, for example. Another representation for the CBPs is given by \citet{Connon-2010}, where for $n\ge 1$ the $Y_n$ can be written as
\begin{equation}
Y_n(\zeta_1,\zeta_2,...\zeta_n)=\sum_{\pi (n)}\frac{n!}{k_1!k_2!...k_n!}\left(\frac{\zeta_1}{1!}\right)^{k_1}\left(\frac{\zeta_2}{2!}\right)^{k_2}...\left(\frac{\zeta_n}{n!}\right)^{k_n},\label{eq:connonrep}
\end{equation}
where the sum is taken over all partitions $\pi(n)$ of $n$, i.e. over all sets of integers $k_j$ such that
\[
k_1+2k_2+...+nk_n =n.
\]
Using CBPs, and a simple scaling argument \citep{Bell-1934, Connon-2010}, immediately seen from equation (\ref{eq:connonrep}),
\begin{equation}
Y_n(a\zeta_1,a^2\zeta_2,...,a^n\zeta_n)=a^nY_n(\zeta_1,\zeta_2,...\zeta_n),\label{eq:CBPscale}
\end{equation}
we can derive the Maclaurin expansion of the transformed pressure, making use of
\[
\cos\left(\frac{2A_x}{B_0L}\right)=\sum_{n=0}^\infty\frac{(-1)^n}{(2n)!}\left(\frac{2A_x}{B_0L}\right)^{2n},\hspace{3mm}\exp\left(\frac{2A_y}{B_0L}\right)=\sum_{n=0}^\infty\frac{(-1)^n}{(2n+1)!}\left(\frac{2A_y}{B_0L}\right)^{2n+1}.
\]
The Maclaurin expansion is found to be
\begin{equation}
\bar{P}_{zz}=P_0e^{-1/(2\beta_{pl})}\sum_{m=0}^\infty a_{2m}\left(\frac{A_x}{B_0L}\right)^{2m}\sum_{n=0}^\infty b_n \left(\frac{A_y}{B_0L}\right)^n, \label{eq:Pmacl}
\end{equation} 
with
\begin{equation}
a_{2m}=e^{1/(2\beta_{pl})}\frac{(-1)^m2^{2m}}{(2m)!}Y_{2m}\left(0,\frac{1}{2\beta_{pl}}, 0 , ... , 0 , \frac{1}{2\beta_{pl}}\right),\label{eq:a2msimple}
\end{equation}
and
 \begin{equation}
b_n=e^{1/\beta_{pl}}\frac{2^n}{n!}Y_n\left(\frac{1}{\beta_{pl}}, ..., \frac{1}{\beta_{pl}}\right).\label{eq:bnsimple}
\end{equation}
This allows us to formally solve the inverse problem for the unknown functions $g_{1s}(p_{xs})$ and $g_{2s}(p_{ys})$ in terms of Hermite polynomials (using results from Chapter \ref{Vlasov}), giving
\begin{eqnarray}
&&f_s(H_s,p_{xs},p_{ys})=\displaystyle \frac{n_{0s}}{\left(\sqrt{2\pi}v_{\text{th},s}\right)^3}e^{-1/(2\beta_{pl})}\times\nonumber\\
&&\left[\displaystyle\sum_{m=0}^\infty C_{2m,s}H_{2m}\left(\frac{p_{xs}}{\sqrt{2}m_sv_{\text{th},s}}\right)\sum_{n=0}^\infty D_{ns}H_n\left(\frac{p_{ys}}{\sqrt{2}m_sv_{\text{th},s}}\right)\right]e^{-\beta_sH_s},\label{eq:result}
\end{eqnarray}
for species-dependent coefficients $C_{2m,s}$ and $D_{ns}$. As discussed in Chapter \ref{Vlasov}, we fix the micro-macroscopic parameter relationships by the following conditions
\begin{eqnarray}
\sigma(A_x,A_y)&=&0,\nonumber\\
P_0\exp\left\{\frac{1}{2\beta_{pl}}\left[\cos\left(\frac{2A_x}{B_0L}\right)+2\exp\left(\frac{2A_y}{B_0L}\right)              -1 \right]                   \right\}&=&m_s\sum_s\int v_z^2f_sd^3v,\nonumber
\end{eqnarray}
for the $f_s$ given by Equation (\ref{eq:result}). After performing the necessary integrations, these conditions are satisfied by fixing the parameters according to
\begin{eqnarray}
n_{0i}&=&n_{0e}=n_0,\hspace{3mm}P_0=n_0\frac{\beta_e+\beta_i}{\beta_e\beta_i}\nonumber\\
C_{2m,s}&=&\left(\frac{\delta_s}{\sqrt{2}}\right)^{2m}a_{2m},\hspace{3mm}D_{ns}=\text{sgn}(q_s)^n\left(\frac{\delta_s}{\sqrt{2}}\right)^{n}b_{n}.\nonumber
\end{eqnarray}
As yet, the distribution of Equation (\ref{eq:result}), together with the micro-macroscopic conditions, is only a formal solution to the inverse problem posed, and we now proceed to confirm the convergence and boundedness properties, using techniques from Chapter \ref{Vlasov}.

\subsection{Convergence and boundedness of the DF}\label{app:A}
Here we include the full details of the calculations that confirm the validity of the Hermite Polynomial representation of the multiplicative FFHS equilibrium in the `original' gauge (Equation (\ref{eq:Agauge1})). We shall first verify the convergence of $g_{2s}$ (expanded over $n$ in Equation (\ref{eq:result})) using the convergence condition from Section \ref{sec:convergence}, and then verify convergence of $g_{1s}$ by comparison with $g_{2s}$. 
\subsubsection{Convergence of the $p_{ys}$ dependent sum}
As Theorem \ref{thm:HermiteConvergence} states, we can verify convergence of $g_{2s}$ provided
\[
\lim_{{n\to\infty}}\sqrt{n}\left|\frac{b_{n+1}}{b_n}\right|<1/\delta_s
.\]
Explicit expansion of the exponentiated exponential series by `twice' using Maclaurin series (as opposed to the CBP formulation of Equation (\ref{eq:bnsimple})) gives
\begin{eqnarray}
\tilde{P}_2(\tilde{A}_y)=\exp \left(\frac{1}{\beta_{pl}}\exp\left(\frac{2A_y}{B_0L}\right)\right)&=&\sum_{k=0}^\infty \frac{1}{\beta_{pl}^kk!} \exp\left(\frac{2kA_y}{B_0L}\right),\nonumber\\
&=&\sum_{k=0}^\infty  \frac{1}{\beta_{pl}^kk!}\sum_{n=0}^\infty\frac{2^nk^n}{n!}\left(\frac{A_y}{B_0L}\right)^n\nonumber\\
&=&\sum_{n=0}^\infty b_n\left(\frac{A_y}{B_0L}\right)^n\nonumber ,
\end{eqnarray}
such that $b_n$ are defined by
\begin{equation}
b_n=\frac{2^n}{n!}\sum_{k=0}^\infty \frac{k^n}{\beta_{pl}^kk!},\label{eq:bmac}
\end{equation}
And for which the sum over $k$ is itself a convergent series, meaning that the $b_n$ are well-defined. Using the definition of $b_n$ and $b_{n+1}$ gives
\begin{eqnarray*}
b_{n+1}/b_n&=&\frac{2}{n+1}\sum_{j=0}^\infty \frac{j^{n+1}}{j!\beta_{pl}^{j}}\biggr/\sum_{j=0}^\infty\frac{j^n}{j!\beta_{pl}^{ j}}\\
&=&\frac{2}{n+1}\left(\frac{\displaystyle 0+\frac{1}{0!\beta_{pl}}+\frac{2^n}{1!\beta_{pl}^{ 2}}+\frac{3^n}{2!\beta_{pl}^{ 3}}+...}{\displaystyle 0+ \frac{1}{1!\beta_{pl}}+\frac{2^n}{2!\beta_{pl}^{ 2}}+\frac{3^n}{3!\beta_{pl}^{ 3}}+...}\right)\\
&=&\frac{2}{n+1}\left(\frac{\displaystyle \frac{1}{\beta_{pl}}+2\frac{2^n}{2!\beta_{pl}^{ 2}}+3\frac{3^n}{3!\beta_{pl}^{ 3}}+...}{\displaystyle \displaystyle \frac{1}{1!\beta_{pl}}+\frac{2^n}{2!\beta_{pl}^{ 2}}+\frac{3^n}{3!\beta_{pl}^{ 3}}+...}\right).
\end{eqnarray*}
The $k$th `partial sum' of this fraction has the form
\[
\mathcal{S}_{n,k}=\frac{p_1+2p_2+3p_3+...+kp_k}{p_1+p_2+p_3+...}
\]
with $p_i\asymp 1/i!$, where we write $g\asymp h$ to mean $g/h$ and $h/g$ are bounded away from $0$. Now since the denominator of the $p_{i}$ increase factorially we have $ i p_{i}\asymp p_{i}$ and hence
\[
0<\sum_{i=1}^{\infty}ip_{i}<\infty \hspace{3mm}{\rm and }\hspace{3mm} 0<\sum_{i=1}^{\infty}p_{i}<\infty.
\]
Thus $\mathcal{S}_{n,k}\to \mathcal{S}_{n,\infty}\in (0,\infty)$ and, more specifically, $\mathcal{S}_{n,\infty}\asymp 1$ in $n$. 
Therefore 
\[
b_{n+1}/b_{n}= \mathcal{S}_{n,\infty}/(n+1)\asymp 1/n.
\]
That is to say $b_{n+1}/b_{n}$ behaves asymptotically like $1/n$. This satisfies the condition of Theorem 1. Hence $g_{2s}(p_{ys})$ converges  for all $\delta_s$ and $p_{ys}$ by the comparison test.

\subsubsection{Convergence of the $p_{xs}$ dependent sum}
We shall now verify convergence of $g_{1s}$, by comparison with $g_{2s}$. By explicitly using the Maclaurin expansion of the exponential, and then the power-series representation for $\cos^nx$ from \citet{Gradshteyn}
\begin{eqnarray*}
\cos ^{2n}x&=&\frac{1}{2^{2n}}\left[\sum_{k=0}^{n-1}2{2n \choose k}\cos (2(n-k)x)+{2n \choose n} \right],\\
\cos ^{2n-1}x&=&\frac{1}{2^{2n-2}}\sum_{k=0}^{n-1}{2n-1 \choose k}\cos ((2n-2k-1)x),
\end{eqnarray*}
one can calculate 
\[
\tilde{P}_1(\tilde{A}_x)=\exp\left(\frac{1}{2\beta_{pl}}\cos\left(\frac{2A_x}{B_0L}\right)\right)=\sum_{m=0}^{\infty}a_{2m}\left(\frac{A_x}{B_0L}\right)^{2m}.
\]
The zeroth coefficient is given by $a_0=\exp\left(1/(2\beta_{pl})\right)$, and the rest are
\[
a_{2m}=\frac{2(-1)^m}{(2m)!}\sum_{k=0}^\infty\sum_{j\in J_{k}} \frac{1}{j!(4\beta_{pl})^j}{j \choose k}(j-2k)^{2m}, \label{eq:appendix1}
\]
for $J_{k}=\left\{2k+1, 2k+2, ...\right\}$ and $m\neq 0$. By rearranging the order of summation, $a_{2m}$ can be written 
\[
a_{2m}=\frac{2(-1)^m}{(2m)!}\sum_{j=1}^\infty \frac{1}{j!(4\beta_{pl})^j}\sum_{k=0}^{\lfloor (j-1)/2 \rfloor}{j \choose k}(j-2k)^{2m}, 
\]
where $\lfloor x \rfloor$ is the floor function, denoting the greatest integer less than or equal to $x$. Recognising an upper bound in the expression for $a_{2m}$;
\[
\sum_{n=0}^{\lfloor (j-1)/2 \rfloor}{j \choose n}(j-2n)^{2m}\leq j^{2m}\sum_{n=0}^j{j \choose n}=2^{j}j^{2m},
\]
gives 
\begin{eqnarray*}
a_{2m}<\frac{2(-1)^m}{(2m)!}\sum_{j=1}^\infty\frac{2^{j+1}j^{2m}}{j!2^j(2\beta_{pl})^j}&=&2\frac{(-1)^m}{(2m)!}\sum_{j=1}^\infty\frac{j^{2m}}{j!(2\beta_{pl})^j},\\
&\le & \frac{2}{(2m)!}\sum_{j=1}^\infty\frac{j^{2m}}{j!(2\beta_{pl})^j},\\
&=&\frac{1}{(2m)!}\sum_{j=1}^\infty\frac{2^{1-j}j^{2m}}{j!\beta_{pl}^j}<b_{2m}.
\end{eqnarray*}
Hence we now have an upper bound on $a_{2m}$ for $m\neq 0$ and we know that $a_{2m+1}=0$, and so is bounded above by $b_{2m+1}$. Note also that $a_0<b_0$. Hence, each term in our series for $g_{1s}(p_{xs})$ is bounded above by a series known to converge for all $\delta_s$ according to
\[
a_l\left(\frac{\delta_s}{\sqrt{2}}\right)^lH_l(x)<b_l\left(\frac{\delta_s}{\sqrt{2}}\right)^lH_l(x),\hspace{3mm}\forall l.
\]
So by the comparison test, we can now say that $g_{1s}\left(p_{xs}\right)$ is a convergent series. Hence the representation of the DF in Equation (\ref{eq:result}) is convergent. 


\subsubsection{Boundedness of the DF}
The boundedness of the DF in Equation (\ref{eq:result}) is now guaranteed by the reasoning from Section \ref{sec:boundedness} for a general solution, and need not be repeated here.

\subsection{Moments of the DF} 

\label{Appendix-PoP} 
The moments of the DF are used to calculate the number density and bulk velocity, and in turn the charge and current densities respectively. It is useful to calculate these quantities from the DF to confirm parity with the required macroscopic quantities not only as a procedural check, but also to derive relations between the micro- and macroscopic parameters. 
\subsubsection{The zeroth order moment}
The number density is found by taking the zeroth moment;
\begin{eqnarray*}
&&n_s(A_x,A_y)=\frac{e^{-\frac{1}{2\beta_{pl}}}}{m_s^3}\frac{n_0}{(\sqrt{2\pi}v_{\text{th},s})^3}\times\\
&&\biggr [\sum_{m=0}^\infty C_{2m,s}\int_{-\infty}^\infty e^{-\frac{\beta_s}{2m_s}(p_{xs}-q_sA_x)^2}H_{2m}\left(\frac{p_{xs}}{\sqrt{2}m_sv_{\text{th},s}}\right)dp_{xs}\times\\
&&\sum_{n=0}^\infty D_{ns}\int_{-\infty}^\infty e^{-\frac{\beta_s}{2m_s}(p_{ys}-q_sA_y)^2}H_{n}\left(\frac{p_{ys}}{\sqrt{2}m_sv_{\text{th},s}}\right)dp_{ys}\int_{-\infty}^\infty e^{-\frac{\beta_s}{2m_s}p_{zs}^2}dp_{zs}\biggr],
\end{eqnarray*}
which, after integrating over $p_{zs}$ and making substitutions,  gives
\begin{eqnarray*}
n_s(A_x,A_y)&=\displaystyle\frac{n_0e^{-\frac{1}{2\beta_{pl}}}}{\pi}\sum_{m=0}^\infty C_{2m,s}\int_{-\infty}^\infty e^{-(X-\frac{q_sA_x}{\sqrt{2}m_sv_{\text{th},s}})^2}H_{2m}(X)dX\\
&\times\displaystyle\sum_{n=0}^\infty D_{ns}\int_{-\infty}^\infty e^{-(Y-\frac{q_sA_y}{\sqrt{2}m_sv_{\text{th},s}})^2}H_n(Y)dY.
\end{eqnarray*}
Use the standard integral \citep{Gradshteyn}, 
\begin{equation}
\int_{-\infty}^\infty e^{-(x-y)^2}H_n(x)dx=\sqrt{\pi}2^ny^n,\nonumber
\end{equation}
to give
\begin{eqnarray*}
n_s(A_x,A_y)&=&n_0e^{-\frac{1}{2\beta_{pl}}}\displaystyle\sum_{m=0}^\infty C_{2m,s}2^{2m}\left(\frac{q_sA_x}{\sqrt{2}m_sv_{\text{th},s}}\right)^{2m}\displaystyle\sum_{n=0}^\infty D_{ns}2^n\left(\frac{q_sA_y}{\sqrt{2}m_sv_{\text{th},s}}\right)^{n}\\
&=&\frac{n_0}{P_0}\bar{P}_{zz}.
\end{eqnarray*}
Using $P_{zz,ff}=P_0$, we see that 
\begin{equation}
n_{ff}=n_0,\nonumber
\end{equation}
and so $n_0$ represents the constant particle number density.

\subsubsection{The $v_x$ moment}
We now take the first moment of the DF by $v_x$ denoted by $[v_xf_s]$;
\begin{eqnarray*}
&&[v_xf_s]=\displaystyle\frac{1}{m_s^3}\int_{-\infty}^{\infty}\int_{-\infty}^{\infty}\int_{-\infty}^{\infty} v_xf_sd^3p,\\
&&=\displaystyle\frac{n_0e^{-\frac{1}{2\beta_{pl}}}}{(\sqrt{2\pi})m_sv_{\text{th},s}}\sum_{n=0}^\infty b_n\left(\frac{A_y}{B_0L}\right)^n\sum_{m=0}^\infty C_{2m,s}\times\\
&&\underbrace{\int_{-\infty}^\infty v_xe^{-\frac{\beta_s}{2m_s}(p_{xs}-q_sA_x)^2}H_{2m}\left(\frac{p_{xs}}{\sqrt{2}m_sv_{\text{th},s}}\right)dp_{xs}}_{I_{v_x}},
\end{eqnarray*}
after both the $p_{ys}$ and $p_{zs}$ integrations. Now, use the Hermite expansion of the exponential \citep{Morse-1953}, to give
\begin{eqnarray*}
&&I_{v_x}=\frac{1}{m_s}\int_{-\infty}^\infty (p_{xs}-q_sA_x)H_{2m}\left(\frac{p_{xs}}{\sqrt{2}m_sv_{\text{th},s}}\right)e^{-\frac{\beta_sp_{xs}^2}{2m_s}}\times\\
&&\left[\sum_{j=0}^\infty\frac{1}{(j)!}\left(\frac{q_sA_x}{\sqrt{2}m_sv_{\text{th},s}}\right)^jH_{j}\left(\frac{p_{xs}}{\sqrt{2}m_sv_{\text{th},s}}\right)\right]dp_{xs}.
\end{eqnarray*}
Now define an inner product according to
\begin{equation}
\langle f_1(x),f_2(x)\rangle =\int_{-\infty}^\infty e^{-x^2}f_1(x)f_2(x)dx.
\end{equation}
Then orthogonality of the Hermite polynomials (Equation (\ref{eq:orthogonal})), and the recurrence relation, $H_{n+1}(x)=2xH_n(x)-2nH_{n-1}(x)$, are used to give
\begin{equation}
\begin{split}
\langle xH_j(x),H_{2m}(x)\rangle&=j\langle H_{j-1}(x),H_{2m}(x)\rangle+\frac{1}{2}\langle H_{j+1}(x),H_{2m}(x)\rangle\\
&=\sqrt{\pi}2^{2m}(2m)!\left(j\delta_{j-1,2m}+\frac{1}{2}\delta_{j+1,2m}\right).
\end{split}
\end{equation}
This allows us to write
\begin{eqnarray*}
&&I_{v_x}=\displaystyle\sqrt{2\pi}v_{\text{th},s}2^{2m}(2m)!\times\\
&&\sum_{j=0}^\infty \frac{1}{j!}\left(\frac{q_sA_x}{\sqrt{2}m_sv_{\text{th},s}}\right)^j\biggr[\sqrt{2}m_sv_{\text{th},s}\left(j\delta_{j-1,2m}+\frac{1}{2}\delta_{j+1,2m}\right)-q_sA_x\delta_{j,2m}\biggr].\nonumber
\end{eqnarray*}
Hence, we have
\begin{eqnarray*}
[v_xf_s]&=&\displaystyle\frac{n_0e^{-\frac{1}{2\beta_{pl}}}}{m_s}\sum_{n=0}^\infty b_n\left(\frac{A_y}{B_0L}\right)^n\sum_{m=0}^\infty C_{2m,s}2^{2m}(2m)!\\
&&\times\sum_{j=0}^\infty \frac{1}{j!}\left(\frac{q_sA_x}{\sqrt{2}m_sv_{\text{th},s}}\right)^j\biggr[\sqrt{2}m_sv_{\text{th},s}\left(j\delta_{j-1,2m}+\frac{1}{2}\delta_{j+1,2m}\right)-q_sA_x\delta_{j,2m}\biggr].
\end{eqnarray*}
reducing to
\begin{eqnarray}
\displaystyle[v_xf_s]&=&\displaystyle\left(\frac{m_sv_{\text{th},s}^2}{q_sB_0L}\right)n_0e^{-\frac{1}{2\beta_{pl}}}\sum_{n=0}^\infty b_n\left(\frac{A_y}{B_0L}\right)^n\sum_{m=1}^\infty a_{2m}2m\left(\frac{A_x}{B_0L}\right)^{2m-1}\nonumber \\
&=&\left(\frac{m_sv_{\text{th},s}^2}{q_sP_0}\right)n_0\frac{\partial \bar{P}_{zz}}{\partial A_x}=\frac{\beta_e\beta_i}{\beta_e+\beta_i}\left(\frac{1}{q_s\beta_s}\right)\frac{\partial \bar{P}_{zz}}{\partial A_x}\label{eq:verify}
\end{eqnarray}

The $x$ component of current density is defined as $j_x=\sum_sq_s[v_xf_s]$, giving
\begin{eqnarray}
j_x&=&\frac{\beta_e\beta_i}{\beta_e+\beta_i}\frac{\partial \bar{P}_{zz}}{\partial A_x}\sum_s\frac{1}{\beta_s}=\frac{\partial \bar{P}_{zz}}{\partial A_x}\nonumber\\
\implies j_x&=&\frac{\partial \bar{P}_{zz}}{\partial A_x},
\end{eqnarray}
reproducing the familiar result from e.g. \citet{Channell-1976, Harrison-2009POP, Schindlerbook, Mynick-1979a}. The first moment of the DF can also be used to calculate the bulk velocity in terms of the microscopic parameters;
\begin{equation}
V_{xs}=\frac{[v_xf_s]}{n_s}=\frac{j_x}{q_s\beta_sP_0},
\end{equation}
using Equation (\ref{eq:verify}). Then, by using the current density for the FFHS (Equation (\ref{eq:FFHScurrent})),
\begin{equation}
\boldsymbol{j}=\frac{B_0}{\mu_0L}\left(\frac{\text{sinh}\left(\frac{z}{L}\right)}{\text{cosh}^2\left(\frac{z}{L}\right)},\frac{1}{\text{cosh}^2\left(\frac{z}{L}\right)}   ,0 \right),
\end{equation}
we have the bulk flow in $x$
\begin{equation}
V_{xs}=\frac{B_0}{\mu_0 Lq_s\beta_sP_0}\frac{\text{sinh}\left(\frac{z}{L}\right)}{\text{cosh}^2\left(\frac{z}{L}\right)}.
\end{equation}

\subsubsection{The $v_y$ moment}
By a completely analogous calculation, we derive the $v_y$ moment of the DF,
\begin{eqnarray*} 
[v_yf_s]&=&\displaystyle\left(\frac{m_sv_{\text{th},s}^2}{P_0q_s}\right)n_0\frac{\partial \bar{P}_{zz}}{\partial A_y}\\
&=&\displaystyle\frac{\beta_e\beta_i}{\beta_e+\beta_i}\left(\frac{m_sv_{\text{th},s}^2}{q_s}\right)\frac{\partial \bar{P}_{zz}}{\partial A_y}
\end{eqnarray*}
Again, the current density $j_y=\sum_sq_s[v_yf_s]$ gives
\begin{eqnarray*}
j_y&=&\displaystyle\frac{\beta_e\beta_i}{\beta_e+\beta_i}\frac{\partial P_{zz}}{\partial A_y}\sum_sm_sv_{\text{th},s}^2=\frac{\partial \bar{P}_{zz}}{\partial A_y}\\
\implies j_y&=&\frac{\partial \bar{P}_{zz}}{\partial A_y}.
\end{eqnarray*}
We can also calculate the bulk velocity in terms of the microscopic parameters;
\begin{equation}
V_{ys}=\frac{B_0}{\mu_0 Lq_s\beta_sP_0}\frac{1}{\text{cosh}^2\left(z/L\right)}.
\end{equation}

\subsection{Properties of the DF}
\subsubsection{Current sheet width}
The nature of the inverse problem is to calculate a microscopic description of a system, given certain prescribed macroscopic data. Hence, one of the main tasks is to find the relationships between the characteristic parameters of each level of description. That is to say, given $(B_0,P_0,L)$ for example, what is their relation to $(m_s, q_s, v_{\text{th},s}, n_{0s})$?

Currently, there are six free parameters that will determine the nature of the equilibrium. These are $n_{0}$, $\beta_{pl}$, $\beta_{th,i}$, $\beta_{th,e}$,  $\delta_i$ and $\delta_e$. $n_{0}$ is in principle fixed by ensuring that the DF is normalised to the total particle number. As yet we have no information regarding the width of the current sheet $L$. To this end we shall consider bulk velocities $V_{xs}$ and $V_{ys}$, obtained from the first moment of the DF. The calculations in Section \ref{Appendix-PoP}, together with the fact that $B_0=\sqrt{2\mu_0P_0/\beta_{pl}}$ give
\begin{eqnarray*}
V_{xs}&=&\frac{[v_xf_s]}{n_0}=\displaystyle\sqrt{\frac{2}{\mu_0\beta_{pl}P_0}}\frac{1}{Lq_s\beta_s}\frac{\text{sinh}\left(z/L\right)}{\text{cosh}^2\left(z/L\right)},\\
V_{ys}&=&\frac{[v_yf_s]}{n_0}=\displaystyle\sqrt{\frac{2}{\mu_0\beta_{pl}P_0}}\frac{1}{Lq_s\beta_s}\frac{1}{\text{cosh}^2\left(z/L\right)}.
\end{eqnarray*}
We can identify the coefficient of the $z$ dependent profiles as the amplitude of the bulk velocities, $V_{xs}$ and $V_{ys}$, as $u_{s}$, given by
\begin{equation}
u_{s}=\sqrt{\frac{2}{\mu_0\beta_{pl}P_0}}\frac{1}{Lq_s\beta_s},\label{eq:us}
\end{equation}
giving
\begin{eqnarray}
(u_{i}-u_{e})^2&=&\frac{2}{\mu_0\beta_{pl}P_0L^2e^2}\left(\frac{\beta_e+\beta_i}{\beta_e\beta_i}\right)^2,\\
\implies L &=&\frac{1}{e}\sqrt{\frac{2(\beta_e+\beta_i)}{\mu_0n_0\beta_e\beta_i(u_{i}-u_{e})^2\beta_{pl}}},\label{eq:length1}
\end{eqnarray}
where $e=|q_s|$. Interestingly, this is almost identical to the expression found in \citet{Neukirch-2009} for the current sheet width of the Harrison-Neukirch equilibrium, with the addition of the $\beta_{pl}^{1/2}$ factor in the denominator. It is readily seen that, given some fixed $B_0$, $L\propto\beta_{pl}^{-1/2}$. This makes sense in that, by raising the number density $n_0$, and hence $\beta_{pl}$, there are simply more current carriers available to produce $\boldsymbol{j}$, and hence the width $L$ can reduce. By manipulating Equation (\ref{eq:us}) one can show that the amplitudes of the fluid velocities are given by 
\begin{equation}
\frac{u_{s}}{v_{\text{th},s}}=2\text{sgn}(q_s)\frac{\delta_s }{\beta_{pl}}=2\text{sgn}(q_s)\frac{\rho_s}{L\beta_{pl}}.
\end{equation}
Once again, this is almost identical to the expression found in \citet{Neukirch-2009}, with the addition of a $\beta_{pl}$ factor in the denominator.

\subsubsection{Plots of the DF}\label{sec:plotsgauge1}
Having found mathematical expressions for the DFs, we now present different plots of their dependence on $v_x$ and $v_y$, for $z/L=0,-1,1$. Plotting $f_s$ in the original gauge is a challenging numerical task, and particularly for the low-$\beta_{pl}$ regime. The reasoning is as follows. When $\beta_{pl}<1/2$, the $C_{2m,s}$ (for example) are readily seen to be of the order
\begin{equation*}
\left(\frac{1}{\sqrt{2}}\right)^{2m}\frac{1}{(2m)!}\left(\frac{\delta_s}{\beta_{pl}}\right)^{2m},
\end{equation*}
since $Y_{2m}$ is a polynomial of order $2m$ in $1/(2\beta_{pl})$. The factorial dependence in the denominator ensures that these terms $\to 0$ as $m \to \infty$. But, for relatively small $m$ there is a competition between the factorial and the $\beta_{pl}^{-2m}$, factor. This means that one must go to many terms in the expansion to get near numerical convergence. As a result, one needs to calculate both incredibly small (e.g. the $1/(2m)!$ factor), and incredibly large numbers (the $Y_{2m}$ factors), and combine them to reach $C_{2m,s}$. 

Furthermore, the Hermite polynomials become very large when the modulus of the argument is large. In normalised parameters, suitable for numerical methods, we have that
\begin{equation}
H_n\left(\frac{p_{js}}{\sqrt{2}m_sv_{\text{th},s}}\right)=H_n\left(\frac{1}{\sqrt{2}}\left(\tilde{v}_{js}+\text{sgn}(q_s)\delta_s^{-1}\tilde{A}_j\right)\right),\label{eq:normalisedHn}
\end{equation}
for $\tilde{p}_{js}=p_{js}/(m_sv_{\text{th},s})$, and $\tilde{A}_j=A_j/(B_0L)$. In particular, small values of $\delta_s$ mean that one needs to calculate $H_n$ of a large number, which can itself be inordinately large since $H_n$ is a polynomial.

So, while it has been proven that the series with which we represent the DFs are convergent for all values of the relevant parameters, attaining numerical convergence is difficult for the low-$\beta_{pl}$ regime, and particularly for the $p_{xs}$ dependent sum. Here we present plots for $\beta_{pl}=0.85$ and $\delta_i=\delta_e=0.15$.  As aforementioned we use Fa\`{a} di Bruno's determinant formula in Equation (\ref{eq:Bruno}) to calculate the CBP's, and a recurrence relation for the Hermite Polynomials. Whilst this $\beta_{pl}$ is only modestly below unity, however it represents a value of which we are confident of our numerics for both the $p_{xs}$ and $p_{ys}$ dependent sums. In Figures (\ref{fig:1aa})-(\ref{fig:1cc}) we plot the $v_x$ variation of our electron DF, as a representative example (the $v_y$ plots are qualitatively similar). First of all we note that the DFs appear to have only a single maximum, and fall off as $v_x\to\pm \infty$. This is to be contrasted with the plots of the DF using the additive pressure, which can have multiple peaks \citep{Neukirch-2009}. Thus far we have not found any indication of multiple peaks in the parameter regime that we have been able to explore. However, this does not mean that multiple peaks can not appear, for example for lower values of the $\beta_{pl}$.

A first look at the plots also seems to indicate that the shape of the DF resembles the shape of a Maxwellian. Motivated by this similarity, we define a Maxwellian DF according to Equation (\ref{eq:Mshift}), and repeated here,
\begin{equation}
f_{Maxw,s}=\frac{n_0}{(\sqrt{2\pi}v_{\text{th},s})^3}\exp\left[\frac{\left(\boldsymbol{v}-\boldsymbol{V}_s(z)\right)^2}{2v_{\text{th},s}^2}\right]. \label{eq:Maxshift}
\end{equation}   
The Maxwellian distribution reproduces the same first order moment in terms of $z$ as the equilibrium solution does, namely $\boldsymbol{V}_s$, and a spatially uniform number density, namely $n_0$. However it is not a solution of the Vlasov equation and hence not an equilibrium solution. PIC simulations for a force-free field were initiated with a distribution of this type in \citet{Hesse-2005,Birn-2010}, for example. To highlight the difference between the two DFs, we plot both the $v_x$ and $v_y$ variation of the ratio of the DF, with the Maxwellian of Equation (\ref{eq:Maxshift}) for both ions and electrons in Figures (\ref{fig:2aa})-(\ref{fig:5cc}). As we can see, in all plots the ratio deviates from unity, and in some cases these deviations are substantial. This shows that the initial impression is somewhat misleading. We also observe a symmetry in that the $v_y$ dependent plots are even in $z$, since $A_y$ and $\langle v_y\rangle _s$ are even in $z$.

To further see the deviations of $f_s$ from the Maxwellian, we present contour plots of the difference $f_s-f_{Maxw,s}$ in Figures (\ref{fig:6aa})-(\ref{fig:7cc}) over $(v_x,v_y)$ space for various $z$ values. One observation we can make from these is that there is a symmetry with respect to both velocity direction and the value of $z$. For example it seems that $f_s$ is symmetric under the transformation $(v_x\to -v_x, \,z\to -z)$. This seems reasonable since $A_x$ is dynamically equivalent to an odd function of $z$, by a gauge transformation, as $B_y$ is even (more on this in Section \ref{Subsec:regauge}). For a plasma-beta modestly below unity, and thermal Larmor radius roughly 15$\%$ of the current sheet width, we find distributions that are roughly Maxwellian in shape, but `shallower' at the centre of the sheet. At the outer edges of the sheet, this shallowness assumes a drop-shaped depression in the $v_x$  direction, with localised differences for large $v_y$.

\begin{figure}
\centering
\begin{subfigure}[b]{0.6\textwidth}
\includegraphics[width=\textwidth]{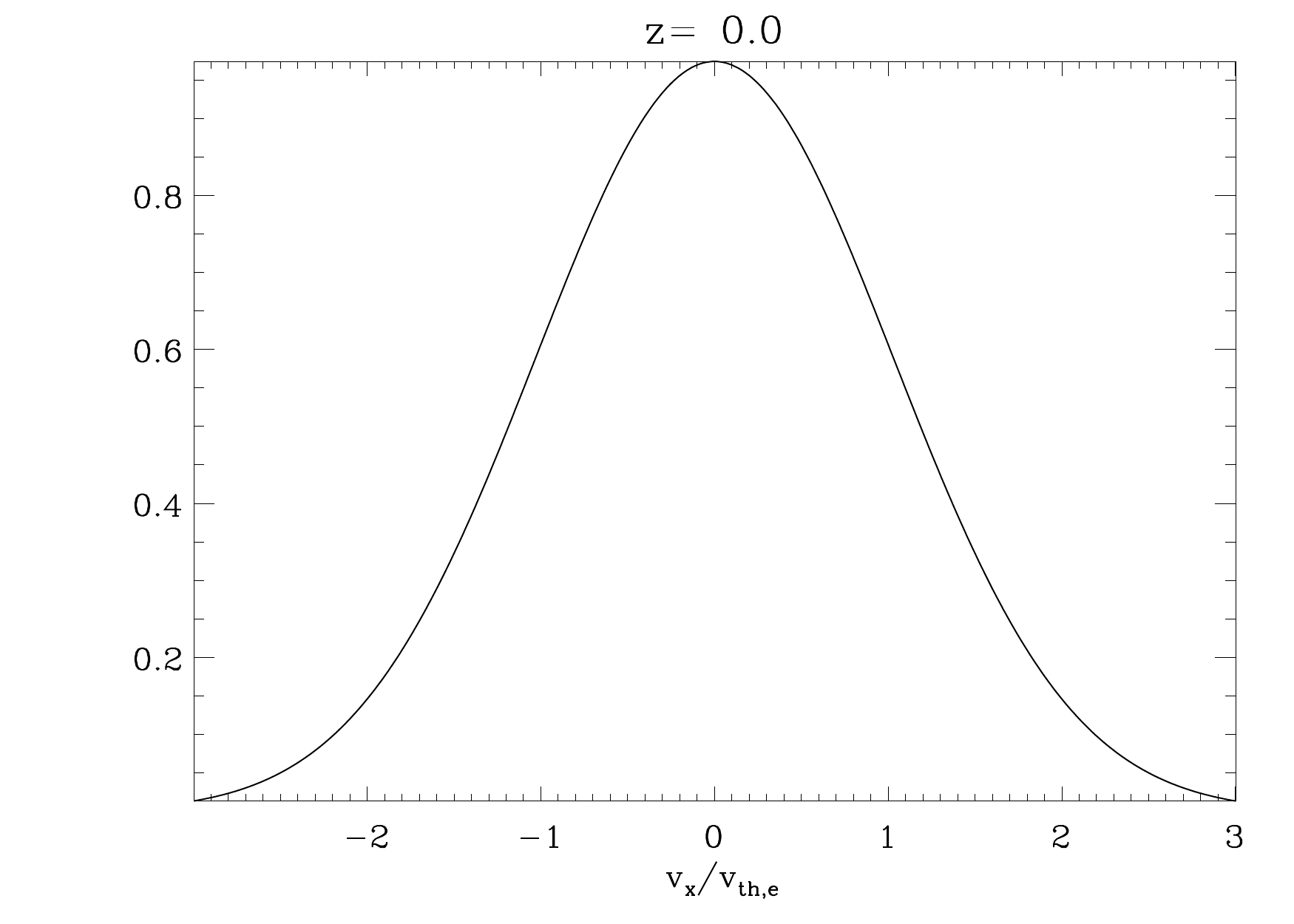}
 \caption{\small \label{fig:1aa}}
\end{subfigure}
\begin{subfigure}[b]{0.6\textwidth}
\includegraphics[width=\textwidth]{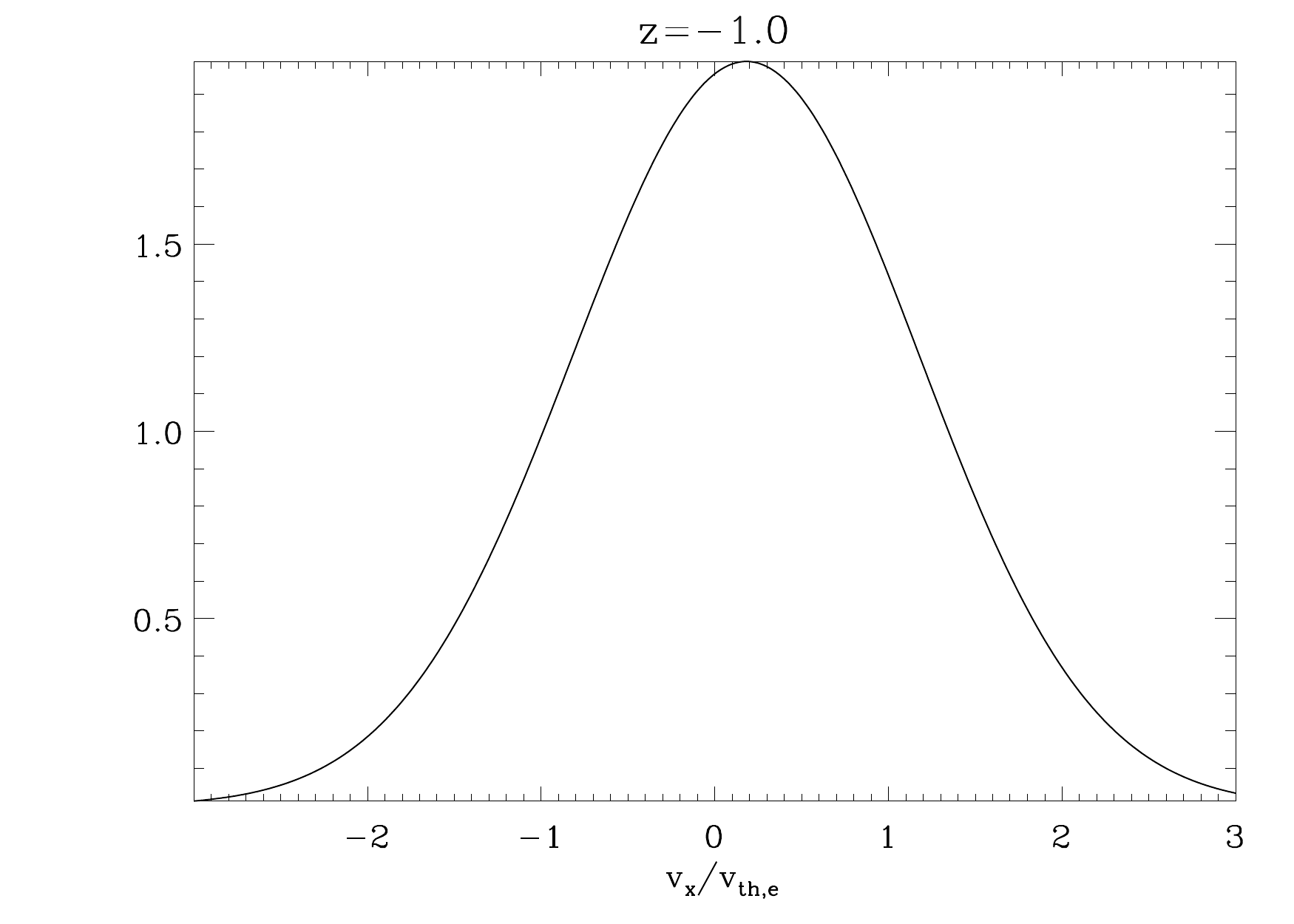}
 \caption{\small \label{fig:1bb}}
\end{subfigure}
\begin{subfigure}[b]{0.6\textwidth}
\includegraphics[width=\textwidth]{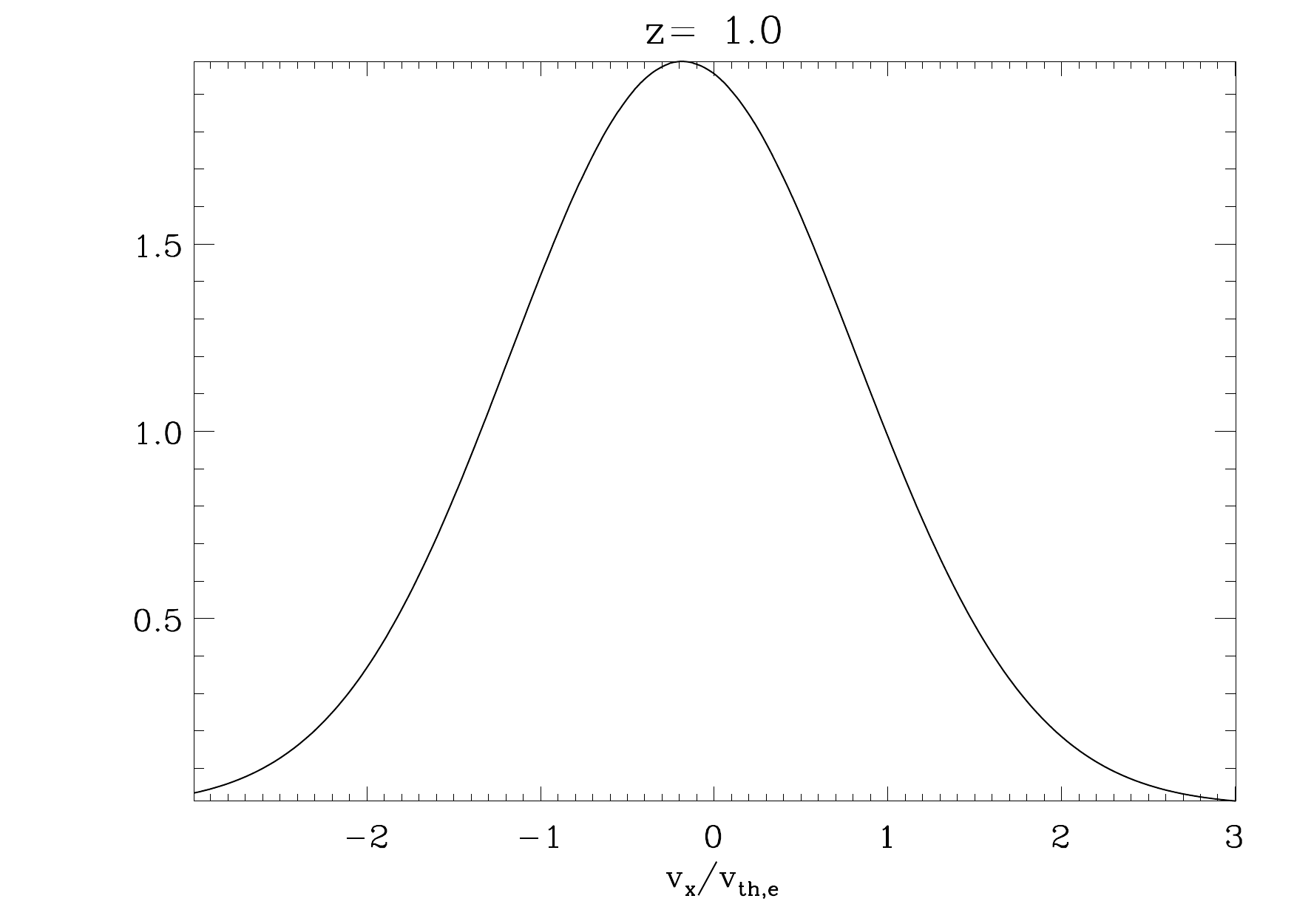}
 \caption{\small \label{fig:1cc}}
\end{subfigure}
\caption{ \small The $v_x$ variation of $f_e$ for $z/L=0$ (\ref{fig:1aa}),  $z/L=-1$ (\ref{fig:1bb}) and  $z/L=1$ (\ref{fig:1cc}). $\beta_{pl}=0.85$ and $\delta_e=0.15$. Note the antisymmetry of the $z=\pm 1$ plots with respect to each other.}
 \end{figure}

\begin{figure}
\centering
\begin{subfigure}[b]{0.6\textwidth}
\includegraphics[width=\textwidth]{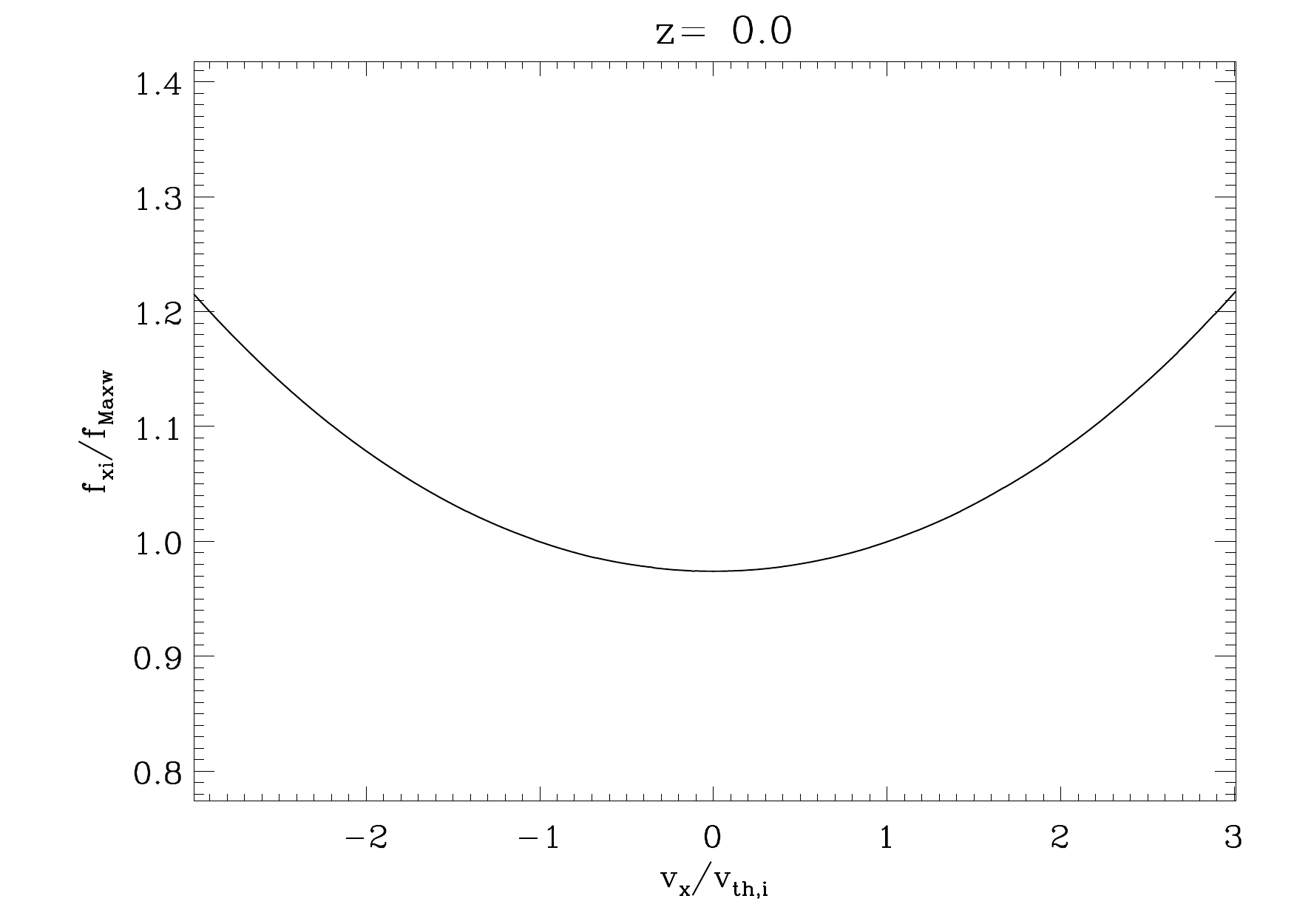}
 \caption{\small \label{fig:2aa}}
\end{subfigure}
\begin{subfigure}[b]{0.6\textwidth}
\includegraphics[width=\textwidth]{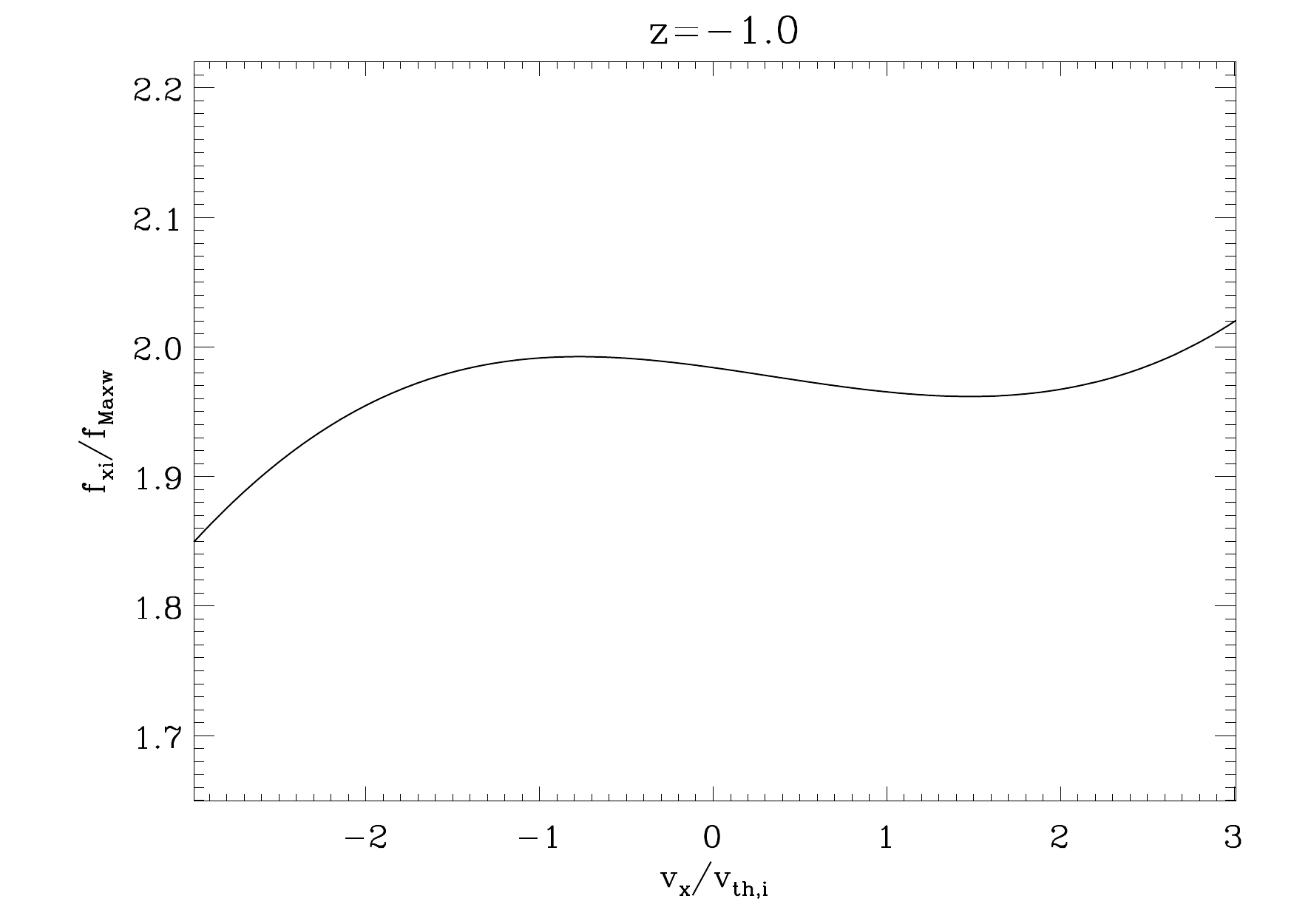}
 \caption{\small \label{fig:2bb}}
\end{subfigure}
\begin{subfigure}[b]{0.6\textwidth}
\includegraphics[width=\textwidth]{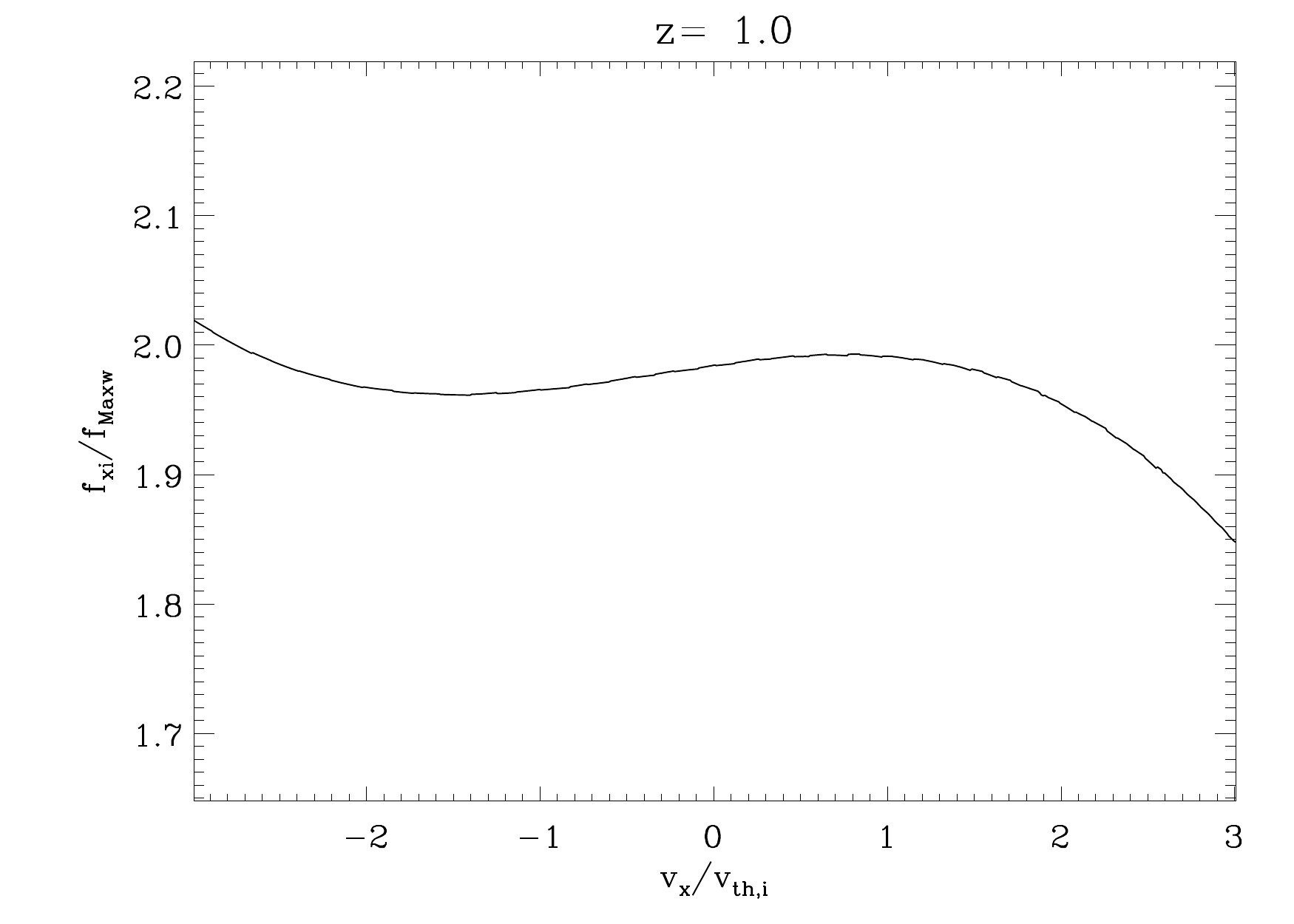}
 \caption{\small \label{fig:2cc}}
\end{subfigure}
\caption{\small The $v_x$ variation of $f_i/f_{Maxw,i}$ for $z/L=0$ (\ref{fig:2aa}),  $z/L=-1$ (\ref{fig:2bb}) and  $z/L=1$ (\ref{fig:2cc}). $\beta_{pl}=0.85$ and $\delta_i=0.15$. Note the antisymmetry of the $z=\pm 1$ plots with respect to each other.}
 \end{figure}

\begin{figure}
\centering
\begin{subfigure}[b]{0.6\textwidth}
\includegraphics[width=\textwidth]{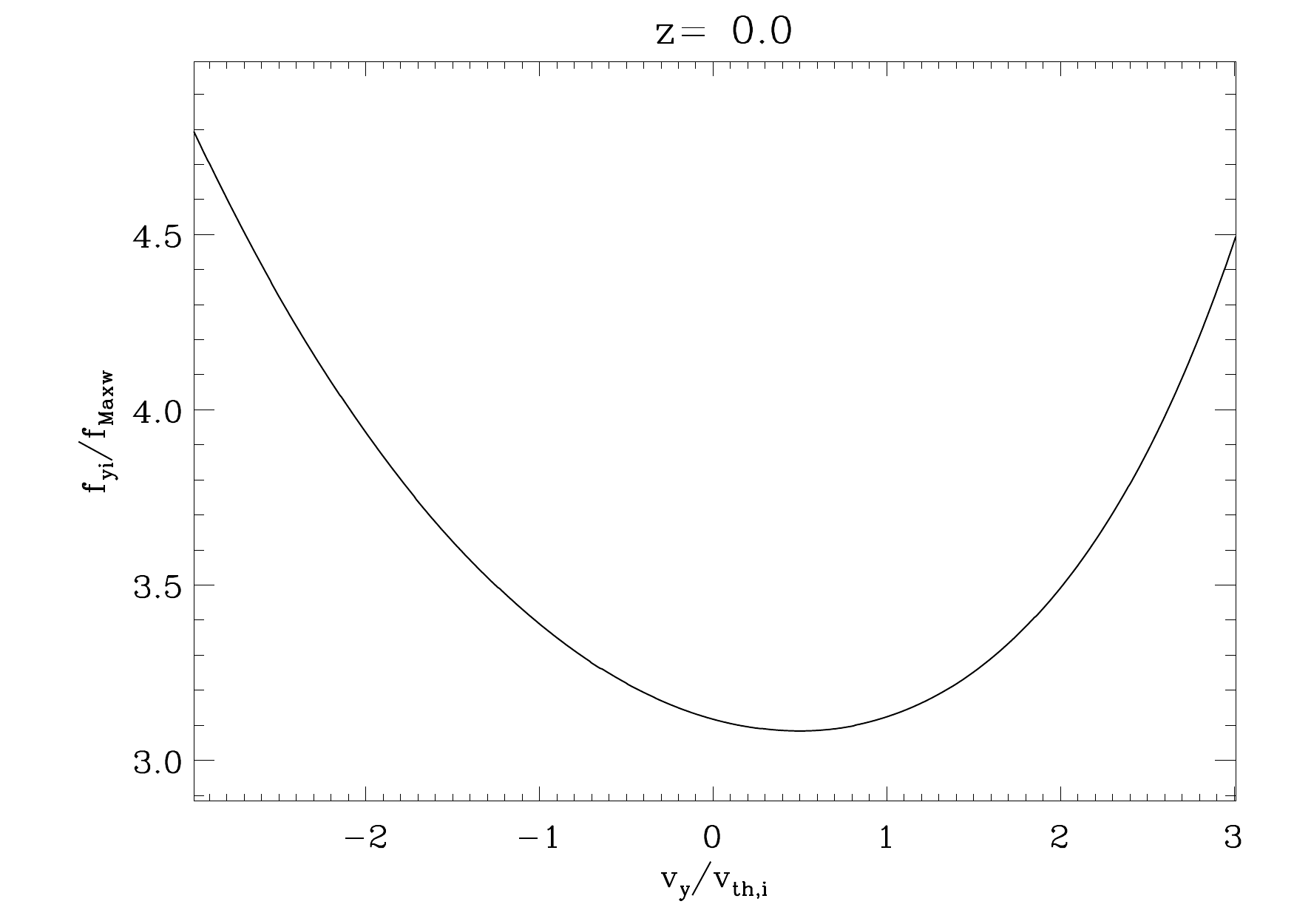}
 \caption{\small \label{fig:3aa}}
\end{subfigure}
\begin{subfigure}[b]{0.6\textwidth}
\includegraphics[width=\textwidth]{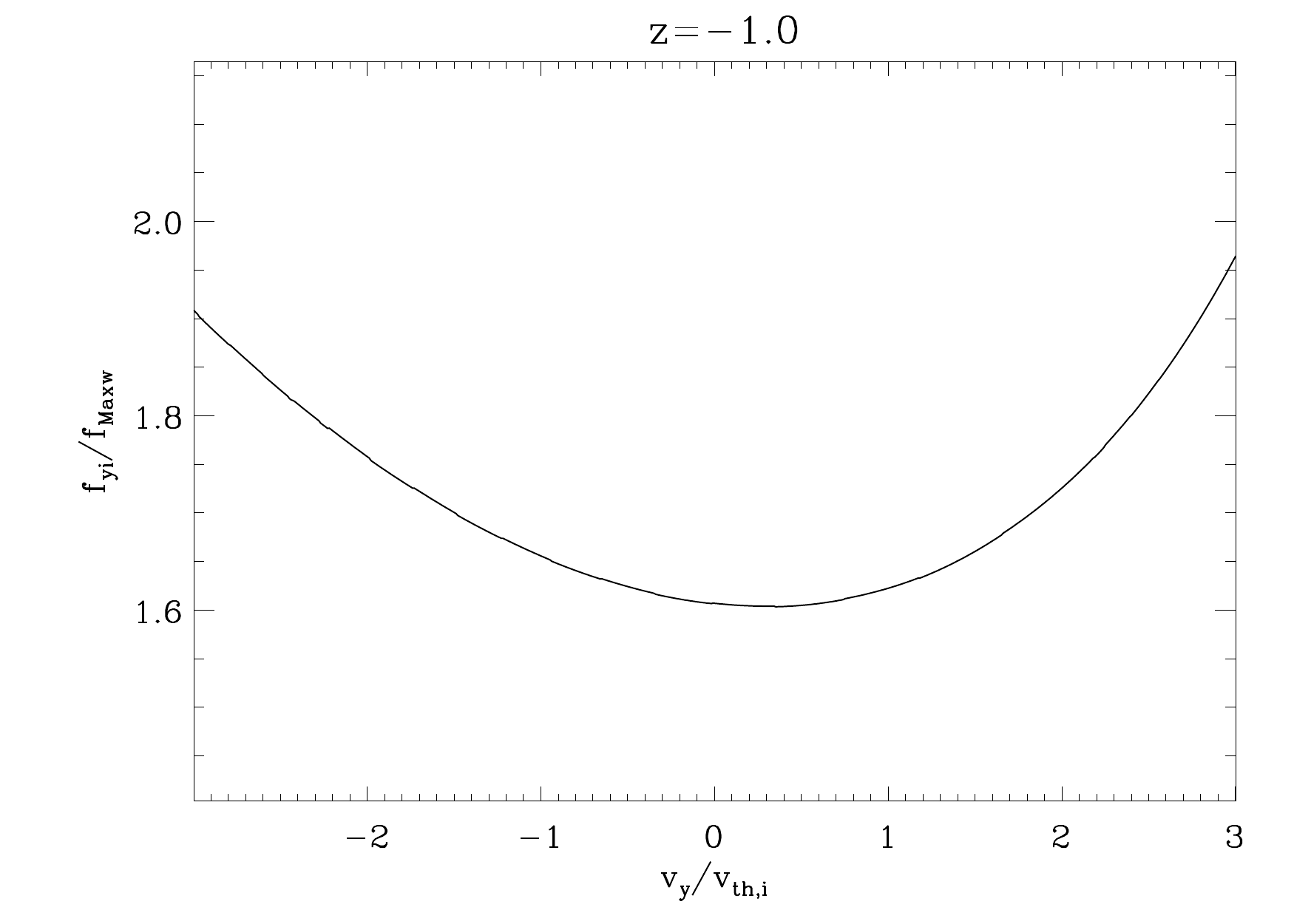}
 \caption{\small \label{fig:3bb}}
\end{subfigure}
\begin{subfigure}[b]{0.6\textwidth}
\includegraphics[width=\textwidth]{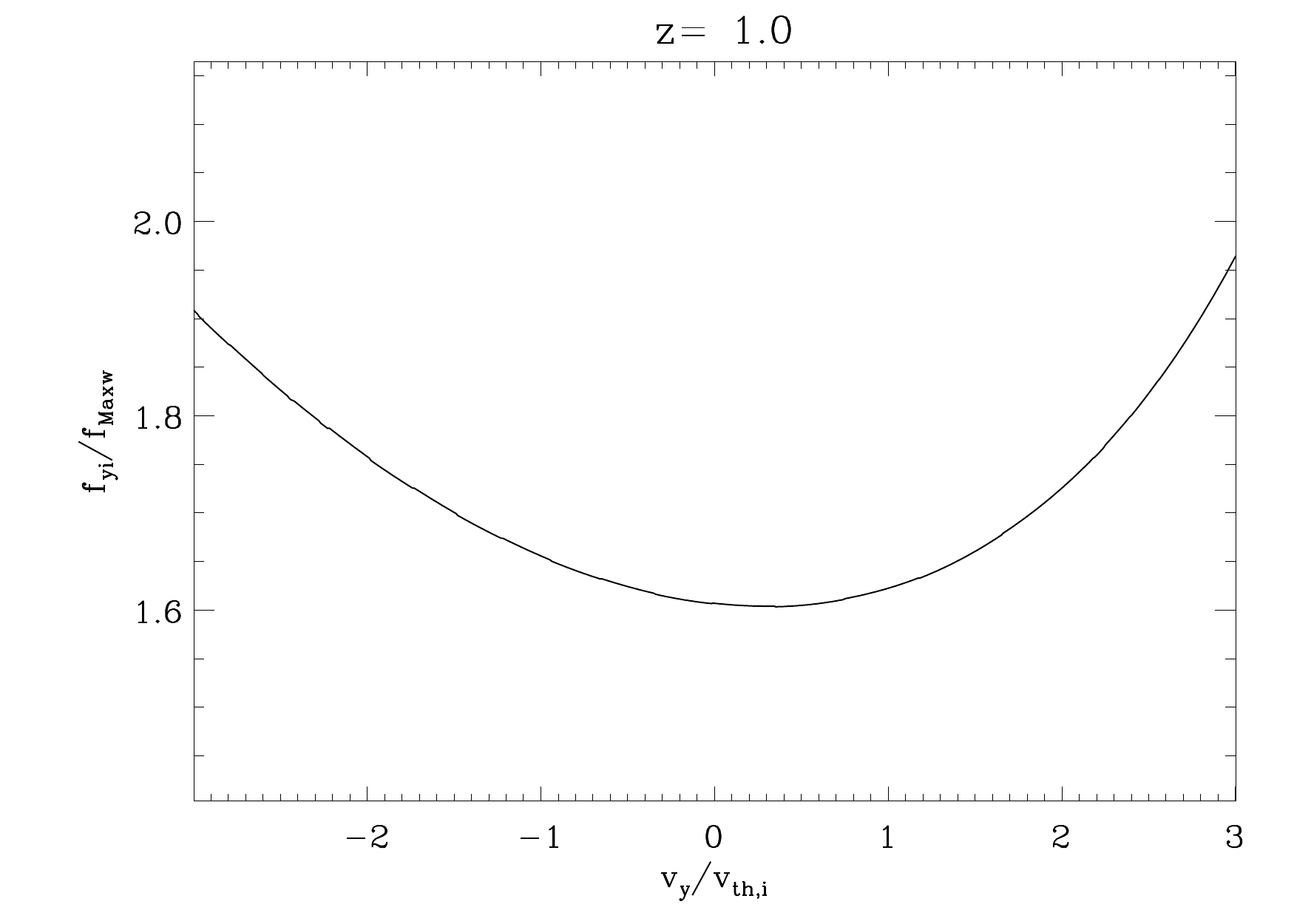}
 \caption{\small \label{fig:3cc}}
\end{subfigure}
\caption{\small The $v_y$ variation of $f_i/f_{Maxw,i}$ for $z/L=0$ (\ref{fig:3aa}),  $z/L=-1$ (\ref{fig:3bb}) and  $z/L=1$ (\ref{fig:3cc}). $\beta_{pl}=0.85$ and $\delta_i=0.15$. Note the symmetry of the $z=\pm 1$ plots with respect to each other.}
\end{figure}

\begin{figure}
\centering
\begin{subfigure}[b]{0.6\textwidth}
\includegraphics[width=\textwidth]{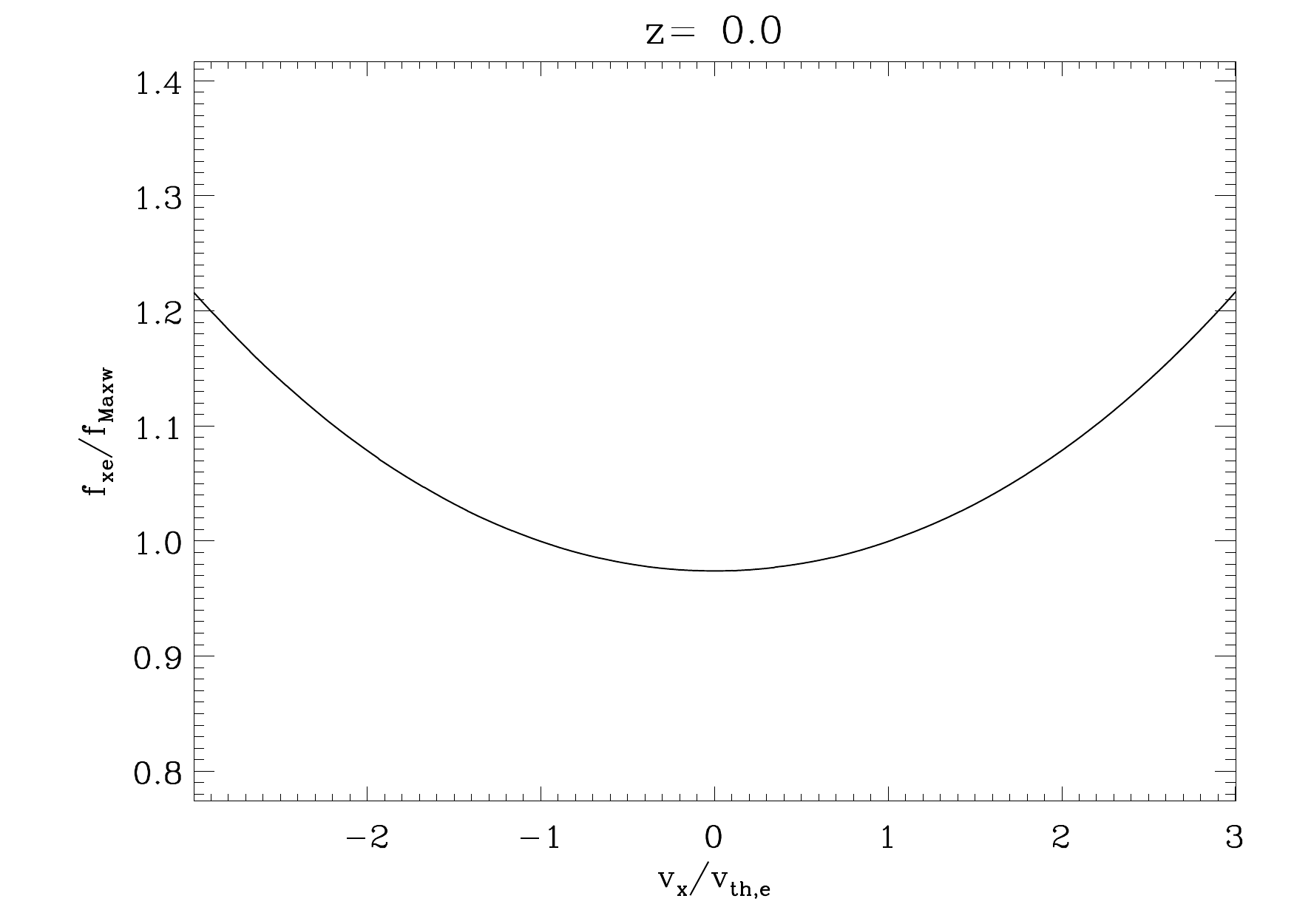}
 \caption{\small \label{fig:4aa}}
\end{subfigure}
\begin{subfigure}[b]{0.6\textwidth}
\includegraphics[width=\textwidth]{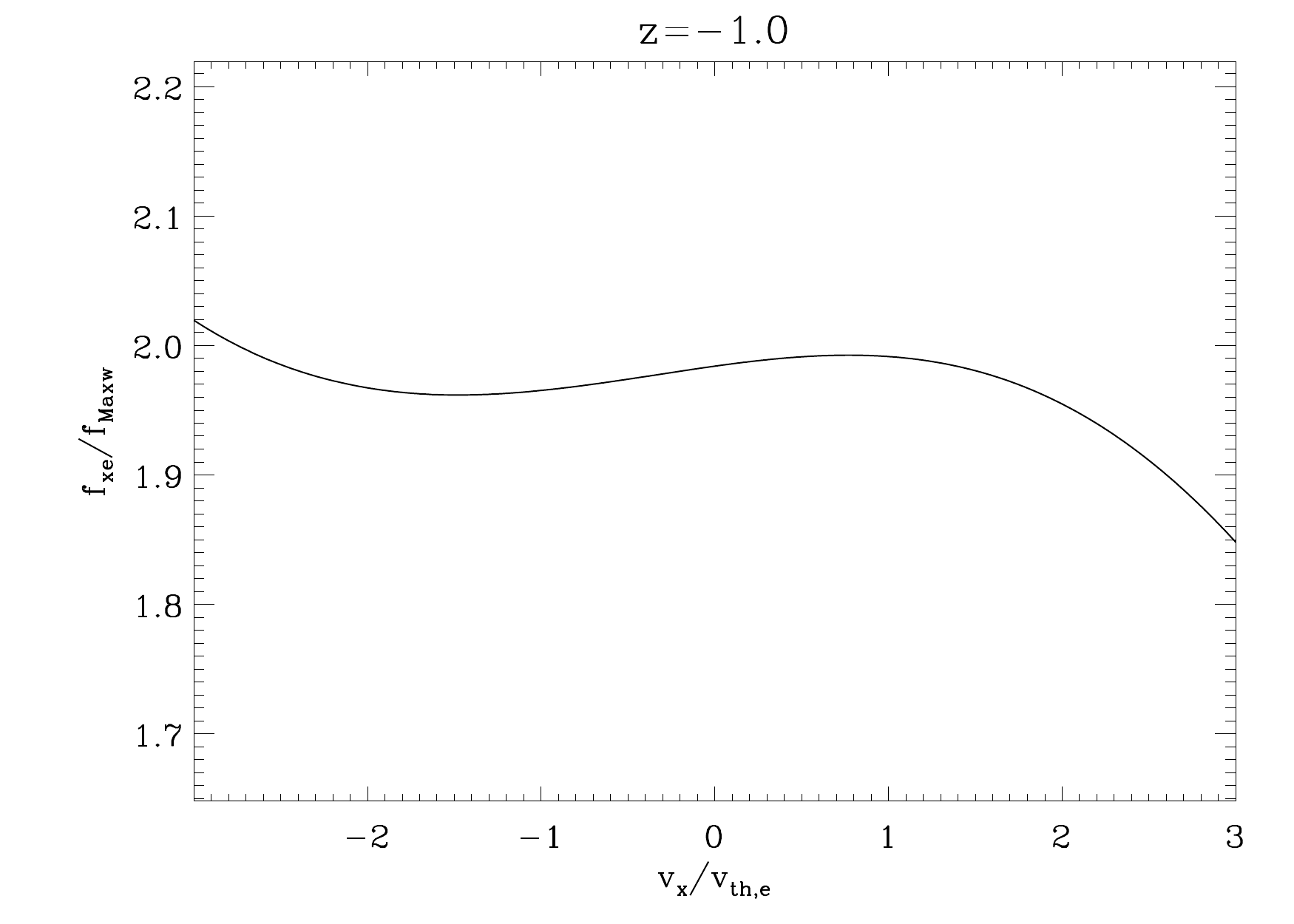}
 \caption{\small \label{fig:4bb}}
\end{subfigure}
\begin{subfigure}[b]{0.6\textwidth}
\includegraphics[width=\textwidth]{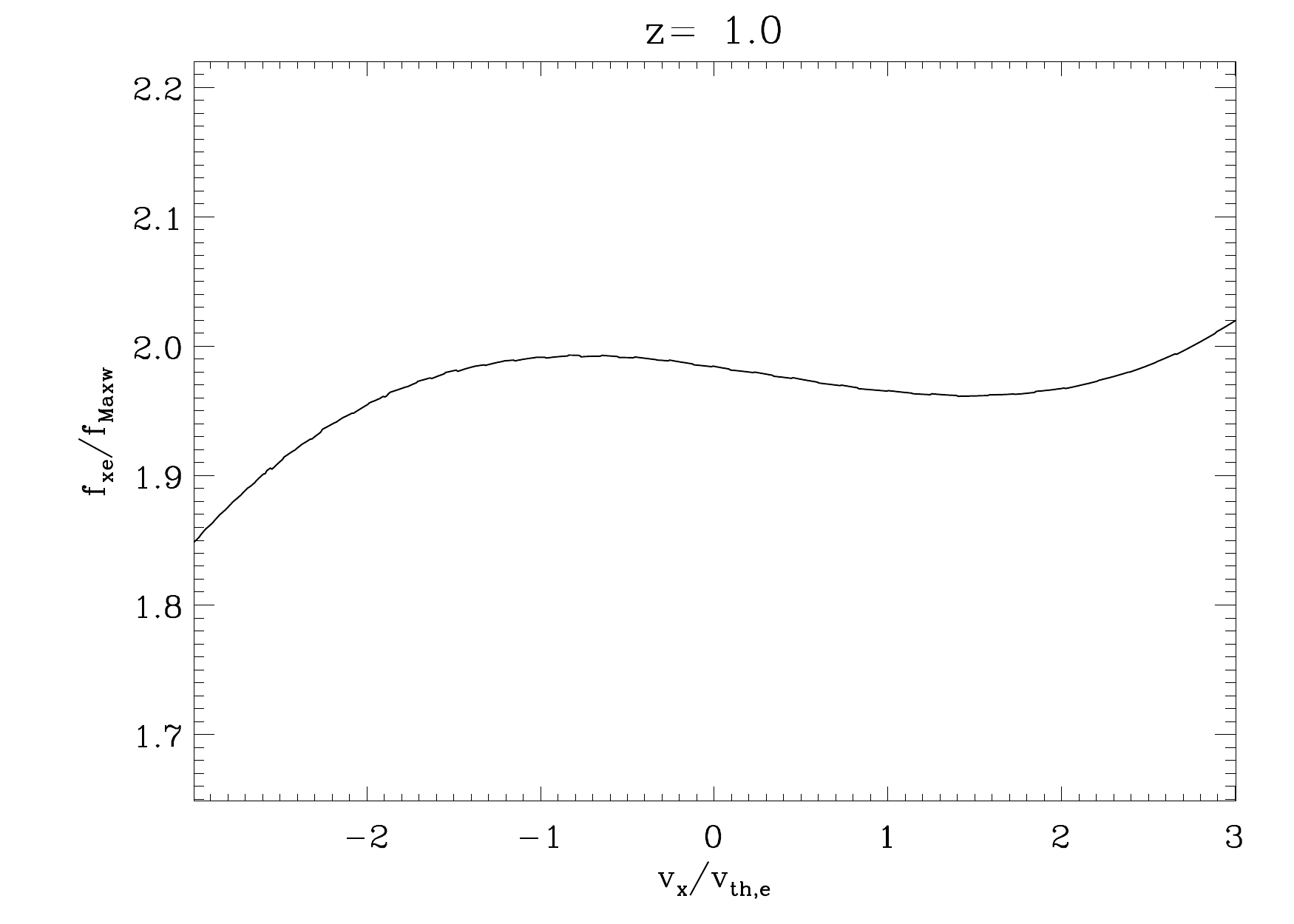}
 \caption{\small \label{fig:4cc}}
\end{subfigure}
\caption{\small The $v_x$ variation of $f_e/f_{Maxw,e}$ for $z/L=0$ (\ref{fig:4aa}),  $z/L=-1$ (\ref{fig:4bb}) and  $z/L=1$ (\ref{fig:4cc}). $\beta_{pl}=0.85$ and $\delta_e=0.15$. Note the antisymmetry of the $z=\pm 1$ plots with respect to each other.}
 \end{figure}

\begin{figure}
\centering
\begin{subfigure}[b]{0.6\textwidth}
\includegraphics[width=\textwidth]{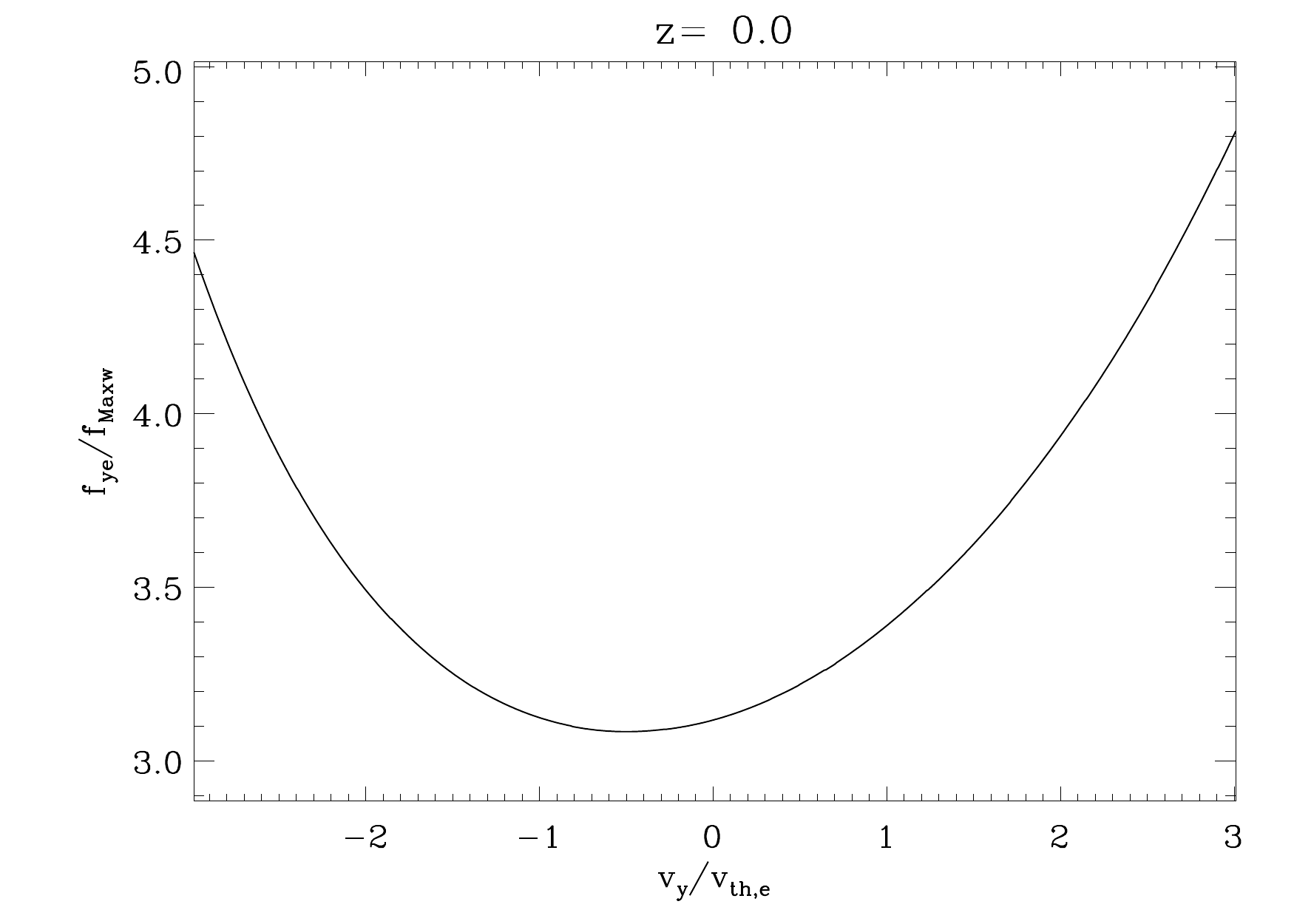}
 \caption{\small \label{fig:5aa}}
\end{subfigure}
\begin{subfigure}[b]{0.6\textwidth}
\includegraphics[width=\textwidth]{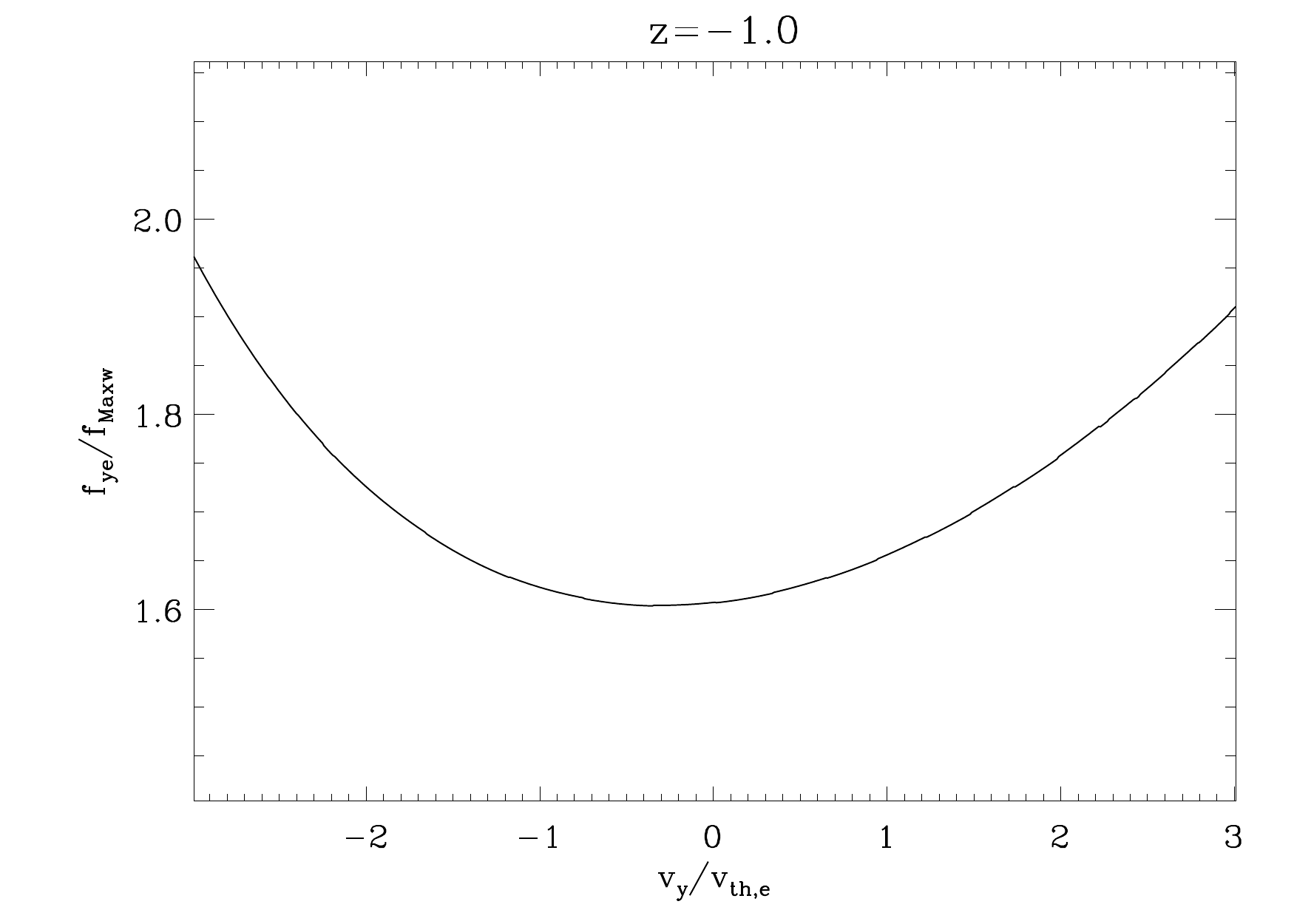}
 \caption{\small \label{fig:5bb}}
\end{subfigure}
\begin{subfigure}[b]{0.6\textwidth}
\includegraphics[width=\textwidth]{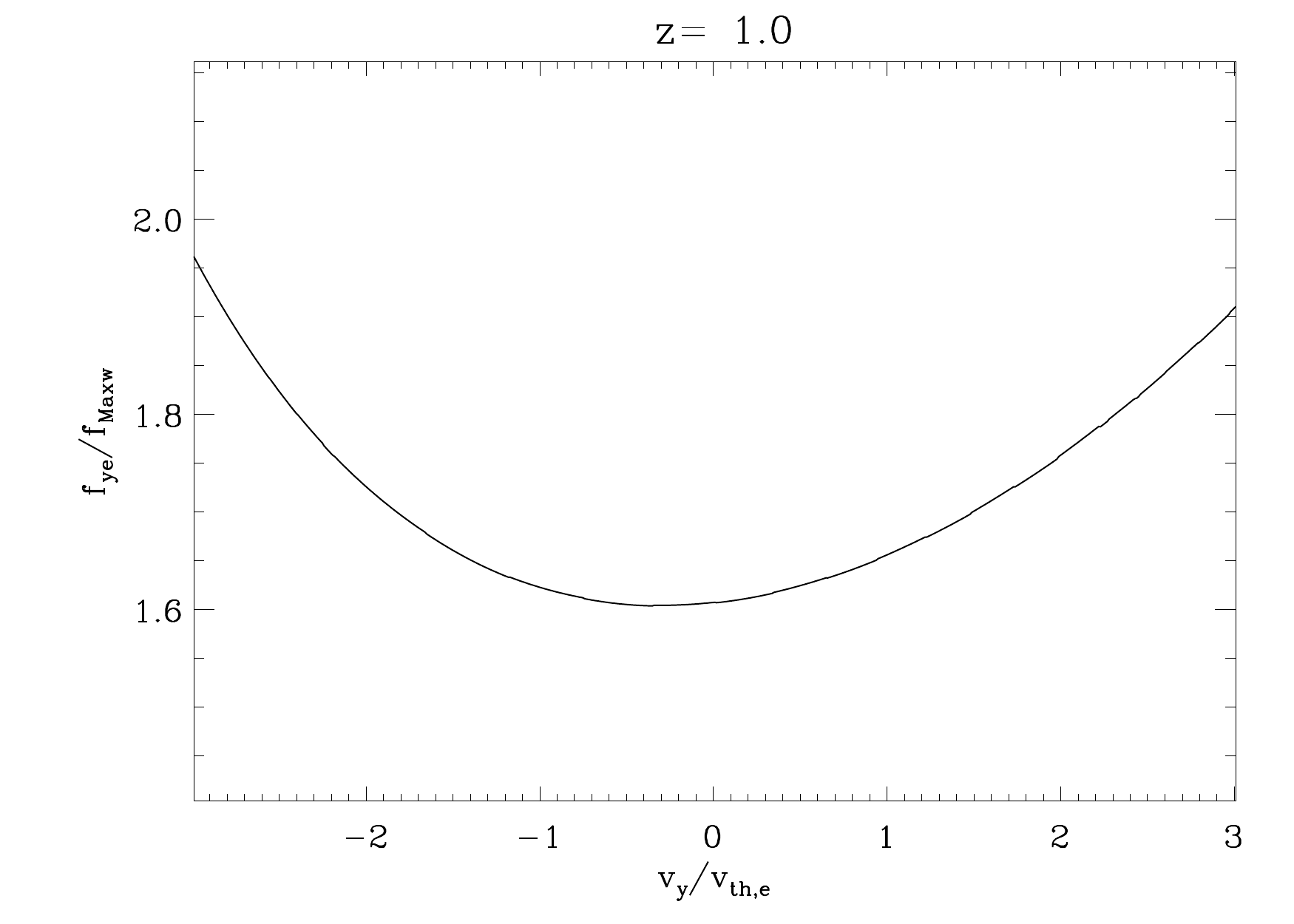}
 \caption{\small \label{fig:5cc}}
\end{subfigure}
\caption{\small The $v_y$ variation of $f_e/f_{Maxw,e}$ for $z/L=0$ (\ref{fig:5aa}),  $z/L=-1$ (\ref{fig:5bb}) and  $z/L=1$ (\ref{fig:5cc}). $\beta_{pl}=0.85$ and $\delta_e=0.15$. Note the symmetry of the $z=\pm 1$ plots with respect to each other.}
 \end{figure}

\begin{figure}
\centering
\begin{subfigure}[b]{0.6\textwidth}
\includegraphics[width=\textwidth]{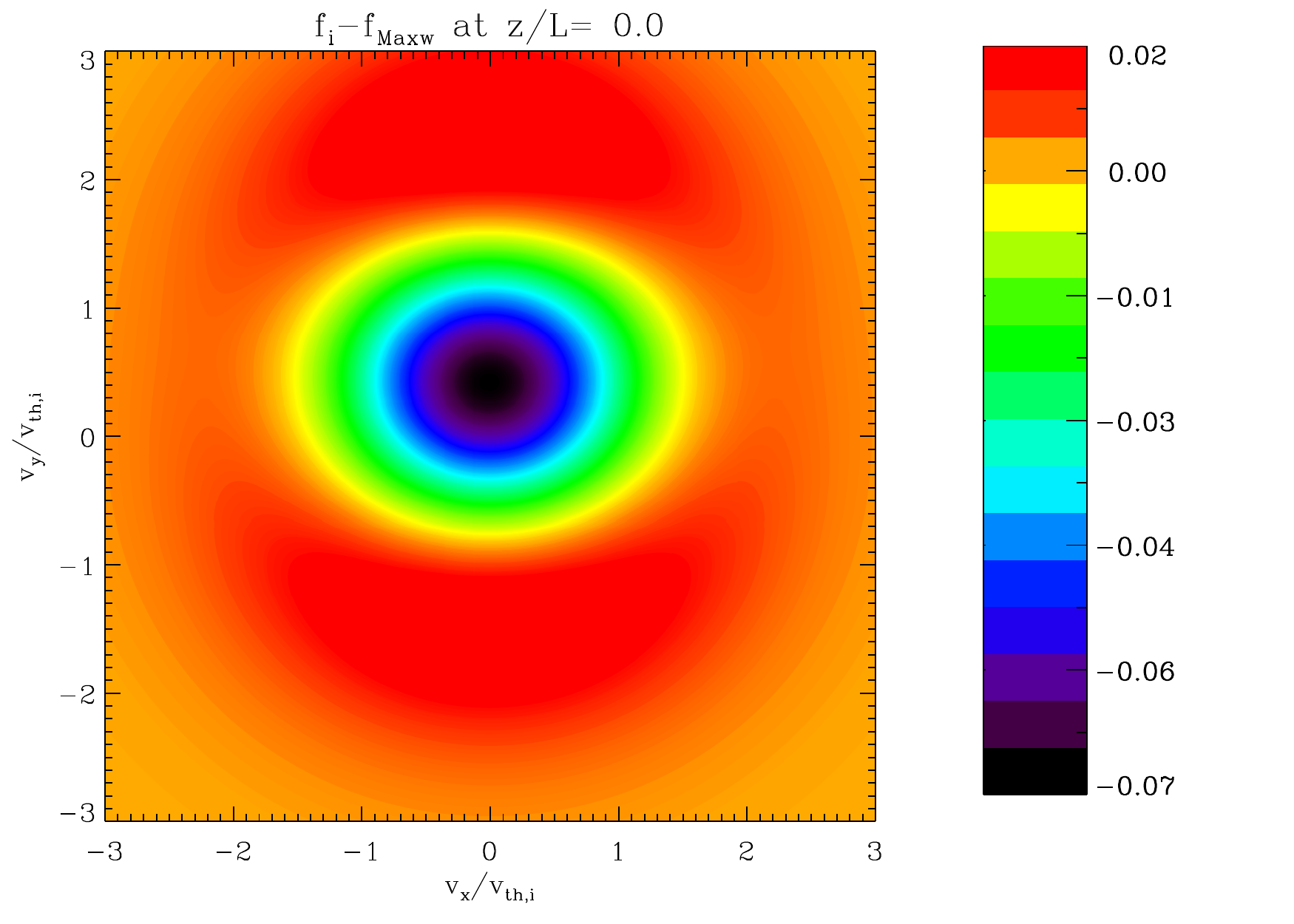}
 \caption{\small \label{fig:6aa}}
\end{subfigure}
\begin{subfigure}[b]{0.6\textwidth}
\includegraphics[width=\textwidth]{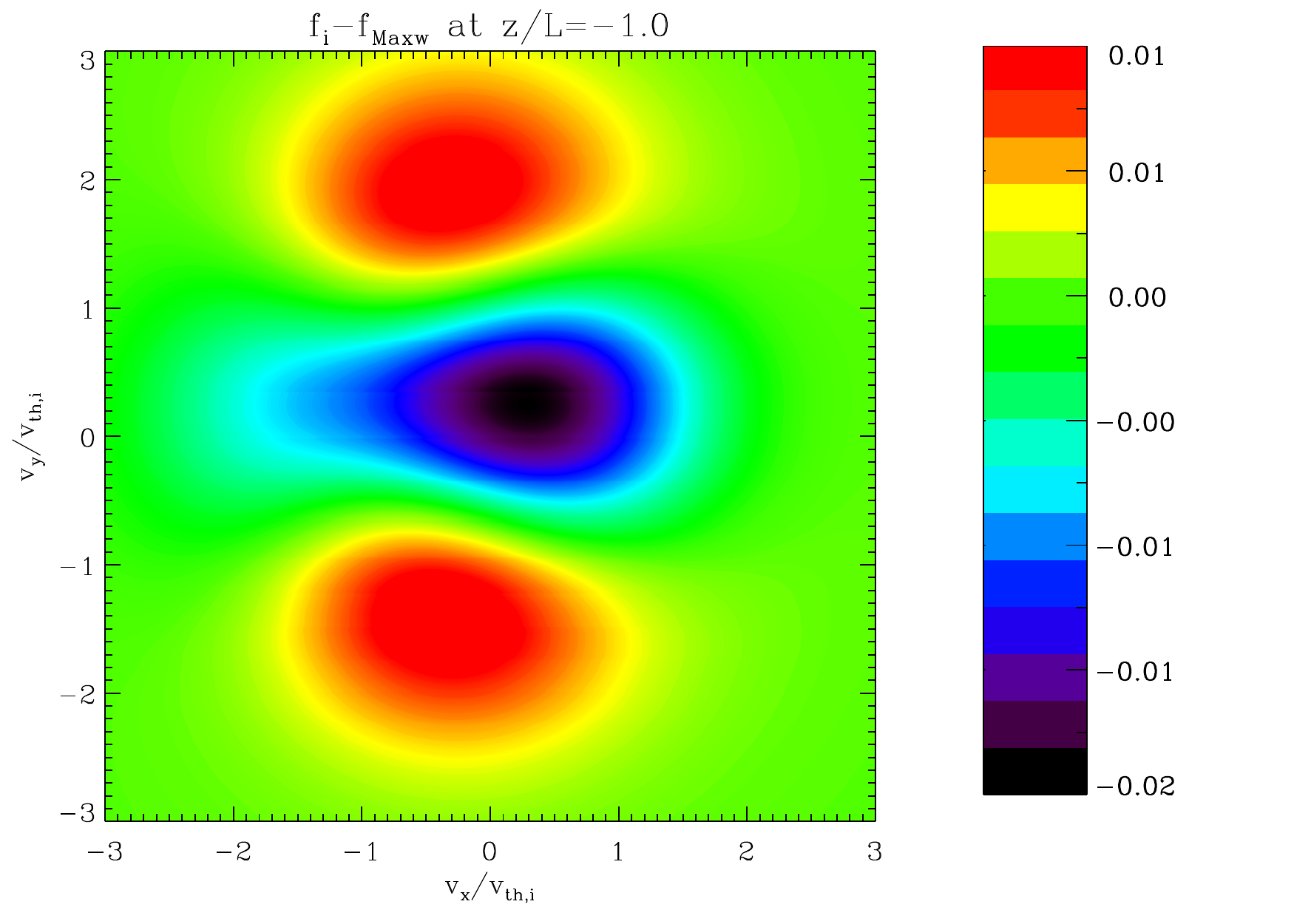}
 \caption{\small \label{fig:6bb}}
\end{subfigure}
\begin{subfigure}[b]{0.6\textwidth}
\includegraphics[width=\textwidth]{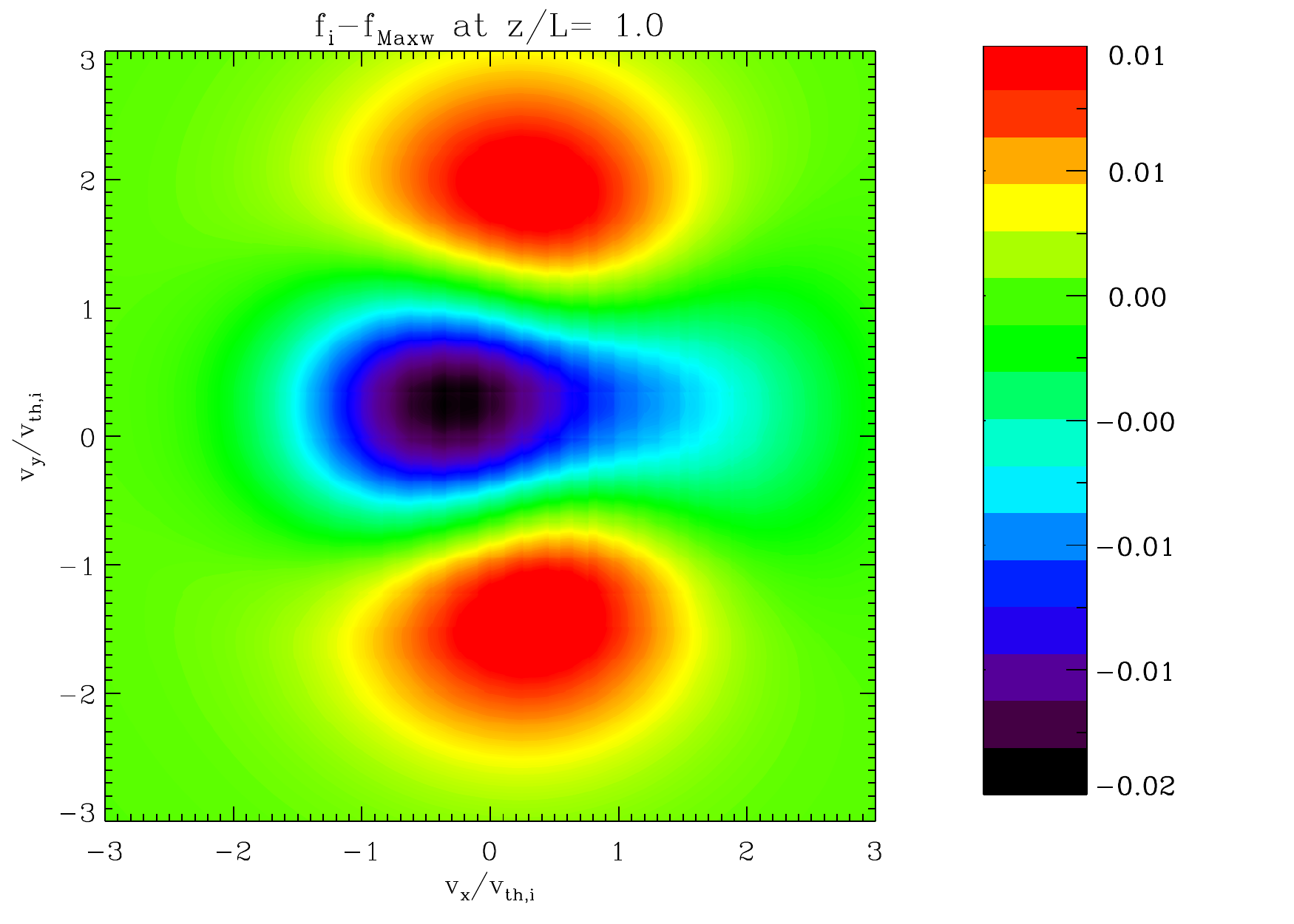}
 \caption{\small \label{fig:6cc}}
\end{subfigure}
\caption{\small Contour plots of $f_i-f_{Maxw,i}$ for $z/L=0$ (\ref{fig:6aa}),  $z/L=-1$ (\ref{fig:6bb}) and  $z/L=1$ (\ref{fig:6cc}). $\beta_{pl}=0.85$ and $\delta_i=0.15$. Note the antisymmetry of the $z=\pm 1$ plots with respect to each other.}
 \end{figure}

\begin{figure}
\centering
\begin{subfigure}[b]{0.6\textwidth}
\includegraphics[width=\textwidth]{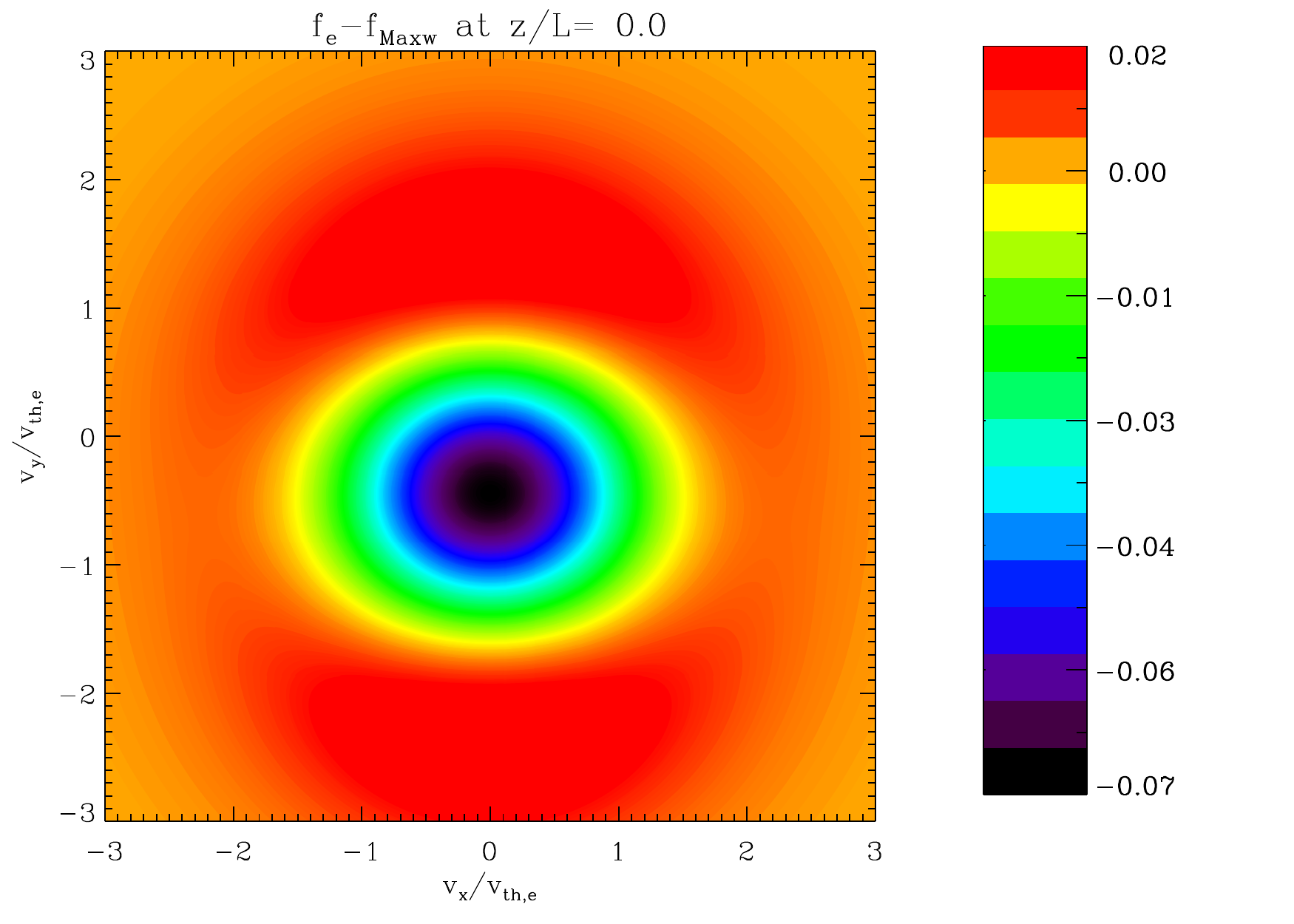}
 \caption{\small \label{fig:7aa}}
\end{subfigure}
\begin{subfigure}[b]{0.6\textwidth}
\includegraphics[width=\textwidth]{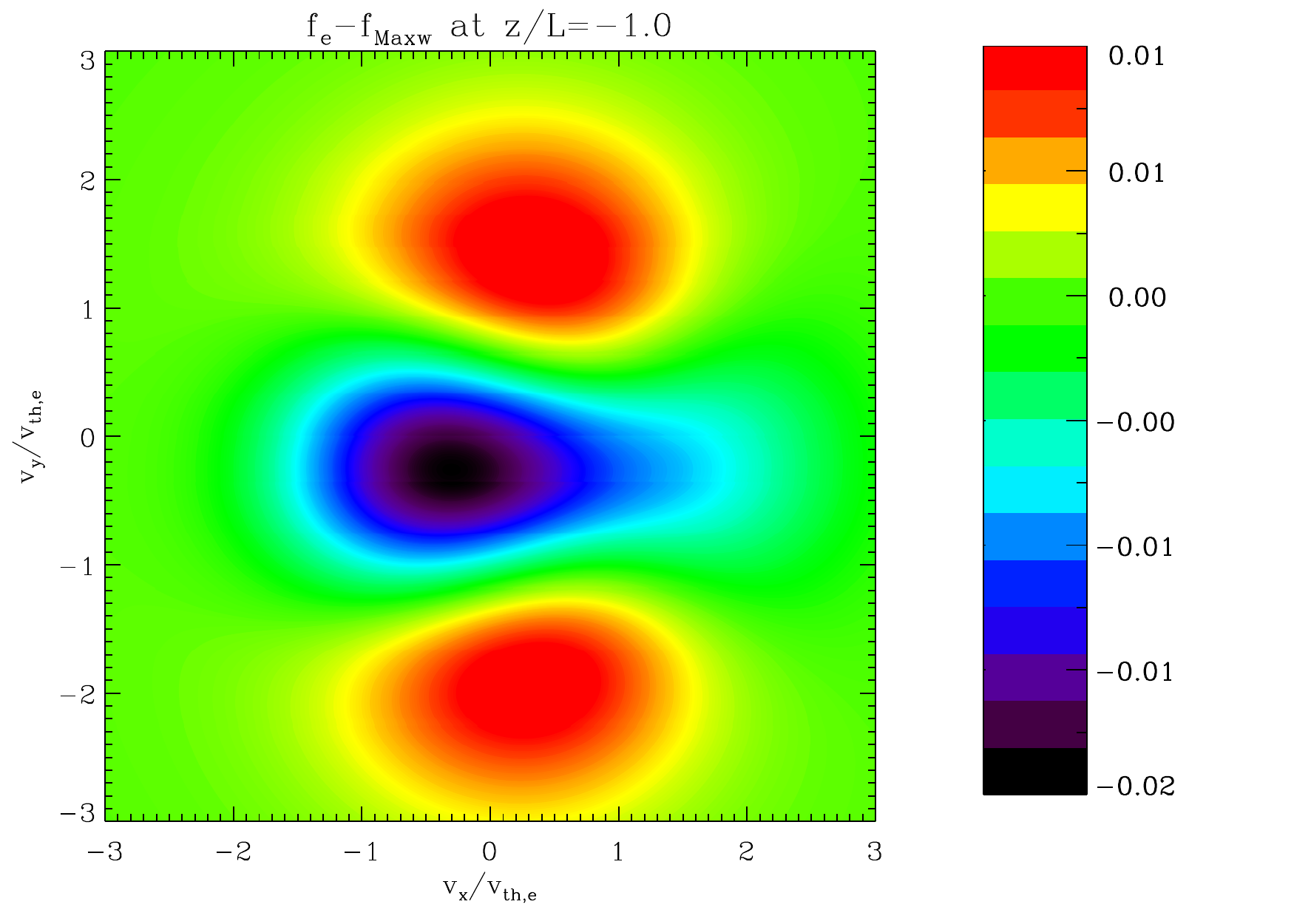}
 \caption{\small \label{fig:7bb}}
\end{subfigure}
\begin{subfigure}[b]{0.6\textwidth}
\includegraphics[width=\textwidth]{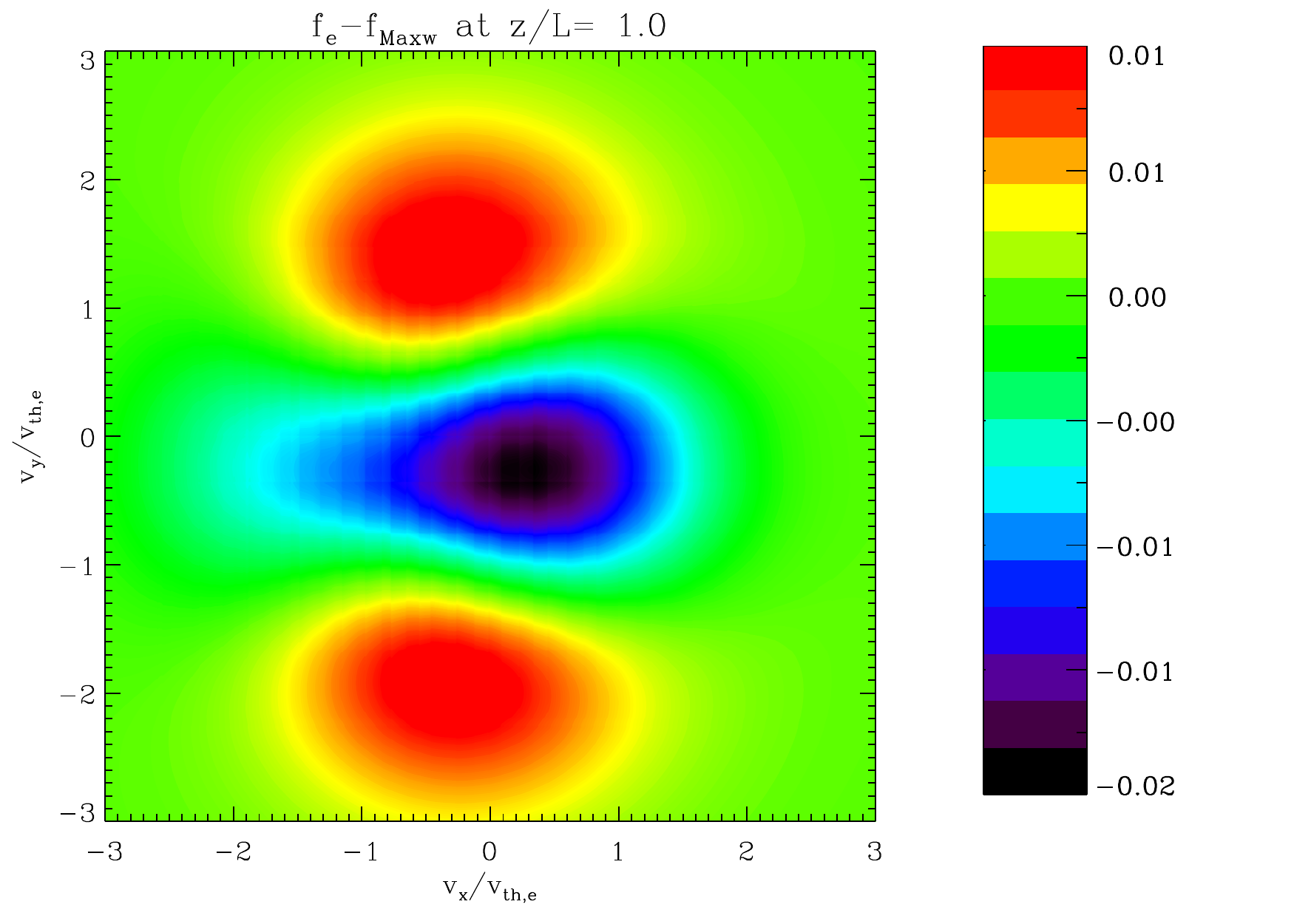}
 \caption{\small \label{fig:7cc}}
\end{subfigure}
\caption{\small Contour plots of $f_e-f_{Maxw,e}$ for $z/L=0$ (\ref{fig:7aa}),  $z/L=-1$ (\ref{fig:7bb}) and  $z/L=1$ (\ref{fig:7cc}). $\beta_{pl}=0.85$ and $\delta_e=0.15$. Note the antisymmetry of the $z=\pm 1$ plots with respect to each other.}
 \end{figure}

\section{`Re-gauged' equilibrium DF for the FFHS}\label{sec:newdf2}
\subsection{On the gauge for the vector potential}\label{Subsec:regauge}
In Section \ref{sec:newdf1} we used the pressure transformation techniques to derive a pressure tensor of `multiplicative form'
\[P_{zz}=P_1(A_x)P_2(A_y),\]
in order to construct a DF self-consistent with any value of the $\beta_{pl}$. However, the exact form of the DF was challenging to calculate numerically for low $\beta_{pl}$, with plots for $\beta_{pl}$ only modestly below unity presented ($\beta_{pl}=0.85$). The `problem terms' are those that depend on $p_{xs}$. The specific problem is that the $A_x$ function in the original gauge is neither even or odd, 
\[A_x=2B_0L\arctan \left(\exp\left(\frac{z}{L}\right)\right),\]
and as a result the range of $p_{xs}$ for which it is necessary to numerically calculate a convergent DF can be obstructive, say over a symmetric range in velocity space. Equation (\ref{eq:normalisedHn}) shows us that when $A_x$ is neither even nor odd, then $|p_{xs}|$ can take on larger than `necessary' values for a given $v_x$.

In this chapter, we shall `re-gauge' the vector potential component $A_x$ to be an odd function, 
\begin{equation}A_x=2B_0L\arctan\left(\tanh\left(\frac{z}{2L}\right)\right),\label{eq:gauge}\end{equation}
which is commensurate with $B_y$ being an even function and results in the same $B_y=B_0\,{\rm sech}(z/L)$ as the one derived from the $A_x$ defined in (\ref{eq:Agauge1}). As a consequence the numerical calculation of the DFs that we shall calculate for the FFHS becomes easier in the low $\beta_{pl}$ regime.

 \subsection{DF for the `re-gauged' FFHS: $\beta_{pl}\in (0\,,\,\infty)$ }\label{subsubsec:ANWTregaugemult}
We will now calculate a multiplicative DF for the `re-gauged' FFHS, in the same style as in Section \ref{sec:newdf1}, in the effort to produce a low-beta DF for the FFHS that is easier to calculate numerically, and plot. The new gauge is defined by
\begin{equation}
\boldsymbol{A}=B_0L\left(2\arctan\left(\tanh\left(\frac{z}{2L}\right)\right),\ln\text{sech}\frac{z}{L},0\right). 
\end{equation}
This re-gauging is equivalent to adding a constant to $A_{x}$ and so corresponds to a shift in the origin of the $A_x$ dependent part of the summative $P_{zz}$ used in \citet{Harrison-2009PRL}. As a result, one can derive a new summative pressure function in the same manner as in \citep{Harrison-2009PRL}, corresponding to this new gauge, as 
\begin{equation}
P_{zz}=\frac{B_0^2}{2\mu_0}\left[\sin ^2\left(\frac{A_x}{B_0L}\right) + \exp\left(\frac{2A_y}{B_0L}\right)   \right]\label{eq:pnewgauge}
\end{equation}
The next step is to construct a multiplicative pressure tensor. Using the same pressure transformation technique as in Section \ref{sec:exppressure}, on the $P_{zz}$ given in Equation (\ref{eq:pnewgauge}), we arrive at the `re-gauged' multiplicative pressure
\begin{eqnarray}P_{zz}&=&P_0e^{-1/\beta_{pl}}\exp\left[ \frac{1}{\beta_{pl}}\left( \sin ^{2}\left(\frac{A_x}{B_0L}\right)+\exp\left(\frac{2A_y}{B_0L}\right)   \right)         \right]  \\
&=&P_0\exp{ \left[ \sum_{n=1}^\infty\frac{1}{(2n)!} \nu_{2n}\left(\frac{A_x}{B_0L}\right)^{2n}   \right]                }  \exp{ \left[\sum_{n=1}^\infty \frac{1}{n!} \xi_n \left(\frac{A_y}{B_0L}\right)^n  \right]                }     ,\end{eqnarray} 
with the coefficients defined by
\[ \nu_{2n}=\frac{(-1)^{n+1}2^{2n-1}}{\beta_{pl}}  ,\hspace{5mm}\xi_n=\frac{2^n}{\beta_{pl}}\, .\]
We now use the theory of CBPs, as in \citep{Allanson-2015POP} and Section \ref{sec:newdf1}, to write the pressure as
\begin{eqnarray}
P_{zz}&=&P_0\sum_{m=0}^\infty \frac{1}{(2m)!}Y_{2m}\left(0\,,\,\nu_2\,,\,0\,,\, \nu_4\, ,\, ...\,,0\,,\, \nu_{2m}\right)\left(\frac{A_x}{B_0L}\right)^{2m}\nonumber\\
&\times&\sum_{n=0}^\infty \frac{1}{n!}Y_{n}\left(\xi_1\,,\,\xi_2\,,\, ...\,,\,\xi_n        \right)\left(\frac{A_y}{B_0L}\right)^{n}.\nonumber
\end{eqnarray}
Once again using the simple scaling argument from Equation (\ref{eq:CBPscale}), we have
\begin{eqnarray}
P_{zz}&=&P_0\sum_{m=0}^\infty \frac{(-1)^m2^{2m}}{(2m)!}Y_{2m}\left(0\,\,\frac{-1}{2\beta_{pl}}\,,\,0\,,\, \frac{-1}{2\beta_{pl}}\, ,\, ...\,,\,0\,,\, \frac{-1}{2\beta_{pl}}\right)\left(\frac{A_x}{B_0L}\right)^{2m}\nonumber\\
&\times&\sum_{n=0}^\infty \frac{2^m}{n!}Y_{n}\left(\frac{1}{\beta_{pl}}\,,\,\frac{1}{\beta_{pl}}\,,\, ...\,,\,\frac{1}{\beta_{pl}}          \right)\left(\frac{A_y}{B_0L}\right)^{n}.\nonumber
\end{eqnarray}
Using the methods established in Chapter \ref{Vlasov}, namely expansion over Hermite polynomials, we calculate a DF that gives the above pressure
\begin{eqnarray}
&&f_{s}=\frac{n_0}{(\sqrt{2\pi}v_{\text{th},s})^3}e^{-\beta_sH_s}\times\nonumber\\
&&\sum_{m=0}^\infty a_{2m}\left(\frac{\delta_s}{\sqrt{2}}\right)^{2m}H_{2m}\left(\frac{p_{xs}}{\sqrt{2}m_{s}v_{\text{th},s}}\right)\times\nonumber\\
&&\sum_{n=0}^\infty b_n{\rm sgn}(q_{s})^n\left(\frac{\delta_s}{\sqrt{2}}\right)^nH_n\left(\frac{p_{ys}}{\sqrt{2}m_{s}v_{\text{th},s}}\right),\label{eq:gammaANWT}
\end{eqnarray}
for
\begin{eqnarray}a_{2m}&=&\frac{(-1)^m2^{2m}}{(2m)!}Y_{2m}\left(0\,\,\frac{-1}{2\beta_{pl}}\,,\,0\,,\, \frac{-1}{2\beta_{pl}}\, ,\, ...\,,\,0\, ,\, \frac{-1}{2\beta_{pl}}\right)\nonumber ,\\
b_n&=&\frac{2^m}{n!}Y_{n}\left(\frac{1}{\beta_{pl}}\,,\,\frac{1}{\beta_{pl}}\,,\, ...\,,\,\frac{1}{\beta_{pl}}          \right).
\end{eqnarray}
One can readily calculate the number density for this DF using standard integral results \citep{Gradshteyn} to be 
\[n_s(A_{x},A_{y})=n_0\sum_{m=0}^\infty a_{2m}\left(\frac{A_x}{B_0L}\right)^{2m}\sum_{n=0}^\infty b_{n}\left(\frac{A_y}{B_0L}\right)^n= P_0\frac{\beta_e\beta_i}{\beta_e+\beta_i}.\]

\subsection{ Convergence and boundedness of the DF }
This DF has identical coefficients for the $p_{ys}$-dependent Hermite polynomials as that derived in Section \ref{sec:newdf1}, and so we need not verify convergence for that series. In fact, all that has changed in the analysis of the coefficients for the $p_{xs}$-dependent sum is that we now have to consider the Maclaurin coefficients of $\sin ^2(A_x/(B_0L))$ as opposed to $\cos(2A_x/(B_0L))$. These Maclaurin coefficients both have the same `factorial dependence' and as such the convergence of the one DF implies the convergence of the other. 


The boundedness argument is exactly analogous to that made above for the DF in original gauge, and need not be repeated here.

\subsection{Plots of the DF}\label{sec:plotsgauge2}
We now present plots for the DF given in Equation (\ref{eq:gammaANWT}), for $\beta_{pl}=0.05$ and $\delta_e=\delta_i=0.03$. This value for $\beta_{pl}$ is substantially lower than the value used in Section \ref{sec:newdf1}, which had $\beta_{pl}=0.85$. The ability to go down to lower values of the plasma beta is due to the re-gauging process as explained in Section \ref{Subsec:regauge}. The plots that we show are intended to demonstrate progress in the numerical evaluation of low-beta DFs for nonlinear force-free fields, and as a proof of principle. 

The value of $\delta_s$ is chosen such that $\delta_s<\beta_{pl}$, since as explained in Section \ref{sec:plotsgauge1}, attaining convergence numerically has not been easy for values of $\delta_s>\beta_{pl}$ when $\beta_{pl}<1$.

Initial investigations of the shape of the variation of the DF in the $v_x$ and $v_y$ directions indicate that the DF seems to have a Gaussian profile, as in the DFs analysed in Section \ref{sec:newdf1}. Hence, as in that work, we shall compare the DFs calculated in this work to drifting Maxwellians, in order to measure the actual difference between the Vlasov equilibrium $f_{s}$, and the Maxwellian $f_{Maxw,s}$. In Figures (\ref{fig:jpp1a}-\ref{fig:jpp1e}) and (\ref{fig:jpp2a}-\ref{fig:jpp2e}) we give contour plots in $(v_x/v_{\text{th},s},v_y/v_{\text{th},s})$ space of the `raw' difference between the DFs defined by Equation (\ref{eq:gammaANWT}) and (\ref{eq:Maxshift}). These figures bear close resemblance to those presented in Section \ref{sec:plotsgauge1}. Specifically, we see `shallower' peaks for the exact Vlasov solution, $f_{s}$, than for $f_{Maxw,s}$. There is also a clear anisotropic effect in that $f_{s}$ falls off more quickly in the $v_x$ direction than in the $v_y$ direction as compared to $f_{Maxw,s}$. Note that whilst the raw differences plotted in these figures may not seem substantial, they can in fact be significant as a proportion of $f_{Maxw,s}$, and even of the order of the magnitude of $f_{Maxw,s}$. As a demonstration of this fact we present plots in Figures (\ref{fig:jpp3a}-\ref{fig:jpp3e}) and (\ref{fig:jpp4a}-\ref{fig:jpp4e}) of the quantity defined by
\[
f_{diff,s}=(f_{s}-f_{Maxw,s})/f_{Maxw,s}
\]
for line cuts through $(v_x/v_{\text{th},s},v_y/v_{\text{th},s}=0)$ and $(v_x/v_{\text{th},s}=0,v_y/v_{\text{th},s})$ respectively, for the ions. As suggested by the contour plots, $f_{diff,i}$ takes on significantly larger values in the $v_y$ direction, indicating that the tail of $f_i$ falls off less quickly than $f_{Maxw,i}$ in $v_y$ than in $v_x$.

We are yet to observe multiple peaks in the multiplicative DFs for the FFHS, derived herein and in Section \ref{sec:newdf1}. However, the summative Harrison-Neukirch equilibria \citep{Harrison-2009PRL} could develop multiple maxima for sufficiently large values of the magnitude of the drift velocities. For the DF derived in this chapter, and as in Section \ref{sec:newdf1}, the `amplitude' of the drift velocity profile across the current sheet is given by
\begin{equation}
\frac{u_{s}}{v_{\text{th},s}}=2\text{sgn}(q_{s})\frac{\delta_s }{\beta_{pl}},\nonumber
\end{equation}
where $u_s$ represents the maximum value of the drift velocities. As a result, large values of the drift velocity correspond to large values of $\delta_s/\beta_{pl}$, and these are exactly the regimes for which we are struggling to attain numerical convergence. This theory suggests that we may not be seeing DFs with multiple maxima because we are not in the appropriate parameter space. 

\begin{figure}
\centering
\begin{subfigure}[b]{0.48\textwidth}
\includegraphics[width=\textwidth]{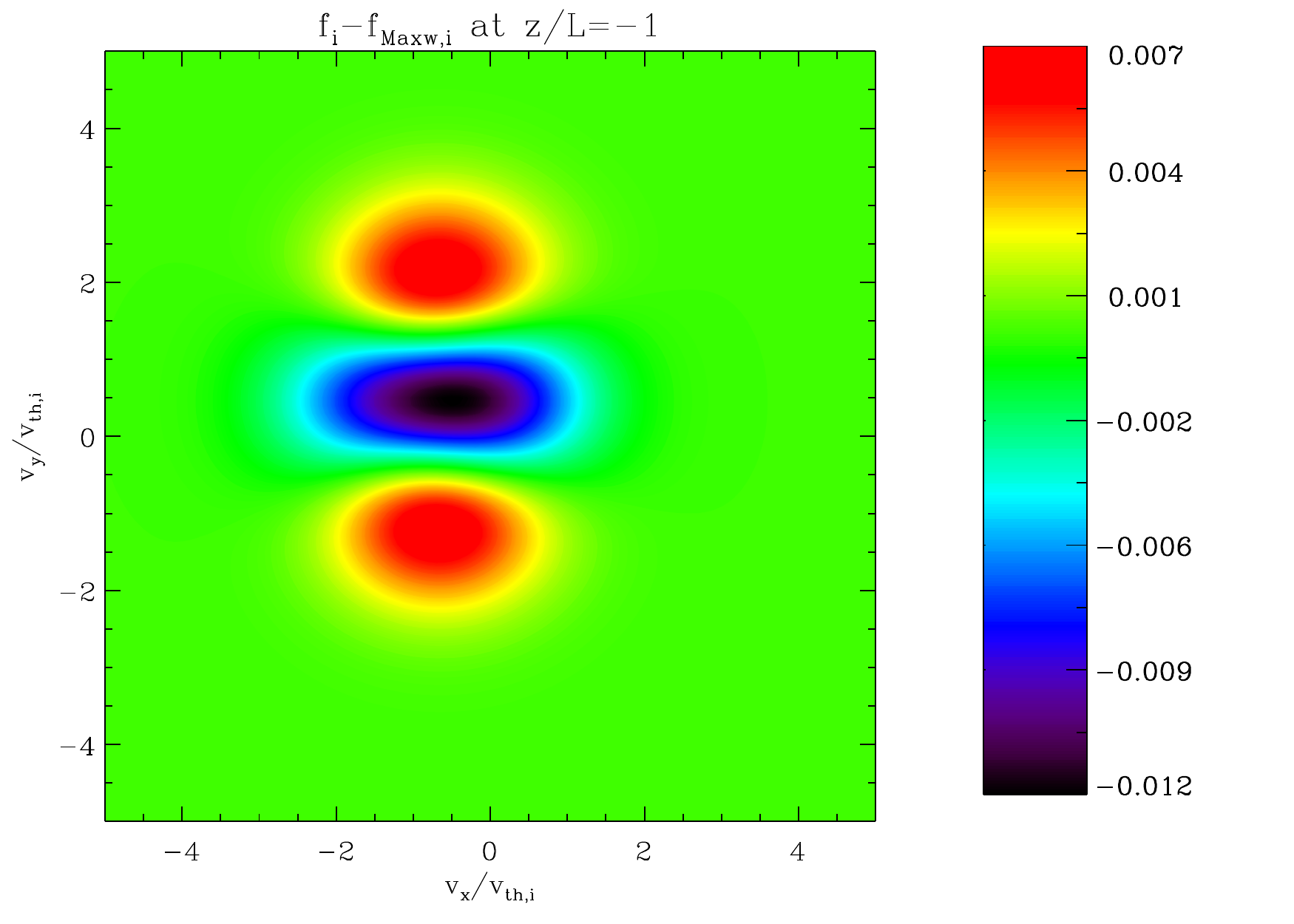}
 \caption{\small \label{fig:jpp1a}}
\end{subfigure}
\begin{subfigure}[b]{0.48\textwidth}
\includegraphics[width=\textwidth]{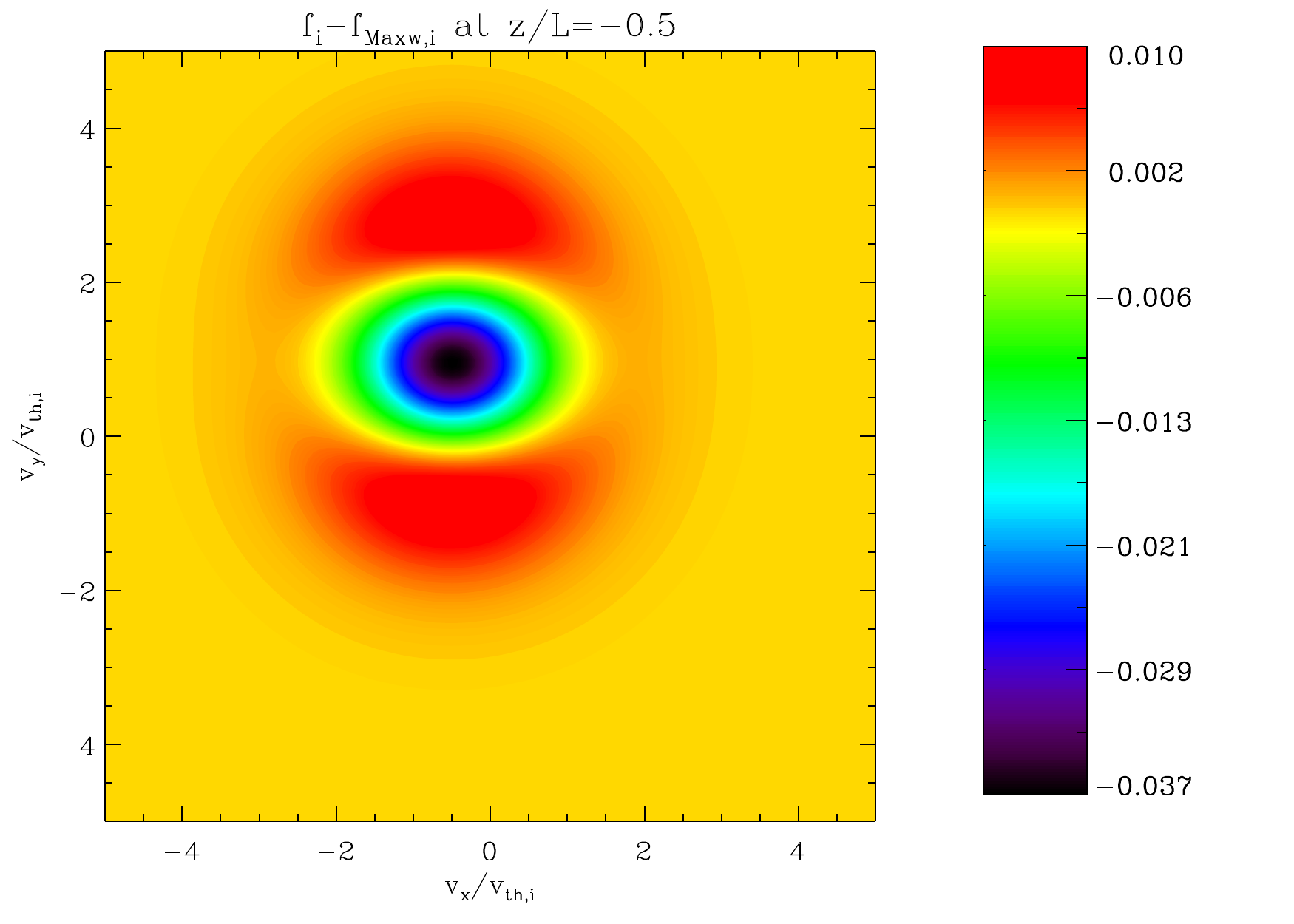}
 \caption{\small \label{fig:jpp1b}}
\end{subfigure}
\begin{subfigure}[b]{0.48\textwidth}
\includegraphics[width=\textwidth]{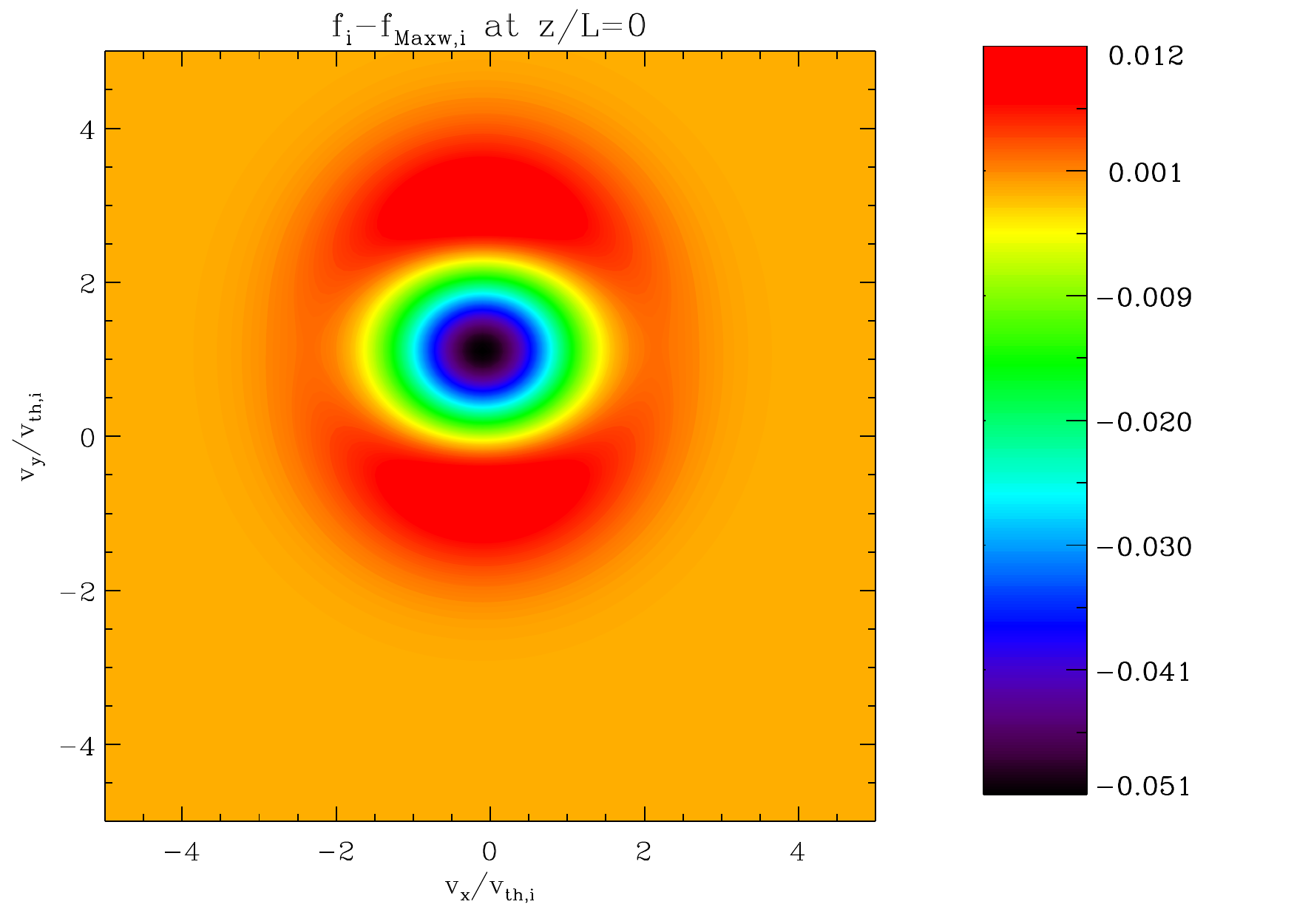}
 \caption{\small \label{fig:jpp1c}}
\end{subfigure}
\begin{subfigure}[b]{0.48\textwidth}
\includegraphics[width=\textwidth]{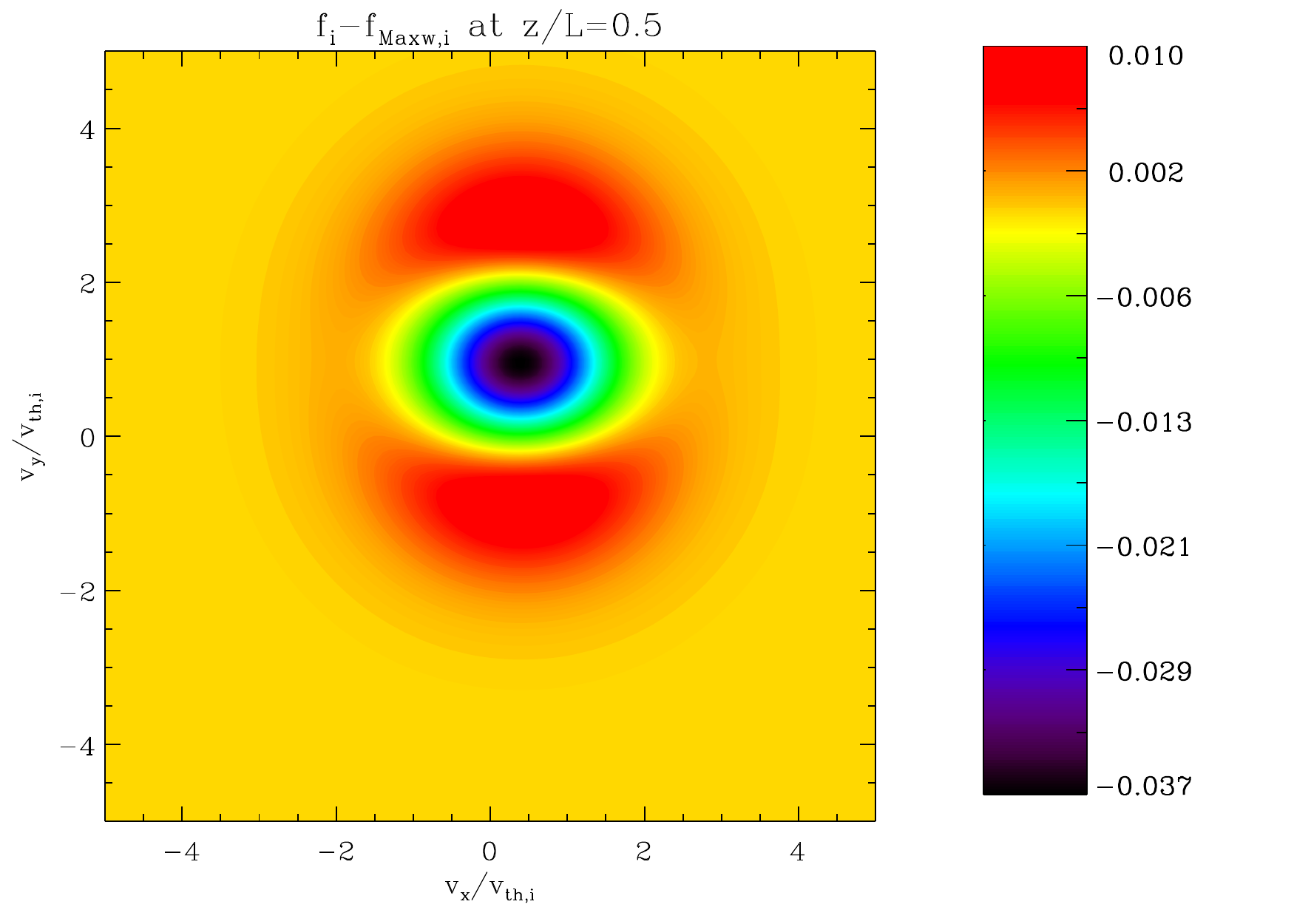}
 \caption{\small \label{fig:jpp1d}}
\end{subfigure}
\begin{subfigure}[b]{0.48\textwidth}
\includegraphics[width=\textwidth]{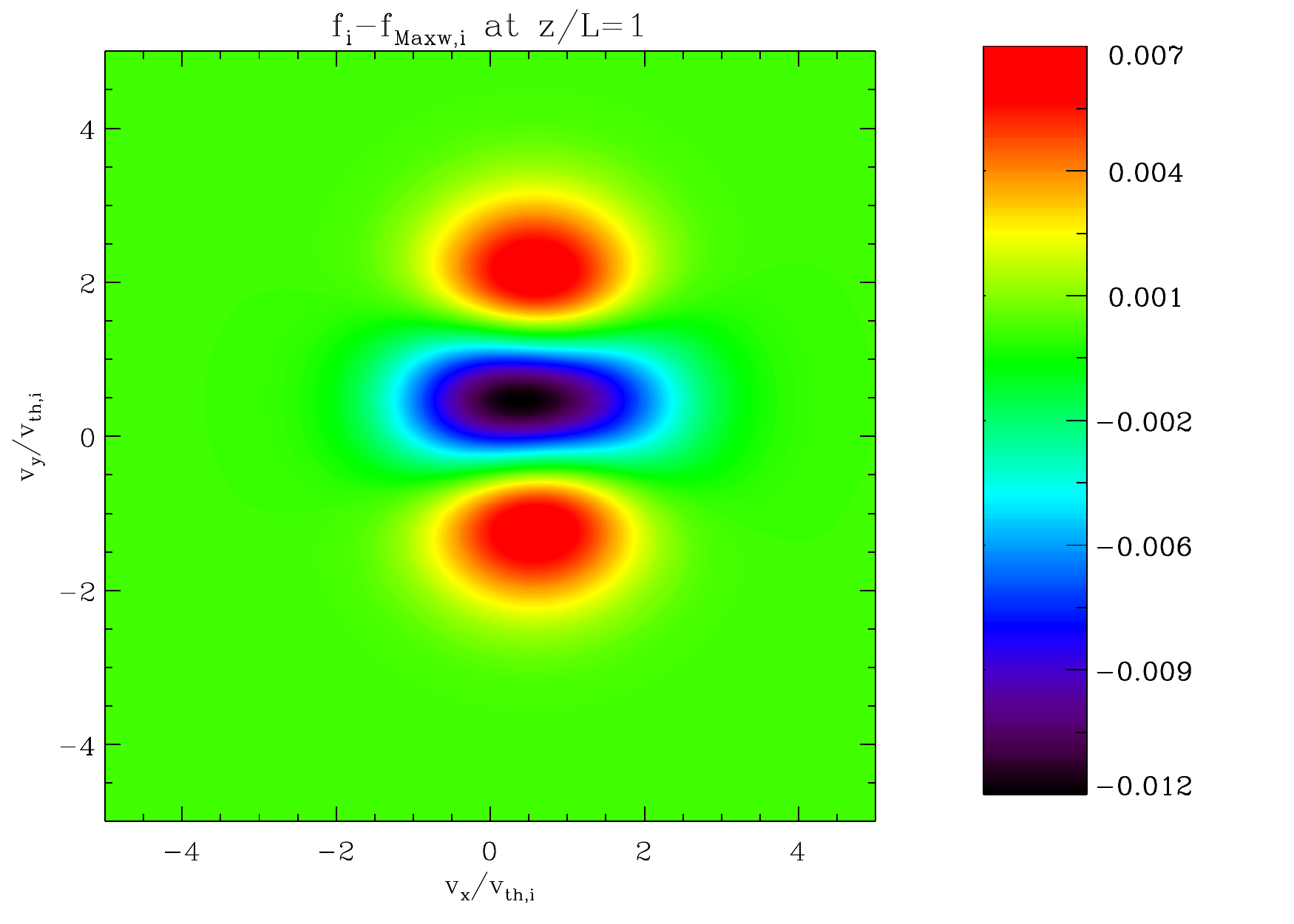}
 \caption{\small \label{fig:jpp1e}}
\end{subfigure}
\caption{\small Contour plots of $f_i-f_{Maxw,i}$ for $z/L=-1$ (\ref{fig:jpp1a}),  $z/L=-0.5$ (\ref{fig:jpp1b}), $z/L=0$ (\ref{fig:jpp1c}), $z/L=0.5$ (\ref{fig:jpp1d}) and $z/L=1$ (\ref{fig:jpp1e}). $\beta_{pl}=0.05$ and $\delta_i=0.03$. }
 \end{figure}

\begin{figure}
\centering
\begin{subfigure}[b]{0.48\textwidth}
\includegraphics[width=\textwidth]{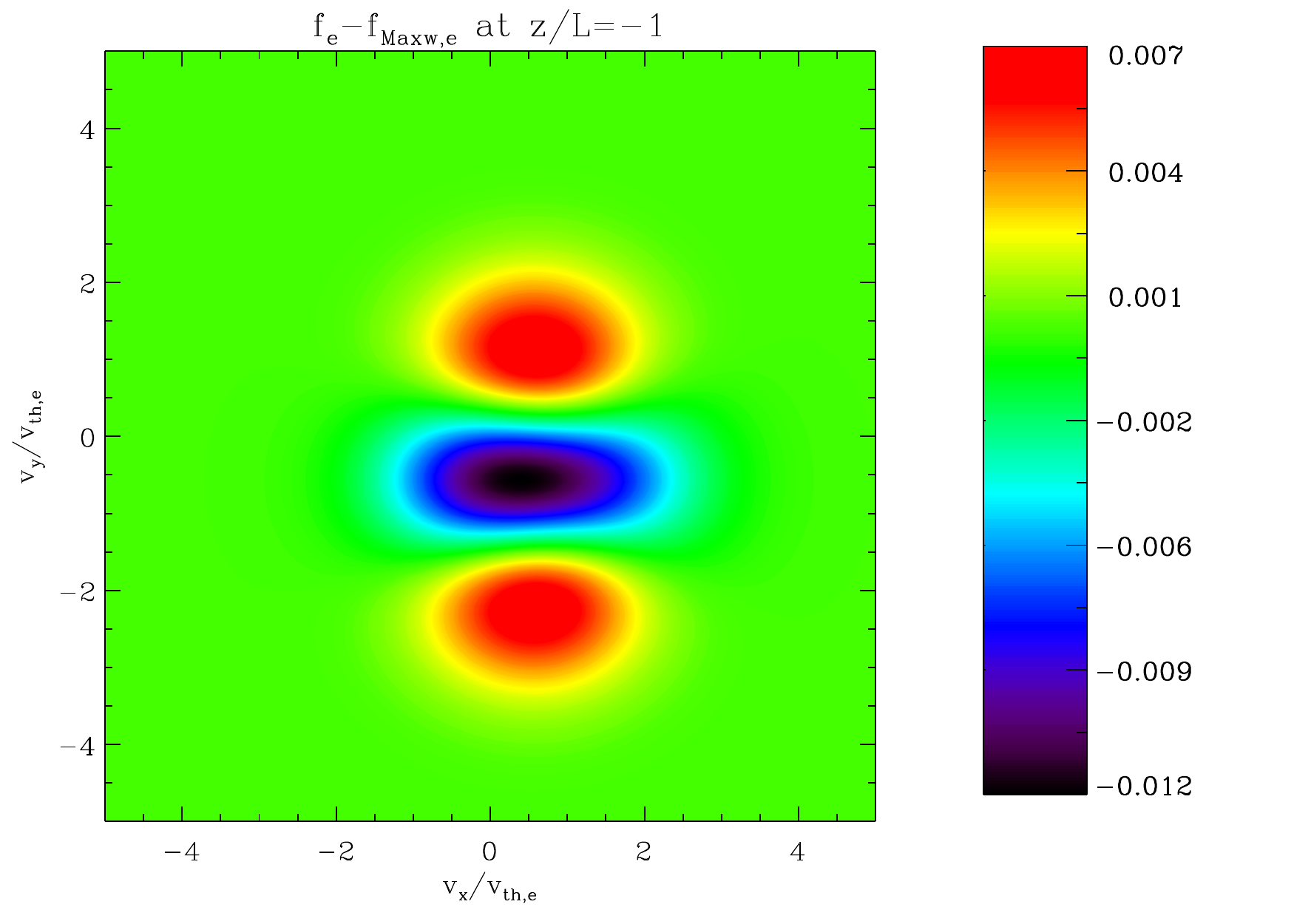}
 \caption{\small \label{fig:jpp2a}}
\end{subfigure}
\begin{subfigure}[b]{0.48\textwidth}
\includegraphics[width=\textwidth]{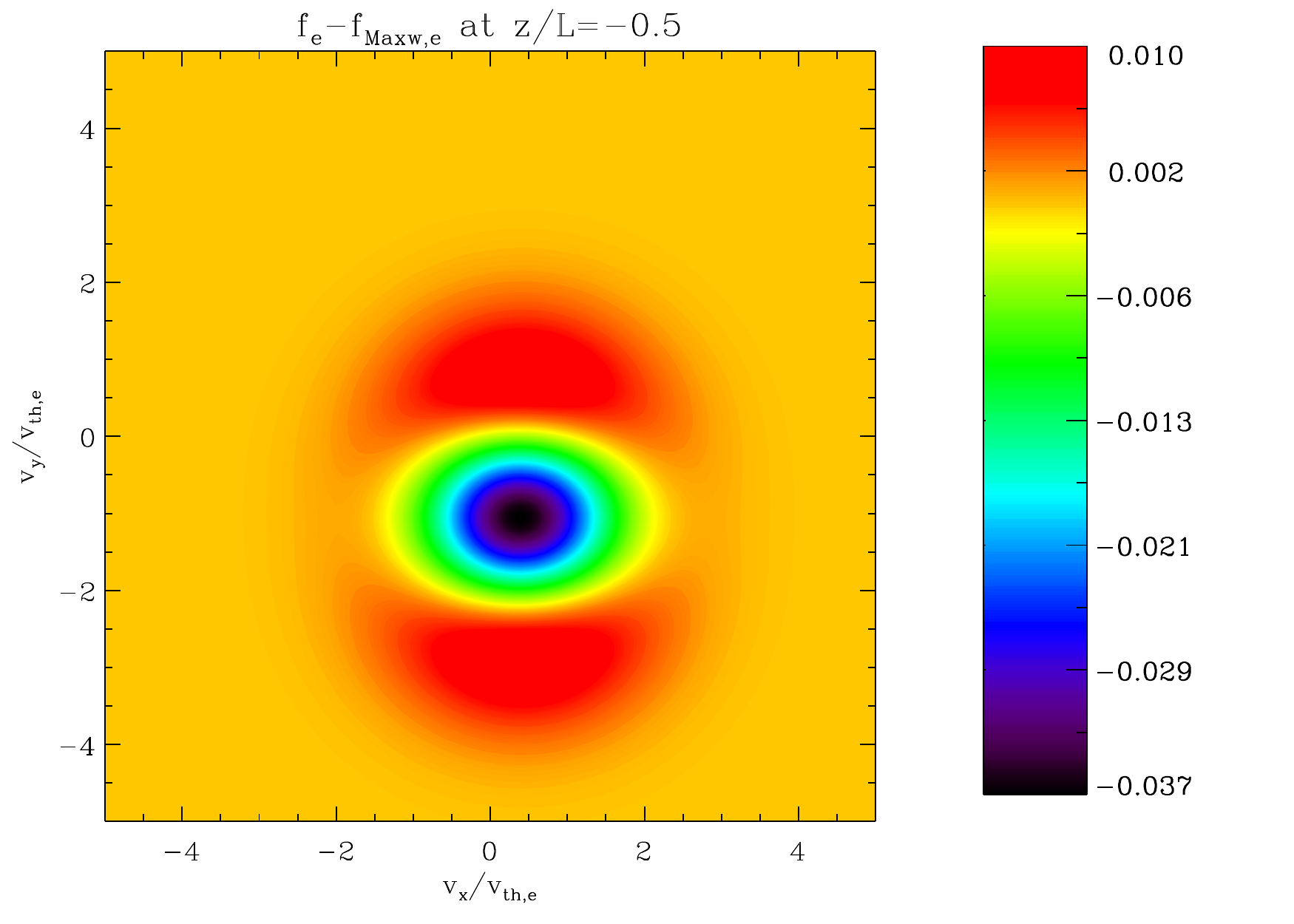}
 \caption{\small \label{fig:jpp2b}}
\end{subfigure}
\begin{subfigure}[b]{0.48\textwidth}
\includegraphics[width=\textwidth]{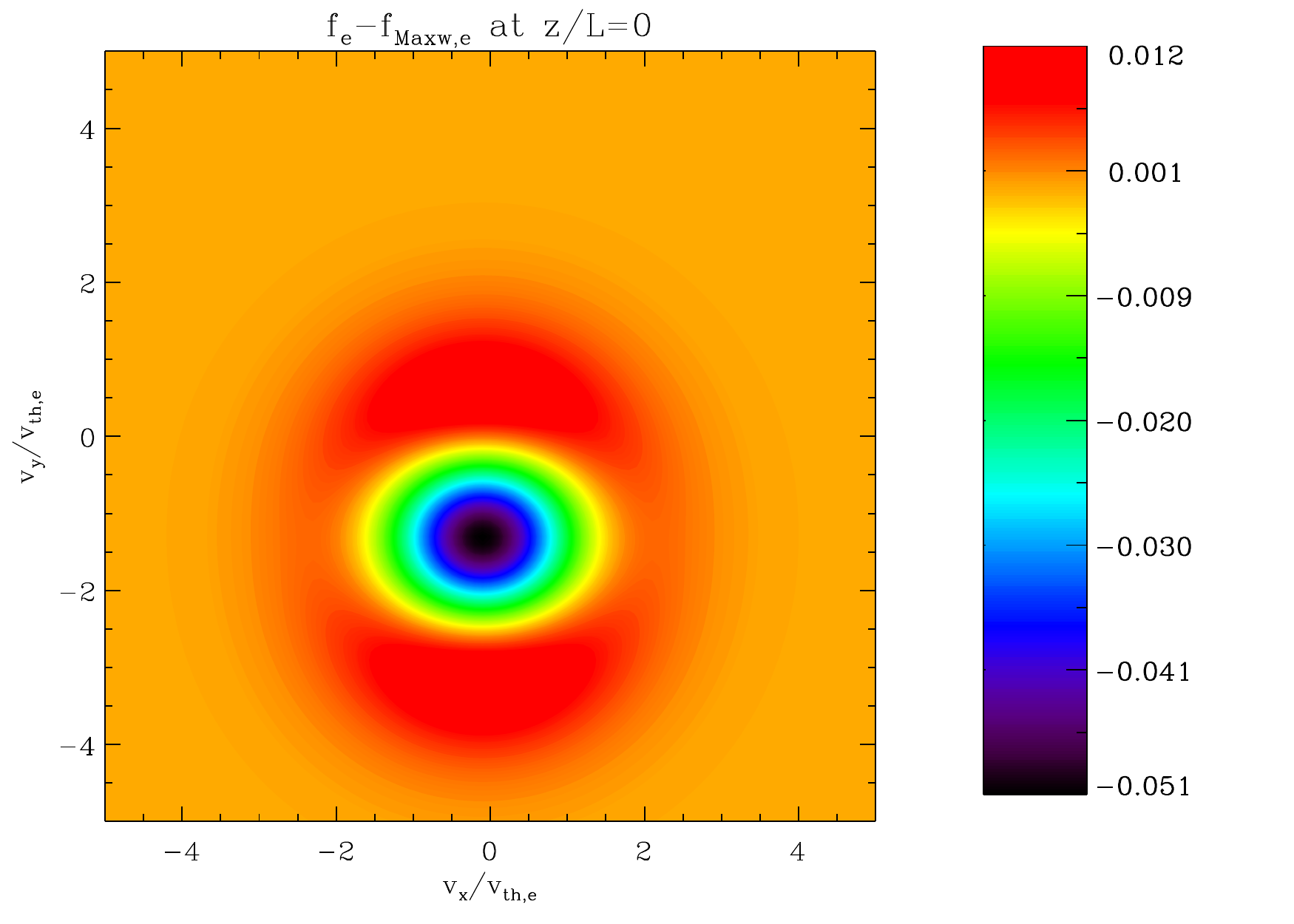}
 \caption{\small \label{fig:jpp2c}}
\end{subfigure}
\begin{subfigure}[b]{0.48\textwidth}
\includegraphics[width=\textwidth]{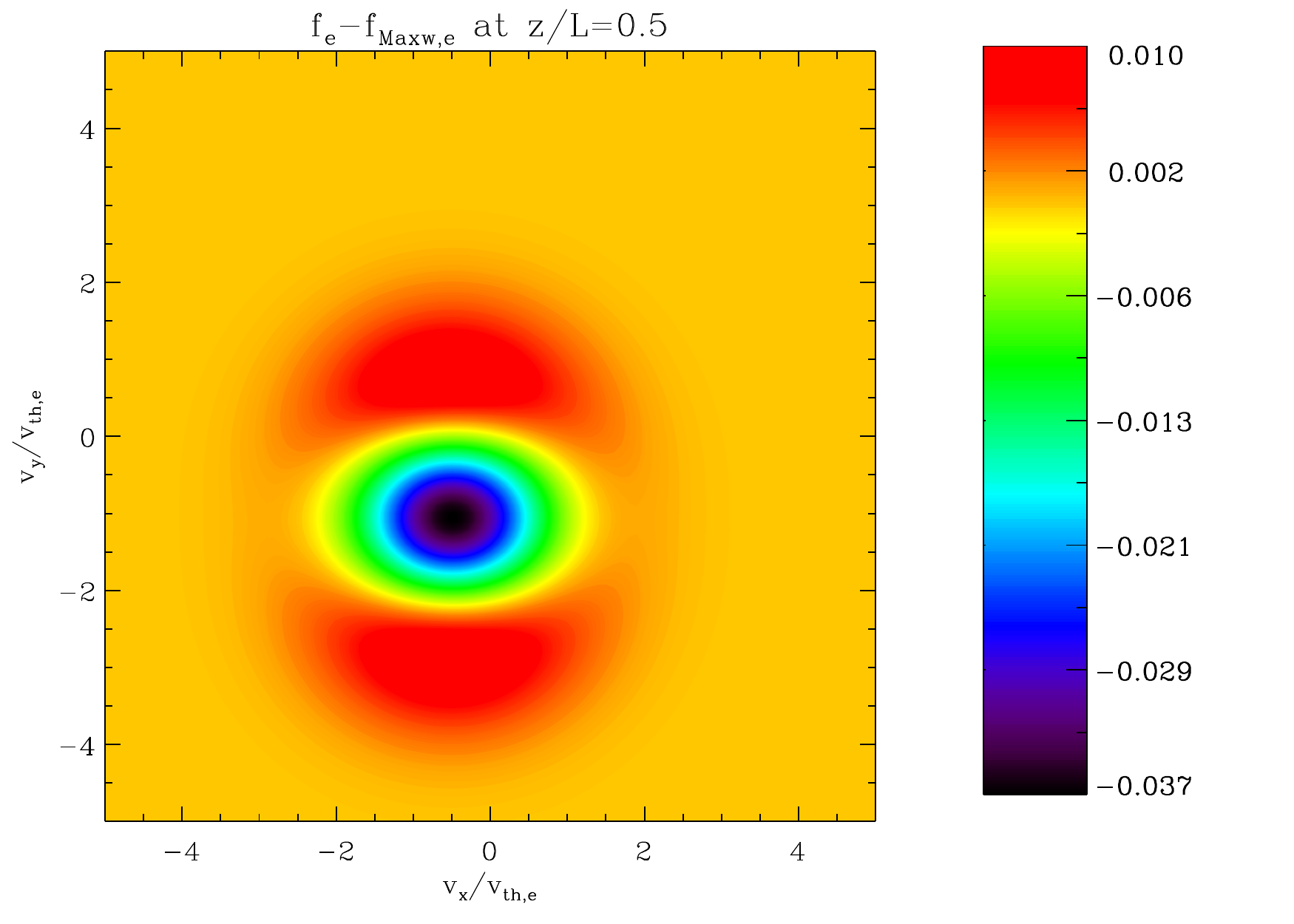}
 \caption{\small \label{fig:jpp2d}}
\end{subfigure}
\begin{subfigure}[b]{0.48\textwidth}
\includegraphics[width=\textwidth]{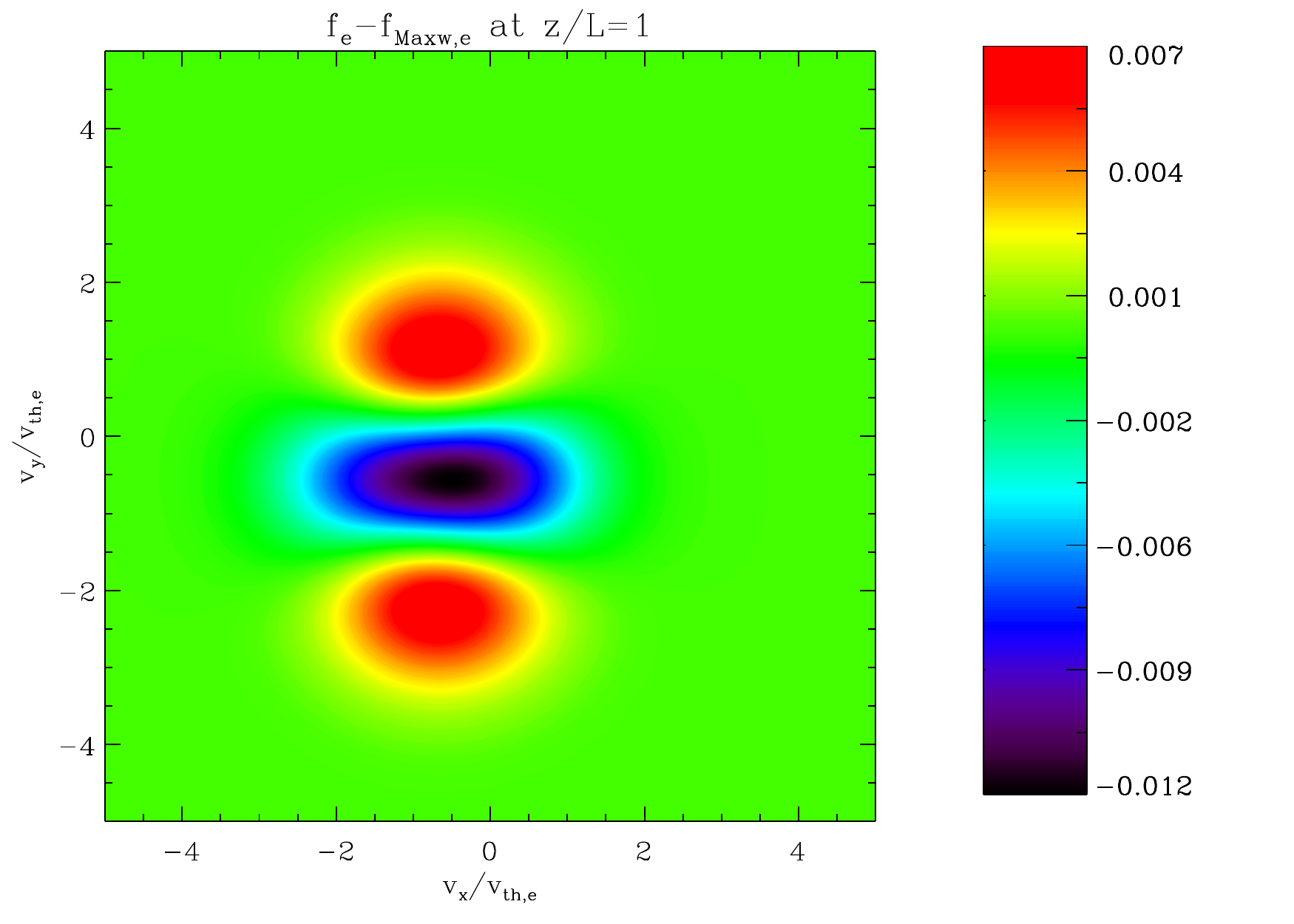}
 \caption{\small \label{fig:jpp2e}}
\end{subfigure}
\caption{\small Contour plots of $f_e-f_{Maxw,e}$ for $z/L=-1$ (\ref{fig:jpp2a}),  $z/L=-0.5$ (\ref{fig:jpp2b}), $z/L=0$ (\ref{fig:jpp2c}), $z/L=0.5$ (\ref{fig:jpp2d}) and $z/L=1$ (\ref{fig:jpp2e}). $\beta_{pl}=0.05$ and $\delta_e=0.03$. }
 \end{figure}

\begin{figure}
\centering
\begin{subfigure}[b]{0.48\textwidth}
\includegraphics[width=\textwidth]{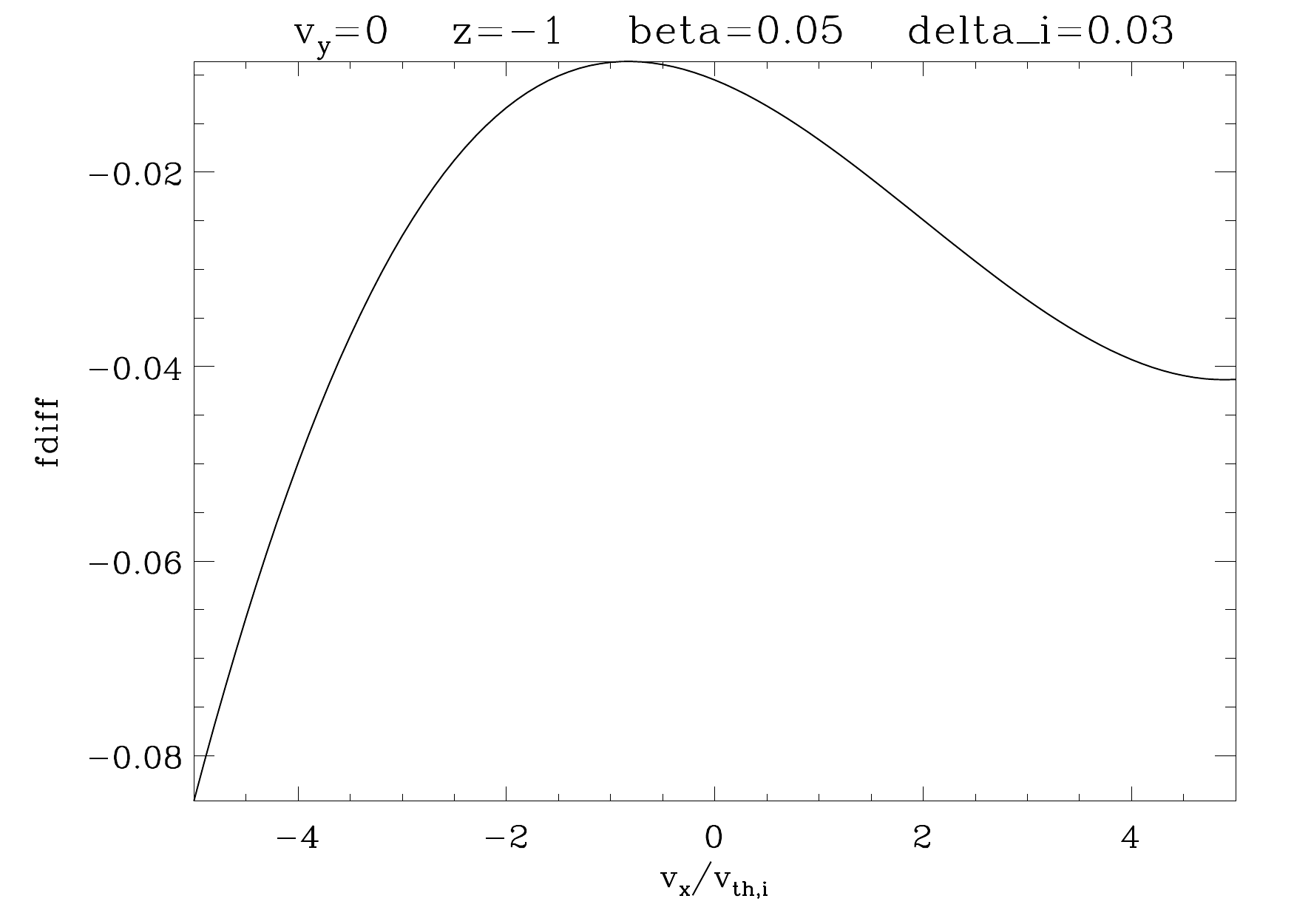}
 \caption{\small \label{fig:jpp3a}}
\end{subfigure}
\begin{subfigure}[b]{0.48\textwidth}
\includegraphics[width=\textwidth]{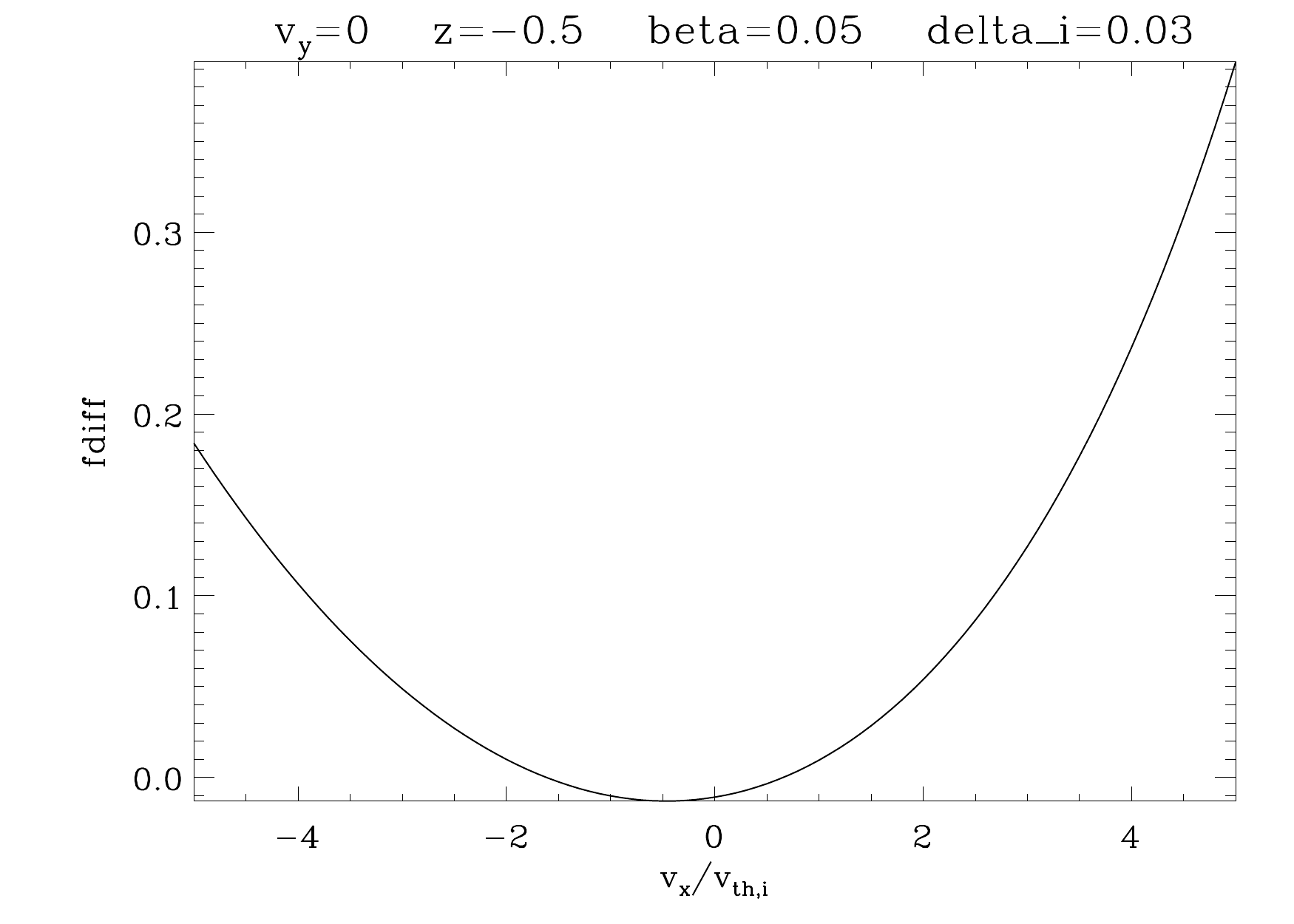}
 \caption{\small \label{fig:jpp3b}}
\end{subfigure}
\begin{subfigure}[b]{0.48\textwidth}
\includegraphics[width=\textwidth]{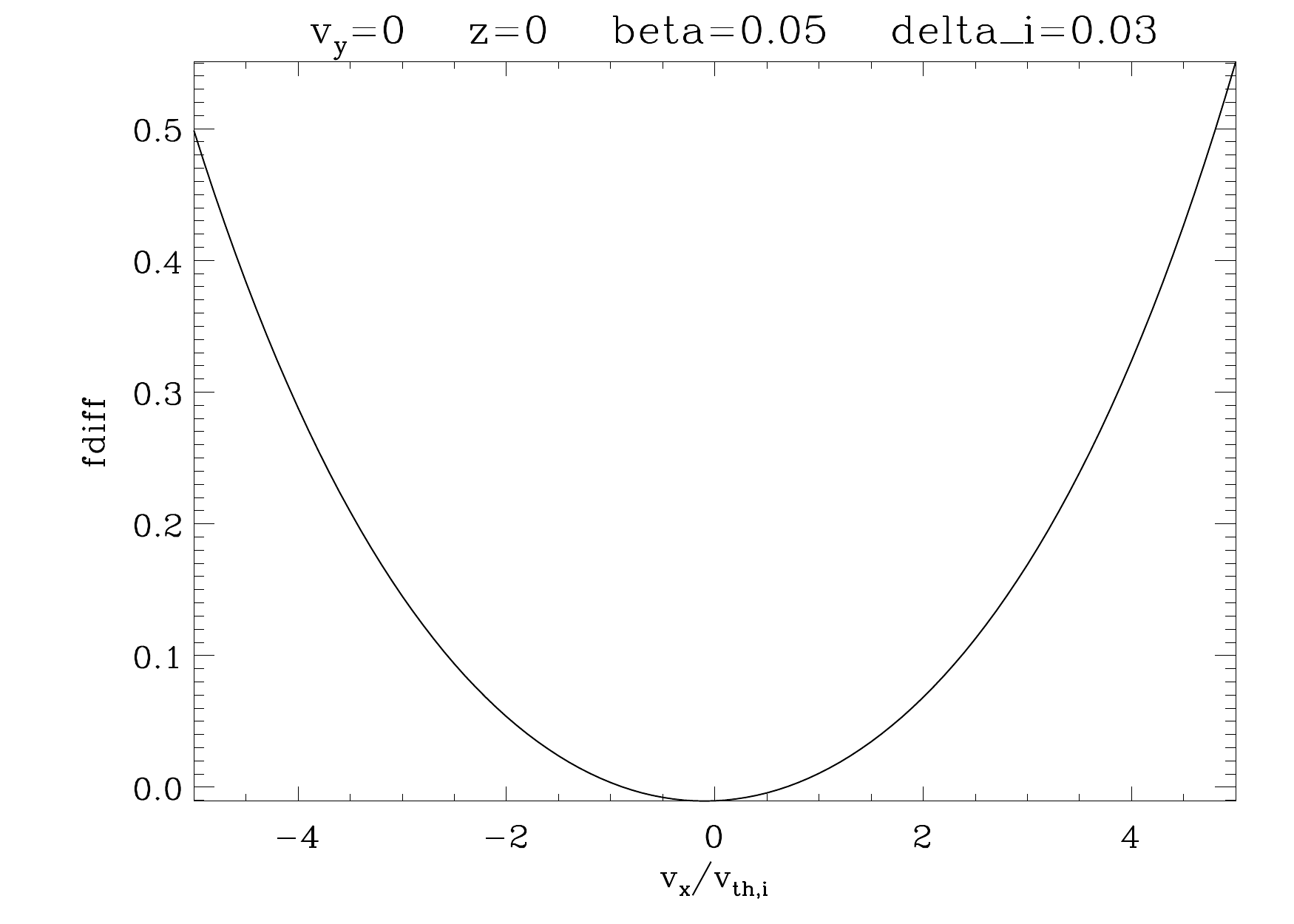}
 \caption{\small \label{fig:jpp3c}}
\end{subfigure}
\begin{subfigure}[b]{0.48\textwidth}
\includegraphics[width=\textwidth]{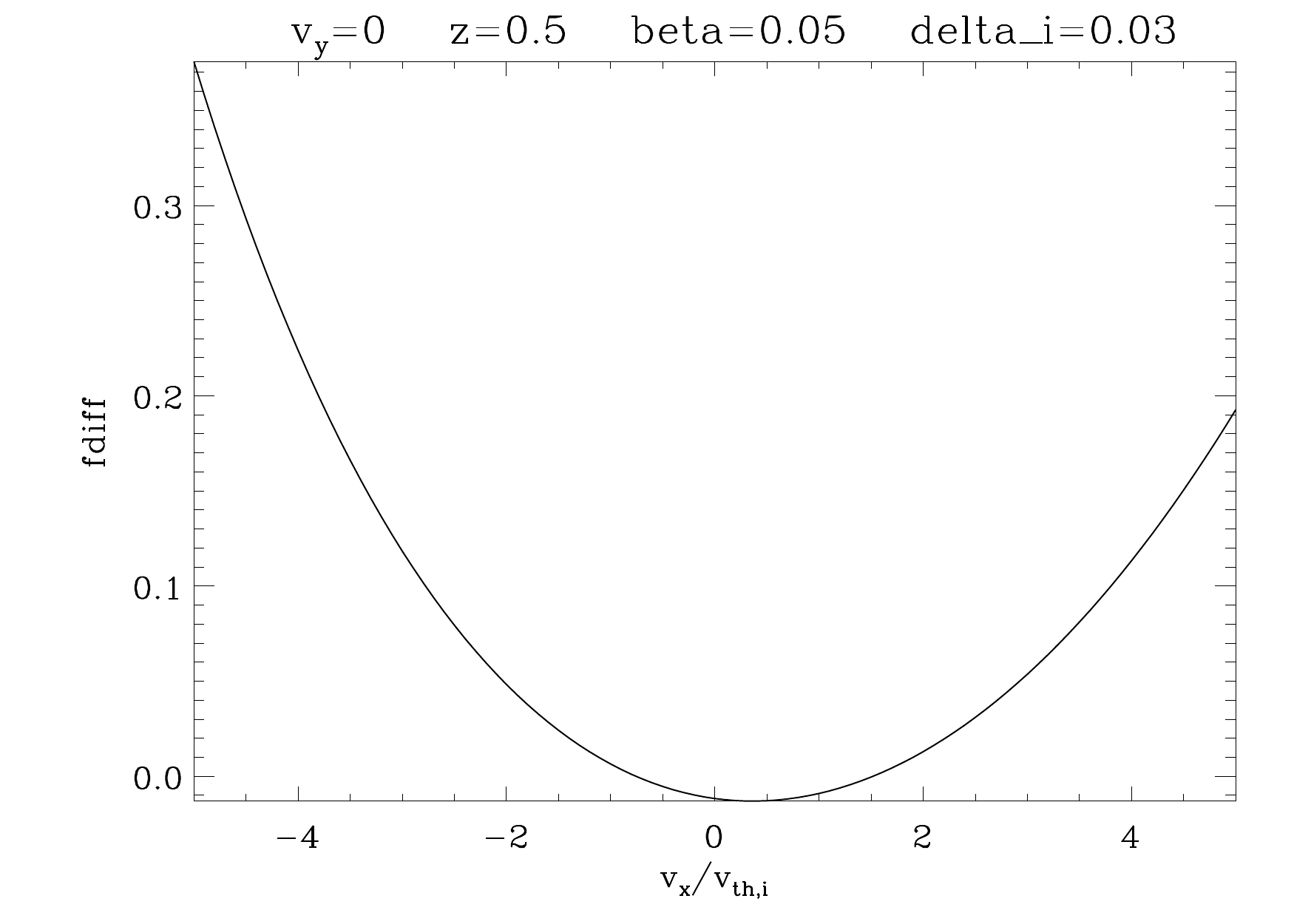}
 \caption{\small \label{fig:jpp3d}}
\end{subfigure}
\begin{subfigure}[b]{0.48\textwidth}
\includegraphics[width=\textwidth]{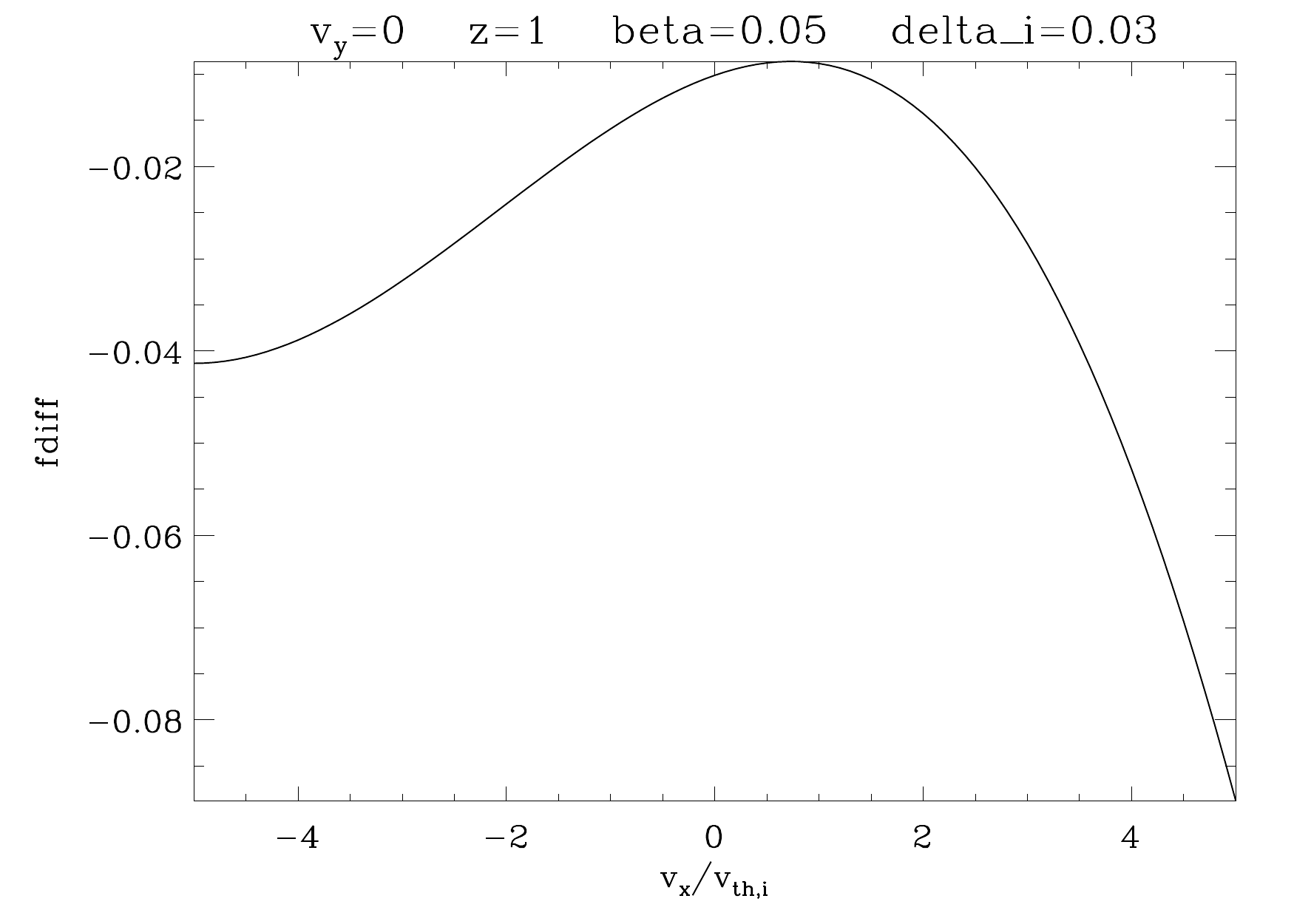}
 \caption{\small \label{fig:jpp3e}}
\end{subfigure}
\caption{\small Line plots of $f_{diff,i}$ against $v_x/v_{\text{th},i}$ at $v_y=0$ for $z/L=-1$ (\ref{fig:jpp3a}),  $z/L=-0.5$ (\ref{fig:jpp3b}), $z/L=0$ (\ref{fig:jpp3c}), $z/L=0.5$ (\ref{fig:jpp3d}) and $z/L=1$ (\ref{fig:jpp3e}). $\beta_{pl}=0.05$ and $\delta_i=0.03$. }
 \end{figure}

\begin{figure}
\centering
\begin{subfigure}[b]{0.48\textwidth}
\includegraphics[width=\textwidth]{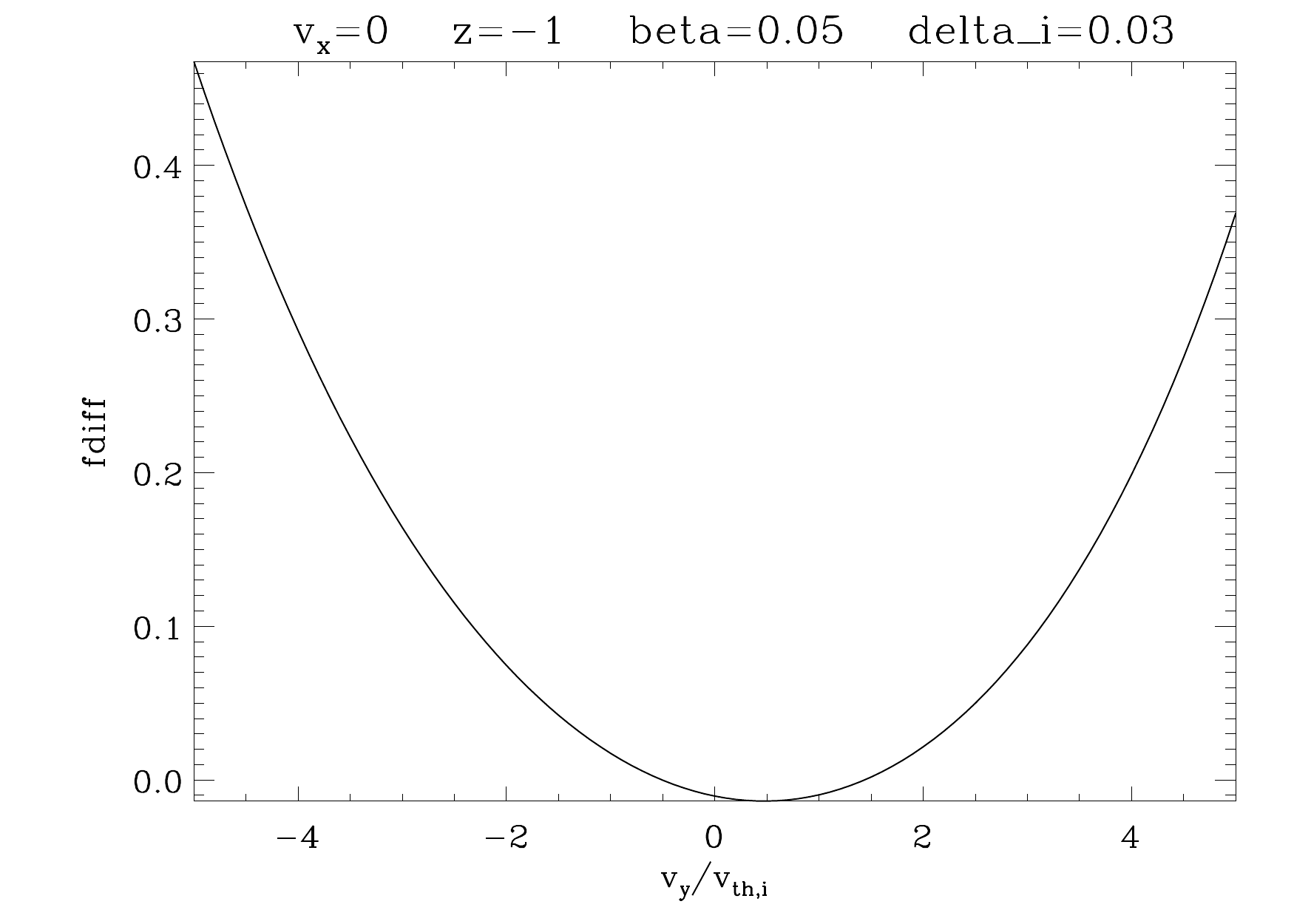}
 \caption{\small \label{fig:jpp4a}}
\end{subfigure}
\begin{subfigure}[b]{0.48\textwidth}
\includegraphics[width=\textwidth]{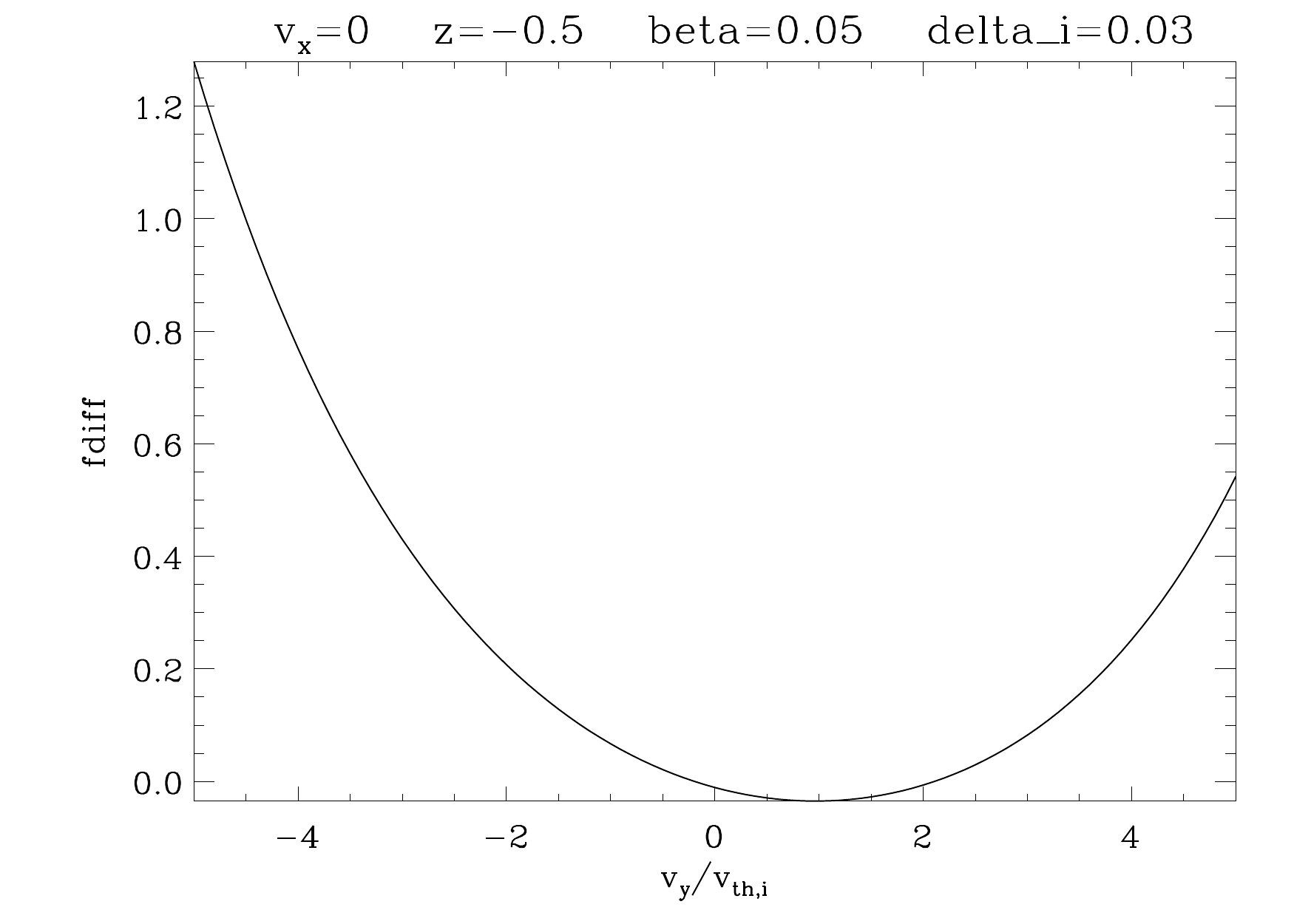}
 \caption{\small \label{fig:jpp4b}}
\end{subfigure}
\begin{subfigure}[b]{0.48\textwidth}
\includegraphics[width=\textwidth]{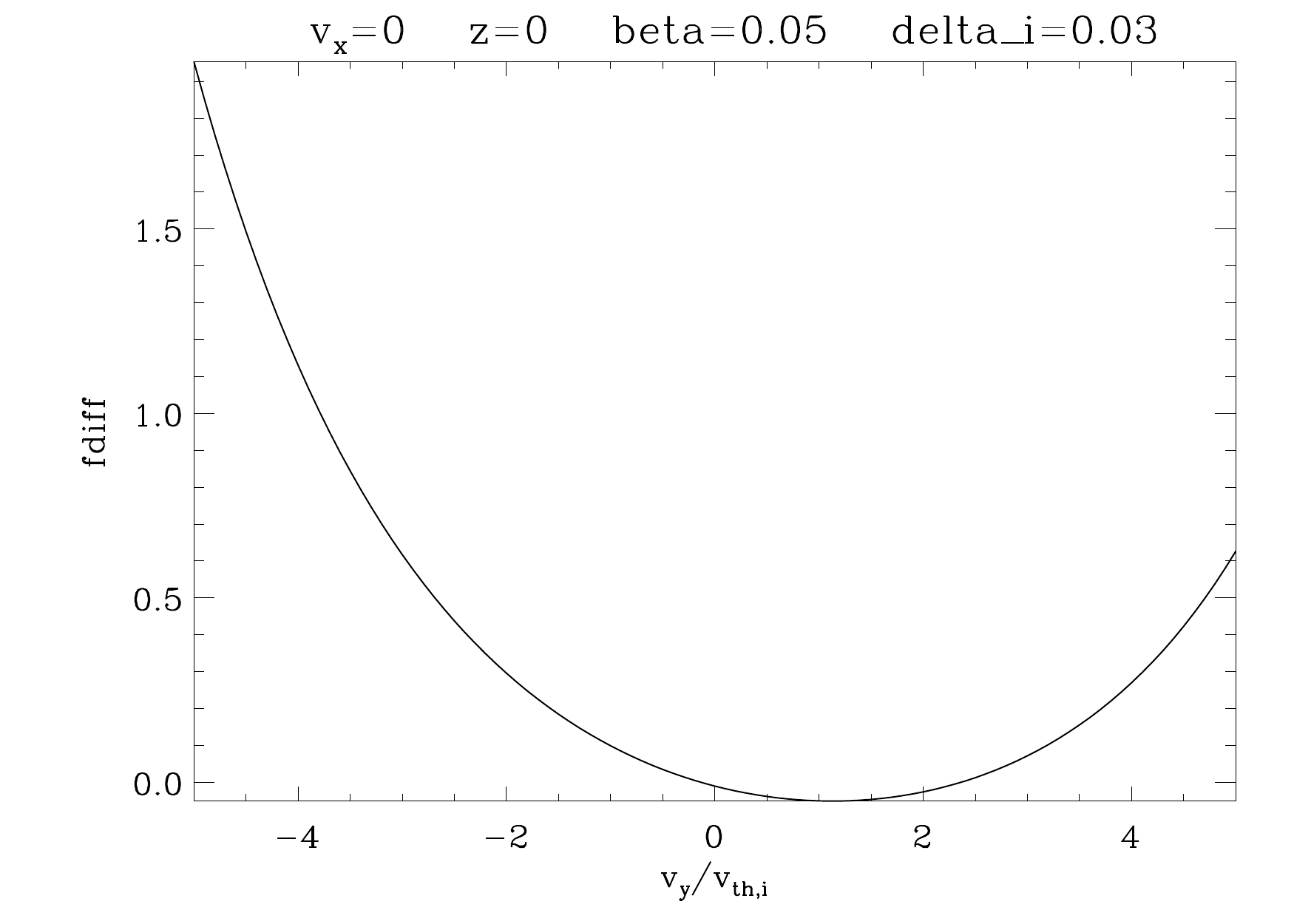}
 \caption{\small \label{fig:jpp4c}}
\end{subfigure}
\begin{subfigure}[b]{0.48\textwidth}
\includegraphics[width=\textwidth]{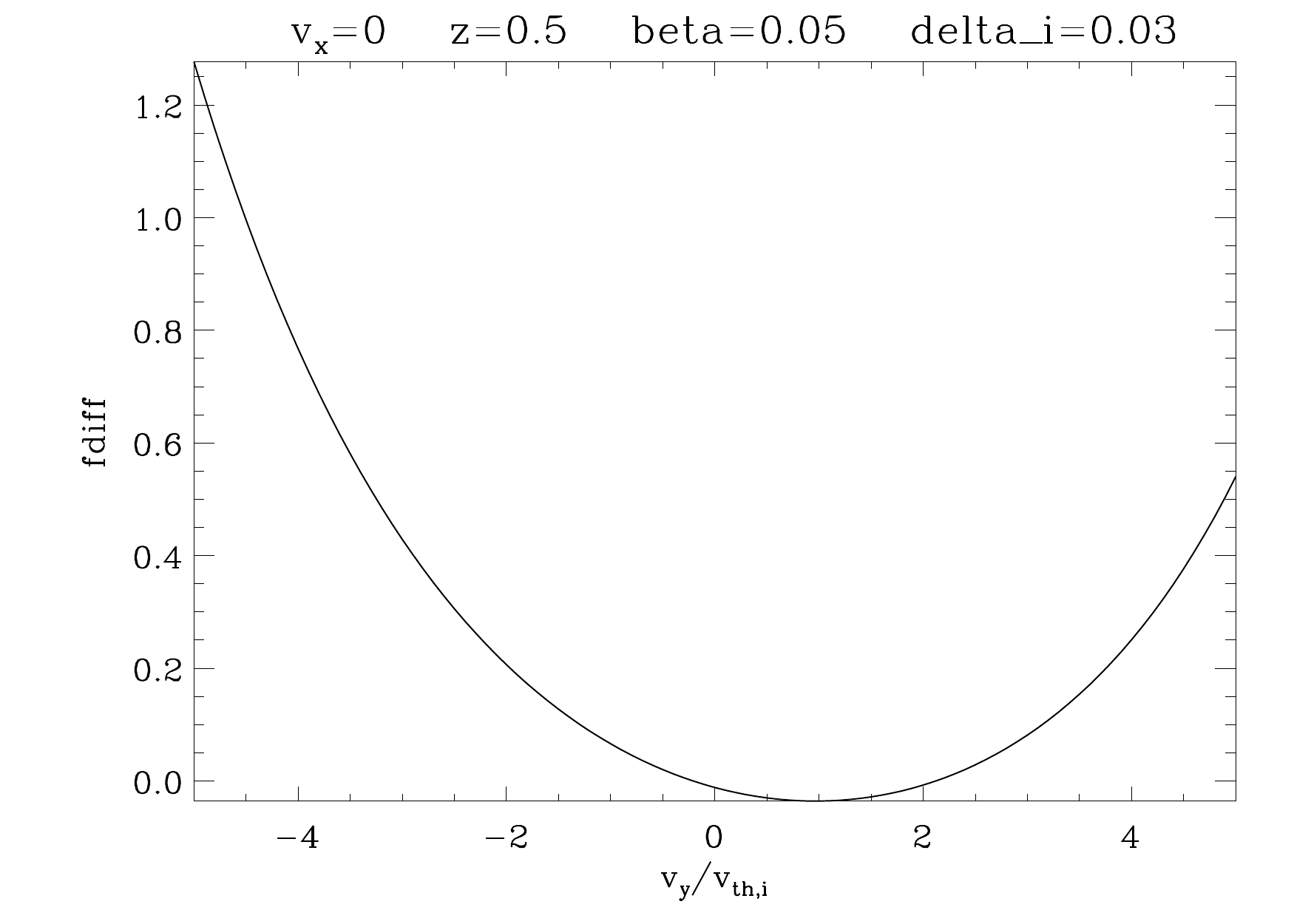}
 \caption{\small \label{fig:jpp4d}}
\end{subfigure}
\begin{subfigure}[b]{0.48\textwidth}
\includegraphics[width=\textwidth]{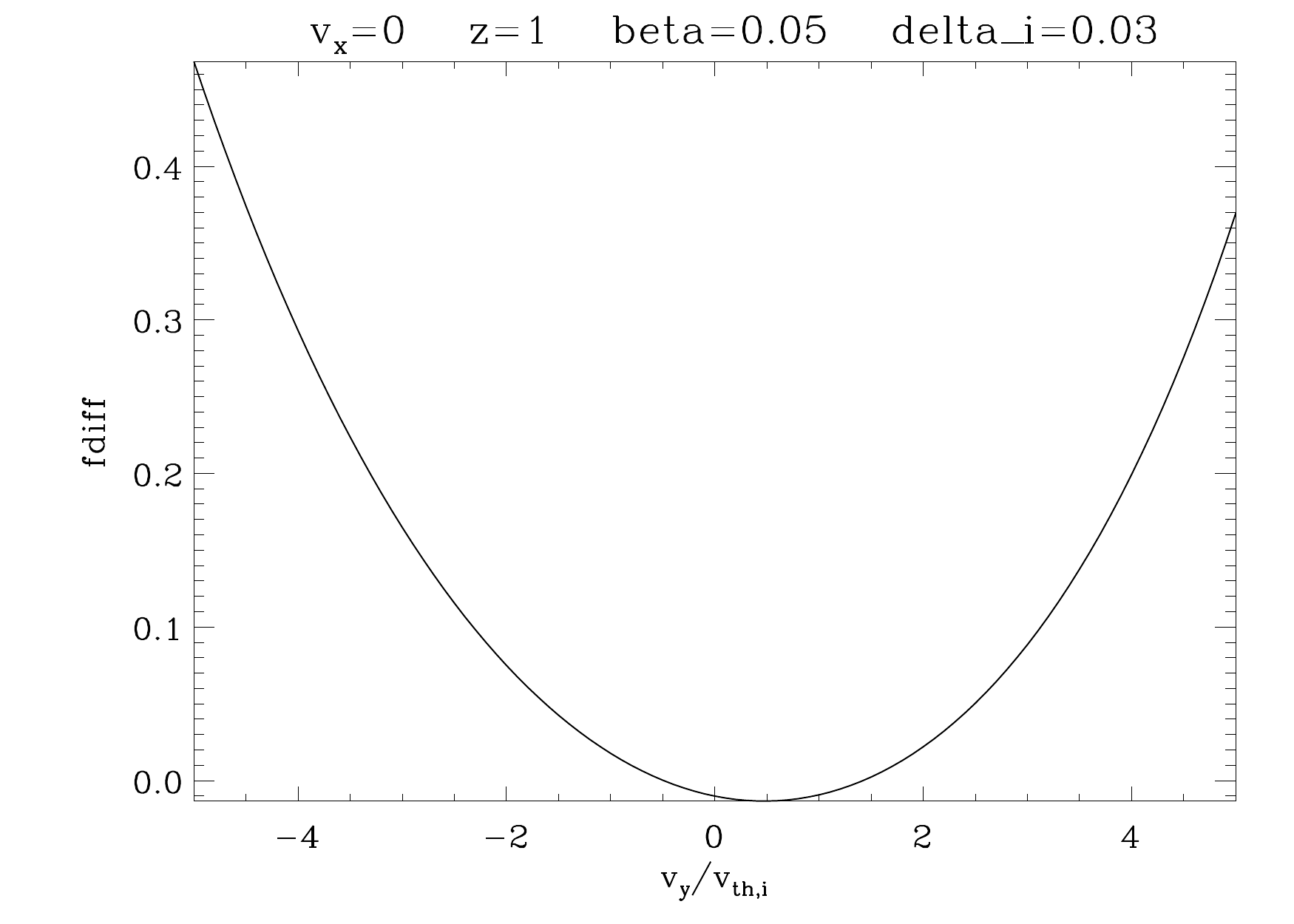}
 \caption{\small \label{fig:jpp4e}}
\end{subfigure}
\caption{\small Line plots of $f_{diff,i}$ against $v_y/v_{\text{th},i}$ at $v_x=0$ for $z/L=-1$ (\ref{fig:jpp4a}),  $z/L=-0.5$ (\ref{fig:jpp4b}), $z/L=0$ (\ref{fig:jpp4c}), $z/L=0.5$ (\ref{fig:jpp4d}) and $z/L=1$ (\ref{fig:jpp4e}). $\beta_{pl}=0.05$ and $\delta_i=0.03$. }
 \end{figure}

\section{Summary}
This chapter contains presentation and analysis of the first DFs capable of describing low plasma beta, nonlinear force-free collisionless equilibria. By using expressions for the moments of the DFs we have derived the relationships between the micro- and macroscopic parameters of the equilibrium, in particular the current sheet width. We have presented line-plots of the electron DF in the $v_x$ direction as a representative example. These show that the DF has a single maximum in the $v_x$ direction, and \emph{seems} to resemble a Maxwellian, at least for the parameter range studied. However, a detailed comparison with a Maxwellian describing the same particle density and average velocity/current density shows that there are significant deviations. This was corroborated by contour plots of the difference between the DF and the Maxwellian in the $(v_x,v_y)$ plane. 

While it has been shown that the infinite series over Hermite polynomials are convergent for all parameter values, plotting the DF in the original gauge,
\[
\boldsymbol{A}=B_0L(2\arctan (\exp(z/L)), \, \ln \text{sech}(z/L),\, 0),
\]
has been difficult for the low-beta regime, and particularly due to the $v_x$ dependent sum. As such, $\beta_{pl}=0.85$ was the lowest value of the plasma beta for which we could be confident in the numerical method. Further work on attaining numerical convergence for a wider parameter range was necessary, with a particular motivation was to find out whether the DF develops multiple peaks similar to the DF found for an additive form of $P_{zz}$ \citep{Neukirch-2009}.

Motivated by the numerical challenges mentioned above, in Section \ref{sec:newdf2} we presented calculations for a DF with a different gauge to that considered in previous studies \citep{Harrison-2009PRL, Neukirch-2009, Wilson-2011, Abraham-Shrauner-2013, Kolotkov-2015}, 
\[
\boldsymbol{A}=B_0L(2\arctan \tanh(z/(2L)), \, \ln \text{sech}(z/L),\, 0).
\]
We have presented some plots of a comparison between the re-gauged DFs and shifted Maxwellian functions, as a proof of principle, namely that numerical convergence for values of $\beta_{pl}$ lower than previously reached in the `original gauge', can now be attained ($\beta_{pl}=0.05$). 

Verification of the analytical properties of convergence and boundedness for both the DFs written as infinite sums over Hermite polynomials have been given. Note that the verification of these DFs is rather involved due to the complex nature of the specific Maclaurin expansions that we consider, and is simpler for more `straightforward' expansions, e.g. for the example considered in Section \ref{Sec:Channell}. 

Future work could involve an in-depth parameter study of the new re-gauged multiplicative DF for the FFHS, with an analysis of how far the exact equilibrium DF differs from an appropriately drifting Maxwellian, frequently used in fully kinetic simulations for reconnection studies. In particular it would be interesting to see how much the DFs differ from drifting Maxwellians as the set of parameters $(\beta_{pl}, \delta_{s})$ are varied across a wide range. Preliminary numerical investigations verify that plotting DFs for the FFHS with a lower $\beta_{pl}$ than previously achieved, namely $\beta_{pl}=0.05$ rather than $\beta_{pl}=0.85$, has been made possible by the theoretical developments in this chapter. We have not yet observed multiple maxima for the DFs, but do see significant deviations from Maxwellian distributions, and an anisotropy in velocity space.

\null\newpage
\chapter{One-dimensional asymmetric current sheets} \label{Asymmetric}

 \epigraph{\emph{Reconnection is now among the most fundamental unifying concepts in astrophysics, comparable in scope and importance to the role of natural selection in biology.}}{\textit{from \citet{Moore-2015}}}

\noindent Much of the work in this chapter is drawn from \citet{Allanson-2017GRL}

\section{Preamble}
The NASA MMS mission has very recently made in situ diffusion region measurements of asymmetric magnetic reconnection for the first time \citep{Burch-2016Science}. In order to compare to the data obtained from kinetic-scale observations (e.g. see \cite{Burch-2016GRL}), it would be useful to have initial equilibrium conditions for PIC simulations that reproduce the physics of the dayside magnetopause current sheet as accurately as possible, i.e. self-consistent VM equilibria that model the magnetosheath-magnetosphere asymmetries in pressure and magnetic field strength. 

In this chapter, we present new `exact numerical' (numerical solutions to equations for exact VM equilibria), and exact analytical equilibrium solutions of the VM system that are self-consistent with 1D and asymmetric Harris-type current sheets, with a constant guide field. The DFs can be represented as a combination of shifted Maxwellian DFs, are consistent with a magnetic field configuration with more freedom than the previously known exact solution \citep{Alpers-1969}, and have different bulk flow properties far from the sheet.

\section{Introduction}

\subsection{Asymmetric current sheets}\label{sec:asymmintro}
Under many circumstances (and unlike the application in Chapter \ref{Sheets}), the plasma conditions can be different on either side of the current sheet, e.g. the magnetic field strength and its orientation. As well as in the magnetopause (e.g. see \cite{Burch-2016GRL}), such asymmetric current sheets are observed at Earth's magnetotail (e.g. \cite{Oieroset-2004}), in the solar wind (e.g. \cite{Gosling-2006}), between solar flux tubes (e.g. \cite{Linton-2006, Murphy-2012, Zhu-2015}), in turbulent plasmas (e.g. \cite{Servidio-2009, Karimabadi-2013}), and inside a tokamak (e.g. \cite{Kadomtsev-1975}).

Regarding the theoretical modelling of dynamical features, various authors have considered the impact of asymmetric current sheets on different aspects of instability and magnetic reconnection, such as the `Sweet-Parker' style analysis carried out by \citet{Cassak-2007}; the development of current driven instabilities (the lower-hybrid instability) \citep{Roytershteyn-2012}; and the suppression of reconnection at Earth's magnetopause \citep{Swisdak-2003, Phan-2013, Trenchi-2015, Liu-2016}. Whilst it can be argued that the general properties (e.g. the reconnection rate) of the nonlinear phase physics of magnetic reconnection are relatively insensitive with regards to the exactitude of the initial conditions, the physics in the linear stage can affect the dynamical evolution of the current sheets, and that can only be confidently studied with exact initial conditions (e.g. see \cite{Dargent-2016}). 

To give some specific examples of the use of exact solutions, setting up a VM equilibrium current sheet in numerical simulations would be helpful for the study of collisionless tearing instabilities, which could be important to understand the role of tearing modes in determining the orientation of the three-dimensional reconnection x-line in an asymmetric geometry \citep{Liu-2015}. This is especially crucial for predicting the location of magnetic reconnection at Earth's magnetopause under diverse solar wind conditions, as discussed in \cite{Komar-2015} for example. Knowledge of an exact equilibirum also facilitates the study of tearing instabilities under the influence of cross-sheet gradients (e.g. see \cite{Zakharov-1992, Kobayashi-2014, Pueschel-2015, Liu-2016}), which can be important for understanding the onset and diamagnetic suppression of sawtooth crashes in fusion devices .

\begin{figure}
    \centering

        \includegraphics[width=0.7\textwidth]{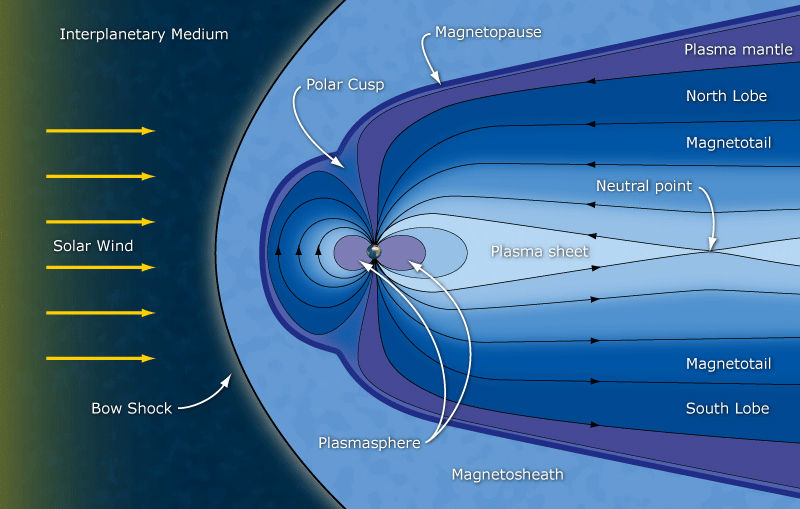}
        \caption{{\small Diagram representing the Earth's magnetic environment, and it's interaction with the solar wind. {\bf Image credit:} NASA, and without copyright.}}
        \label{fig:magnetosphere}
        \end{figure}

\subsection{Modelling the magnetopause current sheet}

\subsubsection{Model paradigm}
The macroscopic equilibrium for which we wish to obtain a self-consistent VM equilibrium is that which describes a current sheet in the Earth's dayside magnetopause. Figure \ref{fig:magnetosphere} depicts the Earth's magnetopause, its relation to the rest of the Earth's magnetosphere, and the interaction with the solar wind. In line with other theoretical approaches (e.g. see \cite{Hesse-2013}) and observational (e.g. see \cite{Burch-2016Science}) conclusions, the equilibrium should be `asymmetric' with respect to either side of the current sheet, i.e. it should be characterised by an enhanced density/pressure on the magnetosheath side of the current sheet, and an enhanced magnetic field magnitude on the magnetosphere side. These basic requirements are shown by Figures \ref{figtwo} and \ref{figone}, in which the coordinates $(x,y,z)$ are related to the ``Boundary Normal'' coordinates, $LMN$, (e.g. see \cite{Hapgood-1992,Burch-2016Science}). Their correspondence is given by $(\hat{x},\hat{y},\hat{z}) \sim (\hat{L},\hat{M},\hat{N})$, with the $xy$ plane tangential to the magnetopause, and $z$ normal to it. As explained by \citet{Hapgood-1992}, \emph{``There is no universal convention to resolve the L and M axes. The relationship between LMN and other systems ... is dependent on position.''} For a heuristic understanding, and in the paradigm of the `square-on' geometry presented by Figure \ref{fig:magnetosphere}, we can think of $x\sim L$ as pointing `Earth North', $y\sim M$ as pointing `Earth West', and $z\sim N$ as pointing `Sunward'. The figures relate to a specific magnetic field, to be defined in Section \ref{sec:prior}, but they portray the basic features that the model should have. Essentially, pressure balance dictates that an enhanced magnitude of magnetic pressure on the magnetosphere side of the current sheet ($z<0, N<0$) relies on a depleted thermal pressure, and vice versa for the magnetosheath side ($z>0, N>0$). However, the current density is modelled to be symmetric.
\begin{figure}
    \centering

       \begin{subfigure}[b]{0.7\textwidth}
        \includegraphics[width=\textwidth]{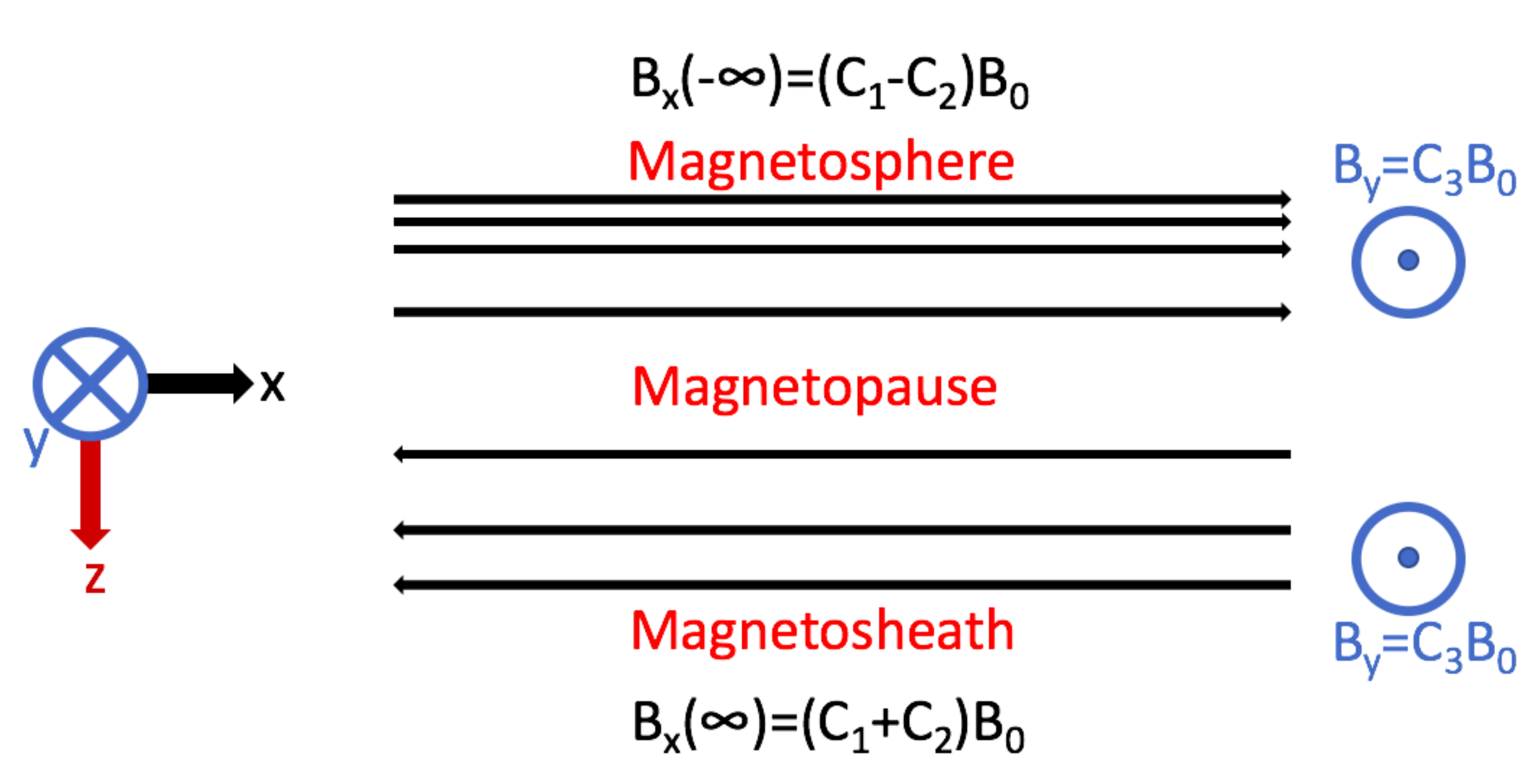}
        \caption{\small A representative diagram of the AH+G equilibrium magnetic field, for $C_{1}+C_{2}<C_{1}-C_{2}$. }
        \label{figtwo}
        \end{subfigure}
          \begin{subfigure}[b]{0.7\textwidth}
        \includegraphics[width=\textwidth]{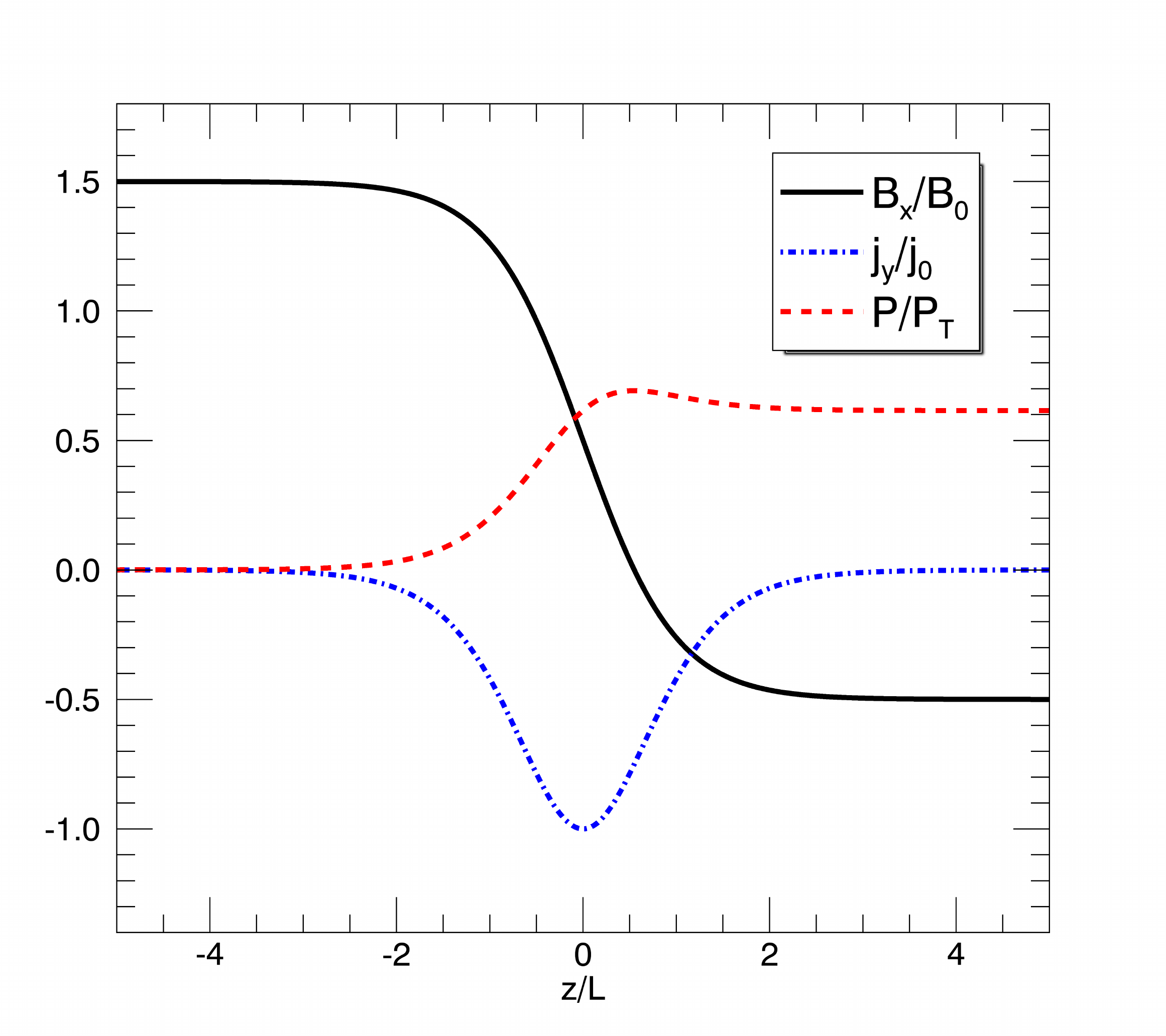}
        \caption{\small Normalised magnetic field $\tilde{B}_x$, current density $\tilde{j}_y$, and scalar pressure $\tilde{p}$ for parameter values $C_1=0.5$, $C_2=-1$, $C_3=1$ and $P_{T}=(C_1^2+C_2^2+C_3^2-2C_1C_2) /2= 1.625$. }
        \label{figone}
    \end{subfigure}
    \caption{\small The AH+G equilibrium configuration (Equation \ref{eq:ahgmagfield}).}
      \label{}
    \end{figure}
As in Chapters \ref{Vlasov} and \ref{Sheets}, we assume a 1D geometry for which $\nabla=(0,0,\partial/\partial z)$, which is justifiable by a separation of scales (e.g. see \cite{Quest-1981B}). In this case, a quasineutral macroscopic equilibrium will obey the following equation,
\begin{equation}
\frac{d}{dz} \left(P_{zz} + \frac{B^2}{2\mu_0}\right) =0,\label{eq:pausebalance}
\end{equation}
but, in contrast to the application to force-free current sheets in Chapter \ref{Sheets}, $P_{zz}$ and $B^2$ must be non-uniform in $z$.

\subsubsection{Typical approach in PIC simulations}
In the effort to model dayside magnetopause reconnection, asymmetric macroscopic equilibria that satisfy Equation (\ref{eq:pausebalance}) have been used in PIC simulations by e.g. \citet{Swisdak-2003, Pritchett-2008, Huang-2008, Malakit-2010, Wang-2013, Aunai-2013, Aunai-2013B,Hesse-2013, Hesse-2014, Dargent-2016, Liu-2016}. All but two \citep{Aunai-2013B, Dargent-2016} of these studies have used drifting Maxwellian DFs as initial conditions (Equation \ref{eq:Mshift}). As discussed in more detail in Section \ref{sec:flowshift}, these DFs can reproduce the same moments $(n(z), \boldsymbol{V}_s(z), p(z))$ necessary for a quasineutral fluid equilibrium, but are not exact solutions of the Vlasov equation and hence do not describe a kinetic equilibrium. The main aim of this chapter is to calculate exact solutions of the equilibrium VM equations consistent with a suitable dayside magnetopause current sheet model, in order to circumvent the need to use non-equilibrium DFs of the form in Equation (\ref{eq:Mshift}).

The work in this chapter is relevant to the main focus of the MMS mission, i.e. asymmetric magnetic reconnection, and so we envisage that this could be the main use of the results at the present time. However, as mentioned in Section \ref{sec:asymmintro}, there are other potential applications of the work to basic equilibrium and instability physics in the magnetotail, solar corona, turbulent plasmas and tokamaks.

\section{Exact VM equilibria for 1D asymmetric current sheets}

\subsection{Theoretical obstacles}
The (symmetric) Harris sheet (Equation (\ref{eq:harrissheet})) can be rendered asymmetric - the \emph{asymmetric Harris sheet} (AHS) - by the simple addition of a constant component to $B_{x}$,
\begin{equation}
\boldsymbol{B}=B_0\left(C_1+C_2\tanh\left(\frac{z}{L}\right),0,0\right),\label{eq:AHS}
\end{equation}
for $C_1$ and $C_2$ dimensionless constants, and and for which there is a field reversal (a change in the sign of $B_x$) only when 
\begin{equation}
\bigg|\frac{C_1}{C_2}\bigg|<1. \label{eq:reversal}
\end{equation}
The current density, $j_{y}$, is indepdendent of $C_{1}$, and so whilst a field-reversal is not essential for the existence of a current sheet in itself, we shall only consider the field-reversal regime. The addition of $C_1$ to $B_x$ leads to an equilibrium described by
\begin{eqnarray}
\frac{B^2}{2\mu_0}(z)&=&\frac{B_0^2}{2\mu_0}\left(C_1^2+2C_1C_2\tanh\left(\frac{z}{L}\right)+C_2^2\tanh^2\left(\frac{z}{L}\right)\right), \nonumber\\
 P_{zz}(z)&=&P_{T}-\frac{B_0^2}{2\mu_0}\left(C_1^2+2C_1C_2\tanh \left( \frac{z}{L}\right)+C_2^2\tanh^2\left( \frac{z}{L}\right)\right),
\end{eqnarray}
with $P_T>B_0^2(|C_1|+|C_2|)^2/(2\mu_0)$ the constant total pressure. 

The VM equilibrium DF self-consistent with the Harris sheet (\cite{Harris-1962} and as discussed in Section \ref{sec:HarrisDF}), 
\[
f_s=\frac{n_{0s}}{(\sqrt{2\pi}v_{th,s})^3}e^{-\beta_s(H_s-u_{ys}p_{ys})},
\]
can also be made to be consistent with the field,
\[
\boldsymbol{B}=B_0\left(\tanh\left(\frac{z}{L}\right),C_3,0\right),
\]
i.e. a Harris sheet plus guide field. This is achieved fairly simply by `sending' $A_x=0\to A_x=C_3B_0z$. This adds no real complications since $j_x=0$, $P_{zz}(A_y)$ remains unchanged, and one essentially just solves Amp\`{e}re's Law with different conditions as $|z|\to\infty$,
\[
\nabla^2A_x=0 \hspace{3mm}\text{s.t.}\hspace{3mm}A_x=0 \hspace{3mm}\to\nabla^2A_x=0 \hspace{3mm}\text{s.t.}\hspace{3mm}A_x=C_3B_0z.
\]
In the analogy of the particle in a potential (see Section \ref{sec:pseudo}), this corresponds to the particle having a non-zero and constant component of `velocity' in the $x\sim A_x$ direction, instead of zero velocity in that direction. As a result, one might expect that it should be relatively straightforward to adapt the Harris DF to be self-consistent with the AHS, but this is not the case in the field-reversal regime.

\subsubsection{$P_{zz}$ must depend on both $A_x$ and $A_y$}\label{sec:twodimensional}
The AHS has only one component of the current density, $j_y$, and since $\boldsymbol{j}=\partial P_{zz}/\partial\boldsymbol{A}$, one might expect that the equilibrium could be described by $P_{zz}=P_{zz}(A_y)$, and hence $f_s=f_s(H_s,p_{ys})$ accordingly. However, using the analogy of the particle in a potential (see Section \ref{sec:pseudo}), in which the following correspondences hold
\begin{eqnarray}
\text{Position:}\,(x,y) & \sim & (A_x,A_y),\nonumber\\
\text{Time:}\, t& \sim & z\nonumber\\
\text{Velocity:}\, (v_x(t),v_y(t))&\sim & \left(\frac{dA_x}{dz}(z),\frac{dA_y}{dz}(z)\right) \sim (B_y,-B_x) , \nonumber\\
\text{Potential:} \, \mathcal{V}(x,y) &\sim & P_{zz}(A_x,A_y),\nonumber\\
\text{Force:} \, \boldsymbol{\mathcal{F}}(x(t),y(t)) &\sim & \mu_0\frac{d^2\boldsymbol{A}}{dz^2},\nonumber\\
\text{Equation of motion:}\, \boldsymbol{\mathcal{F}}=-\nabla\mathcal{V} &\sim &   \mu_0\frac{d^2\boldsymbol{A}}{dz^2}  =-\frac{\partial P_{zz}}{\partial \boldsymbol{A}},  \nonumber
\end{eqnarray}
we note that - crucially - velocity is conjugate to the derivatives of $A_x,A_y$, and hence the magnetic field. The important observation to make is that a single-valued and 1D potential, $P_{zz}(A_y)$, cannot be compatible with a `velocity' of the form of the magnetic field in Equation (\ref{eq:AHS}),
\[
v_y(t)\sim C_1 + C_2\tanh t ,
\] 
when we are in the field-reversal regime ($|C_1|<|C_2|$). The reasoning is as follows. 

Without loss of generality suppose that $C_1, C_2 >0$. The particle begins its journey at $t=-\infty, \, y=\infty$ with velocity $C_1-C_2<0$. It then rolls up a `hill' in the potential, is stationary at $t=\tanh ^{-1} (-C_1/C_2)$, and rolls back down the hill towards $y=\infty$ with final velocity $C_1+C_2$ at $t=\infty$. This trajectory is not possible for a conservative potential that is single-valued in space. Hence we conclude that a 1D asymmetric current sheet with field reversal can not be analytically self-consistent with a pressure tensor that is a function of only one component of the vector potential. 

Despite the fact that $j_x=0$ for the AHS, and hence $\partial P_{zz}/\partial A_x=0$, it has become apparent that we require the `hill' to be 2D, such as the $P_{zz}(A_x,A_y)$ function depicted in Figure \ref{fig:surface}, for which the overlaid line depicts the particle trajectory. (The exact form and derivation of that particular pressure function shall be discussed in Section \ref{sec:pressuretensors}).

    \begin{figure}\centering
        \includegraphics[width=0.7\textwidth]{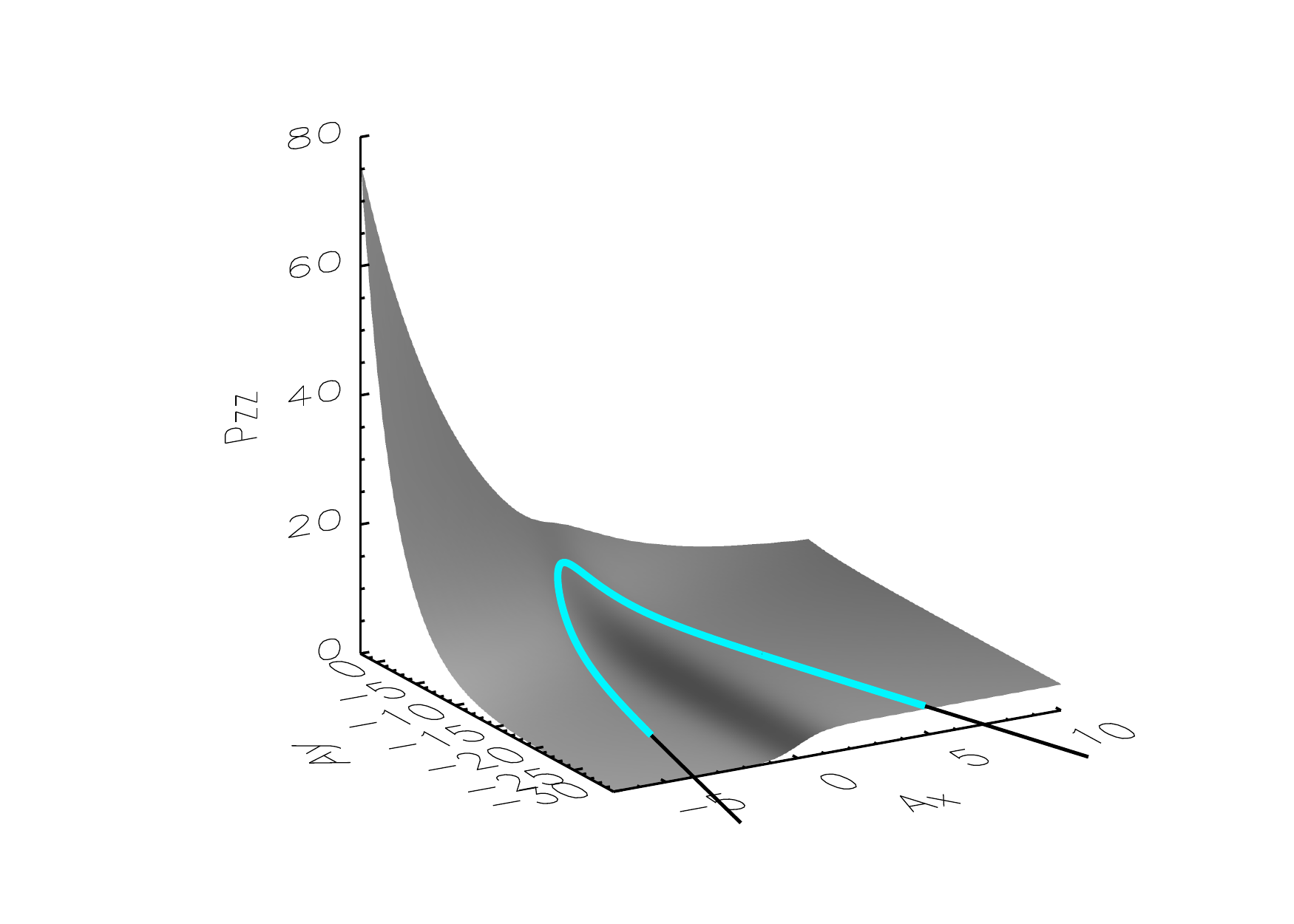}
        \caption{\small The ``$\tanh$'' pressure function for $C_1=-1, C_2=8, C_3=1, C_2=2\sqrt{C_1^2}=2$ }
        \label{fig:surface}

\end{figure}

We note that `exact numerical' VM equilibria have recently been found by \citet{Belmont-2012, Dorville-2015}, using the inverse approach, for the `normal/symmetric' Harris sheet magnetic field, and a modified `force-free Harris sheet' respectively. The equilibria have asymmetries in the number density and temperature either side of the sheet, with \citet{Dorville-2015} including an electric field. Their methods rely on similar notions to those discussed above, for which the DFs were multi-valued functions of the constants of motion. The DF derived by \citet{Belmont-2012} has been used as the initial condition for Hybrid simulations by \citet{Aunai-2013B}, and PIC simulations by \cite{Dargent-2016}. Exact numerical solutions for asymmetric current sheets are more numerous for the forward problem, with examples in e.g. \citet{Kan-1972,Lemaire-1976, Kuznetsova-1995, Roth-1996, Lee-1979JGR}.

\subsubsection{Prior exact analytical VM equilibria}\label{sec:prior}
To our knowledge, there is one known exact VM equilibrium for a magnetic field like the AHS. In the Appendix of \citet{Alpers-1969}, a DF is derived that is consistent with the `Alpers magnetic field', which could be written in a $z$-dependent geometry as
\begin{equation}
\boldsymbol{B}=B_0\left(-\frac{B_2}{2}\left(1+\tanh\left(\frac{z}{L}\right)\right), \frac{B_1}{2}\tanh \left(\frac{z}{L}\right),0\right). \label{eq:Alpers}
\end{equation}
Despite appearances, this magnetic field is almost equivalent to the AHS. To see this, we make a small digression.

First allow the AHS to have a constant \emph{guide field}, with $\boldsymbol{A}, \boldsymbol{B}$ and $\boldsymbol{j}$ given by
\begin{eqnarray}
\boldsymbol{A}&=&B_0L(\,C_3\tilde{z},\, -C_1\tilde{z}-C_2\ln \cosh\tilde{z},\,0),\label{eq:ahgvecpotential}\\
\nabla\times\boldsymbol{A}=\boldsymbol{B}&=&B_0(\,C_1+C_2\text{tanh}\tilde{z},\,C_3,\,0 ),\label{eq:ahgmagfield}\\
\frac{1}{\mu_0}\nabla\times\boldsymbol{B}=\boldsymbol{j}&=&\frac{B_0}{\mu_0L}(\, 0,\, C_2\text{sech}^2\tilde{z},\, 0),\label{eq:ahgcurrent}
\end{eqnarray}
then we have the \emph{Asymmetric Harris sheet plus guide field} (AH+G), with $C_3$ a non-zero constant. The vector potential, magnetic field, current density and length scales are normalised according to $\tilde{\boldsymbol{A}}B_0L=\boldsymbol{A}$, $\tilde{\boldsymbol{B}}B_0=\boldsymbol{B}$, $\boldsymbol{j}=j_0\tilde{\boldsymbol{j}}$ and $z=L\tilde{z}$ respectively, with $j_0=B_0/(\mu_0L)$. Example profiles of $\tilde{B}_x$ and $\tilde{j}_y$ are plotted in Figure \ref{figone} for parameter values $C_1=0.5$, $C_2=-1$ and $C_3=1$, (in line with other theoretical studies, e.g. see \cite{Pritchett-2008,Liu-2016}). For these parameter values, the left and right hand sides of the plot represent the magnetosphere and magnetosheath respectively, whilst the central current sheet is in the magnetopause. The equilibrium is maintained by the `gradient of a scalar pressure', $p(z):=P_{zz}$, according to 
\begin{equation}
P_{zz}(\tilde{z})=P_{T}-\frac{B_0^2}{2\mu_0}\left(C_1^2+2C_1C_2\tanh \tilde{z}+C_2^2\tanh^2\tilde{z}+C_3^2\right),\label{eq:scalar_pressure}
\end{equation}
for $P_{T}$ the total pressure (magnetic plus thermal), and $P_{zz}>0$ for $C_1^2+2|C_1C_2|+C_2^2+C_3^2<2\mu_0P_{T}/B_0^2$. The profile of $\tilde{p}(\tilde{z})=P_{zz}/P_T$ is plotted in Figure \ref{figone}, for $P_T=1.625$.

After a rotation by $\tan\theta=C_1/C_3$, the AH+G field becomes 
\begin{equation}
\mathbf{B}^\prime=B_{0}\left( \frac{C_2C_3}{\sqrt{C_1^2+C_3^2}}\tanh\tilde{z} ,   \sqrt{C_1^2+C_3^2}+\frac{C_1C_2}{\sqrt{C_1^2+C_3^2}}      \tanh \tilde{z} ,0\right),
\end{equation}
which is essentially equivalent to the Alpers magnetic field in Equation (\ref{eq:Alpers}) when $C_1C_2=C_1^2+C_3^2$. As such, the Alpers magnetic field is very similar to the AH+G field, but with one fewer degree of freedom. 

For the DF derived by Alpers, and those to be developed in this chapter, the guide field, $B_y=C_3B_0$, is crucial for making analytical progress. The existence of $B_y$ necessitates a non-trivial $A_x=C_3B_0z$, and as a result the `potential' $P_{zz}$ can now be a function of both $A_x$ and $A_y$. This two-dimensionality was reasoned to be an important feature of analytically described asymmetric fields in Section \ref{sec:twodimensional}, and will allow us to construct exact analytical DFs.

There is one more difference between the equilibrium derived by Alpers, and the one that we shall consider, and it is related to the bulk flows. 

As is necessary for consistency between the microscopic and macroscopic descriptions, the Alpers DF is self-consistent with the prescribed magnetic field, i.e. the sum of the individual species (kinetic) currents are equal to the current prescribed by Amp\`{e}re's Law, i.e. $\sum_s \boldsymbol{j}_{s}=\boldsymbol{j}=\nabla\times\boldsymbol{B}/\mu_{0}$. However, the $\boldsymbol{j}_{s}$ are non-zero at $z=+\infty$ (in our co-ordinates), i.e. the magnetosheath side. In contrast, equation (\ref{eq:current}) shows that the macroscopic current densities vanish as $z\to\pm\infty$, i.e. the Alpers DF gives species currents $\boldsymbol{j}_{s}$ that are not proportional to the macroscopic current $\boldsymbol{j}$. That is to say that there is finite ion and electron mass flow at infinity, \emph{``impinging vertically''} on the magnetosheath side of the current sheet. This could be appropriate if one wishes to consider a larger scale/global model including bulk flows at the boundary, but it is not suitable if one wishes to consider the domain as an isolated `patch', representing a local current sheet structure. 

In summary, the DF that we derive shall be consistent macroscopically with an equilibrium for which there are no mass flows at the boundary (as typically assumed in PIC simulations, e.g. \cite{Aunai-2013, Hesse-2013}), and is self-consistent with a magnetic field that has more degrees of freedom than that in \citet{Alpers-1969}.

\subsection{Outline of basic method}\label{sec:pressuretensors}
In order to find a VM equilibrium, we shall use `Channell's method' \citep{Channell-1976}. As discussed in Chapter \ref{Intro}, this involves the following steps: 
\begin{description}
\item[Pressure tensor:] First calculate a functional form $P_{zz}(A_x,A_y)$ that `reproduces' the scalar pressure of Equation (\ref{eq:scalar_pressure}) as a function of $z$. It must also satisfy $\partial P_{zz}/\partial \boldsymbol{A}=\boldsymbol{j}(z)$. There could in principle be infinitely many functions $P_{zz}(A_x,A_y)$ that satisfy both these criteria, but we shall choose specific $P_{zz}(A_x,A_y)$ functions which allows us to make analytical progress. 

Note that this procedure is - by the analogy of a particle in a potential - contrary to the `typical approach', in which one tries to establish the trajectory in a given potential. We know the `trajectory as a function of time' $\boldsymbol{A}(z)$, and the value of the potential along it $P_{zz}(z)$, and seek to construct a self-consistent `potential function in space', $P_{zz}(A_x,A_y)$.
\item[Inversion:] The second step is to use the assumed form of the DF in Equation (\ref{eq:F_form}) in the definition of the pressure tensor component $P_{zz}$ as the second-order velocity moment of the DF, $P_{zz}=\sum_sm_s\int v_z^2f_sd^3v$, and attempt to invert the integral transforms, either by Fourier transforms, Hermite polynomials, or perhaps some other method.
\item[Macro-micro:] The inversion process must yield an $f_s$ that not only reproduces the macroscopic expression for the pressure tensor (achieved by fixing parameters), but also that is consistent with quasineutrality ($\sigma(A_x,A_y)=0$), and in this case strict neutrality, $\phi=0$.
\end{description}
Let us first consider possible expressions for $P_{zz}(A_x,A_y)$. Pressure balance dictates that
\begin{equation}
P_{zz}(\tilde{z})=P_{T}-\frac{B_0^2}{2\mu_0}\left(C_1^2+2C_1C_2\tanh \tilde{z}+C_2^2\tanh^2\tilde{z}+C_3^2\right).\label{eq:pzzprofile}
\end{equation}
Using the knowledge that exponential functions are eigenfunctions of the Weierstrass transform \citep{Wolf-1977}, we would like to use exponential functions to represent the $P_{zz}$ function wherever possible. In Sections \ref{sec:tanhpressuredf} and \ref{sec:exppressuredf} we present two different attempts at using Channell's method for the AH+G field. The first requires a numerical approach, whereas the second can be completed analytically.

\section{The numerical/``$\tanh$'' equilibrium DF}\label{sec:tanhpressuredf}

\subsection{The pressure function}
From Equation (\ref{eq:ahgvecpotential}) we see that $\exp (2A_y/(C_2B_0L))=\text{sech}^2\tilde{z}\exp(-2C_1\tilde{z}/C_2)$, and so we can construct one part of the RHS of Equation (\ref{eq:pzzprofile}) by
\begin{equation}
\tanh^2\tilde{z}=1-\text{sech}^2\tilde{z}=1-\exp\left(\frac{2\tilde{A}_y}{C_2}\right)\exp\left(\frac{2C_1\tilde{A}_x}{C_2C_3}\right).\label{eq:tanh^2}
\end{equation}
The remaining task is to invert $\tanh\tilde{z}=\tanh\tilde{z}(\tilde{A}_x,\tilde{A}_y)$, and this is most readily achieved by
\begin{equation}
\tanh\tilde{z}=\tanh \left(\frac{\tilde{A}_x}{C_3}\right).\label{eq:tanh^1}
\end{equation}
Note that we have not chosen to take the square root of Equation (\ref{eq:tanh^2}), since we - naively - expect to be able to invert the Weierstrass transform for the expression in Equation (\ref{eq:tanh^1}) more easily (and in fact, it can be shown that one cannot solve Amp\`{e}re's Law by doing so). Substituting Equations (\ref{eq:tanh^1}) and (\ref{eq:tanh^2}) into Equation (\ref{eq:pzzprofile}) gives the pressure tensor
\begin{equation}
P_{zz}=P_0\left[C_2\exp\left(\frac{2\tilde{A}_y}{C_2}\right)\exp\left(\frac{2C_1\tilde{A}_x}{C_2C_3}\right)-2C_1\tanh \left(\frac{\tilde{A}_x}{C_3}\right)+C_b\right],\label{eq:tanh_pressure}
\end{equation}
with $C_b>2C_1$ for positivity of the pressure. There is \emph{a priori} no guarantee that this pressure tensor will satisfy Amp\`{e}re's law, $\partial P_{zz}/\partial\boldsymbol{A}=\boldsymbol{j}$. We can check the validity of the pressure with respect to Amp\`{e}re's law, by
\begin{equation*}
\frac{\partial P_{zz}}{\partial A_x}=\frac{P_0}{B_0L}\frac{\partial \tilde{P}_{zz}}{\partial \tilde{A}_x}=2\frac{P_0}{B_0L}\frac{C_1}{C_3}\left(  e^{2\tilde{A}_y/C_2}e^{2C_1\tilde{A}_x/(C_2C_3)}-\text{sech}^2(\tilde{A}_x/C_3)    \right)=0=j_x,
\end{equation*}
and
\begin{equation*}
\frac{\partial P_{zz}}{\partial A_y}=\frac{P_0}{B_0L}\frac{\partial \tilde{P}_{zz}}{\partial \tilde{A}_y}=\frac{2P_0}{B_0L}e^{2\tilde{A}_y/C_2}e^{2C_1\tilde{A}_x/(C_2C_3)}=\frac{2P_0}{B_0L}\text{sech}^2\tilde{z}=j_y\iff C_2=\frac{2\mu_0P_0}{B_0^2}.
\end{equation*}

\subsection{Inverting the Weierstrass transform}
As aforementioned, we can solve the inverse problem exactly for the exponential functions in Equation (\ref{eq:tanh_pressure}), using the fact that
\[
``g_{js}(p_{js})\propto \exp(\tilde{p}_{js})''\implies ``P_j \propto \exp(\tilde{A}_j)'',
\]
with the terminology of Chapter \ref{Vlasov}. Hence the challenge is to try to solve
\begin{equation}
\tanh (\tilde{A}_x/C_3)=\frac{1}{\sqrt{2\pi}}\int_{-\infty}^\infty \exp \left[ -\frac{1}{2}\left(  \tilde{p}_{xs}-\frac{\text{sgn}(q_{s})}{\delta_s}\tilde{A}_x \right)^2 \right] G_{s}(\tilde{p}_{xs})d\tilde{p}_{xs}.\label{eq:genetic}
\end{equation}
for some unknown $G_s$ function, one component of a DF of the form
\begin{equation}
f_s=\frac{n_{0s}}{(\sqrt{2\pi}v_{th,s})^3}e^{-\beta_sH_s}\left(   a_{0s}e^{\beta_su_{xs}p_{xs}}  e^{\beta_su_{ys}p_{ys}} +a_{1s}G_s(p_{xs})+b_s     \right) ,\label{eq:numericalDF}
\end{equation}
and such that the species-dependent constants are yet to be determined. It turns out that Equation (\ref{eq:genetic}) is not amenable to the Fourier transform method described in Section \ref{sec:ftransform} since there does not exist an analytic expression for the Fourier transform of the $\tanh $ function. Furthermore, one cannot use the Hermite polynomial expansion techniques as developed in Chapter \ref{Vlasov}, because the Maclaurin expansion for $\tanh x$,
\[
\tanh x =\sum_{n=0}^\infty \chi_n x^n,
\]
is only convergent for $|x|<\pi /2$. This is not a purely formal objection, for the following reason. Using the theory developed in Chapter \ref{Vlasov}, we could in principle construct a Hermite polynomial expansion for the $G_s$ function of the form
\[
G_s=\sum_{n=0}^\infty \chi_n \text{sgn}(q_s)^n\left(\frac{\delta_s}{\sqrt{2}}\right)^nH_n\left(\frac{p_{xs}}{\sqrt{2}m_sv_{\text{th},s}}\right),
\]
such that the the Weierstrass transform resulted in a Maclaurin series with the correct coefficients, $\chi_n$. However, the Hermite series is valid for all $p_{xs}$ - assuming that it is convergent - and there is \emph{a priori} no reason to restrict the range of the conjugate variable, $A_x$. Hence the result of the forward procedure is a pressure function that is not convergent for all $\tilde{A}_x$, and cannot equal the closed form on the LHS of Equation (\ref{eq:genetic}). Furthermore, since $\tilde{A}_x/C_3=\tilde{z}\in(-\infty,\infty)$, one can not even make an argument on the basis of \emph{accessibility} (i.e. claiming that this formal argument does not matter), which could possibly be justified if it were the case that $|\tilde{A}_x(\tilde{z})/C_3|<\pi/2\, \forall \tilde{z}$. In the absence of other analytical techniques, one must proceed with this problem numerically. We do not develop that approach in detail in this thesis, but we shall show some indicative results, to demonstrate the principle.

\begin{figure}
    \centering
   \begin{subfigure}[b]{0.55\textwidth}
        \includegraphics[width=\textwidth]{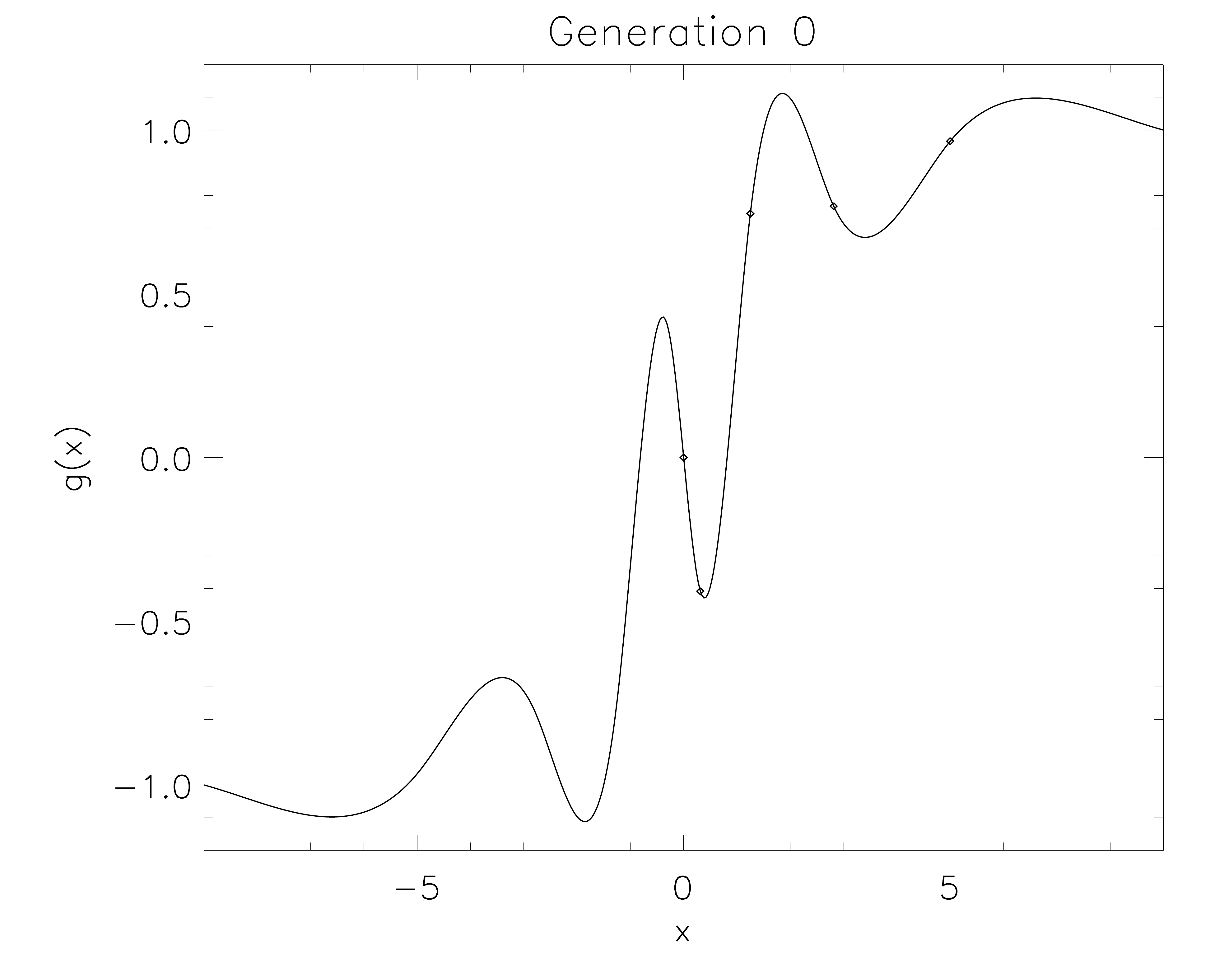}
        \caption{\small The `best fit' in the initial generation}
        \label{fig:000}
        \end{subfigure}
          \begin{subfigure}[b]{0.55\textwidth}
        \includegraphics[width=\textwidth]{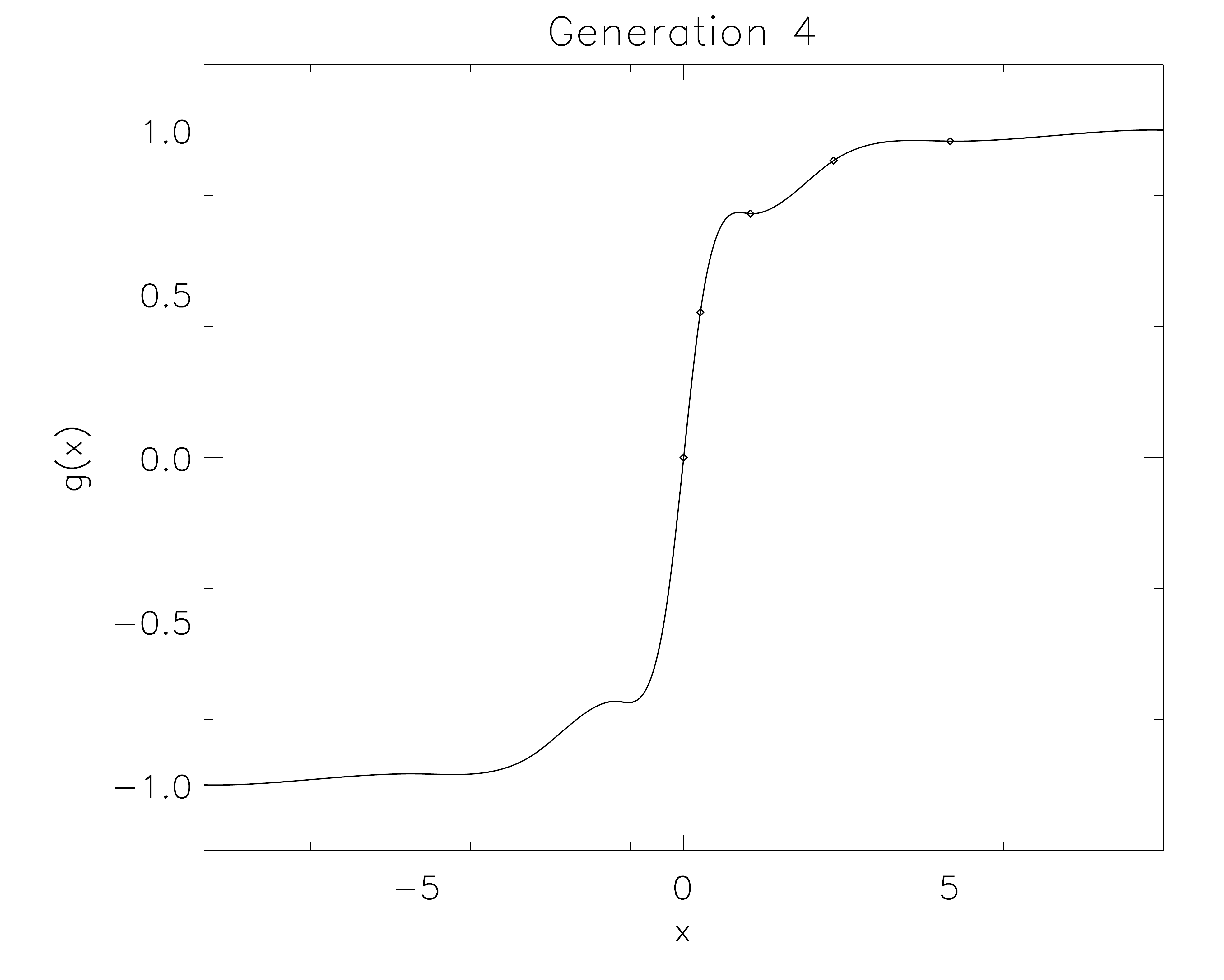}
        \caption{\small The `best fit' in the $4^{\text{th}}$ generation}
        \label{fig:004}
        \end{subfigure}
         \begin{subfigure}[b]{0.55\textwidth}
        \includegraphics[width=\textwidth]{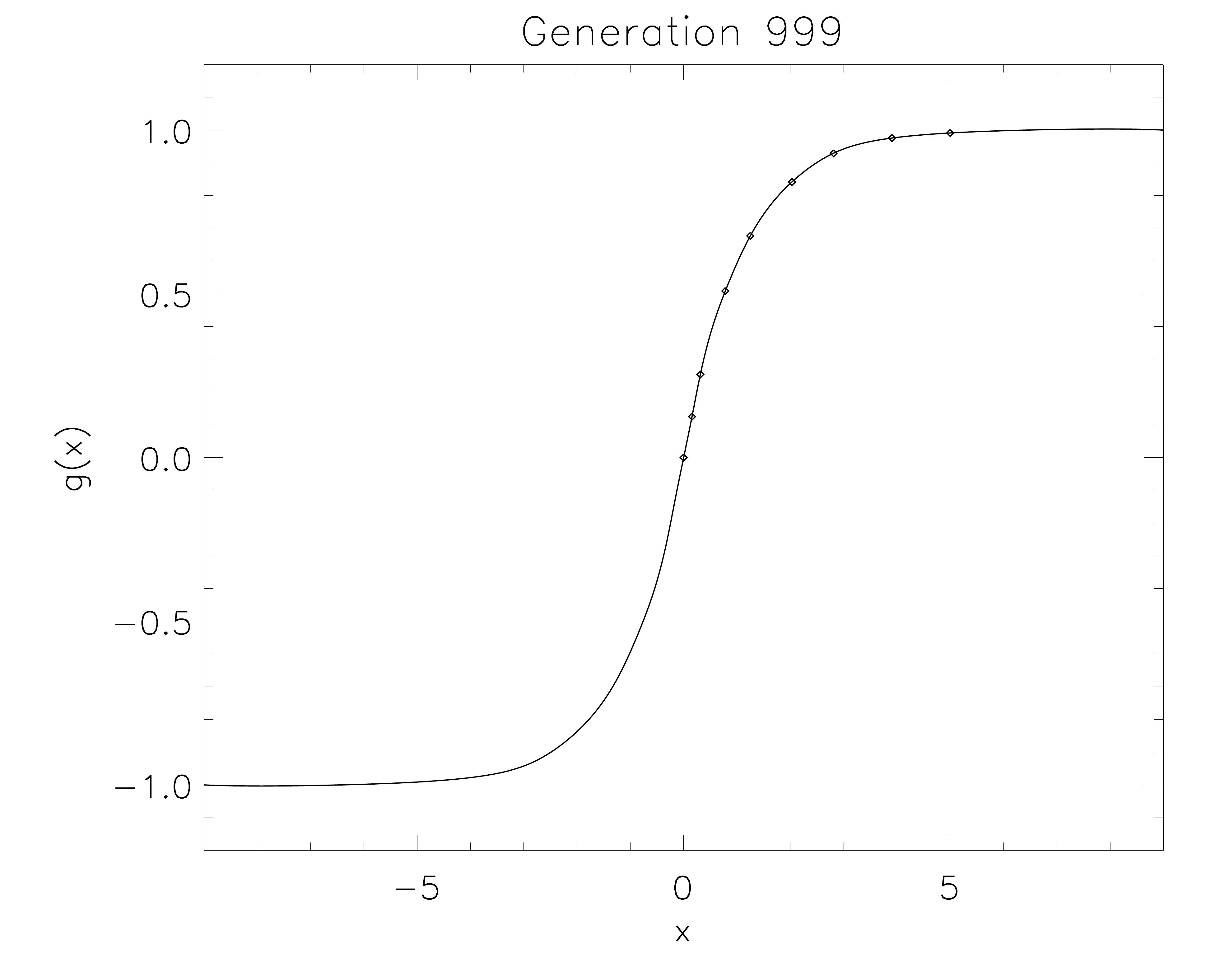}
        \caption{\small The `best fit' in the $999^{\text{th}}$ generation}
        \label{fig:999}
        \end{subfigure}
         \caption{\small The `most fit' numerical solution for the $G_{s}$ function at three separate generations (courtesy of J.D.B. Hodgson).}
 \label{fig:genetic}
        \end{figure}
        
        \begin{figure}
        \centering
        \includegraphics[width=0.7\textwidth]{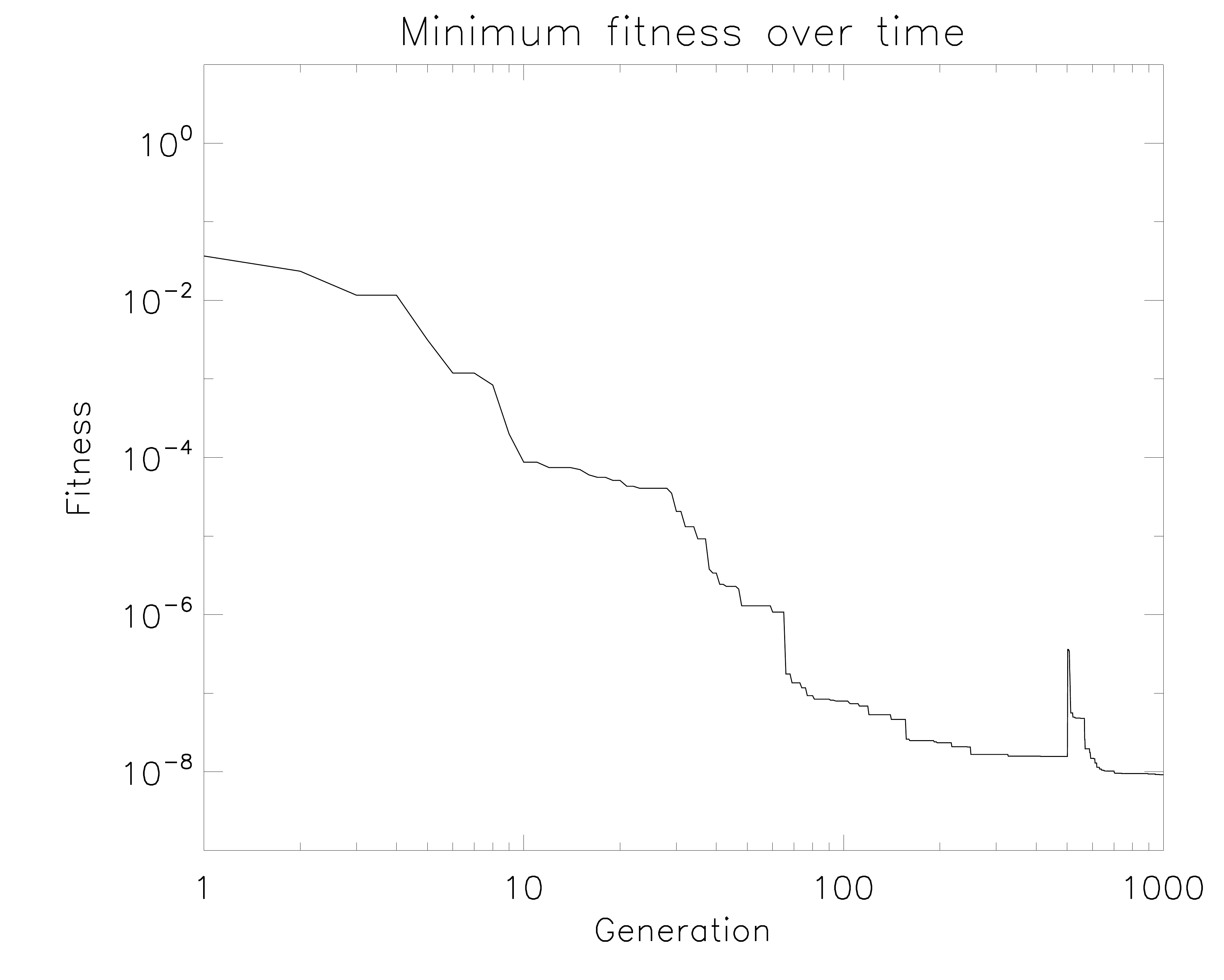}
        \caption{\small The minimum fitness/error through the generations (courtesy of J.D.B. Hodgson).}
        \label{fig:out}
    \end{figure}

In collaboration with J.D.B. Hodgson (who has led this particular effort), we have used \emph{Genetic algorithms} (e.g. see \cite{Holland-1975}) to construct numerical solutions for the $G_s$ function. My contribution to this project has been on the theoretical side, whereas J.D.B. Hodgson's has been the development of the algortithm and numerical approach, as well as Figures (\ref{fig:genetic}) and (\ref{fig:out}). The algorithm works by optimisation through random \emph{mutation}. One starts with an initial \emph{population} of candidate solutions to a problem, i.e. candidate $G_s$ functions that could solve Equation (\ref{eq:genetic}). Each member of the population (or \emph{chromosome}) is ranked according to some \emph{fitness} function. The population is then evolved in discrete steps (\emph{generations}), between which various mutations and \emph{genetic operations} occur, such that the fitness is hopefully optimised as $t\to\infty$.

Since the aim of the algorithm is - in general terms - to find a function $\mathcal{G}(p)$ that satisfies,
\[
P(A)=\int_a^b K(A,p)\mathcal{G}(p)dp,
\]
for known $P(A)$ and $K(A,p)$, a sensible fitness function to choose is
\[
F(\mathcal{G}(p))=\int_{A_0}^{A_1}  \left[  \int_a^b K(A,p)\mathcal{G}(p)dp \, -P(A)     \right]^2dA.\label{eq:fitness}
\]
In analytic terms, one would of course use $\pm \infty$ for all the relevant integral limits, but clearly one cannot do this in numerical computation. Figure \ref{fig:genetic} displays some results for a run of the algorithm through 1000 generations. Figures \ref{fig:000}, \ref{fig:004}, and \ref{fig:999} display the highest ranked chromosome of each population at the initial, $4^{\text{th}}$, and $999^{\text{th}}$ generations respectively. The highest ranked chromosome is the individual that best minimises the fitness function (Equation (\ref{eq:fitness})), which can be thought of as minimising the error. In Figure \ref{fig:out}, we show - on a loglog plot - the trend of the minimum fitness of each population, through the generations. The jump in the fitness around generation 500 identifies the point in the algorithm at which the grid resolution is increased, temporarily resulting in a larger error, which rapidly stabilises. An interesting feature of the `solution' given by Figure \ref{fig:999} is that it almost directly lies over the function $\tanh (\tilde{p}_{xs}/C_3)$, and hence it seems that 
\[
G_s(\tilde{p}_{xs})\approx \tanh (\tilde{p}_{xs}/C_3).
\]
The numerical procedure therefore seems to suggest that $\tanh x$ is close to a `numerical' eigenfunction of the Weierstrass transform, despite the fact that one cannot compute the Weierstrass transform of the $\tanh$ function.
  
Without an analytic expression for the function $G_s(\tilde{p}_{xs})$, we can make some progress in understanding the micro-macroscopic parameter relationships, and in calculating the bulk flow properties. Using standard integrals \citep{Gradshteyn}, we see that the DF in Equation (\ref{eq:numericalDF}) gives a pressure tensor of the form 
\begin{eqnarray}
P_{zz}&=&\sum_s m_s \int v_z^2 f_s d^3 v,\nonumber\\
&=&\sum_s\frac{n_{0s}}{\beta_s}\left(a_{0s}e^{(u_{xs}^2+u_{ys}^2)/(2v_{th,s}^2)}e^{\beta_su_{xs}q_sA_x}  e^{\beta_su_{ys}q_sA_y}+a_{1s}\tanh (\tilde{A}_x/C_3)+b_s\right).\nonumber
\end{eqnarray}
Channell's method dictates that this expression must match up with the macroscopic expression from Equation (\ref{eq:tanh_pressure}). This condition, as well as that of imposing $\sigma=0$ gives the following conditions
\begin{eqnarray}
\frac{2C_1}{C_2C_3B_0L}=e\beta_iu_{xi}=-e\beta_eu_{xe},\label{eq:uxs}\\
\frac{2}{C_2B_0L}=e\beta_iu_{yi}=-e\beta_eu_{ye},\nonumber\\
n_{0s}a_{0s}e^{(u_{xs}^2+u_{ys}^2)/(2v_{th,s}^2)}=:a_0=\frac{\beta_e\beta_i}{\beta_e+\beta_i}P_0C_2,\label{eq:a0}\\
n_{0s}a_{1s}=:a_1=-2\frac{\beta_e\beta_i}{\beta_e+\beta_i}P_0C_1,\label{eq:a1}\\
n_{0s}b_s=:b=\frac{\beta_e\beta_i}{\beta_e+\beta_i}P_0C_b,\nonumber
\end{eqnarray}
and, for completeness, the number density is given by
\begin{equation}
n_i=n_e:=n=a_0\text{sech}^2\tilde{z}+a_1\tanh\tilde{z}+b=\frac{\beta_e\beta_i}{\beta_e+\beta_i}P_{zz}.\nonumber
\end{equation}
The conditions listed above represent 10 constraints for 14 parameters $(\beta_s$, $n_{0s}$, $u_{xs}$, $u_{ys}$, $a_{0s}$, $a_{1s}$, $b_s)$, given macroscopic characteristics $B_0$, $P_0$, $C_1$, $C_2$, $C_3$, $C_b$ and $L$.

We can also calculate the bulk flow properties. In particular one should check that $j_x=0$. Using standard integrals \citep{Gradshteyn}, we see that
\begin{eqnarray}
j_x=0&=&\sum_sq_s\int f_s v_x d^3v,\nonumber\\
&=&\sum_sq_sn_{0s}\bigg[  a_{0s}u_{xs}e^{(u_{xs}^2+u_{ys}^2)/(2v_{th,s}^2)}e^{\beta_sq_su_{ys}A_y}e^{\beta_sq_su_{xs}A_x}\nonumber\\
&+&\frac{a_{1s}}{\sqrt{2\pi}v_{th,s}}\int_{-\infty}^{\infty} e^{-v_x^2/(2v_{th,s}^2)}v_xG_s(\tilde{p}_{xs})dv_x      \bigg].\nonumber
\end{eqnarray}
By differentiating Equation (\ref{eq:genetic}) with respect to $\tilde{A}_x$, we can see that
\begin{equation}
\int_{-\infty}^{\infty} e^{-v_x^2/(2v_{th,s}^2)}v_xG_s(\tilde{p}_{xs})dv_x=\frac{\sqrt{2\pi}\delta_s\text{sgn}(q_s)v_{th,s}^2}{C_3}\text{sech}^2(\tilde{A}_x/C_3).
\end{equation}
Plugging this back into the equation for $j_x$ gives
\begin{equation}
j_x=\sum_sq_s\left[ a_0u_{xs} \text{sech}^2\tilde{z} +\frac{a_1}{C_3\beta_sq_sB_0L}\text{sech}^2\tilde{z} \right],
\end{equation}
and substituting in Equation (\ref{eq:uxs}), and then Equations (\ref{eq:a0}) and (\ref{eq:a1}) gives
\begin{equation}
j_x=\frac{\beta_e\beta_i \text{sech}^2\tilde{z}}{(\beta_e+\beta_i)C_3B_0L}\sum_s\frac{1}{\beta_s}\left[ 2C_1P_0 -2P_0C_1 \right]=\frac{\text{sech}^2\tilde{z}P_0}{C_3B_0L}\sum_s0=0.
\end{equation}
In contrast to the solution found by \citet{Alpers-1969}, we see that this DF gives $V_{xs}\propto j_x =0$.

Similarly, we can calculate $j_y$,
\begin{eqnarray}
j_y&=&\sum_sq_s\int f_s v_y d^3v,\nonumber\\
&\vdots&\nonumber\\
&=&\frac{C_2B_0^2}{2\mu_0}\frac{2}{B_0L}\text{sech}^2\tilde{z}=j_y.\nonumber
\end{eqnarray}
The individual bulk velocities in the $y$ direction are proportional to the total current density, and go to zero at $\infty$, i.e. $V_{ys}\propto j_y$.

\section{The analytical/``exponential'' equilibrium DF}\label{sec:exppressuredf}
\subsection{The pressure tensor}
In this section we derive one more pressure tensor consistent with the AH+G field, that allows an exact analytical solution for the DF. The key step for analytic progress is to find distinct representations of $\tanh\tilde{z}=\tanh\tilde{z}(\tilde{A}_x,\tilde{A}_y)$ that allow inversion of the Weierstrass transform.

In a similar vein to the method in \citet{Alpers-1969}, we achieve this crucial step by identifying two distinct representations of $\tanh\tilde{z}(A_x,A_y)$,
\begin{eqnarray}
\tanh\tilde{z}&=&1-e^{-\tilde{z}}\text{sech}\tilde{z}=1-e^{\frac{C_1-C_2}{C_2C_3}\tilde{A}_x}e^{\frac{1}{C_2}\tilde{A}_y},\nonumber\\
\tanh\tilde{z}&=&\sqrt{1-\text{sech}^2\tilde{z}}=\sqrt{1-e^{\frac{2C_1}{C_2C_3}\tilde{A}_x}e^{\frac{2}{C_2}\tilde{A}_y}},\nonumber
\end{eqnarray}
These are composed as a linear combination, and then substituted into Equation (\ref{eq:scalar_pressure}) to give
\begin{eqnarray}
&&P_{zz}=P_{T}-\frac{B_0^2}{2\mu_0}\Bigg\{    C_1^2+C_3^2+2C_1C_2\left(  1-e^{\frac{C_1-C_2}{C_2C_3}\tilde{A}_x}e^{\frac{1}{C_2}\tilde{A}_y}    \right)   \nonumber\\
&&+C_2^2\left[ k\left( 1-e^{\frac{C_1-C_2}{C_2C_3}\tilde{A}_x}e^{\frac{1}{C_2}\tilde{A}_y}   \right)^2+(1-k)\left(  1-e^{\frac{2C_1}{C_2C_3}\tilde{A}_x}e^{\frac{2}{C_2}\tilde{A}_y} \right)      \right]    \Bigg\},\label{eq:pressure_tensor}
\end{eqnarray}
with $k$ a `separation constant'. Amp\`{e}re's Law implies that $P_{zz}$ must satisfy $\partial P_{zz}/\partial A_x(\tilde{z})=0$ and $\partial P_{zz}/\partial A_y(\tilde{z})=B_0C_2/(\mu_0L)\text{sech}^2\tilde{z}$, and it is seen to do so when $k=C_1/C_2$. In this case, Equation (\ref{eq:pressure_tensor}) can be re-written
\begin{eqnarray}
P_{zz}&=&P_{T}-\frac{B_0^2}{2\mu_0}\Bigg\{    C_1^2+C_3^2+-C_1C_2+C_1C_2\left(1+  \left(1-e^{\frac{C_1-C_2}{C_2C_3}\tilde{A}_x}e^{\frac{1}{C_2}\tilde{A}_y}   \right) \right)^2   \nonumber\\
&+&C_2(C_2-C_1)\left(  1-e^{\frac{2C_1}{C_2C_3}\tilde{A}_x}e^{\frac{2}{C_2}\tilde{A}_y} \right)          \Bigg\},\label{eq:pressure_tensor2}
\end{eqnarray}
An examination of the coefficients of the exponential functions in Equation (\ref{eq:pressure_tensor2}) tells us that $P_{zz}>0\, \forall\, (A_x,A_y)$ under the following conditions
\begin{eqnarray}
C_1C_2&<&0,\label{eq:c1c2neg}\\
P_T&>&\frac{B_0^2}{2\mu_0}\left[ C_1^2+C_3^2-C_1C_2+C_2^2-C_1C_2   \right]\nonumber\\
&&=\frac{B_0^2}{2\mu_0}\left[ C_1^2+C_2^2 +C_3^2-2C_1C_2  \right]\label{eq:ptotbound}
\end{eqnarray}
Now that a $P_{zz}>0$ has been found that satisfies Amp\`{e}re's Law and pressure balance, we can attempt to solve the inverse problem.

\subsection{The DF}
By comparison with Equation (\ref{eq:pressure_tensor2}) (in which $P_{zz}$ is written as a sum of exponential functions), we can suggest a form for the DF by using either `inspection and standard integral formulae' \citep{Gradshteyn}, Fourier transforms (see Section \ref{sec:ftransform}), or knowledge of eigenfunctions \citep{Wolf-1977}. The form that we choose is
\begin{eqnarray}
&&f_s=\frac{n_{0s}}{(\sqrt{2\pi}v_{\text{th},s})^3}e^{-\beta_{s}H_{s}}\bigg(a_{0s}e^{\beta_s(u_{xs}p_{xs}+u_{ys}p_{ys})}\nonumber\\
&&+a_{1s}e^{2\beta_s(u_{xs}p_{xs}+u_{ys}p_{ys})}+a_{2s}e^{\beta_s(v_{xs}p_{xs}+v_{ys}p_{ys})}+b_s  \bigg),\label{eq:DF_ansatz}
\end{eqnarray}
for $a_{0s}, a_{1s}, a_{2s}, b_s, u_{xs}, u_{ys}, v_{xs}$ and $v_{ys}$ as yet arbitrary constants, with the ``$a,b$'' constants dimensionless, and the ``$u,v$'' constants the bulk flows of particular particle populations (e.g. see \cite{Davidsonbook, Schindlerbook} and Section \ref{sec:Harristype}).

\subsubsection{Equilibrium parameters and their relationships}\label{secparameters}
We proceed with the necessary task of ensuring that the DF in Equation (\ref{eq:DF_ansatz}) exactly reproduces the correct pressure tensor expression of Equation (\ref{eq:pressure_tensor}). After some algebra we find the `micro-macroscopic' consistency relations by taking the $v_z^2$ moment of the DF, and these are displayed in Equations (\ref{eq:micromacro1} - \ref{eq:micromacro4}). 
\begin{eqnarray}
P_{T}-\frac{B_0^2}{2\mu_0}\left[(C_1+C_2)^2+C_3^2\right]=b\frac{\beta_e+\beta_i}{\beta_e\beta_i}, &\displaystyle \frac{C_1-C_2}{C_2C_3B_0L}=e\beta_iu_{xi}=-e\beta_eu_{xe}, \label{eq:micromacro1}\\
4C_1C_2\frac{B_0^2}{2\mu_0}=a_0\frac{\beta_e+\beta_i}{\beta_e\beta_i}, &   \displaystyle\frac{1}{C_2B_0L}=e\beta_iu_{yi}=-e\beta_eu_{ye} ,\label{eq:micromacro2}\\
 -C_1C_2\frac{B_0^2}{2\mu_0}=a_1\frac{\beta_e+\beta_i}{\beta_e\beta_i}, & \displaystyle\frac{2C_1}{C_2C_3B_0L}=e\beta_iv_{xi}=-e\beta_ev_{xe},\\
  C_2(C_2-C_1)\frac{B_0^2}{2\mu_0}=a_2\frac{\beta_e+\beta_i}{\beta_e\beta_i},& \displaystyle\frac{2}{C_2B_0L}=e\beta_iv_{yi}=-e\beta_ev_{ye}, \label{eq:micromacro4}
 \end{eqnarray}
We must also ensure that $n_i(A_x,A_y)=n_e(A_x,A_y)$ (for $n_{s}(A_{x},A_{y})$ the number density of species $s$) in order to be consistent with our assumption that $\phi=0$. The constants $a_0,a_1,a_2$ and $b$ are defined by these neutrality relations that complete this final step of the method, are found by calculating the zeroth order moment of the DF, and are written in Equations (\ref{eq:neutral1} - \ref{eq:neutral2}).
\begin{eqnarray}
a_{0}=n_{0s}a_{0s}e^{(u_{xs}^2+u_{ys}^2)/(2v_{\text{th},s}^2)},  &  a_{2}= n_{0s}a_{2s}e^{(v_{xs}^2+v_{ys}^2)/(2v_{\text{th},s}^2)}, \label{eq:neutral1}\\
 a_{1}=n_{0s}a_{1s}e^{2(u_{xs}^2+u_{ys}^2)/v_{\text{th},s}^2},   &   b=n_{0s}b_s.\label{eq:neutral2}
 \end{eqnarray}
These constraints are 16 in number, with 20 microscopic parameters $(\beta_s$, $n_{0s}$, $a_{0s}$, $a_{1s}$, $a_{2s}$, $b_s$, $u_{xs}$, $u_{ys}$, $v_{xs}$, $v_{ys})$, given chosen macroscopic parameters $(B_0$, $P_T$, $L$, $C_1$, $C_2$, $C_3)$.

\subsubsection{Non-negativity of the DF}
Since we integrate $f_s$ over velocity space to calculate $P_{zz}$, it is clear that non-negativity of $P_{zz}$ does not imply non-negativity of $f_s$. Furthermore, it is clear from Equations (\ref{eq:micromacro2}) and (\ref{eq:neutral1}) that $C_1C_2<0\implies a_{0s}<0$ (as well as $a_{1s}>0$, $a_{2s}>0$). We can also see by consideration of Equations (\ref{eq:ptotbound}) and (\ref{eq:neutral2}) that $b_s$>0. The fact that $a_{0s}<0$ is a cause for concern, regarding the positivity of the DF, given its form (Equation (\ref{eq:DF_ansatz})). However, by completing the square, the DF can be re-written as
\begin{eqnarray}
f_s=\frac{n_{0s}}{(\sqrt{2\pi}v_{\text{th},s})^3}e^{-\beta_{s}H_{s}}&\bigg[&\frac{1}{a_{1s}}\left(-\frac{a_{0s}}{2}+a_{1s}e^{\beta_s(u_{xs}p_{xs}+u_{ys}p_{ys})}\right)^2-\frac{a_{0s}^2}{4a_{1s}}\nonumber\\
&&+a_{2s}e^{\beta_s(v_{xs}p_{xs}+v_{ys}p_{ys})}+b_s\bigg].\nonumber
\end{eqnarray}
Hence we see that non-negativity of the DF is assured provided 
\begin{equation}
b_s\ge\frac{a_{0s}^2}{4a_{1s}}.\label{eq:bsbound}
\end{equation}

\subsubsection{The DF is a sum of Maxwellians}\label{sec:3max}
The equilibrium DF in equation (\ref{eq:DF_ansatz}) is written as a function of the constants of motion ($H_s, p_{xs},p_{ys}$), and this was suitable for constructing an exact equilibrium solution to the Vlasov equation. However, we can write $f_s$ explicitly as a function over phase-space ($z,\boldsymbol{v}$), in a form similar to that of the drifting Maxwellian in Equation (\ref{eq:Maxshift}). The DF can be re-written as 
\begin{eqnarray}
f_s(z,\boldsymbol{v})=\frac{1}{(\sqrt{2\pi}v_{\text{th},s})^3}&\bigg[&\mathcal{N}_{0s}(z)e^{-\frac{(\boldsymbol{v}-\boldsymbol{V}_{0s})^2}{2v_{\text{th},s}^2}}+\mathcal{N}_{1s}(z)e^{-\frac{(\boldsymbol{v}-\boldsymbol{V}_{1s})^2}{2v_{\text{th},s}^2}}\nonumber\\
&&+\mathcal{N}_{2s}(z)e^{-\frac{(\boldsymbol{v}-\boldsymbol{V}_{2s})^2}{2v_{\text{th},s}^2}}+be^{-\frac{\boldsymbol{v}^2}{2v_{\text{th},s}^2}}\bigg],\label{eq:sumMax}
\end{eqnarray}
for the density and bulk flow variables (``$\mathcal{N},\boldsymbol{V}$''), defined by 
\begin{eqnarray*}
\mathcal{N}_{0s}(z)&=&a_{0}e^{q_s\beta_{s} \boldsymbol{A}\cdot\boldsymbol{V}_{0s}}=a_0e^{-\tilde{z}}\text{sech}\tilde{z},\hspace{3mm} \boldsymbol{V}_{0s}=(u_{xs},u_{ys},0),  \\
\mathcal{N}_{1s}(z)&=&a_{1}e^{q_s\beta_{s} \boldsymbol{A}\cdot\boldsymbol{V}_{1s}}=a_1e^{-2\tilde{z}}\text{sech}^2\tilde{z},\hspace{3mm}\boldsymbol{V}_{1s}=(2u_{xs},2u_{ys},0),  \\
\mathcal{N}_{2s}(z)&=&a_{2}e^{q_s\beta_{s} \boldsymbol{A}\cdot\boldsymbol{V}_{2s}}=a_2\text{sech}^2\tilde{z},  \hspace{3mm} \boldsymbol{V}_{2s}=(v_{xs},v_{ys},0),\\
\end{eqnarray*}
respectively. The $u,v$ variables are normalised by $v_{\text{th},s}$ ($\tilde{u}_{xs}=u_{xs}/v_{\text{th},s}$ etc). This representation of $f_s$ has the advantages of having a clear visual/physical interpretation, and of being in a form readily implemented into PIC simulations as initial conditions. Despite the fact that each term of $f_s$ as written in Equation (\ref{eq:sumMax}) bears a strong resemblance to $f_{Maxw,s}$ as defined by Equation (\ref{eq:Maxshift}), $f_s$ is an exact Vlasov equilibrium DF, whereas $f_{Maxw,s}$ is not.

\subsection{Plots of the DF}
In order to plot the normalised DF, $\tilde{f}_s=f_s/\max f_s$, it is more convenient for Equations (\ref{eq:micromacro1}) - (\ref{eq:micromacro4}) to be expressed in dimensionless form. Making use of the dimensionless parameters also defined in Section \ref{sec:3max}, we have the following relationships
\begin{eqnarray}
\beta_{T}-\left[(C_1+C_2)^2+C_3^2\right]=bR,&&\hspace{3mm} \frac{(C_1-C_2)\delta_s^\star}{C_2C_3}=\tilde{u}_{xs}\nonumber\\
4C_1C_2=a_0R,&&\hspace{3mm}\frac{\delta_s^\star}{C_2}=\tilde{u}_{ys}, \nonumber\\
-C_1C_2=a_1R,&&\hspace{3mm} \frac{2C_1\delta_s^\star}{C_2C_3}=\tilde{v}_{xs},\nonumber\\
C_2(C_2-C_1)=a_2R,&&\hspace{3mm} \frac{2\delta_s^\star}{C_2}=\tilde{v}_{ys}.\nonumber
\end{eqnarray}
The signed magnetisation parameter $\delta_s^\star = m_sv_{\text{th},s}/(q_sB_0L)$ is the ratio of the (signed) thermal Larmor radius to the current sheet width, and the constants $R$ and $\beta_{T}$ defined by
\begin{eqnarray}
R&=&\frac{\beta_e+\beta_i}{\beta_e\beta_i} \frac{2\mu_0}{B_0^2},\nonumber\\
\beta_{T}&=&P_T\frac{2\mu_0}{B_0^2}\nonumber.
\end{eqnarray}
Hence, the normalised bulk flow parameters, $\tilde{u}_{xs}, \tilde{u}_{ys}, \tilde{v}_{xs}, \tilde{v}_{ys}$ are fixed by choosing the magnetisation, $\delta_s^\star$, and the magnetic field configuration, $C_1,C_2,C_3$. If in addition one chooses $n_{0s}$, and the ratio $R$ (note that $n_{0s}$R is dimensionless), then we see that the `density parameters' $a_0, a_1$ and $a_2$ are also fixed. In turn $a_{0s},a_{1s}$ and $a_{2s}$ are then fixed by Equations (\ref{eq:neutral1}) and (\ref{eq:neutral2}). Then, the lower bound on $b_s$ (for positivity of the DF) is determined by Equation (\ref{eq:bsbound}), and in turn we see a lower bound for $b$ and hence $\beta_{T}$.

Note that when $T_e=T_i:=T$ and $n_{0i}=n_{0e}:=n_0$, it is the case that $n_0R=2\beta_{pl}^\star$, for
\[
\beta_{pl}^\star=\frac{n_0k_BT}{B_0^2/(2\mu_0)},
\] 
a constant reference value for $\beta_{pl}$, which itself is spatially dependent. We shall also assume that $b_s=a_{0s}^2/(4a_{1s})$, and hence 
\[
\inf \tilde{f}_s =0 .
\]

In Figure \ref{fig:asymmetricdfs} we present plots of the DF in $(v_x/v_{\text{th},s},v_y/v_{\text{th},s})$ space, for $z/L=(0,0.1,1,10)$, and for the parameters 
\begin{eqnarray}
(\delta_i,R,n_0,C_1,C_2,C_3)&=&(0.2,0.1,1,-0.1,0.2,0.1),\nonumber\\
\implies(\tilde{u}_{xi},\tilde{u}_{yi},\tilde{v}_{xi},\tilde{v}_{yi})&=&(-3,1,-2,2)\nonumber.
\end{eqnarray}
We have chosen this particular parameter set, in order to clearly see that the VM equilibrium permits multiple maxima in velocity space, as is to be expected by a sum of drifting Maxwellians. However, whilst the plots of $\tilde{f}_i$ permit multiple maxima for $z/L=0,0.1,1$ in the parameter range chosen, we see that for large $z/L$ the DF is an isotropic Maxwellian, centred on $(0,0)$. This is consistent with no bulk flows $V_{xi},V_{xe}$ for large $\tilde{z}$, in contrast to the DF found by \citet{Alpers-1969}.

In particular, Figure \ref{fig:e1} shows $\tilde{f}_e$ for $\delta_e=\delta_i$, and hence
\[
(\tilde{u}_{xe},\tilde{u}_{ye},\tilde{v}_{xe},\tilde{v}_{ye})=-(\tilde{u}_{xi},\tilde{u}_{yi},\tilde{v}_{xi},\tilde{v}_{yi}),
\] 
with other parameters unchanged. As a result, we see that sending ``$q_i\to q_e$'' seems equivalent to sending ``$f_i(v_x/v_{\text{th},i},v_y/v_{\text{th},i})\to f_e(-v_x/v_{\text{th},e},-v_y/v_{\text{th},e})$''. However, for Figures \ref{fig:e2}, \ref{fig:e3} and \ref{fig:e4} we take $T_e=T_i$, and hence $\delta_e=\sqrt{m_e/m_i}\delta_i$, giving
\[
(\tilde{u}_{xe},\tilde{u}_{ye},\tilde{v}_{xe},\tilde{v}_{ye})\approx (  0.07,-0.02,0.05,-0.05     ).
\] 
The normalised bulk electron flow is now much smaller in magnitude, and this is represented in the figures. 

We note that there is a large portion of parameter space for which one sees no multiple maxima in velocity space (although we have not plotted these), indicating that the VM equilibrium that we present permits locally Maxwellian/thermalised - and hence micro-stable -DFs.

\begin{figure}
    \centering
    \begin{subfigure}[b]{0.45\textwidth}
        \includegraphics[width=\textwidth]{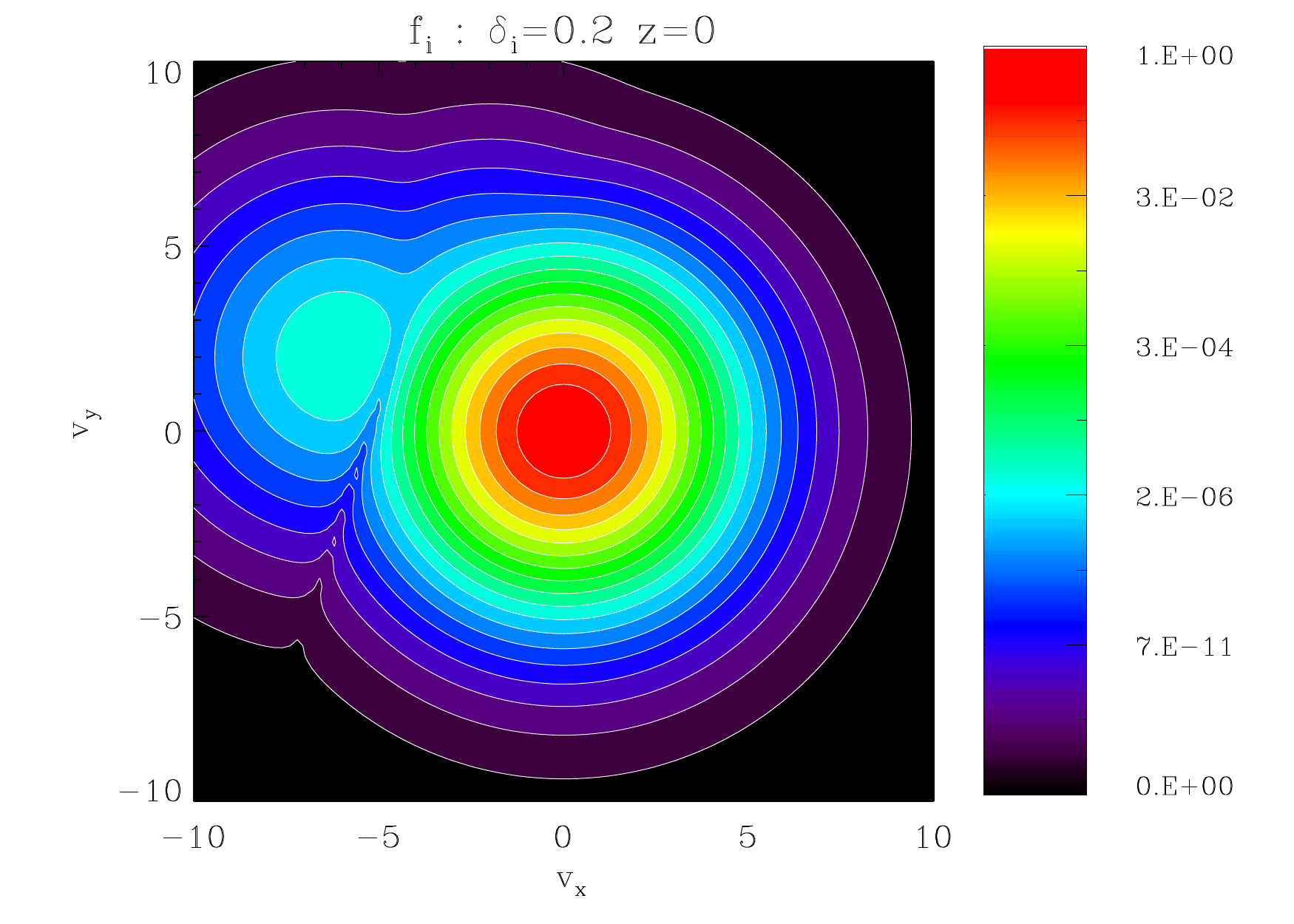}
        \caption{}
        \label{fig:ion1}
    \end{subfigure}
     \begin{subfigure}[b]{0.45\textwidth}
        \includegraphics[width=\textwidth]{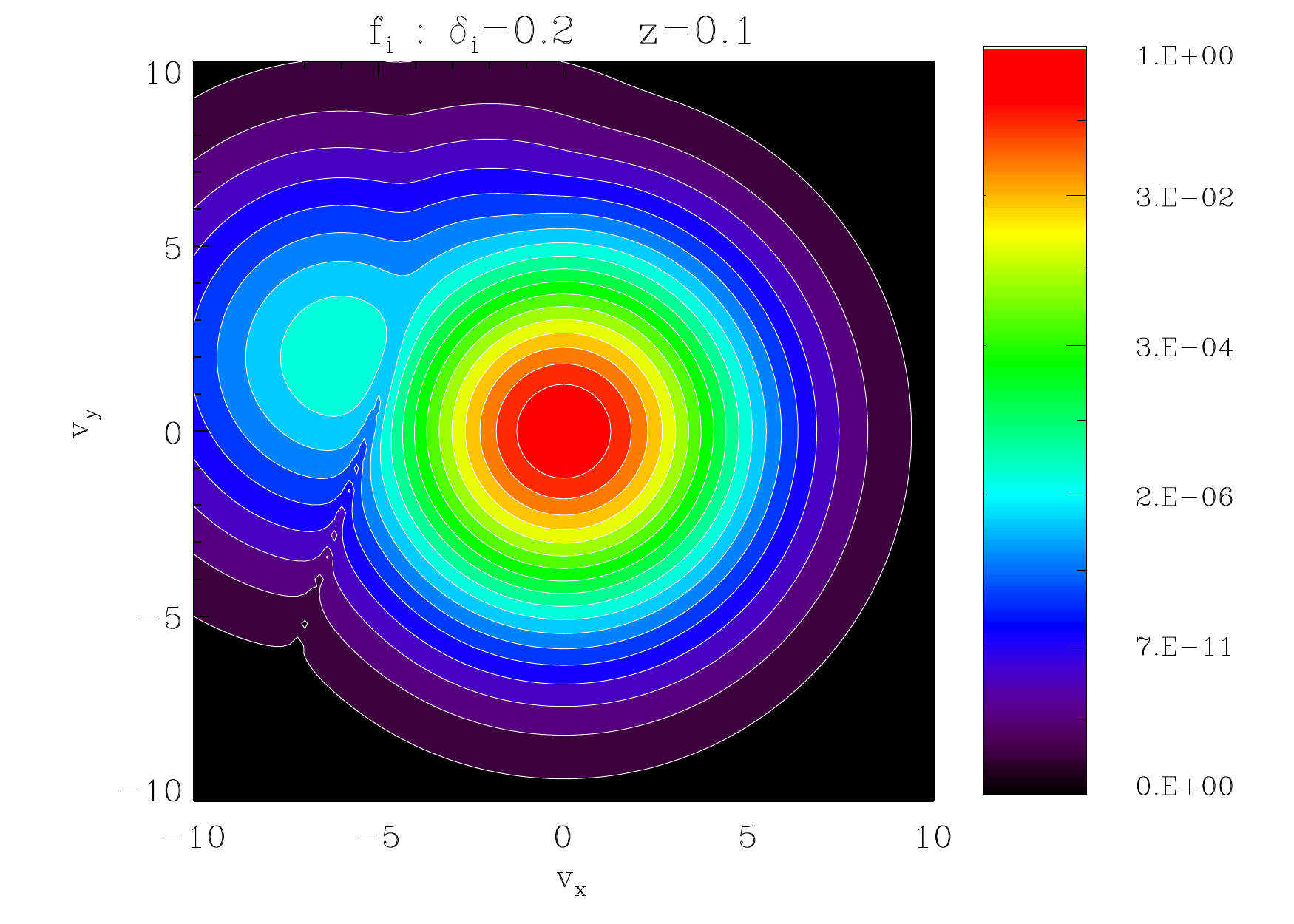}
        \caption{}
        \label{fig:ion2}
    \end{subfigure}
       \begin{subfigure}[b]{0.45\textwidth}
        \includegraphics[width=\textwidth]{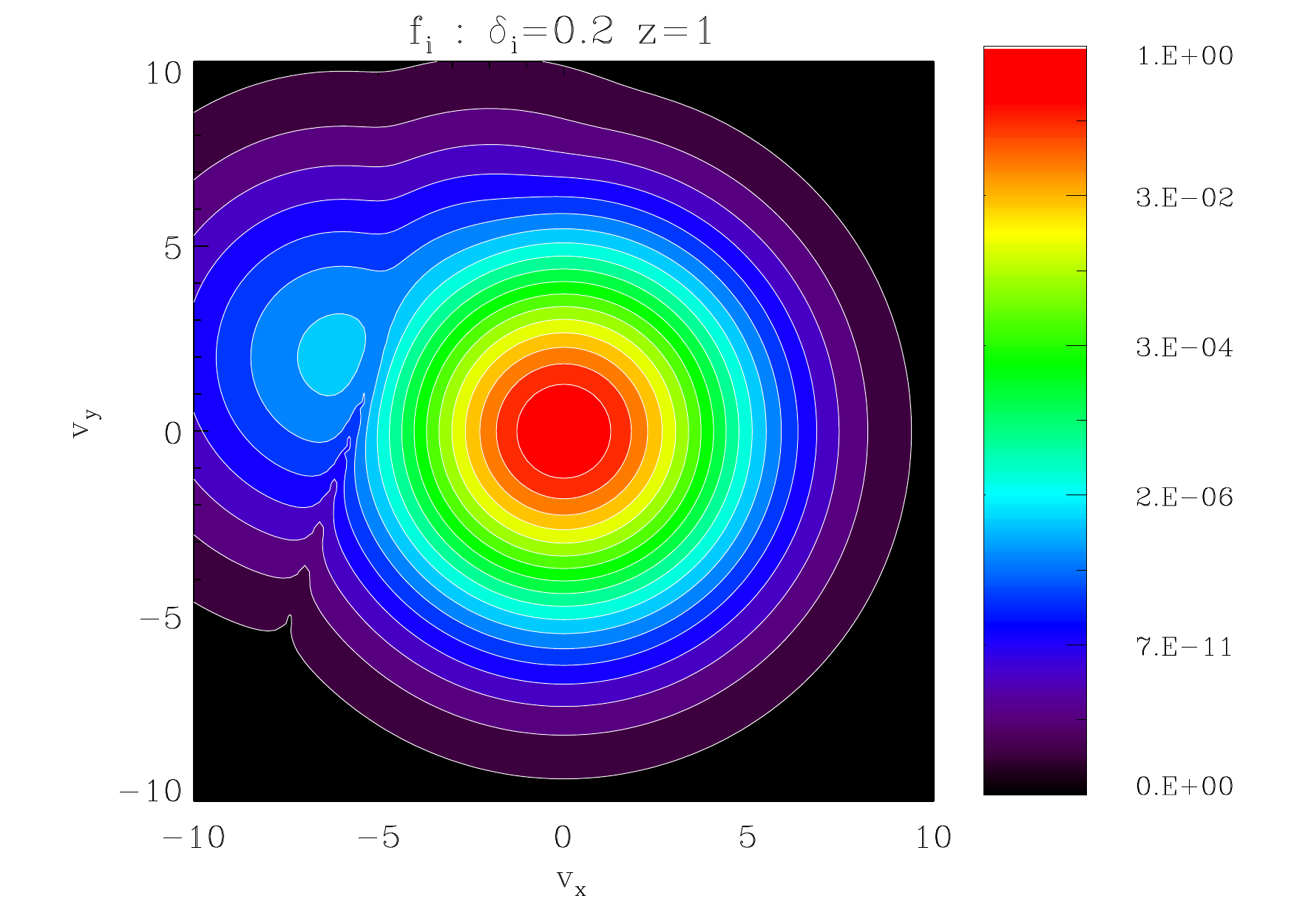}
        \caption{}
        \label{fig:ion3}
    \end{subfigure}
        \begin{subfigure}[b]{0.45\textwidth}
        \includegraphics[width=\textwidth]{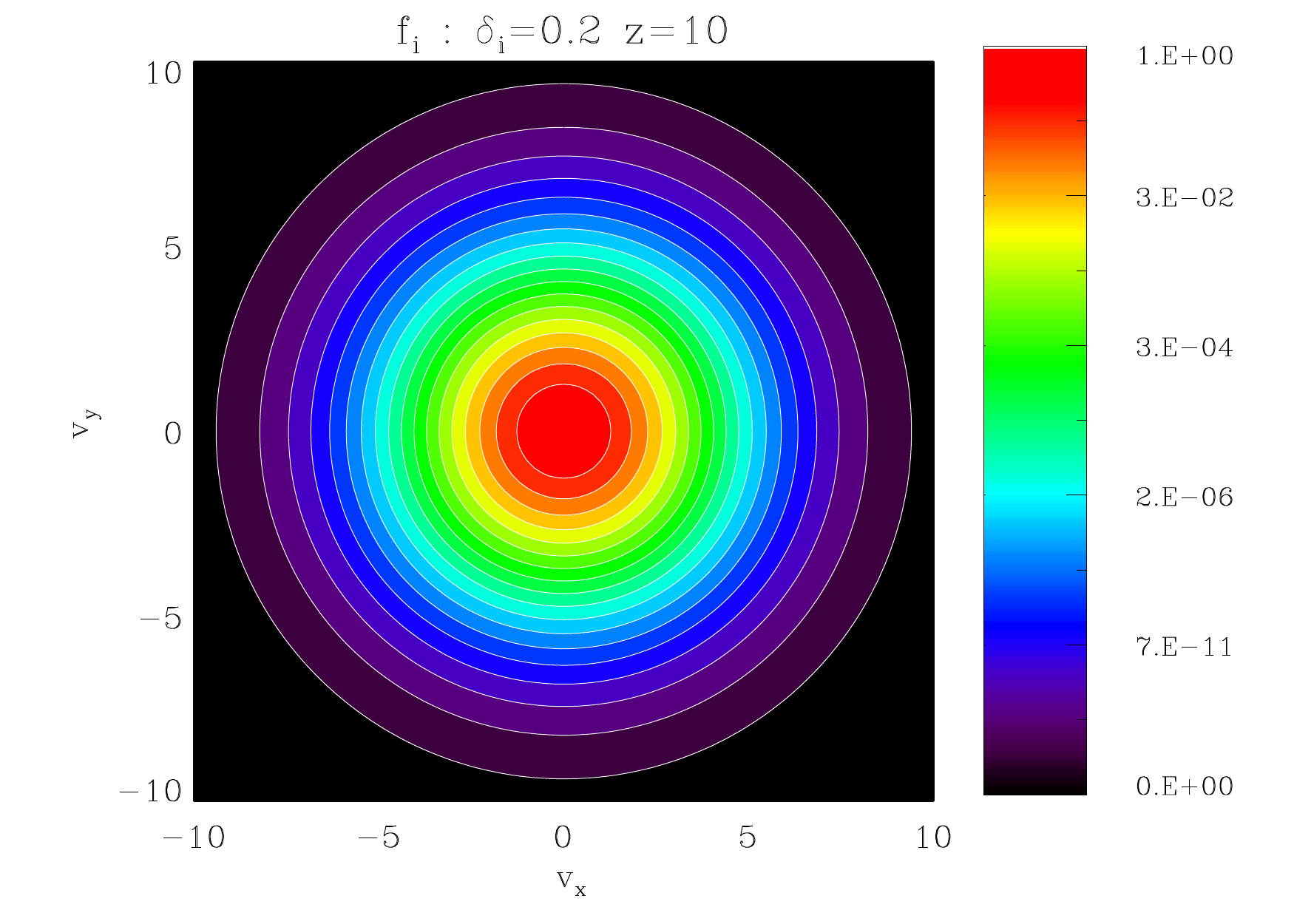}
        \caption{}
        \label{fig:ion4}
    \end{subfigure}
    \begin{subfigure}[b]{0.45\textwidth}
        \includegraphics[width=\textwidth]{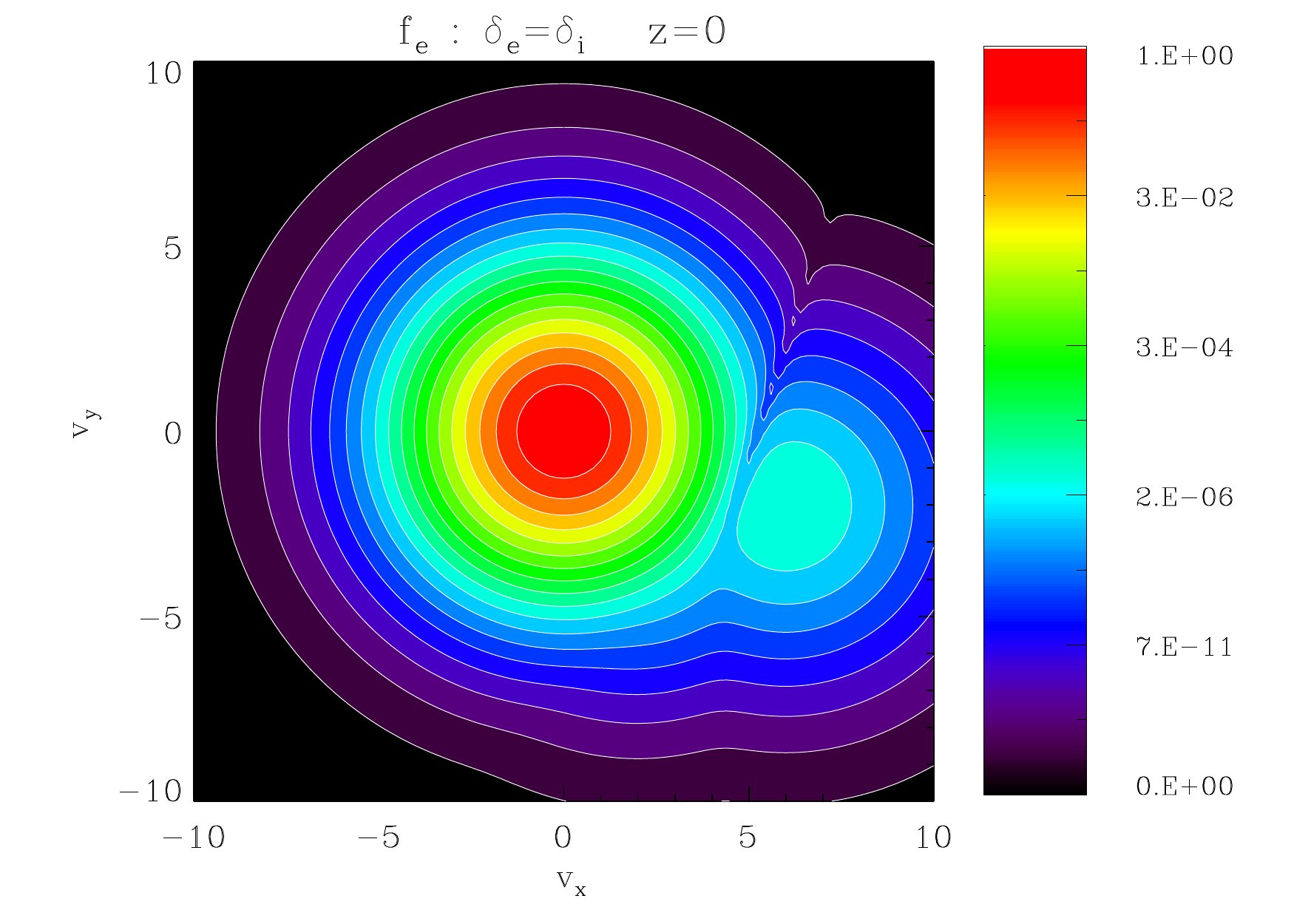}
        \caption{}
        \label{fig:e1}
    \end{subfigure}
     \begin{subfigure}[b]{0.45\textwidth}
        \includegraphics[width=\textwidth]{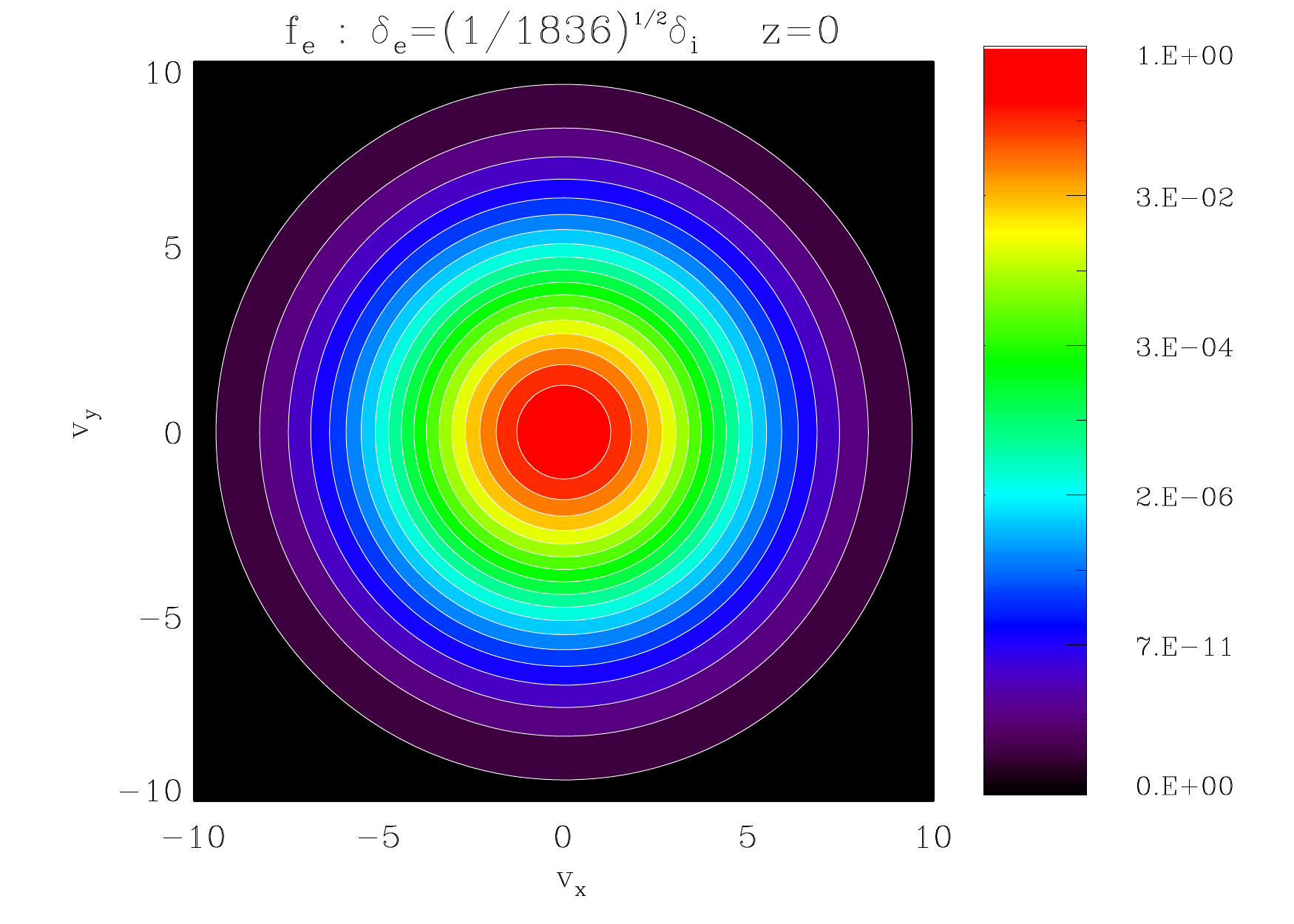}
        \caption{}
        \label{fig:e2}
    \end{subfigure}
       \begin{subfigure}[b]{0.45\textwidth}
        \includegraphics[width=\textwidth]{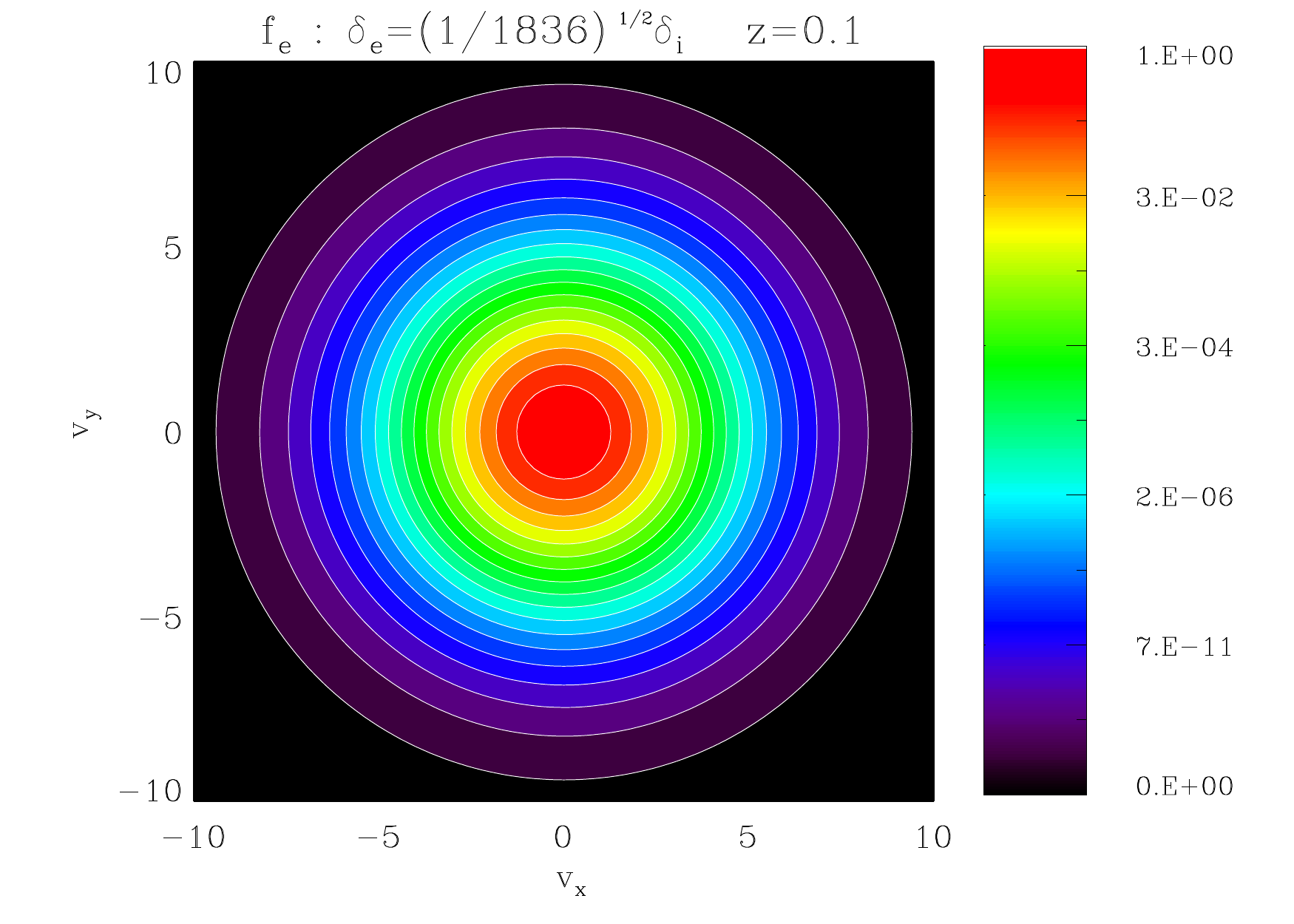}
        \caption{}
        \label{fig:e3}
    \end{subfigure}
        \begin{subfigure}[b]{0.45\textwidth}
        \includegraphics[width=\textwidth]{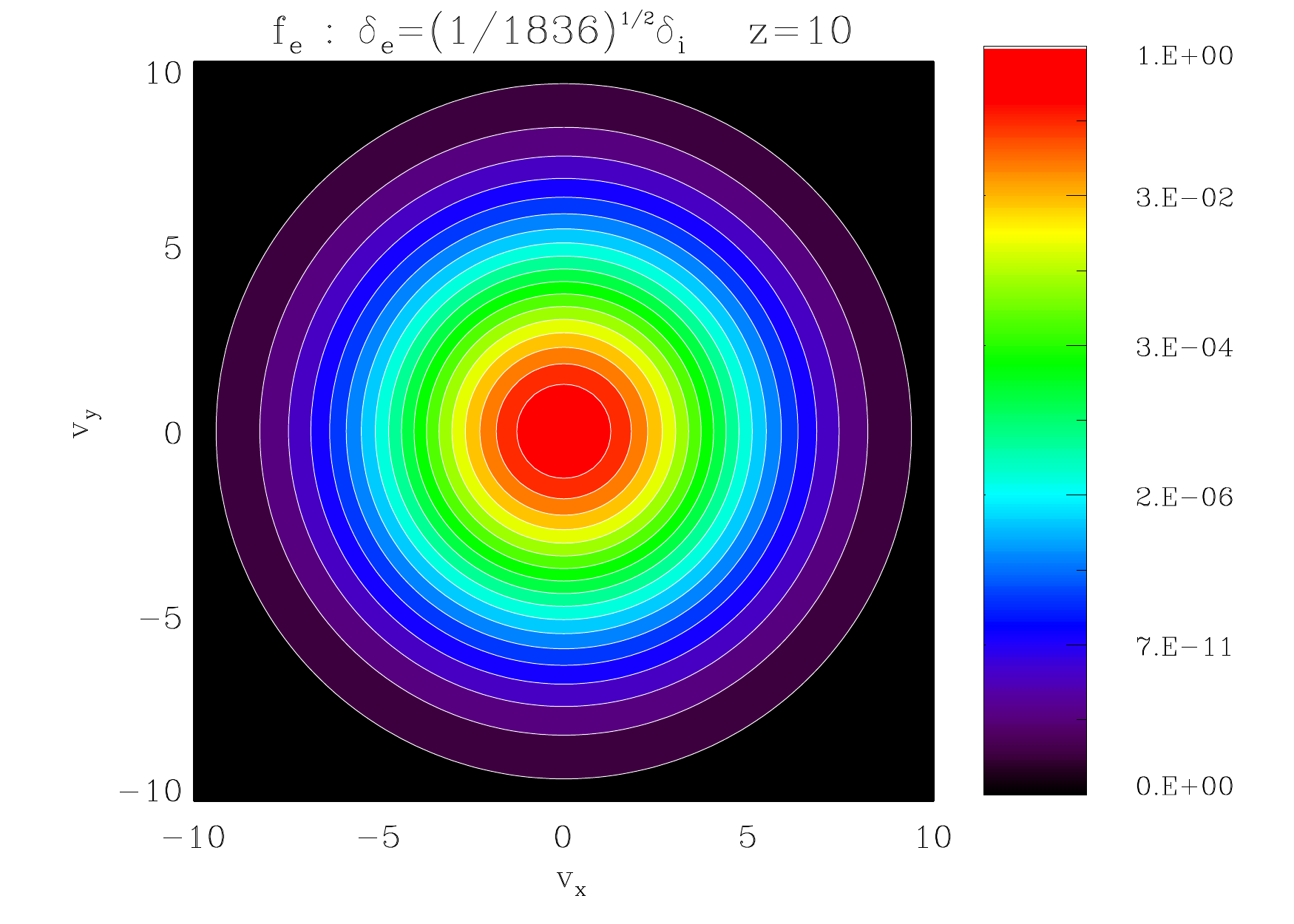}
        \caption{}
        \label{fig:e4}
    \end{subfigure}
        \caption{\small In Figures \ref{fig:ion1}, \ref{fig:ion2}, \ref{fig:ion3} and \ref{fig:ion4} we plot $\tilde{f}_i$ for $\delta_i=0.2$ and $z/L=0,0.1,1$ and $10$ respectively. In Figure \ref{fig:e1} we plot $\tilde{f}_e$ for $\delta_e=0.2$ and $z/L=0$. In Figures \ref{fig:e2}, \ref{fig:e3} and \ref{fig:e4} we plot $\tilde{f}_e$ for $\delta_e=\sqrt{m_e/m_i}\delta_i$ and $z/L=0,0.1$ and $10$ respectively.  }\label{fig:asymmetricdfs}
\end{figure}

\section{Discussion}
By considering the theory of the pressure tensor in vector-potential space (and its analogy with the problem of a particle in a potential), we have deduced that $P_{zz}$ must be a function of both $A_x$ and $A_y$, to describe a 1D asymmetric Harris current sheet with field reversal. This is - at first glance - a surprise, since there is only one component of the current density.

We have presented two valid $P_{zz}(A_x,A_y)$ functions that are self-consistent with an asymmetric Harris sheet plus guide field. One of these necessitated a numerical approach in order to solve for the DF, whereas the second allowed an analytical solution. The magnetic fields described by our models have often been used as asymmetric current sheet models for reconnection studies, and should be particularly suited to studying reconnection in Earth's dayside magnetopause.

The expression for the exact analytical VM equilibrium DF is elementary in form, and is written as a sum of exponential functions of the constants of motion, which can be re-written in $(z,\boldsymbol{v})$ space as a weighted sum of drifting Maxwellian DFs. This form for the DF can be readily used as initial conditions in particle-in-cell simulations. The equilibrium has zero mass flow far from the sheet, which is corroborated by the plots of the DF, and this is in contrast to the known exact analytical DF in the literature \citep{Alpers-1969}.

\null\newpage

\chapter{Neutral and non-neutral flux tube equilibria} \label{Cylindrical} 

\epigraph{\emph{Things are the way they are because they were the way they were.}}{\textit{Fred Hoyle}}

\noindent Much of the work in this chapter is drawn from \citet{Allanson-2016POP}

\section{Preamble}
In this chapter we calculate exact 1D collisionless plasma equilibria for a continuum of flux tube models, for which the total magnetic field is made up of the `force-free' Gold-Hoyle (GH) magnetic flux tube embedded in a uniform and anti-parallel background magnetic field. For a sufficiently weak background magnetic field, the axial component of the total magnetic field reverses at some finite radius. The presence of the background magnetic field means that the total system is not exactly force-free, but by reducing its magnitude, the departure from force-free can be made as small as desired. The DF for each species is a function of the three constants of motion; namely, the Hamiltonian and the canonical momenta in the axial and azimuthal directions. Poisson's equation and Amp\`{e}re's law are solved exactly, and the solution allows either electrically neutral or non-neutral configurations, depending on the values of the bulk ion and electron flows. These equilibria have possible applications in various solar, space, and astrophysical contexts, as well as in the laboratory.

The work in this chapter pertains to a cylindrical geometry, in which $r$ is the horizontal distance from the $z$ axis, and $\theta$ the azimuthal angle.

\section{Introduction}
Magnetic flux tubes and flux ropes are prevalent in the study of plasmas, with a wide variety of observed forms in nature and experiment, as well as uses and applications in numerical experiments and theory. Some examples of the environments and fields of study in which they feature include solar (e.g. \cite{Priest-2002,Magara-2003}); solar wind (e.g. \cite{Wangflux-1990, Borovsky-2008}); planetary magnetospheres (e.g. \cite{Sato-1986, Pontius-1990}) and magnetopauses (e.g. \cite{Cowley-1989}); astrophysical plasmas (e.g. \cite{Rogava-2000, Li-2006}); tokamak (e.g. \cite{Bottino-2007, Ham-2016}), laboratory pinch experiments (e.g. \cite{Rudakov-2000}), and the basic study of energy release in magnetised plasmas (e.g. \cite{Cowley-2015}), to give a small selection of references. 

One application of flux tubes is in the study of solar active regions (e.g. \cite{Fan-2009}) and the onset of solar flares and coronal mass ejections (e.g. \cite{Torok-2003, Titov-2003,Hood-2016}). A classic magnetohydrodynamic (MHD) model for magnetic flux tubes was first presented by T. Gold and F. Hoyle \citep{Gold-1960}, initially intended for use in the study of solar flares. The GH model is an infinite, straight, 1D and nonlinear force-free magnetic flux tube with constant `twist' \citep{Birn-2007}. Mathematically, the GH magnetic field could be regarded as the cylindrical analogue of the Force-Free Harris sheet \citep{Tassi-2008}, as the Bennett Pinch \citep{Bennett-1934} might be to the `original' Harris Sheet.  

It is typical to consider solar, space and astrophysical flux tubes within the framework of MHD (e.g. see \cite{Priest-2014}). However, many of these plasmas can be weakly collisional or collisionless, with values of the collisional free path large against any fluid scale \citep{Marsh-2006}, making a description using collisionless kinetic theory necessary. In this chapter, it is our intention to study the GH flux tube model beyond the MHD description, since - apart from the very recent work in \citet{Vinogradov-2016} - we see no attempt in the literature of a microscopic description of the GH field. 

The work in Chapters \ref{Vlasov} and \ref{Sheets}, as well as \citet{Alpers-1969, Harrison-2009POP, Harrison-2009PRL, Neukirch-2009, Wilson-2011, Abraham-Shrauner-2013, Kolotkov-2015}, used methods like Channell's \citep{Channell-1976} to tackle the VM inverse problem in Cartesian geometry. Channell described the extension of his work to cylindrical geometry as `not possible in a straightforward manner.' As explained in \citet{Tasso-2014} (in which cylindrical coordinates are used to model a torus), this is due in part to the `toroidicity' of the problem, i.e. the $1/r$ factor in the equations. As we shall see in this chapter, another potential complication is the need to allow -- at least in principle -- a non-zero charge density. 

There has been significant recent work on VM equilibria that are consistent with nonlinear force-free \citep{Harrison-2009POP, Harrison-2009PRL, Neukirch-2009, Wilson-2011, Abraham-Shrauner-2013, Kolotkov-2015, Allanson-2015POP, Allanson-2016JPP} and `nearly force-free' \citep{Artemyev-2011} magnetic fields in Cartesian geometry. VM equilibria for linear force-free fields have also been found in \citet{Sestero-1967, Bobrova-1979, Bobrova-2001}. Therein, force-free refers to a magnetic field for which the associated current density is exactly parallel, which is the definition we shall also use,
\begin{eqnarray}
\boldsymbol{j}\times\boldsymbol{B}=\frac{1}{\mu_0}(\nabla\times\boldsymbol{B})\times\boldsymbol{B}=\boldsymbol{0}.\nonumber
\end{eqnarray}
These works consider 1D collisionless current sheets, and so a natural question to consider is whether it is also possible to find self-consistent force-free (or nearly force-free) VM equilibria for other geometries, in particular cylindrical geometry. In this chapter we shall present particular VM equilibria for 1D magnetic fields which are nearly force-free in cylindrical geometry, i.e. flux tubes/ropes. These kinetic models and the the theory that follows are of potential applicability in the solar corona (e.g. see \cite{Wiegelmann, Hood-2016}), Earth's magnetotail (e.g. see \cite{Kivelson-1995, Khurana-1995, Slavin-2003, Yang-2014}) and magnetopause (e.g. \cite{Eastwood-2016}), planetary magnetospheres (e.g. \cite{DiBraccio-2015}), tokamak (e.g. \citep{Tasso-2007, Tasso-2014}) and laboratory (e.g. \cite{Davidsonbook}) plasmas.

\subsection{Previous work}
Two of the archetypal field configurations in cylindrical geometry are the $z$-Pinch and the $\theta$-pinch. The $z$-pinch has axial current and azimuthal magnetic field, 
\begin{equation}
\boldsymbol{j}\times\boldsymbol{B}=\nabla p \iff \frac{d}{dr}\left(p+\frac{B_{\theta}^2}{2\mu_0}\right)+\frac{B_{\theta}^2}{\mu_0r}=0,\nonumber
\end{equation}
 \citep{Freidbergbook}, a classical example of which is the Bennett Pinch
\begin{equation}
\tilde{B}_{\theta}=-\frac{\tilde{r}}{1+\tilde{r}^2},
\end{equation}
written in non-dimensional units, and for which a Vlasov equilibrium is well known \citep{Bennett-1934, Harris-1962}. In contrast, the $\theta$-Pinch has azimuthal current and axial magnetic field,
\begin{equation}
\boldsymbol{j}\times\boldsymbol{B}=\nabla p \iff \frac{d}{dr}\left(p+\frac{B_{z}^2}{2\mu_0}\right)=0.\nonumber
\end{equation}
Pinches that have both axial and azimuthal magnetic fields are known as screw or cylindrical pinches, e.g. see \citet{Freidbergbook, Carlqvist-1988}.

Consideration of `Vlasov-fluid' models of $z$-Pinch equilibria was given in \citet{Channon-2001}, with \citet{Mahajan-1989} calculating $z$-Pinch equilibria and an extension with azimuthal ion-currents. Others have also constructed kinetic models of the $\theta$-pinch, see \citet{Nicholson-1963, Batchelor-1975} for examples. In the same year as \citet{Pfirsch-1962}, cylindrical kinetic equilibria with only azimuthal currents were studied in \citet{Komarov-1962}. For examples of treatments of the stability of fluid and kinetic linear pinches, see \citet{Newcomb-1960,Pfirsch-1962,Davidsonbook} respectively. 

Recently there have been studies on `tokamak-like' VM equilibria with flows \citep{Tasso-2007, Tasso-2014}, starting from the VM equation in cylindrical geometry and working towards Grad-Shafranov equations for the vector potential. We also note two Vlasov equilibrium DFs in the literature that are close in style to the one that we shall present. The first is described in a brief paper \citep{El-Nadi-1976}, with an equilibrium presented for a cylindrical pinch. However, their distribution describes a different magnetic field and the DF appears not to be positive over all phase space. The second DF is a very recent paper that actually describes a magnetic field much like the one that we discuss  \citep{Vinogradov-2016}. Their DF is designed to model `ion-scale' flux tubes in the Earth's magnetosphere. Formally, their quasineutral model approaches a nonlinear force-free configuration in the limit of a vanishing electron to ion mass ratio. In their model, current is carried exclusively by electrons and the non-negativity of the DF depends on a suitable choice of microscopic parameters. Finally, we mention that in beam physics (e.g. see \cite{Morozov-1961, Hammer-1970, Gratreau-1977, Uhm-1985}), much work on constructing cylindrical VM equilibria is done by looking for mono-energetic distributions with conserved angular momentum,
\begin{equation}
f_s=\delta(H_s-H_{0s})g(p_{\theta s}),\nonumber
\end{equation}
for $H_{0s}$ a fixed energy, $H_s$ and $p_{\theta s}$ the Hamiltonian and angular momentum respectively. 



This chapter is structured as follows. In Section \ref{sec:theory} we first review the theory of the equation of motion consistent with a collisionless DF in cylindrical geometry, and discuss the question of the possibility of 1D force-free equilibria. Then we introduce the magnetic field to be used. We note that whilst the work in this chapter is applied to a particular magnetic field from Section \ref{subsec:magfield} onwards, the steps taken to calculate the equilibrium DF seem as though they could be adaptable to other cases. In Section \ref{sec:thedf} we present the form of the DF that gives the required macroscopic equilibrium, and proceed to `fix' the parameters of the DF by explicitly solving Amp\`{e}re's Law and Poisson's Equation. Note that whilst we choose to consider a two-species plasma of ions and electrons, we see no obvious reason preventing the work in this chapter being used to describe plasmas with a different composition. In Section \ref{sec:analysis} we present a preliminary analysis of the physical properties of the equilibrium. The analysis includes discussions on non-neutrality and the electric field; the equation of state and the plasma beta; the origin of individual terms in the equation of motion; plots of the DF; as well as particularly technical calculations in Sections \ref{sec:moments}, \ref{app:vthetamax} and \ref{app:vzmax}. Section \ref{sec:moments} contains the zeroth and first order moment calculations, used to find the number densities and bulk flows directly, and in turn the charge and current densities. Sections \ref{app:vthetamax} and \ref{app:vzmax} contain the mathematical details of the existence and location of multiple maxima of the DF in velocity-space.

The work in this chapter does not present a generalised method for the VM inverse problem in cylindrical geometry, but instead some particular solutions for a specific given magnetic field. Other than any interesting theoretical advances, a possible application of the results of this study could be to implement the obtained model in kinetic (particle) numerical simulations.

\section{General theory}\label{sec:theory}
\subsection{Vlasov equation in time-independent orthogonal coordinates}
A collisionless equilibrium is characterised by the 1-particle DF, $f_s$, a solution of the steady-state Vlasov Equation (e.g. see \cite{Schindlerbook}). The Vlasov equation can be written \citep{Santini-1970} in index notation as
\begin{equation}
\frac{\partial f_s}{\partial t} +\frac{1}{\sqrt{g}}\frac{\partial}{\partial x^i}\left(\sqrt{g}\frac{dx^i}{d t}f_s\right)+\frac{\partial}{\partial v^i}\left(\frac{d v^i}{d t}f_s\right)=0, \label{eq:Vlasov_index}
\end{equation}
for $i\in\{1,2,3\}$; time-independent orthogonal coordinates given by $x^i\in (x^1,x^2,x^3)$; orthogonal and orthonormal basis vectors defined by $e_i$ and $\hat{e}_i$ respectively; the diagonal metric tensor $g_{ij}=g_{ij}(x^1,x^2,x^3)=e_ie_j$, such that distances in configuration-space obey
\[
ds^2=g_{11}(dx^1)^2+g_{22}(dx^2)^2+g_{33}(dx^3)^2;
\]
$g=\text{Det}[g_{ij}]=g_{11}g_{22}g_{33}$; velocities given by $\boldsymbol{v}=v^i\hat{e}_i=\sqrt{g_{ii}}\,dx^i/dt\, \hat{e}_i$; and the Einstein summation convention applied such that repeated indicies are summed over, i.e.
\[
A_iB^i=A_1B^1+A_2B^2+A_3B^3.
\]
Superscript and subscript indices represent contra- and co-variant tensor components respectively, with the metric tensor able to raise or lower these indices, e.g.
\[
x_j=g_{ij}x^i,
\]
such that 
\[
\nabla=e_i\nabla^i=g_{ij}e^j\nabla^i=e^j\nabla_j(=e^i\nabla_i=\nabla ),
\]
and 
\[
\nabla_i=\frac{\partial}{\partial x^i}
\]
(see e.g. \cite{Leonhardtbook, Landaufields} for good introductions to index notation).

Equation (\ref{eq:Vlasov_index}) can be re-written in vector notation \citep{Santini-1970} as 
\begin{equation}
\frac{\partial f_s}{\partial t} +\boldsymbol{v}\cdot\nabla f +\frac{q_s}{m_s} \left[(\boldsymbol{E}+\boldsymbol{v}\times\boldsymbol{B})-\boldsymbol{v}\times(\nabla\times\boldsymbol{v})\right]\cdot \frac{\partial f_s}{\partial \boldsymbol{v}}=0, \label{eq:Vlasov_vector}
\end{equation}
for 
\begin{equation}
\boldsymbol{v}=v^i\hat{e}_i,\hspace{3mm}\frac{\partial f_s}{\partial \boldsymbol{v}}=\hat{e}^i\frac{\partial f_s}{\partial v^i},\nonumber
\end{equation}
In Cartesian geometry, equation (\ref{eq:Vlasov_vector}) reduces to a familiar form since the Cartesian basis vectors are position-independent, i.e.
\[
\nabla\times\boldsymbol{v}=v^x\nabla\times\hat{e}_x+v^y\nabla\times\hat{e}_y+v^z\nabla\times\hat{e}_z=\boldsymbol{0}.
\]
\subsection{Vlasov equation in cylindrical geometry}
In cylindrical geometry $(x^1=r,x^2=\theta,x^3=z)$, $(\hat{e}_r=e_r=\hat{r},\hat{e}_\theta=\frac{1}{r}e_\theta=\hat{\theta},\hat{e}_z=e_z=\hat{z})$, $\nabla\times\boldsymbol{v}=v^\theta \hat{r}\times\nabla\theta$, and equation (\ref{eq:Vlasov_vector}) can be shown to reduce to
\begin{eqnarray}
\frac{\partial f_s}{\partial t}+\boldsymbol{v}\cdot\nabla f_s+\frac{q_s}{m_s}\left(\boldsymbol{E}+\boldsymbol{v}\times\boldsymbol{B}\right)\cdot\frac{\partial f_s}{\partial \boldsymbol{v}}+\left[\frac{v_{\theta}^2}{r}\frac{\partial f_s}{\partial v_{r}}-\frac{v_{r}v_{\theta}}{r}\frac{\partial f_s}{\partial v_{\theta}}     \right]=0,\label{eq:Vlasov_cyl}
\end{eqnarray}
e.g. see \citet{Komarov-1962, Santini-1970} and \citet{Tasso-2007}. Note that the gradient operator in cylindrical coordinates is given by
\[ 
\nabla = \hat{r} \frac{\partial}{\partial r} + \hat{\theta}\frac{1}{r}\frac{\partial}{\partial \theta}+ \hat{z}\frac{\partial}{\partial z}  ,
\] 
such that the matrix representation of the metric tensor, $\underline{\underline{g}}=\text{Mat}[g_{ij}]$, is given by
\begin{eqnarray}
\underline{\underline{g}}=\left(\begin{matrix}
  1 & 0 & 0 \\
  0 & r^2 & 0 \\
  0 & 0 & 1
 \end{matrix}\right).\nonumber
\end{eqnarray}
The `fluid' equation of motion of a particular species $s$ is found by taking first-order velocity moments of the Vlasov equation. For the purposes of completeness and future reference the full first order moment-taking calculation is performed in Section \ref{subsec:eom}, since it is not easily found in the literature, to our knowledge. The result is that for an arbitrary DF that only depends spatially on $r$, the equation of motion can almost be written in a familiar form, as compared to the equation in Cartesian geometry \citep{Mynick-1979a, Greene-1993, Schindlerbook}, but with some `additional' terms. This is to be expected, given the form of equation (\ref{eq:Vlasov_cyl}).

\subsection{Equation of Motion in cylindrical geometry}\label{subsec:eom}
It will be useful to-rewrite the Vlasov equation from Equation (\ref{eq:Vlasov_cyl}) in index notation, in order to take the velocity moments. 
As such, the Vlasov equation can be written according to
\begin{equation}
\frac{\partial f_s}{\partial t}+v^{i}\nabla_if_s+\frac{q_s}{m_s}\left(E_i+\varepsilon_{ijk}v^jB^k\right)\nabla_{v_i}f_s+\left[\frac{(v^{\theta})^2}{r}\nabla_{v^r}f_s-\frac{v^rv^{\theta}}{r}\nabla_{v^\theta}f_s     \right]=0.\label{eq:vlasov2}
\end{equation}
The totally antisymmetric unit tensor of rank 3 (the Levi-Civita tensor) is $\varepsilon_{ijk}$, and it takes the value $0$ when any of its indices are repeated (e.g. $\varepsilon_{131}=0$), $+\sqrt{g}$ for an `ordered triplet' (e.g. $\varepsilon_{231}=\sqrt{g}$), and $-\sqrt{g}$ for a `disordered triplet' (e.g. $\varepsilon_{213}=-\sqrt{g}$). 
The first moment of the Vlasov equation (Equation (\ref{eq:vlasov2})), and multiplied by $m_s$, gives
\begin{eqnarray}
&&m_s\int \bigg\{\underbrace{v_i\frac{\partial f_s}{\partial t}}_{ \text{ A      }    }+\underbrace{v_iv_j\nabla^{j}f_s}_{\text{B}}+\frac{q_s}{m_s}\bigg(\underbrace{v_iE_j\nabla_{v_j}f_s}_{\text{C}}+\underbrace{v_i\epsilon_{jkl}v^kB^l\nabla_{v_j}f_s}_{\text{D}}\bigg)\nonumber\\
&&+\underbrace{v_i\frac{(v^{\theta})^2}{r}\nabla_{v^r}f_s}_{\text{E}}-\underbrace{v_i\frac{v^rv^{\theta}}{r}\nabla_{v^\theta}f_s}_{\text{F}}     \bigg\}d^3v=0,
\end{eqnarray}
with the triple integral written in shorthand by
\[
\int d^3v :=\int_{-\infty}^\infty\int_{-\infty}^\infty\int_{-\infty}^\infty dv_rdv_\theta dv_z.
\]

The first term, `A', gives $\partial/\partial t (\rho_sV_{js})$. Next, we notice that the spatial derivative in `B' can be taken outside of the integral. Then, if we write $v_i=V_{is}+w_{js}$, we see that by Leibniz' rule, for $\nabla$ a derivative
\begin{equation*}
\nabla\langle v_iv_j \rangle=\nabla(V_{is}V_{js})+\nabla \langle w_{is}w_{js}\rangle+\nabla ( V_{is}\langle w_{js}\rangle)+\nabla ( V_{js}\langle w_{is}\rangle),
\end{equation*}
with the angle brackets denoting an integral over velocity space (by definition $\langle w_i\rangle=0$). As a result, `B' becomes
\begin{equation*}  
\nabla^{j}(\rho_sV_{js}V_{is})+\nabla^jP_{ij}   .
\end{equation*}

We shall integrate terms `C-F' by parts and neglect surface terms, i.e. we assume that
\[
\lim_{|\boldsymbol{v}|\to\infty}\mathcal{G}(\boldsymbol{x},\boldsymbol{v},t)f_s(\boldsymbol{x},\boldsymbol{v},t)=0,
\]
for $\mathcal{G}$ representing the different variables multiplying the DF in terms `C-F'. As a result `C' and `D' become $-\sigma_sE_i$ and $-\sigma_s\epsilon_{ijk}V^j_{s}B^k$ respectively. If again, we rewrite $v_i=V_{is}+w_{is}$, and use Leibniz' rule, `E' becomes
\begin{equation*}  
-\frac{\delta_{ir}}{r}\rho_sV_{\theta s}^2-\frac{\delta_{ir}}{r}P_{\theta \theta,s}       ,
\end{equation*}
with $\delta_{ij}$ the Kronecker delta. Similarly, `F' becomes
\begin{equation*}  
\frac{1}{r}\pi_{ir,s}+\frac{\delta_{i\theta}}{r}\pi_{r\theta ,s}      ,
\end{equation*}
for $\pi_{ij,s}=m_s\int v_iv_jf_sd^3v$. Putting this all together gives
\begin{eqnarray}
&&\rho_s\frac{\partial V_{is}}{\partial t}+\nabla^j P_{ij,s}+\nabla^j(\rho_sV_{js}V_{is})-\sigma_sE_i-\sigma_s\epsilon_{ijk}V^j_{s}B^k\nonumber\\
&&-\rho_s(V_{\theta s})^2\frac{\delta_{ir}}{r}-\frac{\delta_{ir}}{r}P_{\theta\theta,s}+\frac{1}{r}\pi_{ir,s}+\frac{\delta_{i\theta}}{r}\pi_{r\theta ,s}=0.\label{eq:gen_eq_motion}
\end{eqnarray}
Taking the $r$-component, in equilibrium ($\partial /\partial t=0$), assuming a 1D configuration with only radial dependence ($\partial /\partial \theta=\partial /\partial z=0)$, letting $f_s$ be an even function of $v_r$ ($V_{rs}=P_{r\theta}=P_{zr}=0$), and noticing that $\pi_{rr,s}=\rho_sV_{rs}^2+P_{rr,s}=P_{rr,s}$ gives
\begin{equation*}
\frac{\partial P_{rr,s}}{\partial r}+\frac{1}{r}(P_{rr,s}-P_{\theta\theta,s})=\sigma_s(\boldsymbol{E}+\boldsymbol{V}_s\times\boldsymbol{B})_r+\rho_s\frac{V_{\theta s}^2}{r}.
\end{equation*}
We now consider the general expression for the $r$ component of the divergence of a rank-2 tensor in cylindrical coordinates \citep{NRL}
\begin{equation}
(\nabla\cdot\boldsymbol{T})_r= \frac{1}{r}\frac{\partial}{\partial r}\left(rT_{rr}\right)+\frac{1}{r}\frac{\partial T_{\theta r}}{\partial \theta}+\frac{\partial T_{z r}}{\partial z}-\frac{T_{\theta\theta}}{r}.                     \label{eq:divergence}
\end{equation}
Since the $P_{r\theta}$ and $P_{zr}$ terms of the pressure tensor are zero, this becomes
\begin{equation}
(\nabla\cdot\boldsymbol{P})_r= \frac{1}{r}\frac{\partial}{\partial r}\left(rP_{rr}\right)-\frac{P_{\theta\theta}}{r},  \label{eq:pressurediv}            
 \end{equation}
and so force balance for species $s$ is maintained - in equilibrium ($\partial /\partial t=0$), assuming a 1D configuration with only radial dependence ($\partial /\partial \theta=\partial /\partial z=0)$, and letting $f_s$ be an even function of the radial velocity $v_r$  -  according to
\begin{equation}
(\nabla\cdot\boldsymbol{P}_s)_r=(\boldsymbol{j}_s\times\boldsymbol{B})_r+\sigma_s E_r+\frac{\rho_s}{r}V_{\theta s}^2.\label{eq:species}
\end{equation}
Equation (\ref{eq:species}) can be summed over species to give
\begin{equation}
(\nabla\cdot\boldsymbol{P})_r+\boldsymbol{\mathcal{F}}_{\text{c}}=(\boldsymbol{j}\times\boldsymbol{B})_r+\sigma \boldsymbol{E}, \label{eq:CMHDE_alt}
\end{equation}
where
\begin{equation*}
\boldsymbol{\mathcal{F}}_{\text{c}}=\sum_s\boldsymbol{\mathcal{F}}_{\text{c},s}=-\frac{1}{r}\left(\rho_iV_{\theta i}^2+\rho_eV_{\theta e}^2\right)\hat{\boldsymbol{e}}_r
\end{equation*}
is the force density associated with the rotating bulk flows of the ions and electrons, and is in fact a centripetal force. Equation (\ref{eq:CMHDE_alt}) is a cylindrical analogue of the force balance equation in Cartesian geometry (e.g. see \cite{Mynick-1979a}). However, in the cylindrical case there are extra terms due to centripetal forces. Note that in a non-inertial frame that is co-moving with the respective species bulk flows, the species $s$ will also feel a fictitious force equal to $-\boldsymbol{\mathcal{F}}_{\text{c},s}$ (as well as any other forces), and this is known as the centrifugal force. 

From the point of view of a particular magnetic field $\boldsymbol{B}$ (which is the point we take by specifying a particular macroscopic equilibrium), we see that equilibrium is maintained by a combination of density/pressure variations as in the case of Cartesian geometry, but with additional contributions from centripetal forces and as an inevitable result of the resultant charge separation, an electric field. This effect is represented in Figure \ref{fig:E_inout}, with Figure \ref{fig:E_in} depicting the case for $E_r<0$, such that $-\boldsymbol{\mathcal{F}}_{ci}>-\boldsymbol{\mathcal{F}}_{ce}$. Whereas Figure \ref{fig:E_out} depicts the case for $E_r>0$, such that $-\boldsymbol{\mathcal{F}}_{ce}>-\boldsymbol{\mathcal{F}}_{ci}$. This demonstrates that `sourcing' an exactly force-free macroscopic equilibrium with an equilibrium DF in a 1D cylindrical geometry is inherently a more difficult task than in the Cartesian case. The presence of `extra' centripetal forces, and almost inevitably forces associated with charge separation, raises the question of whether exactly force-free ($\boldsymbol{j}\times\boldsymbol{B}=\boldsymbol{0}$) equilibria are possible at all in this geometry.

\begin{figure}
    \centering
    \begin{subfigure}[b]{0.7\textwidth}
        \includegraphics[width=\textwidth]{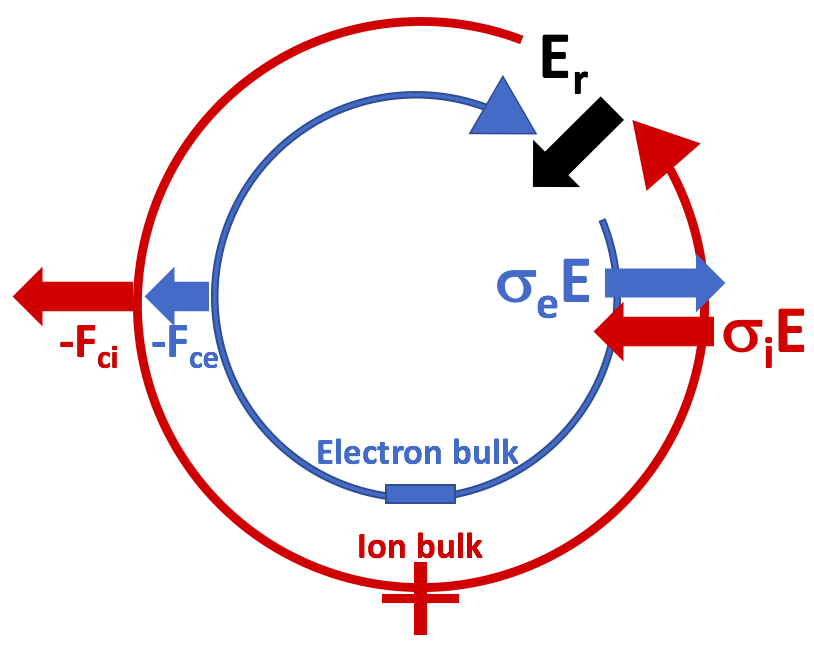}
        \caption{\small Force balance with $E_r<0$ and $ -\boldsymbol{\mathcal{F}}_{ci}>-\boldsymbol{\mathcal{F}}_{ce}$}
        \label{fig:E_in}
    \end{subfigure}
       \begin{subfigure}[b]{0.7\textwidth}
        \includegraphics[width=\textwidth]{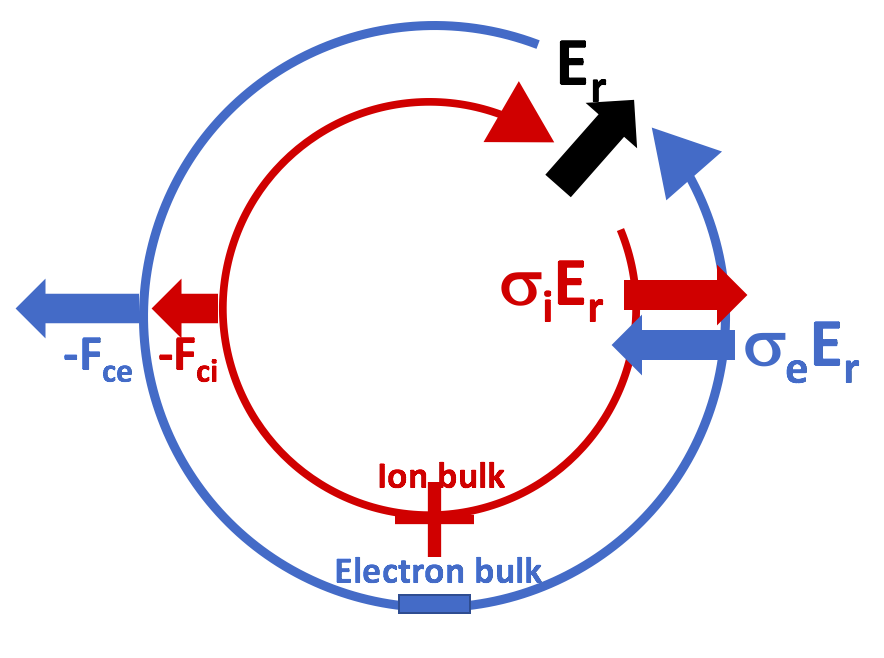}
        \caption{\small Force balance with $E_r>0$ and $-\boldsymbol{\mathcal{F}}_{ce}>-\boldsymbol{\mathcal{F}}_{ci}$}
        \label{fig:E_out}
    \end{subfigure}
    \caption{\small  A schematic representation of how, in force balance, the electric field, $E_r$ exists in order to balance the `charge separation' effect caused by the forces associated with the ion and electron rotational bulk flows, $\boldsymbol{\mathcal{F}}_{ci}$   and $\boldsymbol{\mathcal{F}}_{ce}$ respectively. Figure \ref{fig:E_in} depicts the case for $E_r<0$, such that $-\boldsymbol{\mathcal{F}}_{ci}>-\boldsymbol{\mathcal{F}}_{ce}$, whilst Figure \ref{fig:E_out} depicts the case for $E_r>0$, such that $-\boldsymbol{\mathcal{F}}_{ce}>-\boldsymbol{\mathcal{F}}_{ci}$. }\label{fig:E_inout}
\end{figure}

Before proceeding, we comment that given certain macroscopic constraints on the electromagnetic fields or fluid quantities - such as the force-free condition, or a specific given magnetic field (for example) - it is not \emph{a priori} known how to calculate a self-consistent Vlasov equilibrium, or if one even exists within the framework of the assumptions made. Hence one has to proceed more or less on a case by case basis, with the intention of achieving consistency with the required macroscopic conditions, upon taking moments of the DF.

\subsection{The Gold-Hoyle (GH) magnetic field}
The GH magnetic field \citep{Gold-1960} is a 1D ($\partial/\partial\theta=\partial/\partial z=0$), nonlinear force-free ($\nabla\times\boldsymbol{B}=\alpha(r)\boldsymbol{B}$) and uniformly twisted flux-tube model, with 
\begin{eqnarray}
\boldsymbol{A}_{GH}(\tilde{r})&=&\frac{B_0}{2\tau}\left(0,\frac{1}{\tilde{r}}\ln\left(1+\tilde{r}^2\right),-\ln\left(1+\tilde{r}^2\right)\right),\nonumber\\
\boldsymbol{B}_{GH}(\tilde{r})&=&B_0\left(0,\frac{\tilde{r}}{1+\tilde{r}^2},\frac{1}{1+\tilde{r}^2}\right),\nonumber\\
\boldsymbol{j}_{GH}(\tilde{r})&=&2\frac{\tau B_0}{\mu_0}\left(0,\frac{\tilde{r}}{(1+\tilde{r}^2)^2},\frac{1}{(1+\tilde{r}^2)^2}\right),\nonumber\\
\boldsymbol{j}_{GH}(\boldsymbol{A},\tilde{r})&=&2\frac{\tau B_0}{\mu_0}\left(0,\tilde{r}e^{-\frac{4\tau}{B_0}\tilde{r}A_{\theta}},e^{\frac{4\tau}{B_0}A_z}\right),\label{eq:GHjofA}
\end{eqnarray}
The constant $\tau$ has units of inverse length, and we use $1/\tau$ to represent the characteristic length scale of the system ($\tilde{r}=\tau r$). The parameter $B_0$ gives the magnitude of the magnetic field at $\tilde{r}=0$.  Note that the representation of $\boldsymbol{j}_{GH}(\boldsymbol{A})$ chosen in Equation (\ref{eq:GHjofA}) is representative and non-unique. In fact there are other possible representations, that include `mixtures' of $A_{\theta}$ and $A_z$ in each component of the current density. 

Furthermore, $\tau$ is a direct measure of the `twist' of the embedded flux tube (see \cite{Birn-2007}), with the number of turns per unit length (in $z$) along a field line given by $\tau/(2\pi)$ \citep{Gold-1960}. A diagram representing the qualitative interior structure of such a flux tube is given in Figure \ref{fig:interior}, and reproduced from \citet{Russell-1979} (their magnetic field was in fact not quite uniformly twisted, but close enough that the diagram still serves a purpose). The most important feature to note is how the $B_{z}$ component of the field dominates at small radii, whereas the $B_{\theta}$ component dominates for larger radii. This characteristic ensures that you travel the same distance in $z$, for each $2\pi$ revolution, regardless of how far from the central axis you are ($d\theta/dz=\text{const.}$). The force-free parameter for the magnetic field is 
\[
\alpha (r)=\frac{\nabla\times\boldsymbol{B}\cdot\boldsymbol{B}}{|\boldsymbol{B}|^{2}}=\frac{2\tau}{1+\tilde{r}^2}.
\]
Should one wish to consider the GH field in an MHD context ($\nabla p =\boldsymbol{j}\times\boldsymbol{B}=\boldsymbol{0}$) then the scalar pressure $p=\text{const.}$. This is seen by considering the 1D force-balance equation \citep{Freidbergbook},  
\begin{equation}
\boldsymbol{j}\times\boldsymbol{B}=\nabla p \iff \frac{d}{dr}\left(p+\frac{B_{\theta}^2}{2\mu_0}+\frac{B_{z}^2}{2\mu_0}\right)+\frac{B_{\theta}^2}{\mu_0r}=0,\nonumber
\end{equation}
for the GH field.
\begin{figure}
    \centering
        \includegraphics[width=0.7\textwidth]{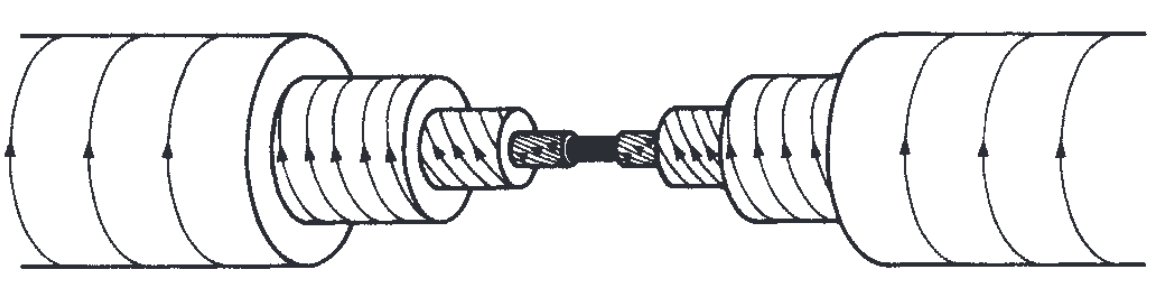}
    \caption{{\small The interior structure of a flux tube, from \citet{Russell-1979}, and similar to the GH model. {\bf Image Copyright:} Nature Publishing Group. Reprinted by permission from Macmillan Publishers Ltd:      \href{http://www.nature.com}{\emph{Nature}} {\bf 279} (June 1979), pp. 616-618., copyright (1979). }}\label{fig:interior}
\end{figure}

\subsection{Methods for calculating an equilibrium DF}
In \citet{Channell-1976, Harrison-2009PRL} for example, a method used to calculate a DF, given a prescribed 1D magnetic field was Inverse Fourier Transforms (IFT). This method was also discussed in Section \ref{sec:ftransform}. A DF of the form 
\begin{equation}
f_s\propto e^{-\beta_sH_s}g_s(p_{xs},p_{ys}),
\end{equation}
was used, with $H_{s}$, $p_{xs}$ and $p_{ys}$ the conserved particle Hamiltonian and canonical momenta in the $x$ and $y$ directions, and $g_s$ an unknown function, to be determined. Since our problem is one of a 1D equilibrium with variation in the radial direction, the three constants of motion are the Hamiltonian, and the canonical momenta in the $\theta$ and $z$ directions:
\begin{eqnarray}
&&H_s=\frac{m_s}{2}\left(v_{r}^2+v_{\theta}^2+v_{z}^2\right)+q_s\phi,\nonumber\\
&&p_{\theta s}=r\left(m_sv_{\theta}+q_sA_{\theta}\right),\hspace{3mm} p_{zs}=m_sv_{z}+q_sA_z.
\end{eqnarray}
One can try to calculate an equilibrium distribution for the GH force-free flux tube without a background field by a similar method, assuming a DF of the form
\begin{equation}
f_s\propto e^{-\beta_sH_s}g_s(p_{\theta s},p_{zs}). \label{eq:fansatz}
\end{equation}
By exploiting the convolution in the definition of the current density, 
\begin{eqnarray}
\boldsymbol{j}(\boldsymbol{A},r)&=&\sum_sq_s\int \,\boldsymbol{v}\,f_s(H_s,p_{\theta s},p_{zs})\,d^3v,\nonumber\\
&=&r\sum_s \frac{q_s}{m_s^4}\int \,(\boldsymbol{\mathfrak{p}}_s-q_s\boldsymbol{A})\, f_s(H_s,r\mathfrak{p}_{\theta s},\mathfrak{p}_{z s})\,d^3\mathfrak{p}_s,\nonumber
\end{eqnarray}
Amp\`{e}re's law can be solved formally by IFT (\emph{cf.} \cite{Harrison-2009PRL} and Section \ref{sec:ftransform}), or informally by `inspection' (\emph{cf.} \cite{Neukirch-2009}), with the quantity $\boldsymbol{\mathfrak{p}}_s$ defined by
\begin{equation}
\mathfrak{p}_{rs}=p_{rs},\hspace{3mm}\mathfrak{p}_{\theta s}=\frac{p_{\theta s}}{r},\hspace{3mm}\mathfrak{p}_{zs}=p_{zs}.\nonumber
\end{equation}
Notice how when written in this integral form, $\boldsymbol{j}$ is not only a function of $\boldsymbol{A}$, but - in contrast with the Cartesian case - also of the relevant spatial co-ordinate, $r$.

\subsubsection{Problems with equilibrium DFs for the GH field}
We shall now reproduce the calculations, representatively, for the $j_{\theta}$ case. These calculations are representative in that the choice of expression for the current density as a function of the vector potential is non-unique, as indicated previously. However, this calculation should demonstrate the inherent obstacle in calculating a Vlasov equilibrium DF for the GH field. 

The definition of the current density, along with the ansatz of Equation (\ref{eq:fansatz}) gives
\begin{eqnarray}
j_{\theta}=r\sum_s\frac{q_s}{m_s^4}\frac{n_{0s}}{(\sqrt{2\pi}v_{\text{th},s})^3}e^{-\beta_sq_s\phi}\int (\mathfrak{p}_{\theta s}-q_sA_{\theta}) e^{-(\boldsymbol{\mathfrak{p}}_s-q_s\boldsymbol{A})^2/(2m_s^2v_{\text{th},s}^2)}g_s(r\mathfrak{p}_{\theta s},\mathfrak{p}_{z s})d^3\mathfrak{p}_s.\nonumber
\end{eqnarray}
If we now take a representative (i.e. one possible) expression for the current density, chosen as a more `general' form than that in Equation (\ref{eq:GHjofA}),
\[j_{\theta}=c_1\frac{\tau^2 B_0 r}{\mu_0}\exp \left(\frac{c_2\tau^2 rA_{\theta }}{B_0}+\frac{c_3\tau A_{z}}{B_0} \right),\] 
for $c_{1}, c_{2}$ and $c_{3}$ constants, and re-write $\mathfrak{p}_{\theta s}=p_{\theta s}/r$, then we obtain    
\begin{eqnarray}
&&c_1\frac{\tau^2 B_0 r}{\mu_0}\exp\left(\frac{c_2\tau^2 rA_{\theta }}{B_0}+\frac{c_3\tau A_{z}}{B_0} \right)=\sum_s\frac{q_s}{m_s^3}\frac{n_{0s}}{(\sqrt{2\pi}v_{\text{th},s})^2}e^{-\beta_sq_s\phi}\times\nonumber\\
&&\int (p_{\theta s}-q_srA_{\theta})e^{-(p_{\theta s}-q_srA_{\theta})^2/(2m_s^2v_{\text{th},s}^2r^2)-(p_{z s}-q_sA_{z})^2/(2m_s^2v_{\text{th},s}^2)}g_s(p_{\theta s},p_{z s})dp_{\theta s}dp_{z s}.\nonumber
\end{eqnarray}

In the case of zero scalar potential, the result of the calculation is to give a $g_s$ function (and hence a DF) that is not a solution of the Vlasov equation as it is not a function of the constants of motion only. In essence, an additional ``$\exp (-r^2)$'' factor would be required in the DF to counter ``$\exp (+r^2)$'' terms that manifest by completing the square in the integration. That is to say that the `solution' would be of the form 
\[
g_{s}(p_{\theta s},p_{zs})=g_0   \exp\left(-\frac{\omega_s^2}{2\tau^2v_{\text{th},s}^2}\delta_s^2\tau^2r^2\right)   \exp\left(  \frac{\omega_s}{\tau v_{\text{th},s}}\frac{\tau^2 p_{\theta s}}{q_sB_0} + \frac{V}{v_{\text{th},s}}\frac{\tau p_{zs}}{q_sB_0} \right),\nonumber\\
\]
and hence the DF can be written as
\begin{equation}
f_s\propto g_0\exp\left(-\frac{\omega_s^2}{2\tau^2v_{\text{th},s}^2}\delta_s^2\tau^2r^2\right)e^{-\beta_sH_s}\exp\left(  \frac{\omega_s}{\tau v_{\text{th},s}}\frac{\tau^2 p_{\theta s}}{q_sB_0} + \frac{V}{v_{\text{th},s}}\frac{\tau p_{zs}}{q_sB_0} \right)\label{eq:GHDF},
\end{equation}
for some $g_0, \omega_s$ and $V$ related to $c_1,c_2$ and $c_3$ respectively. The ratio of the thermal Larmor radius, $r_L=m_sv_{\text{th},s}/(e|B|)$ (for $e=|q_s|$) to the macroscopic length scale of the system $L(=1/\tau)$, is given by
\begin{equation*}
\delta_s(r)=\frac{r_L}{L}=\frac{m_sv_{\text{th},s}\tau}{eB(r)},
\end{equation*}    
typically known as the `magnetisation parameter' \citep{Fitzpatrickbook} (see Table \ref{tab:param} for a concise list of the micro and macroscopic parameters of the equilibrium). Note that in our system, the magnitude of the magnetic field and hence $\delta_{s}$ itself is spatially variable. For the purposes of the calculations in this chapter however, we set
\begin{equation*}
\frac{m_sv_{\text{th},s}\tau}{eB_0}=\delta_s={\rm const.}
\end{equation*}
as a characteristic value.

The DF in Equation (\ref{eq:GHDF}) is not a solution of the Vlasov equation, but would approximate one in the limit 
\[\frac{\omega_s}{\tau v_{\text{th},s}}\delta_s=\frac{\omega_s }{q_s B_0/m_s}\to 0,\]
i.e. the vanishing ratio of the bulk angular frequency to the gyrofrequency of the individual particles (\emph{cf.} \cite{Vinogradov-2016} and more on this later). It is now apparent that the physical cause for the extra ``$\exp (+r^2)$'' term here would appear to be the forces associated with the rotational bulk flow, since the term is non-negligible when $\omega_s$ is of a sufficient magnitude.
 
 \begin{table*}
 \centering
  \caption{\small The fundamental parameters of the equilibrium.\\ The $s$ subscript refers to particles of species $s$.}
  \begin{center}
  \begin{tabular}{cc}
  \hline
Macroscopic& \\
parameter&Meaning\\
\hline
$B_0$&Characteristic magnetic field strength\\
$\tau$&Measure of the twist of flux tube\\
$k$&Strength of the background field\\
$\gamma_1\ne 0,1$, $0<\gamma_2 <1$&Gauge for scalar potential\\
${U}_{zs}, {V}_{zs}$&Bulk rectilinear flows\\
$ \omega_{s}$&Bulk angular frequency\\
\hline
Microscopic& \\
parameter&Meaning\\
\hline
$m_s$& Mass of particle\\
$q_s$, $e$& Charge, magnitude of charge\\
$\beta_s=1/(k_BT_s)$&  Thermal beta\\
$v_{\text{th},s}$&  Thermal velocity\\
$\delta_s(r), \delta_s$&   Magnetisation parameters\\
$n_{0s}$&Normalisation of particle number\\

\hline
\label{tab:param}
\end{tabular}
\end{center}
\end{table*}

If one assumes a non-zero scalar potential, then the above considerations would seem to imply that 
\[
-\beta_iq_i\phi=-\beta_eq_e\phi=-\frac{\omega_s^2}{2\tau^2v_{\text{th},s}^2}\delta_s^2\tau^2r^2,
\]
for there to be an exact Vlasov solution. This equation cannot be satisfied. The physical cause seems to be that, in the case of force-free fields, one would require a `different' electrostatic potential to balance the forces for the ions and electrons, which is of course nonsensical. Thus, our investigation seems to suggest that it is not possible to calculate a DF of the form of Equation (\ref{eq:fansatz}) for the exact GH field.

\subsection{GH flux tube plus background field (GH+B)}\label{subsec:magfield}
To make progress, we introduce a background field in the negative $z$ direction. The mathematical motivation for this change is to balance the `$\exp (r^2)$ problem'. Physically, it seems that the background field introduces an extra term (whose sign depends on species) into the force-balance, to allow for both the ion and electrons to be in force balance simultaneously, given one unique expression for the scalar potential.

The vector potential, magnetic field and current density used are as follows:
\begin{eqnarray}
\boldsymbol{A}_{GH+B}(\tilde{r})&=&\frac{B_0}{2\tau}\left(0,\frac{1}{\tilde{r}}\ln\left(1+\tilde{r}^2\right)-2k\tilde{r},-\ln\left(1+\tilde{r}^2\right)\right),\nonumber\\
&=&\boldsymbol{A}_{GH}-\left( 0,B_0k\tau^{-1}\tilde{r},0\right).\label{eq:vecfield}\\
\boldsymbol{B}_{GH+B}(\tilde{r})&=&B_0\left(0,\frac{\tilde{r}}{1+\tilde{r}^2},\frac{1}{1+\tilde{r}^2}-2k\right),\nonumber\\
&=&\boldsymbol{B}_{GH}-(0,0,2kB_0)\label{eq:magfield}.\\
\boldsymbol{j}_{GH+B}(\tilde{r})&=&2\frac{\tau B_0}{\mu_0}\left(0,\frac{\tilde{r}}{(1+\tilde{r}^2)^2},\frac{1}{(1+\tilde{r}^2)^2}\right),\nonumber\\
&=&\boldsymbol{j}_{GH}.\label{eq:current}
\end{eqnarray}
The dimensionless constant $k>0$ controls the strength of the background field in the $z$ direction, and as a result there are now two different interpretations to be made. We could either consider the system as a GH flux tube of uniform twist embedded in an untwisted uniform background field, or consider the whole GH+B magnetic field as a non-uniformly twisted flux tube. 

In the first interpretation, $\tau$ is (as aforementioned) a direct measure of the `twist' of the embedded flux tube (see \cite{Birn-2007}), with the number of turns per unit length (in $z$) along a field line given by $\tau/(2\pi)$ \citep{Gold-1960}. In the second interpretation, we see that the system is not uniformly twisted, with the $z$ distance traversed when following a field line (e.g. \cite{Marshbook}), given by 
\begin{equation*}
\int \frac{rB_z}{B_{\theta}}d\theta=\frac{1}{\tau}\left(1-2k(1+\tilde{r}^2)\right)\int d\theta.
\end{equation*}
The fact that this depends on $r$ demonstrates that the system as a whole has non-uniform twist. The number of turns per unit length in $z$ of the GH+B field: the `twist' is given by
\begin{equation*}
\left( \int_{\theta=0}^{\theta=2\pi} \frac{rB_z}{B_{\theta}}d\theta   \right)^{-1}=\frac{\tau}{2\pi}\left(1-2k(1+\tilde{r}^2)\right)^{-1},
\end{equation*}
and is plotted in Figure \ref{fig:twist} for three values of $k$. Since $k<1/2$ corresponds to the field-reversal regime, we see a mixture of positive and negative twists (Figure \ref{fig:twist1}). However, for $k\ge 1/2$ we see only negative values of the twist (Figures \ref{fig:twist2} and \ref{fig:twist3}), i.e. we travel in the negative $z$ direction as we wind round the GH+B flux tube in the anti-clockwise direction.
\begin{figure}
    \centering
    \begin{subfigure}[b]{0.6\textwidth}
        \includegraphics[width=\textwidth]{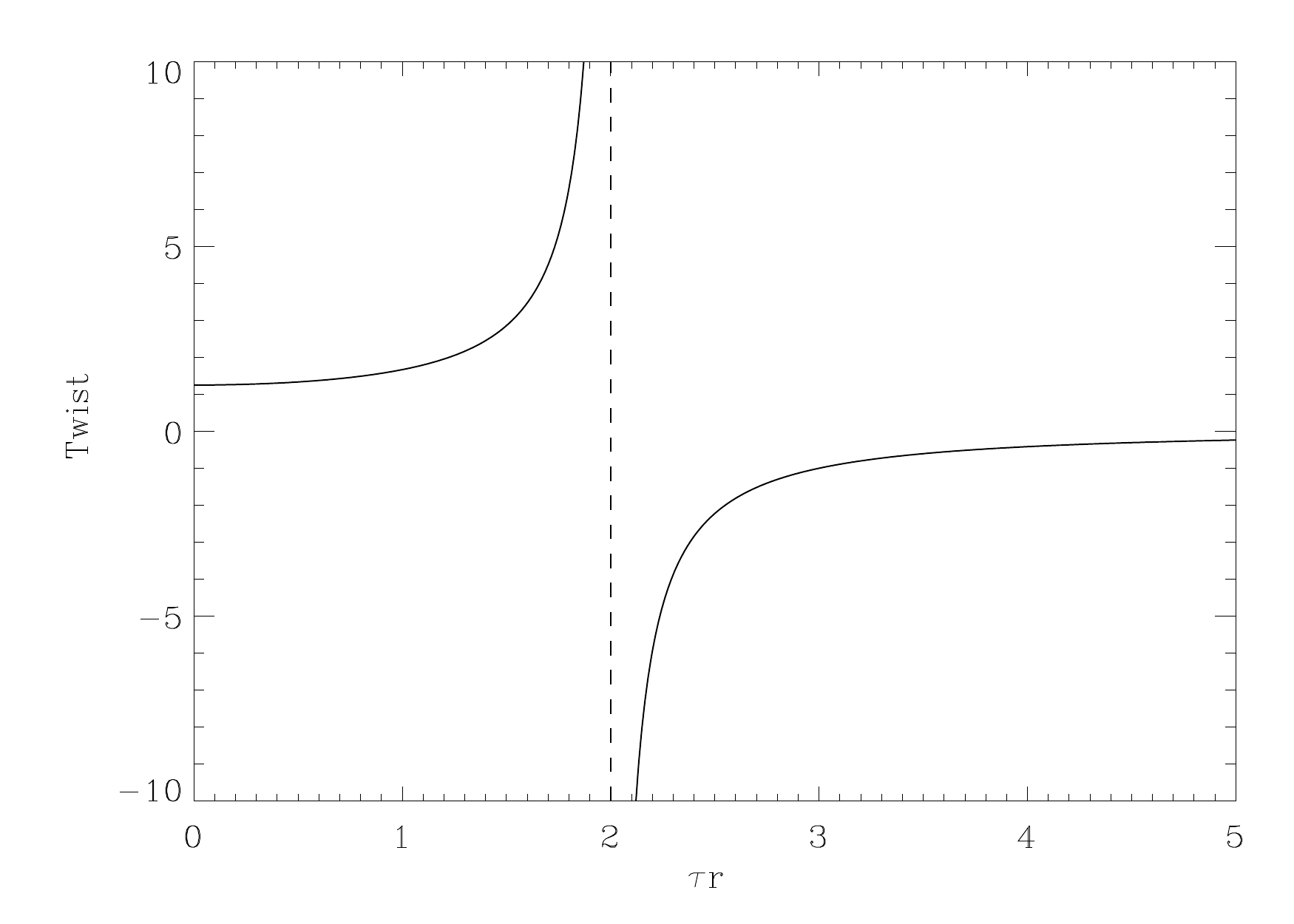}
        \caption{\small $k=0.1$}
        \label{fig:twist1}
    \end{subfigure}
       \begin{subfigure}[b]{0.6\textwidth}
        \includegraphics[width=\textwidth]{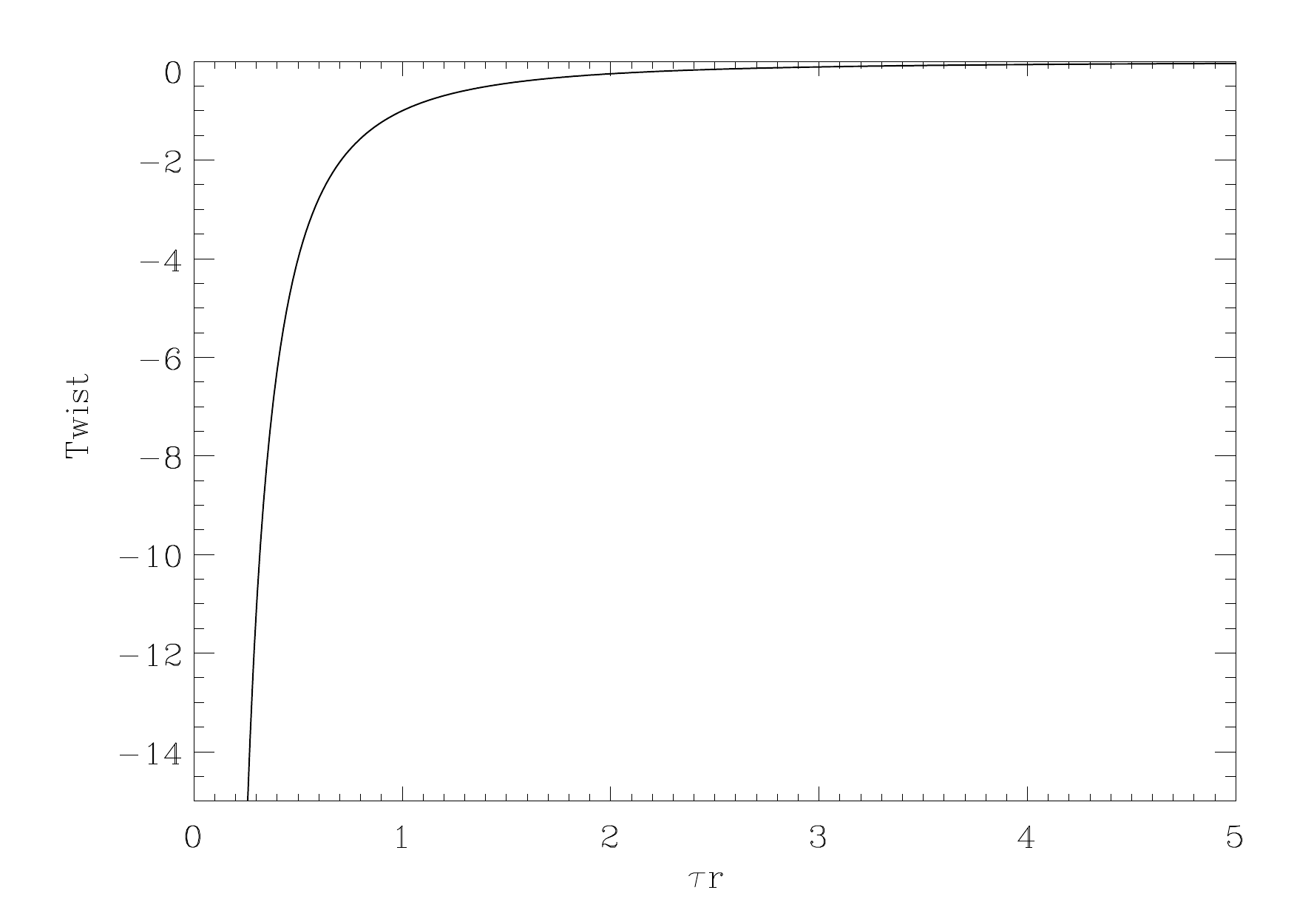}
        \caption{\small $k=0.5$}
        \label{fig:twist2}
    \end{subfigure}
        \begin{subfigure}[b]{0.6\textwidth}
        \includegraphics[width=\textwidth]{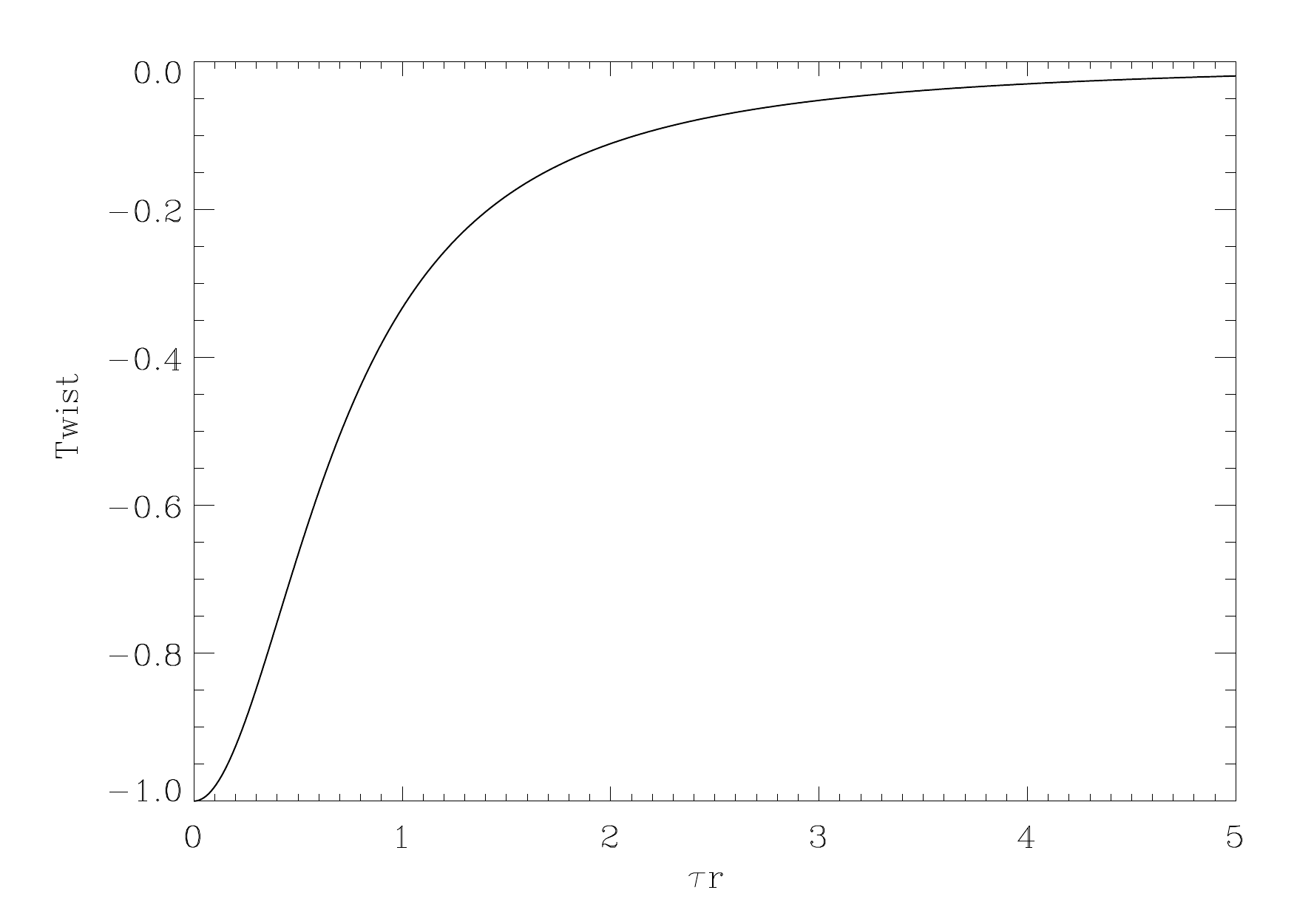}
        \caption{\small $k=1$}
        \label{fig:twist3}
    \end{subfigure}
    \caption{\small The twist (normalised by $\tau/(2\pi)$) of the GH+B field for three values of $k$. \ref{fig:twist1} shows the twist for $k<1/2$, and as such there are both negative and positive twists, due to the field reversal. \ref{fig:twist2} and \ref{fig:twist3} both show negative twist, since there is no magnetic field reversal. }\label{fig:twist}
\end{figure}
The magnetic field is plotted in Figures \ref{fig:0.3}-\ref{fig:0.5} for two values of $k$.  The $k=0.3$ case contains a reversal of the $\tilde{B}_z$ field direction and as such is akin to a Reversed Field Pinch (e.g. see \cite{Escande-2015} for a laboratory interpretation): this configuration may be of use in the study of astrophysical jets, see \citet{Li-2006} for example. The value $k=1/2$ gives zero $\tilde{B}_z$ at $\tilde{r}=0$, and as such is the value that distinguishes the two different classes of field configuration, namely unidirectional ($k\ge 1/2$) or including field reversal ($k<1/2$). The value of $\tilde{r}$ for which the $\tilde{B}_z$ field reverses is plotted in Figure \ref{fig:rev}. The magnitude of the GH+B magnetic field is plotted in Figure \ref{fig:bmag} for three values of $k$. For all values of $k$, $|\tilde{\boldsymbol{B}}|\to 2k$ for large $\tilde{r}$, i.e. to a potential field. 
\begin{figure}
    \centering
    \begin{subfigure}[b]{0.45\textwidth}
        \includegraphics[width=\textwidth]{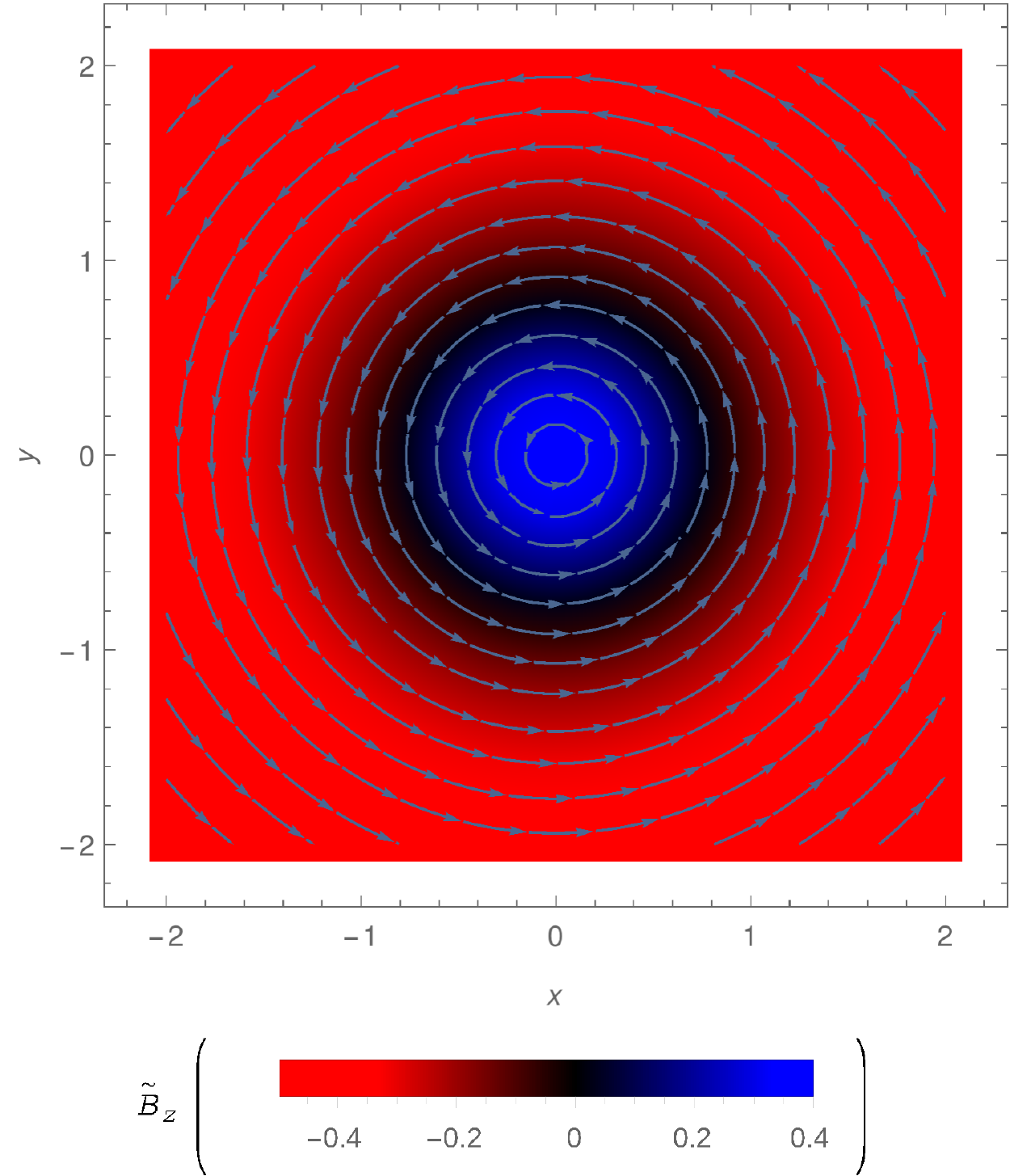}
        \caption{\small $\boldsymbol{B}$ for $k=0.3$}
        \label{fig:0.3}
    \end{subfigure}
       \begin{subfigure}[b]{0.45\textwidth}
        \includegraphics[width=\textwidth]{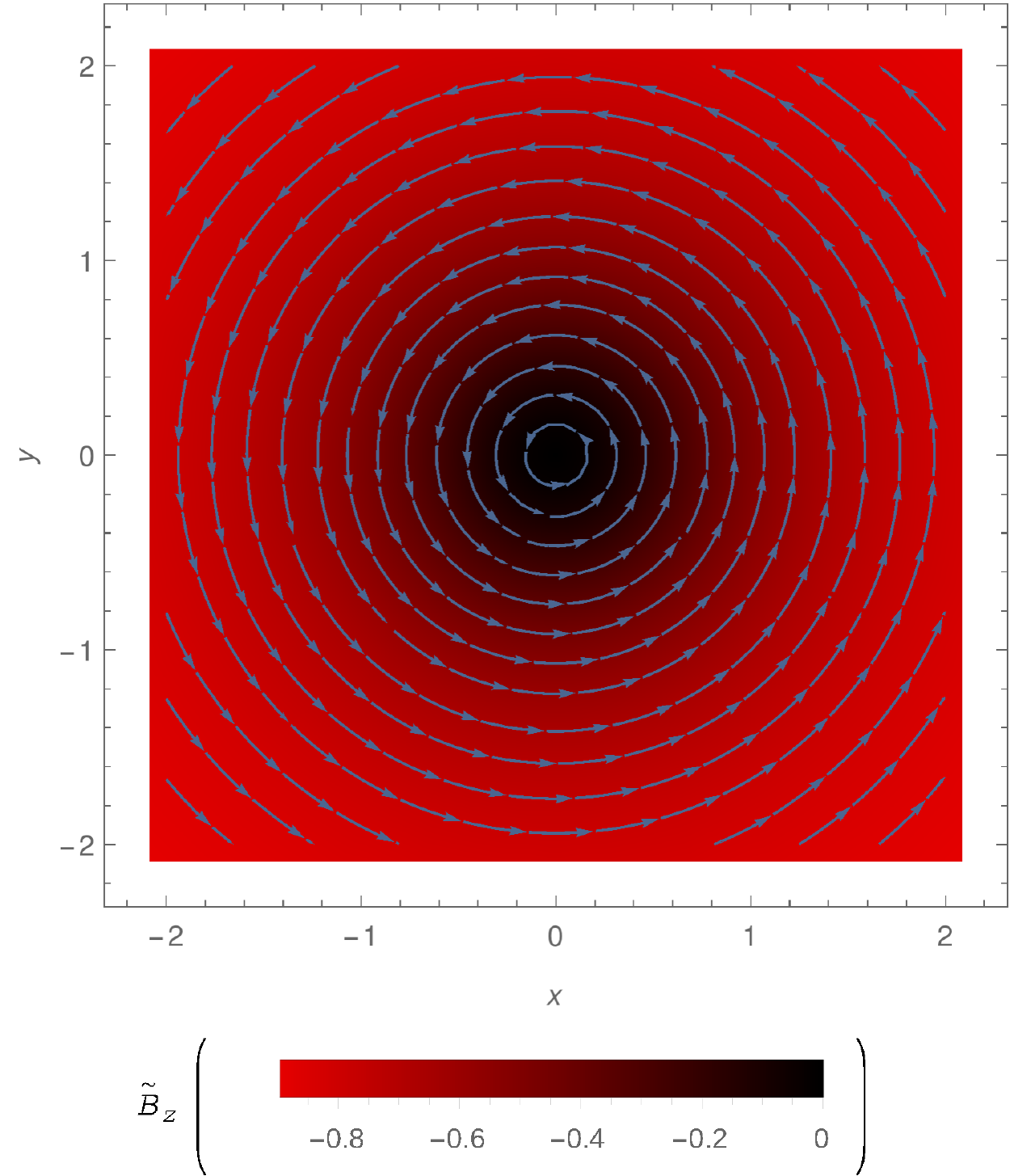}
        \caption{\small $\boldsymbol{B}$ for $k=0.5$}
        \label{fig:0.5}
    \end{subfigure}
        \begin{subfigure}[b]{0.7\textwidth}
        \includegraphics[width=\textwidth]{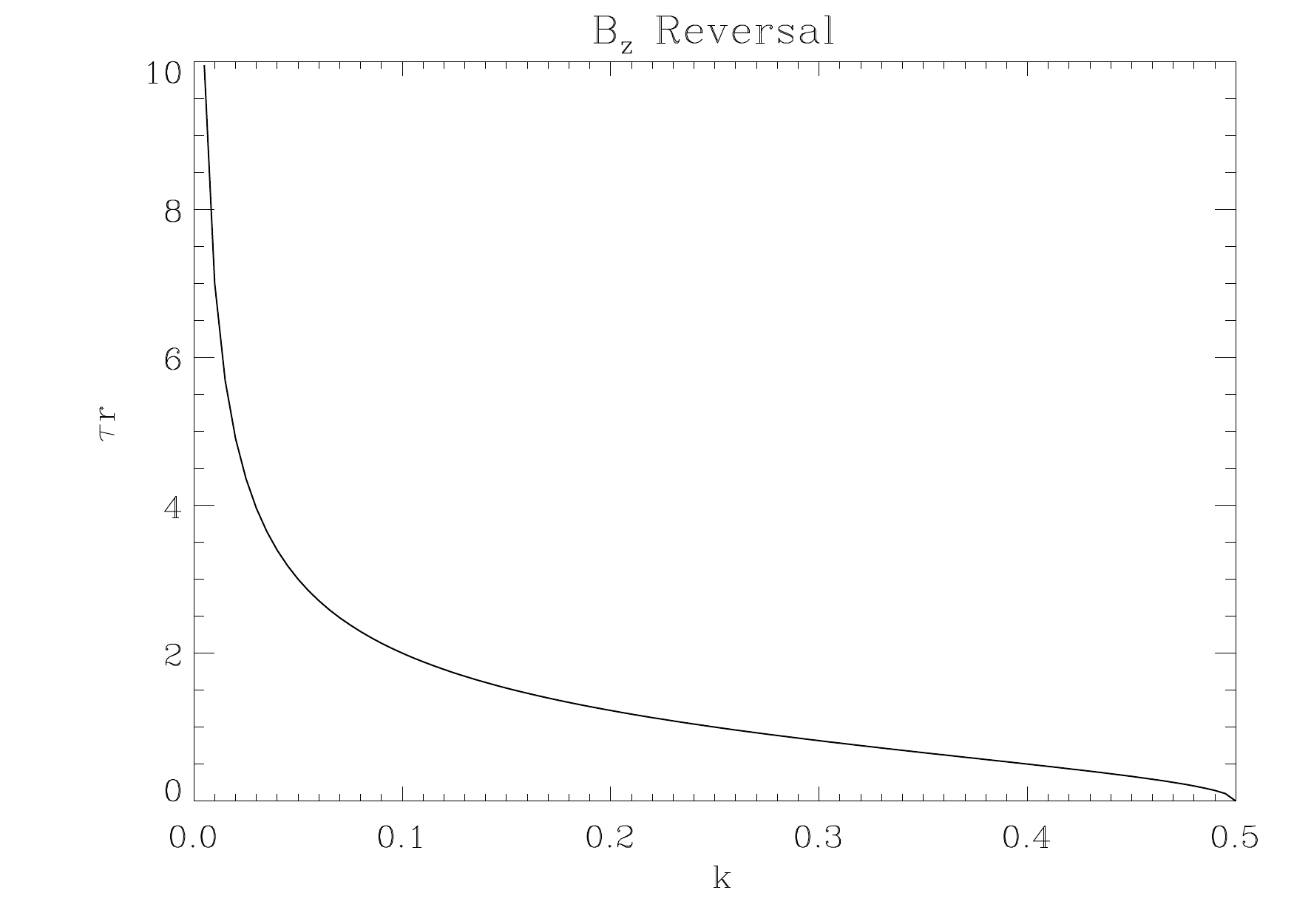}
        \caption{\small Radius of $B_z$ reversal, given $0<k<1/2$}
        \label{fig:rev}
    \end{subfigure}
    \caption{\small \ref{fig:0.3} and \ref{fig:0.5} show the GH+B magnetic field in the $xy$ plane, for two values of $k$. The curved arrows indicate the direction of the $\tilde{B}_\theta$ components, whilst the blue-black-red shading denotes the magnitude and direction of the $\tilde{B}_z$ component. The $k=0.3$ case contains a reversal of the $\tilde{B}_z$ field direction and as such is a Reversed Field Pinch whilst $k=0.5$ gives zero $\tilde{B}_z$ at $\tilde{r}=0$. \ref{fig:rev} shows the radius at which $\tilde{B}_z$ changes its direction, for $0<k<1/2$. $\tilde{B}_z$ does not reverse for $k\ge 1/2$. }\label{fig:magfield}
\end{figure}
We also note here that flux tubes embedded in an axially directed background field have recently been observed during reconnection events in the Earth's magnetotail, by the Cluster spacecraft (e.g. \cite{Borg-2012}), and that recent numerical modelling of `magnetohydrodynamic (MHD) avalanches' in the low-beta solar corona has used multiple flux ropes embedded in a uniform background magnetic field \citep{Hood-2016}. The magnetic field model used \citep{Hood-2009} is similar to the model in this chapter, as it is force-free and 1D. 

The primary task of this chapter is to calculate self-consistent collisionless equilibrium DFs for the GH+B field. This problem essentially reduces to solving Amp\`{e}re's Law such that Equation (\ref{eq:Vlasov_cyl}) is satisfied. We assume nothing about the electric field however, and in fact use that degree of freedom to solve Amp\`{e}re's Law. The resultant form of the scalar potential is then substituted into Poisson's equation, to establish the final relationships between the microscopic and macroscopic parameters of the equilibrium.

\section{The equilibrium DF}\label{sec:thedf}
Although the IFT method did not yield a self-consistent equilibrium DF for the GH field without a background field, the outcome of the calculation can still be used as an indication of possible forms for the DF for the GH+B field. Using trial and error we arrived at the DF
\begin{eqnarray}
f_s=\frac{n_{0s}}{(\sqrt{2\pi}v_{\text{th},s})^3}\left[e^{-\left(\tilde{H}_s-\tilde{\omega}_s\tilde{p}_{\theta s}-\tilde{{U}}_{zs}\tilde{p}_{zs}\right)}+{C}_se^{-\left(\tilde{H}_s-\tilde{{V}}_{zs}\tilde{p}_{zs}\right)}\right],\label{eq:ansatz}
\end{eqnarray}
which is a superposition of two terms that are consistent macroscopically with a `Rigid-Rotor' \citep{Davidsonbook}. A Rigid-Rotor is microscopically described by a DF of the form $F(\mathcal{H}-\omega p_{\theta}-Vp_z)$. Each $F(H-\omega p_{\theta}-Vp_z)$ term corresponds to an average macroscopic motion of rigid rotation with angular frequency $\omega$, and rectilinear motion with velocity $V$ (with $\omega=0$ in the second term of the DF in Equation (\ref{eq:ansatz})). This can be shown in a manner similar to that shown in Section \ref{sec:Harristype}.

The dimensionless constants $\tilde{\omega}_s$, $\tilde{{U}}_{zs}$, $\tilde{{V}}_{zs}$ and ${C}_s$ are yet to be determined, with ${C}_s>0$ for positivity of the distribution (see Table \ref{tab:norm} for a concise list of the dimensionless quantities used in this chapter). 
\begin{table*}
\begin{center}
\centering
  \caption{\small Dimensionless form of some important variables.\\ The $s$ subscript refers to particles of species $s$.}
  \begin{tabular}{cc}
  \hline
       Variable & Dimensionless form\\
       \hline
  Particle Hamiltonian      & $ \tilde{H}_s=\beta_sH_s       $  \\
  Particle angular momentum     &$      \tau p_{\theta s}=m_sv_{\text{th},s}\tilde{p}_{\theta s}  $     \\
  Particle $z$-Momentum          &$          p_{zs}=m_sv_{\text{th},s}\tilde{p}_{zs}   $ \\
 Vector potential  &        $ q_s\boldsymbol{A}=m_sv_{\text{th},s}\tilde{\boldsymbol{A}}_s   $ \\
Scalar Potential  &         $ \tilde{\phi}_s=q_s\beta_s\phi  $\\
Bulk rectilinear flows  &    $      v_{\text{th},s}\tilde{{U}}_{zs}={U}_{zs},\hspace{3mm}v_{\text{th},s}\tilde{{V}}_{zs}={V}_{zs}$    \\
Bulk angular frequency  &     $      \tau v_{\text{th},s}\tilde{\omega}_s=\omega_{s}     $ \\
Particle position (radial)  &       $\tau r=\tilde{r}      $ \\
Particle velocity  &       $\boldsymbol{v}=v_{\text{th},s}\tilde{\boldsymbol{v}}_s       $ \\
\hline
\label{tab:norm}
\end{tabular}
\end{center}
\end{table*}

\subsection{Moments of the DF}\label{sec:moments}
In order to satisfy Maxwell's equations, we shall require the charge and current densities. Hence we will require the zeroth- and first-order moments of the DF in Equation (\ref{eq:ansatz}), and these calculations follow. See Table \ref{tab:norm} for a clarification of all dimensionless quantities denoted by a tilde, $^{\tilde{}}$.
\subsubsection{Zeroth order moments}\label{app:number}
The number density of species $s$ is given by the zeroth moment of the DF;
\begin{eqnarray}
&&n_s=\int f_sd^3v_s=\frac{n_{0s}}{(\sqrt{2\pi})^3}\int e^{-\tilde{H}_s}\bigg(e^{{\tilde{{U}}_{zs}} \tilde{p}_{zs}}e^{{\tilde{\omega}_{s}}\tilde{p}_{\theta s}}+{C}_se^{\tilde{{V}}_{zs} \tilde{p}_{zs}}\bigg)d^3\tilde{v}_s\label{eq:nbasic}\\
&&=\frac{n_{0s}}{(\sqrt{2\pi})^2}e^{-\tilde{\phi}_s}\bigg[e^{\left(\tilde{{U}}_{zs}^2+\tilde{r}^2\tilde{\omega}_{s}^2\right)/2}e^{{\tilde{{U}}_{zs}}\tilde{A}_{zs}}e^{{\tilde{\omega}_{s}}\tilde{r}\tilde{A}_{\theta s}}\int_{-\infty}^{\infty} e^{-\left(\tilde{v}_{zs}-\tilde{{U}}_{zs}\right)^2/2}d\tilde{v}_{zs}\times\nonumber\\
&&\int_{-\infty}^{\infty} e^{-\left(\tilde{v}_{\theta s}-\tilde{\omega}_{s}\tilde{r}\right)^2/2}d\tilde{v}_{\theta s}+{C}_s\sqrt{2\pi}e^{\tilde{{V}}_{zs}^2/2}e^{\tilde{{V}}_{zs}\tilde{A}_{zs}}\int_{-\infty}^{\infty} e^{-\left(\tilde{v}_{zs}-\tilde{{V}}_{zs}\right)^2/2}d\tilde{v}_{zs}\bigg] \nonumber\\
&&=n_{0s}e^{-\tilde{\phi}_s}\left[e^{\left(\tilde{{U}}_{zs}^2+\tilde{r}^2\tilde{\omega}_{s}^2\right)/2}e^{{\tilde{{U}}_{zs}}\tilde{A}_{zs}}e^{{\tilde{\omega}_{s}}\tilde{r}\tilde{A}_{\theta s}}+{C}_se^{\tilde{{V}}_{zs}^2/2}e^{\tilde{{V}}_{zs}\tilde{A}_{zs}}\right] \label{eq:numdensity}
\end{eqnarray}
We take the following sum to calculate the charge density,
\begin{eqnarray}
\sigma=\sum_sq_sn_s=\sum_s n_{0s}q_se^{-\tilde{\phi}_s}\left[e^{\left(\tilde{{U}}_{zs}^2+\tilde{r}^2\tilde{\omega}_{s}^2\right)/2}e^{{\tilde{{U}}_{zs}}\tilde{A}_{zs}}e^{{\tilde{\omega}_{s}}\tilde{r}\tilde{A}_{\theta s}}+{C}_se^{\tilde{{V}}_{zs}^2/2}e^{\tilde{{V}}_{zs}\tilde{A}_{zs}}\right]\label{eq:sigmaapp}
\end{eqnarray}
\begin{figure}
    \centering
    \begin{subfigure}[b]{0.6\textwidth}
        \includegraphics[width=\textwidth]{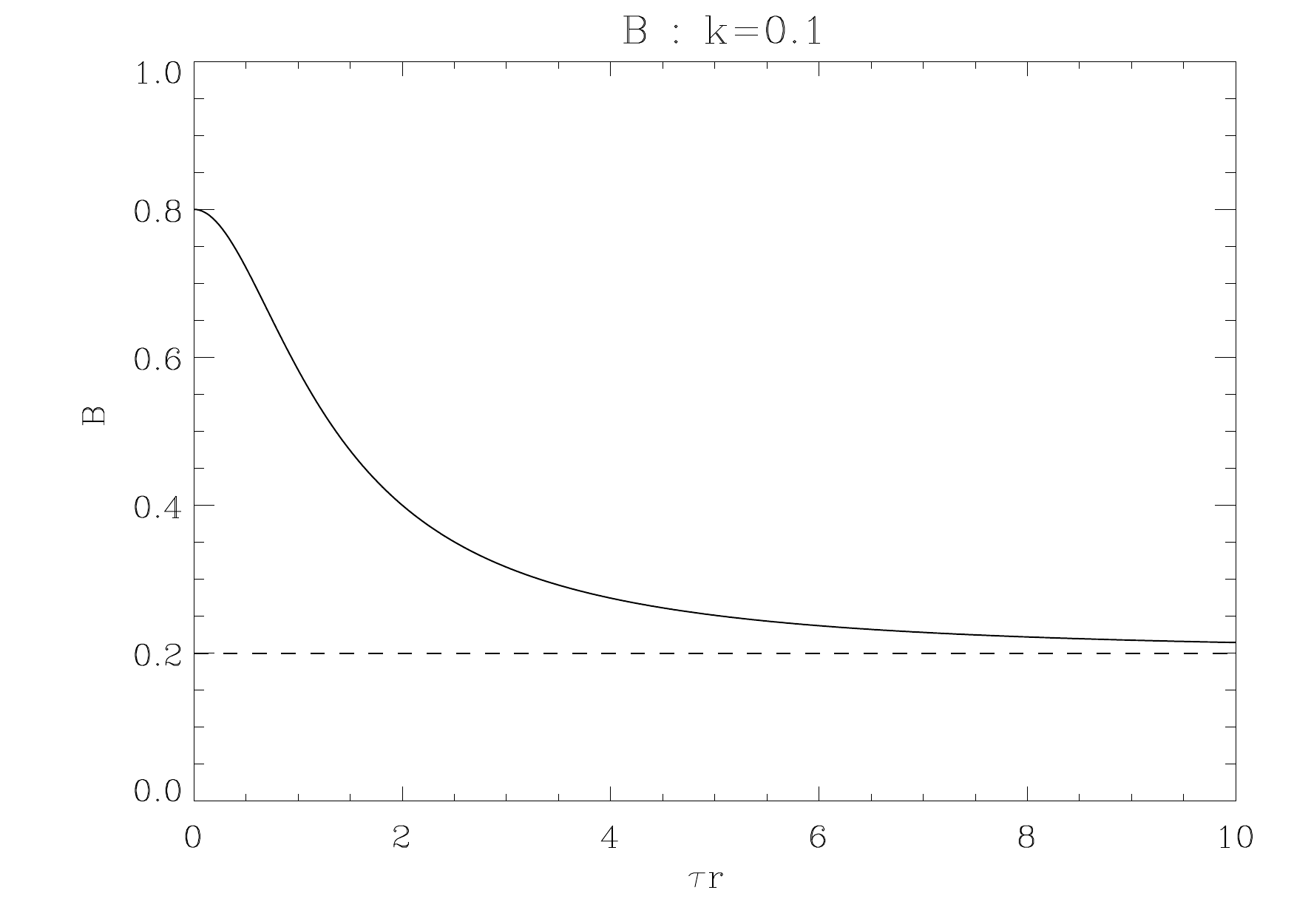}
        \caption{\small $|\boldsymbol{B}|$ for $k=0.1$}
        \label{fig:bmag1}
    \end{subfigure}
       \begin{subfigure}[b]{0.6\textwidth}
        \includegraphics[width=\textwidth]{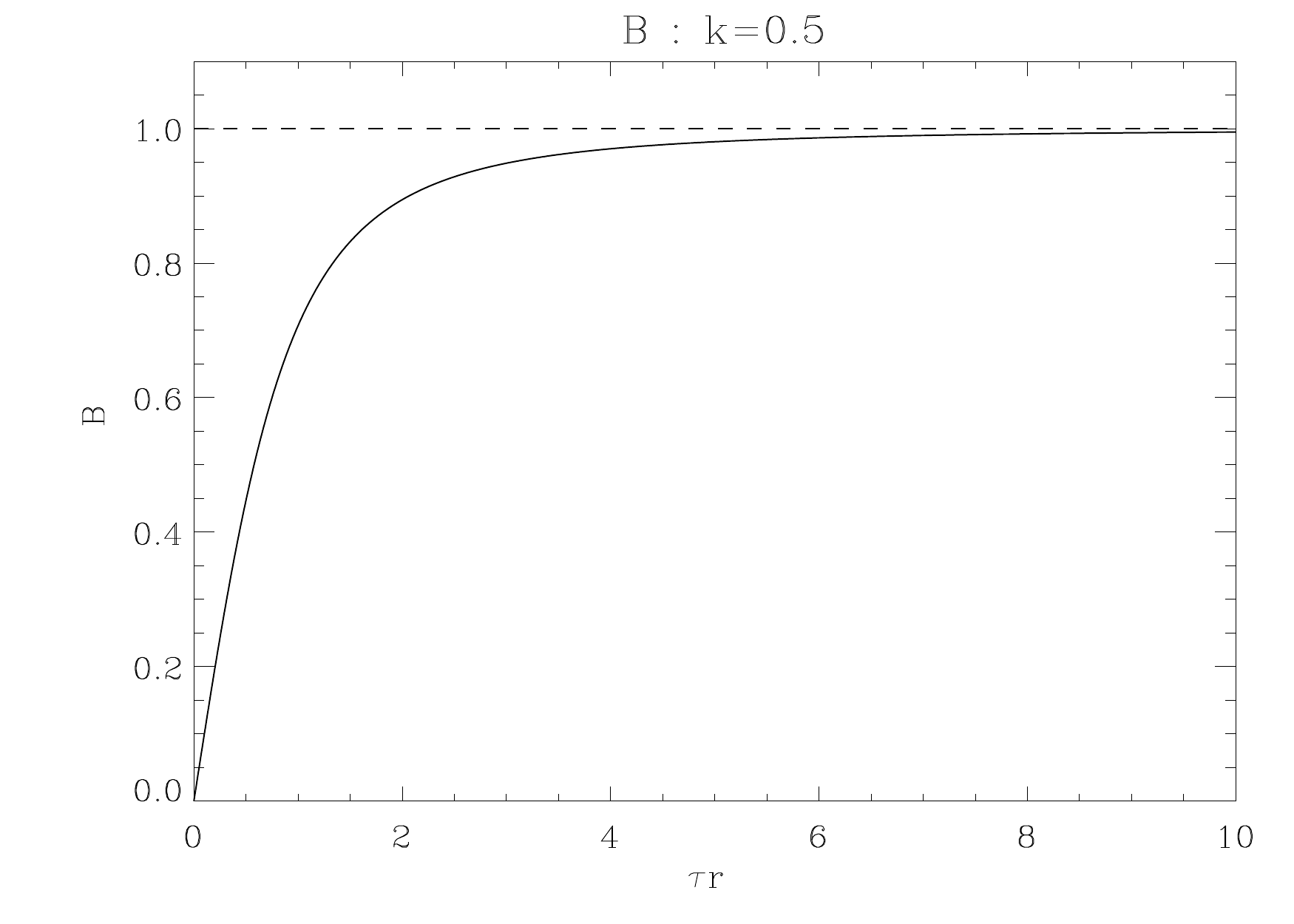}
        \caption{\small  $|\boldsymbol{B}|$ for $k=0.5$ }
        \label{fig:bmag2}
    \end{subfigure}
        \begin{subfigure}[b]{0.6\textwidth}
        \includegraphics[width=\textwidth]{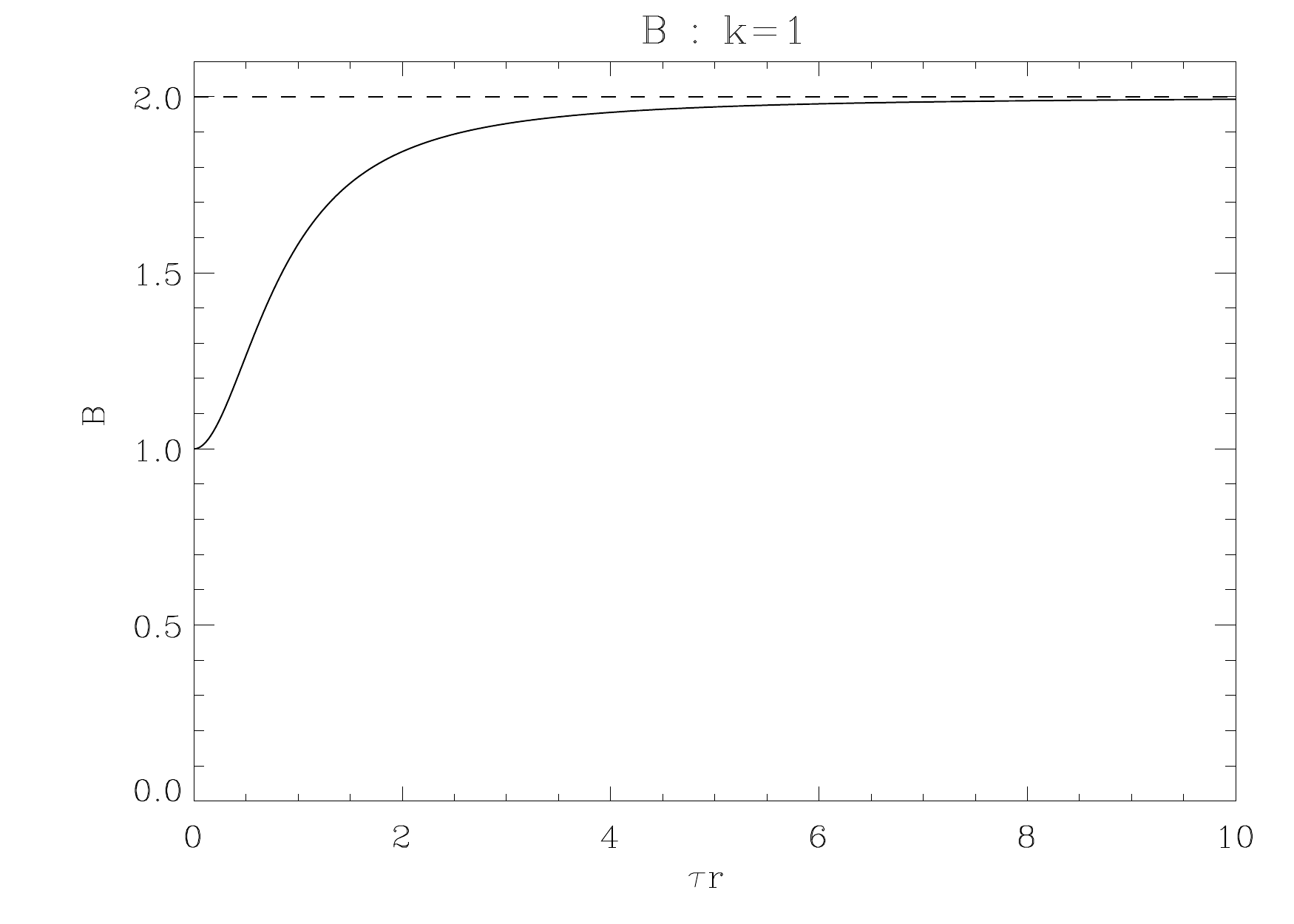}
        \caption{\small  $|\boldsymbol{B}|$ for $k=1$  }
        \label{fig:bmag3}
    \end{subfigure}
    \caption{\small \ref{fig:bmag1}-\ref{fig:bmag3} show the magnitude of the GH+B magnetic field for $k=0.1, 0.5$ and $k=1$ respectively, normalised by $B_0$. For $k<0.5$, $|\tilde{\boldsymbol{B}}|\to 2k$ from above, whereas for $k\ge 1/2$, $|\tilde{\boldsymbol{B}}|\to 2k$ from below.}\label{fig:bmag}
\end{figure}
\subsubsection{First order moments}\label{app:jz}
We take the $v_{z}$ moment of the DF to calculate the $z-$ component of the bulk velocity,
\begin{eqnarray}
&&V_{zs}=\frac{v_{\text{th},s}^4}{n_s}\int\tilde{v}_{z s} f_s d^3\tilde{v}_s,\nonumber\\
&&=\frac{v_{\text{th},s}}{n_s}\frac{n_{0s}}{(\sqrt{2\pi})^2}e^{-\tilde{\phi}_s}\bigg[e^{\left(\tilde{{U}}_{zs}^2+\tilde{r}^2\tilde{\omega}_{s}^2\right)/2}e^{{\tilde{{U}}_{zs}}\tilde{A}_{zs}}e^{{\tilde{\omega}_{s}}\tilde{r}\tilde{A}_{\theta s}}\int_{-\infty}^{\infty} \tilde{v}_{zs}e^{-\left(\tilde{v}_{zs}-\tilde{{U}}_{zs}\right)^2/2}d\tilde{v}_{zs}\times\nonumber\\
&&\int_{-\infty}^{\infty} e^{-\left(\tilde{v}_{\theta s}-\tilde{\omega}_{s}\tilde{r}\right)^2/2}d\tilde{v}_{\theta s}+{C}_s\sqrt{2\pi}e^{\tilde{{V}}_{zs}^2/2}e^{\tilde{{V}}_{zs}\tilde{A}_{zs}}\int_{-\infty}^{\infty} \tilde{v}_{zs}e^{-\left(\tilde{v}_{zs}-\tilde{{V}}_{zs}\right)^2/2}d\tilde{v}_{zs}\bigg] \nonumber\\
&&=\frac{n_{0s}v_{\text{th},s}}{n_s}e^{-\tilde{\phi}_s}\bigg[{\tilde{{U}}_{zs}}e^{{\tilde{{U}}_{zs}}\tilde{A}_{zs}}e^{\left(\tilde{{U}}_{zs}^2+\tilde{r}^2\tilde{\omega}_{s}^2\right)/2}e^{{\tilde{\omega}_{s}}\tilde{r}\tilde{A}_{\theta s}}+\tilde{{V}}_{zs}{C}_se^{\tilde{{V}}_{zs}^2/2}e^{\tilde{{V}}_{zs}\tilde{A}_{zs}}\bigg],
\end{eqnarray}
for $n_s$ the number density. We take the following sum to calculate the $z-$ component of the current density,
\begin{eqnarray}
&&j_z=\sum_sq_sn_sV_{zs}=\sum_s n_{0s}q_sv_{\text{th},s}e^{-\tilde{\phi}_s}\times\nonumber\\
&&\bigg({\tilde{{U}}_{zs}}e^{{\tilde{{U}}_{zs}}\tilde{A}_{zs}}e^{\left(\tilde{{U}}_{zs}^2+\tilde{r}^2\tilde{\omega}_{s}^2\right)/2}e^{{\tilde{\omega}_{s}}\tilde{r}\tilde{A}_{\theta s}}+\tilde{{V}}_{zs}{C}_se^{\tilde{{V}}_{zs}^2/2}e^{\tilde{{V}}_{zs}\tilde{A}_{zs}}\bigg).\label{eq:jzapp}
\end{eqnarray}

By taking the $v_{\theta}$ moment of the DF we can calculate the $\theta-$ component of the bulk velocity,
\begin{eqnarray}
&&V_{\theta s}=\frac{v_{\text{th},s}^4}{n_s}\int\tilde{v}_{\theta s} f_s d^3\tilde{v}_s=\frac{v_{\text{th},s}}{n_s}\frac{n_{0s}}{(\sqrt{2\pi})^2}e^{-\tilde{\phi}_s}\bigg[e^{\left(\tilde{{U}}_{zs}^2+\tilde{r}^2\tilde{\omega}_{s}^2\right)/2}e^{{\tilde{{U}}_{zs}}\tilde{A}_{zs}}e^{{\tilde{\omega}_{s}}\tilde{r}\tilde{A}_{\theta s}}\times\nonumber\\
&&\int_{-\infty}^{\infty} e^{-\left(\tilde{v}_{zs}-\tilde{{U}}_{zs}\right)^2/2}d\tilde{v}_{zs}\int_{-\infty}^{\infty}\tilde{v}_{\theta s} e^{-\left(\tilde{v}_{\theta s}-\tilde{\omega}_{s}\tilde{r}\right)^2/2}d\tilde{v}_{\theta s}\nonumber \\
&&=\frac{\tilde{r}\tilde{\omega}_sn_{0s}v_{\text{th},s}e^{-\tilde{\phi}_s}}{n_s}e^{\left(\tilde{{U}}_{zs}^2+\tilde{r}^2\tilde{\omega}_{s}^2\right)/2}e^{{\tilde{{U}}_{zs}}\tilde{A}_{zs}}e^{{\tilde{\omega}_{s}}\tilde{r}\tilde{A}_{\theta s}} ,
\end{eqnarray}
for $n_s$ the number density. This gives the $\theta-$ component of the current density,
\begin{eqnarray}
j_\theta=\sum_sq_sn_sV_{\theta s}=\sum_s n_{0s}q_sv_{\text{th},s}\tilde{r}\tilde{\omega}_se^{-\tilde{\phi}_s}e^{{\tilde{{U}}_{zs}}\tilde{A}_{zs}}e^{\left(\tilde{{U}}_{zs}^2+\tilde{r}^2\tilde{\omega}_{s}^2\right)/2}e^{{\tilde{\omega}_{s}}\tilde{r}\tilde{A}_{\theta s}}.\label{eq:jthetageneral}
\end{eqnarray}

\subsection{Maxwell's equations: \\
fixing the parameters of the DF}
By insisting on a specific magnetic field configuration (the GH+B field) we have made a statement on the macroscopic physics. In searching for the equilibrium DF, we are trying to understand the microscopic physics. In this sense we are tackling an `inverse problem'. Once an assumption on the form of the DF is made then -- should the assumed form be able to reproduce the correct moments -- this inverse problem reduces to establishing the relationships between the microscopic and macroscopic parameters of the equilibrium. In this Section we `fix' the free parameters of the DF in Equation (\ref{eq:ansatz}), such that Maxwell's equations are satisfied;
\begin{eqnarray}
\nabla\cdot\boldsymbol{E}&=&\frac{1}{\varepsilon_0}\sum_sq_s\int f_sd^3v,\\
\nabla\times\boldsymbol{B}&=&\mu_0\sum_sq_s\int \boldsymbol{v}f_sd^3v.
\end{eqnarray}
Note that the solenoidal constraint and Faraday's law are automatically satisfied for the GH+B field in equilibrium, since $\boldsymbol{B}=\nabla\times\boldsymbol{A}$ implies that $\nabla\cdot\boldsymbol{B}=0$ and $\boldsymbol{E}=-\nabla\phi$ implies that $\nabla\times\boldsymbol{E}=\boldsymbol{0}=-\frac{\partial\boldsymbol{B}}{\partial t}$.

\subsubsection{Amp\`{e}re's Law}
In Section \ref{app:jz} we have calculated the $j_{z}$ current density, found by summing first order moments in $v_z$  of the DF. We now substitute in the macroscopic expressions for $j_z(\tilde{r})$, $A_{\theta}(\tilde{r})$ and $A_z(\tilde{r})$ from (\ref{eq:current}) and (\ref{eq:vecfield}) into the expression for the $j_z$ current density of Equation (\ref{eq:jzapp}). After this substitution, we can calculate a $\phi(r)$ that makes the system consistent. The substitution of the known expressions for $j_{z}$, $A_{z}$ and $A_{\theta}$ gives
\begin{eqnarray}
&&j_{z}(\tilde{r})=\frac{2\tau B_0}{\mu_0}\frac{1}{(1+\tilde{r}^2)^2}=\sum_sn_{0s}q_sv_{\text{th},s}e^{-q_s\beta_s\phi}\times\nonumber\\
&&\bigg({\tilde{{U}}_{zs}}e^{({\tilde{{U}}_{zs}}^2+\tilde{r}^2{\tilde{\omega}_{s}}^2)/2-\text{sgn}(q_s){\tilde{\omega}_{s}}\tilde{r}^2k/\delta_s}\left(1+\tilde{r}^2\right)^{\text{sgn}(q_s)({\tilde{\omega}_{s}}-{\tilde{{U}}_{zs}})/(2\delta_s)}\nonumber\\
&&+{\tilde{{V}}_{zs}}{C}_se^{{\tilde{{V}}_{zs}}^2/2}\left(1+\tilde{r}^2\right)^{-\text{sgn}(q_s){\tilde{{V}}_{zs}}/(2\delta_s)}\bigg)\nonumber\\
&&=\text{``ion terms''} + \text{``electron terms''}
\end{eqnarray}
In order to satisfy the above equality  we can construct a solution by introducing a `separation constant' $\gamma_1\ne 0, 1$. We multiply the above equation by $(1+\tilde{r}^2)^2$ which makes the left-hand side constant, whilst the right-hand side is a sum of two (sets of) terms, one depending on ion parameters and the second depending on electron parameters. Then we can define $\gamma_1$ by
\begin{equation}
\frac{2\tau B_0}{\mu _0}=\underbrace{\frac{2\tau B_0}{\mu _0} (1-\gamma_1)}_{\text{ion terms}} + \underbrace{\frac{2\tau B_0}{\mu _0}\gamma_1}_{\text{electron terms}} ,\label{eq:a1a2}
\end{equation}
associating the `ion term' with the first term on the right-hand side of (\ref{eq:a1a2}), and the `electron term' with the second term on the right-hand side of (\ref{eq:a1a2}). After some algebra we can rearrange these two associations to give two expressions for the scalar potential, one in terms of the ion parameters, and one in terms of the electron parameters:
\begin{eqnarray}
\phi(r)&=&\frac{1}{q_i\beta_i}\ln\bigg\{\frac{\mu_0n_{0i}q_iv_{\text{th},i}}{2\tau B_0(1-\gamma_1)}\times\nonumber\\
&&\bigg[{\tilde{{U}}_{zi}}e^{({\tilde{{U}}_{zi}}^2+\tilde{r}^2{\tilde{\omega}_{i}}^2)/2-{\tilde{\omega}_{i}}\tilde{r}^2k/\delta_i}\left(1+\tilde{r}^2\right)^{2+({\tilde{\omega}_{i}}-{\tilde{{U}}_{zi}})/(2\delta_i)}\nonumber\\
&&+{\tilde{{V}}_{zi}}{C}_ie^{{\tilde{{V}}_{zi}}^2/2}\left(1+\tilde{r}^2\right)^{2-{\tilde{{V}}_{zi}}/(2\delta_i)}     \bigg]\bigg\}\nonumber\\
\phi(r)&=&\frac{1}{q_e\beta_e}\ln\bigg\{\frac{\mu_0n_{0e}q_ev_{\text{th},e}}{2\tau B_0\gamma_1}\bigg[{\tilde{{U}}_{ze}}e^{({\tilde{{U}}_{ze}}^2+\tilde{r}^2{\tilde{\omega}_{e}}^2)/2+{\tilde{\omega}_{e}}\tilde{r}^2k/\delta_e}\left(1+\tilde{r}^2\right)^{2-({\tilde{\omega}_{e}}-{\tilde{{U}}_{ze}})/(2\delta_e)}\nonumber\\
&&+\tilde{{V}}_{ze}{C}_ee^{\tilde{{V}}_{ze}^2/2}\left(1+\tilde{r}^2\right)^{2+\tilde{{V}}_{ze}/(2\delta_e)}     \bigg]\bigg\}\nonumber
\end{eqnarray}
The two values of the scalar potential above must be made identical by a suitable choice of relationships between the ion and electron parameters. Given enough freedom in parameter space, we could say that the $z$ component of Amp\`{e}re's Law is \emph{implicitly} solved the above equations, in that one just needs to choose a consistent set of parameters. However, we seek a solution in an \emph{explicit} sense. 

In order to make progress we non-dimensionalise the above equations by multiplying both sides by $e\beta_r$ with 
\begin{equation*}
\beta_r=\frac{\beta_i\beta_e}{\beta_e+\beta_i}.
\end{equation*}
Once this is done we can write the scalar potential in the form
\begin{eqnarray}
e\beta_r\phi(r)&=&\ln\left\{\left[\text{ion terms}\right]^{\frac{e\beta_r}{q_i\beta_i}}\right\},\label{eq:phi3}\\
e\beta_r\phi(r)&=&\ln\left\{\left[\text{electron terms}\right]^{\frac{e\beta_r}{q_e\beta_e}}\right\}.\label{eq:phi4}
\end{eqnarray}
Specifically, Equations (\ref{eq:phi3}) and (\ref{eq:phi4}) require the equality of the arguments of the logarithm to hold in order for a meaningful solution to be obtained for the scalar potential. A first step towards this is made by requiring consistent powers of the $1+\tilde{r}^2$ `profile' in the right-hand side of the above expression to allow factorisation. Hence
\begin{eqnarray}
({\tilde{\omega}_{i}}-{\tilde{{U}}_{zi}})/(2\delta_i)&=&-{\tilde{{V}}_{zi}}/(2\delta_i),\hspace{3mm} -({\tilde{\omega}_{e}}-{\tilde{{U}}_{ze}})/(2\delta_e)=\tilde{{V}}_{ze}/(2\delta_e),\nonumber\\
&\implies & {\tilde{\omega}_{i}}=\tilde{{U}}_{zi}-\tilde{{V}}_{zi},\hspace{3mm}\tilde{\omega}_e=\tilde{{U}}_{ze}-\tilde{{V}}_{ze},\label{eq:omegafix}
\end{eqnarray}
and hence the rigid-rotation, $\tilde{\omega}_s$, is fixed by the difference of the rectilinear motion, $\tilde{U}_{zs}-\tilde{V}_{zs}$. On top of this, we require that the power of the $1+\tilde{r}^2$ `profile' on the right-hand side is the same for both the ions and electrons, thus
\begin{equation}
 \frac{e\beta_r}{q_i\beta_i}\left(2-\tilde{{V}}_{zi}/(2\delta_i)\right)=\mathcal{E}=\frac{e\beta_r}{q_e\beta_e}\left(2+{\tilde{{V}}_{ze}}/(2\delta_e)\right).\label{eq:Gfix}
\end{equation}
This condition seems to be a statement on an average potential energy associated with the particles. Once more to allow factorisation of the $1+\tilde{r}^2$ `profile', we insist that net $\exp(r^2)$ terms cancel, i.e.
\begin{equation}
 \frac{{\tilde{\omega}_{i}}}{2}=\frac{k}{\delta_i}>0,\hspace{3mm}\frac{{\tilde{\omega}_{e}}}{2}=-\frac{k}{\delta_e} <0.\label{eq:kfix}
\end{equation}
The physical meaning of this condition seems to be that the frequencies of the rigid rotor for each species are matched according to the relevant magnetisation, and the background field magnitude. The remaining task is to ensure equality of the `coefficients'
\begin{eqnarray}
&&\left\{\frac{1}{4\delta_i(1-\gamma_1)} \frac{n_{0i}m_iv_{\text{th},i}^2}{B_0^2/(2\mu_0)}\left[{\tilde{{U}}_{zi}}e^{{\tilde{{U}}_{zi}}^2/2}+\tilde{{V}}_{zi}{C}_ie^{\tilde{{V}}_{zi}^2/2}    \right]\right\}^{\frac{e\beta_r}{q_i\beta_i}}={\mathcal{D}}\nonumber\\
&&=\left\{- \frac{1}{4\delta_e\gamma_1 }\frac{n_{0e}m_ev_{\text{th},e}^2}{B_0^2/(2\mu_0) }  \left[{\tilde{{U}}_{ze}}e^{{\tilde{{U}}_{ze}}^2/2}+\tilde{{V}}_{ze}{C}_ee^{\tilde{{V}}_{ze}^2/2} \right]\right\}^{\frac{e\beta_r}{q_e\beta_e}}\label{eq:D1fix}
\end{eqnarray}
These seem to be conditions on the ratios of the energy densities associated with the bulk rectilinear motion and the magnetic field respectively. Thus far we have 8 constraints and 12 unknowns (${\tilde{{U}}_{zs}}, {\tilde{{V}}_{zs}}, {\tilde{\omega}_{s}}, {C}_s, n_{0s}, \beta_{s}$), given fixed characteristic macroscopic parameters of the equilibrium; $B_0$, $\tau$, and $k$. We can now write down an expression for $\phi$ that explicitly solves the $z$ component of Amp\`{e}re's law;
\begin{equation}
\phi(\tilde{r})=\frac{1}{e\beta_r}\mathcal{E}\ln\left(1+\tilde{r}^2\right)+\phi (0),\label{eq:scalarpot}
\end{equation}
with 
\[
\phi(0)=\frac{1}{e\beta_r}\ln {\mathcal{D}}.
\] 
Clearly, we require that $\mathcal{D}>0$ for the expression above to make sense. It is clear that the sign of $\gamma_1$ could, in principle, affect the sign of $\mathcal{D}$. It is seen from (\ref{eq:D1fix}) that positivity of $\mathcal{D}$ implies that
\begin{eqnarray}
\frac{1}{1-\gamma_1}\left[ \tilde{U}_{zi}e^{\tilde{U}_{zi}^2/2}+\tilde{V}_{zi}C_ie^{\tilde{V}_{zi}^2/2}    \right]&>&0,\\
\frac{1}{\gamma_1}\left[ \tilde{U}_{ze}e^{\tilde{U}_{ze}^2/2}+\tilde{V}_{ze}C_ee^{\tilde{V}_{ze}^2/2}    \right]&<&0.
\end{eqnarray}
By rearranging the above inequalities to make $C_s$ the subject, it can be seen after some algebra that positivity of $\mathcal{D}$ and $C_s$ is guaranteed when 
\[
\gamma_1>1, \hspace{3mm} \text{sgn}(\tilde{U}_{zs})=-\text{sgn}(\tilde{V}_{zs}).
\]
Note that these conditions are sufficient, but not necessary, i.e. it is possible to have $\mathcal{D}>0$ and $C_s>0$ for any value of $\gamma_1\ne 0,1$, and even for $\text{sgn}(\tilde{U}_{zs})=\text{sgn}(\tilde{V}_{zs})$ in the case of $\gamma_1<0$. 

Thus far we have only considered the $j_z$ component, and it is premature to consider all components of Amp\`{e}re's Law satisfied. Let us move on to consider the $\theta$ component. In a process similar to that above, we substitute in the macroscopic expressions for $j_\theta(\tilde{r})$, $A_{\theta}(\tilde{r})$ and $A_z(\tilde{r})$ for the GH+B field into the expression for the $j_\theta$ current density of Equation (\ref{eq:jthetageneral}) in Section \ref{app:jz}. After this substitution, we can once more calculate the $\phi$ that makes the system consistent. The substitution gives
\begin{eqnarray}
&&j_{\theta}=\frac{2\tau B_0}{\mu_0}=\sum_{s}n_{0s}q_sv_{\text{th},s}\tilde{\omega}_se^{-q_s\beta_s\phi}\times\nonumber\\
&&e^{({\tilde{{U}}_{zs}}^2+\tilde{r}^2{\tilde{\omega}_{s}}^2)/2-\text{sgn}(q_{s)}{\tilde{\omega}_{s}}\tilde{r}^2k/\delta_s}\left(1+\tilde{r}^2\right)^{2+\text{sgn}(q_{s})({\tilde{\omega}_{s}}-{\tilde{{U}}_{zs}})/(2\delta_s)}
\end{eqnarray}
Using the parameter relations as above, we determine that the scalar potential is again given in the form of (\ref{eq:scalarpot}), 
\begin{equation*}
\phi(\tilde{r})=\frac{1}{e\beta_r}\mathcal{E}\ln\left(1+\tilde{r}^2\right)+\phi(0).
\end{equation*} 
Hence, this form of the scalar potential is consistent provided
\begin{equation}
\left[\frac{1}{1-\gamma_2}\frac{1}{4\delta_i} \frac{n_{0i}m_iv_{\text{th},i}\omega_i/\tau}{B_0^2/(2\mu_0)} e^{{\tilde{{U}}_{zi}}^2/2}   \right]^{\frac{e\beta_r}{q_i\beta_i}}={\mathcal{D}}=\left[ -\frac{1}{\gamma_2} \frac{1}{4\delta_e}\frac{n_{0e}m_ev_{\text{th},e}\omega_e/\tau}{B_0^2/(2\mu_0)  } e^{{\tilde{{U}}_{ze}}^2/2}  \right]^{\frac{e\beta_r}{q_e\beta_e}}\label{eq:D2fix}
\end{equation}
for $\gamma_2\ne 1$ another separation constant. These seem to be conditions on the ratios of the energy densities associated with the bulk rotation and the magnetic field respectively. This has added two more constraints. 

Once again we must ensure that $\mathcal{D}>0$. Since $\omega_e<0$, the right-hand side of the above equation implies that $\gamma_2>0$ to ensure that $\mathcal{D}>0$. Whilst the left-hand side implies that $\gamma_2<1$ for positivity of $\mathcal{D}$ since $\omega_i>0$. Hence we can say that for positivity
\[
0<\gamma_2<1.
\]
We can now consider Amp\`{e}re's Law satsified, given a $\phi$ that solves Poisson's equation. That is to say that we have satisfied the equation
\begin{eqnarray}
&&\left(\sum_sq_s\int\boldsymbol{v}f_sd^3v=\right)\boldsymbol{j}_{\text{micro}}(\phi,\boldsymbol{A})=\boldsymbol{j}_{\text{macro}}(r)\left(=\frac{1}{\mu_0}\nabla\times\boldsymbol{B}\right),\nonumber\\
&&\text{s.t.}\hspace{3mm}\phi(\tilde{r})=\frac{1}{e\beta_r}\mathcal{E}\ln\left(1+\tilde{r}^2\right)+\phi(0)\hspace{3mm}\text{and}\hspace{3mm}\boldsymbol{A}=\boldsymbol{A}_{GH+B},\nonumber
\end{eqnarray}
with $\boldsymbol{A}_{GH+B}$ defined by Equation (\ref{eq:vecfield}). As a result, the problem of consistency is now shifted to solving Poisson's Equation, where the remaining degrees of freedom lie. 

\subsubsection{Poisson's Equation}
The final step in `self-consistency' is to solve Poisson's Equation. Frequently in such equilibrium studies, this step is replaced by satisfying quasineutrality and in essence solving a first order approximation of Poisson's equation, see for example \citet{Schindlerbook, Harrison-2009POP,Tasso-2014} and Section \ref{sec:quasi} of this thesis. Here we solve Poisson's equation exactly, i.e. to all orders. Poisson's equation in cylindrical coordinates with only radial dependence gives
\begin{equation}
\nabla\cdot\boldsymbol{E}=-\frac{1}{r}\frac{\partial}{\partial r}\left(r\frac{\partial \phi}{\partial r}\right)=\frac{\sigma}{\varepsilon_0}.\label{eq:Poiss}
\end{equation}
The electric field is calculated as $\boldsymbol{E}=-\nabla \phi$, giving
\begin{equation}
E_r=-\partial_r\phi=-\frac{2\tau \mathcal{E}}{e\beta_r}\frac{\tilde{r}}{(1+\tilde{r}^2)}.\label{eq:elecfield}
\end{equation}
We can now take the divergence of the electric field $\nabla\cdot\boldsymbol{E}=\tau\tilde{r}^{-1}\partial_{\tilde{r}}(\tilde{r}E_{r})$ and so
\begin{equation}
\nabla\cdot\boldsymbol{E}=-\frac{4\tau^2\mathcal{E}}{e\beta_r}\frac{1}{(1+\tilde{r}^2)^2}\implies\sigma=-\frac{4\varepsilon_{0}\tau^2\mathcal{E}}{e\beta_r}\frac{1}{(1+\tilde{r}^2)^2}.\label{eq:Edivergence}
\end{equation}
This gives a non-zero net charge - per unit length in $z$ - of    
\begin{equation}
\mathcal{Q}=\int_{\theta=0}^{\theta=2\pi}\int_{r=0}^{r=\infty} \,\sigma \,r\,dr\,d\theta= -\frac{4\pi\varepsilon_{0}\mathcal{E}}{e\beta_r}.\label{eq:coulomb}
\end{equation}
The charge density derived in Equation (\ref{eq:Edivergence}) must equal the charge density calculated by taking the zeroth moment of the DF. The expression for the charge density calculated in (\ref{eq:sigmaapp}) gives
\begin{eqnarray}
\sigma&=&\sum_sn_{0s}q_se^{-q_s\beta_s\phi}\times\nonumber\\
&&\left(e^{({\tilde{{U}}_{zs}}^2+\tilde{r}^2{\tilde{\omega}_{s}}^2)/2}e^{{\tilde{{U}}_{zs}}\tilde{A}_{zs}}e^{{\tilde{\omega}_{s}}\tilde{r}\tilde{A}_{\theta s}}+{C}_se^{({\tilde{{U}}_{zs}}-{\tilde{\omega}_{s}})^2/2}e^{({\tilde{{U}}_{zs}}-{\tilde{\omega}_{s}})\tilde{A}_{zs}}\right) ,  \nonumber\\
&=&\sum_sn_{0s}q_se^{-q_s\beta_s\phi}\times\nonumber\\
&&\left(1+\tilde{r}^2\right)^{{\rm sgn}(q_s)({\tilde{\omega}_{s}}-{\tilde{{U}}_{zs}})/(2\delta_s)}\left(e^{{\tilde{{U}}_{zs}}^2/2}+{C}_se^{({\tilde{{U}}_{zs}}-{\tilde{\omega}_{s}})^2/2}\right),\nonumber\\
&=&\frac{1}{\left(1+\tilde{r}^2\right)^{2}}\sum_sn_{0s}q_s{\mathcal{D}}^{-\frac{q_s\beta_s}{e\beta_r}}\left(e^{{\tilde{{U}}_{zs}}^2/2}+{C}_se^{({\tilde{{U}}_{zs}}-{\tilde{\omega}_{s}})^2/2}\right).  \label{eq:chargedensity}
\end{eqnarray}
The second equality is found by substituting the form of the vector potential from Equation (\ref{eq:vecfield}), and the final equality is reached by using the conditions derived in Equations (\ref{eq:omegafix}) - (\ref{eq:scalarpot}).

We can now match Equations (\ref{eq:Edivergence}) and (\ref{eq:chargedensity}) to get
\begin{equation}
(\sigma(0)=)-\frac{4\varepsilon_0\tau^2\mathcal{E}}{e\beta_r}=\sum_{s}n_{0s}q_s  {\mathcal{D}}^{-\frac{q_s\beta_s}{e\beta_r}}\left(e^{{\tilde{{U}}_{zs}}^2/2}+{C}_se^{{\tilde{{V}}_{zs}}^2/2}\right).\label{eq:poisscon}
\end{equation}
We now have 12 physical parameters (${\tilde{{U}}_{zs}}, {\tilde{{V}}_{zs}}, {\tilde{\omega}_{s}}, {C}_s, n_{0s},\beta_{s}$) with 11 constraints (\ref{eq:omegafix}-\ref{eq:D1fix}), (\ref{eq:D2fix}) \& (\ref{eq:poisscon}). For example, if one picks $B_0$, $\tau$, $k$ and one microscopic parameter, say $\beta_i$, then the remaining parameters of the equilibrium, (${\tilde{{U}}_{zs}}, {\tilde{{V}}_{zs}}, {\tilde{\omega}_{s}}, {C}_s$, $n_{0s}$, $\beta_{e}$), are now determined. One could of course choose the values of a different set of parameters, and determine those that remain by using the constraints derived. Note that whilst the constants $\gamma_1\ne 0,1$ and $0< \gamma_2 <1 $ are system parameters, they are not physically meaningful as they only represent a change in the gauge of the scalar potential.

\section{Analysis of the equilibrium}\label{sec:analysis}

\subsection{Non-neutrality \& the electric field}
It is seen from equations (\ref{eq:Edivergence}) and (\ref{eq:coulomb}) that basic electrostatic properties of the equilibrium described by $f_s$ are encoded in $\mathcal{E}$. The equilibrium is electrically neutral only when $\mathcal{E}=0$, and non-neutral otherwise. Specifically, there is net negative charge when $\mathcal{E}>0$, and net positive charge when $\mathcal{E}<0$. This net charge is finite in the $(r,\theta)$ plane and given by $\mathcal{Q}$ in Equation (\ref{eq:coulomb}).

Physically, the sign of $\mathcal{E}$ seems to be related to the respective magnitudes of the bulk rotation frequencies, $\tilde{\omega}_s$. From equations (\ref{eq:omegafix}) and (\ref{eq:Gfix}) we see that $\mathcal{E}>0$ implies that
\begin{eqnarray}
\tilde{\omega}_i>\omega_{i}^{\star}&=&\tilde{U}_{zi}-4\delta_i,\nonumber\\
|\tilde{\omega}_e|<\omega_{e}^{\star}&=&-\tilde{U}_{ze}-4\delta_e,\nonumber
\end{eqnarray}
and $\mathcal{E}<0$ implies that
\begin{eqnarray}
\tilde{\omega}_i<\omega_{i}^{\star}&=&\tilde{U}_{zi}-4\delta_i,\nonumber\\
|\tilde{\omega}_e|>\omega_{e}^{\star}&=&-\tilde{U}_{ze}-4\delta_e.\nonumber
\end{eqnarray}
Hence, $\mathcal{E}>0$ is seen to occur for `sufficiently large' bulk ion rotation frequencies, and `sufficiently small' (in magnitude) bulk electron rotation frequencies. A positive $\mathcal{E}$ corresponds to an electric field directed radially `inwards'. This seems to make sense physically, by the following argument. A `larger' ($\tilde{\omega}_i>\omega_{i}^{\star}$) bulk ion rotation freqency gives a `larger' centrifugal force (in the co-moving frame), and a `smaller' ($ |\tilde{\omega}_e|<\omega_{e}^{\star} $) bulk electron rotation frequency gives a `smaller' centrifugal force (in the co-moving frame). For a dynamic interpretation, at a fixed $r$, the ions are forced to a slightly larger radius than the electrons, i.e. a charge separation manifests on small scales. This charge separation results in an inward electric field, $E_r<0$. An equally valid interpretation is to say that for an equilibrium to exist, an electric field must exist to counteract the differences in the forces associated with the bulk ion and electron rotational flows. This effect is represented in Figure \ref{fig:E_in}.

In a similar manner, $\mathcal{E}<0$ is seen to occur for `sufficiently small' ($\tilde{\omega}_i<\omega_{i}^{\star}$) bulk ion rotation frequencies, and `sufficiently large' ($ |\tilde{\omega}_e|>\omega_{e}^{\star} $) bulk electron rotation frequencies. A negative $\mathcal{E}$ corresponds to an electric field directed radially `outwards'. We can then interpret these result physically, in a manner like that above. This effect is represented in Figure \ref{fig:E_out}.

Finally, we can interpret the neutral case, $\mathcal{E}=0$, as the intermediary between the two circumstances considered above. That is to say that the equilibrium is neutral when the bulk rotation flows are just matched accordingly, such that there is no charge separation and hence no electric field.

\subsection{The equation of state and the plasma beta}
For certain considerations, e.g. the solar corona, it would be advantageous if the DF had the capacity to describe plasmas with sub-unity values of the plasma beta: the ratio of the thermal energy density to the magnetic energy density
\begin{equation}
\beta_{pl}(\tilde{r})=\frac{2\mu _0k_B}{B^2}\sum _sn_sT_s\label{eq:beta}.
\end{equation}
For our configuration, the number density is seen to be proportional to the $rr$ component of the pressure tensor, $P_{rr,s}=n_sk_BT_s$. This is demonstrated by the following calculation. In order to calculate $P_{rr}$, we must consider the integral
\begin{equation}
P_{rr}=\sum_s m_s\int_{-\infty}^\infty \, w_{rs}\,w_{rs}\,f_s\,d^3v.\label{eq:Prr}
\end{equation}
However, we do not have to consider a bulk velocity in the $r$ direction here $(V_{rs}=0)$, since $f_s$ is an even function of $v_r$. Using the fact that
\begin{eqnarray}
\int_{-\infty}^{\infty} v_r^2e^{-v_r^2/(2v_{\text{th},s}^2)}dv_r&=&v_{\text{th},s}^2\int_{-\infty}^{\infty} e^{-v_r^2/(2v_{\text{th},s}^2)}dv_r,\nonumber
\end{eqnarray}
and by consideration of Equations (\ref{eq:Prr}) and the number density, we see that
\begin{eqnarray}
P_{rr,s}=m_sv_{\text{th},s}^2n_s,\label{eq:temp}
\end{eqnarray}
that is to say that $k_BT_s=m_sv_{\text{th},s}^2$. Note that if $n_i=n_e:= n$ and hence $\mathcal{E}=0$ (neutrality), then we have an equation of state given by
\begin{equation*}
P_{rr}=\frac{\beta_e+\beta_i}{\beta_e\beta_i}n.
\end{equation*}
This resembles expressions found in the Cartesian case, in \citet{Channell-1976, Neukirch-2009, Allanson-2015POP} for example. Incidentally, we can use the connection between $n_s$ and $P_{rr}$ to give an expression for the $\beta_{pl}$ that is perhaps more typically seen,
\begin{equation*}
\beta_{pl}(\tilde{r})=\frac{2\mu _0}{B^2}\sum _sP_{rr,s}.
\end{equation*}
The square magnitude of the magnetic field (Equation (\ref{eq:magfield})) is given by
\begin{equation*}
B^2=\frac{B_0^2}{(1+\tilde{r}^2)}\left(1-4k+4k^2(1+\tilde{r}^2)\right).
\end{equation*}
Using the number density from Equation (\ref{eq:numdensity}) in the definition of the plasma beta from Equation (\ref{eq:beta}), as well as the equilibrium conditions (\ref{eq:omegafix}) - (\ref{eq:scalarpot}) gives
\begin{eqnarray}
&&\beta_{pl}(\tilde{r})=\frac{2\mu _0}{B_0^2(1+\tilde{r}^2)\left(1-4k+4k^2(1+\tilde{r}^2)\right)}\times\nonumber\\
&&\sum _s\frac{n_{0s}}{\beta_s}{\mathcal{D}}^{-\frac{q_s\beta_s}{e\beta_r}}\left(e^{{\tilde{{U}}_{zs}}^2/2}+{C}_se^{{\tilde{{V}}_{zs}}^2/2}\right).
\end{eqnarray}
It is not immediately obvious from the above equation what values $\beta_{pl}$ can have. However it is readily seen that as $\tilde{r}\to\infty$ then $\beta_{pl}\to 0$, essentially since the number density is vanishing at large radii. On the central axis of the tube we see that 
\begin{eqnarray}
&&\beta_{pl}(0)=\frac{2\mu _0}{B_0^2\left(1-4k+4k^2\right)}\times\nonumber\\
&&\sum _s\frac{n_{0s}}{\beta_s}{\mathcal{D}}^{-\frac{q_s\beta_s}{e\beta_r}}\left(e^{{\tilde{{U}}_{zs}}^2/2}+{C}_se^{{\tilde{{V}}_{zs}}^2/2}\right),
\end{eqnarray}
suggesting that for a suitable choice of parameters, it should be possible to attain any value of $\beta_{pl}$ on the axis.

\subsection{Origin of terms in the equation of motion}\label{subsec:origin}
It could be instructive to now consider the individual terms in the equation of motion for this equilibrium, Equation (\ref{eq:CMHDE_alt}), and repeated here,
\begin{equation*}
(\nabla\cdot\boldsymbol{P})_r=(\boldsymbol{j}\times\boldsymbol{B})_r+\sigma E_r-\boldsymbol{\mathcal{F}}_{\text{c}}\cdot \hat{\boldsymbol{e}}_r.
\end{equation*}
We will seek to see if, at least mathematically, that certain terms have their origin in other particular terms in the equation, and what these are. Rather than this suggesting `what balances what', it is an attempt to see the physical origin of the forces, i.e. \emph{which forces arise from which system configurations?}

\subsubsection{Centripetal forces and non-inertial motion}
Let's first consider the divergence of the pressure, Equation (\ref{eq:pressurediv}), and repeated here
\begin{equation*}
(\nabla\cdot\boldsymbol{P})_r= \frac{1}{r}\frac{\partial}{\partial r}\left(rP_{rr}\right)-\frac{P_{\theta\theta}}{r}.         
 \end{equation*}
 As mentioned in Section \ref{subsec:eom}, $P_{\theta\theta}=\pi_{\theta\theta}-\sum_sn_sV_{\theta s}^2$, since 
 \begin{eqnarray}
 P_{\theta\theta}&=&\sum_s\int(V_{\theta s}-v_\theta)^2f_sd^3v,\nonumber\\
 &=&\sum_s\left[n_sV_{\theta s}^2  -2n_sV_{\theta s}^2+\int v_{\theta}^2f_sd^3v   \right],\nonumber\\
 &=&\sum_s\left[\int v_\theta^2f_sd^3v-n_sV_{\theta s}^2\right],\nonumber\\
 &=&\pi_{\theta\theta}-\sum_sn_sV_{\theta s}^2=\pi_{\theta\theta}+\boldsymbol{\mathcal{F}}_c\cdot r\hat{\boldsymbol{e}}_r.\nonumber
 \end{eqnarray}
Hence the centripetal forces, $\boldsymbol{\mathcal{F}}_c=-\frac{1}{r}\sum_s\rho_sV_{\theta s}^2\hat{\boldsymbol{e}}_r$ are seen to have their origin in the terms in $P_{\theta\theta}/r$, from $\nabla\cdot\boldsymbol{P}$. This seems to say that in a lab frame, the centripetal forces arise from the stresses associated with the differences between the particle and bulk velocities, i.e. the $\int v_\theta V_{\theta s} f_s d^3v$ terms. So far we have accounted for the following terms,
\[
 \underbrace{\frac{1}{r}\frac{\partial}{\partial r}\left(rP_{rr}\right)-\frac{1}{r}\pi_{\theta\theta}+\textcolor{red}{\frac{1}{r}\sum_sn_sV_{\theta s}^2} }_{\text{``Derivatives'' of potentials}} =\underbrace{(\boldsymbol{j}\times\boldsymbol{B})_r+\sigma E_r-\textcolor{red}{\boldsymbol{\mathcal{F}}_{\text{c}}\cdot \hat{\boldsymbol{e}}_r}}_{\text{Forces}}.
 \]
 
 \subsubsection{Electric fields and pressure gradients}
 We now consider the $P_{rr}$ terms. Using the `equation of state' (\ref{eq:temp}), and the $n_s$ implicit from Equation (\ref{eq:chargedensity}) we see that
 \begin{eqnarray}
 \frac{1}{r}\frac{\partial}{\partial r}(rP_{rr})= \frac{1}{r}\frac{\partial}{\partial r}\left[r\left(   \frac{n_i}{\beta_i}+\frac{n_e}{\beta_e}\right)\right]&\propto &\frac{1}{\tilde{r}}\frac{\partial}{\partial \tilde{r}}\left[\tilde{r}       \frac{1}{(1+\tilde{r}^2)^2}\right],\nonumber\\
 &=&\underbrace{\frac{1}{\tilde{r}}\frac{1}{(1+\tilde{r}^2)^2}}_{\propto P_{rr}/r}-\underbrace{\frac{4\tilde{r}}{(1+\tilde{r}^2)^3}}_{\propto\partial P_{rr}/\partial r}\label{eq:prrexpand}
 \end{eqnarray}
 We can see from equations (\ref{eq:elecfield}), (\ref{eq:chargedensity}) and (\ref{eq:poisscon}), that 
\[
\sigma E_r=\frac{8\epsilon_0\tau^3\mathcal{E}^2}{e^2\beta_r^2}\frac{\tilde{r}}{(1+\tilde{r}^2)^3}.
\]
Hence the electric fields have their origin in the density/pressure gradients $\partial P_{rr}/\partial r$, and we have accounted for the following terms,
\[
 \underbrace{\frac{1}{r}P_{rr}+\textcolor{blue}{\frac{\partial }{\partial r}\sum_s\frac{n_s}{\beta_s}}-\frac{1}{r}\pi_{\theta\theta}+\textcolor{red}{\frac{1}{r}\sum_sn_sV_{\theta s}^2} }_{\text{``Derivatives'' of potentials}} =\underbrace{(\boldsymbol{j}\times\boldsymbol{B})_r+\textcolor{blue}{\sigma E_r}-\textcolor{red}{\boldsymbol{\mathcal{F}}_{\text{c}}\cdot \hat{\boldsymbol{e}}_r}}_{\text{Forces}}.
 \]

\subsubsection{`Lorentz forces' and $\pi _{\theta\theta}$}
Using the definition of the DF (Equation (\ref{eq:ansatz})), let's now consider the form of $\pi_{\theta\theta}/r$,
\begin{eqnarray}
-\frac{1}{r}\pi_{\theta\theta}=-\frac{1}{r}\sum_s\int v_{\theta}^2f_sd^3v&=& -\sum_s\frac{1}{r}\left(n_sr^2\omega_s^2+K_1n_s\right),\nonumber\\
&\propto&-\sum_s\frac{1}{\tilde{r}}\left(   \frac{\tilde{r}^2\tilde{\omega}_s^2}{(1+\tilde{r}^2)^2}+\frac{K_1}{(1+\tilde{r}^2)^2}         \right),\nonumber
\end{eqnarray}
for $K_1$ a positive constant, and using elementary integrals. The second term on the RHS is seen to cancel with the first term on the RHS of Equation (\ref{eq:prrexpand}), i.e. $ P_{rr}/ r$. Also, we see from equations (\ref{eq:magfield}) and (\ref{eq:current}) that
\[
\left(\boldsymbol{j}\times\boldsymbol{B}\right)_r=-\frac{4k\tau B_0^2}{\mu_0}\frac{\tilde{r}}{(1+\tilde{r}^2)^2},
\]
and so we see that the $\boldsymbol{j}\times\boldsymbol{B}$ force has it's origins in $\pi_{\theta\theta}$. Now we are in a position to account for all the terms in force balance,
\[
 \underbrace{ \textcolor{blue}{\frac{\partial }{\partial r}\sum_s\frac{n_s}{\beta_s}}+\textcolor{red}{\frac{1}{r}\sum_sn_sV_{\theta s}^2} -  \textcolor{green}{\frac{1}{r}\sum_sn_sr^2\omega_s^2} }_{\text{``Derivatives'' of potentials}} =\underbrace{\textcolor{blue}{\sigma E_r}-\textcolor{red}{\boldsymbol{\mathcal{F}}_{\text{c}}\cdot \hat{\boldsymbol{e}}_r}  + \textcolor{green}{(\boldsymbol{j}\times\boldsymbol{B})_r} }_{\text{Forces}}.
 \]
 
 \subsubsection{Summary of force balance analysis}
 The conclusions reached from this analysis are somewhat general since some results did not depend on the specific electromagnetic fields $(\boldsymbol{E},\boldsymbol{B})$. Regardless, we see that 
 \begin{itemize}
 \item The electric field sources/balances gradients in the particle number densities
 \item The centripetal forces are sourced/balanced by the bulk angular flows, $V_{\theta s} (r)$
 \item The Lorentz force is sourced/balanced by a centripetal-type force, that treats the flow as uniform circular motion, $V_{\theta s}=r\tilde{\omega}_s$, i.e. rotational flows consistent with a rigid-rotor (see Section \ref{sec:thedf}).
 \end{itemize}

\begin{figure}
    \centering
    \begin{subfigure}[b]{0.45\textwidth}
        \includegraphics[width=\textwidth]{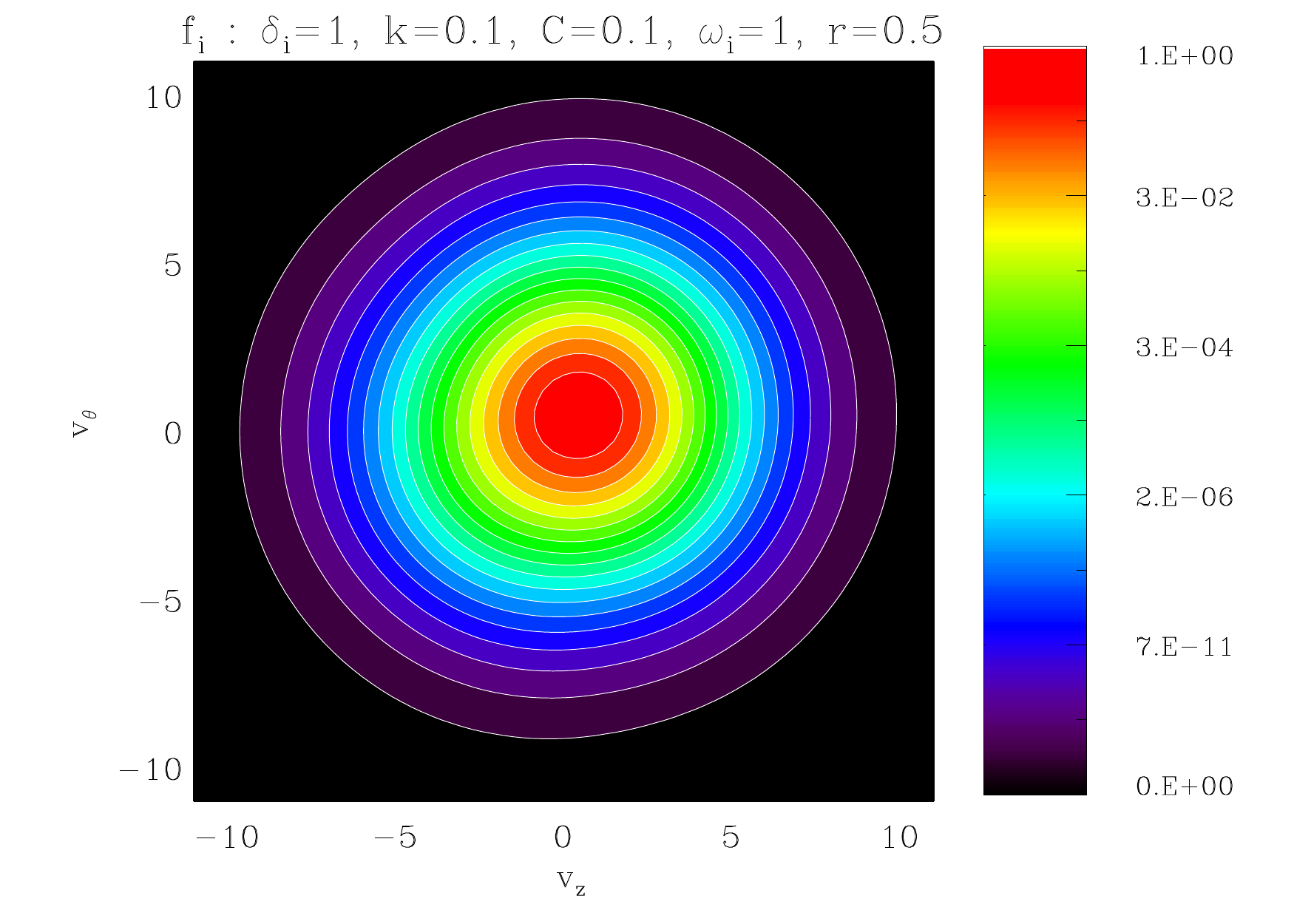}
        \caption{\small $(\tilde{\omega}_i,\tilde{r},C_i)=(1,0.5,0.1)$}
        \label{fig:4a}
    \end{subfigure}
       \begin{subfigure}[b]{0.45\textwidth}
        \includegraphics[width=\textwidth]{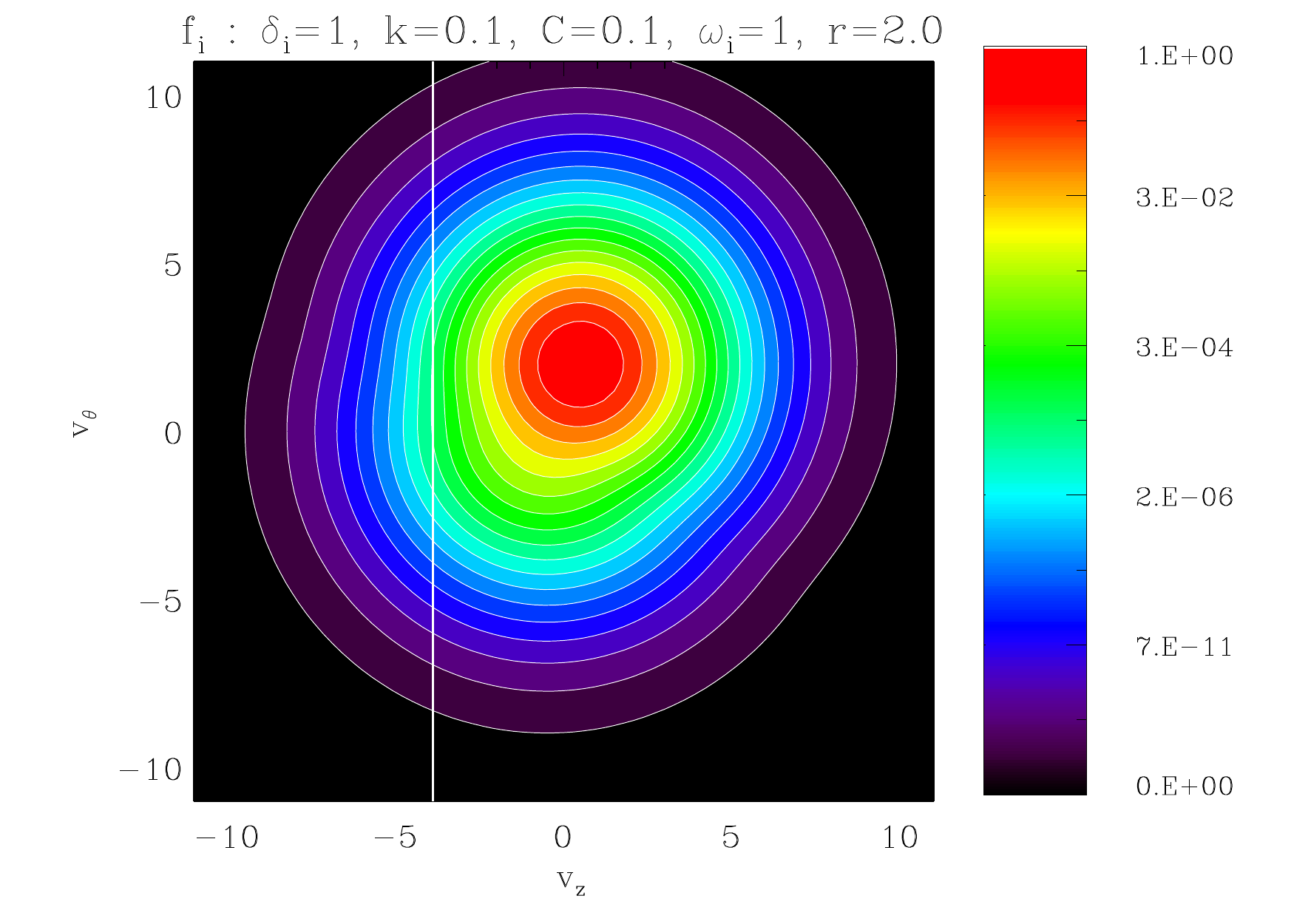}
        \caption{\small $(\tilde{\omega}_i,\tilde{r},C_i)=(1,2,0.1)$}
        \label{fig:4b}
    \end{subfigure}
        \begin{subfigure}[b]{0.45\textwidth}
        \includegraphics[width=\textwidth]{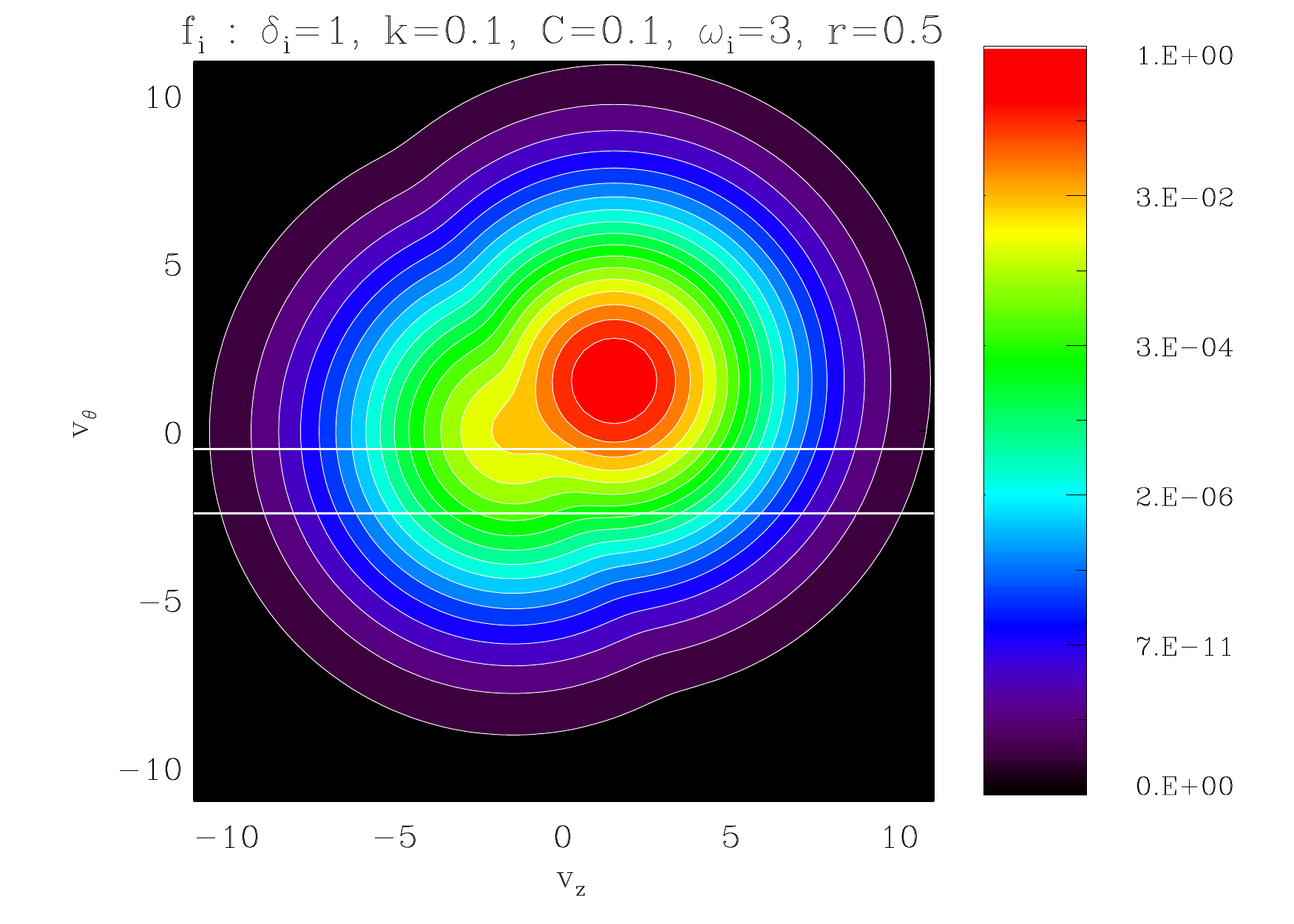}
        \caption{\small $(\tilde{\omega}_i,\tilde{r},C_i)=(3,0.5,0.1)$}
        \label{fig:4c}
    \end{subfigure}
    \begin{subfigure}[b]{0.45\textwidth}
        \includegraphics[width=\textwidth]{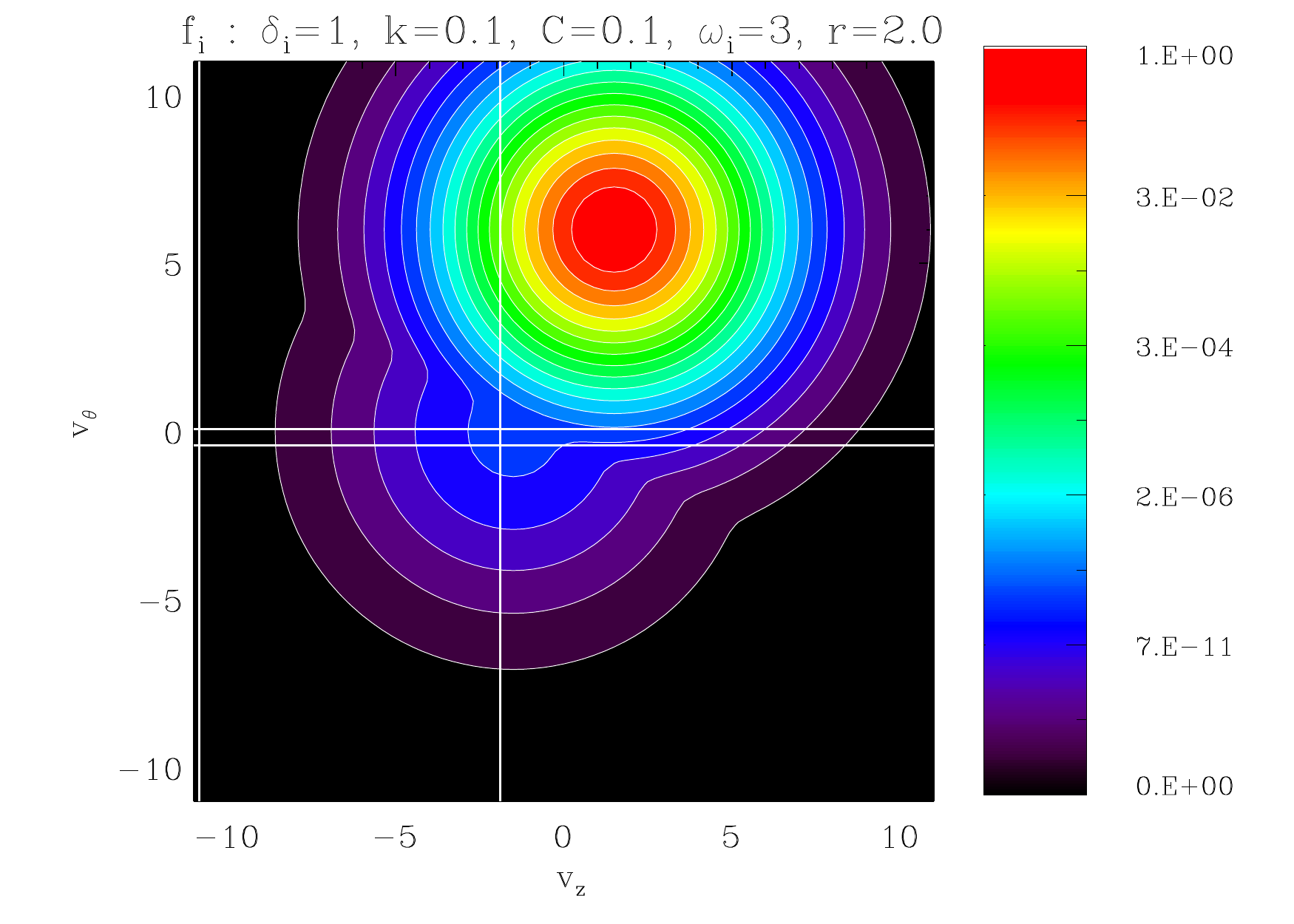}
        \caption{\small $(\tilde{\omega}_i,\tilde{r},C_i)=(3,2,0.1)$}
        \label{fig:4d}
    \end{subfigure}
        \begin{subfigure}[b]{0.45\textwidth}
        \includegraphics[width=\textwidth]{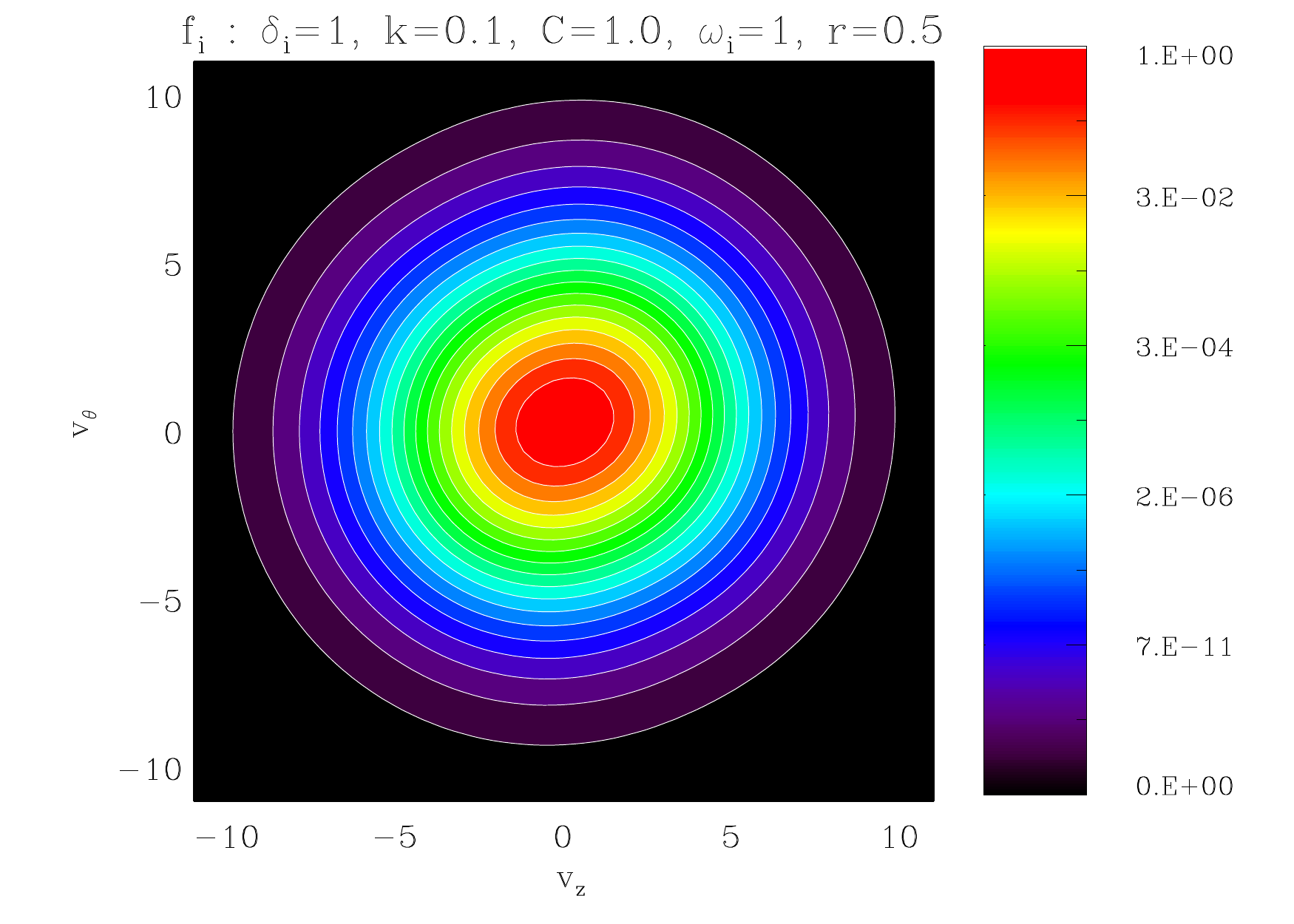}
        \caption{\small $(\tilde{\omega}_i,\tilde{r},C_i)=(1,0.5,1)$}
        \label{fig:4e}
    \end{subfigure}
    \begin{subfigure}[b]{0.45\textwidth}
        \includegraphics[width=\textwidth]{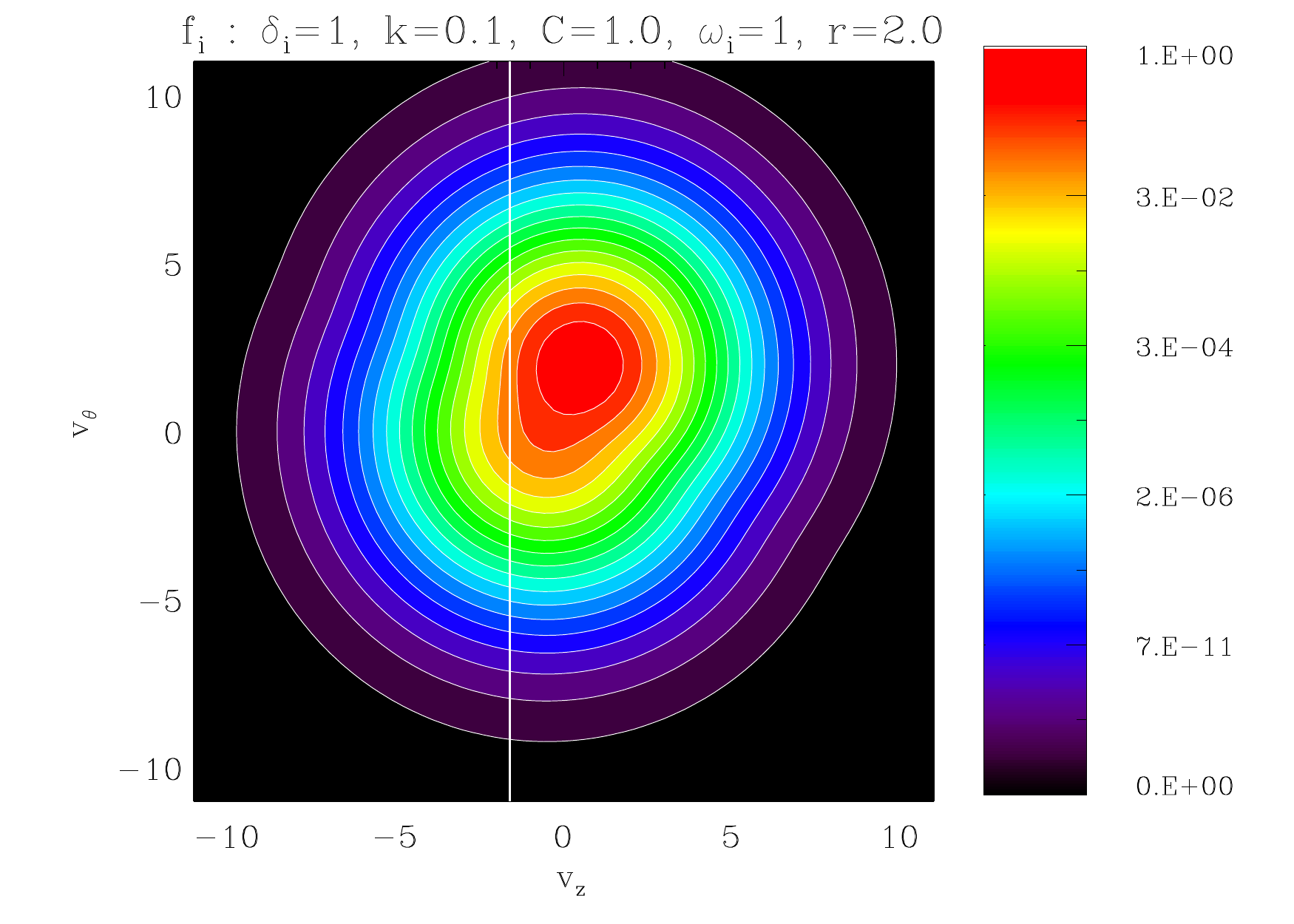}
        \caption{\small $(\tilde{\omega}_i,\tilde{r},C_i)=(1,2,1)$}
        \label{fig:4f}
    \end{subfigure}
        \begin{subfigure}[b]{0.45\textwidth}
        \includegraphics[width=\textwidth]{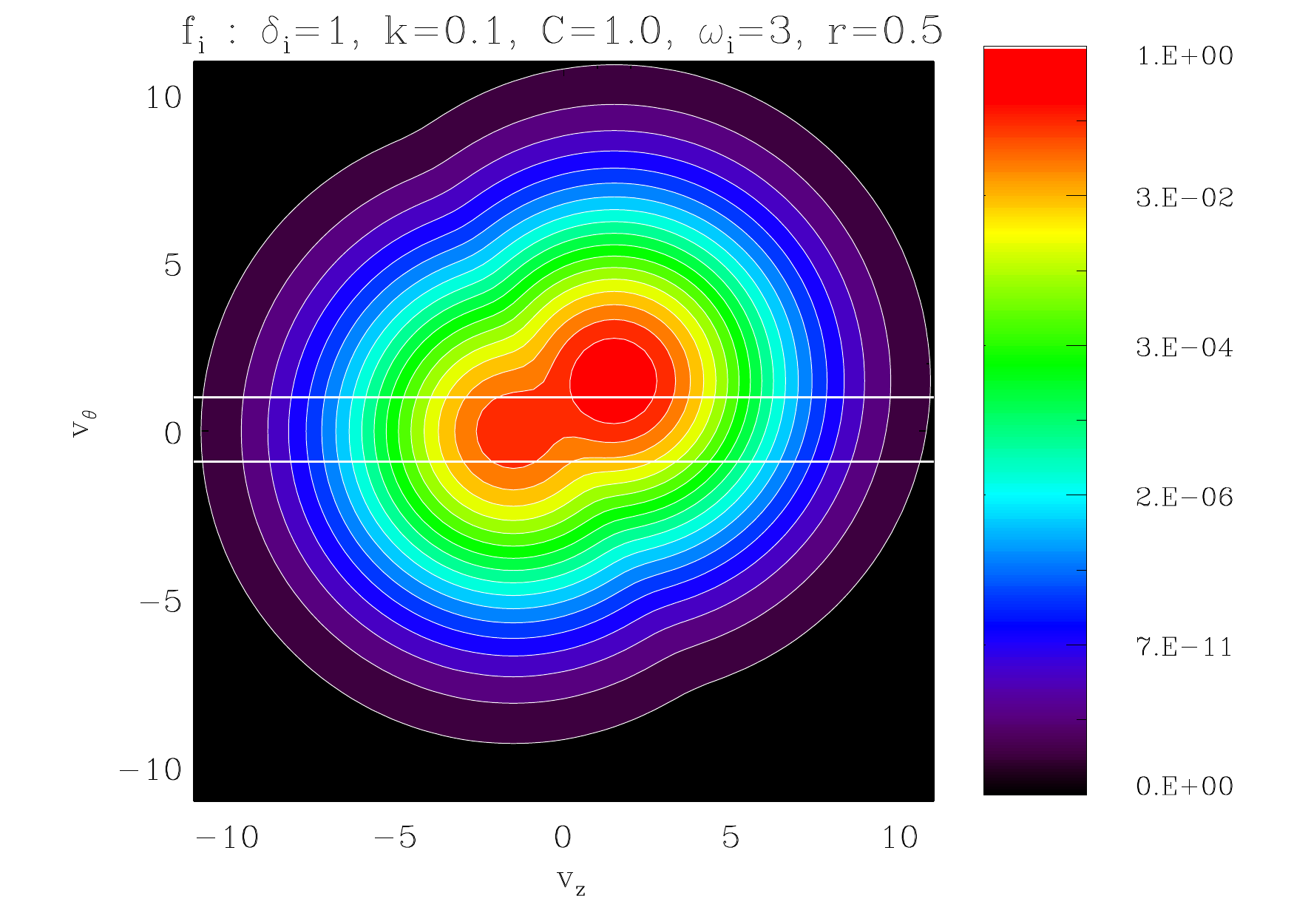}
        \caption{\small $(\tilde{\omega}_i,\tilde{r},C_i)=(3,0.5,1)$}
        \label{fig:4g}
    \end{subfigure}
     \begin{subfigure}[b]{0.45\textwidth}
        \includegraphics[width=\textwidth]{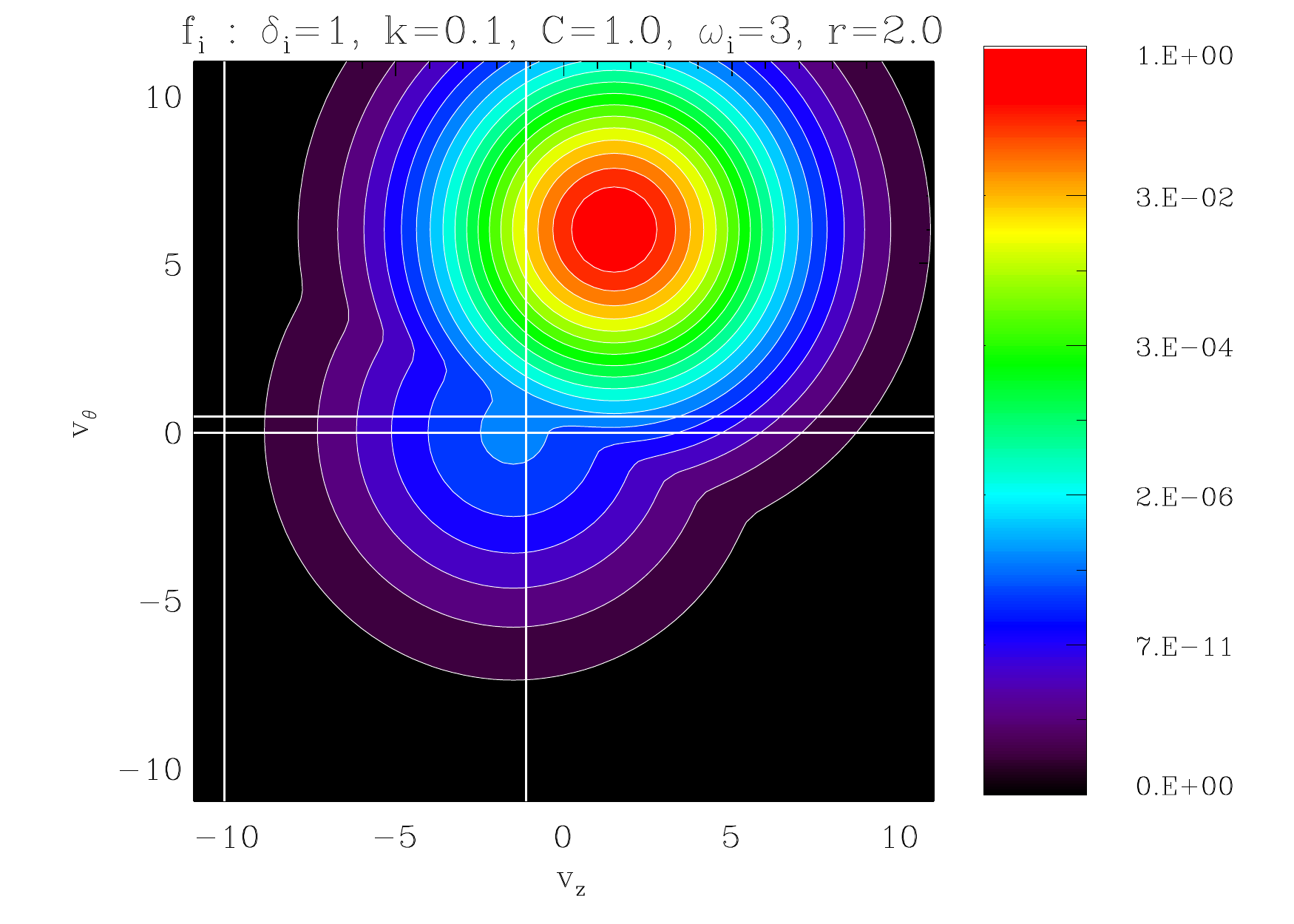}
        \caption{\small $(\tilde{\omega}_i,\tilde{r},C_i)=(3,2,1)$}
        \label{fig:4h}
    \end{subfigure}
    \caption{\small Contour plots of the $f_i$ in $(\tilde{v}_z,\tilde{v}_{\theta})$ space for an equilibrium with field reversal ($k=0.1<0.5$), for a variety of parameters ($\tilde{\omega}_i,\tilde{r}, C_i$) and $\delta_i=1$. The white horizontal/vertical lines indicate the regions in which multiple maxima in either the $\tilde{v}_z$ or $\tilde{v}_{z}$ directions can occur, if at all. A single line indicates that the `region' is a line.   }\label{fig:4}
\end{figure}

\begin{figure}
    \centering
    \begin{subfigure}[b]{0.45\textwidth}
        \includegraphics[width=\textwidth]{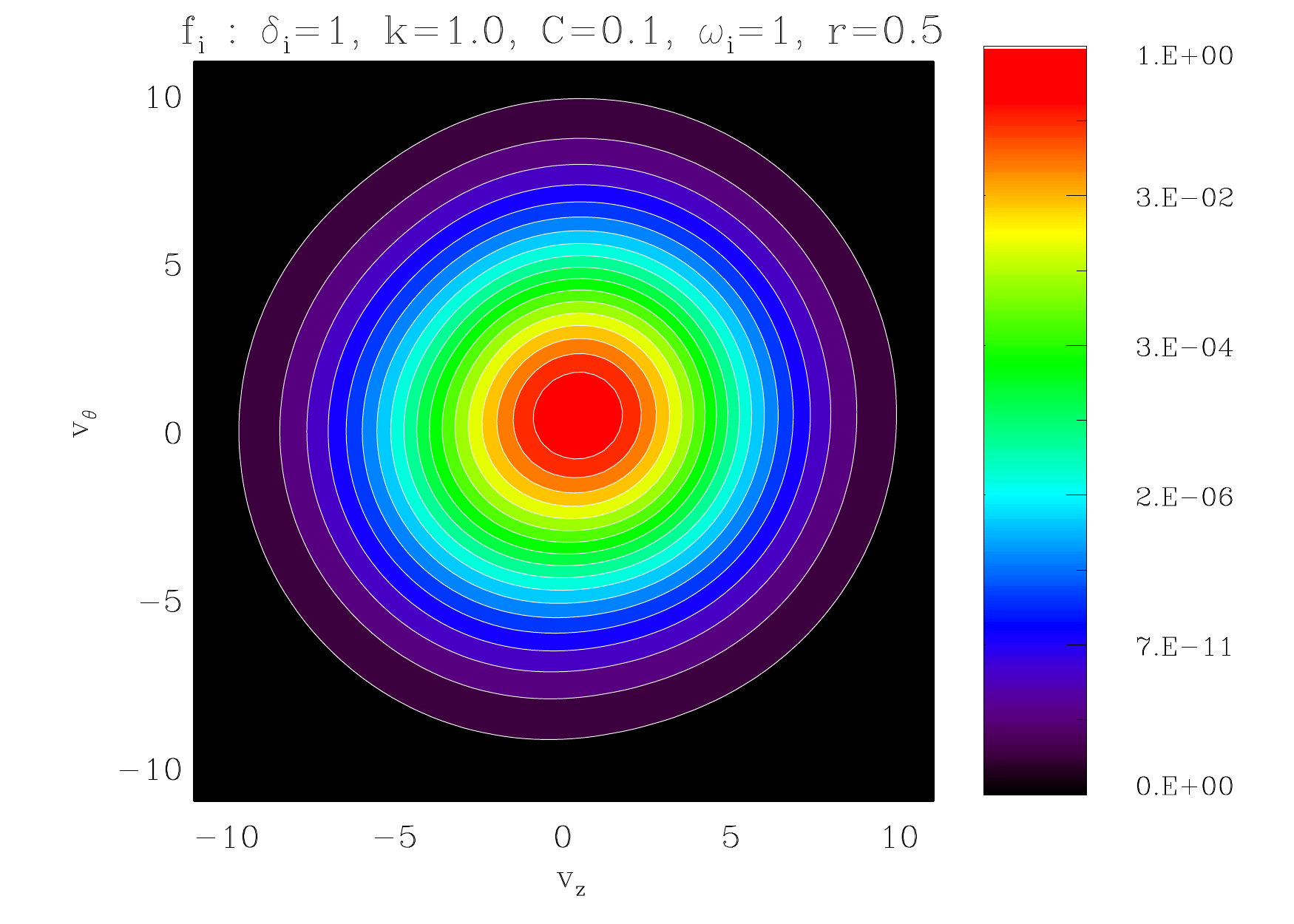}
        \caption{\small $(\tilde{\omega}_i,\tilde{r},C_i)=(1,0.5,0.1)$}
        \label{fig:5a}
    \end{subfigure}
       \begin{subfigure}[b]{0.45\textwidth}
        \includegraphics[width=\textwidth]{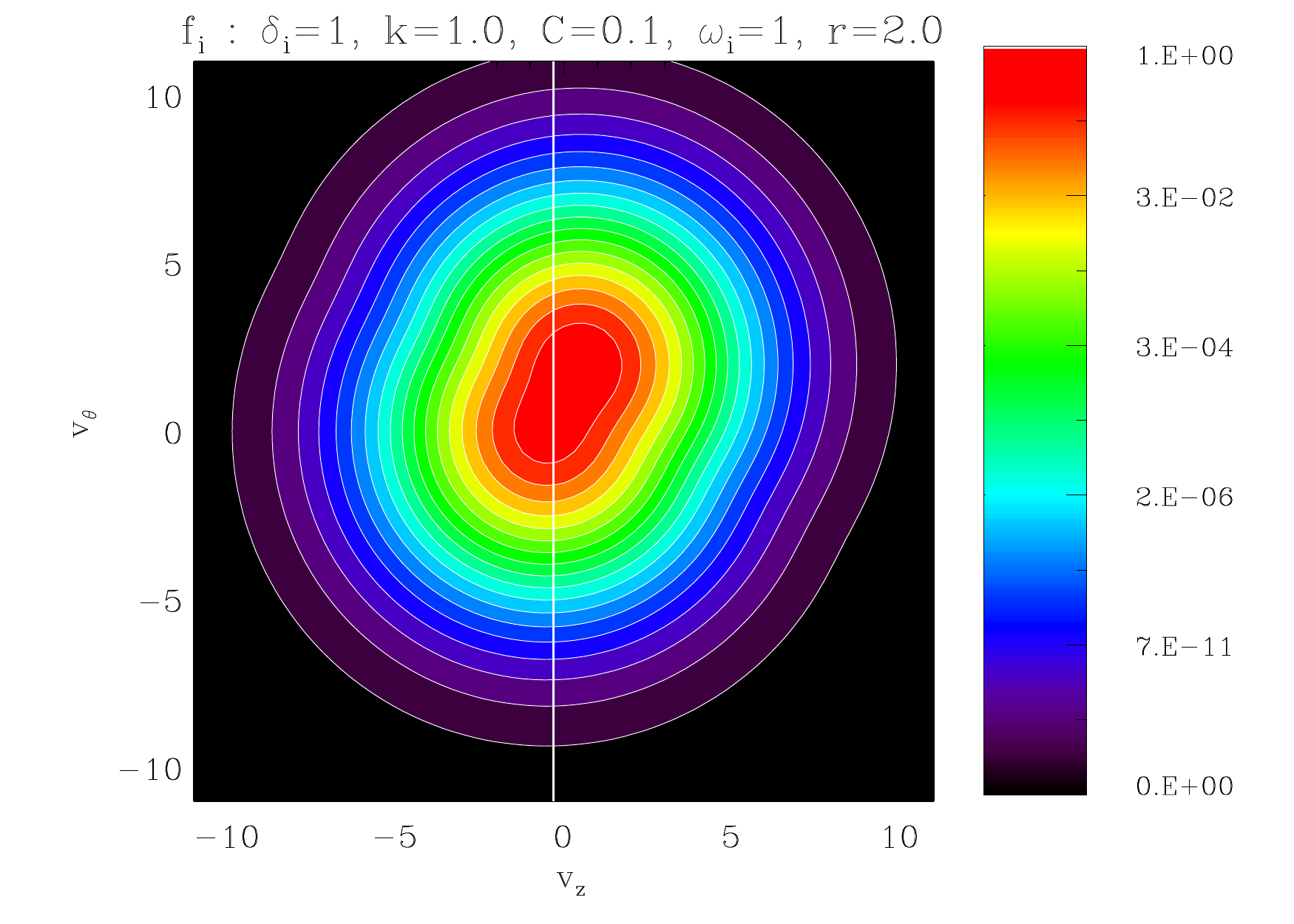}
        \caption{\small $(\tilde{\omega}_i,\tilde{r},C_i)=(1,2,0.1)$}
        \label{fig:5b}
    \end{subfigure}
        \begin{subfigure}[b]{0.45\textwidth}
        \includegraphics[width=\textwidth]{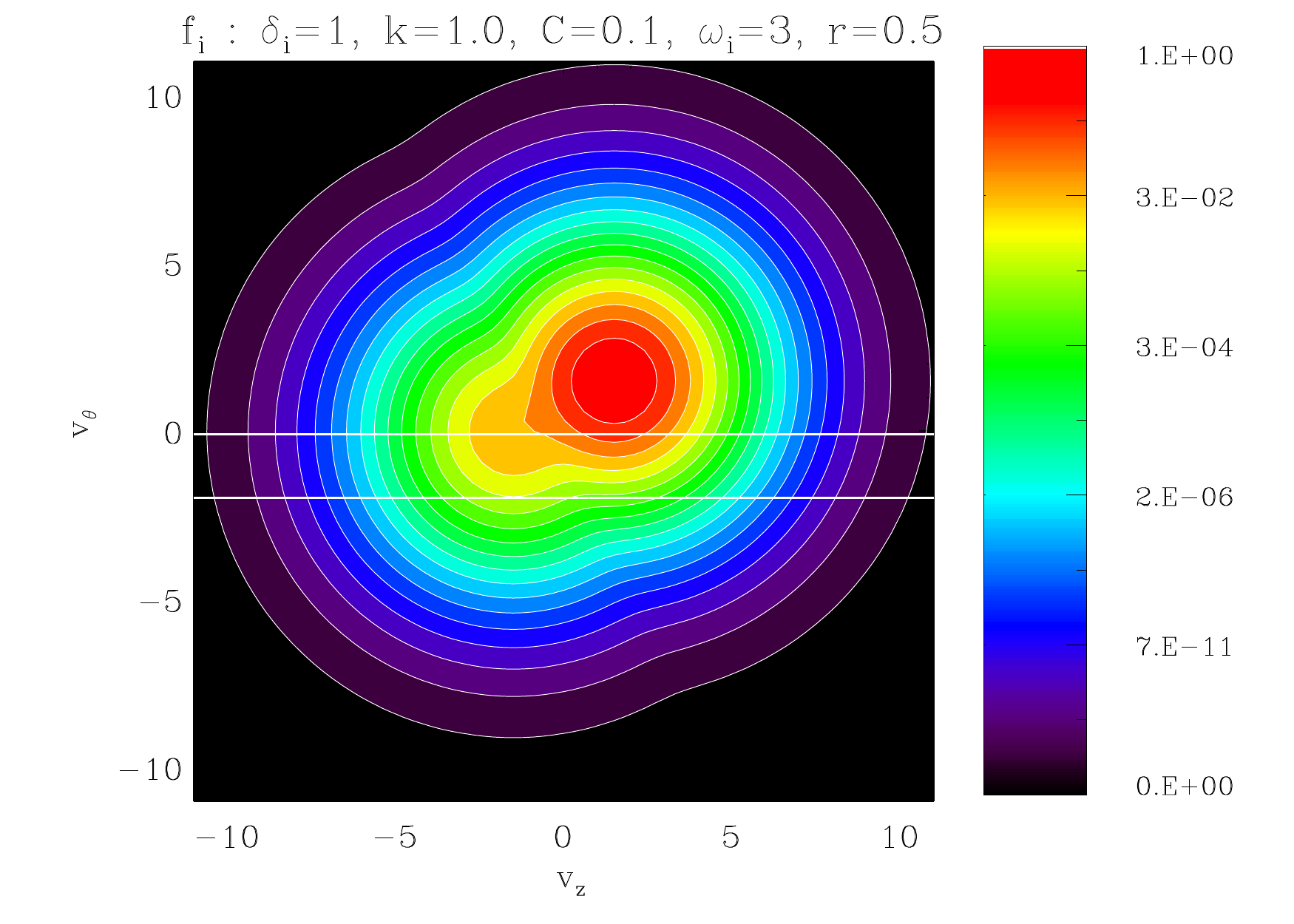}
        \caption{\small $(\tilde{\omega}_i,\tilde{r},C_i)=(3,0.5,0.1)$}
        \label{fig:5c}
    \end{subfigure}
    \begin{subfigure}[b]{0.45\textwidth}
        \includegraphics[width=\textwidth]{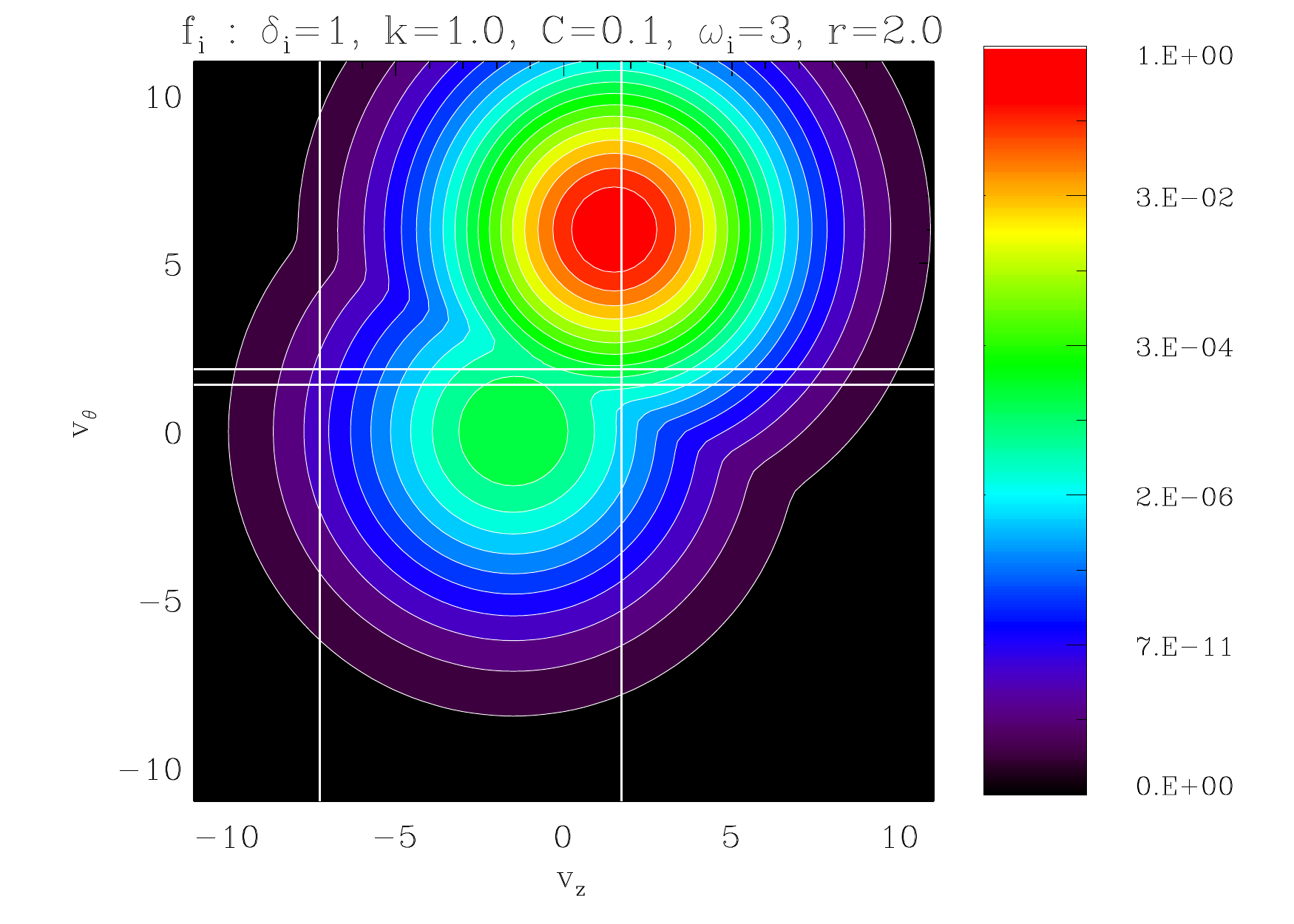}
        \caption{\small $(\tilde{\omega}_i,\tilde{r},C_i)=(3,2,0.1)$}
        \label{fig:5d}
    \end{subfigure}
        \begin{subfigure}[b]{0.45\textwidth}
        \includegraphics[width=\textwidth]{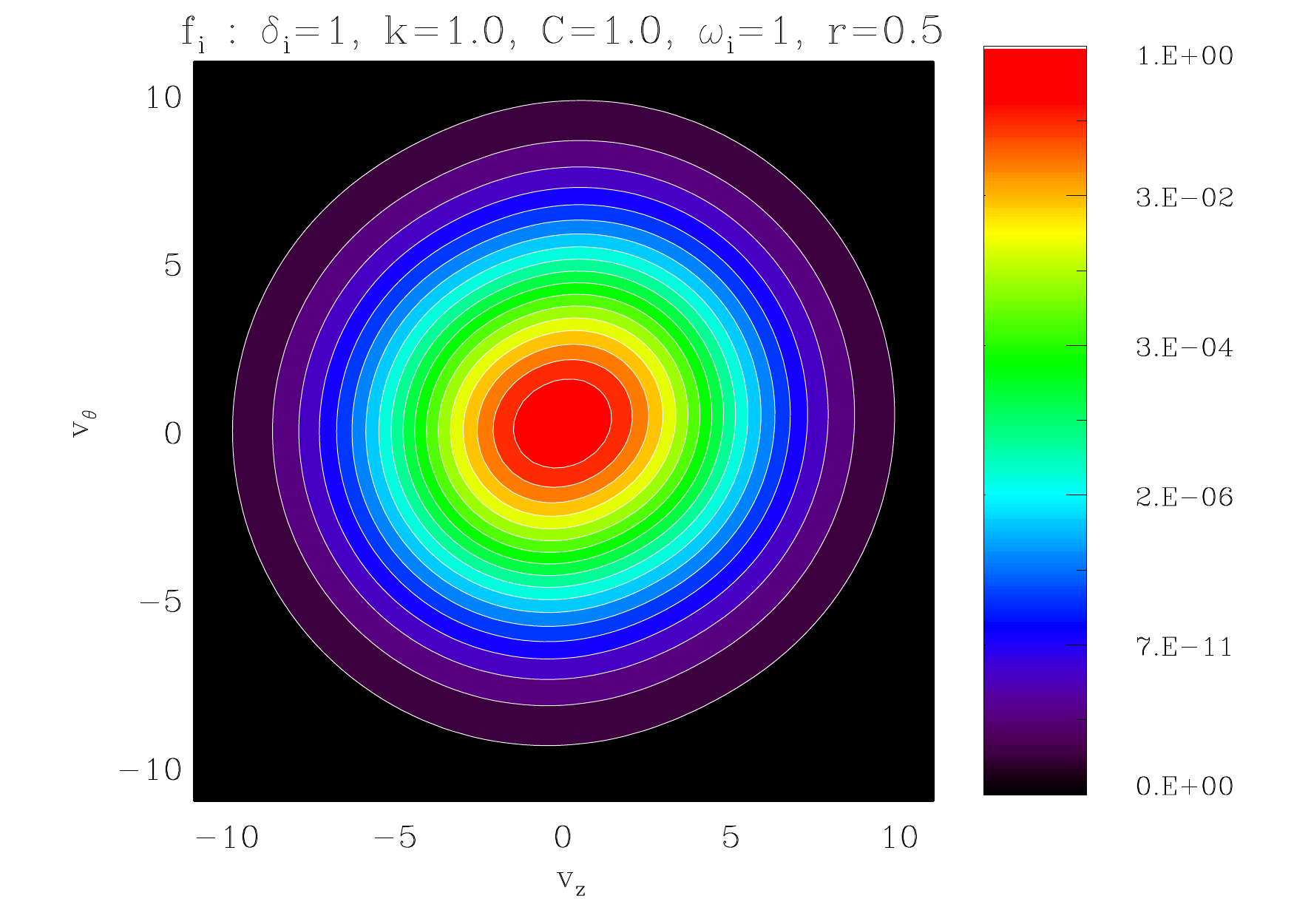}
        \caption{\small $(\tilde{\omega}_i,\tilde{r},C_i)=(1,0.5,1)$}
        \label{fig:5e}
    \end{subfigure}
    \begin{subfigure}[b]{0.45\textwidth}
        \includegraphics[width=\textwidth]{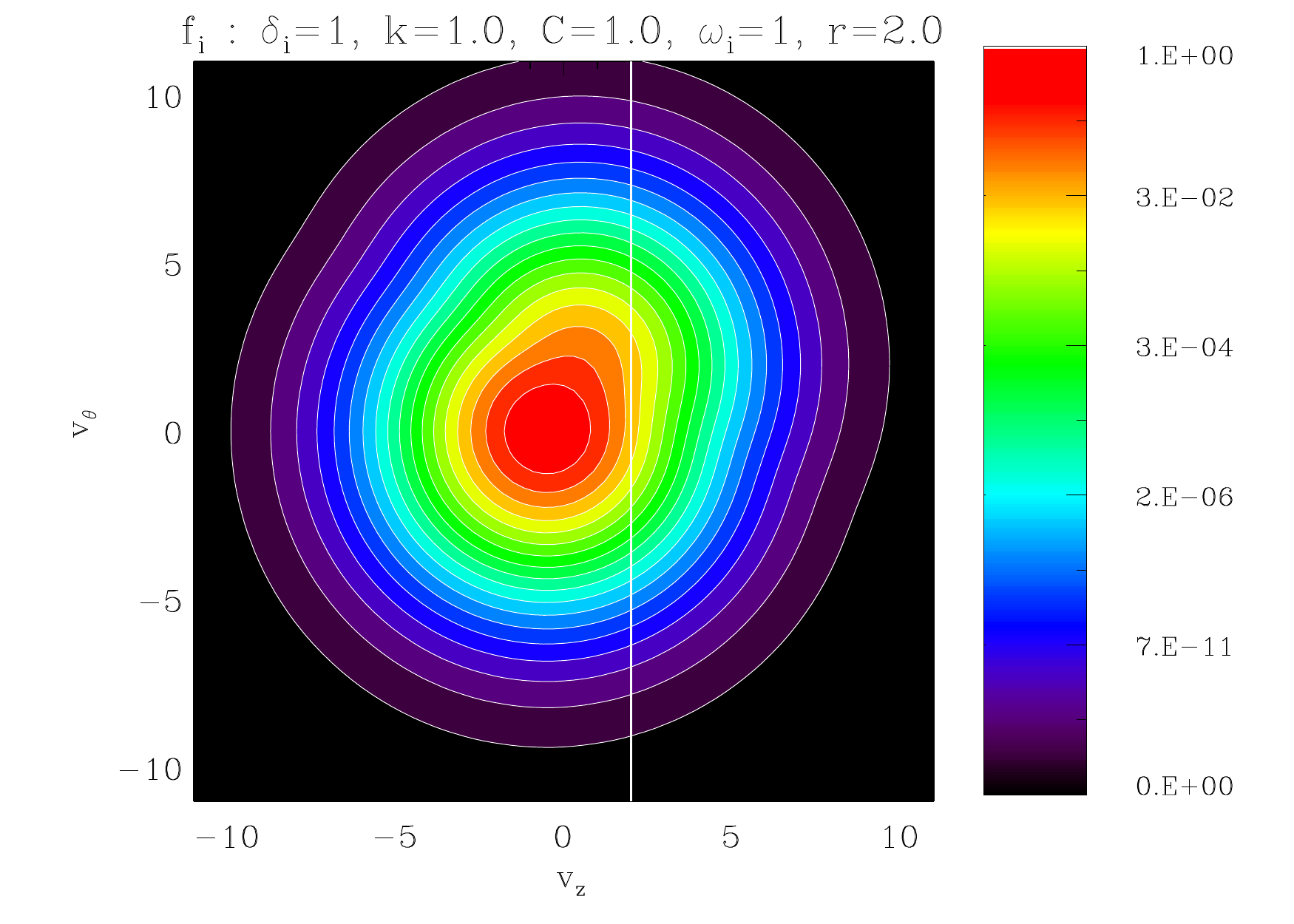}
        \caption{\small $(\tilde{\omega}_i,\tilde{r},C_i)=(1,2,1)$}
        \label{fig:5f}
    \end{subfigure}
        \begin{subfigure}[b]{0.45\textwidth}
        \includegraphics[width=\textwidth]{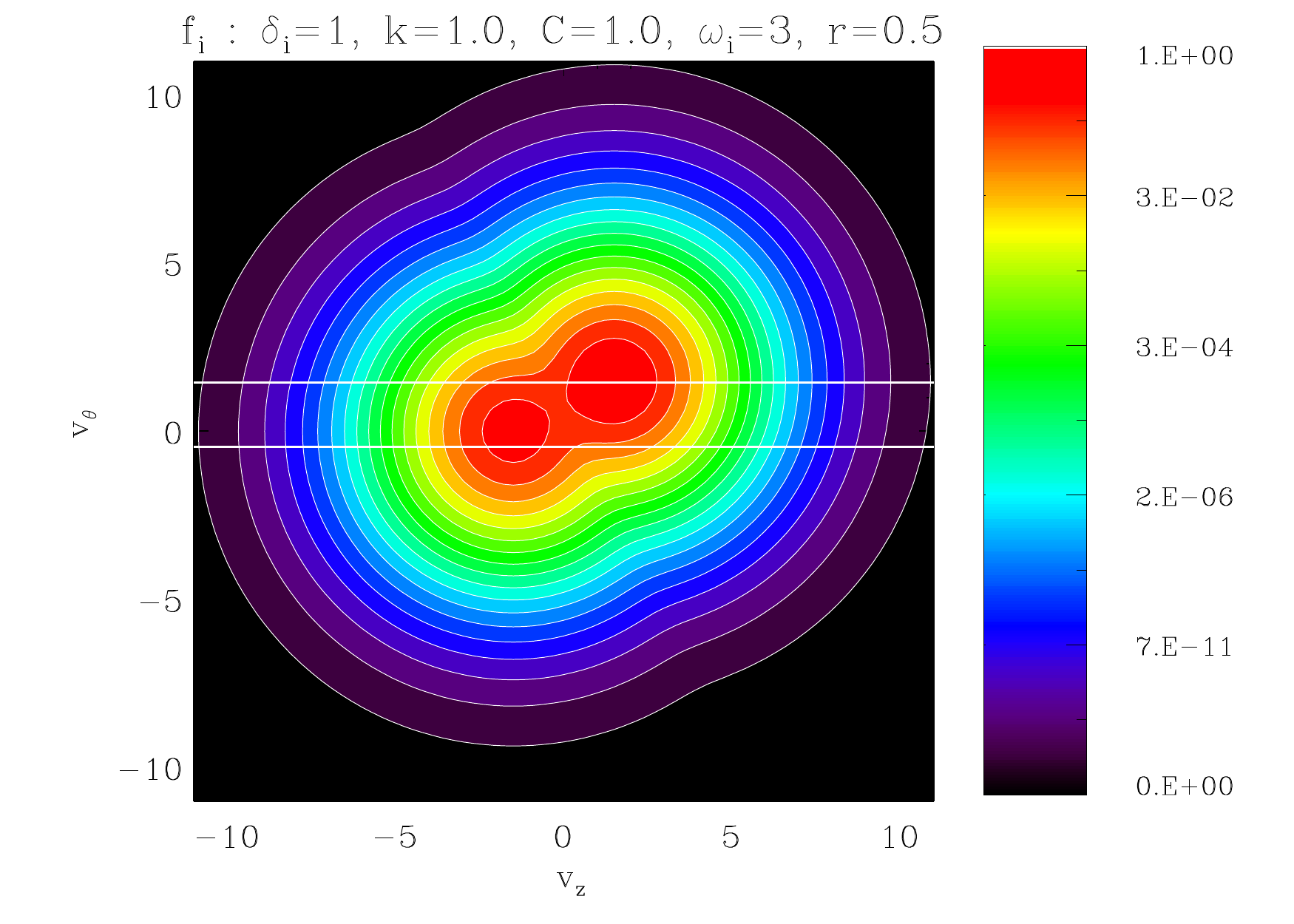}
        \caption{\small $(\tilde{\omega}_i,\tilde{r},C_i)=(3,0.5,1)$}
        \label{fig:5g}
    \end{subfigure}
     \begin{subfigure}[b]{0.45\textwidth}
        \includegraphics[width=\textwidth]{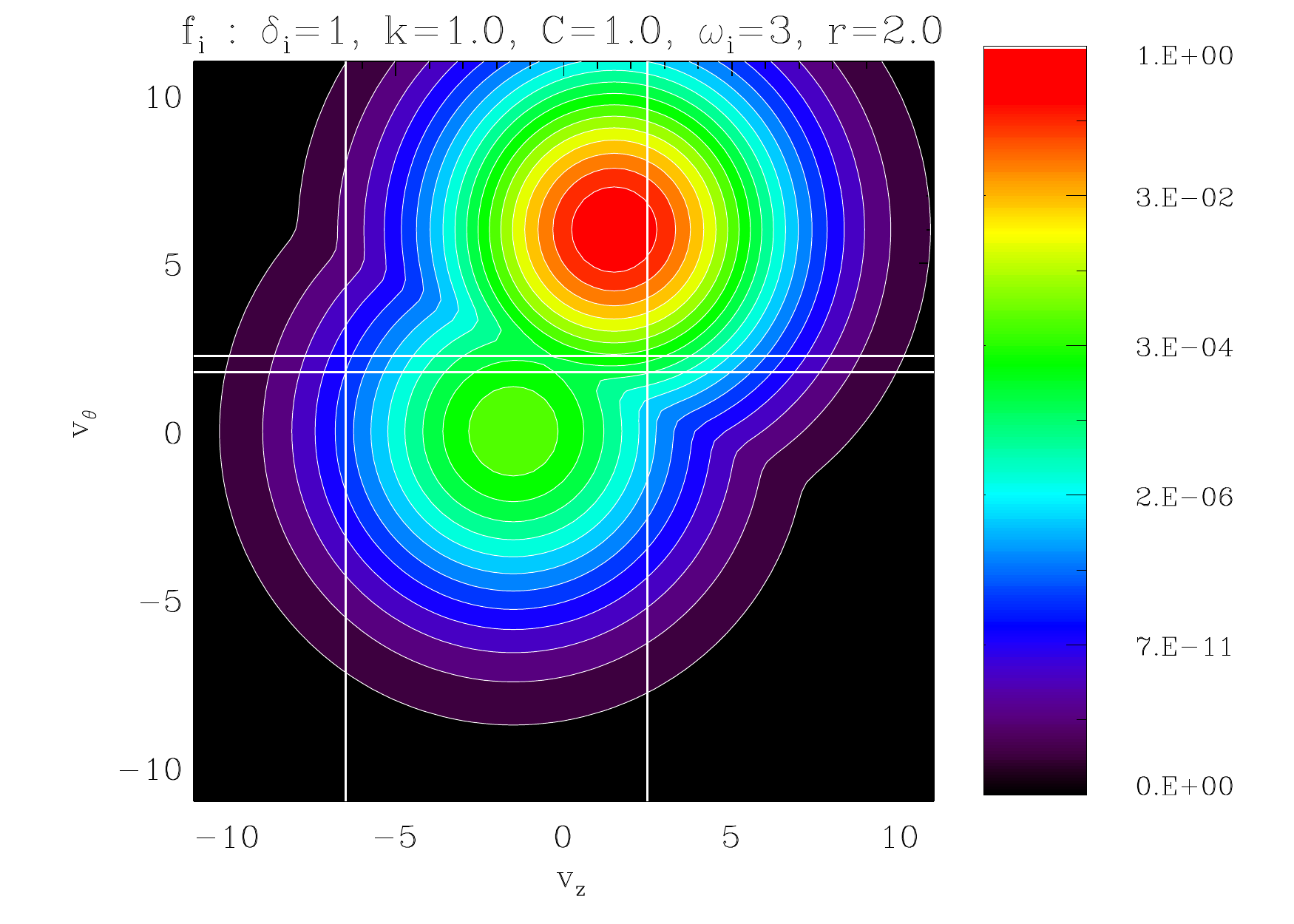}
        \caption{\small $(\tilde{\omega}_i,\tilde{r},C_i)=(3,2,1)$}
        \label{fig:5h}
    \end{subfigure}
    \caption{\small Contour plots of $f_i$ in $(\tilde{v}_z,\tilde{v}_{\theta})$ space for an equilibrium without field reversal ($k=1>0.5$), for a variety of parameters ($\tilde{\omega}_i,\tilde{r},C_i$) and $\delta_i=1$. The white horizontal/vertical lines indicate the regions in which multiple maxima in either the $\tilde{v}_z$ or $\tilde{v}_{z}$ directions can occur, if at all. A single line indicates that the `region' is a line.  }\label{fig:5}
\end{figure}

\begin{figure}
    \centering
    \begin{subfigure}[b]{0.45\textwidth}
        \includegraphics[width=\textwidth]{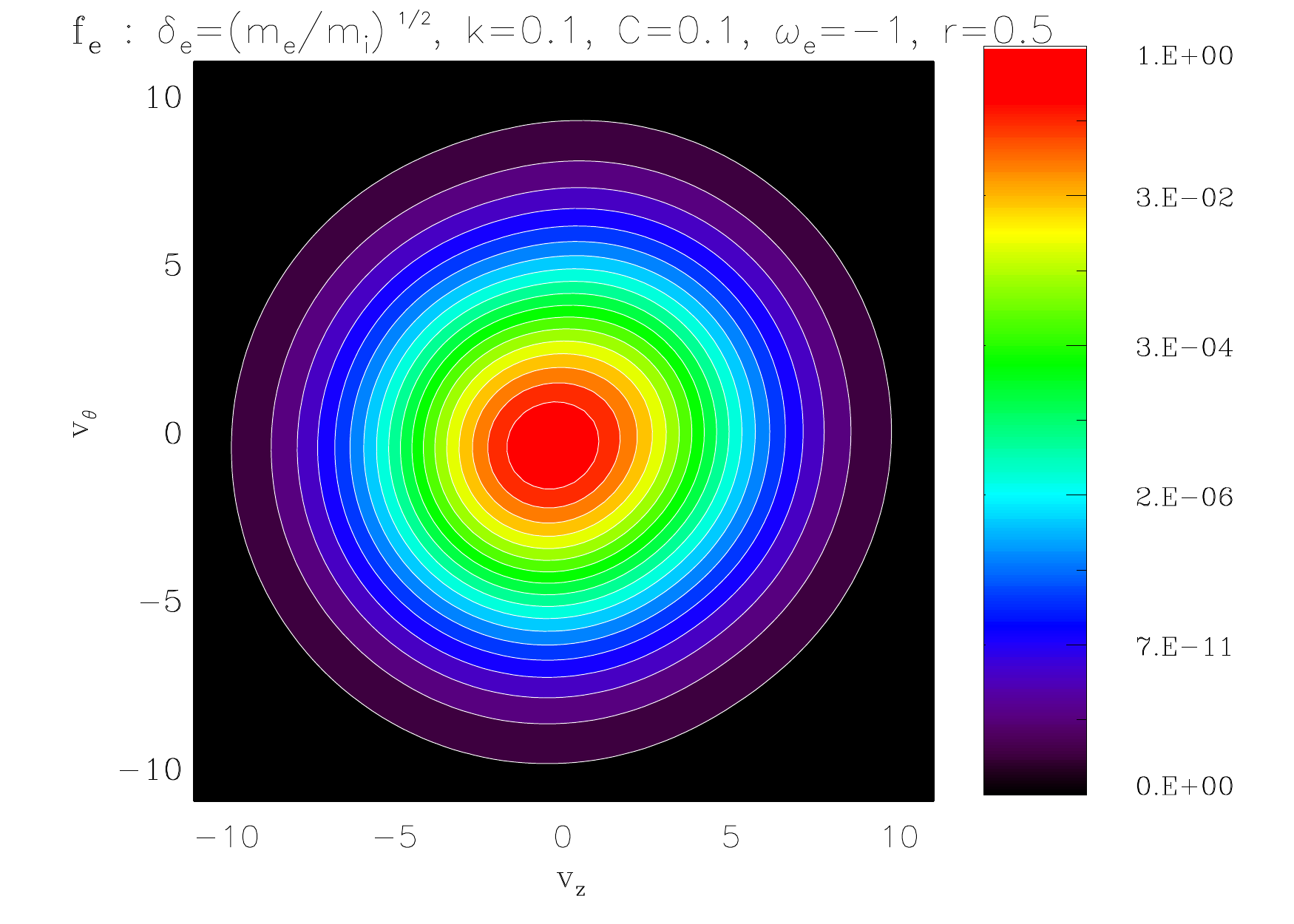}
        \caption{\small $(\tilde{\omega}_e,\tilde{r},C_e)=(-1,0.5,0.1)$}
        \label{fig:}
    \end{subfigure}
       \begin{subfigure}[b]{0.45\textwidth}
        \includegraphics[width=\textwidth]{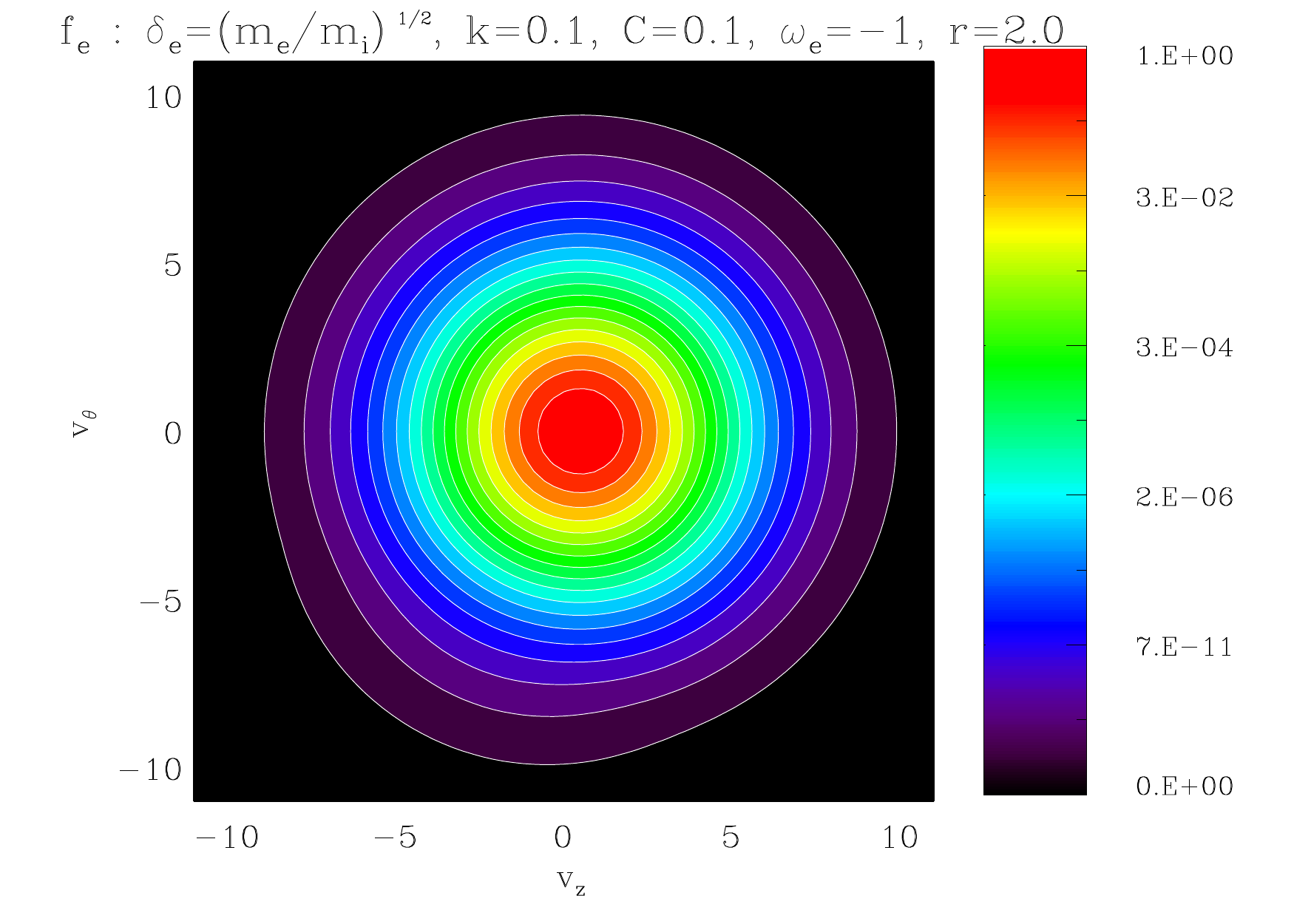}
        \caption{\small $(\tilde{\omega}_e,\tilde{r},C_e)=(-1,2,0.1)$}
        \label{fig:}
    \end{subfigure}
        \begin{subfigure}[b]{0.45\textwidth}
        \includegraphics[width=\textwidth]{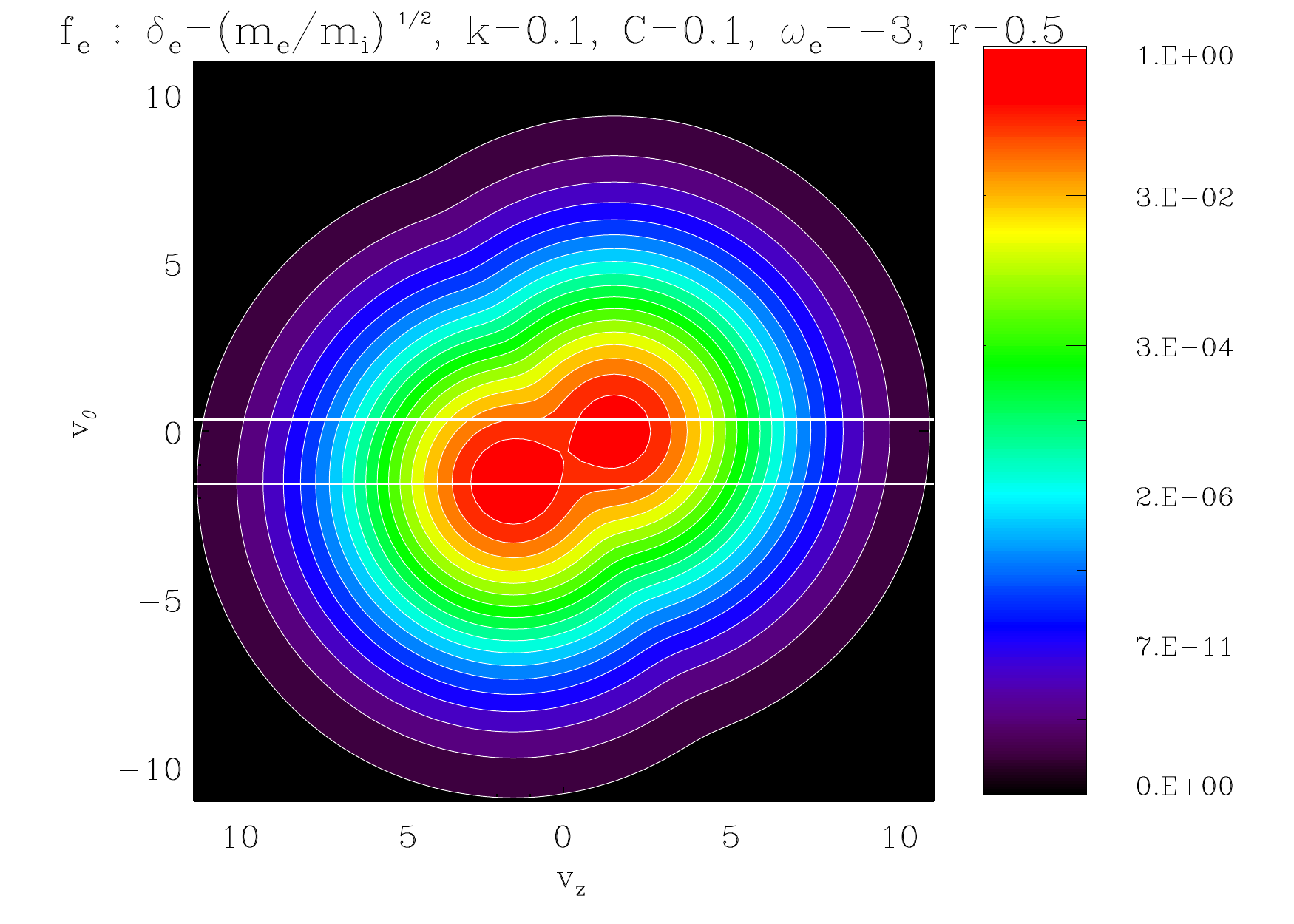}
        \caption{\small $(\tilde{\omega}_e,\tilde{r},C_e)=(-3,0.5,0.1)$}
        \label{fig:}
    \end{subfigure}
    \begin{subfigure}[b]{0.45\textwidth}
        \includegraphics[width=\textwidth]{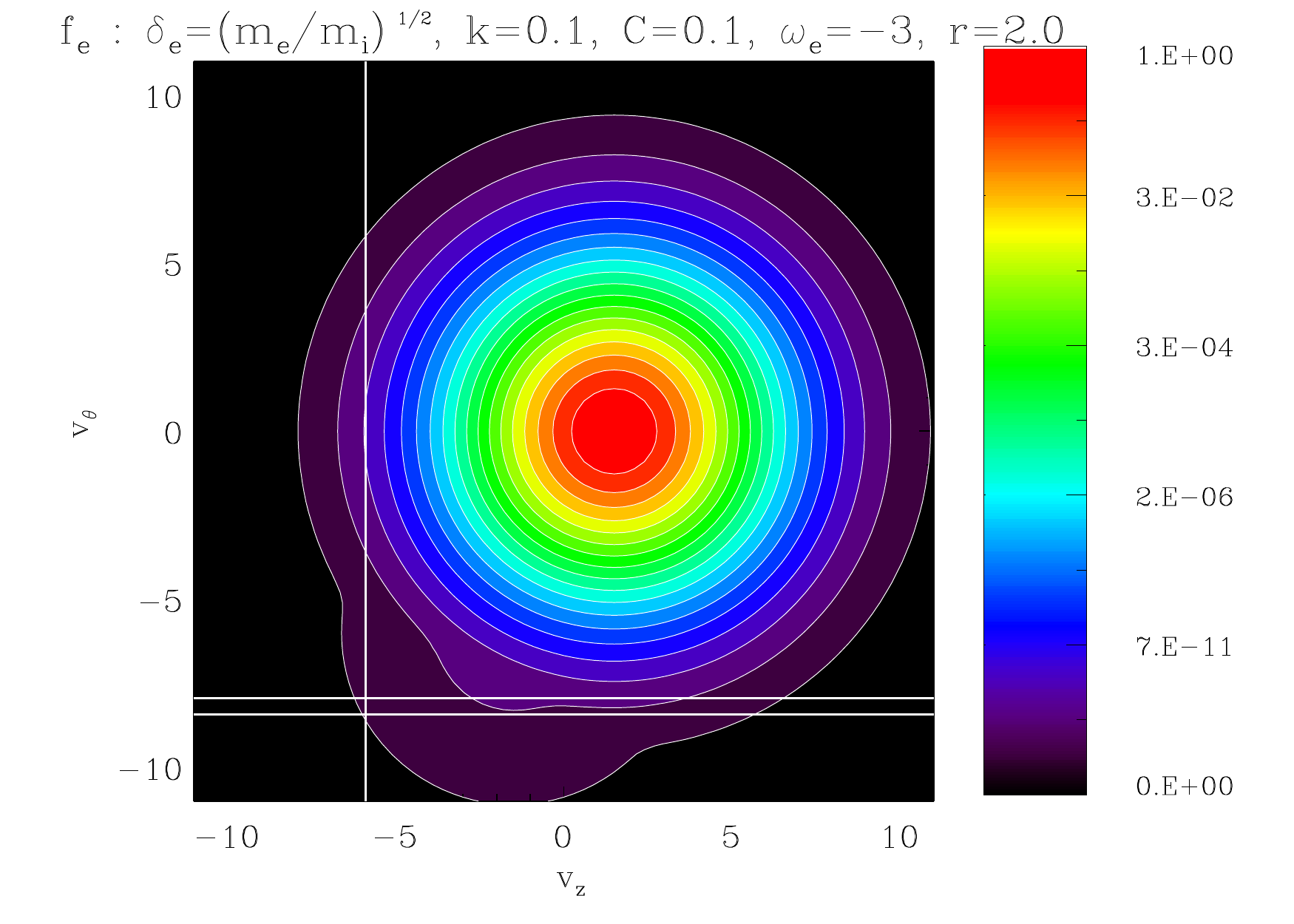}
        \caption{\small $(\tilde{\omega}_e,\tilde{r},C_e)=(-3,2,0.1)$}
        \label{fig:}
    \end{subfigure}
        \begin{subfigure}[b]{0.45\textwidth}
        \includegraphics[width=\textwidth]{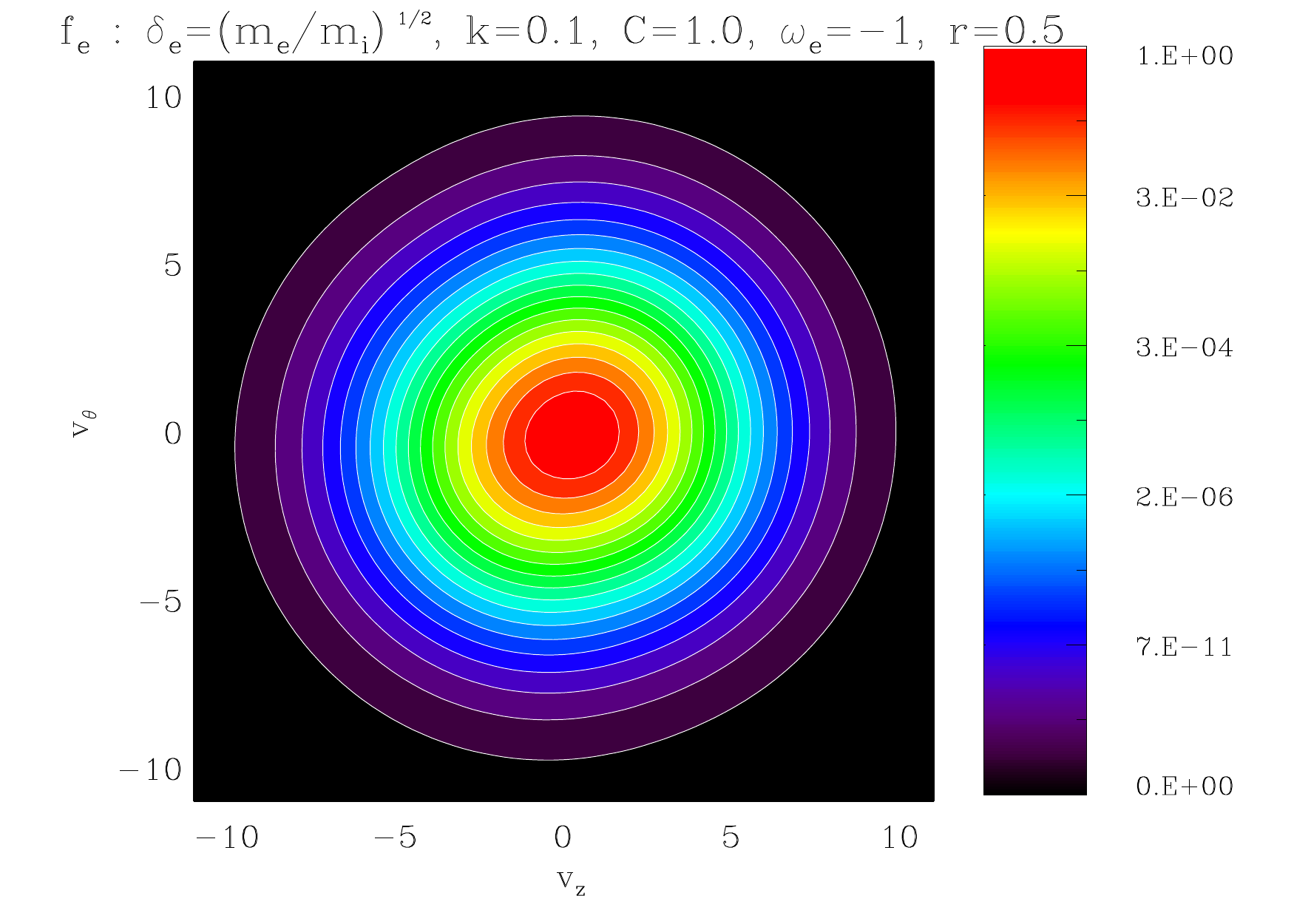}
        \caption{\small $(\tilde{\omega}_e,\tilde{r},C_e)=(-1,0.5,1)$}
        \label{fig:}
    \end{subfigure}
    \begin{subfigure}[b]{0.45\textwidth}
        \includegraphics[width=\textwidth]{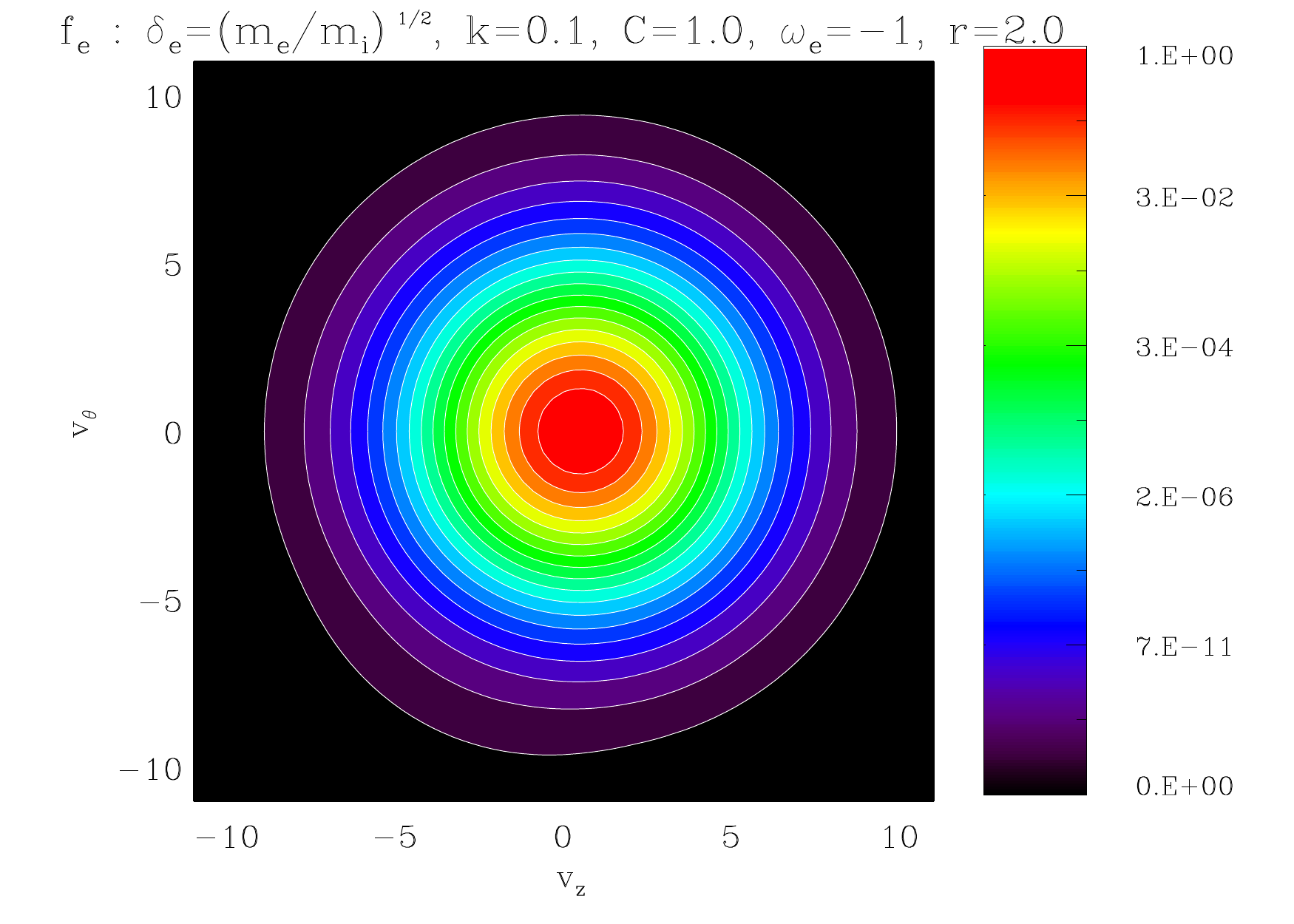}
        \caption{\small $(\tilde{\omega}_e,\tilde{r},C_e)=(-1,2,1)$}
        \label{fig:}
    \end{subfigure}
        \begin{subfigure}[b]{0.45\textwidth}
        \includegraphics[width=\textwidth]{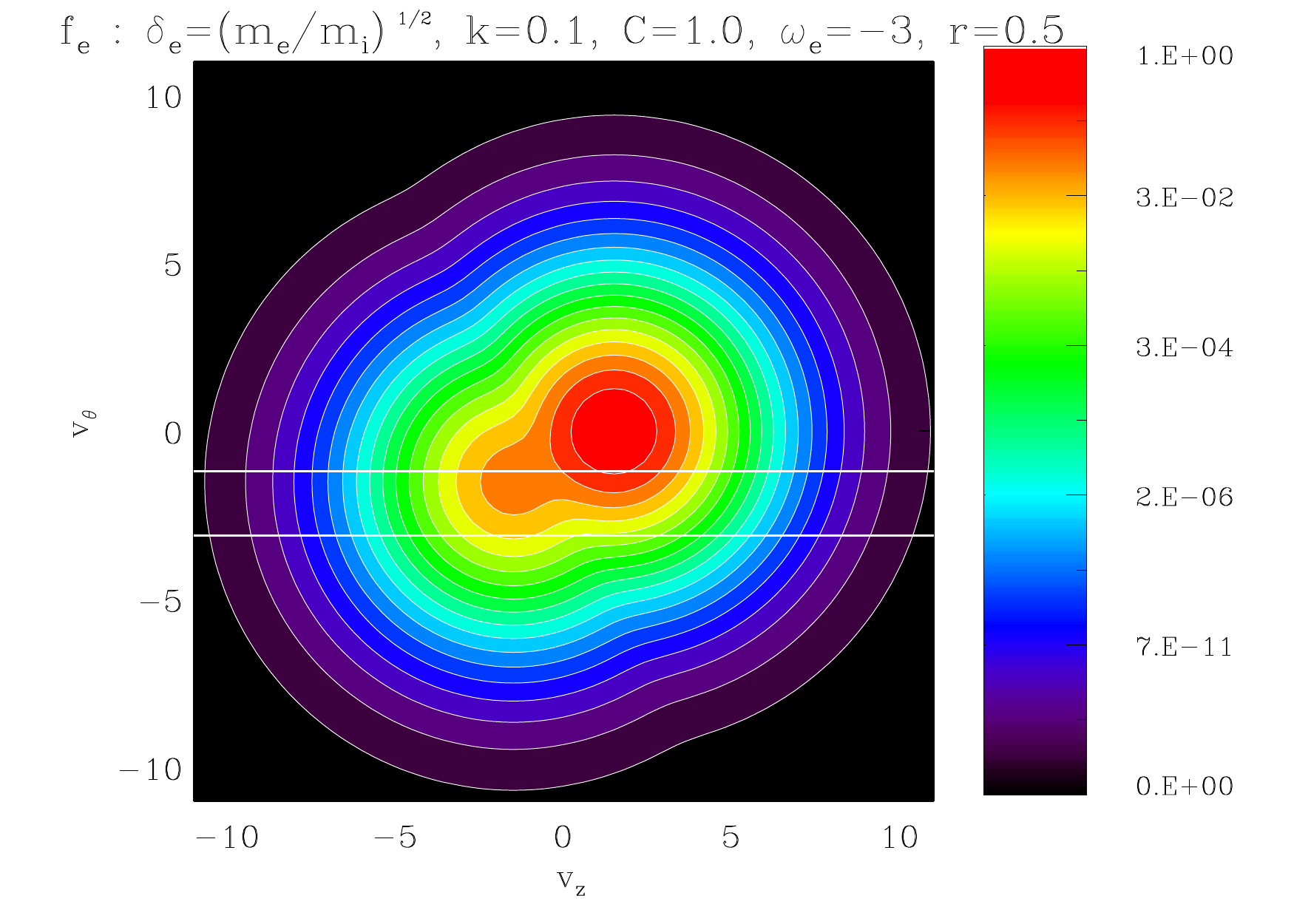}
        \caption{\small $(\tilde{\omega}_e,\tilde{r},C_e)=(-3,0.5,1)$}
        \label{fig:}
    \end{subfigure}
     \begin{subfigure}[b]{0.45\textwidth}
        \includegraphics[width=\textwidth]{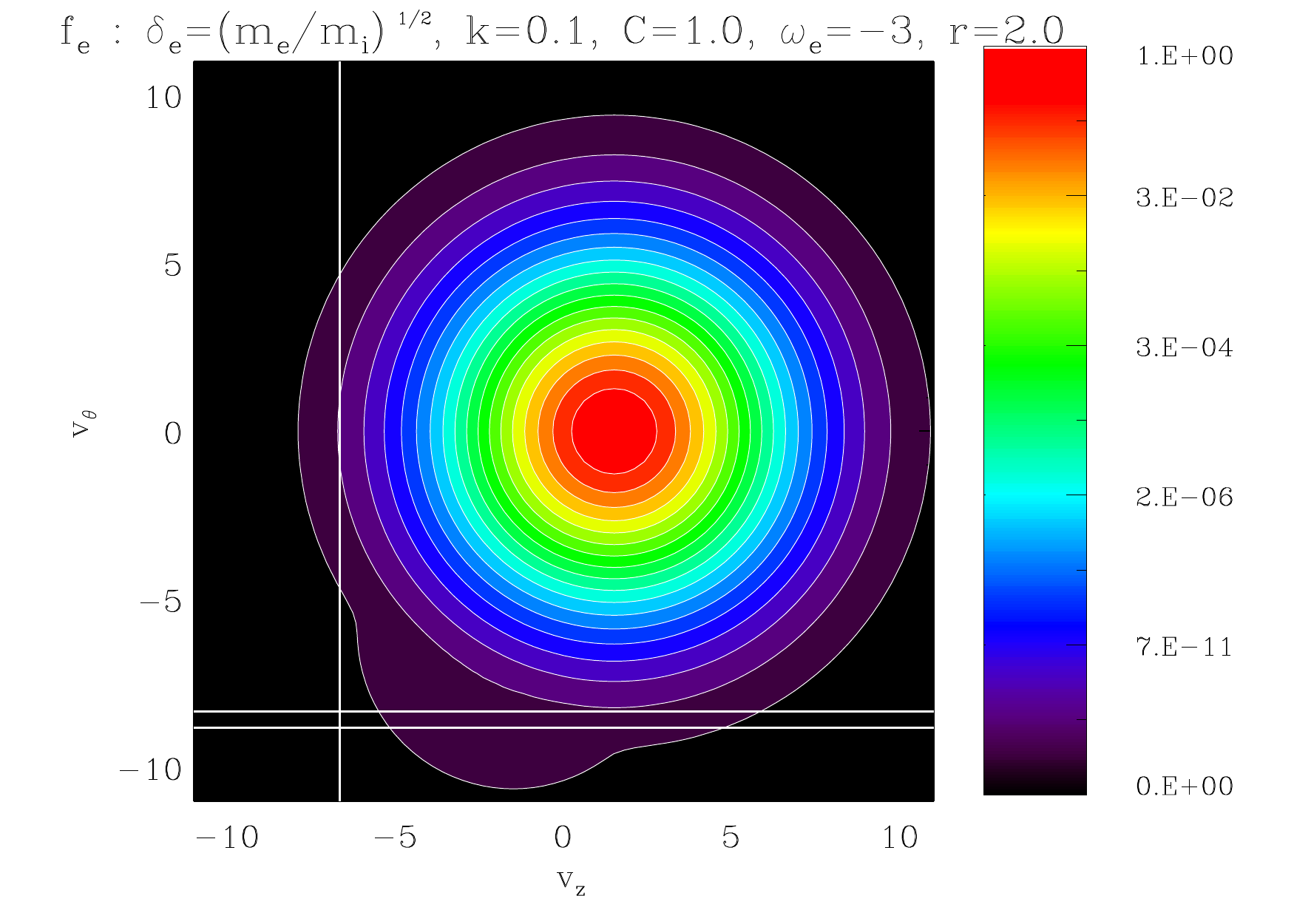}
        \caption{\small $(\tilde{\omega}_e,\tilde{r},C_e)=(-3,2,1)$}
        \label{fig:}
    \end{subfigure}
    \caption{\small     Contour plots of $f_e$ in $(\tilde{v}_z,\tilde{v}_{\theta})$ space for an equilibrium with field reversal ($k=0.1<0.5$), for a variety of parameters ($\tilde{\omega}_e,\tilde{r},C_e$) and $\delta_e\approx1/\sqrt{1836}$. The white horizontal/vertical lines indicate the regions in which multiple maxima in either the $\tilde{v}_z$ or $\tilde{v}_{z}$ directions can occur, if at all. A single line indicates that the `region' is a line.   }\label{fig:6}
\end{figure}

\begin{figure}
    \centering
    \begin{subfigure}[b]{0.45\textwidth}
        \includegraphics[width=\textwidth]{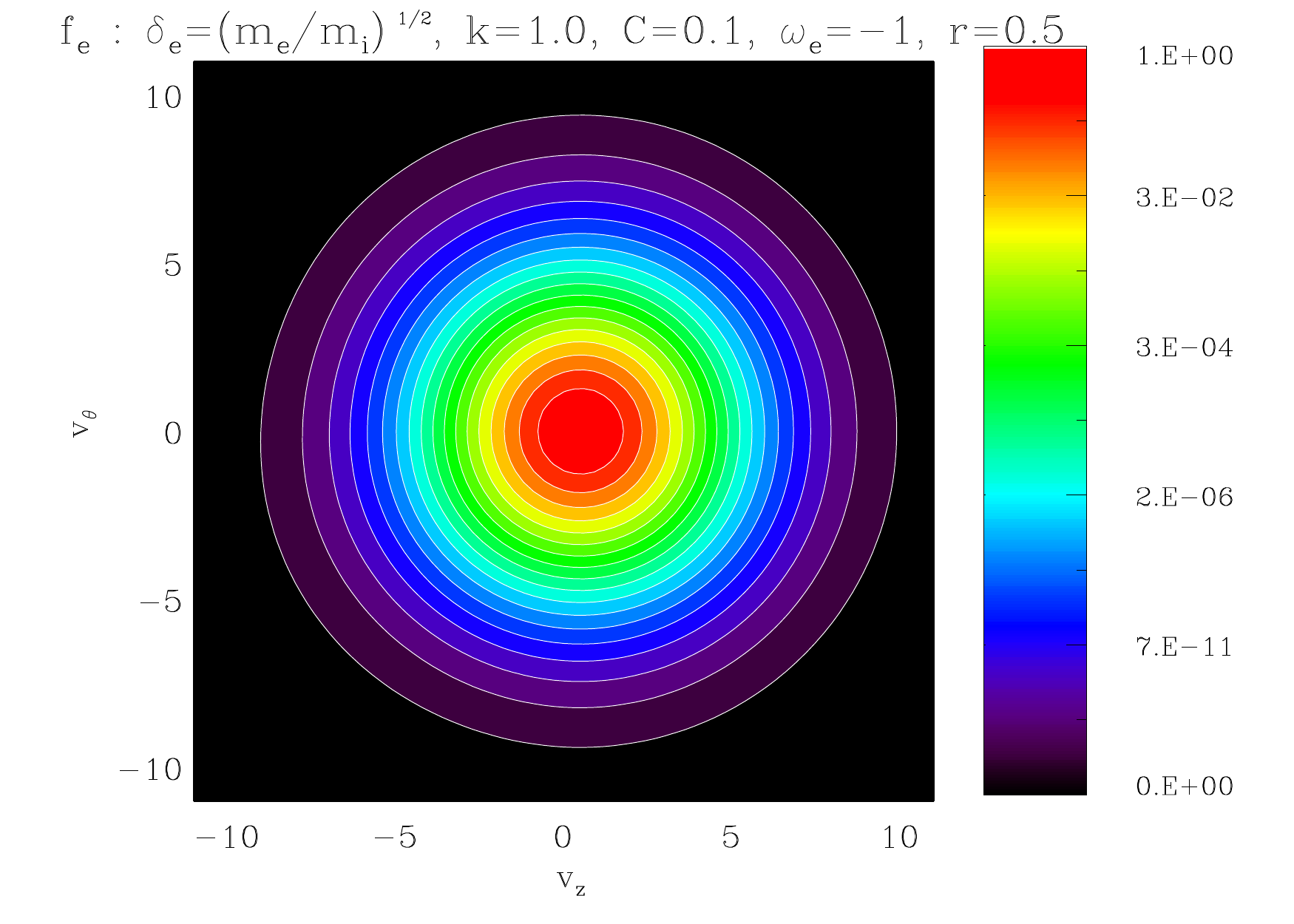}
        \caption{\small $(\tilde{\omega}_e,\tilde{r},C_e)=(-1,0.5,0.1)$}
        \label{fig:}
    \end{subfigure}
       \begin{subfigure}[b]{0.45\textwidth}
        \includegraphics[width=\textwidth]{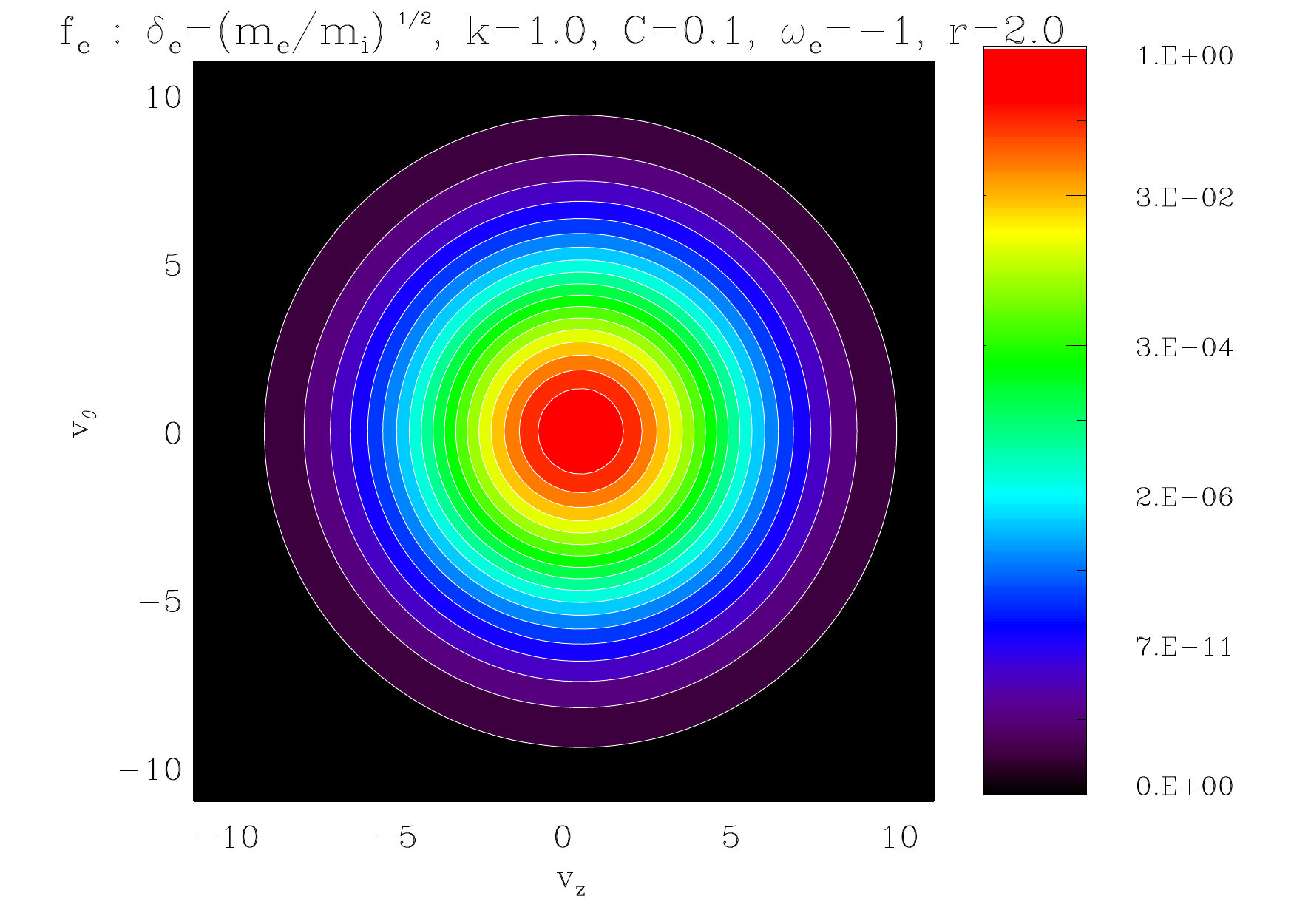}
        \caption{\small $(\tilde{\omega}_e,\tilde{r},C_e)=(-1,2,0.1)$}
        \label{fig:}
    \end{subfigure}
        \begin{subfigure}[b]{0.45\textwidth}
        \includegraphics[width=\textwidth]{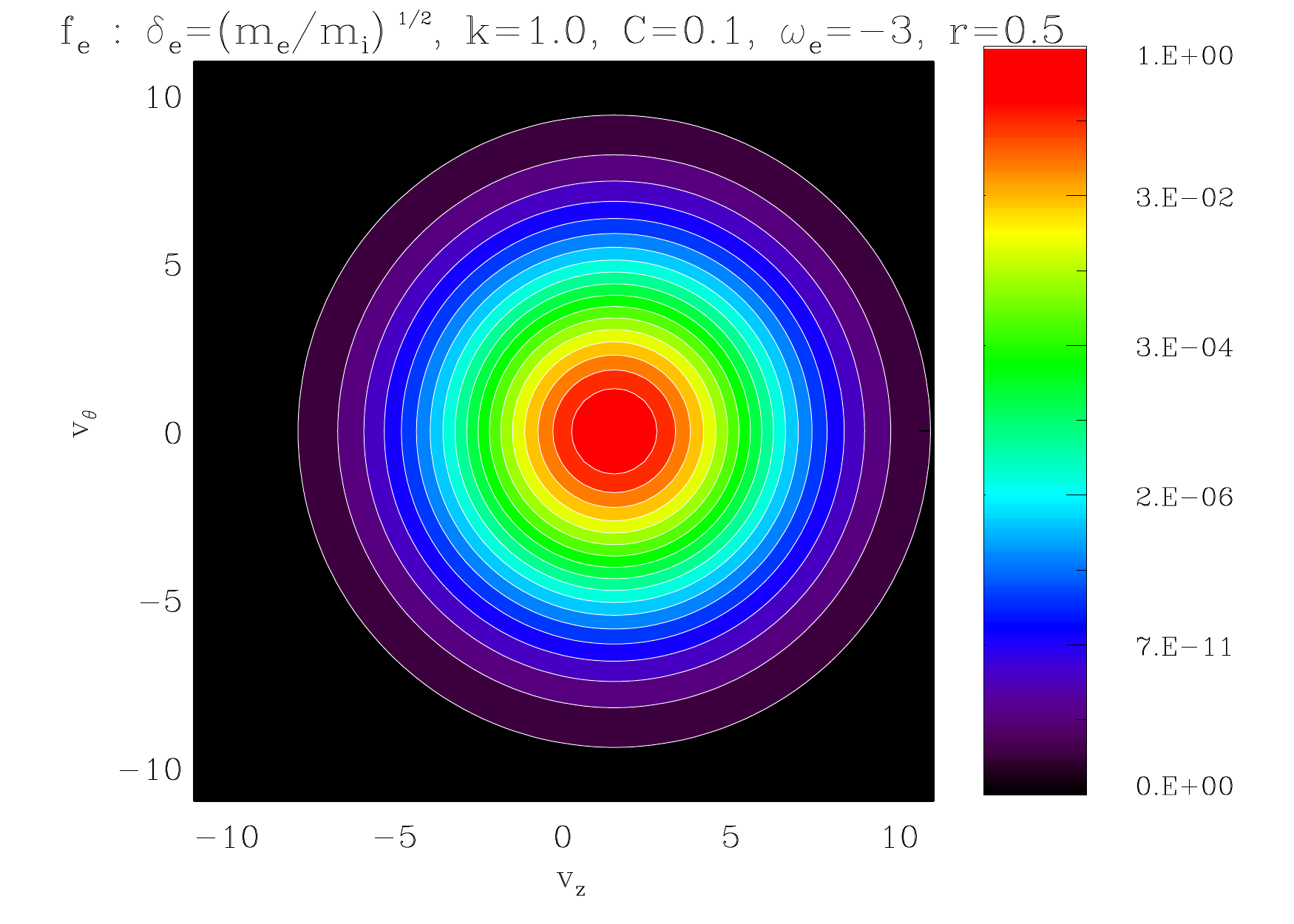}
        \caption{\small $(\tilde{\omega}_e,\tilde{r},C_e)=(-3,0.5,0.1)$}
        \label{fig:}
    \end{subfigure}
    \begin{subfigure}[b]{0.45\textwidth}
        \includegraphics[width=\textwidth]{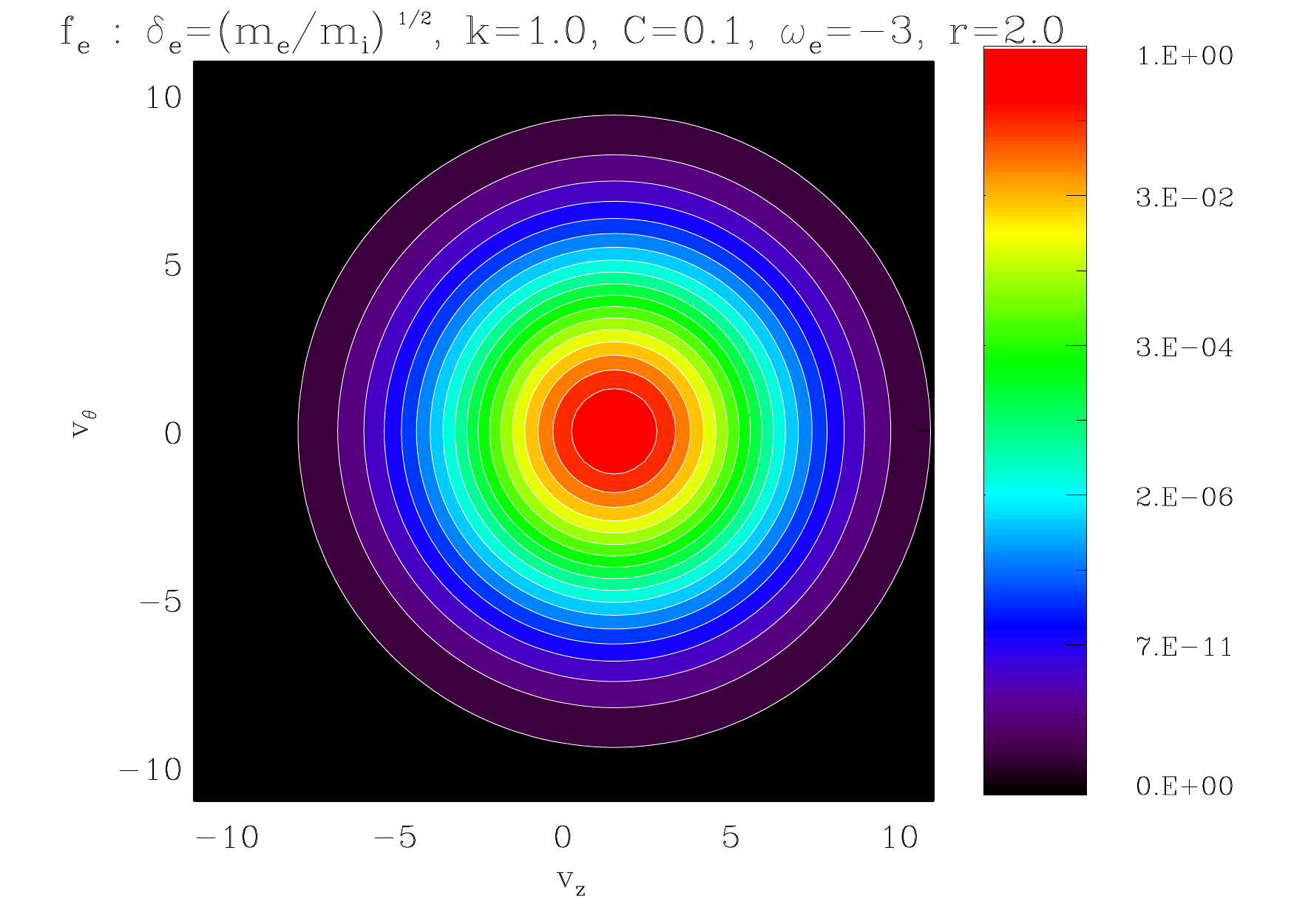}
        \caption{\small $(\tilde{\omega}_e,\tilde{r},C_e)=(-3,2,0.1)$}
        \label{fig:}
    \end{subfigure}
        \begin{subfigure}[b]{0.45\textwidth}
        \includegraphics[width=\textwidth]{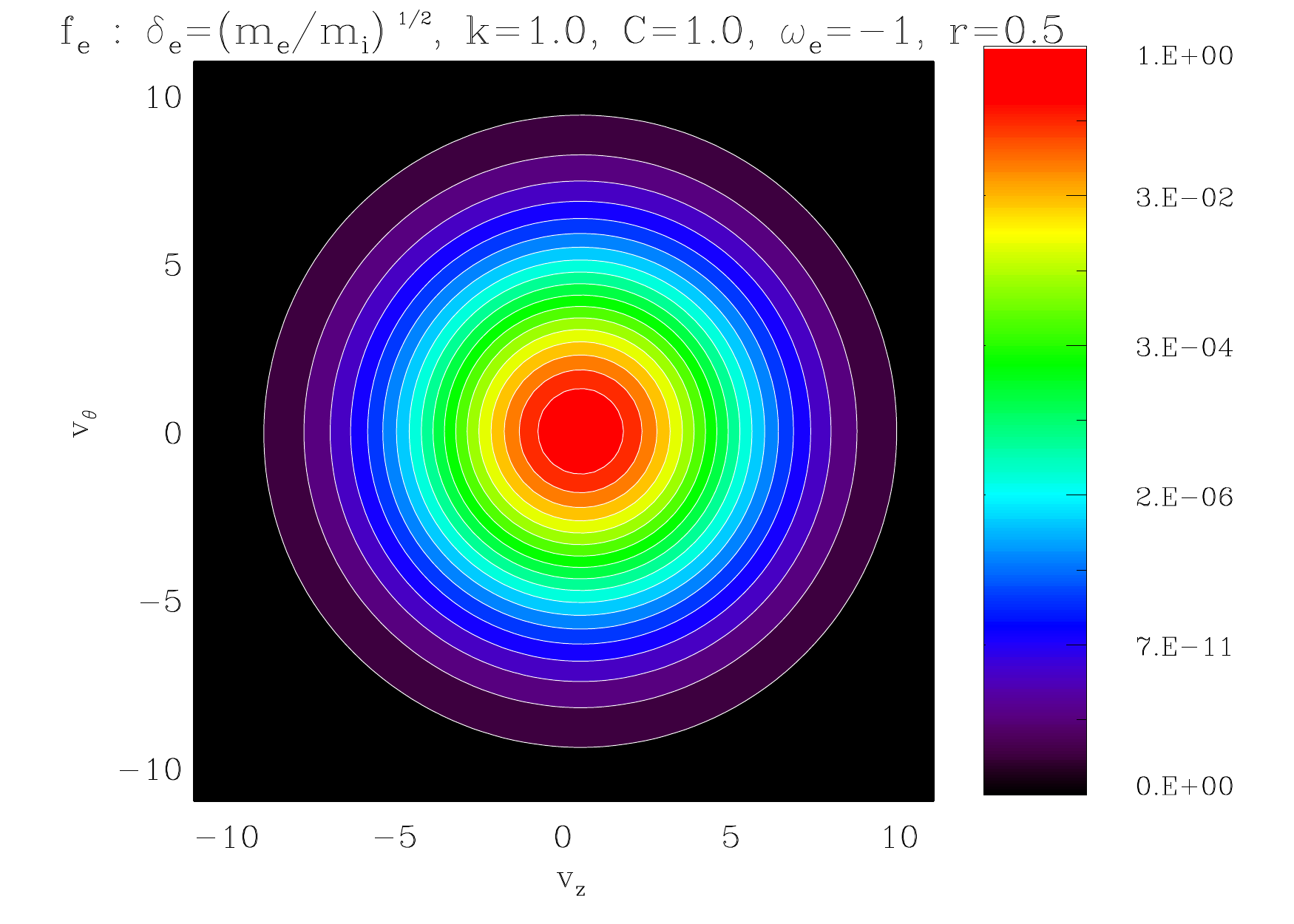}
        \caption{\small $(\tilde{\omega}_e,\tilde{r},C_e)=(-1,0.5,1)$}
        \label{fig:}
    \end{subfigure}
    \begin{subfigure}[b]{0.45\textwidth}
        \includegraphics[width=\textwidth]{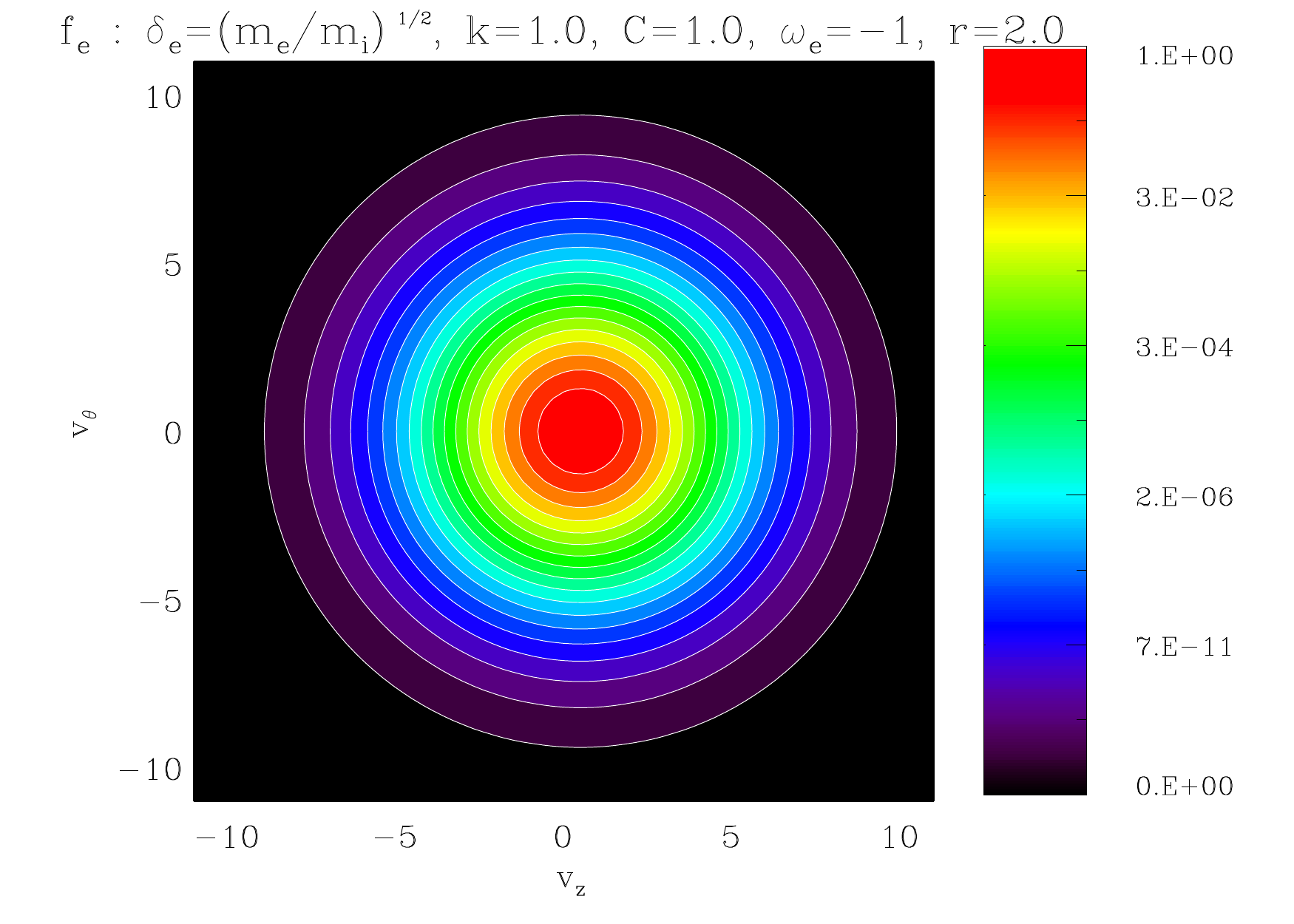}
        \caption{\small $(\tilde{\omega}_e,\tilde{r},C_e)=(-1,2,1)$}
        \label{fig:}
    \end{subfigure}
        \begin{subfigure}[b]{0.45\textwidth}
        \includegraphics[width=\textwidth]{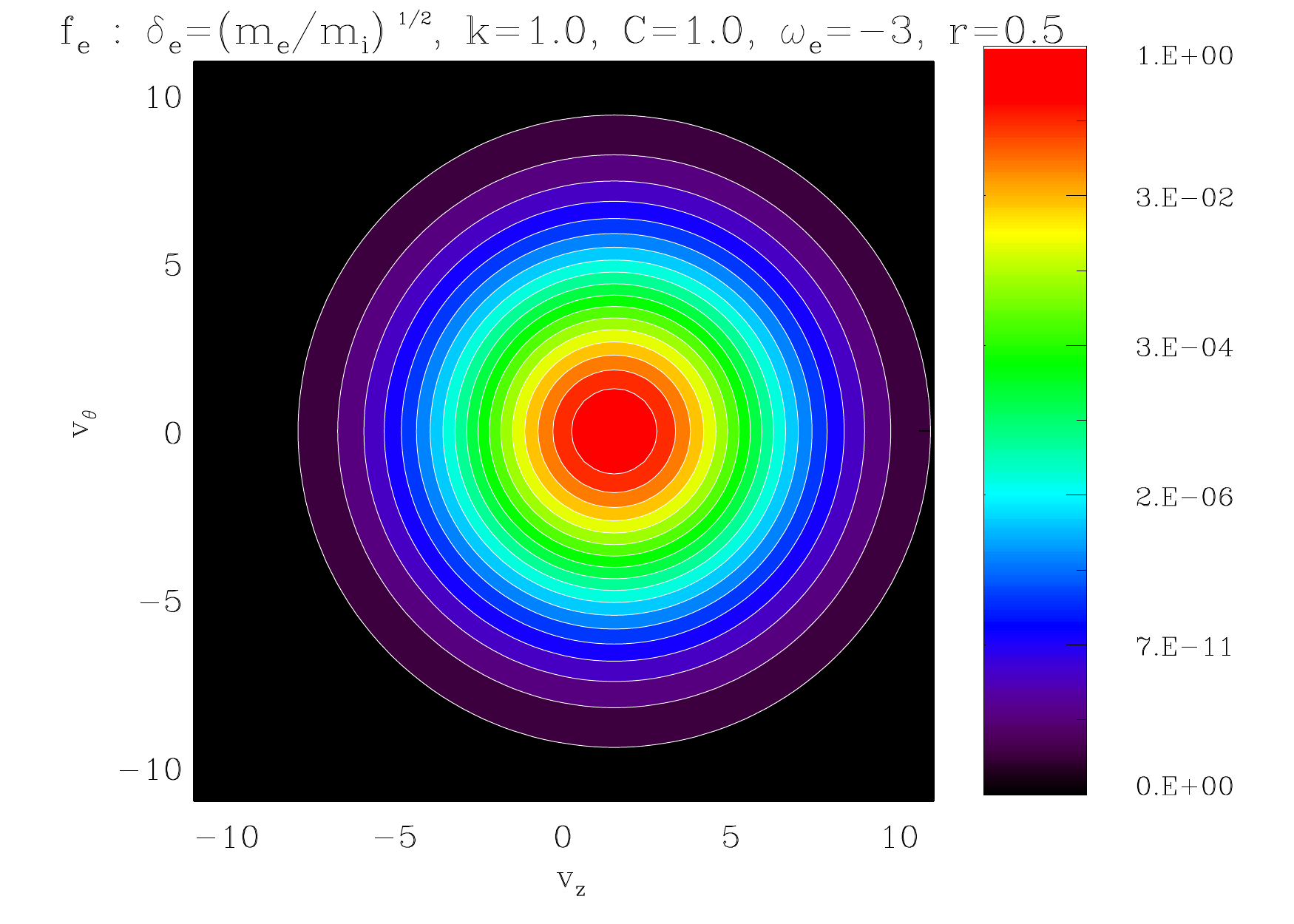}
        \caption{\small $(\tilde{\omega}_e,\tilde{r},C_e)=(-3,0.5,1)$}
        \label{fig:}
    \end{subfigure}
     \begin{subfigure}[b]{0.45\textwidth}
        \includegraphics[width=\textwidth]{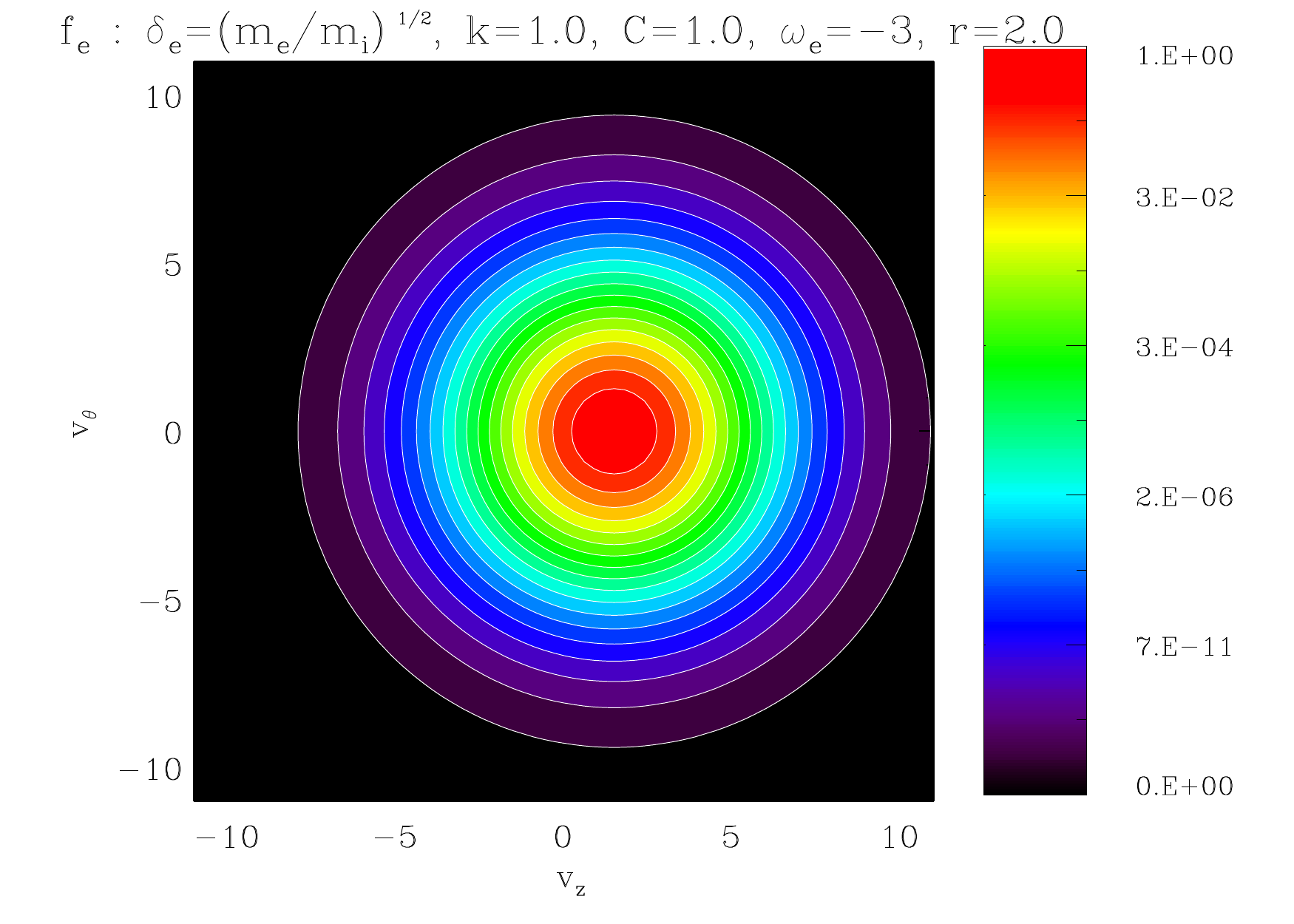}
        \caption{\small $(\tilde{\omega}_e,\tilde{r},C_e)=(-3,2,1)$}
        \label{fig:}
    \end{subfigure}
    \caption{\small       Contour plots of $f_e$ in $(\tilde{v}_z,\tilde{v}_{\theta})$ space for an equilibrium without field reversal ($k=1>0.5$), for a variety of parameters ($\tilde{\omega}_e,\tilde{r},C_e$) and $\delta_e\approx1/\sqrt{1836}$. Note that there are not any multiple maxima in this case.   }\label{fig:7}
\end{figure}

\subsection{Plots of the DF}
A characteristic that one immediately looks for in a new DF is the existence of multiple maxima in velocity space, which are a direct indication of non-thermalisation, relevant for the existence of micro-instabilities (e.g. see \cite{Gary-2005}). Using an analysis very similar to that in \citet{Neukirch-2009}, we can derive - for a given value of $\tilde{\omega}_{s}$ - conditions on $\tilde{r}$ and either $\tilde{v}_{z}$ or $\tilde{v}_{\theta}$, for the existence of multiple maxima in the $\tilde{v}_{\theta}$ or $\tilde{v}_{z}$ direction respectively. We present these calculations in Sections \ref{app:vthetamax} and \ref{app:vzmax}. The most readily understood results are that multiple maxima in the $\tilde{v}_{\theta}$ direction can only occur for $\tilde{r}>2/|\tilde{\omega}_{s}|$, and in the $\tilde{v}_z$ direction for $|\tilde{\omega}_s|>2$. Given these necessary conditions, one can then calculate that multiple maxima of $f_s$ will occur in the $\tilde{v}_{\theta}$ direction for $\tilde{v}_z$ bounded above and below, and vice versa. 

In Figures (\ref{fig:4}-\ref{fig:7}) we present plots of the DFs over a range of parameter values. Figures (\ref{fig:4}) and (\ref{fig:5}) show the ion DFs for $k=0.1$ and $k=1$ respectively, for all combinations of $\tilde{\omega}_{i}=1,3$, $\tilde{r}=0.5,2$ and $C_s=0.1, 1$, and with the magnetisation parameter $\delta_i=1$. As a graphical confirmation of the above discussion, we can only see multiple maxima in the $\tilde{v}_{\theta}$ direction for $\tilde{r}>2/|\tilde{\omega}_{s}|$, and in the $\tilde{v}_z$ direction for $|\tilde{\omega}_s|>2$, with the appropriate bounds marked by the horizontal/vertical white lines.  

Aside from multiple maxima in the orthogonal directions, the DF can also be `two-peaked'. That is, the DF can have two isolated peaks in $(\tilde{v}_z,\tilde{v}_{\theta})$ space. This is seen to occur for Figures (\ref{fig:5d}, \ref{fig:5g}, \ref{fig:5h}). Hence, $f_i$ is seen to be `two-peaked' when $k=1$ for both $\tilde{r}>2/\tilde{\omega}_i$ and $\tilde{r}<2/\tilde{\omega}_i$. However, we do not see a two-peaked DF for $k=0.1$. This seems to suggest that the stronger guide field ($k=1$) correlates with multiple peaks. Physically, this may correspond to the fact that a homogeneous guide field is consistent with a Maxwellian DF centred on the origin in $(\tilde{v}_z,\tilde{v}_{\theta})$ space, given that a Maxwellian contributes zero current. Hence, if the `main' part/peak of the DF is centred away from the origin, then the Maxwellian contribution from the guide field could contribute a secondary peak. These secondary peaks are seen to be more pronounced when $\tilde{C}_i$ is larger, i.e. the contribution from the second term from the DF is greater.  

Figures (\ref{fig:6}) and (\ref{fig:7}) show the electron DFs for $k=0.1$ and $k=1$ respectively, for all combinations of $\tilde{\omega}_{e}=1,3$; $\tilde{r}=0.5,2$, and $C_e=0.1,1$, and with the magnetisation parameter $\delta_e=\delta_i\sqrt{m_e/m_i}\approx 1/\sqrt{1836}$. This choice of magnetisation corresponds to $T_i=T_e$. In general we see DFs with fewer multiple maxima in velocity space than the ion plots, which is physically consistent with the electrons being more magnetised, i.e. more `fluid-like'. In particular we see no multiple maxima in Figure \ref{fig:7}, the case with the stronger background field. 

Note that when the electrons to have the same magnetisation as the ions, i.e. $\delta_e=\delta_i=1$, then these marked differences in the velocity-space plots disappear, and we observe a qualitative symmetry $f_i(\tilde{v}_\theta,\tilde{v}_z,r)\propto f_e(-\tilde{v}_\theta,-\tilde{v}_z,r)$.

\subsubsection{Maxima in $v_{\theta}$ space}\label{app:vthetamax}
The $\tilde{p}_{rs}$ dependence of the DF is irrelevant to our discussion, and as such can be integrated out. We can also neglect the scalar potential $\phi$. The reduced DF, $\tilde{F}_s$, in dimensionless form is 
\begin{eqnarray}
\tilde{F}_s=((\sqrt{2\pi}v_{\text{th},s})^2/n_{0s})\,e^{\tilde{\phi}_s}\,\int_{-\infty}^{\infty} \,f_s\, dv_{r},\nonumber
\end{eqnarray}
which then reads
\begin{eqnarray}
&&\tilde{F}_s=\exp\left\{-\frac{1}{2}\left[\left(\frac{\tilde{p}_{\theta s}}{\tilde{r}}-\tilde{A}_{\theta s}\right)^2+\left(\tilde{p}_{zs}-\tilde{A}_{zs}\right)^2\right]\right\}\times\nonumber\\
&&\left[\exp\left(\tilde{\omega}_s\tilde{p}_{\theta s}+\tilde{{U}}_{zs}\tilde{P}_{zs}  \right)+{C}_s\exp\left(\tilde{{V}}_{zs}\tilde{P}_{zs}   \right)\right].\label{eq:reduced}
\end{eqnarray}
We have written $\tilde{F}_s$ in terms of the canonical momenta, and so we search for stationary points given by $\partial \tilde{F}_s/\partial \tilde{p}_{\theta s}=0$, equivalent to $\partial \tilde{F}_s/\partial \tilde{v}_{\theta s}=0$.
Setting $\partial \tilde{F}_s/\partial \tilde{p}_{\theta s}=0$ gives
\begin{eqnarray}
\tilde{p}_{\theta s}-\tilde{r}\tilde{A}_{\theta s}&=&\frac{\tilde{\omega}_s\tilde{r}^2}{1+{C}_se^{-\tilde{\omega}_s\tilde{p}_{z s}}e^{-\tilde{\omega}_s\tilde{p}_{\theta s}}}\nonumber\\
&=&\frac{A}{1+Be^{-\tilde{\omega}_s\tilde{p}_{\theta s}}}:=R(\tilde{p}_{\theta s}).\label{eq:lineartheta}
\end{eqnarray}
To derive a necessary condition for multiple maxima, we analyse the RHS of Equation (\ref{eq:lineartheta}), $R(\tilde{p}_{\theta s})$. This function is bounded between 0 and A, and is monotonically increasing. Hence, using techniques similar to those in \citet{Neukirch-2009}, a necessary condition for multiple maxima in the DF is that 
\begin{eqnarray}
\max_{\tilde{p}_{\theta s}}R^{\prime}(\tilde{p}_{\theta s})>1,
\end{eqnarray} 
since the LHS of Equation (\ref{eq:lineartheta}) is a linear function of unit slope in $\tilde{p}_{\theta s}$. This condition can be shown to be equivalent to $A\tilde{\omega}_s/4>1$ and so 
\begin{equation}
\tilde{\omega}_s^2>4\tilde{r}^{-2}\iff\tilde{r}>2/|\tilde{\omega}_s|
\label{eq:rcond}
\end{equation}
This demonstrates that for sufficiently small $\tilde{r}$, there cannot exist multiple maxima. Equivalently, this condition will always be satisfied for some $\tilde{r}$, and as such is just a condition on the domain, in $\tilde{r}$, for which multiple maxima can occur. This condition is not sufficient however, as it could still be the case that there exists only one point of intersection (and hence one maximum), depending on the value of $B$. It is seen that $R$ has unit slope at 
\begin{eqnarray}
&&\tilde{p}_{\theta s}^{\pm}=\frac{1}{\tilde{\omega}_s}\times\nonumber\\
&&\left[\ln\left(2B\right)-\ln\left(A\tilde{\omega}_s-2\pm\sqrt{A\tilde{\omega}_s\left(A\tilde{\omega}_s-4\right)}\right)\right].
\label{eq:pthetastar}    
\end{eqnarray}
Clearly $R$ has unit slope for two values of $\tilde{p}_{\theta s}$. After some graphical consideration of the problem, it becomes apparent that $B$ should be bounded above and below for multiple maxima. After elementary consideration of the functional form of (\ref{eq:lineartheta}), for example with graph plotting software, we see that multiple maxima in the $\tilde{v}_{\theta}$ direction can only occur, for a given $\tilde{r}$, when $B$ (and hence $\tilde{v}_{z}$) satisfies these inequalities for ions
\begin{eqnarray}
&&\tilde{p}_{\theta i}^{ +}-R(\tilde{p}_{\theta i}^{ +})-\tilde{r}\tilde{A}_{\theta i}>0,\nonumber\\
&&\tilde{p}_{\theta i}^{ -}-R(\tilde{p}_{\theta i}^{ -})-\tilde{r}\tilde{A}_{\theta i}<0,
\end{eqnarray}
and these for electrons
\begin{eqnarray}
&&\tilde{p}_{\theta e}^{ +}-R(\tilde{p}_{\theta e}^{ +})-\tilde{r}\tilde{A}_{\theta e}<0,\nonumber\\
&&\tilde{p}_{\theta e}^{ -}-R(\tilde{p}_{\theta e}^{ -})-\tilde{r}\tilde{A}_{\theta e}>0.
\end{eqnarray}

\subsubsection{Maxima in $v_{z}$ space}\label{app:vzmax}
We shall once again use the reduced DF defined in Equation (\ref{eq:reduced}) in our analysis. Thus, we shall consider $\partial \tilde{F}_s/\partial \tilde{p}_{z s}=0$, which is equivalent to $\partial \tilde{F}_s/\partial \tilde{v}_{z s}=0$. Setting $\partial \tilde{F}_s/\partial \tilde{p}_{z s}=0$ gives
\begin{eqnarray}
\tilde{p}_{z s}-\tilde{A}_{z s}&=&\frac{\tilde{{U}}_{zs}+{C}_s\tilde{{V}}_{zs}e^{-\tilde{\omega}_s(\tilde{p}_{zs}+\tilde{p}_{\theta s})}}{1+{C}_se^{-\tilde{\omega}_s(\tilde{p}_{zs}+\tilde{p}_{\theta s})}}\nonumber\\
&=&\frac{A_1}{1+B_1e^{-D_1\tilde{p}_{z s}}}+\frac{A_2}{1+B_2e^{-D_2\tilde{p}_{z s}}}\nonumber\\
&:=&R_1(\tilde{p}_{z s})+R_2(\tilde{p}_{z s})=R(\tilde{p}_{zs})\label{eq:linearz},\nonumber
\end{eqnarray}
such that
\begin{eqnarray}
A_1&=&\tilde{{U}}_{zs},\hspace{3mm}A_2=\tilde{{V}}_{zs},\nonumber\\
B_1&=&{C}_se^{-\tilde{\omega}_s\tilde{p}_{\theta s}}=B_2^{-1},\hspace{3mm}D_1=\tilde{\omega}_s=-D_2.\nonumber
\end{eqnarray}
 To derive a necessary condition for multiple maxima, we analyse the RHS of Equation (\ref{eq:linearz}). Each $R$ function is bounded and  monotonic. Once again using techniques similar to those in \citet{Neukirch-2009}, a necessary condition for multiple maxima in the DF is that 
\begin{eqnarray}
\max_{\tilde{p}_{z s}}\left(R_1^{\prime}(\tilde{p}_{z s})+R_2^{\prime}(\tilde{p}_{z s})\right)>1.
\end{eqnarray} 
After some algebra this condition can be shown to be equivalent to $\tilde{\omega}_s^2/4>1$ and so 
\begin{eqnarray}
|\tilde{\omega}_s|>2.
\label{eq:omegacond}
\end{eqnarray}
This condition is not sufficient however, as it could still be the case that there exists only one point of intersection, depending on the value of $B_1(=1/B_2)$. The transition between 3 points of intersection and one occurs at the value of $B_1$ for which the straight line of slope unity through $\tilde{p}_{z s}=0$ just touches $R_1(\tilde{p}_{z s})+R_2(\tilde{p}_{z s})$ at the point where it also has unit slope. It is readily seen that $R_1+R_2$ has unit slope at 
\begin{eqnarray}
\tilde{p}_{z s}^{\pm}=\frac{1}{\tilde{\omega}_s}\left[\ln\left(  2B_{1} \right)      -   \ln\left(     \tilde{\omega}_s^2-2 \pm \sqrt{\tilde{\omega}_{s}^{2}(  \tilde{\omega}_{s}^{2}-4   )}    \right)     \right].
\end{eqnarray}
Clearly $R$ has unit slope for two values of $\tilde{p}_{z s}$. Once again, after some graphical consideration of the problem, it becomes apparent that $B_{1}$ should be bounded above and below for multiple maxima. After elementary consideration of the functional form of (\ref{eq:linearz}), for example with graph plotting software we see that multiple maxima in the $\tilde{v}_{z}$ direction can only occur, for a given $\tilde{r}$, when $B_1$ (and hence $\tilde{v}_{\theta}$) satisfies these inequalities for ions
\begin{eqnarray}
&&\tilde{p}_{z i}^{ +}-R(\tilde{p}_{z i}^{ +})-\tilde{A}_{z i}>0,\nonumber\\
&&\tilde{p}_{z i}^{ -}-R(\tilde{p}_{z i}^{ -})-\tilde{A}_{z i}<0,
\end{eqnarray}
and these for electrons
\begin{eqnarray}
&&\tilde{p}_{z e}^{ +}-R(\tilde{p}_{z e}^{ +})-\tilde{A}_{z e}<0,\nonumber\\
&&\tilde{p}_{z e}^{ -}-R(\tilde{p}_{z e}^{ -})-\tilde{A}_{z e}>0.
\end{eqnarray}

\section{Summary}

In this chapter we have calculated 1D collisionless equilibria for a continuum of magnetic field models based on the GH flux tube, with an additional constant background field in the axial direction. This study was motivated by a desire to extend the existing methods for solutions of the `inverse problem in Vlasov equilibria' in Cartesian geometry, to cylindrical geometry. 

In Section \ref{subsec:eom} we calculated the fluid equations of motion for a 1D system with azimuthal and axial flows, found by taking the first order velocity moment of the Vlasov equation in cylindrical coordinates. The presence of centripetal forces in the equation of motion demonstrated that it may be difficult to find Vlasov equilibrium DFs self-consistent with force-free fields.

However, initial efforts focussed on solving for the exact force-free GH field, but this seems impossible due to the centripetal forces, and this conclusion is somewhat corroborated by \citet{Vinogradov-2016}. The GH field in particular was chosen as it represents the `natural' analogue of the Force-Free Harris Sheet in cylindrical geometry, a magnetic field whose VM equilibria have been the subject of recent study, \citep{Harrison-2009PRL, Neukirch-2009, Wilson-2011, Abraham-Shrauner-2013, Kolotkov-2015}, as well as the work detailed in Chapters \ref{Vlasov} and \ref{Sheets}, featuring work from \citet{Allanson-2015POP,Allanson-2016JPP}

A background field was introduced, and an equilibrium DF was found that reproduces the required magnetic field, i.e. solves Amp\`{e}re's Law. It is the presence of the background field that allows us to solve Vlasov's equation and Amp\`{e}re's Law, and it appears physically necessary as it introduces an `asymmetry'; namely an extra term into the equation of motion whose sign depends explicitly on species. In contrast to the `demands' of insisting on a particular magnetic field, no condition was made on the electric field. The DF allows both electrically neutral and non-neutral configurations, and in the case of non-neutrality we find an exact and explicit solution to Poisson's equation for an electric field that decays like $1/r$ far from the axis. We note here that the type of solutions derived in this chapter could - after a Galilean transformation - be interpreted as 1D BGK modes with finite magnetic field (see \cite{Abraham-Shrauner-1968,Ng-2005,Grabbe-2005,Ng-2006} for example, to provide some context).

An analysis of the physical properties of the DF was given in Section \ref{sec:analysis}, with some particularly detailed calculations in Sections \ref{app:vthetamax} and \ref{app:vzmax}. The dependence of the sign of the charge density (and hence the electric field) on the bulk ion and electron rotational flows was analysed, with a physical interpretation given. Essentially the argument states that the electric field exists in order to balance the difference in the centrifugal forces (in the co-moving frame) between the two species. The DF was found to be able to give sub-unity values of the plasma beta, should this be required/desirable given the relevant physical system that it is intended to model. In Section \ref{subsec:origin} we performed a detailed analysis of the relationship between individual terms in the equation of motion. For clarity, the conclusions drawn for the macroscopic equilibrium considered in this chapter are that the electric field sources/balances gradients in the particle number densities; the centripetal forces are sourced/balanced by the bulk angular flows; and the $\boldsymbol{j}\times\boldsymbol{B}$ force is sourced/balanced by a centripetal-type force, that treats the flow as uniform circular motion, i.e. rotational flows consistent with a rigid-rotor (see Section \ref{sec:thedf}). The final part of the analysis focussed on plotting the DF in velocity space, for certain parameter values, and at different radii. Mathematical conditions were found that determine whether or not the DF could have multiple maxima in the orthogonal directions in velocity space, and these are corroborated by the plots of the DFs. For certain parameter values, the DF was also seen to have two separate, isolated peaks. This non-thermalisation suggests the existence of microinstabilities, for a certain choice of parameters.

Further work could involve a deeper analysis of the properties of the DFs and their stability. This work has also raised a fundamental question: `is it possible to describe a 1D force-free collisionless equilibrium in cylindrical geometry?' Preliminary investigations seem to suggest that it is not possible. It would also be of value to find out whether the relationships derived between individual terms in the equation of motion are totally general in nature, and if not, to what extent to they apply?

\null\newpage

\chapter{Discussion} 

\epigraph{\emph{For God's sake, stop researching for a while and begin to think.}}{\textit{Walter Hamilton Moberley}}

\label{Discussion} 
The details of the main results of this thesis have been explained in the preambles and summaries of Chapters \ref{Vlasov}, \ref{Sheets}, \ref{Asymmetric} and \ref{Cylindrical}, and as such we shall not duplicate that information. Here, it is the intention to place the motivation of the work and the results in context with regards to personal research direction, broader questions, and suggestions for future work.

\section{Context}

The overarching physical motivation for the work in this thesis is perhaps embodied by - and has its roots in - the `GEM challenge': \emph{`The goal is to identify the essential physics which is required to model collisionless magnetic reconnection'}, \citep{Birn-2001}. However, this thesis does not focus on the analysis of instability and reconnection itself. The results in this thesis are on the theoretical modelling of Vlasov-Maxwell equilibria, with the approach being a mixture of `general scientific curiosity' (e.g. Chapters \ref{Vlasov} and \ref{Cylindrical}), and the application to particular physical problems (e.g. Chapters \ref{Sheets}, \ref{Asymmetric} and \ref{Cylindrical}).

\subsection{Current sheets}
Much of the research effort in tackling the GEM challenge has been spent on antiparallel (i.e. $B_{x}(z)=-B_x(-z)$) reconnection, with initial equilibrium conditions as symmetric 1D current sheets (e.g. see \cite{Hesse-2001,Birn-2005} for examples with and without guide fields $B_y$ respectively). In particular, the Harris current sheet model (or some modification) is very frequently used, in no small part due to the well-known exact Vlasov-Maxwell equilibrium DF \citep{Harris-1962},
\[
f_s=\frac{n_{0s}}{(\sqrt{2\pi}v_{\text{th},s})^3}e^{-\beta_s(H_s-u_{ys}p_{ys})}.
\]

It is possible to approximate force-free ($\boldsymbol{j}\times\boldsymbol{B}=\boldsymbol{0}$) conditions, relevant to the $\beta_{pl}\ll 1$ conditions in the solar corona, by assuming a strong, uniform guide field $B_y(z)=B_{y0}\gg B_{x0}$,
\[
\boldsymbol{B}=(B_{x0}\tanh\tilde{z},B_{y0},0).
\] 
However, as discussed in Chapter \ref{Sheets}, the nature of such an equilibrium does not accurately represent a true force-free equilibrium, such as the force-free Harris sheet,
\[
\boldsymbol{B}=B_0(\tanh\tilde{z},\text{sech}\tilde{z},0).
\] 
Until the discovery of the first VM equilibrium DF for a nonlinear force-free field (the \emph{Harrison-Neukirch} equilibrium for the force-free Harris sheet) by \citet{Harrison-2009PRL}, the analysis of reconnection and instability of force-free fields had to be limited to the use of exact initial conditions for a uniform strong guide field configuration, e.g. \citet{Ricci-2004}; the use of inexact initial conditions (drifting Maxwellians) for an exact nonlinear force-free field (e.g. \cite{Birn-2010}); or one would have to use a linear force-free model (e.g. \cite{Bobrova-2001}), for which one cannot isolate and study a single current sheet. We are now beginning to see the first analyses of linear stability \citep{Wilson-2017}, and reconnection \citep{Wilson-2016} for exact nonlinear force-free current sheet models.

The Harrison-Neukirch equilibrium does have one fairly significant drawback, with regards to its use in a low plasma beta environment. Due to technical reasons regarding the manner in which the Vlasov-Maxwell equilibrium was constructed, $\beta_{pl}$ is bounded below by unity. This feature motivated our investigations of low-beta Vlasov-Maxwell equilibria for the force-free Harris sheet \citep{Allanson-2015POP, Allanson-2016JPP}, as discussed in Chapter \ref{Sheets}. The key step in reducing the lower bound for $\beta_{pl}$, was the use of pressure tensor transformation techniques, as discussed in \citet{Harrison-2009POP}, and for which we chose an exponential function. This transformation made the inverse problem \citep{Channell-1976} difficult to solve, and confidence in the solution necessitated some rigorous mathematical work (see \cite{Allanson-2016JPP}) and Chapter \ref{Vlasov}.

It is now established that \emph{`magnetic reconnection relies on the presence of a diffusion region, where collisionless or collisional plasma processes facilitate the changes in magnetic connection through the generation of dissipative electric fields'} \citep{Hesse-2011}. The very recent (and current) NASA MMS mission is able to make in-situ diffusion region measurements on kinetic scales for the very first time \citep{Burch-2016Science, Hesse-2016}. The satellite will focus on the dayside magnetopause in the first phase of its mission, and the magnetotail in the second phase. Current sheets in the dayside magnetopause are typically of a rather different nature than those of the symmetric Harris sheet type, by virtue of the asymmetric conditions either side of the current sheet. The magnetosheath side is characterised by an enhanced thermal pressure and depleted magnetic pressure, and vice versa for the magnetosphere side. Exact analytical \citep{Alpers-1969} and numerical \citep{Belmont-2012,Dorville-2015} Vlasov-Maxwell equilibria are few in number, and so the work in Chapter \ref{Asymmetric} and \citet{Allanson-2017GRL} is targeted towards improving this situation. In particular, the exact analytical solution due to \citet{Alpers-1969} has different bulk flow properties to the one that we present.

\subsection{Flux tubes}
Localised currents need not always obey a planar geometry; flux tubes play an important role in confinement and subsequent energy release in many areas of plasma physics (see Chapter \ref{Cylindrical}), and particularly in the solar corona (e.g. see \cite{Wiegelmann, Hood-2016}), as well as the extended structure of magnetic islands, perpendicular to current sheets in the magnetopause and magnetotail (e.g. see \cite{Kivelson-1995,Vinogradov-2016}). Hence it was with a combination of mathematical curiosity, and a desire to model nonlinear force-free flux tubes, that we attempted to calculate exact Vlasov-Maxwell equilibria for the Gold-Hoyle flux tube \citep{Gold-1960}, the natural analogue of the force-free Harris sheet in cylindrical geometry \citep{Tassi-2008}. The work is detailed in Chapter \ref{Cylindrical} and \citet{Allanson-2016POP}, and in fact we were unable to find solutions for the exact nonlinear force-free Gold-Hoyle model. However, the magnetic field can be arbitrarily close to a force-free field if desired. An interesting feature of the analysis focussed on the need to include non-neutrality and non-zero electric fields in the equilibrium, brought about by charge separation effects, inherent in the rotational motion of particles with different masses.





\section{Broader theoretical questions}

\subsection{The pressure tensor}
In a one-dimensional and $z$-dependent geometry, the `keystone' of the inverse problem is the pressure tensor component $P_{zz}(A_x,A_y)$: given a magnetic field, one first attempts to calculate $P_{zz}$, and then self consistent distribution functions. The main theoretical/mathematical developments in this thesis (related to Cartesian geometry) have focussed on the second step in this process, i.e. calculating self-consistent DFs, of the form
\[
f_s=\frac{n_{0s}}{(\sqrt{2\pi}v_{\text{th},s})^3}e^{-\beta_sH_s}g_s(p_{xs},p_{ys}),\label{eq:summarydf}
\]
given a $P_{zz}(A_x,A_y)$. However, there remain important questions about the determination of the $P_{zz}$ function itself. 

As discussed in Chapters \ref{Intro} and \ref{Asymmetric}, the problem of determining $P_{zz}(A_x,A_y)$ given a magnetic field (in force balance) is analogous to that of determining the shape of a conservative potential function, $\mathcal{V}(\boldsymbol{x})$, given the knowledge of the particle trajectory, $\boldsymbol{x}(t)$, and the value of the potential along the trajectory, $\mathcal{V}(t)$. In the case of 1D force-free fields there is an algorithmic path that determines a valid form of $P_{zz}$ (e.g. see Chapter \ref{Sheets}, \cite{Harrison-2009POP}). The question remains: `to what extent is it possible to find self-consistent $P_{zz}$ functions for a given magnetic field, and what are they?'

One other feature of interest is the solubility of Amp\`{e}re's Law,
\[
\frac{\partial P_{zz}}{\partial{\boldsymbol{A}}}=-\frac{1}{\mu_0}\frac{d^2\boldsymbol{A}}{dz^2},
\]
with respect to different $P_{zz}$ expressions. As demonstrated in Chapter \ref{Sheets} and \citet{Harrison-2009POP} for the case of force-free fields; given one $P_{zz}$ that satisfies Amp\`{e}re's Law, there exist infinitely many others. There are two obvious questions here. Firstly, it would be interesting to investigate if there are ways to transform the Harrison-Neukirch pressure function to allow sub-unity values of the plasma beta, in a way that is more readily soluble and easier to manipulate numerically than the result found in Chapter \ref{Sheets} and \citet{Allanson-2015POP, Allanson-2016JPP}. Secondly, is it in any way possible to extend the pressure transformation theory for force-free equilibria to non force-free equilibria? If so, then the theory is to be expected to be more complicated than for force-free fields, which relies on $P_{zz}$ being a constant when evaluated along the force-free trajectory $(A_x(z),A_y(z))$.

\subsection{Non-uniqueness}
One clear challenge is to marry together the need for individual, exact solutions of the inverse problem for Vlasov-Maxwell equilibria, versus the fact that there are in principle infinitely many solutions. In essence, how do we know that a given Vlasov-Maxwell equilibrium is appropriate physically? In Chapter \ref{Intro} we gave arguments for suggesting why distribution functions of the form in Equation (\ref{eq:summarydf}) were reasonable on both physical and mathematical grounds. In particular, this form of distribution function bears a strong resemblance to a (drifting) Maxwellian. Hence, provided the $g_s$ function is not too `exotic', it seems reasonable that these distribution functions can - for a certain choice of microscopic parameters - minimise the free energy (maximise the entropy) in a plasma, given certain constraints such as the conservation of energy in a closed system (e.g. see \cite{Schindlerbook}).

The inverse problem is characterised by non-uniqueness on the level of the $P_{zz}$ for a given $\boldsymbol{B}$, and on the level of $f_s$ for a given $P_{zz}$. It would be of interest to see if - given a distribution function of the form in Equation (\ref{eq:summarydf}) - the inversion of the Weierstrass transform gives a unique solution and if not, whether the inversion method (e.g. Fourier transform or Hermite polynomial expansion) has an effect on the outcome. As discussed in Chapter \ref{Vlasov}, these considerations are related to the `backwards uniqueness of the heat equation' \citep{Evansbook}, with $g_s$ and $P_{zz}$ somewhat equivalent to the initial and final `heat' distributions over a two-dimensional surface. 

\begin{figure}
    \centering
      \includegraphics[width=0.7\textwidth]{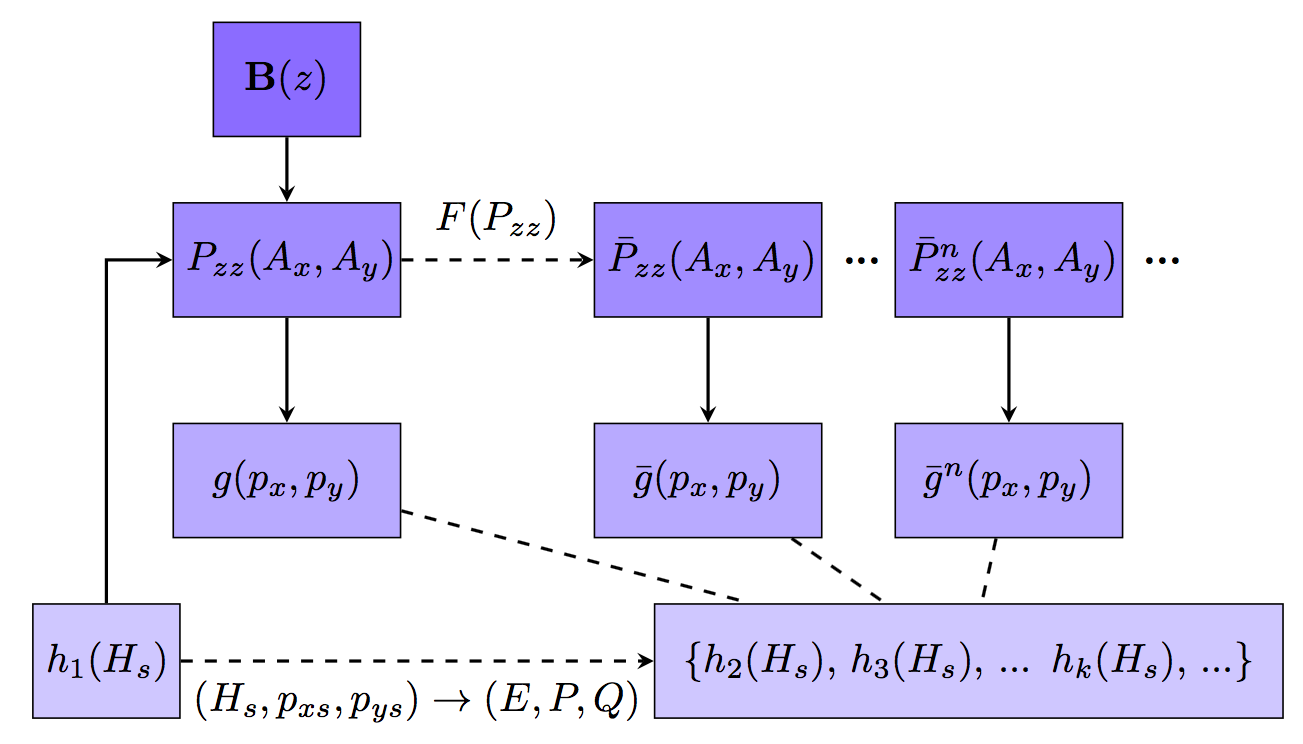}
    \caption{\small A schematic representation of the inverse problem in Vlasov-Maxwell equilibria.}\label{fig:structure}
\end{figure}

An explicit demonstration of the non-uniqueness of the inverse problem (on the level of $f_s$ for a given $P_{zz}$) was given by \citet{Wilson-2011} for the case of the force-free Harris sheet, and using ideas from \citet{Schmid-Burgk-1965}. As discussed in the Appendix, it is possible to rewrite the relevant integral equations in $dH_sdp_{xs}dp_{ys}$ space. When this is done, it soon becomes apparent that - in the case of $\phi=0$ - that there is considerable freedom in the dependency of the DF on $H_s$, for a given $P_{zz}(A_x,A_y)$. This is related to the `convoluted' nature of the $(A_x,A_y)$ and $(p_{xs},p_{ys})$ variables, and as such the $g_s(p_{xs},p_{ys})$ function and the $P_{zz}(A_x,A_y)$ function can be considered `tied' together, with flexibility in the function of energy. 

Putting all of this together, we see that the non-uniqueness of the inverse problem can be represented by Figure \ref{fig:structure}, which works as follows. For a given $\boldsymbol{B}$, one can attempt to find a self-consistent $P_{zz}$. In that case, one might assume the energy dependence of the DF to be of a certain form, e.g. $h=\exp (-\beta_sH_s)$, and then solve the inverse problem for $g_s$. Once these $g_s$ functions are found, it may be possible to find other $h$ functions that are self-consistent with the same $P_{zz}$, and hence $\boldsymbol{B}$. On top of all this, there could in practice be infinitely many such compatible $P_{zz}$ functions (which in the force-free case can be found using established pressure transformation theory). For each of these $P_{zz}$ functions one could then attempt to solve the inverse problem for $g_s$, given an assumed form of $h_s$. Once this is achieved, it may be possible to generalise the energy dependency once more.

In summary, we believe that there is more work to be done regarding the non-uniqueness of Vlasov-Maxwell equilibria. It would be desirable to be able to have a `road-map' of the variety of solutions to the inverse problem, with a clearer understanding of how they relate to one another in their mathematical structure, and their suitability for physical applications. In particular, can the somewhat complicated structure of the diagram in Figure (\ref{fig:structure}) be simplified, or brought in to a more holistic form, and to what extent can the heat/diffusion equation analogy be brought to bear on the problem at hand?

\subsection{Extensions to other physical systems and geometries}
Clearly, not all collisionless plasma equilibria can be modelled in a one-dimensional, Cartesian, strictly neutral and non-relativistic framework. For example, one might really need to consider two-dimensional current sheets in the Earth's magnetotail (e.g. see \cite{Artemyev-2013}), cylindrical geometry in a tokamak (e.g. see \cite{Tasso-2014}), non-neutral plasmas in nonlinear electrostatic structures (e.g. see \cite{Ng-2006, Vasko-2016}), and relativistic equilibria in pulsar magnetospheres (e.g. see \cite{DeVore-2015}). In contrast to the `forward problem', the theory for the `inverse problem' is only really well-developed for one-dimensional quaineutral plasmas in a Cartesian geometry, like those considered in this thesis. It would clearly be of interest to try and develop the methods of the inverse problem in some or all of these directions.

The generalisation that seems - at a first `glance' - to be the most readily made, is to two-dimensional plasmas. In fact, this is the paradigm in which the `forward problem' is most usually considered (e.g. see \cite{Schindler-2002,Schindlerbook, Artemyev-2013}). However, if one uses Jeans' theorem with the constants of motion of Hamiltonian and the canonical momenta, there is a clear trade-off between spatial invariance, and the number of non-zero components of the current density. To be precise, if we now let the system depend on both $x$ and $z$, then $p_{xs}$ is no longer a conserved quantity. In the absence of other conserved quantities, we now only have $H_s$ and $p_{ys}$ for the variables in the distribution function, and as such we can only model plasmas with a current density in the $y$ direction, and fields that are of the form
\begin{eqnarray}
\boldsymbol{A}&=&(0,A_y(x,z),0),\nonumber\\
\boldsymbol{B}&=&(B_x(x,z),0,B_z(x,z)),\nonumber\\
\boldsymbol{j}&=&(0,j_y(x,z),0).\nonumber
\end{eqnarray}
Note that since $j_x=j_z=0$, we could in principle add a constant $B_y$ field, and hence $A_x,A_z$ that are linear functions of $x,z$. This would not break the self-consistency with the Vlasov approach, provided the distribution function had no dependence on $A_x$ or $A_z$. This is somewhat similar to the realisation that the distribution function for the Harris sheet, is also self consistent with the Harris sheet plus guide field.

So we see there is a challenge if one wishes to maintain flexibility in both the spatial variance of the plasma considered, as well as more than one current carrying component. Formally speaking, one would have to proceed by identifying further exact (or approximate/adiabatic) constants of motion, in order to have more than one current component (e.g. see \cite{Schindlerbook, Zelenyi-2011} for discussions of these topics).

The `grand goal' of all of this theoretical work is, in my mind, some sort of unification of the forward and inverse approaches. Can we establish a framework that includes physically meaningful Vlasov-Maxwell equilibria, for which there are clear and well-understood routes from the microscopic Vlasov description of particles, to the macroscopic description of fluids and fields, and vice versa? First of all, I would be motivated to develop the forward/inverse theory - beyond quasineutrality - for distribution functions of the form described in \citet{Mottez-2004}
\[
f_s(H_s,p_{xs},p_{ys})=\int_{a_1}^{a_2}\frac{n_{0s}(a)}{(\sqrt{2\pi}v_{\text{th},sa})}e^{-\beta_{sa}(H_s-u_{xsa}p_{xs}-u_{ysa}p_{ys})} da,
\]
for $a_1,a_2$ constants, and $f_s$ the distribution function, which is formed by a continuous superposition over the index/variable $a$, and for which the $g_s$ functions have been written as exponentials, i.e. eigenfunctions of the Weierstrass transform. The $a$ variable indexes the thermal velocity, thermal beta, and the drift parameters, and $f_s$ reduces to a more immediately recognisable distribution function when $n_{0s}(a)=\delta(a-c)n_{0s}$, for $a_1<c<a_2$ and $n_{0s}$ a constant. A first step in this direction might be to consider a discrete superposition rather than a continuous one, i.e. for $n_{0s}(a)=\sum_j\delta (a-a_j)$.

\subsection{Stability}
As mentioned throughout this thesis, but never really explored, a theoretical understanding of equilibria is not complete without understanding their stability properties. Knowledge of Vlasov-Maxwell equilibria allows one to study micro-instabilities in phase space \citep{Gary-2005}, for which non-thermal distribution functions are a pre-condition (i.e. multiple maxima and/or anisotropic distributions in velocity space). And keeping in mind the `main' physical motivation for this body of work, we would be interested in considering instabilities involved in magnetic reconnection, e.g. the tearing mode (e.g. see \citet{Furth-1963,Drake-1977}).  

There are two main approaches to assess the stability of a (kinetic) equilibrium
\begin{description}
\item[Normal mode analysis] (e.g. see \citet{Daughton-1999, Gary-2005}): Linearise the Vlasov-Maxwell equations by expressing quantities in the form $f_{s}=f_{0s}+f_{1s}, \boldsymbol{B}=\boldsymbol{B}_0+\boldsymbol{B}_1$ etc, for the first order quantities as small perturbations to the zeroth order ones, to arrive at,
\[
\frac{df_{1s}}{dt}=-\frac{q_s}{m_s}\left(\boldsymbol{E}_1+\boldsymbol{v}\times\boldsymbol{B}_1\right)\cdot\frac{\partial f_{0s}}{\partial \boldsymbol{v}}.
\]
One then subjects this equation to a Laplace/Fourier analysis in time/space (perturbed quantities $\sim e^{i(\boldsymbol{k}\cdot\boldsymbol{x}-\omega t)}$, for $\boldsymbol{k}$ the real wave-vector, and $\omega$ the complex frequency), with the aim being to solve for $f_{1s}$, by integrating the RHS over the `unperturbed orbits'. One can then - in principle - use the knowledge of $f_{1s}$ to calculate the source terms, $\sigma_1$ and $\boldsymbol{j}_1$. The source terms and the perturbed distribution function can then be substituted into the linearised Maxwell equations, from which one attempts to calculate a dispersion relation, $\omega=\omega (\boldsymbol{k})$. The results of this analysis is that for certain $\boldsymbol{k}$, and $\omega=\omega_r+i\gamma$, one should see that the equilibrium is linearly stable to some perturbations ($\gamma<0$), and unstable to others ($\gamma>0$). This approach does not only tell the analyst the perturbations for which the equilibrium is unstable, but it also yields the `damping/growth-rate', $|\gamma |$, which tells us how quickly the perturbation damps/grows.
\item[The (linear and nonlinear) energy principles]: This approach counts a system as stable if \emph{``a suitably selected test energy remains bounded by the energy supplied from external sources.'} \citep{Schindlerbook}. In the linear approach, the method essentially rests on first calculating the total energy over the spatial domain (for which there is no energy flux across the boundaries). For example, assuming the electric energy density is vanishing (consistent with quasineutrality), the energy is given by
\[ 
W=\sum_s\int \frac{m_s}{2}v^2 f_sd^3vd^3x+\int\frac{1}{2\mu_0}B^2d^3x.
\]
Then, assuming linear perturbations of the form $f_{s}=f_{0s}+f_{1s}, \boldsymbol{B}=\boldsymbol{B}_0+\boldsymbol{B}_1$ etc, one tries to ascertain whether - under certain dynamical constraints - there is a \emph{``dynamic conversion of equilibrium energy into kinetic energy''} \citep{Schindlerbook}. If there is no dynamic conversion, then the equilibrium is said to be linearly stable. The energy approach typically provides sufficient criteria for stability, as opposed to necessary ones.
\end{description}

Preliminary analysis of the kinetic stability properties of the force-free Harris sheet have been conducted in \citet{Harrison-thesis, Wilson-thesis}. In \citet{Wilson-2016} the first particle-in-cell simulations were performed with exact intiial conditions for a nonlinear force-free field. In \citet{Wilson-2017} we carry out a normal-mode analysis for the collisionless tearing mode, of the manner described above, and for the Harrison-Neukirch equilibrium \citep{Harrison-2009PRL}. It is of interest to study the stability properties of exact force-free tangential equilibria - for which $\boldsymbol{B}\cdot\nabla=0$ and $\nabla n=0$ - since `density-driven/drift instabilities' (e.g. the lower hybrid drift instability) will not be present \citep{Gary-2005}.

Possible future work could include normal mode/energy principle and/or numerical (i.e. particle-in-cell) instability analyses of the specific equilibria presented in this thesis, and particularly that presented in Chapter \ref{Asymmetric}, given the timely relevance to the MMS mission. One might also wish to study the stability analysis of distribution functions in a general sense, viz: ``given a distribution function function that is a solution of the inverse problem, what are its necessary/sufficient stability properties, and how does it grow/damp?''




\null\newpage
\appendix
\chapter{Schmid-Burgk variables} 
 \label{Appendix}
\noindent This Appendix is based on results in \citet{Schmid-Burgk-1965,Wilson-2011}.
\section{Species-independent integrals}
For a general DF of the form $f_s=f_s(H_s,p_{xs},p_{ys})$, we see that $P_{zz}$ is given by
\[
P_{zz}=2\sum_s \frac{1}{m_s^3}\int_{-\infty}^\infty \int_{-\infty}^\infty \int_{H_{s,\text{min}}}^\infty \sqrt{2m_s(H_s-H_{s,\text{min}})}f_sdH_sdp_{xs}dp_{ys},
\]
for $H_{s,\text{min}}=[(p_{xs}-q_sA_x)^2+(p_{ys}-q_sA_y)^2]/(2m_s)+q_s\phi$. At this stage it seems clear that the result of the integral is species-dependent. If one makes substitutions using \emph{Schmid-Burgk variables},
\begin{eqnarray}
(E_s,P_s,Q_s)&=&\left(\frac{m_sH_s}{q_s^2},\frac{p_{xs}}{q_s},\frac{p_{ys}}{q_s}\right),\nonumber\\
F_s(E_s,P_s,Q_s)&=&\frac{m_s^3}{q_s^4}f_s(H_s,p_{xs},p_{ys}),\nonumber
\end{eqnarray}
then $P_{zz}$ is now written
\[
P_{zz}=2\sum_s \frac{e}{m_s}\int_{-\infty}^\infty \int_{-\infty}^\infty \int_{E_{s,\text{min}}}^\infty \sqrt{(E_s-E_{s,\text{min}})}F_sdE_sdP_sdQ_s,
\]
for $E_{s,\text{min}}=[(P_s-A_x)^2+(Q_s-A_y)^2]/2+\frac{q_s}{m_s}\phi$. As yet, we have only made substitutions, and there have been no restrictions. However, if we now assume strict neutrality, $\phi=0$, and - crucially - assume that the \emph{functional form} of the $F_s$ function is independent of species, then the above expression has an interesting property. Note that when we say `functional form is independent of species', we mean that regardless of the species $s$, the function $F_s$ maps the inputs ($E_s,P_s,Q_s$) according to the same rules, i.e.
\[
F_s(E_s,P_s,Q_s)=F(E_s,P_s,Q_s),
\] 
(for example, it cannot use an exponential function for ions, and a quadratic function for electrons). Under these assumptions, the triple integral in the $P_{zz}$ expression actually becomes species-independent. The $(E_s,P_s,Q_s)$ variables are nothing but dummy variables, and the integrand itself is now of the same form, regardless of $s$. As a result, $P_{zz}$ becomes
\begin{equation}
P_{zz}(A_x,A_y)=2e\left(\frac{1}{m_e}+\frac{1}{m_i}\right)\int_{-\infty}^\infty \int_{-\infty}^\infty \int_{E_{s,\text{min}}}^\infty \sqrt{(E_s-E_{s,\text{min}})}FdE_sdP_sdQ_s.\label{eq:Schmid}
\end{equation}
Similarly it can be shown that the charge density is given by
\[
\sigma(A_x,A_y)=2\int_{-\infty}^\infty \int_{-\infty}^\infty \int_{E_{s,\text{min}}}^\infty (E_s-E_{s,\text{min}})^{-1/2}FdE_sdP_sdQ_s\sum_{s}\frac{q_{s}}{e}=0,
\]
and we see that the DF is automatically self-consistent with the assumption of strict neutrality. 

The Schmid-Burgk variables have helped us to demonstrate that the species-dependency of velocity moments of the DF enter through a $q_s/m_s$ factor that multiplies the scalar potential, and through any `innate' species-dependency that the DFs may have in themselves. In particular, the assumption of strict neutrality is automatically self-consistent if $F_s=F$ (in the case of an electron-ion plasma, or any plasma for which $\sum_s q_s/|q_s|=0$).

\subsection{Freedom in the energy dependency}
Using the Schmid-Burgk variables and the assumptions explained above ($\phi=0,F_s=F$), \citet{Wilson-2011} show - for the the example of the FFHS - that it is possible under certain conditions to solve the inverse problem with a DF of the general form
\[
F = h(E_s)g(p_{xs},p_{ys}),
\]
and with the $h$ function not only of the typically assumed exponential form, but of a reasonably arbitrary nature. This process is demonstrated for $h$ functions that are in Dirac delta form ($\delta(E_s-E_{0})$), Step function form $(\Theta (E_{0}-E_s))$, and polynomial form ($\Theta (E_{0}-E_s)\,(E_{0}-E_s)^\chi $, for $\chi>-1$). 

As such, we can consider the $P_{zz}(A_x,A_y)$ and $g_s(p_{xs},p_{ys})$ functions as `tied' together. This `tie' is evidenced by the convoluted nature of the variables $\boldsymbol{A}$ and $\boldsymbol{p}_s$ in the relevant integral equations, i.e. velocity moments of the DF, in general form, are given by
\begin{eqnarray}
&&\langle v_j^kf_s\rangle(A_x,A_y):=\frac{n_{0s}}{(\sqrt{2\pi}v_{\text{th},s})^3} \frac{2}{m_s^{k+2}}\times\nonumber\\
&&\int_{-\infty}^\infty \int_{-\infty}^\infty \int_{H_{s,\text{min}}}^\infty\frac{(p_{js}-q_sA_j)^k}{\sqrt{2m_s(H_{s,\text{min}}-H_s)}} f_s(H_s,p_{xs},p_{ys}) dH_sdp_{xs}dp_{ys}\nonumber .
\end{eqnarray}

\subsection{Summary}
In summary, the Schmid-Burgk variables have helped us to see that in the case of strictly neutral plasmas, there is evidence to suggest that the inverse problem should be framed as as: \emph{``for a given macroscopic equilibrium, i.e. a $P_{zz}(a_x,A_y)$, what are the self-consistent $g$ functions''}, for
\[
f_s\propto h(E_s)g(P_s,Q_s),
\]
as opposed to: \emph{``for a given macroscopic equilibrium, i.e. a $P_{zz}(a_x,A_y)$, what are the self-consistent DFs?''}


\null\newpage

\listoffigures \addchaptertocentry{List of figures}






\end{document}